\documentclass[a4paper,twoside,12pt,titlepage]{mybook}
\usepackage{a4}

\usepackage{amsmath,amsfonts,amssymb,amsthm}
\usepackage{url}
\usepackage{float}
\usepackage{graphicx}
\usepackage{subfigure}
\usepackage[sectionbib]{bibunits}
\usepackage{algorithmic}
\usepackage{setspace}
\usepackage{fancyhdr}
  \renewcommand{\chaptermark}[1]{\markboth{\chaptername \ \thechapter \ \ #1}{}}
  
\usepackage{adfatitlepage}
\usepackage{rotating}
\usepackage[none]{hyphenat}
\usepackage[intoc]{nomen}
\usepackage{pstricks}
\usepackage{array}
\usepackage{substr}
\usepackage{hyphenat}
\usepackage[chapter]{algorithm}
\usepackage{url}
\usepackage{verbatim}
\usepackage{cite}
\usepackage{mathrsfs}
\usepackage{multirow}
\usepackage{bm}
\usepackage{lineno}
\usepackage{array}
\usepackage{color}
\usepackage{latexsym}
\usepackage{amsxtra}
\usepackage{acronym}
\usepackage{longtable}
\usepackage{nomencl}
\usepackage{titletoc}
\usepackage{stfloats}
\usepackage{hhline}
\usepackage{tabularx}
\usepackage{makecell}
\usepackage{enumerate}
\usepackage{epstopdf}
\usepackage{epsfig}
\usepackage{hyperref}
\usepackage{pifont}
\usepackage{bbding}
\usepackage{multicol}
\usepackage{tablefootnote}
\newcommand{\abs}[1]{\lvert{#1}\rvert}

\newcommand{\iid}{i.\@i.\@d.\ }

\newcommand{\tabitem}{~~\llap{\textbullet}~~}

\DeclareMathOperator{\Tr}{Tr}

\DeclareMathOperator{\maxo}{maximize}
\DeclareMathOperator{\mino}{minimize}

\newtheorem{Thm}{Theorem}[chapter]
\newtheorem{Lem}{Lemma}[chapter]
\newtheorem{Cor}{Corollary}[chapter]
\newtheorem{Def}{Definition}[chapter]

\newtheorem{T-Prob}{Transformed Problem}[chapter]

\newtheorem{remark}{Remark}[chapter]

\makeatletter
\usepackage{etoolbox}
\pretocmd{\tableofcontents}{%
  \if@openright\cleardoublepage\else\clearpage\fi
  \pdfbookmark[0]{\contentsname}{toc}%
}{}{}%
\makeatother

\makenomenclature

\usepackage{geometry}
\begin{document}
% Title page
\begin{titlepage}
\begin{center}
\vspace*{1cm}
\Huge \hspace{-15mm}\textbf{Performance Analysis and Design of Non-orthogonal Multiple Access for Wireless Communications}\\
\vspace{1.5cm}
\hspace{-15mm}\normalsize\textbf{Zhiqiang Wei}
\\
\vspace{2cm}
\normalsize
{\hspace{-15mm}A thesis submitted to the Graduate Research School of\\
\hspace{-15mm}The University of New South Wales\\
\hspace{-15mm}in partial fulfillment of the requirements for the degree of\\
\text{ \ }\\
\hspace{-15mm}\textbf{Doctor of Philosophy}}\\
\vspace{2cm}
\hspace{-15mm}\includegraphics[width=0.4\textwidth]{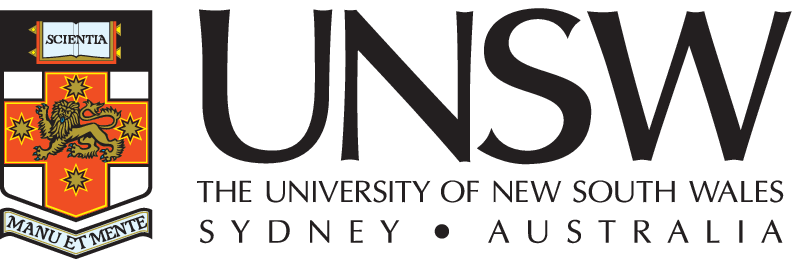}\\
\vspace{2cm}
\textbf{\hspace{-15mm}School of Electrical Engineering and Telecommunications\\
\hspace{-15mm}Faculty of Engineering\\
\hspace{-15mm}The University of New South Wales}\\
\vspace{2cm}
\hspace{-15mm}{June 2019}
\end{center}
\end{titlepage}

% Front pages acknowledgements/declaration/publications/abstract
\frontmatter
\onehalfspacing
\pagestyle{empty}

\section*{Copyright Statement}

I hereby grant The University of New South Wales or its agents the right
to archive and to make available my thesis or dissertation in whole or part
in the University libraries in all forms of media, now or hereafter known,
subject to the provisions of the Copyright Act 1968. I retain all proprietary
rights, such as patent rights. I also retain the right to use in future works
(such as articles or books) all or part of this thesis or dissertation.

I also authorise University Microfilms to use the abstract of my thesis
in Dissertation Abstract International (this is applicable to doctoral
thesis only).

I have either used no substantial portions of copyright material in my thesis
or I have obtained permission to use copyright material; where permission
has not been granted I have applied/will apply for a partial restriction
of the digital copy of my thesis or dissertation.

\vspace*{1.5cm}
\noindent
\parbox{3in}{Signed \dotfill}

\vspace*{0.5cm}
\noindent
\parbox{3in}{Date  \dotfill}

\vspace{1cm}

\section*{Authenticity Statement}

I certify that the Library deposit digital copy is a direct equivalent of
the final officially approved version of my thesis. No emendation of content
has occurred and if there are any minor variations in formatting, they are
the result of the conversion to digital format.

\vspace*{1.5cm}
\noindent
\parbox{3in}{Signed \dotfill}

\vspace*{0.5cm}
\noindent
\parbox{3in}{Date \dotfill}

\clearpage{\thispagestyle{empty}\cleardoublepage}
\section*{Originality Statement}

I hereby declare that this submission is my own work and to the best
of my knowledge it contains no material previously published or written
by another person, or substantial portions of material which have been
accepted for the award of any other degree or diploma at UNSW or any
other educational institute, except where due acknowledgment is made
in the thesis.  Any contribution made to the research by others, with
whom I have worked at UNSW or elsewhere, is explicitly acknowledged in
the thesis. I also declare that the intellectual content of this thesis
is the product of my own work, except to the extent that assistance from
others in the project's design and conception or in style, presentation
and linguistic expression is acknowledged.

\vspace*{1.5cm}
\noindent
\parbox{3in}{Signed \dotfill}

\vspace*{0.5cm}
\noindent
\parbox{3in}{Date  \dotfill}

\clearpage{\thispagestyle{empty}\cleardoublepage}

\pagenumbering{roman}
\pagestyle{fancy}
\fancyhf{}
\setlength{\headheight}{15pt}
\fancyhead[LE,RO]{\footnotesize \textbf \thepage}

\chapter*{}

\vspace{+5cm}
%\large\hspace{+1.5cm}
\begin{center}
    \large \emph{Dedicated to my parents, wife, and daughter.}
\end{center}
%}

\doublespacing
\chapter*{Abstract}
\addcontentsline{toc}{chapter}{\protect\numberline{}{Abstract}}
In this thesis, we study performance analysis and resource allocation design for non-orthogonal multiple access (NOMA) in wireless communication systems.
In contrast to conventional orthogonal multiple access (OMA) schemes, NOMA allows multiple users to share the same degree of freedom via superposition coding and successive interference cancelation (SIC) decoding.
Inspired by the solid foundations from information theory, NOMA has rekindled the interests of researchers as a benefit of the recent advancement in signal processing and silicon technologies.
However, comprehensive performance analysis on NOMA and practical resource allocation designs to exploit potential gains of NOMA in terms of spectral and energy efficiency have not been fully studied and investigated in the literature.
This thesis attempts to address these problems by providing a unified performance analysis and a systematic resource allocation design for NOMA in wireless communication systems.
The research work of this thesis can be divided into five parts.

In the first part of our research, we focus on investigating the ergodic sum-rate gain (ESG) in uplink communications attained by NOMA over OMA in single-antenna, multi-antenna, and massive MIMO systems with both single-cell and multi-cell deployments.
Employing the asymptotic analysis, we quantify the ESG of NOMA over OMA when adopting practical signal reception schemes, e.g. maximal ratio combining SIC (MRC-SIC) and minimum mean square error SIC (MMSE-SIC), at the base station.
The distinctive behaviors of ESG under different scenarios are unveiled which provide some important and interesting insights for the practical implementation of NOMA in future wireless networks.
In addition, extensive simulations are performed to confirm the accuracy of our performance analyses and validate the insights obtained in this work.

Our second work proposes a joint pilot and payload power allocation (JPA) scheme for uplink NOMA communications to reduce the effect of error propagation during SIC decoding due to channel estimation error, via exploiting the trade-off between pilot and payload power allocation.
In particular, reducing payload power would introduce less inter-user interference (IUI) for other users in NOMA systems.
Meanwhile, this also saves the power for transmitting pilots which yields a better channel estimation.
To this end, we propose an optimal joint power allocation design for pilots and payload to exploit this underlying trade-off to enhance NOMA's robustness against channel uncertainty.
Simulation results demonstrate that our proposed JPA scheme can effectively alleviate the error propagation and enhance the data detection performance, especially in the moderate energy budget regime.

In the third part of our research, we focus on the power-efficient resource allocation design for downlink multi-carrier NOMA (MC-NOMA) systems.
We consider a joint design in power allocation, rate allocation, user scheduling, as well as the SIC decoding policy.
Taking into account the imperfect channel state information and heterogeneous traffic demands, the resource allocation design is formulated as a non-convex optimization problem.
A globally optimal design via the branch-and-bound approach is proposed.
Furthermore, a low-complexity suboptimal iterative resource allocation algorithm is developed which can converge to a close-to-optimal solution rapidly.
Simulation results show that our proposed resource allocation schemes provide significant transmit power savings and enhanced robustness against channel uncertainty via exploiting the heterogeneity in channel conditions and traffic requirements of users in MC-NOMA systems.

Our fourth work extends NOMA from microwave communications to millimeter wave (mmWave) communications.
In particular, we propose a multi-beam NOMA framework for hybrid mmWave systems and investigate its resource allocation design.
Our proposed scheme is more practical than the existing conventional single-beam mmWave-NOMA schemes, as our scheme can flexibly pair NOMA users with an arbitrary angle-of-departure distribution.
A two-stage resource allocation design is proposed to maximize the system sum-rate utilizing the coalition formation game theory and the difference of convex (D.C.) programming technique, respectively.
Simulation results demonstrate that the proposed resource allocation can achieve a close-to-optimal performance in each stage.
Our proposed multi-beam mmWave NOMA scheme can offer a substantial system sum-rate improvement compared to the conventional mmWave-OMA scheme and the single-beam mmWave-NOMA scheme.

Finally, in the last part of our research, we further investigate the application of NOMA to downlink mmWave communications.
To fully exploit the potential performance gain brought by the combination of NOMA and mmWave communication, we propose a novel beamwidth control-based mmWave NOMA scheme and study its energy-efficient resource allocation design.
Enjoying the widened analog beamwidth, our proposed scheme can increase the number of potential NOMA groups compared to conventional scheme, which improves the energy efficiency of mmWave NOMA systems.
Specifically, a user grouping algorithm based on coalition formation game theory is developed and a low-complexity iterative digital precoder design algorithm is proposed to obtain an efficient suboptimal solution.
Simulation results show that the proposed scheme with beamwidth control offers a substantial energy efficiency gain over the conventional OMA and NOMA schemes without beamwidth control.

\doublespacing

\singlespacing
%\fancyhead[CE,CO]{\MakeUppercase{List of Publications}}

\chapter*{List of Publications} \label{listofpublications}
\addcontentsline{toc}{chapter}{\protect\numberline{}{List of Publications}}
%\nomenclature{$a$}{The number of angels per unit area}%
%\nomenclature{$N$}{The number of angels per needle point}%
%\nomenclature{$A$}{The area of the needle point}%
%

\ifpdf
    \graphicspath{{1_introduction/figures/PNG/}{1_introduction/figures/PDF/}{1_introduction/figures/}}
\else
    \graphicspath{{1_introduction/figures/EPS/}{1_introduction/figures/}}
\fi

The work in this thesis has been published or has been submitted for publication as journal papers. These papers are:

\vspace{0.1in}\noindent{\textbf{Journal Articles:}} \vspace{0.1in}
\begin{enumerate}[J1]
\item \textbf{Z. Wei}, J. Yuan, D. W. K. Ng, and H.-M. Wang
    ``Optimal Resource Allocation for Power-Efficient MC-NOMA with Imperfect Channel State Information,'' {\em IEEE Trans. Commun.}, vol.~65, no.~9, pp. 3944--3961, Sep. 2017.

\item \textbf{Z. Wei}, D. W. K. Ng, and J. Yuan,
    ``Joint Pilot and Payload Power Control for Uplink MIMO-NOMA with MRC-SIC Receivers,'' {\em IEEE Commun. Lett.}, vol.~22, no.~4, pp. 692--695, Apr. 2018.

\item \textbf{Z. Wei}, L. Zhao, J. Guo, D. W. K. Ng, and J. Yuan, ``Multi-Beam NOMA for Hybrid mmWave Systems'', {\em IEEE Trans. Commun.}, vol.~67, no.~2, pp. 1705--1719, Feb. 2019.

\item \textbf{Z. Wei}, D. W. K. Ng, and J. Yuan, ``NOMA for Hybrid MmWave Communication Systems with Beamwidth Control'', {\em IEEE J. Select. Topics Signal Process.}, vol. 13, no. 3, pp. 567-583, Jun. 2019.

\item \textbf{Z. Wei}, L. Yang, D. W. K. Ng, J. Yuan, and L. Hanzo ``On the Performance Gain of NOMA over OMA in Uplink Communication Systems'', {\em IEEE Trans. Commun.}, major revision, 30th Apr. 2019.
\end{enumerate}

The following publications are also the results from my Ph.D. study but not included in this thesis:

\vspace{0.1in} \noindent{\textbf{Journal Articles:}} \vspace{0.1in}
\begin{enumerate}[J1]
\setcounter{enumi}{5}
\item \textbf{Z. Wei}, J. Yuan, D. W. K. Ng, M. Elkashlan, and Z. Ding
    ``A Survey of Downlink Non-orthogonal Multiple Access for 5G Wireless
    Communication Networks,'' {\em ZTE Commun.}, vol.~14, no.~4, pp. 17--25, Oct. 2016.

\item Z. Sun, \textbf{Z. Wei}, L. Yang, J. Yuan, X. Cheng, and L. Wan, ``Exploiting Transmission Control for Joint User Identification and Channel Estimation in Massive Connectivity'', {\em IEEE Trans. Commun.}, accepted, Apr. 2019.

\item L. Zhao, \textbf{Z. Wei}, D. W. K. Ng, J. Yuan, and M. Reed, ``Multi-cell Hybrid Millimeter Wave Systems: Pilot Contamination and Interference Mitigation,'' {\em IEEE Trans. Commun.}, vol.~66, no.~11, pp. 5740--5755, Nov. 2018.
\end{enumerate}

\vspace{0.1in} \noindent{\textbf{Conference Articles:}} \vspace{0.1in}
\begin{enumerate}[C1]
\item \textbf{Z. Wei}, D. W. K. Ng, and J. Yuan ``Power-Efficient Resource Allocation for MC-NOMA with Statistical Channel State Information,'' in \emph{Proc. IEEE Global Commun. Conf.}, pp. 1--7, Dec. 2016.

\item \textbf{Z. Wei}, J. Guo, D. W. K. Ng, and J. Yuan, ``Fairness Comparison of Uplink NOMA and OMA,'' in \emph{Proc. IEEE Veh. Techn. Conf.}, pp. 1--6, 2017.

\item \textbf{Z. Wei}, L. Dai, D. W. K. Ng, and J. Yuan, ``Performance Analysis of a Hybrid Downlink-Uplink Cooperative NOMA Scheme,'' in \emph{Proc. IEEE Veh. Techn. Conf.}, pp. 1--7, 2017.

\item \textbf{Z. Wei}, L. Zhao, J. Guo, D. W. K. Ng, and J. Yuan, ``A Multi-Beam NOMA Framework for Hybrid mmWave Systems,'' in \emph{Proc. IEEE Intern. Commun. Conf.}, pp. 1--7, 2018.

\item \textbf{Z. Wei}, L. Zhao, J. Guo, D. W. K. Ng, and J. Yuan, ``On the Performance Gain of NOMA over OMA in Uplink Single-cell Systems,'' in \emph{Proc. IEEE Global Commun. Conf.}, pp. 1--7, 2018.

\item \textbf{Z. Wei}, D. W. K. Ng, and J. Yuan, ``Beamwidth Control for NOMA in Hybrid mmWave Communication Systems,'' in \emph{Proc. IEEE Intern. Commun. Conf.}, pp. 1--6, 2019.

\item \textbf{Z. Wei}, M. Qiu, D. W. K. Ng, and J. Yuan, ``A Two-Stage Beam Alignment Framework for Hybrid MmWave Distributed Antenna Systems,'' in \emph{Proc. IEEE Intern. Workshop on Signal Process. Advances in Wireless Commun.}, pp. 1--6, 2019.

\item L. Zhao, \textbf{Z. Wei}, D. W. K. Ng, J. Yuan, and M. Reed, ``Mitigating Pilot Contamination in Multi-cell Hybrid Millimeter Wave Systems,'' in \emph{Proc. IEEE Intern. Commun. Conf.}, pp. 1-7, 2018.

\item Z. Sun, \textbf{Z. Wei}, L. Yang, J. Yuan, X. Cheng, and L. Wan ``Joint User Identification and Channel Estimation in Massive Connectivity with Transmission Control,'' in \emph{Proc. IEEE Intern. Symp. on Turbo Codes $\&$ Iterative Information Process.}, pp. 1-6, 2018.

\item L. Zhao, J. Guo, \textbf{Z. Wei}, D. W. K. Ng, and J. Yuan, ``A Distributed Multi-RF Chain Hybrid mmWave Scheme for Small-cell Systems,'' in \emph{Proc. IEEE Intern. Commun. Conf.}, pp. 1-7, 2019.

\item T. Zheng, Q. Yang, K. Huang, H. Wang, \textbf{Z. Wei}, and J. Yuan, ``Physical-Layer Secure Transmissions in Cache-Enabled Cooperative Small Cell Networks,'' in \emph{Proc. IEEE Global Commun. Conf.}, pp. 1-6, 2018.

\item R. Li, \textbf{Z. Wei}, L. Yang, D. W. K. Ng, N. Yang, J. Yuan, and J. An, ``Joint Trajectory and Resource Allocation Design for UAV Communication Systems,'' in \emph{Proc. IEEE Global Commun. Conf.}, pp. 1-6, 2018.

\item Y. Cai, \textbf{Z. Wei}, R. Li, D. W. K. Ng, and J. Yuan, ``Energy-Efficient Resource Allocation for Secure UAV Communication Systems,'' in \emph{Proc. IEEE Wireless Commun. and Networking Conf.}, pp. 1-8, 2019.
\end{enumerate}
\vspace{0.1in}

\chapter*{Abbreviations} \label{abbreviations}
\addcontentsline{toc}{chapter}{\protect\numberline{}{Abbreviations}}
\markboth{ABBREVIATIONS}{}

\begin{longtable}[t]{ll}
\textbf{1G} \quad\quad&\mbox{The first generation} \vspace{0.1in}\\
\textbf{2G} \quad\quad&\mbox{The second generation} \vspace{0.1in}\\
\textbf{3G} \quad\quad&\mbox{The third generation} \vspace{0.1in}\\
\textbf{3GPP} \quad\quad&\mbox{3rd Generation Partnership Project} \vspace{0.1in}\\
\textbf{4G} \quad\quad&\mbox{The fourth generation} \vspace{0.1in}\\
\textbf{5G} \quad\quad&\mbox{The fifth generation} \vspace{0.1in}\\
\textbf{AOA} \quad\quad&\mbox{Angle-of-arrival} \vspace{0.1in}\\
\textbf{AOD} \quad\quad&\mbox{Angle-of-departure} \vspace{0.1in}\\
\textbf{AR} \quad\quad&\mbox{Augmented reality} \vspace{0.1in}\\
\textbf{ASINR} \quad\quad&\mbox{Average signal-to-interference-plus-noise ratio} \vspace{0.1in}\\
\textbf{AWGN} \quad\quad&\mbox{Additive white Gaussian noise} \vspace{0.1in}\\
\textbf{BER} \quad\quad&\mbox{Bit error rate} \vspace{0.1in}\\
\textbf{BS} \quad\quad&\mbox{Base station} \vspace{0.1in}\\
\textbf{CDF} \quad\quad&\mbox{Cumulative distribution function} \vspace{0.1in}\\
\textbf{CDMA} \quad\quad&\mbox{Code division multiple access} \vspace{0.1in}\\
\textbf{CSI} \quad\quad&\mbox{Channel state information} \vspace{0.1in}\\
\textbf{dB} \quad\quad&\mbox{Decibel} \vspace{0.1in}\\
\textbf{D2D} \quad\quad&\mbox{Device-to-device} \vspace{0.1in}\\
\textbf{DoF} \quad\quad&\mbox{Degree of freedom } \vspace{0.1in}\\
\textbf{DPC} \quad\quad&\mbox{Dirty paper coding} \vspace{0.1in}\\
\textbf{EE} \quad\quad&\mbox{Energy efficiency} \vspace{0.1in}\\
\textbf{eMBB} \quad\quad&\mbox{Enhanced mobile broadband} \vspace{0.1in}\\
\textbf{ESG} \quad\quad&\mbox{Ergodic sum-rate gain} \vspace{0.1in}\\
\textbf{FDD} \quad\quad&\mbox{Frequency division duplex} \vspace{0.1in}\\
\textbf{FD} \quad\quad&\mbox{Full-duplex} \vspace{0.1in}\\
\textbf{FDMA} \quad\quad&\mbox{Frequency division multiple access} \vspace{0.1in}\\
\textbf{FFT} \quad\quad&\mbox{Fast Fourier transform} \vspace{0.1in}\\
%\textbf{FTPA} \quad\quad&\mbox{Fractional transmission power allocation} \vspace{0.1in}\\
\textbf{H-CRAN} \quad\quad&\mbox{Heterogeneous cloud radio access networks} \vspace{0.1in}\\
\textbf{ICT} \quad\quad&\mbox{Information and communication technology} \vspace{0.1in}\\
\textbf{IDMA} \quad\quad&\mbox{Interleave division multiple access} \vspace{0.1in}\\
\textbf{ICI} \quad\quad&\mbox{Inter-cell interference} \vspace{0.1in}\\
\textbf{i.i.d.} \quad\quad&\mbox{Independent and identically distributed} \vspace{0.1in}\\
\textbf{I/Q}  \quad\quad&\mbox{In-phase/Quadrature} \vspace{0.1in}\\
\textbf{IoT}  \quad\quad&\mbox{Internet-of-things} \vspace{0.1in}\\
\textbf{ISI} \quad\quad&\mbox{Inter-symbol interference} \vspace{0.1in}\\
\textbf{IUI} \quad\quad&\mbox{Inter-user interference} \vspace{0.1in}\\
\textbf{JPA} \quad\quad&\mbox{Joint pilot and payload power allocation } \vspace{0.1in}\\
\textbf{LDS} \quad\quad&\mbox{Low-density spreading} \vspace{0.1in}\\
\textbf{LOS} \quad\quad&\mbox{Line-of-sight} \vspace{0.1in}\\
\textbf{LTE} \quad\quad&\mbox{Long-term evolution} \vspace{0.1in}\\
\textbf{MAI} \quad\quad&\mbox{Multiple access interference} \vspace{0.1in}\\
\textbf{MC}  \quad\quad&\mbox{Multi-carrier} \vspace{0.1in}\\
\textbf{MIMO} \quad\quad&\mbox{Multiple-input multiple-output} \vspace{0.1in}\\
\textbf{MISO} \quad\quad&\mbox{Multiple-input single-output} \vspace{0.1in}\\
\textbf{ML}  \quad\quad&\mbox{Maximum-likelihood} \vspace{0.1in}\\
\textbf{MMSE}  \quad\quad&\mbox{Minimum mean square error} \vspace{0.1in}\\
\textbf{mMTC} \quad\quad&\mbox{Massive machine type communications} \vspace{0.1in}\\
\textbf{mmWave} \quad\quad&\mbox{Millimeter wave} \vspace{0.1in}\\
\textbf{MRC} \quad\quad&\mbox{Maximal-ratio combining} \vspace{0.1in}\\
\textbf{MRT} \quad\quad&\mbox{Maximal-ratio transmission} \vspace{0.1in}\\
\textbf{MSE}  \quad\quad&\mbox{Mean square error} \vspace{0.1in}\\
\textbf{MUD}  \quad\quad&\mbox{Multi-user detection} \vspace{0.1in}\\
\textbf{NLOS}  \quad\quad&\mbox{Non-line-of-sight} \vspace{0.1in}\\
\textbf{NMSE}  \quad\quad&\mbox{Normalized mean squared error} \vspace{0.1in}\\
\textbf{NOMA} \quad\quad&\mbox{Non-orthogonal multiple access} \vspace{0.1in}\\
\textbf{OFDM} \quad\quad&\mbox{Orthogonal frequency division multiplexing} \vspace{0.1in}\\
\textbf{OFDMA} \quad\quad&\mbox{Orthogonal frequency division multiple access} \vspace{0.1in}\\
\textbf{OMA} \quad\quad&\mbox{Orthogonal multiple access} \vspace{0.1in}\\
\textbf{OVSF} \quad\quad&\mbox{Orthogonal variable spreading factor} \vspace{0.1in}\\
\textbf{PDF} \quad\quad&\mbox{Probability density function} \vspace{0.1in}\\
\textbf{PF} \quad\quad&\mbox{Proportional fairness} \vspace{0.1in}\\
\textbf{PDMA} \quad\quad&\mbox{Pattern division multiple
access} \vspace{0.1in}\\
\textbf{PS} \quad\quad&\mbox{Phase shifter} \vspace{0.1in}\\
\textbf{QAM} \quad\quad&\mbox{Quadrature amplitude modulation} \vspace{0.1in}\\
\textbf{QoS} \quad\quad&\mbox{Quality-of-service} \vspace{0.1in}\\
\textbf{RF} \quad\quad&\mbox{Radio frequency} \vspace{0.1in}\\
\textbf{SC} \quad\quad&\mbox{Superposition coding} \vspace{0.1in}\\
\textbf{SCMA} \quad\quad&\mbox{Sparse code multiple access} \vspace{0.1in}\\
\textbf{SDP} \quad\quad&\mbox{Semi-definite programming} \vspace{0.1in}\\
\textbf{SDMA} \quad\quad&\mbox{Spatial division multiple access} \vspace{0.1in}\\
\textbf{SE} \quad\quad&\mbox{Spectral efficiency} \vspace{0.1in}\\
\textbf{SER} \quad\quad&\mbox{Symbol error rate} \vspace{0.1in}\\
\textbf{SIC} \quad\quad&\mbox{Successive interference cancelation}\vspace{0.1in}\\
\textbf{SNR}  \quad\quad&\mbox{Signal-to-noise ratio} \vspace{0.1in}\\
\textbf{SINR}  \quad\quad&\mbox{Signal-to-interference-plus-noise ratio} \vspace{0.1in}\\
\textbf{SISO}  \quad\quad&\mbox{Single-input single-output} \vspace{0.1in}\\
\textbf{SVD}  \quad\quad&\mbox{Singular value decomposition} \vspace{0.1in}\\
\textbf{TDD} \quad\quad&\mbox{Time division duplex} \vspace{0.1in}\\
\textbf{TDMA} \quad\quad&\mbox{Time division multiple access} \vspace{0.1in}\\
%\textbf{TTPA} \quad\quad&\mbox{Tree-search based transmission power allocation } \vspace{0.1in}\\
\textbf{s.t.}  \quad\quad&\mbox{Subject to} \vspace{0.1in}\\
\textbf{UMi} \quad\quad&\mbox{Urban micro}\vspace{0.1in}\\
\textbf{URLLC} \quad\quad&\mbox{Ultra reliable and low latency communications}\vspace{0.1in}\\
\textbf{V2V} \quad\quad&\mbox{Vehicle-to-vehicle}\vspace{0.1in}\\
\textbf{V2X} \quad\quad&\mbox{Vehicle-to-everything}\vspace{0.1in}\\
\textbf{V-BLAST} \quad\quad&\mbox{Vertical-bell laboratories layered space-time}\vspace{0.1in}\\
\textbf{VR} \quad\quad&\mbox{Virtual reality}\vspace{0.1in}\\
\textbf{w.r.t.} \quad\quad&\mbox{With respect to}\vspace{0.1in}\\
\textbf{ZF} \quad\quad&\mbox{Zero-forcing}\vspace{0.1in}\\
\end{longtable}

\chapter*{List of Notations} \label{listofnotations}
\addcontentsline{toc}{chapter}{\protect\numberline{}{List of Notations}}
\markboth{LIST OF NOTATIONS}{}
Scalars, vectors and matrices are written in italic, boldface lower-case and upper-case letters, respectively, e.g. $x$, $\mathbf{x}$, and $\mathbf{X}$.
\begin{longtable}[t]{ll}
$\mathbf{X}^\mathrm{T}$ \quad\quad&\mbox{Transpose of $\mathbf{X}$} \vspace{0.1in}\\
%$\mathbf{X}^{\ast}$ \quad\quad&\mbox{Complex conjugate of $\mathbf{X}$} \vspace{0.1in}\\
$\mathbf{X}^\mathrm{H}$ \quad\quad&\mbox{Hermitian transpose of $\mathbf{X}$ } \vspace{0.1in}\\
$\mathbf{X}^{-1}$ \quad\quad&\mbox{Inverse of $\mathbf{X}$ } \vspace{0.1in}\\
$\mathbf{X}_{i,j}$ \quad\quad&\mbox{The element in the row $i$ and the column $j$ of $\mathbf{X}$ } \vspace{0.1in}\\
$\det\left(\mathbf{X}\right)$ \quad\quad&\mbox{Determinant of a square matrix $\mathbf{X}$} \vspace{0.1in}\\
$\mbox{Tr}\left(\mathbf{X}\right)$ \quad\quad&\mbox{Trace of a square matrix $\mathbf{X}$} \vspace{0.1in}\\
$\mbox{Rank}\left(\mathbf{X}\right)$ \quad\quad&\mbox{Rank of a matrix $\mathbf{X}$} \vspace{0.1in}\\
$|x|$ \quad\quad&\mbox{Absolute value (modulus) of the complex scalar $x$} \vspace{0.1in}\\
$\|\mathbf{x}\|$ \quad\quad&\mbox{The Euclidean norm of a vector $\mathbf{x}$} \vspace{0.1in}\\
$\|\mathbf{X}\|_{\mathrm{F}}$ \quad\quad&\mbox{The Frobenius norm of a matrix $\mathbf{X}$} \vspace{0.1in}\\
$[x]^{+}$ \quad\quad&\mbox{$\max\left(0,x\right)$.} \vspace{0.1in}\\
$\mathbb{R}^{M\times N}$ \quad\quad&\mbox{The set of all $M\times N$ matrices with real entries} \vspace{0.1in}\\
$\mathbb{C}^{M\times N}$ \quad\quad&\mbox{The set of all $M\times N$ matrices with complex entries} \vspace{0.1in}\\
$\mathbb{B}^{M\times N}$ \quad\quad&\mbox{The set of all $M\times N$ matrices with binary entries} \vspace{0.1in}\\
%$\mathbb{R}^n$ \quad\quad&\mbox{The Euclidean space} \vspace{0.1in}\\
$\mathrm{Pr}\{E\}$ \quad\quad&\mbox{The probability of event $E$ occurs} \vspace{0.1in}\\
%$p_{\msf{X}}(x),p(x)$ \quad\quad&\mbox{Probability density function of the random variable $\msf{X}$} \vspace{0.1in}\\
%$p_{\msf{X}|\msf{Y}}(x|y),p(x|y)$ \quad\quad&\mbox{Conditional distribution of $\msf{X}$ given $\msf{Y}$} \vspace{0.1in}\\
%$p_{\msf{X},\msf{Y}}(x,y),p(x,y)$ \quad\quad&\mbox{Joint distribution of $\msf{X}$ and $\msf{Y}$} \vspace{0.1in}\\
$x \sim p(x)$ \quad\quad&\mbox{$x$ is distributed according to $p(x)$} \vspace{0.1in}\\
$\left\lfloor \cdot \right\rfloor$ \quad\quad&\mbox{The maximum integer smaller than the input value} \vspace{0.1in}\\
$\Re(z)$ \quad\quad&\mbox{The real part of a complex number $z$} \vspace{0.1in}\\
$\Im(z)$ \quad\quad&\mbox{The imaginary part of a complex number $z$} \vspace{0.1in}\\
%$\mathbb{C}^{M\times N}$ \quad\quad&\mbox{The space of all $M\times N$ matrices with complex
%entries} \vspace{0.1in}\\
$\emptyset$ \quad\quad&\mbox{An empty set, a null set} \vspace{0.1in}\\
$|\mathcal{S}|$ \quad\quad&\mbox{The cardinality of a set $\mathcal{S}$} \vspace{0.1in}\\
%$\doteq$ quad\quad&\mbox{Asymptotically equivalent.} \vspace{0.1in}\\
$\mathbf{0}$ \quad\quad&\mbox{A vector or a matrix with all-zero entries} \vspace{0.1in}\\
$\mathbf{I}_{\mathrm{N}}$ \quad\quad&\mbox{$N$ dimension identity matrix} \vspace{0.1in}\\
%$j$ \quad\quad&\mbox{Imaginary unit $j = \sqrt{-1}$} \vspace{0.1in}\\
${\mathrm{E}}\left\{ \cdot \right\}$ \quad\quad&\mbox{Statistical expectation} \vspace{0.1in}\\
$\mathcal{N}(\mu,\sigma^2)$ \quad\quad&\mbox{Real Gaussian random variable with mean $\mu$ and variance $\sigma^2$} \vspace{0.1in}\\
$\mathcal{CN}(\mu,\sigma^2)$ \quad\quad&\mbox{Circularly symmetric complex Gaussian random variable:}\vspace{0.1in} \\
\quad\quad&\mbox{the real and imaginary parts are i.i.d. $\mathcal{N}(\mu/2,\sigma^2/2)$} \vspace{0.1in}\\
$\mathcal{CN}(\mathbf{0},\mathbf{K})$ \quad\quad&\mbox{Circularly symmetric Gaussian random vector with}\vspace{0.1in} \\
\quad\quad&\mbox{a zero-mean vector and covariance matrix $\mathbf{K}$} \vspace{0.1in}\\
$\ln(\cdot)$ \quad\quad&\mbox{Natural logarithm} \vspace{0.1in}\\
$\log_a(\cdot)$ \quad\quad&\mbox{Logarithm in base $a$} \vspace{0.1in}\\
%$\log(\cdot)$ \quad\quad&\mbox{Logarithm in base ten} \vspace{0.1in}\\
%$\lambda _{i}(\cdot )$ \quad\quad&\mbox{The $i$-th maximum eigenvalue of a matrix} \vspace{0.1in}\\
$\mathrm{diag}\left\{\bm{a}\right\}$ \quad\quad&\mbox{A diagonal matrix with the entries of $\bm{a}$ on its diagonal} \vspace{0.1in}\\
$\lim$ \quad\quad&\mbox{Limit} \vspace{0.1in}\\
%$\mathrm{sinc}(x)$ \quad\quad&\mbox{Sinc function with input $x$} \vspace{0.1in}\\
$\max\left\{\cdot\right\}$ \quad\quad&\mbox{Return the maximum element of its input} \vspace{0.1in}\\
$\min\left\{\cdot\right\}$ \quad\quad&\mbox{Return the minimum element of its input} \vspace{0.1in}\\
$e^x,\exp(x)$ \quad\quad&\mbox{Natural exponential function} \vspace{0.1in}\\
%$\tanh$ \quad\quad&\mbox{Hyperbolic tangent function} \vspace{0.1in}\\
%$d_H$ \quad\quad&\mbox{Hamming distance} \vspace{0.1in}\\
%$w_H$ \quad\quad&\mbox{Hamming weight} \vspace{0.1in}\\
%$d_{E,\min}$ \quad\quad&\mbox{Minimum Euclidean distance} \vspace{0.1in}\\
%$d_{P,\min}$ \quad\quad&\mbox{Minimum Product distance} \vspace{0.1in}\\
%$\oplus$ \quad\quad&\mbox{Modulo lattice addition} \vspace{0.1in}\\
%$\ominus$ \quad\quad&\mbox{Modulo lattice subtraction} \vspace{0.1in}\\
%$I(\msf{X};\msf{Y})$ \quad\quad&\mbox{The mutual information between $\msf{X}$ and $\msf{X}$} \vspace{0.1in}\\
%$h(\msf{X})$ \quad\quad&\mbox{The entropy of a continous random variable $\msf{X}$} \vspace{0.1in}\\
%$H(\msf{X})$ \quad\quad&\mbox{The entropy of a discrete random variable $\msf{X}$} \vspace{0.1in}\\
%$\mathcal{S}_1 \setminus \mathcal{S}_2$ \quad\quad&\mbox{Obtain the elements that only belong to set $\mathcal{S}_1$} \vspace{0.1in}\\
$\mathcal{O}(g(x))$ \quad\quad&\mbox{An asymptotic upper bound, i.e., $f(x) = \mathcal{O}(g(x))$ }\vspace{0.1in} \\
\quad\quad&\mbox{if $\lim_{x\to \infty} |\frac{f(x)}{g(x)}| \le N$ for $0 < N < \infty$} \vspace{0.1in}\\
%$f_{\msf{X}}(x)$ \quad\quad&\mbox{PDF of a random variable $\msf{X}$} \vspace{0.1in}\\
%$\mathrm{s.t.}$ \quad\quad&\mbox{Subject to} \vspace{0.1in}\\
\end{longtable}

\fancyhead[CE]{\footnotesize \leftmark}
\fancyhead[CO]{\footnotesize \rightmark}

% -- Header modification --
\renewcommand{\chaptermark}[1]{%
\markboth{\MakeUppercase{%
\thechapter.%
\ #1}}{}}
%--------------

% Table of Contents
\tableofcontents

% List of figures
\listoffigures
%\addcontentsline{toc}{chapter}{\protect\numberline{}{List of Figures}}

% List of tables
\listoftables
%\addcontentsline{toc}{chapter}{\protect\numberline{}{List of Tables}}

% List of algorithms
\listofalgorithms
\addcontentsline{toc}{chapter}{\protect\numberline{}{List of Algorithms}}

% List of Abbreviations
%\renewcommand{\nomname}{List of Abbreviations}
%\renewcommand{\nomitemsep}{2pt}
%\include{misc/abbr}

%\markboth{\MakeUppercase\nomname}{\MakeUppercase\nomnamue}
%\printnomenclature[2cm]
%\clearpage{\thispagestyle{empty}\cleardoublepage}

% chapters
\mainmatter
\doublespacing

\chapter{Introduction}\label{C1:chapter1}
In this chapter, we first provide an overview on the fifth-generation (5G) wireless networks before introducing the motivation of our research for studying non-orthogonal multiple access (NOMA) in this thesis.
Then, a literature review on NOMA and a brief introduction on the research challenges of NOMA are presented.
In addition, we also outline the main contributions of the thesis.

\section{Overview of 5G}
A networked world is on its way.
Wireless communications have become one of the revolutionary technologies in modern societies and have offered many business opportunities across industrial, public, and government sectors\cite{Andrews2014,wong2017key}.
In particular, the development of wireless communications worldwide fuels the massive growth in the number of wireless communication devices and sensors for emerging applications, such as smart logistics \& transportation, environmental monitoring, energy management, safety management, and industry automation, just to name a few.
It is expected that in the Internet-of-Things (IoT) era \cite{Zorzi2010}, there will be $50$ billion wireless communication devices connected worldwide with a connection density up to a million devices per $\mathrm{km}^{\mathrm{2}}$ \cite{QualComm,SunALOHA}.

A widely held view is that the 5G is not just an evolutionary version of the current fourth-generation (4G) communication systems \cite{Andrews2014}, due to not only the exponentially increasing demand of data traffic but also the energy-hungry emerging services and functionalities.
To be more specific, three main envisioned usage scenarios have been proposed for 5G as follows, which are expected to revolutionize our future daily life.
\begin{itemize}
  \item Enhanced Mobile Broadband (eMBB)\cite{PopovskiAccess}: high-resolution video streaming, virtual reality (VR), augmented reality (AR), etc.
  \item Massive Machine Type Communications (mMTC)\cite{PopovskiAccess,SunJUICE,SunALOHA,SunPNC,liu2019deep,SunCSTCOM}: IoT services, metering, monitoring, and measuring, smart agriculture, smart city, smart port, etc.
  \item Ultra-Reliable and Low Latency Communications (URLLC)\cite{PopovskiAccess}: vehicle-to-\newline vehicle (V2V) and vehicle-to-everything (V2X) communications, autonomous driving, remote control surgery, etc.
\end{itemize}
These new services impose unprecedentedly challenges for the development of 5G wireless communication systems, such as the requirement of ultra-high data rates ($100\sim1000 \times$ of current 4G technology), lower latency (1 ms for a roundtrip latency), massive connectivity ($10^6 \;\mathrm{devices/km}^2$), and the support of diverse quality of service (QoS)\cite{Andrews2014}.

In addition, energy-efficient communications have become an important focus in both academia and industry due to the growing demands of energy consumption and the arising environmental concerns around the world\cite{QingqingEE,WuEESE2018,DerrickEEOFDMA,DerrickEESWIPT,DerrickEERobust,DerrickLimitedBackhaul,NgDUALII,cai2019energy}.
It is predicted that billions of information and communication technology (ICT) devices could create up to 3.5\% of global emissions by 2020 and up to 14\% by 2040\cite{ClimateHN}.
Also, in 2025, it is expected that the communications industry will be responsible for 20\% of all the world’s electricity\cite{ClimateHN}.
The escalating energy costs and the associated global carbon dioxide ($\mathrm{CO}_2$) emission of information and communication technology (ICT) devices have stimulated the interest of researchers in an emerging area of energy-efficient radio management.
To this end, studying energy-efficient wireless communications is critical to strike a balance between system energy consumption and throughput\cite{QingqingEE,WuEESE2018,DerrickEEOFDMA,DerrickEESWIPT,DerrickEERobust,DerrickLimitedBackhaul,NgDUALII,cai2019energy}.

On the path to 5G, compelling new technologies, such as massive multiple-input multiple-output (MIMO)\cite{Marzetta2010,ngo2013energy}, energy harvesting communications\cite{BoshkovskaEE,AhmedEH2013,Zlatanov2017}, non-orthogonal multiple access (NOMA)\cite{Dai2015,Ding2015b,WeiSurvey2016,QiuLOMA,QiuLOMASlowFading}, millimeter wave (mmWave) communications\cite{Rappaport2013,XiaoMing2017,zhao2017multiuser,wei2018multibeam,WEIMultiBeamFramework,WeiBeamWidthControl,wei2019beamwidthConference}, and ultra densification and offloading, \cite{Andrews2012,Andrews2013,Ramasamy2013} etc., have been proposed to address the aforementioned issues and have been identified as the key technologies by researchers\cite{wong2017key}.
In particular, it is highly expected to employ a future radio access technology, which is flexible, reliable\cite{DerrickEERobust}, and efficient in terms of energy and spectrum\cite{DerrickEEOFDMA,DerrickEESWIPT} for 5G or beyond 5G networks.
In fact, multiple access technology is the most fundamental aspect in physical layer and it significantly affects the whole system performance in each generation of wireless networks.
As such, this thesis focuses on non-orthogonal multiple access (NOMA), which is a promising multiple access technology and is expected to address some of the aforementioned key challenges of 5G.

\section{Motivation}
The main motivations of studying NOMA in the thesis are listed as follows.
In fact, the motivations are to exploit the main advantages and to avoid the disadvantages of applying NOMA in future wireless networks.
{The advantages and disadvantages of NOMA have been intensively studied, e.g. \cite{Ding2015b,Dai2015}, but are not limited to the following issues.}
Here, we only present the most related points of interest of the thesis and the details on the features of NOMA can be found in Table \ref{NOMAvsOMA} in Section \ref{C1:lITERATURE}.

\noindent\underline{\textbf{Supporting Overload Transmission}}

The massive growth in the number of wireless communication devices and sensors creates a serious performance bottleneck in realizing reliable and ubiquitous wireless
communication networks\cite{SunALOHA,SunPNC,SunJUICE}.
It is highly desirable to have an innovative multiple access technique which is able to accommodate more users with limited system resource, i.e., time-frequency resource block, spatial subspace, and/or radio frequency (RF) chain.
{A set of system resources can be interpreted as a kind of system degrees of freedom (DoF), which allows the multiple data streams to be transmitted on each DoF without interfering each other.}
In other words, the system DoF is equal to the maximum number of the independently deliverable data streams.
{When the number of users is larger than the system DoF, an overload scenario occurs.}
In such scenarios, NOMA is an efficient multiple access scheme which paves the way to support the overload transmission via DoF sharing.
In particular, NOMA transmission technique allows multiple users to simultaneously share the same DoF and consequently is beneficial to increase the number of supported connections by introducing controllable interference.
Despite of the inherent DoF-sharing feature of NOMA, it highly relies on user scheduling and resource allocation design to support overload transmission in a spectral-efficient and/or energy-efficient way, which is one of the main focuses of the thesis.
Conceptually, NOMA is a generalization of orthogonal multiple access (OMA) where the latter is more conservative which does not allow DoF sharing among users.
Therefore, a thorough study of NOMA and OMA is highly expected to reveal some insights about the performance gain of NOMA over OMA, which can shed light on the practical implementation of NOMA in future wireless communication networks.

\noindent\underline{\textbf{Improving Flexibility in Resource Allocation}}

%Improving the spectral efficiency has been a centric and continuous research topic in the field of wireless communications for decades, due to the scarce and expensive spectrum resource.
%
In contrast to allocate a resource block exclusively to a single user in OMA scheme, NOMA can utilize the spectrum more efficiently by admitting strong users into the resource block occupied by weak users without compromising much their performance via utilizing appropriate power allocation and successive interference cancelation (SIC) techniques.
From the optimization point of view, by relaxing the orthogonal constraint of OMA, NOMA enables a more flexible management of radio resources and offers an efficient way to improve the spectral efficiency via appropriate resource allocation\cite{Timotheou2015}.
To this end, studying the radio management via exploiting limited system resources is critical to improve the spectral efficiency of NOMA, which is one of the main focuses of this thesis.

On the other hand, although some existing contributions\cite{Cover1991,Tse2005,WangPowerEfficiency,Ding2014} have demonstrated the NOMA's spectral improvement compared to OMA when the QoSs of users are taken into account, the source of the performance gain of NOMA over OMA has not been well understood.
More importantly, the impact of specific system parameters on the spectral efficiency of NOMA, such as the number of NOMA users, the signal-to-noise ratio (SNR), and the cell size, have not been revealed in the open literature.
To fill this gap, the thesis offers a unified analysis on the ergodic sum-rate gain of NOMA over OMA in single-antenna, multi-antenna, and massive MIMO systems with both single-cell and multi-cell deployments.

\noindent\underline{\textbf{Improving Energy Efficiency}}

The advantages of NOMA for supporting massive connectivity and improving the resource allocation flexibility do not come for free.
There are significant concerns of financial implications to the service providers due to the huge power consumption in wireless communication networks\cite{QingqingEE,WuEESE2018,DerrickEEOFDMA,DerrickEESWIPT,DerrickEERobust,DerrickLimitedBackhaul,NgDUALII,cai2019energy}.
In particular, in the absence of a cautious design on user scheduling and SIC decoding order, NOMA may consume more power than that of OMA due to the existence of inter-user interference (IUI) and to meet the QoS requirement of each user\cite{Wei2017}.
In addition, it is well-known that the spectral efficiency gain of NOMA over OMA is larger when channel conditions of the multiplexed users become more distinctive \cite{Dingtobepublished}.
However, a higher transmit power is needed to satisfying the QoS of the weak user if its channel condition is too poor.
Therefore, there is a non-trivial trade-off between the spectral efficiency and energy efficiency which should be taken into account for resource allocation algorithm design for NOMA.

On the other hand, applying NOMA to millimeter wave (mmWave) communications has a high potential to improve the system energy efficiency.
In particular, the tremendous energy consumption associated with RF chains constitutes a large part of the total system energy consumption in mmWave communication systems.
Therefore, clustering users within the same analog beam as a NOMA group and allowing them to share the energy-hungry RF chain can substantially improve the system energy efficiency.
However, how to clustering NOMA users effectively and how to perform resource allocation in an mmWave-NOMA system are still open problems and deserve our efforts to explore.

\noindent\underline{\textbf{Improving Robustness of NOMA Against Channel Uncertainty}}

The spectral efficiency and energy efficiency gain of NOMA arises from exploiting the heterogeneity of channel conditions of NOMA users\cite{Ding2014,Xu2017,Wei2018PerformanceGain}.
As a result, the performance of NOMA is sensitive to the channel uncertainty.
In practice, the user scheduling strategy and the SIC decoding order are determined by the CSI.
However, most of existing works \cite{Lei2016NOMA,Di2016sub,Sun2016Fullduplex} on resource allocation of NOMA are based on the assumption of perfect channel state information (CSI) at the transmitter side, which is difficult to obtain in practice due to either the estimation error or the feedback delay.
Therefore, it is necessary to investigate how CSI error affects the performance of NOMA and to consider robust resource allocation under imperfect CSI\cite{Wei2016NOMA}.

On the other hand, the implementation of SIC at the receiver of NOMA system suffers from the inherent problem of error propagation, which is mainly caused by inaccurate channel estimation and subsequently degraded detection quality \cite{Nara2005Error}.
Specifically, imperfect channel estimates deteriorate the system performance via not only encouraging the inter-users crosstalk during interference cancelation process but also affecting the optimal SIC decoding order, which undermining the benefits of NOMA.
Therefore, it is expected to design an innovative NOMA scheme to alleviate the error propagation during SIC detection in NOMA systems.
\section{Literature Review}\label{C1:lITERATURE}
In this section, the related topics of this thesis, e.g. orthogonal multiple access, non-orthogonal multiple access, and millimeter wave communications, are discussed and reviewed.

\subsection{From OMA to NOMA}
Radio access technologies for cellular communications are characterized by multiple access schemes, such as frequency-division multiple access (FDMA) for the first generation (1G), time-division multiple access (TDMA) for the second generation (2G), code-division multiple access (CDMA) used by both 2G and the third generation (3G), and orthogonal frequency-division multiple access (OFDMA) for 4G.
Equipping multiple antennas at transceivers, spatial division multiple access (SDMA) is a potential technology to be exploited for supporting multiple users.
All these conventional multiple access schemes are categorized as OMA technologies, where different users are allocated to orthogonal resources in either time, frequency, spatial, or code domain in order to mitigate multiple access interference (MAI).
However, OMA schemes are not sufficient to support the massive connectivity with diverse QoS requirements.
In fact, due to the limited system DoF, some users with better channel quality have a higher priority to be served while other users with poor channel quality have to wait to access, which leads to high unfairness and large latency.
Besides, it is inefficient in terms of system resources when allocating a DoF solely to users with poor channel quality.

\begin{figure}[t]
\centering
\includegraphics[width=5.5in]{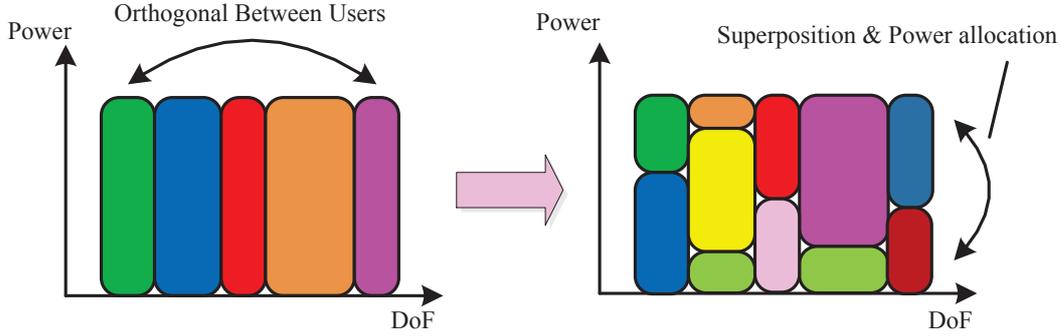}
\caption{From OMA to NOMA via power domain multiplexing.}%
\label{NOMAPrinciple}%
\end{figure}

In contrast to conventional OMA, NOMA transmission techniques intend to share a DoF among users via superposition and consequently need to employ multiple user detection (MUD)\cite{verdu1998multiuser,IDMAPingLI} to separate interfered users sharing the same DoF, as illustrated in Figure \ref{NOMAPrinciple}.
Therefore, NOMA is beneficial to supporting
a large number of connections to achieve spectrally-efficient communications by introducing controllable symbol collisions in the same DoF.
In addition, multiple users with different types of traffic requests can be multiplexed to transmit concurrently on the same DoF to improve the latency and fairness performance.
The comparison of OMA and NOMA is summarized in Table \ref{NOMAvsOMA}.
As a result, NOMA has been recognized as a promising multiple access technique for the 5G wireless networks due to its high spectral efficiency, massive connectivity, low latency, and high user fairness \cite{Dai2015}.
For example, the industrial community has proposed up to 16 various forms of NOMA as the potential multiple access schemes for the forthcoming fifth-generation (5G) networks\cite{ChenNOMAScheme}.

\begin{table}[t]
\center\small
\caption{Comparison of OMA and NOMA}
  \begin{tabular}{lll}
  \hline
    & Advantages & Disadvantages \\ \hline
  \multirow{3}{*}{OMA}
    & \tabitem Simpler receiver detection & \tabitem Lower spectral efficiency\\
    &  & \tabitem Limited number of users\\
    &  & \tabitem Unfairness for users\\\hline
  \multirow{5}{*}{NOMA}
    & \tabitem Higher spectral efficiency & \tabitem Increased complexity of receivers \\
    & \tabitem Higher connection density & \tabitem Higher sensitivity to channel uncertainty\\
    & \tabitem Enhanced user fairness & \\
    & \tabitem Lower latency & \\
    & \tabitem Supporting diverse QoS & \\\hline
  \end{tabular}
\label{NOMAvsOMA}
\end{table}

The concept of non-orthogonal transmissions dates back to the 1990s, e.g. \cite{Cover1991,Tse2005}, which serves as a foundation for the development of the power-domain NOMA.
Indeed, NOMA schemes relying on non-orthogonal spreading sequences have led to popular CDMA arrangements, even though eventually the so-called orthogonal variable spreading factor (OVSF) code was selected for the global 3G wireless systems\cite{Verdu1999,YangCDMA,Hanzo2003,WangPowerEfficiency}.
To elaborate a little further, the spectral efficiency of CDMA was analyzed in \cite{Verdu1999}.
In \cite{YangCDMA}, the authors compared the benefits and deficiencies of three typical CDMA schemes: single-carrier direct-sequence CDMA (SC DS-CDMA), multi-carrier CDMA (MC-CDMA), and multi-carrier DS-CDMA (MC DS-CDMA).
Furthermore, a comparative study of OMA and NOMA was carried out in \cite{WangPowerEfficiency}.
It has been shown that NOMA possesses a spectral-power efficiency advantage over OMA\cite{WangPowerEfficiency} and this theoretical gain can be realized with the aid of the interleave division multiple access (IDMA) technique proposed in \cite{Ping2006IDMA}.
Despite the initial efforts on the study of NOMA, the employment of NOMA in practical systems has been relatively slow due to the requirement of sophisticated hardware for its implementation.
Recently, NOMA has rekindled the interests of researchers as a benefit of the recent advances in signal processing and silicon technologies\cite{Access2015,NOMADOCOMO}.

According to the domain of multiplexing, we can divide the existing NOMA techniques into two categories\cite{WeiSurvey2016}, i.e., code domain multiplexing and power domain multiplexing.
The code domain NOMA techniques, including low-density spreading (LDS)\cite{Hoshyar2008,Hoshyar2010,Razavi2012,HuangTaoLDS}, sparse code multiple access (SCMA)\cite{Nikopour2013}, pattern division multiple
access (PDMA)\cite{Dai2014PDMA}, etc, introduce redundancy via coding/spreading to facilitate the users separation at the receiver. For instance, LDS-CDMA\cite{Hoshyar2008} intentionally arranges each user to spread its data over a small number of chips and then interleave uniquely, which makes optimal MUD affordable at receiver and exploits the intrinsic interference diversity.
LDS-OFDM\cite{Hoshyar2010,Razavi2012}, as shown in Figure \ref{LDSOFDMBLOCK}, can be interpreted as a system which applies LDS for multiple access and OFDM for multi-carrier modulation.
Besides, SCMA is a generalization of LDS methods where the modulator and LDS spreader are merged.
The principle of power domain NOMA is to exploit the users' power difference for multiuser multiplexing together with superposition coding at the transmitter, while applying SIC at the receivers for alleviating the IUI\cite{WeiSurvey2016}.

\begin{figure}[t]
\centering
\includegraphics[width=5.5in]{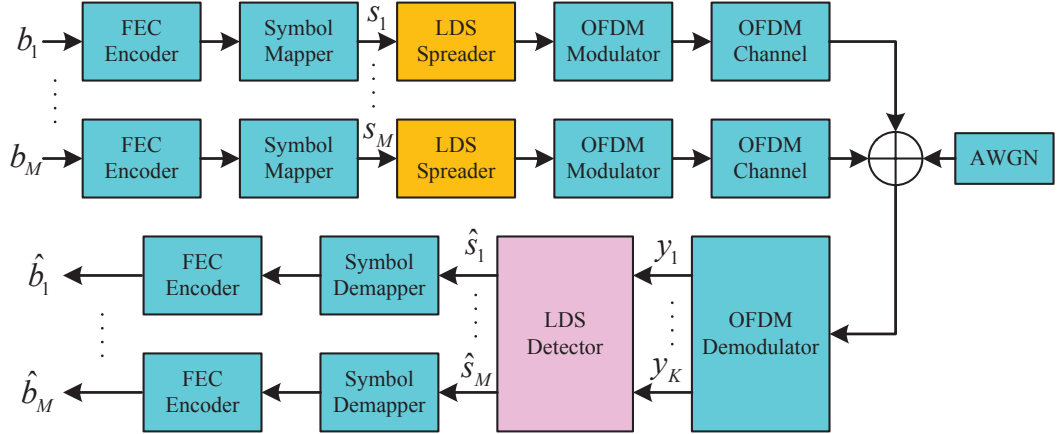}
\caption{Block diagram of an uplink LDS-OFDM system.}%
\label{LDSOFDMBLOCK}%
\end{figure}

In fact, the non-orthogonal feature can be introduced either in the power domain only or in the hybrid code and power domain.
{Although code domain NOMA offers a potential coding gain to improve spectral efficiency, power domain NOMA is simpler since only minor changes are required in the physical layer operations at the transmitter side compared to current 4G technologies.}
In addition, power domain NOMA paves the way for flexible resource allocation via relaxing the orthogonality requirement to improve the performance of NOMA, such as spectral efficiency\cite{Ding2014,Hanif2016,sun2016optimal}, energy efficiency\cite{Sun2015a,zhang2016energy,Zhang2018NOMA,Fuhui2018NOMA}, and user fairness\cite{Timotheou2015,LiuFairnessNOMA,wei2017fairness}.
Therefore, this thesis focuses on the power domain NOMA, including the system performance analysis and some specific designs when applying NOMA in diverse types of wireless communication systems\footnote{We note that the authors of \cite{ProceedingLiu,DaiCST2018} surveyed the state-of-the-art research on NOMA and offered a high-level discussion of the challenges and research opportunities for NOMA systems.
Interested readers are referred to \cite{ProceedingLiu,DaiCST2018} for more details.}.

\subsection{Performance Analysis on NOMA}
It has been shown that NOMA offers considerable performance gain over OMA in terms of spectral efficiency and outage probability\cite{Otao2012,Saito2013,Saito2013a,Xu2015,Ding2014,Dingtobepublished,Chen2017,Yang2016}.
Initially, the performance of NOMA was evaluated through simulations with perfect CSI by utilizing the proportional fairness scheduler\cite{Otao2012,Saito2013}, fractional transmission power allocation (FTPA)\cite{Saito2013}, and tree-search based transmission power allocation (TTPA)\cite{Saito2013a}.
These works showed that the overall cell throughput, cell-edge user throughput, and the degrees of proportional fairness achieved by NOMA are all superior to those of OMA.
More specifically, based on the achievable rate region, Xu \emph{et al.} proved in \cite{Xu2015} that NOMA outperforms time division multiple access (TDMA) with a high probability in terms of both its overall sum-rate and the individual user-rate.
Furthermore, the ergodic sum-rate of single-input single-output NOMA (SISO-NOMA) was derived and the performance gain of SISO-NOMA over SISO-OMA was demonstrated via simulations by Ding \emph{et al.} in a cellular downlink scenario with randomly deployed users\cite{Ding2014}.
With the proposed asymptotic analysis, it was shown that NOMA is asymptotically equivalent to the opportunistic multiple access technique.
Upon relying on their new dynamic resource allocation design, Chen \emph{et al.} \cite{Chen2017} proved that SISO-NOMA always outperforms SISO-OMA using a rigorous optimization technique.
In \cite{Yang2016}, Yang \emph{et al.} analyzed the outage probability degradation and the ergodic sum-rate of SISO-NOMA systems by taking into account the impact of partial channel state information (CSI).
It it shown that NOMA based on the second order statistical CSI always achieves a better performance than that of NOMA based on imperfect CSI, while it can achieve similar performance to the NOMA with perfect CSI in the low SNR region.
A remarkable work in \cite{Dingtobepublished} characterized the impact of user pairing on the performance of a two-user SC-NOMA system with fixed power allocation and cognitive radio inspired power allocation, respectively.
The authors proved that, for fixed power allocation, the performance gain of NOMA over OMA increases when the difference in channel gains between the paired users becomes larger.
The aforementioned contributions studied the performance of NOMA systems or discussed the superiority of NOMA over OMA in different contexts.
However, analytical results quantifying the ergodic sum-rate gain (ESG) of SISO-NOMA over SISO-OMA have not been reported.
Further exploration on performance analysis of NOMA system should be carried out to investigate the underlying insights of NOMA.

To achieve a higher spectral efficiency, the concept of NOMA has also been extended to the case of multi-antenna systems, resulting in the notion of multiple-input multiple-output NOMA (MIMO-NOMA), for example, by invoking the signal alignment technique of Ding \emph{et al.} \cite{DingSignalAlignment} and the quasi degradation based precoding design of Chen \emph{et al.} \cite{ChenQuasiDegradation}.
Although the performance gain of MIMO-NOMA over MIMO-OMA has indeed been shown in \cite{DingSignalAlignment,ChenQuasiDegradation} with the aid of simulations, the performance gain due to additional antennas has not been quantified analytically.
Moreover, how the ESG of NOMA over OMA increases upon upgrading the system from having a single antenna to multiple antennas is still an open problem at the time of writing, which deserves our efforts to explore.
The answers to these questions can shed light on the practical implementation of NOMA in future wireless networks.
On the other hand, there are only some preliminary results on applying the NOMA principle to massive MIMO systems.
For instance, Zhang \emph{et al.} \cite{ZhangDiMassive} investigated the outage probability of massive MIMO-NOMA (\emph{m}MIMO-NOMA).
Furthermore, Ding and Poor \cite{DingMassive} analyzed the outage performance of \emph{m}MIMO-NOMA relying on realistic limited feedback and demonstrated a substantial performance improvement for \emph{m}MIMO-NOMA over \emph{m}MIMO-OMA.
Upon extending NOMA to an mmWave massive MIMO system, the capacity attained in the high-SNR regime and low-SNR regime were analyzed by Zhang \emph{et al.} \cite{ZhangDimmWave}.
Yet, the ESG of \emph{m}MIMO-NOMA over \emph{m}MIMO-OMA remains unknown and the investigation of \emph{m}MIMO-NOMA has the promise attaining NOMA gains in large-scale systems in the networks of the near future.

While single-cell NOMA has received significant research attention \cite{Xu2015,Ding2014,Chen2017,Yang2016,Sun2016Fullduplex,WeiTCOM2017,DingSignalAlignment,ChenQuasiDegradation,ZhangDiMassive,DingMassive,ZhangDimmWave}, the performance of NOMA in multi-cell scenarios remains unexplored but critically important for practical deployment, where the inter-cell interference becomes a major obstacle\cite{ShinMulticellNOMA}.
Centralized resource optimization of multi-cell NOMA was proposed by You \emph{et al.} in \cite{YouMulticellNOMA}, while a distributed power control scheme was studied in \cite{FuMulticellNOMA}.
The transmit precoder design of MIMO-NOMA aided multi-cell networks designed for maximizing the overall sum throughput was proposed by Nguyen \emph{et al.} \cite{NguyenMulticellNOMA} and a computationally efficient algorithm was proposed for achieving a locally optimal solution.
Despite the fact that the simulation results provided by \cite{ShinMulticellNOMA,YouMulticellNOMA,FuMulticellNOMA,NguyenMulticellNOMA} have demonstrated a performance gain for applying NOMA in multi-cell cellular networks, the analytical results quantifying the ESG of NOMA over OMA for multi-cell systems relying on single-antenna, multi-antenna, and massive MIMO arrays at the base station (BS) have not been reported in the open literature.
Furthermore, the performance gains disseminated in the literature have been achieved for systems having a high transmit power or operating in the high-SNR regime.
However, a high transmit power inflicts a strong inter-cell interference, which imposes a challenge for the design of inter-cell interference management.
Therefore, there are many practical considerations related to the NOMA principle in multi-cell systems to be investigated.

\subsection{Resource Allocation Design for NOMA}
Resource allocation has received significant attention since it is critical in exploiting the potential performance gain of NOMA.
However, optimal resource allocation is very challenging for NOMA systems, since user scheduling and power allocation design couple with each other and generally result in NP-hard optimization problems\cite{Lei2016NOMA}.
Some initial works on resource allocation in \cite{Otao2012,Saito2013,Saito2013a} have been reported, but their results and achieved performance are generally far from optimal.
In \cite{Hojeij2015,Hojeij2015a}, Hojeij \emph{et al.} studied a two-user multi-carrier NOMA (MC-NOMA) system by minimizing the number of subcarriers assigned under the constraints of maximum allowed transmit power and requested data rates, and further introduced a hybrid orthogonal-nonorthogonal scheme.
In \cite{Di2015}, Di \emph{et al.} formulated the resource allocation problem to maximize the sum rate, which is a non-convex optimization problem due to the binary constraint and the existence of the interference term in the objective function.
As a compromise solution, they proposed a suboptimal algorithm by employing matching theory and water-filling power allocation.
In \cite{Lei2015}, Lei \emph{et al.} presented a systematic approach for NOMA resource allocation from a mathematical optimization point of view.
They formulated the joint power and channel allocation problem of a downlink multiuser MC-NOMA system, and proved its NP-hardness based on \cite{Liu2014}.
Furthermore, they proposed a competitive suboptimal algorithm based on Lagrangian duality and dynamic programming, which significantly outperforms OFDMA as well as NOMA with FTPA.
To fully reveal the maximum performance gain brought by NOMA, a globally optimal resource allocation solution is highly expected.

Most of aforementioned contributions  focused on the optimal resource allocation for maximizing the sum rate.
However, fairness is another important objective to optimize for resource allocation of NOMA.
Proportional fairness (PF) has been adopted as a metric to balance the transmission efficiency and user fairness in many works\cite{Kim2004PF,Wengerter2005PF}.
In \cite{Liu2015b}, Liu \emph{et al.} proposed a user pairing and power allocation scheme for downlink two-user MC-NOMA based on a PF objective.
A prerequisite for user pairing was given and a closed-form optimal solution for power allocation was derived.
Apart from PF, max-min or min-max methods are usually adopted to achieve user fairness.
Given a preset user group, Timotheou \emph{et al.} in \cite{Timotheou2015} studied the power allocation problem from a fairness standpoint by maximizing the minimum achievable user rate with instantaneous CSI and minimizing the maximum outage probability with statistical CSI.
Although the resulting problems are non-convex, simple low-complexity algorithms were developed to provide close-to-optimal solutions.
Similarly, another paper \cite{Shi2015} studied the outage balancing problem of a downlink multiuser MC-NOMA system to maximize the minimum weighted success probability with and without user grouping.
Joint power allocation and decoding order selection solutions were given, and the inter-group power and resource allocation solutions were also provided.

In addition to the spectral efficiency maximization, the power-efficient or energy-efficient resource allocation for NOMA has become an important focus to enable environment friendly and cost-effective wireless communication networks.
Note that power-efficient resource allocation has been extensively studied in OMA systems\cite{DerrickEEOFDMA,DerrickFD2016}.
Ng \emph{et al.} in \cite{DerrickEEOFDMA,DerrickFD2016} studied the power-efficient resource allocation design for large number of base station antennas and full-duplex radio base stations, respectively.
To address the green radio design for NOMA systems, Zhang \emph{et al.} in \cite{zhang2016energy} proposed an optimal power allocation strategy for a single-carrier NOMA system to maximize the energy efficiency, while a separate subcarrier assignment and power allocation scheme was proposed for MC-NOMA systems in \cite{FangEnergyEfficientNOMAJournal}.
To minimize the total power consumption, Lei \emph{et al.} in \cite{Lei2016NPM} designed a suboptimal ``relax-then-adjust" algorithm for MC-NOMA systems.
Most recently, Zhang \emph{et al.} \cite{Zhang2018NOMA} and Zhou \emph{et al.} \cite{Fuhui2018NOMA} studied the energy-efficient resource allocation design in NOMA heterogeneous networks and NOMA-enabled heterogeneous cloud radio access networks (H-CRAN), respectively.

In summary, many existing works focused on the resource allocation for NOMA systems under the assumption of perfect CSI at the transmitter side.
However, there are only few works on the joint user scheduling and power allocation problem for MC-NOMA systems under imperfect CSI, not to mention the SIC decoding order selection problem.
In fact, under imperfect CSI, the SIC decoding order cannot be determined by channel gain order, and some other metrics, such as distance, priority, and target rates, are potential criteria to decide the SIC decoding order.

To further enhance the system performance, MIMO-NOMA has become one of the main research focuses in the literature, where the BS and users are equipped with multiple antennas\cite{DingSignalAlignment,ChenQuasiDegradation}.
In MIMO-NOMA systems, basically, multiple users are clustered into several groups, where multiple user groups are multiplexed in spatial domain and users within a group are multiplexed in power domain.
As a result, MIMO-NOMA suffers from both the inter-cluster interference and intra-cluster interference.
Therefore, user clustering, precoding design, and power allocation are critical for improving the system performance of MIMO-NOMA.

Initially, Kim \emph{et al.} in \cite{Kim2013} proposed a clustering and power allocation algorithm for a two-user MIMO-NOMA system, where there are at most two users in each user group.
Both the channel correlation and channel gain difference were taken into consideration to reduce intra-cluster interference and inter-cluster interference simultaneously \cite{Kim2013}.
In \cite{Choi2015}, Choi proposed a minimum power multicast beamforming scheme and applied to two-user NOMA systems for multi-resolution broadcasting.
The proposed two-stage beamforming method outperforms the zero-forcing (ZF) beamforming scheme in \cite{Kim2013}.
Sun \emph{et al.} in \cite{Sun2015} studied the ergodic sum capacity maximization problem of a two-user MIMO-NOMA system under statistical CSI with the total power constraint and minimum rate constraint for the weak user.
This work derived the optimal input covariance matrix and proposed the optimal power allocation scheme as well as a low complexity suboptimal solution.
Furthermore, Sun \emph{et al.} \cite{Sun2015b} further extend their study to the sum-rate optimization problem, where NOMA users have different precoders.
The optimal precoding covariance matrix was derived by utilizing the duality between uplink and downlink, and a low complexity suboptimal solution based on singular value decomposition (SVD) was also provided.

Then, with relaxing the constraint of having no more than two users in each NOMA group, some researchers studied the resource allocation in a multi-user MIMO-NOMA system.
In \cite{Dingtobepublisheda}, Ding \emph{et al.} proposed a new design of precoding and detection matrices for a downlink multiuser MIMO-NOMA system, then analyzed the impact of user clustering as well as power allocation on the sum rate and outage probability of MIMO-NOMA system.
Furthermore, in \cite{DingSignalAlignment}, a transmission framework based on signal alignment was proposed for downlink and uplink two-user MIMO-NOMA systems.
Hanif \emph{et al.} in \cite{Hanif2016} studied the sum rate maximization problem of a downlink multiuser multiple-input single-output NOMA (MISO-NOMA) system.
Recently, a multiuser MIMO-NOMA scheme based on limited feedback was proposed and analyzed in \cite{DingMassive}.
Invoking the concept of quasi-degradation in a MISO-NOMA system, Chen \emph{et al.} \cite{ChenQuasiDegradation} derived the closed-form expressions of the optimal precoding vectors for the minimization of total power consumption and proposed a hybrid-NOMA precoding scheme to yield a practical transmission scheme.
The notion of NOMA was further extend to the case of secure communications.
Sun \emph{et al.} \cite{SunMISONOMA} studied the optimal resource allocation for maximization of the weighted system throughput of a MISO-NOMA system while the information leakage is constrained and artificial noise is injected to guarantee secure communication in the presence of multiple potential eavesdroppers.

\subsection{The Coexistence of NOMA and MmWave Communications}
Recently, spectrum congestion under 6 GHz frequency bands in current cellular systems creates a fundamental bottleneck for capacity improvement and sustainable system evolution\cite{Rappaport2013,Andrews2014}.
Subsequently, it is necessary and desirable to use high frequency bands, where a wider frequency bandwidth is available, such as mmWave bands \cite{Rappaport2013,ZhaommWaveMulticell,ZhaoMulticellMMwAVE} ranging from 30 GHz to 300 GHz.
On the other hand, multiple access technology is fundamentally important to support multi-user communications in wireless cellular networks.
Although communication systems utilizing microwave bands, i.e., sub-6 GHz, have been widely investigated, the potential multiple access scheme for mmWave communication systems is still unclear.
Meanwhile, NOMA has been recognized as a promising multiple access technique for the 5G wireless networks due to its higher spectral efficiency and capability to support massive connectivity\cite{Dai2015,Ding2015b,WeiSurvey2016,Mao2018,QiuLOMA,QiuLOMASlowFading}.
The interwork between the two important techniques via applying the NOMA concept in mmWave communications are explored by the thesis, which offers a viable solution to overcome the limitation of the limited number of RF chains and to improve the system energy efficiency.

In the literature, two kinds of architectures have been proposed for mmWave communications, i.e., fully digital architecture and hybrid architecture \cite{XiaoMing2017,zhao2017multiuser,zhao2019distributed,GaoSubarray,lin2016energy,Dai2017}.
Specifically, the fully digital architecture requires a dedicated RF chain\footnote{A RF chain consists of an ADC/DAC, a power amplifier, a mixer, and a local oscillator, etc.\cite{lin2016energy}.} for each antenna.
Hence, the tremendous energy consumption of RF chains, and the dramatically increased signal processing complexity and cost become a major obstacle in applying the fully digital architecture to mmWave systems in practical implementations.
Hybrid architectures, including fully access\cite{zhao2017multiuser} and subarray structures \cite{GaoSubarray}, provide a feasible and compromise solution for implementing mmWave systems which strike a balance between energy consumption, system complexity, and system performance.
In particular, for the fully access hybrid structures, each RF chain is connected to all antennas through an individual group of phase shifters \cite{zhao2017multiuser}.
For subarray hybrid structures\cite{GaoSubarray}, each RF chain has access to only a disjoint subset of antennas through an exclusive phase shifter for each antenna.
In essence, the two kinds of hybrid structures separate the signal processing into a digital precoder operated in baseband and an analog beamformer operated in RF band.
The comparison of the two types of hybrid structures can be found in \cite{lin2016energy}.
In general, the hybrid mmWave architectures are practical for implementation due to the promising system performance, which is also the focus of the thesis.
Lately, most of existing works \cite{zhao2017multiuser,GaoSubarray,lin2016energy} have investigated the channel estimation and hybrid precoding design for hybrid mmWave architectures.
However, the design of potential and efficient multiple access schemes for hybrid mmWave systems is rarely discussed.

Conventional OMA schemes adopted in previous generations of wireless networks cannot be applied directly to the hybrid mmWave systems, due to the associated special propagation features and hardware constraints\cite{wong2017key,Rappaport2013}.
For instance, in hybrid mmWave systems, an analog beamformer is usually shared by all frequency components in the whole frequency band.
Subsequently, FDMA and OFDMA are only applicable to the users covered by the same analog beam.
Unfortunately, the beamwidth of an analog beam in mmWave frequency band is typically narrow with a large antenna array\footnote{The $-3$ dB beamwidth of a uniform linear array with $M$ half wavelength spacing antennas is about $\frac{{102.1}}{M}$ degrees \cite{van2002optimum}.} and hence only limited number of users can be served via the same analog beam.
As a result, the limited beamwidth in practical systems reduces the capability of accommodating multiple users via FDMA and OFDMA, despite the tremendous bandwidth in mmWave frequency band.
Similarly, CDMA also suffers from the problem of narrow beamwidth, where only the users located within the same analog beam can be served via CDMA.
Even worse, the CDMA system performance is sensitive to the power control quality due to the near-far effects.
Another OMA scheme, TDMA, might be a good candidate to facilitate multi-user communication in hybrid mmWave systems, where users share the spectrum via orthogonal time slots.
However, it is well-known that the spectral efficiency of TDMA is inferior to that of NOMA \cite{Ding2015b,zhang2016energy}.
Moreover, the key challenge of implementing TDMA in hybrid mmWave systems is the requirement of high precision in performing fast timing synchronization since mmWave communications usually provide a high symbol rate.
On the other hand, SDMA \cite{zhao2017multiuser} is a potential technology for supporting multi-user communications, provided that the base station is equipped with enough number of RF chains and antennas.
However, in hybrid mmWave systems, the limited number of RF chains restricts the number of users that can be served simultaneously via SDMA, i.e., one RF chain can serve at most one user.
In particular, in overloaded scenarios, i.e., the number of users is larger than the number of RF chains, SDMA fails to accommodate all the users.
More importantly, in order to serve a large number of users via SDMA, more RF chains are required which translates to a higher implementation cost, hardware complexity, and energy consumption.
Thus, the combination of SDMA and mmWave \cite{zhao2017multiuser} is unable to cope with the emerging need of massive connectivity required in the future 5G communication systems\cite{Andrews2014}.
Therefore, we attempt to overcome the limitation incurred by the small number of RF chains in hybrid mmWave systems.
To this end, we introduce the concept of NOMA into hybrid mmWave systems, which allows the system to serve more users with a limited number of RF chains.

Another benefit of applying NOMA in hybrid mmWave communications is the enhanced system energy efficiency, especially considering the practical blockage effect in mmWave channels.
In particular, when the line-of-sight (LOS) path of a user is blocked, the RF chain serving this user can only rely on non-line-of-sight (NLOS) paths to provide communications.
However, the channel gain of the NLOS paths is usually much weaker than that of the LOS path in mmWave frequency bands, due to the high attenuation in reflection and penetration\cite{Rappaport2013}.
As a result, the dedicated RF chain serving this user consumes a lot of system resources for achieving only a low data rate, which translates into a low energy efficiency.
In fact, when the LOS path of a user is blocked, it should be treated as a weak user in the context of NOMA.
Thus, the weak user can be clustered with a strong user possessing a LOS link to form a NOMA group serving by only one RF chain to exploit the near-far diversity.
Meanwhile, the original RF chain and its associated phase shifter (PS) dedicated to the weak user become idle, which can substantially save the associated circuit power consumption.
Inspired by these observations, the combination of NOMA and hybrid mmWave technology is a promising solution to improve the system energy efficiency.

In contrast to conventional OMA schemes mentioned above, NOMA can serve multiple users via the same DoF and achieve a higher spectral efficiency\cite{Ding2014}.
Several preliminary works considered NOMA schemes for mmWave communications, which were shown to offer a higher spectral efficiency compared to conventional OMA schemes\cite{DingFRAB,Ding2017RandomBeamforming,Cui2017Optimal, WangBeamSpace2017}.
Ding \emph{et al.} in \cite{DingFRAB} proposed a NOMA transmission scheme in massive MIMO mmWave communication systems via exploiting the features of finite resolution of PSs.
Then, the outage performance analysis for mmWave NOMA systems with random beamforming was studied in \cite{Ding2017RandomBeamforming}.
Subsequently, Cui \emph{et al.} in \cite{Cui2017Optimal} proposed a user scheduling and power allocation design for the random beamforming mmWave NOMA  scheme.
In \cite{WangBeamSpace2017}, a beamspace mmWave MIMO-NOMA scheme was proposed by using a lens antenna array and a power allocation algorithm was designed to maximize the system sum-rate.

Although all these works have demonstrated the potential spectral efficiency gain of applying NOMA in mmWave systems, the energy efficiency gain brought by NOMA for mmWave communication systems is rarely discussed in the literature.
More importantly, due to the high carrier frequency in mmWave frequency band, massive numbers of antennas are usually equipped at transceivers and thus the beamwidths of the associated analog beams are typically narrow \cite{BusariSurveyonMMWave}.
As a result, the number of users that can be served concurrently by the mmWave-NOMA scheme is very limited and it depends on the users' angle-of-departure (AOD) distribution.
This reduces the potential spectral efficiency and energy efficiency gain brought by NOMA, which calls for new designs of applying NOMA to hybrid mmWave communication systems.

{\subsection{Summary}
Based on the literature review, we summarize the main research questions and knowledge gaps of NOMA for future wireless networks, which motivates this thesis.
\begin{itemize}
  \item There is a paucity of literature on the comprehensive
   analysis of performance gain of NOMA over OMA in various specifically considered system setups, including single-antenna, multi-antenna, and massive-MIMO systems with both single-cell and
   multi-cell deployments.
   Some basic questions about the performance gain of NOMA over OMA are remained to be answered, such as where it comes from and what its behavior is in different system setups.
  \item SIC decoding plays an important role for realizing the performance gain of NOMA, but it suffers from the well-known error propagation\cite{Nara2005Error}. Therefore, how to mitigate the error propagation of SIC decoding in NOMA systems is worth to explore.
  \item As NOMA can be realized via power domain multiplexing, resource allocation design is crucial for exploiting the potential performance gain of NOMA systems.
      However, most of existing resource allocation designs \cite{Lei2016NOMA,Di2016sub,Sun2016Fullduplex,Liu2015b} are based on perfect CSIT, which is unlikely to acquire in practice.
      Therefore, how to design the robust resource allocation strategy for NOMA systems with CSIT uncertainty is interesting and remained for further investigation.
  \item Due to the spectrum congestion in the currently microwave frequency band, it is expected to push the carrier frequency of future wireless network to high frequency bands, where a wider frequency bandwidth is available, such as mmWave bands ranging from 30 GHz to 300 GHz.
      Therefore, the coexistence of NOMA and mmWave communications is remained to explore, including the feasibility of NOMA in mmWave frequency band and the potential benefits of applying NOMA in mmWave communication systems.
\end{itemize}}

\section{Thesis Outline and Main Contributions}

\subsection{Thesis Organization}
In this subsection, the outline of each chapter in this thesis is given.
There are eight chapters in total, including an introduction of the thesis, the necessary background knowledge, the technical details of the conducted research, and the conclusion of this thesis.
The technical parts from Chapter 3 to Chapter 7 of the thesis are ordered in a logic from uplink NOMA systems to downlink NOMA systems as well as from NOMA in microwave communications to NOMA in mmWave communications.

\textbf{Chapter 1}

This chapter provides an overview of 5G communication usage scenarios, key requirements, and enabled technologies.
Then, the motivation of the considered research questions of the thesis and the relevant existing works are presented.
It also provides the outline and the main contributions of this thesis.

\textbf{Chapter 2}

In this chapter, some related background knowledges of the thesis are presented, including the channel models, the fundamental concepts of NOMA, spectral efficiency, energy efficiency, and resource allocation design methodologies. Toy examples with relevant figures are provided to help understanding these concepts. The materials presented in this chapter will be used throughout the rest of this thesis.

\textbf{Chapter 3}

In this chapter, we investigate and reveal the ESG of NOMA over OMA in uplink cellular communication systems.
{A unified performance analysis on ESG of NOMA over OMA in single-antenna, multi-antenna and massive MIMO systems considering both single-cell and multi-cell deployments are presented.}
The distinctive behaviors of ESG of NOMA over OMA under different scenarios are unveiled to provide some interesting insights.

\textbf{Chapter 4}

In this chapter, we propose a joint pilot and payload power allocation (JPA) scheme to reduce the effect of error propagation issue for an uplink MIMO-NOMA system with a maximal ratio combining and SIC (MRC-SIC) receiver.
Details on how to formulate the JPA design as an optimization problem and how to solve the formulated problem are presented in this chapter.

\textbf{Chapter 5}

In this chapter, we study the power-efficient resource allocation design for an downlink MC-NOMA system by taking into account the imperfection of CSI
at transmitter.
A joint design of power allocation, rate allocation, user scheduling, and SIC decoding policy for
minimizing the total transmit power is formulated and solved.
A globally optimal solution is obtained as a performance benchmark, while a low-complexity suboptimal iterative algorithm is developed for practical implementations.

\textbf{Chapter 6}

In this chapter, we propose a multi-beam NOMA scheme for downlink hybrid mmWave communication systems and study its resource allocation.
The proposed scheme can perform NOMA transmission for the users with an arbitrary angle-of-departure distribution, which provides a higher flexibility for NOMA user clustering and thus can efficiently exploit the potential multi-user diversity.
A suboptimal two-stage resource allocation design for maximizing the system sum-rate is also presented in this chapter.

\textbf{Chapter 7}

In this chapter, we present a further work of applying NOMA in downlink mmWave communications.
In particular, we propose a novel NOMA scheme with beamwidth control for hybrid mmWave communication systems and study the resource allocation design to maximize the system energy efficiency.
The details on how to control the analog beamwidth and how to design energy-efficient resource allocation policy are presented in this chapter.

\textbf{Chapter 8}

This chapter concludes the thesis by summarizing the main ideas of each chapter and the contributions of all the works conducted during my Ph.D. research.
{The potential future works arising from this thesis are also outlined.}

\subsection{Research Contributions}
This thesis studies the performance analysis and design for NOMA in different types of wireless communication systems, covering from uplink to downlink communication systems and from microwave to mmWave communication systems.
The developed theoretical results can serve as guidelines for the practical implementation of NOMA in future wireless communication systems.
The designed spectral-efficient and energy-efficient resource allocation algorithms have high potential to find applications in several current or upcoming wireless
communication standards.
In what follows, a detailed list of the research contributions in Chapters 3-7 are presented.

\begin{enumerate}
  \item Chapter 3 presents {a unified} performance analysis on ESG of NOMA over OMA in single-antenna, multi-antenna and massive MIMO uplink communication systems.
      We first focus on the ESG analysis in single-cell systems and then extend our analytical results to multi-cell systems by taking into account the inter-cell interference (ICI).
      We quantify the ESG of NOMA over OMA relying on practical signal reception schemes at the base station for both NOMA as well as OMA and unveil its behavior under different scenarios.
      Our simulation results confirm the accuracy of our performance analyses and provide some interesting insights.

      In particular, in single-antenna systems, we identify two types of gains brought about by NOMA: 1) a large-scale near-far gain arising from the distance discrepancy between the base station and users; 2) a small-scale fading gain originating from the multi-path channel fading.
      Furthermore, we reveal that the large-scale near-far gain increases with the normalized cell size, while the small-scale fading gain is a constant, given by $\gamma = 0.57721$ nat/s/Hz, in Rayleigh fading channels.
      When extending single-antenna NOMA to $M$-antenna NOMA, we prove that both the large-scale near-far gain and small-scale fading gain achieved by single-antenna NOMA can be increased by a factor of $M$ for a large number of users.
      Moreover, given a massive antenna array at the base station and considering a fixed ratio between the number of antennas, $M$, and the number of users, $K$, the ESG of NOMA over OMA increases linearly with both $M$ and $K$.
      Compared to the single-cell case, the ESG in multi-cell systems degrades as NOMA faces more severe inter-cell interference due to the non-orthogonal transmissions.
      Besides, we unveil that a large cell size is always beneficial to the ergodic sum-rate performance of NOMA in both single-cell and multi-cell systems.
      Numerical results verify the accuracy of the analytical results derived and confirm the insights revealed about the ESG of NOMA over OMA in different scenarios.

The results in Chapter 3 have been presented in the following publications:
\begin{itemize}
\item \textbf{Z. Wei}, L. Yang, D. W. K. Ng, and J. Yuan, ``On the Performance Gain of NOMA over OMA in Uplink Single-cell Systems,'' in \emph{Proc. IEEE Global Commun. Conf.}, pp. 1--7, 2018.

\item \textbf{Z. Wei}, L. Yang, D. W. K. Ng, J. Yuan, and L. Hanzo ``On the Performance Gain of NOMA over OMA in Uplink Communication Systems'', {\em IEEE Trans. Commun.}, major revision, 30th Apr. 2019.
\end{itemize}

  \item In Chapter 4, a JPA scheme for uplink MIMO-NOMA with an MRC-SIC receiver is proposed to alleviate the error propagation in SIC decoding of NOMA.
      The average signal-to-interference-plus-noise ratio (ASINR) of each user during the MRC-SIC decoding is analyzed by taking into account the error propagation due to the channel estimation error.
      Furthermore, the JPA design is formulated as a non-convex optimization problem to maximize the minimum weighted ASINR.
      The formulated problem is transformed into an equivalent geometric programming problem and is solved optimally.
      Simulation results confirm the developed performance analysis and show that our proposed scheme can effectively alleviate the error propagation of MRC-SIC and enhance the detection performance, especially for users with moderate energy budgets.

      The results in Chapter 4 have been presented in the following publication:
        \begin{itemize}
            \item \textbf{Z. Wei}, D. W. K. Ng, J. Yuan,
                ``Joint Pilot and Payload Power Control for Uplink MIMO-NOMA with MRC-SIC Receivers,'' {\em IEEE Commun. Lett.}, vol.~22, no.~4, pp. 692--695, Apr. 2018.
        \end{itemize}
  \item In Chapter 5, the problem of power-efficient resource allocation design for downlink MC-NOMA systems is studied.
      The resource allocation design is formulated as a non-convex optimization problem which jointly designs the power allocation, rate allocation, user scheduling, and SIC decoding policy for minimizing the total transmit power.
      The proposed framework takes into account the imperfection of CSI at transmitter and QoS requirements of users.
      A \emph{channel-to-noise ratio outage threshold} is defined to facilitate the design of optimal SIC decoding policy on each subcarrier.
      Subsequently, the considered non-convex optimization problem is recast as a generalized linear multiplicative programming problem, for which a globally optimal solution is obtained via employing the branch-and-bound approach.
      The optimal resource allocation policy serves as a system performance benchmark due to its high computational complexity.
      To strike a balance between system performance and computational complexity, we propose a suboptimal iterative resource allocation algorithm based on difference of convex programming.
      Simulation results demonstrate that the suboptimal scheme achieves a close-to-optimal performance. Also, both proposed schemes provide significant transmit power savings, compared to that of conventional OMA schemes.

      The results in Chapter 5 have been presented in the following publications:
        \begin{itemize}
            \item \textbf{Z. Wei}, D. W. K. Ng, J. Yuan ``Power-Efficient Resource Allocation for MC-NOMA with Statistical Channel State Information,'' in \emph{Proc. IEEE Global Commun. Conf.}, pp. 1--7, Dec. 2016.
            \item \textbf{Z. Wei}, J. Yuan, D. W. K. Ng, H.-M. Wang
    ``Optimal Resource Allocation for Power-Efficient MC-NOMA with Imperfect Channel State Information,'' {\em IEEE Trans. Commun.}, vol.~65, no.~9, pp. 3944--3961, Sep. 2017.
        \end{itemize}
  \item In Chapter 6, we propose a multi-beam NOMA scheme for hybrid mmWave systems and study its resource allocation.
        A beam splitting technique is designed to generate multiple analog beams to serve multiple NOMA users on each radio frequency chain.
        In contrast to the recently proposed single-beam mmWave-NOMA scheme which can only serve multiple NOMA users within the same analog beam, the proposed scheme can perform NOMA transmission for the users with an arbitrary angle-of-departure distribution.
        This provides a higher flexibility for applying NOMA in mmWave communications and thus can efficiently exploit the potential multi-user diversity.
        Then, we design a suboptimal two-stage resource allocation for maximizing the system sum-rate.
        In the first stage, assuming that only analog beamforming is available, a user grouping and antenna allocation algorithm is proposed to maximize the conditional system sum-rate based on the coalition formation game theory.
        In the second stage, with the zero-forcing digital precoder, a suboptimal solution is devised to solve a non-convex power allocation optimization problem for the maximization of the system sum-rate which takes into account the quality of service constraints.
        Simulation results show that our designed resource allocation can achieve a close-to-optimal performance in each stage.
        In addition, we demonstrate that the proposed multi-beam mmWave-NOMA scheme offers a substantial spectral efficiency improvement compared to that of the single-beam mmWave-NOMA and the mmWave orthogonal multiple access schemes.

        The results in Chapter 6 have been presented in the following publications:
        \begin{itemize}
        \item \textbf{Z. Wei}, L. Zhao, J. Guo, D. W. K. Ng, and J. Yuan, ``A Multi-Beam NOMA Framework for Hybrid mmWave Systems,'' in \emph{Proc. IEEE Intern. Commun. Conf.}, pp. 1--7, 2018.

        \item \textbf{Z. Wei}, L. Zhao, J. Guo, D. W. K. Ng, and J. Yuan, ``Multi-Beam NOMA for Hybrid mmWave Systems'', {\em IEEE Trans. Commun.}, vol.~67, no.~2, pp. 1705--1719, Feb. 2019.
        \end{itemize}
  \item In Chapter 7, we propose a novel beamwidth control-based mmWave NOMA scheme and studied its energy-efficient resource allocation design.
      The proposed beamwidth control can increase the number of served NOMA groups by widening the beamwidth which can further exploit the energy efficiency gain brought by NOMA.
        To this end, two beamwidth control methods, based on the conventional beamforming and the Dolph-Chebyshev beamforming, respectively, are proposed.
        We firstly characterize the main lobe power losses due to the two beamwidth control methods and propose an effective analog beamformer design to minimize the power loss.
        Then, we formulate the energy-efficient resource allocation design as a non-convex optimization problem which takes into account the minimum required user data rate.
        A NOMA user grouping algorithm based on the coalition formation game theory is developed and a low-complexity iterative digital precoder design is proposed to achieve a locally optimal solution utilizing the quadratic transformation.
        Simulation results verify the fast convergence and effectiveness of our proposed algorithms.
        In addition, our results demonstrate the superior energy efficiency achieved by our proposed beamwidth controlling NOMA scheme compared to the conventional orthogonal multiple access and NOMA schemes without beamwidth control.

        The results in Chapter 7 have been presented in the following publications:
        \begin{itemize}
        \item \textbf{Z. Wei}, D. W. K. Ng, and J. Yuan, ``Beamwidth Control for NOMA in Hybrid mmWave Communication Systems'', in \emph{Proc. IEEE Intern. Commun. Conf.}, pp. 1--6, 2019.
        \item \textbf{Z. Wei}, D. W. K. Ng, and J. Yuan, ``NOMA for Hybrid MmWave Communication Systems with Beamwidth Control'', {\em IEEE J. Select. Topics Signal Process.}, vol. 13, no. 3, pp. 567-583, Jun. 2019.
        \end{itemize}
\end{enumerate}

\chapter{Background}\label{C2:chapter2}

\section{Introduction}
In this part, we first introduce channel models for wireless communications and then present the fundamental concepts of NOMA.
%
%The fundamental concepts of NOMA are also briefly discussed.
%
After that, the definitions of spectral efficiency and energy efficiency and their fundamental trade-off are also discussed.
Subsequently, we briefly introduce the optimization-based resource allocation design framework.
The materials in this chapter serve as the technical guidelines to provide the necessary background
to understand the works in the later chapters.

\section{Channel Models}
Wireless communication systems operate with electromagnetic waves, where the transmitted signal propagates through the physical medium from the transmitter to the receiver\cite{Tse2005}.
The radio propagation characteristics have a fundamental impact on each aspect of the wireless communication systems, in particular, cell-site planning, system
performance, equipment design, and signal processing requirements, etc.\cite{Shafi2018microwave}.
A defining characteristic of mobile wireless channels is the variations of the channel strength over time and over frequency, i.e., channel fading.
Roughly, channel fading can be divided into two types:
\begin{itemize}
  \item \emph{Large-scale fading}: Large-scale fading characterizes the channel attenuations which is mainly due to path loss of signal as a function of distance and shadowing by large objects such as buildings and hills.
      This occurs over a substantially larger distance compared to the carrier wavelength and is typically frequency independent.
  \item \emph{Small-scale fading}: Small-scale fading characterizes the channel fluctuations which is mainly caused by the constructive and destructive additions of the multi-path signals between the transmitter and receiver.
      This occurs at the spatial scale of the order of the carrier wavelength, and is usually frequency dependent.
\end{itemize}
To describe the statistics of small-scale and large-scale variations in different environments, many channel models have been suggested for wireless communications in various environments \cite{Goldsmith2005wireless,Tse2005,AkdenizChannelMmWave,HemadehmmWaveChannelModels}, just to cite a few.
However, the nature of radio propagation between microwave frequency band, i.e., sub-6GHz, and mmWave frequency band, i.e., 30-300 GHz is severely different, which results in different channel models.
Here, we review several probabilistic channel models which are widely used in wireless communications, mainly due to their simplicity and their ability to form a good approximation of the actual physical channels.
Also, some specific empirical channel models employed for microwave and mmWave communications in this thesis are presented.

\subsection{Large-scale Fading Models}

\underline{\textbf{Free Space Path Loss Model:}}

Consider a signal transmitted through free space without any obstructions between the transmitter and receiver.
Free-space path loss is defined as the path loss of the free-space model\cite{Goldsmith2005wireless}:
\begin{equation}\label{C2:FreeSpaceChannelModel}
  {P_{\mathrm{L}}}\left(d\right)\left[ {{\mathrm{dB}}} \right] = 10{\log _{10}}\frac{{{P_{\mathrm{t}}}}}{{{P_{\mathrm{r}}}}} = 10{\log _{10}}\frac{{{{\left( {4\pi d} \right)}^2}}}{{{G_t}{G_r}{\lambda ^2}}},
\end{equation}
where $\lambda$ is the wavelength of the information-carrying signal, $d$ is the distance between transmitter and receiver, $P_{\mathrm{t}}$ denotes the transmit power, $P_{\mathrm{r}}$ denotes the received power,
$G_{\mathrm{t}}$ represents the transmit antenna gain, and $G_{\mathrm{r}}$ is the receiving antenna gain.

\noindent\underline{\textbf{Simplified Distance-based Path Loss Model}}

A further simplified distance based path loss model has been widely used in the literature\cite{Ding2014,wei2019performanceGainJournal} to emphasize the impact of propagation distance $d$, which is given by:
\begin{equation}\label{C2:Distance-based}
  {P_{\mathrm{L}}}\left(d\right) = \frac{1}{{1+d^{n_{\mathrm{L}}}}},
\end{equation}
where $n_{\mathrm{L}}$ is the path loss exponent.
In this model, the path loss is only a simple function with respect to (w.r.t.) the distance $d$ between transmitter and receiver, without taking into account the effect of shadowing.
This simplified model is usually adopted to facilitate the performance analysis to reveal some underlying insights about the impact the large-scale fading on the system performance\cite{Ding2014,wei2019performanceGainJournal}, which is used in Chapter 3 of this thesis.

\noindent\underline{\textbf{Path Loss Channel Model with Shadowing}}

The free space model and the simplified distance-based model are the foundation of most large-scale fading models but they are inaccurate for mobile systems.
Therefore, a path loss model with shadowing has been proposed to account the large-scale variation at a given propagation distance $d$\cite{Tse2005,Goldsmith2005wireless}:
\begin{equation}\label{C2:LogDistanceChannelModel}
  {P_{\mathrm{L}}}\left(d\right)\left[ {{\mathrm{dB}}} \right] = {P_{\mathrm{L}}}\left(d_0\right)\left[ {{\mathrm{dB}}} \right] + 10n_{\mathrm{L}}\log_{10}\left(\frac{d}{d_0}\right) + X_{\sigma_{\mathrm{L}}},
\end{equation}
where $d_0$ is a reference distance for the antenna far-field and ${P_{\mathrm{L}}}\left(d_0\right)\left[ {{\mathrm{dB}}} \right]$ represents the reference path loss measured at $d_0$.
In addition, $X_{\sigma_{\mathrm{L}}}$ is a zero mean Gaussian random variable with variance $\sigma_{\mathrm{L}}^2$, i.e., $X_{\sigma_{\mathrm{L}}} \sim \mathcal{N}(0,\sigma_{\mathrm{L}}^2)$ which captures the effect of log-normal shadowing in wireless channels.
Note that the parameters $\sigma_{\mathrm{L}}$ and $n_{\mathrm{L}}$ would be found from empirical measurements.

\noindent\underline{\textbf{Empirical Large-scale Fading Models}}

Extensive measurement campaigns have been conducted to understand the physical characteristics in microwave\cite{Access2010,Goldsmith2005wireless} and mmWave frequency bands\cite{AkdenizChannelMmWave,HemadehmmWaveChannelModels}.
As a result, several empirical channel models have been adopted in both academia and industry.

For microwave communications, the standard 3rd Generation Partnership Project (3GPP) urban micro (UMi) path
loss model with hexagonal deployments \cite{Access2010} is given by
\begin{equation}\label{C2:UMiModel}
  {P_{\mathrm{L}}}\left( d \right)\left[ {{\mathrm{dB}}} \right] = 22.7 + 36.7{\log _{10}}\left( d \right) + 26{\log _{10}}\left( {{f_c}} \right)
\end{equation}
where $10 \;\mathrm{m} < d < 2000 \;\mathrm{m}$ is distance in meters and $0.45 \;\mathrm{GHz} \le f_c \le 6 \; \mathrm{GHz}$ is the carrier frequency in GHz.

For mmWave communications, the New York City (NYC) empirical model \cite{AkdenizChannelMmWave} has been widely used since this location is representative of likely initial deployment of mmWave cellular systems due to the high user density.
In particular, the path loss at a propagation distance $d$ is given by\cite{AkdenizChannelMmWave}
\begin{equation}\label{C2:NYCModel}
  {P_{\mathrm{L}}}\left( d \right)\left[ {{\mathrm{dB}}} \right] = \alpha_{\mathrm{L}}  + \beta_{\mathrm{L}} 10{\log _{10}}\left( d \right) + X_{\sigma_{\mathrm{L}}}, \;\; X_{\sigma_{\mathrm{L}}} \sim \mathcal{N}(0,\sigma_{\mathrm{L}}^2),
\end{equation}
where the propagation distance $d$ is in the range $30 \;\mathrm{m} < d < 200 \;\mathrm{m}$ with the unit of meters, $\alpha_{\mathrm{L}}$ and $\beta_{\mathrm{L}}$ are the least square
fits of floating intercept and slope over the measured distances, and $\sigma_{\mathrm{L}}^2$ is the log-normal shadowing variance.
The values of $\alpha_{\mathrm{L}}$, $\beta_{\mathrm{L}}$, and $\sigma_{\mathrm{L}}^2$ are fitted according to the NYC measurement data\cite{AkdenizChannelMmWave}.
Since the path losses for line-of-sight (LOS) link and non-line-of-sight (NLOS) link in mmWave channels are significantly different,  it is common to fit the LOS and NLOS path losses separately.
The model parameters for LOS and NLOS environments at carrier frequency $f_c = 28$ GHz and $f_c = 73$ GHz are given in Table \ref{C2:NYCModelParameters} according to the reference \cite{AkdenizChannelMmWave}.

\begin{table}[t]
\begin{center}
\small
\caption{NYC Empirical Model Parameters}
  \begin{tabular}{lll}
  \hline
                           & $f_c$ = 28 GHz                                       & $f_c$ = 73 GHz                                        \\ \hline
     \multirow{2}{*}{NLOS} & $\alpha_{\mathrm{L}} = 72.0$, $\beta_{\mathrm{L}} = 2.92$, $\sigma_{\mathrm{L}} = 8.7$ dB   & $\alpha_{\mathrm{L}} = 86.6$, $\beta_{\mathrm{L}} = 2.45$, $\sigma_{\mathrm{L}} = 8.0$ dB (\dag) \\
                           &                                                      & $\alpha_{\mathrm{L}} = 82.7$, $\beta_{\mathrm{L}} = 2.69$, $\sigma_{\mathrm{L}} = 7.7$ dB (\ddag) \\ \hline
     LOS                   & $\alpha_{\mathrm{L}} = 61.4$, $\beta_{\mathrm{L}} = 2$, $\sigma_{\mathrm{L}} = 5.8$ dB      & $\alpha_{\mathrm{L}} = 69.8$, $\beta_{\mathrm{L}} = 2$, $\sigma_{\mathrm{L}} = 5.8$ dB       \\ \hline
  \end{tabular}
  \label{C2:NYCModelParameters}
  \end{center}
\vspace{-8mm}
\begin{flushleft}
\footnotesize
\hspace{6mm}(\dag): Parameters for the 2 m-RX-height data and 4.06 m-RX-height data combined.\\
\hspace{6mm}(\ddag): Parameters for the 2 m-RX-height data only.\\
\end{flushleft}
\end{table}

\subsection{Small-scale Fading Models}
The large-scale fading model parameters are associated with the macro-scattering environment and change with time relatively slowly\cite{AkdenizChannelMmWave}.
In contrast, small-scale fading models are mainly built to characterize the constructive and destructive addition of different multi-path components introduced by the channel, which is usually to show rapid fluctuations in the signal's envelope.
The channel fading gain resulting from the effects of path loss, shadowing, and multi-path are compared in Figure \ref{C2:LargeScaleVersusSmallScale}.

\begin{figure}[t]
\centering
\includegraphics[width=4.5in]{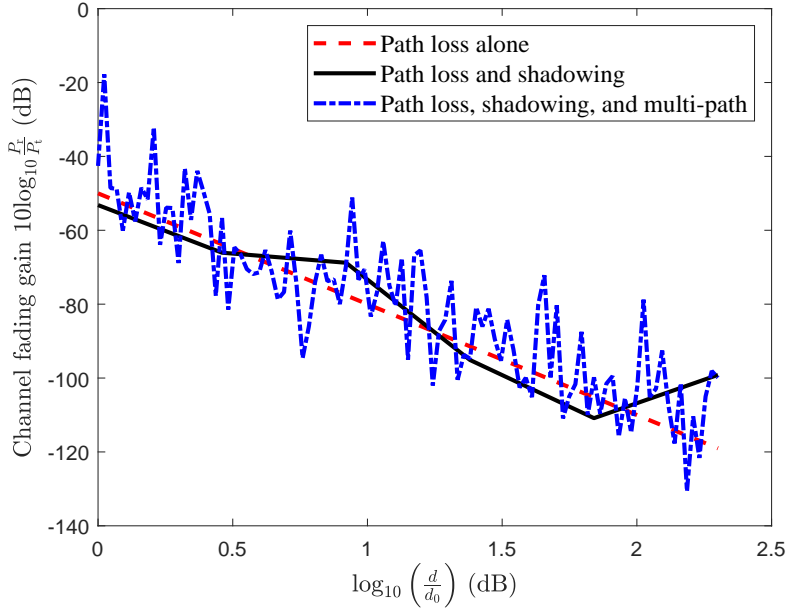}
\caption{An illustration of channel fading gain versus propagation distance resulting from large-scale fading and small-scale fading. For large-scaling fading, the log distance path loss with shadowing model in \eqref{C2:LogDistanceChannelModel} is adopted with $d_0 = 10$ m, ${P_{\mathrm{L}}}\left(d_0\right) = 50$ dB, $n_{\mathrm{L}} = 3$, and $\sigma_{\mathrm{L}} = 8$ dB. The Rayleigh fading model in \eqref{C2:RayleighFading} with $\sigma_l^2 = 1$ is adopted for small-scale fading.}%
\label{C2:LargeScaleVersusSmallScale}%
\end{figure}

\noindent\underline{\textbf{Rayleigh Fading Model}}

In rich-scattering environment with many small reflectors, the $l$-th channel filter tap $h_l[m]$ at time stamp $m$ is the sum of a large number of small independent random variables.
As a result, according to the Central Limit Theorem, it can reasonably be modeled as a zero-mean Gaussian random variable.
With the assumed Gaussian distribution of channel coefficients, the magnitude $\left|h_l[m]\right|$ of the $l$-th tap can be modeled by a Rayleigh random variable, i.e.,
\begin{equation}\label{C2:RayleighFading}
  \left|h_l[m]\right| \sim \frac{x}{{\sigma _l^2}}\exp \left\{ { - \frac{{{x^2}}}{{2\sigma _l^2}}} \right\}, x\ge 0,
\end{equation}
and the squared magnitude $\left|h_l[m]\right|^2$ is exponentially distributed, i.e.,
\begin{equation}\label{C2:ExponentFading}
  \left|h_l[m]\right|^2 \sim  \frac{1}{{\sigma _l^2}}\exp \left\{ { - \frac{{{x}}}{{2\sigma _l^2}}} \right\}, x\ge 0,
\end{equation}
where ${\sigma _l^2}$ denotes the variance of the complex channel coefficient of the $l$-th path.

\noindent\underline{\textbf{Rician Fading Model}}

\begin{figure}[ptb]
\centering
\includegraphics[width=4.5in]{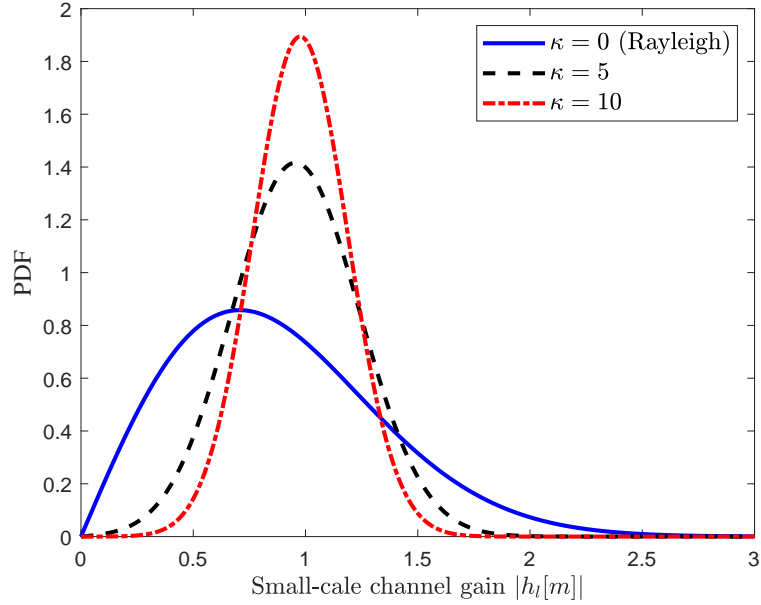}
\caption{PDF comparison of Rayleigh and Rician distributions with different Rician K-factor, $\kappa$, and $\sigma _l^2 = 1$.}%
\label{C2:RayleighRician}%
\end{figure}

This is a frequently used alternative channel model in which the LOS path dominates the signal propagation compared to other NLOS paths.
In this case, the $l$-th channel filter tap $h_l[m]$ at time $m$ can be modeled as
\begin{equation}\label{C2:RicianFading}
  h_l[m]  =  \sqrt{\frac{\kappa}{\kappa+1}}\sigma _l e^{j\theta} + \sqrt{\frac{1}{\kappa+1}}\mathcal{CN}(0,\sigma _l^2),
\end{equation}
where the first term corresponding to the LOS path arriving with a uniform distributed phase $\theta \in [0, 2\pi]$ and the second term corresponding to the aggregation of the large number of reflected and scattered NLOS paths, independent of $\theta$.
The parameter $\kappa$ (so-called Rician K-factor) is the ratio of the energy in the LOS path to the energy in the scattered paths.
The magnitude $\left|h_l[m]\right|$ follows a Rician
distribution.
Thus, this channel model is referred to as the Rician fading model.
Note that it is often a better channel fading model than the Rayleigh model since practical environments are usually not that rich-scattering.
Figure \ref{C2:RayleighRician} compares the probability density functions (PDFs) of the small-scale channel gain under Rayleigh and Rician fading models.
We can observe that when $\kappa = 0$, Rician fading model is equivalent to Rayleigh fading model.
In addition, the larger $\kappa$ is, the more deterministic the channel is.

\noindent\underline{\textbf{Saleh-Valenzuela Model}}

The above Rayleigh and Rician models are commonly adopted to model the small-scale fading channels in microwave communication systems, depending on the scattering environments.
However, in mmWave communications, an important feature is channel sparsity in the sense that majority of reflected multi-path components do not carry sufficient electromagnetic energy\cite{Shafi2018microwave}, due to the high attenuation in reflection\cite{Rappaport2013}.
Therefore, a Saleh-Valenzuela model has been widely adopted in the literature to model the sparse channel in mmWave communication systems\cite{XinyuGaoLetter,Brady2013,WangBeamSpace2017}.
In particular, considering an mmWave communication system with $M_{\mathrm{t}}$ transmit antennas and $M_{\mathrm{r}}$ receiving antennas, the channel matrix, ${{\bf{H}}} \in \mathbb{C}^{{ M_{\mathrm{r}} \times M_{\mathrm{t}}}}$, can be modelled as
\begin{equation}\label{C2:NarrowBandModel}
{{\bf{H}}} = {a _{0}}{{\bf{H}}_{0}} + \sum\limits_{l = 1}^L {{a _{l}}{{\bf{H}}_{l}}},
\end{equation}
where ${\mathbf{H}}_{0} \in \mathbb{C}^{ M_{\mathrm{r}} \times M_{\mathrm{t}} }$ is the LOS channel matrix with ${a _{0}}$ denoting the LOS complex path gain, ${\mathbf{H}}_{l} \in \mathbb{C}^{ M_{\mathrm{r}} \times M_{\mathrm{t}} }$ denotes the $l$-th NLOS path channel matrix with ${a _{l}}$ denoting the corresponding NLOS complex path gains, $1 \le l \le L$, and $L$ denoting the total number of NLOS paths.
In particular, ${\mathbf{H}}_{l}$, $\forall l \in \{0,\ldots,L\}$, is given by
\begin{equation}
{\mathbf{H}}_{l} = {\mathbf{a}}_{\mathrm{r}} \left(  \phi _{l} \right){\mathbf{a}}_{\mathrm{t}}^{\mathrm{H}}\left( \theta _{l} \right),
\end{equation}
with
\begin{equation}
{\mathbf{a}}_{\mathrm{t}}\left( \theta _{l} \right) = \left[ {1,{e^{ - j2\pi  \cos \theta _{l} }}, \ldots ,{e^{ - j{\left({M_{{\mathrm{t}}}} - 1\right)}\pi  \cos \theta _{l} }}}\right]^{\mathrm{T}},
\end{equation}
denoting the array response vector \cite{van2002optimum} for the angle-of-departure (AOD) of the $l$-th path ${\theta _{l}}$ at the transmitter and
\begin{equation}
{\mathbf{a}}_{\mathrm{r}}\left( \phi _{l} \right) = \left[ 1, {e^{ - j2\pi \cos \phi _{l} }}, \ldots ,{e^{ - j{\left({M_{{\mathrm{r}}}} - 1\right)}\pi \cos \phi _{l} }} \right] ^ {\mathrm{T}},
\end{equation}
denoting the array response vector for the angle-of-arrival (AOA) of the $l$-th path ${\phi _{l}}$ at the receiver.

\subsection{Blockage Model for MmWave Communications}
One important feature of mmWave communications is their vulnerability to blockage due to its higher penetration loss and deficiency of diffraction in mmWave frequency band\cite{Andrews2017}.
To capture the blockage effects in mmWave systems, a probabilistic model was proposed in \cite{Andrews2017} and has been widely adopted in the literature \cite{Ding2017RandomBeamforming,WeiBeamWidthControl}.
The probability for an arbitrary link with a distance $d$ having a LOS link is given by
\begin{equation}
{{\mathrm{P}}_{{\mathrm{LOS}}}}\left( d \right) = {e^{{-d}/\varrho}},
\end{equation}
where $\varrho$ is determined by the building density and the shape of the buildings, etc..

\section{Fundamental Concepts of NOMA}
In this section, we first introduce two key enabling technologies for NOMA, including superposition coding (SC) and SIC.
Then, the basic system models and concepts for downlink and uplink NOMA are presented.
In the last part, the multi-user capacity regions for downlink and uplink NOMA are provided and discussed.

\subsection{Superposition Coding and Successive Interference Cancelation}
The concept of NOMA was originally proposed for downlink multi-user communications or so called broadcast channel\cite{CoverBC}, where the two fundamental building blocks are SC and SIC.
However, NOMA inherently exists in uplink communications, since the electromagnetic waves are naturally superimposed at a receiving base station (BS) and the implementation of SIC is more affordable at the BS than at user terminals.
Therefore, the NOMA concept was recently generalized to uplink multi-user communications\cite{Al-Imari2014,Al-Imari2015,Yang2016NOMA,wei2017fairness}, i.e., multiple access channel, which will be detailed in the latter subsection.
In this subsection, we introduce the concepts of SC and SIC in the context of downlink communications.

At the transmitter side, the coded signals intended for different users are superimposed with different power levels.
Without loss of generality, the channel gains of users are assumed to be w.r.t. a particular ordering.
In the literature\cite{Ding2014, WeiSurvey2016, Liu2016}, the user with a better channel quality is usually called a \emph{strong user}, while the user with a worse channel quality is called a \emph{weak user}.
The transmit powers for the strong and weak users are allocated according to their channel gain order.
To provide fairness and to facilitate the SIC decoding, the transmitter usually allocate more power to the weak user with a poor channel condition\footnote{However, allocating a higher power to the user with the worse channel is not necessarily required in NOMA, as shown in \cite{Vaezi2018non}, especially when there is an explicit minimum data rate requirement for each user.}.

At the receiver, SIC decoding is employed to exploit the heterogeneity in channel gains and transmit powers.
For an SIC receiver, it first decodes other users' signals one by one based on a decoding order before decoding its own signal.
The decoding order depends on the order of the received signal powers, where the user with a higher received power is decoded first during SIC decoding.
Upon finishing decoding one user's signal, the receiver subtracts it from its received signal.
As a result, the interference can be successively removed and the achievable data rate is improved.
For downlink communications, users with better channel conditions can perform SIC to mitigate the
inter-user interference (IUI).
Due to its advantages, SIC has been employed in practical systems such as CDMA \cite{PatelCDMA} and vertical-bell laboratories layered space-time (V-BLAST) \cite{WolnianskyVBLAST}.

Compared to OMA schemes, SC combined with SIC can provide a comparable rate to the strong user, while achieving close to the single-user bound for the weak user.
Intuitively, the strong user, being at a high signal-to-noise ratio (SNR), is degree-of-freedom (DoF) limited and superposition coding allows it to use the full DoF of the channel while being
allocated only a small amount of transmit power, thus causing small amount of interference to the weak user.
In contrast, an orthogonal scheme has to allocate a significant fraction of the degrees of freedom to the weak user to achieve near single-user bound performance, and this causes a large degradation in the performance of the strong user.
In fact, it has been proved that NOMA with SC and SIC is capable of achieving the capacities of general degraded broadcast channels \cite{Cover:2006:EIT:1146355}.

\subsection{Downlink NOMA}

\begin{figure}[t]
\centering
\includegraphics[width=4.5in]{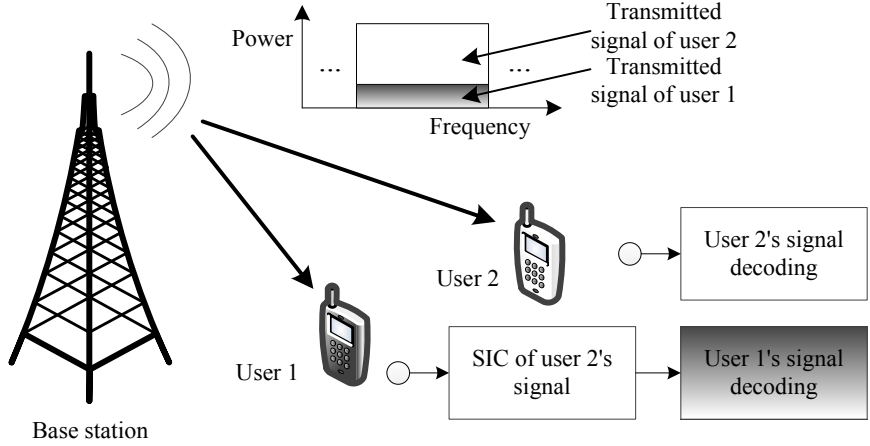}
\caption{The downlink NOMA system model with one BS and two users.}%
\label{C2:DownlinkNOMA}%
\end{figure}

In this chapter, to facilitate the presentation for the basic concepts of NOMA, we consider a simple single-carrier two-user downlink NOMA system.
The generalization to the case of multi-carrier, multi-antenna, and multi-user communications will be presented in the following technical chapters when necessary.
The generic system model for downlink NOMA is illustrated in Figure \ref{C2:DownlinkNOMA} with one BS and two users.
The BS transmits the messages of both user 1 and user 2, i.e., $s_1$ and $s_2$, with different transmit powers $p_1$ and $p_2$, on the same frequency band, respectively.
The corresponding transmitted signal is represented by
\begin{equation}\label{C2:TWO_UE_TRANSMISSION}
  x = \sqrt {{p_1}} {s_1} + \sqrt {{p_2}} {s_2},
\end{equation}
where transmit power is constrained by $p_1+p_2 = 1$.
The received signal at user $k$ is given by
\begin{equation}\label{C2:TWO_UE_RECEIVE1}
  {y_k} = {h_k}x + {n_k}, k \in \{1,2\},
\end{equation}
where $h_k$ denotes the complex channel coefficient between the BS and user $k$, including the joint effect of large-scale fading and small-scale fading. Variable ${n_k}$ denotes the additive white Gaussian noise (AWGN) at user $k$ with a noise power of $\sigma _k^2$, i.e., ${n_k} \sim \mathcal{CN}\left( {0,\;\sigma _k^2} \right)$.
We assume that user 1 is the strong user with a better channel quality, while user 2 is the weak user with a worse channel quality, i.e., $\frac{{{{\left| {{h_1}} \right|}^2}}}{{\sigma _1^2}} \ge \frac{{{{\left| {{h_2}} \right|}^2}}}{{\sigma _2^2}}$.

In downlink NOMA systems, the SIC decoding is implemented at the user side.
The optimal SIC decoding order is in the descending order of channel gains normalized by noise.
It means that user 1 decodes ${{s}_2}$ first and removes the IUI of user 2 by subtracting ${{s}_2}$ from the received signal ${y_1}$ before decoding its own message ${{s}_1}$.
On the other hand, user 2 does not perform interference cancelation and directly decodes its own message ${{s}_2}$ with interference from user 1.
As a result, we can easily obtain the following achievable rates \cite{Ding2014}:
\begin{align}
{R_{1,2}} &= {\log _2}\left( {1 + \frac{{{p_2}{{\left| {{h_1}} \right|}^2}}}{{{p_1}{{\left| {{h_1}} \right|}^2} + {\sigma_1 ^2}}}} \right), \label{C2:AchievableRateDownlink1}\\
{R_{1}} &= {\log _2}\left( {1 + \frac{{{p_1}{{\left| {{h_1}} \right|}^2}}}{{{\sigma_1 ^2}}}} \right),\; \text{and} \label{C2:AchievableRateDownlink2}\\
{R_{2}} &= {\log _2}\left( {1 + \frac{{{p_2}{{\left| {{h_2}} \right|}^2}}}{{{p_1}{{\left| {{h_2}} \right|}^2} + {\sigma_2^2}}}} \right),\label{C2:AchievableRateDownlink3}
\end{align}
where ${R_{1,2}}$ denotes the achievable rate at user 1 to decode the message of user 2, ${R_{1}}$ denotes the achievable rate for user 1 to decode its own message after decoding and subtracting the signal of user 2, and ${R_{2}}$ denotes the achievable rate for user 2 to decode its own message.
Note that the necessary condition for achieving the data rate ${R_{1}}$ in \eqref{C2:AchievableRateDownlink2} is
\begin{equation}\label{C2:SICDecodingCondition}
{R_{1,2}} \ge {R_{2}},
\end{equation}
which implies the possibility of successful interference cancelation at user 1.
Furthermore, we can observe that for any arbitrary power allocation $p_1$ and $p_2$, the condition in \eqref{C2:SICDecodingCondition} can always be satisfied owing to $\frac{{{{\left| {{h_1}} \right|}^2}}}{{\sigma _1^2}} \ge \frac{{{{\left| {{h_2}} \right|}^2}}}{{\sigma _2^2}}$.

\begin{figure}[t]
\centering
\includegraphics[width=4.5in]{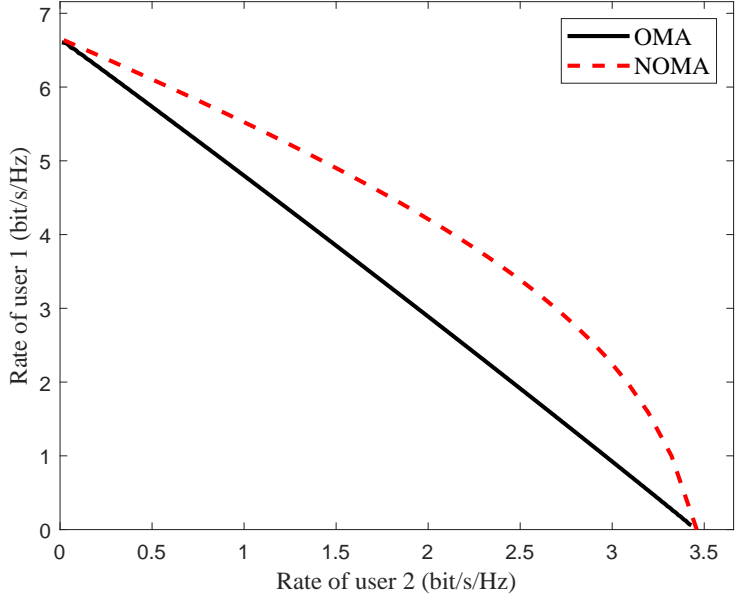}
\caption{The multi-user achievable rate region for downlink NOMA with one BS and two users. User 1 is the strong user with $\frac{{{{\left| {{h_1}} \right|}^2}}}{{\sigma _1^2}} = 100$, while user 2 is the weak user with $\frac{{{{\left| {{h_2}} \right|}^2}}}{{\sigma _2^2}} = 10$.}%
\label{C2:NOMA_over_OMA}%
\end{figure}

Note that SIC is not able to eliminate the interference caused by user 1 for user 2.
Fortunately, if the power allocated to user 2 is larger than that to user 1 in the aggregate received signal ${y_2}$, it does not introduce much performance degradation compared to allocating user 2 on this frequency band exclusively.
The achievable rate region of downlink NOMA is illustrated in Figure \ref{C2:NOMA_over_OMA} in comparison with that of OMA.
We can observe that the achievable rate region of OMA is only a subset of that of NOMA.
As a result, NOMA provides a higher flexibility in resource allocation for improving the system spectral efficiency, especially considering the diverse QoS requirements of users.

\subsection{Uplink NOMA}

\begin{figure}[t]
\centering
\includegraphics[width=4.5in]{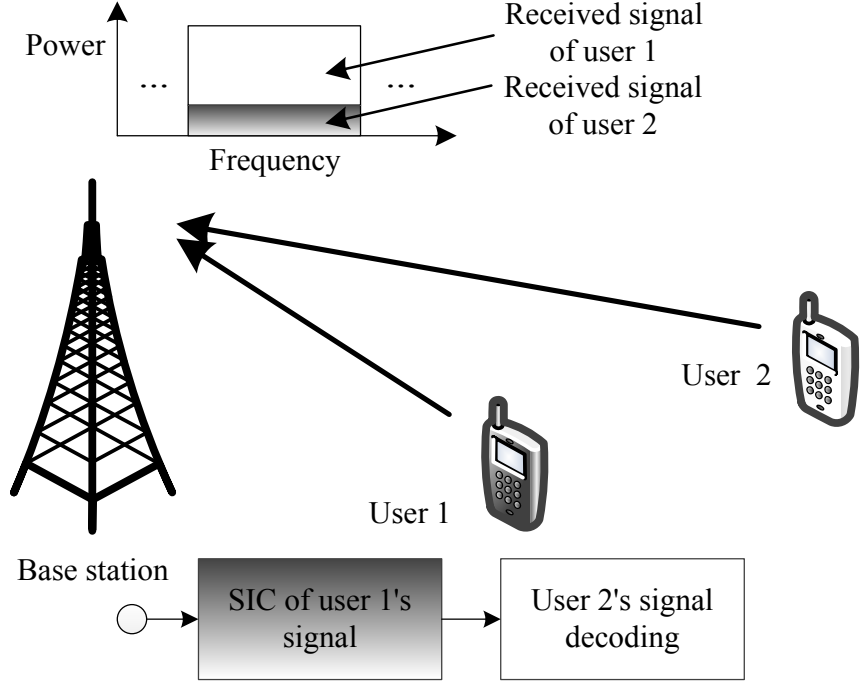}
\caption{The uplink NOMA system model with one BS and two users.}%
\label{C2:UplinkNOMA}%
\end{figure}

The generic system model for uplink NOMA is illustrated in Figure \ref{C2:UplinkNOMA} with one BS and two users.
Both users are transmitting their messages within the same frequency band with the same transmit power $p$.
The received signal at the BS is given by
\begin{equation}\label{C2:SystemModelUplinkNOMA}
  y = \sqrt{p} h_1 s_1 + \sqrt{p} h_2 s_2 + n,
\end{equation}
where $h_k \in \mathbb{C}$ denotes the channel coefficient between the BS and user $k$, $s_k$ denotes the modulated symbol of user $k$, i.e., $k \in \{1,2\}$, and $n\sim{\cal CN}(0,\sigma^2)$ denotes the AWGN at the BS.
Without loss of generality, we assume that user 1 is the strong user and user 2 is the weak user, i.e.,$\left|h_1\right|^2 \ge \left|h_2\right|^2$.
Due to their distinctive channel gains, the received signal power of user 1 is higher than that of the user 2, as shown in Figure \ref{C2:UplinkNOMA}.

In uplink NOMA, the SIC decoding is performed at the BS to retrieve the messages $s_1$ and $s_2$ from the superimposed signal $y$.
In particular, the BS first decodes the message of user 1 and then subtracts $s_1$ from the superimposed signal $y$.
Then, the BS can decode the message of user 2 without IUI.
As a result, we can obtain the individual data rates for user 1 and user 2 as follows:
\begin{align}\label{C2:UplinkNOMARates}
  R_{1} &= {\log _2}\left( 1 + \frac{p \left| {h_1} \right|^2}{ p {\left| {h_2} \right|}^2 + \sigma ^2}\right) \;\text{and} \\
  R_{2} &= {\log _2}\left( 1 + \frac{p \left| {h_2} \right|^2}{ \sigma ^2}\right)
\end{align}

We note that, with the same transmit power, the optimal SIC decoding order of uplink NOMA is the ascending order of channel gains, which is opposite to that of downlink NOMA.
The same transmit power constraint for all uplink NOMA users can be relaxed if employing a closed-loop uplink power control to satisfy the diverse quality-of-service (QoS) requirement of each user.
In addition, compared to downlink NOMA, uplink NOMA requires a tight time synchronization among users and a closed-loop power control, which both occur a system overhead.
However, uplink NOMA is appealing as advanced signal detection/decoding algorithms, e.g. SIC and minimum mean square error SIC (MMSE-SIC), are more affordable at the BS rather than at the user equipments.

\begin{figure}[t]
\centering
\includegraphics[width=4.5in]{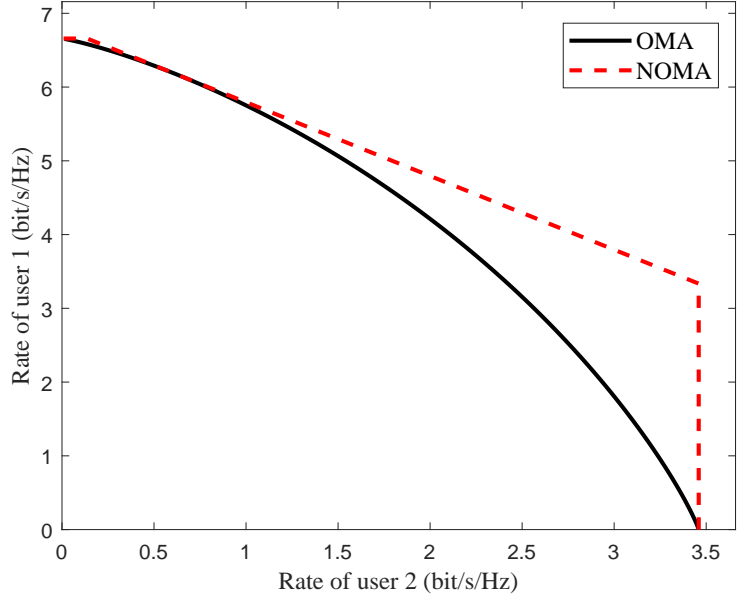}
\caption{The multi-user achievable rate region for uplink NOMA with one BS and two users. User 1 is the strong user with $\frac{{{{\left| {{h_1}} \right|}^2}}}{{\sigma ^2}} = 100$, while user 2 is the weak user with $\frac{{{{\left| {{h_2}} \right|}^2}}}{{\sigma ^2}} = 10$.}%
\label{C2:NOMA_over_OMA_Uplink}%
\end{figure}

The achievable rate region of uplink NOMA is illustrated in Figure \ref{C2:NOMA_over_OMA_Uplink} in comparison with that of OMA.
Similar to downlink NOMA, the achievable rate region of OMA is also a subset of that of NOMA, which provides a higher flexibility in resource allocation design.

\section{Spectral Efficiency and Energy Efficiency}
In this section, we introduce the basic definitions of spectral efficiency (SE) and energy efficiency (EE) adopted in this thesis.
The fundamental trade-off between SE and EE is then discussed.
\subsection{Spectral Efficiency}
Improving the spectral efficiency has been a centric and continuous research topic in the filed of wireless communications for decades, due to the scarce and expensive spectrum resource \cite{Kwan_AF_2010,DerrickOFDMARelay,li2018joint}.
The SE quantifies how many information bits can be delivered within a unit of time and system bandwidth (bit/s/Hz).
Therefore, based on the celebrated formula attained by Shannon\cite{Shannon1948mathematical}, the SE is defined as
\begin{equation}\label{C2:SE}
  \mathrm{SE} = {{{\log }_2}\left( {1 + \frac{p \left|h\right|^2}{\sigma^2}} \right)} (\mathrm{bit/s/Hz}),
\end{equation}
where $\sigma^2$ denotes the noise power, $p$ represents the transmit power, and $\left|h\right|^2$ denotes the channel gain.

\subsection{Energy Efficiency}
The energy efficiency has emerged as a new prominent and fundamental figure of merit for wireless communication systems as the energy consumptions and related environment problems become a major issue currently.
In general, EE is essentially in the form of a benefit-cost ratio\footnote{We need to note that there are alternative types of EE definitions such as from facility level, equipment level, and network level, respectively \cite{HasanGreen}, depending on the design of particular systems.} to evaluate the amount of data delivered by utilizing the limited system energy resource (bits/Joule).
The EE is defined as\cite{DerrickEEOFDMA,DerrickEERobust,DerrickEESWIPT}
\begin{equation}\label{C2:EE}
\mathrm{EE} = {\frac{W\cdot\mathrm{SE}}{{\delta {p} + {P_{{\mathrm{C}}}}}}},
\end{equation}
where ${P_{\mathrm{C}}}$ denotes the static circuit power consumption associated with communications and $\delta>1$ captures the inefficiency of the transmit power amplifier\footnote{Here, we assume that the power amplifier operates in its linear region and the hardware power consumption ${P_{\mathrm{C}}}$ is a constant.}.

\begin{figure}[ptb]
	\centering
	\subfigure[EE (bit/Joule) versus transmit power $p$ (dBm).]
	{\label{C2:EEVsPower} %% label for first subfigure
		\includegraphics[width=4.5in]{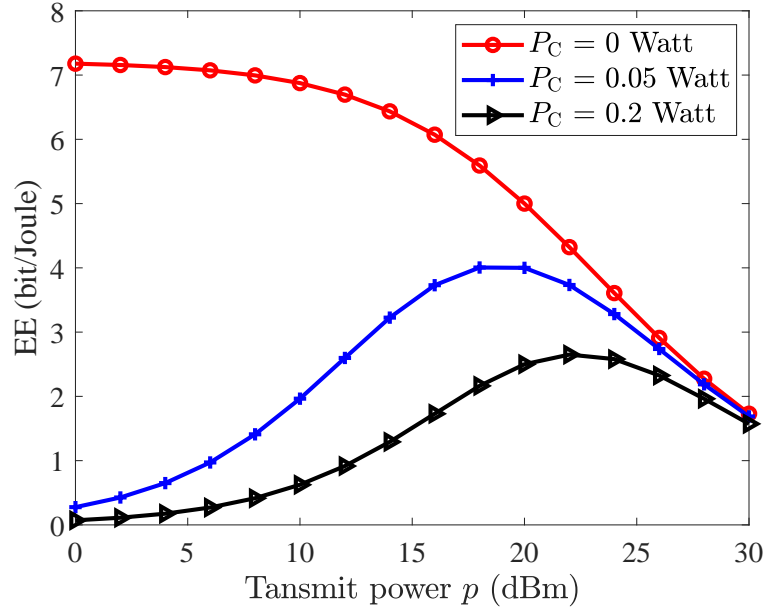}}
	\subfigure[EE (bit/Joule) versus SE (bit/s/Hz).]
	{\label{C2:EEVsSE} %% label for second subfigure
		\includegraphics[width=4.5in]{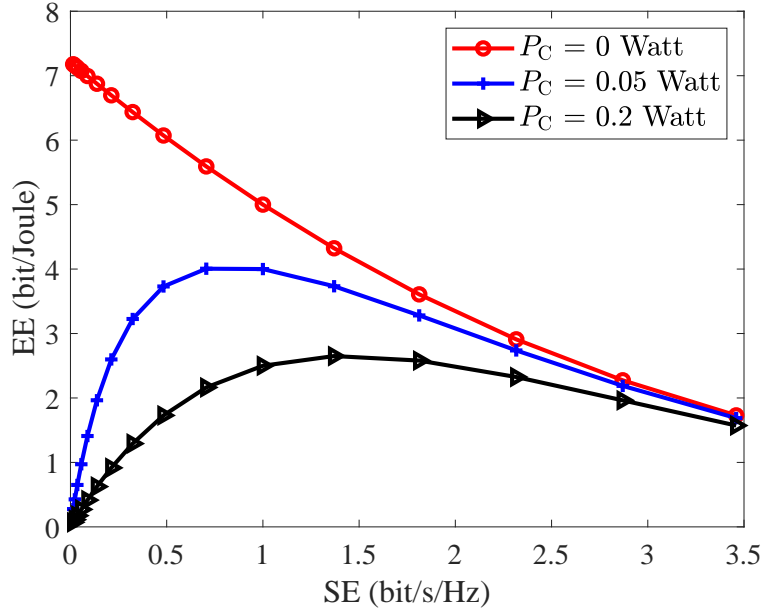}}
	\caption{An illustration of the trade-off between EE, transmit power, and SE. The simulation setups are $W = 1$ Hz,  $\frac{\left|h\right|^2}{N_0} = 10$ dB, $P_{\mathrm{C}} = [0, 0.05, 0.2]$ Watt, and $\delta = 2$.}
	\label{Tradeoff}
\end{figure}

\subsection{The Trade-off Between Energy Efficiency and Spectral Efficiency}
Based on \eqref{C2:SE}, it is clear that the SE monotonically increases with an increased transmit power ${p}$, but with a diminishing return due to the logarithmic nature of the achievable rate function.
On the other hand, the denominator of \eqref{C2:EE} is a linear function of $p$.
As a result, there is a non-trivial trade-off between the EE and SE which should be taken into account for resource allocation algorithm design in wireless communication systems.

Figure \ref{Tradeoff} illustrates the trade-off between EE, transmit power, and SE.
When the total circuit power consumption is negligibly small, i.e., $P_{\mathrm{C}} = 0$ Watt, the EE is a monotonically decreasing function of both the transmit power and the SE, as illustrated in Figure \ref{C2:EEVsPower} and Figure \ref{C2:EEVsSE}, respectively.
In other words, transmission with an arbitrarily low power, i.e., $p \to 0$, is the optimal operation point for maximizing the system EE and the resulting system EE is $\mathop {\lim }\limits_{p \to 0,P_{\mathrm{C}} = 0} \mathrm{EE}= \frac{ \left|h\right|^2}{\delta N_0}$.
In addition, we can observe that when $P_{\mathrm{C}} > 0$, the system EE first increases with the transmit power and then decreases with it.
In fact, in the low SNR regime, the EE is mainly limited by the fixed circuit power consumption ${P_{{\mathrm{C}}}}$ and SE scales almost linearly w.r.t. the transmit power $p$.
Hence, increasing the transmit power can effectively increase both the SE and the EE.
On the other hand, in the high SNR regime, the transmit power, $p$, dominates the total power consumption and there is only a marginal gain in SE when increasing the transmit power.
As a result, after reaching the maximum system EE, as shown in Figure \ref{C2:EEVsPower}, further increasing the transmit power decreases the system EE.
Furthermore, in Figure \ref{C2:EEVsPower}, we can observe that with increasing the circuit power consumption, the optimal operation point is pushed towards the high SNR regime.
It is due to the fact that the larger the circuit power consumption, the higher transmit power is needed to outweigh the impact of the circuit power consumption on the EE.
As a result, for a practical case of $P_{\mathrm{C}} > 0$, there is always a non-trivial trade-off between EE and SE.
Hence, finding the optimal operation point to maximize the system EE has attracted significant attention in the literature in the past few years\cite{DerrickLimitedBackhaul,NgDUALII}.

\section{Optimization and Resource Allocation Designs}

\begin{figure}[ptb]
\centering
\includegraphics[width=4.5in]{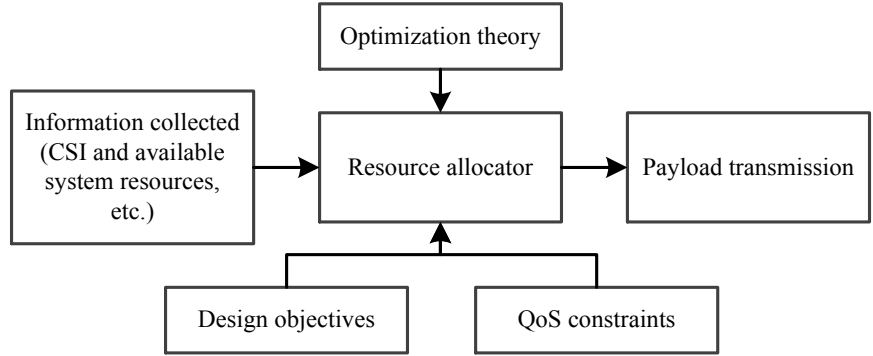}
\caption{Optimization-based resource allocation design framework.}%
\label{C2:ResourceAllocationDesign}%
\end{figure}

In wireless communications, resource allocation is the concept of making the best use of limited communication resources based on the information available at the resource allocator to improve the system performance, as shown in Figure \ref{C2:ResourceAllocationDesign}.
In general, the system resources are the transmit power, the available bandwidth and time, as well as the available space if multiple antennas are employed at terminals.
The available information at the resource allocator usually includes channel state information (CSI) and the available system resources.
In particular, the CSI can be obtained from user feedback in frequency division duplex (FDD) systems or from uplink channel estimation in time division duplex (TDD) systems.
Besides, the QoS requirements, such as the minimum data rate requirement and outage probability requirement, act as constraints in resource allocation optimization framework.
The design objectives are usually maximizing the system sum-rate\cite{Kwan_AF_2010,DerrickOFDMARelay,li2018joint}, maximizing the system energy efficiency\cite{QingqingEE,WuEESE2018,DerrickEEOFDMA,DerrickEESWIPT,DerrickEERobust,DerrickLimitedBackhaul,NgDUALII,cai2019energy}, or minimizing the system power consumption\cite{LengPowerEFFICIENT,DerrickPowerEFFCIEINT2015,DerrickFD2016,Wei2016NOMA,WeiTCOM2017,Sun2016MultipleObj,Boshkovska2018}.
Specifically, resource allocation designs rely on the application of the optimization theory to optimize the system performance taking into account various QoS constraints.
Finally, the transmitter will transmit the payload according to its obtained results of resource allocation designs, i.e., power allocation, user scheduling, and rate allocation, etc..

In the next section, some classic problem formulations used for resource allocation designs in wireless communications, including the commonly adopted design objectives and QoS constraints, are presented.
Optimization theories, including convex \cite{Boyd2004} and non-convex optimization\cite{horst2013global}, are the key tools being used in resource allocation designs.
Interested readers are referred to \cite{Boyd2004,horst2013global} for more details.
\subsection{Design Objectives}
In the following, we briefly introduce three kinds of design objectives commonly used in the literature.

\noindent\underline{\textbf{Weighted System Sum-rate Maximization}}

Considering a communication system with $K$ users, the objective function (utility function) for maximizing the weighted system sum-rate is given by
\begin{equation}\label{C2:Objective1}
  {\mathcal{U}_1}\left( \mathbf{x} \right) = \sum_{k=1}^{K} \omega_k R_k \left( \mathbf{x} \right),
\end{equation}
where $\mathbf{x}$ denotes the optimization variable and it may include the variables for power allocation, user scheduling, and rate allocation, etc..
The individual rate of user $k$, $R_k \left( \mathbf{x} \right)$, is a function of the optimization variables.
The weights $\omega_k \ge 0$ with $\sum_{k=1}^{K} \omega_k  = 1$ are non-negative constants provided by the upper layers, which allow the resource allocator to give different priorities to different users and to enforce certain
notions of fairness.
When $\omega_1 = \omega_2 =,\ldots, = \omega_K = \frac{1}{K}$, it degenerates to the design objective of maximizing the system sum-rate which does not take into account fairness in resource allocation design.

\noindent\underline{\textbf{System Power Consumption Minimization}}

If the resource allocation design is applied to a communication system which aims to minimize the system radiation power consumption, the design objective is given by
\begin{equation}\label{C2:Objective2}
  {\mathcal{U}_2}\left( \mathbf{x} \right) = \sum_{k=1}^{K} {{p_k}}\left( \mathbf{x} \right),
\end{equation}
where ${{p_k}}\left( \mathbf{x} \right)$ represents the transmit power for user $k$ which is a function of other optimization variables, such as power allocation for other users, user scheduling, and rate allocation, etc..

\noindent\underline{\textbf{System Energy Efficiency Maximization}}

When the resource allocation is designed to maximize the system energy efficiency, the objective function can be expressed as
\begin{equation}\label{C2:Objective3}
  {\mathcal{U}_3}\left( \mathbf{x} \right) = \frac{\sum_{k=1}^{K} \omega_k R_k \left( \mathbf{x} \right)}{{ \sum\limits_{k = 1}^K \left(\delta{{p_k}}\left( \mathbf{x} \right)  + {P_{\mathrm{C},k}}\right)}},
\end{equation}
where the numerator is actually the weighted system sum-rate ${\mathcal{U}_1}\left( \mathbf{x} \right)$.
In the denominator, ${P_{\mathrm{C},k}}$, denotes the static circuit power consumption associated with user $k$.

\subsection{QoS Constraints}
To satisfy diverse QoS requirements for different applications, different types of QoS constraints can be incorporated in the problem formulation of resource allocation designs.
In fact, QoS constraints combined with the system resource limitations usually span the feasible solution set $\mathcal{X}$ for the optimization variable, i.e., $\mathbf{x} \in \mathcal{X}$.
In general, two kinds of QoS constraints commonly adopted in the literature are introduced in the following.

\noindent\underline{\textbf{Minimum Data Rate Requirement}}

As the name implies, the minimum data rate requirement is the constant minimum data rate $R_k^{\mathrm{min}}$ to support the application of communication user $k$, i.e.,
\begin{equation}\label{C2:QoSConstrain}
  R_k \left( \mathbf{x} \right) \ge R_k^{\mathrm{min}},
\end{equation}
where $R_k^{\mathrm{min}}$ is usually preset and obtained during the information collection phase in Figure \ref{C2:ResourceAllocationDesign}.
The minimum data rate requirement is imposed for resource allocation design to guarantee the QoS of each user.
In particular, for the design objective in \eqref{C2:Objective1} with $\omega_1 = \omega_2 =,\ldots, = \omega_K = \frac{1}{K}$, according the water-filling principle\cite{Tse2005}, it results in the optimal utilization of the system resources from the channel capacity point of view.
However, users with poor channel conditions may suffer from starvation since they are rarely selected for transmission.
Hence, introducing a minimum data rate requirement can effectively balance the system performance and each user's QoS requirement.

\noindent\underline{\textbf{Outage Probability Requirement}}

For a communication system with imperfect CSI at transmitter side, there exists a non-zero probability that the scheduled data rate exceeds the instantaneous channel capacity.
In this case, even applying powerful error correction coding cannot prevent packet error and thus an outage occurs.
As a result, the outage probability for the communication link of user $k$ can be defined as
\begin{equation}\label{C2:OutagePro}
  \mathrm{P}_{\mathrm{out},k}\left(\mathbf{x}\right) = \mathrm{Pr}\{C_k\left(\mathbf{x}\right) < R_k\left(\mathbf{x}\right)\},
\end{equation}
where $R_k\left(\mathbf{x}\right)$ is the allocated data rate for user $k$ and $C_k\left(\mathbf{x}\right)$ is the channel capacity of the communication link for user $k$.
They both depend on the resource allocation policy and thus the outage probability is a function of the resource allocation variables $\mathbf{x}$.
The outage probability constraint has been employed for resource allocation design to enhance the communication reliability\cite{ZhangStochastic,WeiTCOM2017}.
In particular, the outage probability of user $k$ should be smaller than the maximum tolerable outage probability $\overline{\mathrm{P}}_{\mathrm{out},k}$, i.e.,
\begin{equation}\label{C2:OutageConstraint}
  \mathrm{P}_{\mathrm{out},k}\left(\mathbf{x}\right) \le \overline{\mathrm{P}}_{\mathrm{out},k}.
\end{equation}
This probabilistic constraint takes the CSI imperfectness into consideration and hence is very useful for robust resource allocation in wireless communications.
We note that the robust resource allocation design based on outage probability only needs to know the statistical CSI at the transmitter, rather than the instantaneous CSI.
This makes the outage-constrained resource allocation design more practical since statistical CSI is usually available based on the long term measurements and does not change so fast as the instantaneous CSI.

\section{Summary}
In this chapter, we present the background materials on wireless communications and resource allocation which are closely related to and required by the research work in this thesis.
The main points presented in this chapter are summarized as follows.
\begin{itemize}
\item We briefly introduced the channel characteristics in wireless communications, including large-scale fading and small-scale fading and presented the typical channel models adopted in the later chapters.
\item We introduced some fundamental concepts of NOMA, including SC and SIC, system models for downlink NOMA as well as uplink NOMA, and their multi-user achievable rate regions.
\item We provided some basic knowledge and definitions on spectral efficiency and energy efficiency.
\item We also presented the optimization-based resource allocation design framework, including the commonly adopted design objectives and QoS constraints.
\end{itemize}

\chapter{On the Performance Gain of NOMA over OMA in Uplink Communication Systems}\label{C3:chapter3}

%\nomenclature{$a$}{The number of angels per unit area}%
%\nomenclature{$N$}{The number of angels per needle point}%
%\nomenclature{$A$}{The area of the needle point}%
%
%\ifpdf
%    \graphicspath{{1_introduction/figures/PNG/}{1_introduction/figures/PDF/}{1_introduction/figures/}}
%\else
%    \graphicspath{{1_introduction/figures/EPS/}{1_introduction/figures/}}
%\fi
\section{Introduction}
The massive number of devices and explosive data traffic has posed challenging requirements, such as massive connectivity \cite{SunJUICE} and ultra-high spectral efficiency for future wireless networks\cite{Andrews2014,wong2017key}.
Compared to the conventional orthogonal multiple access (OMA) schemes, non-orthogonal multiple access (NOMA) allows users to simultaneously share the same resource blocks and hence it is beneficial for supporting a large number of connections for spectrally efficient communications.
In this chapter, we analyze the performance of NOMA for furthering the understanding on the performance gain of NOMA over OMA, which serves as building blocks for the specific designs in the next few chapters.
Although the existing treatises have investigated the system performance of NOMA from different perspectives, such as the outage probability \cite{Ding2014,DingSignalAlignment,ChenQuasiDegradation,DingMassive} and the ergodic sum-rate \cite{Ding2014}, in various specifically considered system setups, no unified analysis has been published to discuss the performance gain of NOMA over OMA.
To fill this gap, our work offers a unified analysis on the ergodic sum-rate gain (ESG) of NOMA over OMA in single-antenna, multi-antenna, and massive
antenna array aided systems relying on both single-cell and multi-cell deployments.
{In particular, the ESG of NOMA over OMA is defined as difference between the ergodic sum-rates of NOMA and OMA schemes. Hence, the ESG is an additive gain.}
We first focus our attention on the ESG analysis in single-cell systems and then extend our analytical results to multi-cell systems by taking into account the inter-cell interference (ICI).
We quantify the ESG of NOMA over OMA relying on practical signal reception schemes at the base station for both NOMA as well as OMA and unveil its behavior under different scenarios.
Our simulation results confirm the accuracy of our performance analyses and provide some interesting insights, which are summarized in the following:
\begin{itemize}
  \item In all the cases considered, a high ESG can be achieved by NOMA over OMA in the high-SNR regime, but the ESG vanishes in the low-SNR regime.
  \item In the single-antenna scenario, we identify two types of gains attained by NOMA and characterize their different behaviors.
      In particular, we show that the \emph{large-scale near-far gain} achieved by exploiting the distance-discrepancy between the base station and users increases with the cell size, while the \emph{small-scale fading gain} is given by an Euler-Mascheroni constant\cite{abramowitz1964handbook} of $\gamma = 0.57721$ nat/s/Hz in Rayleigh fading channels.
  \item When applying NOMA in multi-antenna systems, compared to the MIMO-OMA utilizing zero-forcing detection, we analytically show that the ESG of SISO-NOMA over SISO-OMA can be increased by $M$-fold, when the base station is equipped with $M$ antennas and serves a sufficiently large number of users $K$.
  \item Compared to MIMO-OMA utilizing a maximal ratio combining (MRC) detector, an $\left(M-1\right)$-fold degrees of freedom (DoF) gain can be achieved by MIMO-NOMA.
      In particular, the ESG in this case increases linearly with the system's SNR quantified in dB with a slope of $\left(M-1\right)$ in the high-SNR regime.
  \item For massive MIMO systems with a fixed ratio between the number of antennas, $M$, and the number of users, $K$, i.e., $\delta = \frac{M}{K}$, the ESG of \emph{m}MIMO-NOMA over \emph{m}MIMO-OMA increases linearly with both $K$ and $M$ using MRC detection.
  \item In practical multi-cell systems operating without joint cell signal processing, the ESG of NOMA over OMA is degraded due to the existence of ICI, especially for a small cell size with a dense cell deployment.
      Furthermore, no DoF gain can be achieved by NOMA in multi-cell systems due to the lack of joint multi-cell signal processing to handle the ICI.
      In other words, all the ESGs of NOMA over OMA in single-antenna, multi-antenna, and massive MIMO multi-cell systems saturate in the high-SNR regime.
  \item For both single-cell and multi-cell systems, a large cell size is always beneficial to the performance of NOMA.
      In particular, in single-cell systems, the ESG of NOMA over OMA is increased for a larger cell size due to the enhanced large-scale near-far gain.
      For multi-cell systems, a larger cell size reduces the ICI level, which prevents a severe ESG degradation.
\end{itemize}

\section{System Model}
\begin{figure}[t]
\centering
\includegraphics[width=3in]{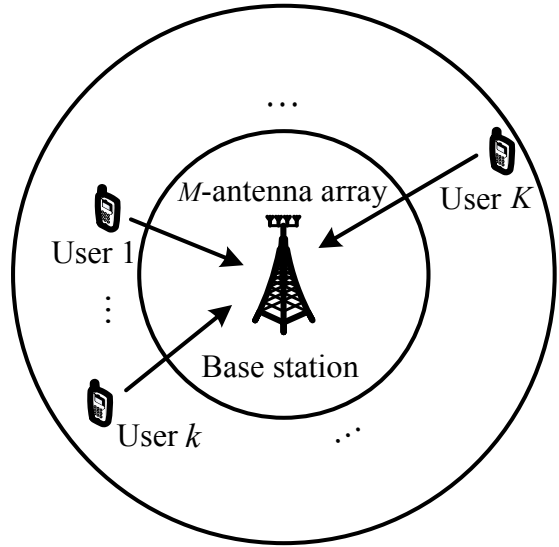}
\caption{The system model of the single-cell uplink communication with one base station and $K$ users.}
\label{C3:NOMA_Uplink_Model}
\end{figure}
\subsection{System Model}
We first consider the uplink of a single-cell\footnote{We first focus on the ESG analysis for single-cell systems, which serves as a building block for the analyses for multi-cell systems presented in Section \ref{C3:DiscussionsMulticell}.} NOMA system with a single base station (BS) supporting $K$ users, as shown in Figure \ref{C3:NOMA_Uplink_Model}.
The cell is modeled by a ring-shaped disc.
{The BS is located at the center of the ring-shaped disc with the inner radius of $D_0$ and outer radius of $D$, where all the $K$ users are scattered uniformly within the ring-shaped disc.}
%
%Note that, the circle cell deployment is an approximated model for the practical hexagon cellular system.
%%
%However, it can simplify our problem a lot and the conclusions obtained with circle cell deployment are approximatively applicable to the hexagon cellular systems.
%
For the NOMA scheme, all the $K$ users are multiplexed on the same frequency band and time slot, while for the OMA scheme, $K$ users utilize the frequency or time resources orthogonally.
Without loss of generality, we consider frequency division multiple access (FDMA) as a typical OMA scheme.

In this chapter, we consider three typical types of communication systems:
\begin{itemize}
  \item SISO-NOMA and SISO-OMA: the BS is equipped with a single-antenna ($M=1$) and all the $K$ users also have a single-antenna.
  \item MIMO-NOMA and MIMO-OMA: the BS is equipped with a multi-antenna array ($M>1$) and all the $K$ users have a single-antenna associated with $K > M$.
  \item Massive MIMO-NOMA (\emph{m}MIMO-NOMA) and massive MIMO-OMA \\(\emph{m}MIMO-OMA): the BS is equipped with a large-scale antenna array ($M \to \infty$), while all the $K$ users are equipped with a single antenna, associated with $\frac{M}{K} = \delta < 1$, i.e., {the number of antennas $M$ at the BS is smaller than the number of users $K$}, but with a fixed ratio of $\delta < 1$.
\end{itemize}

%In this chapter, we focus on the case of single-antenna users to simplify the following analytical study.
%%
%We note that it is actually a general system model, which is applicable to the case of multi-antenna users.
%%
%In particular, when a user has multiple antennas, it can transmit multiple data streams via spatial multiplexing, which can be viewed as multiple co-located single-antenna users, each with a single data stream.
%%
%On the other hand, if the multi-antenna user transmits a single data stream via the maximum ratio transmission, it can be treated as a single-antenna user with a higher transmit power.
%%
%Note that, in the second and third cases, we still denote them as MIMO-NOMA/MIMO-OMA and \emph{m}MIMO-NOMA/\emph{m}MIMO-OMA even each user is equipped with a single-antenna.
%%
%This is because there are multiple users uploading their independent data streams to the BS.
\subsection{Signal and Channel Model}
The signal received at the BS is given by
\begin{equation}\label{C3:MIMONOMASystemModel}
{\bf{y}} = \sum\limits_{k = 1}^K {{{\bf{h}}_k}} \sqrt {{p_k}} {x_k} + {\bf{v}},
\end{equation}
where ${\bf{y}}\in \mathbb{C}^{ M \times 1}$, $p_k$ denotes the power transmitted by user $k$, $x_k$ is the normalized modulated symbol of user $k$ with ${\mathrm{E}}\left\{ \left|{x_k}\right|^2 \right\} = 1$, and ${\bf{v}}\sim \mathcal{CN}\left(\mathbf{0},N_0 {{\bf{I}}_M}\right)$ represents the additive white Gaussian noise (AWGN) at the BS with zero mean and covariance matrix of $N_0 {{\bf{I}}_M}$.
To emphasize the impact of the number of users $K$ on the performance gain of NOMA over OMA, we fix the total power consumption of all the uplink users and thus we have
\begin{equation}\label{C3:SumPowerConstraint}
\sum\limits_{k = 1}^K {{p_{k}}}  \le {P_{\mathrm{max}}},
\end{equation}
where ${P_{\mathrm{max}}}$ is the maximum total transmit power for all the users.
Note that the sum-power constraint is a commonly adopted assumption in the literature\cite{Vishwanath2003,WangMUG,Xu2017} for simplifying the performance analysis of uplink communications.
In fact, the sum-power constraint is a reasonable assumption for practical cellular communication systems, where a total transmit power limitation is intentionally imposed to limit the ICI.

The uplink (UL) channel vector between user $k$ and the BS is modeled as
\begin{equation}\label{C3:ChannelModel}
{{{\bf{h}}_k}} = \frac{{\bf{g}}_k}{\sqrt{1+d_k^{\alpha}}},
\end{equation}
where ${\bf{g}}_k \in \mathbb{C}^{ M \times 1}$ denotes the Rayleigh fading coefficients, i.e., ${\bf{g}}_k \sim \mathcal{CN}\left(\mathbf{0},{{\bf{I}}_M}\right)$, $d_k$ is the distance between user $k$ and the BS in the unit of meter, and $\alpha$ represents the path loss exponent\footnote{In this chapter, we ignore the impact of shadowing to simplify our performance analysis.
Note that, shadowing only introduces an additional power factor to ${{\bf{g}}_k}$ in the channel model in \eqref{C3:ChannelModel}.
Although the introduction of shadowing may change the resulting channel distribution of ${{{\bf{h}}_k}}$, the distance-based channel model is sufficient to characterize the large-scale near-far gain exploited by NOMA, as will be discussed in this chapter.}.
We denote the UL channel matrix between all the $K$ users and the BS by ${\bf{H}} = \left[ {{{\bf{h}}_1}, \ldots ,{{\bf{h}}_K}} \right] \in \mathbb{C}^{ M \times K}$.
Note that the system model in \eqref{C3:MIMONOMASystemModel} and the channel model in \eqref{C3:ChannelModel} include the cases of single-antenna and massive MIMO aided BS associated with $M = 1$ and $M \to \infty$, respectively.
For instance, when $M=1$, ${{{{h}}_k}} = \frac{{{g}}_k}{\sqrt{1+d_k^{\alpha}}}$ denotes the corresponding channel coefficient of user $k$ in single-antenna systems.
We assume that the channel coefficients are independent and identically distributed (i.i.d.) over all the users and antennas.
Since this chapter aims for providing some insights concerning the performance gain of NOMA over OMA, we assume that perfect UL CSI knowledge is available at the BS for coherent detection.
%
%Without loss of generality, we assume $\left\| {{{\bf{h}}_1}} \right\| \ge, \ldots, \ge\left\| {{{\bf{h}}_K}} \right\|$, i.e., users are indexed based on their channel power gain.
\subsection{Signal Detection and Resource Allocation Strategy}\label{C3:ResourceAllocation}
To facilitate our performance analyses, we focus our attention on the following efficient signal detection and practical resource allocation strategies.

\noindent\textbf{\underline{Signal Detection}}

\begin{table}
\center\small
\caption{Signal Detection Techniques for NOMA and OMA Systems}
\centering
\begin{tabular}{cc|cc}
  \hline
  NOMA system  & Reception technique & OMA system & Reception technique\\ \hline
  SISO-NOMA  & SIC & SISO-OMA & FDMA-SUD \\
  MIMO-NOMA & MMSE-SIC & MIMO-OMA & FDMA-ZF, FDMA-MRC\\
  \emph{m}MIMO-NOMA & MRC-SIC & \emph{m}MIMO-OMA & FDMA-MRC\\\hline
\end{tabular}\label{C3:TransceiverProtocol}
\end{table}

The signal detection techniques adopted in this chapter for NOMA and OMA systems are shown in Table \ref{C3:TransceiverProtocol}, which are detailed in the following.

For SISO-NOMA, we adopt the commonly used successive interference cancelation (SIC) receiver \cite{wei2017performance} at the BS, since its performance approaches the capacity of single-antenna systems\cite{Tse2005}.
On the other hand, given that all the users are separated orthogonally by different frequency subbands for SISO-OMA, the simple single-user detection (SUD) technique can be used to achieve the optimal performance.

For MIMO-NOMA, the minimum mean square error criterion based successive interference cancelation (MMSE-SIC) constitutes an appealing receiver algorithm, since its performance approaches the capacity \cite{Tse2005} at an acceptable computational complexity for a finite number of antennas $M$ at the BS.
On the other hand, two types of signal detection schemes are considered for MIMO-OMA, namely FDMA zero forcing (FDMA-ZF) and FDMA maximal ratio combining (FDMA-MRC).
Exploiting the extra spatial degrees of freedom (DoF) attained by multiple antennas at the BS, ZF can be used for multi-user detection (MUD), as its achievable rate approaches the capacity in the high-SNR regime\cite{Tse2005}.
In particular, all the users are categorized into $G = K/M$ groups\footnote{Without loss of generality, we consider the case with $G$ as an integer in this chapter.} with each group containing $M$ users.
Then, ZF is utilized for handling the inter-user interference (IUI) within each group and FDMA is employed to separate all the $G$ groups on orthogonal frequency subbands.
In the low-SNR regime, the performance of ZF fails to approach the capacity \cite{Tse2005}, thus a simple low-complexity MRC scheme is adopted for single user detection on each frequency subband.
We note that there is only a single user in each frequency subband of our considered FDMA-MRC aided MIMO-OMA systems, i.e., no user grouping.

With a massive number of UL receiving antennas employed at the BS, we circumvent the excessive complexity of matrix inversion involved in ZF and MMSE detection by adopting the low-complexity MRC-SIC detection \cite{WeiLetter2018} for \emph{m}MIMO-NOMA systems and the FDMA-MRC scheme for \emph{m}MIMO-OMA systems.
Given the favorable propagation property of massive MIMO systems\cite{Ngo2013}, the orthogonality among the channel vectors of multiple users holds fairly well, provided that the number of users is sufficiently lower than the number of antennas.
Therefore, we can assign $W \ll M$ users to every frequency subband and perform the simple MRC detection while enjoying negligible IUIs in each subband.
In this chapter, we consider a fixed ratio between the group size and the number of antennas, namely, $\varsigma = \frac{W}{M} \ll 1$, and assume that the above-mentioned favorable propagation property holds under the fixed ratio $\varsigma$ considered.

\noindent\textbf{\underline{Resource Allocation Strategy}}

To facilitate our analytical study in this chapter, we consider an equal resource allocation strategy for both NOMA and OMA schemes.
In particular, equal power allocation is adopted for NOMA schemes\footnote{As shown in \cite{Vaezi2018non}, allocating a higher power to the user with the worse channel is not necessarily required in NOMA\cite{Vaezi2018non}.}.
On the other hand, equal power and frequency allocation is adopted for OMA schemes.
%%
%Note that, we have the constraints $\sum\limits_{k = 1}^K {{q_{k}}}  \le {P_{\mathrm{max}}}$ for the power allocation and $\sum\limits_{k = 1}^K {{f_{k}}}  \le 1$ or $\sum\limits_{g = 1}^G {{f_{g}}}  \le 1$ for the frequency allocation.
%
Note that the equal resource allocation is a typical selected strategy for applications bearing only a limited system overhead, e.g. machine-type communications (MTC).
%
%With equal resource allocation strategy, we can investigate the inherent source of performance gain of NOMA over OMA without the help of the resource allocation design.

We note that beneficial user grouping design is important for the MIMO-OMA system relying on FDMA-ZF and for the \emph{m}MIMO-OMA system using FDMA-MRC.
In general, finding the optimal user grouping strategy is an NP-hard problem and the performance analysis based on the optimal user grouping strategy is generally intractable.
Furthermore, the optimal SIC decoding order of NOMA in multi-antenna and massive MIMO systems is still an open problem in the literature, since the channel gains on different antennas are usually diverse.
To avoid tedious comparison and to facilitate our performance analysis, we adopt a random user grouping strategy for the OMA systems considered and a fixed SIC decoding order for the NOMA systems investigated.
%
%For FDMA-ZF scheme in the MIMO-OMA systems, consecutive $M$ users are grouped together, i.e., the $g$-th group consists of $\mathcal{S}_g = \left\{(g-1)M+1:gM\right\}$ users with their channel matrix as ${\bf{H}}_g = \left[ {{{\bf{h}}_{(g-1)M+1}}, \ldots ,{{\bf{h}}_{gM}}} \right] \in \mathbb{C}^{ M \times M}$.
%
In particular, we randomly select $M$ and $W$ users for each group on each frequency subband for the MIMO-OMA and \emph{m}MIMO-OMA systems, respectively.
For NOMA systems, without loss of generality, we assume $\left\| {{{\bf{h}}_1}} \right\| \ge \left\| {{{\bf{h}}_2}} \right\|, \ldots, \ge\left\| {{{\bf{h}}_K}} \right\|$, that the users are indexed based on their channel gains, and the SIC/MMSE-SIC/MRC-SIC decoding order\footnote{Note that, in general, the adopted decoding order is not the optimal SIC decoding order for maximizing the achievable sum-rate of the considered MIMO-NOMA and \emph{m}MIMO-NOMA systems.} at the BS is $1,2,\ldots,K$.
%
%On the other hand, for fair comparison, we consider a random SIC and MMSE-SIC decoding order at the BS for SISO-NOMA, MIMO-NOMA, and \emph{m}MIMO-NOMA, respectively\footnote{Note that the SIC and MMSE-SIC decoding order at the BS do not affect the system sum-rate in uplink SISO-NOMA systems\cite{WeiLetter2018} and uplink MIMO-NOMA systems\cite{Tse2005}, respectively.
%%
%Although the MRC-SIC decoding order at the BS determines the system sum-rate for uplink \emph{m}MIMO-NOMA systems, we focus on the average system performance with a random SIC decoding order to facilitate the performance analysis.}.
%
Additionally, to unveil insights about the performance gain of NOMA over OMA, we assume that there is no error propagation during SIC/MMSE-SIC/MRC-SIC decoding at the BS.
\section{ESG of SISO-NOMA over SISO-OMA}
In this section, we first derive the ergodic sum-rate of SISO-NOMA and SISO-OMA.
Then, the asymptotic ESG of SISO-NOMA over SISO-OMA is discussed under different scenarios.

\subsection{Ergodic Sum-rate of SISO-NOMA and SISO-OMA}
When decoding the messages of user $k$, the interferences imposed by users $1,2,\ldots,(k-1)$ have been canceled in the SISO-NOMA system by SIC reception.
Therefore, the instantaneous achievable data rate of user $k$ in the SISO-NOMA system considered is given by:
\begin{equation}\label{C3:SISONOMAIndividualAchievableRate}
R_{k}^{\mathrm{SISO-NOMA}} = {\ln}\left( 1 + \frac{{{p_k}{{\left| {{h_k}} \right|}^2}}}{{\sum\limits_{i = k + 1}^K {{p_i}{{\left| {{h_i}} \right|}^2}}  + {N_0}}} \right).
\end{equation}
On the other hand, in the SISO-OMA system considered, user $k$ is allocated to a subband exclusively, thus there is no inter-user interference (IUI).
As a result, the instantaneous achievable data rate of user $k$ in the SISO-OMA system considered is given by:
\begin{equation}\label{C3:SISOOMAIndividualAchievableRate}
R_{k}^{{\mathrm{SISO-OMA}}} = f_k {\ln}\left(1+ {\frac{{{p_{k }}{{\left| {{{{h}}_k}} \right|}^2}}}{{f_k N_0}}} \right),
\end{equation}
with ${p_k}$ and $f_k$ denoting the power allocation and frequency allocation of user $k$.
Note that we consider a normalized frequency bandwidth for both the NOMA and OMA schemes in this chapter, i.e., $\sum \limits_{k=1}^{K} f_k = 1$.
Under the identical resource allocation strategy, i.e., for ${{p_{k}}} = \frac{{P_{\mathrm{max}}}}{K}$ and $f_k = 1/K$, we have the instantaneous sum-rate of SISO-NOMA and SISO-OMA given by
\begin{align}
R_{\mathrm{sum}}^{{\mathrm{SISO-NOMA}}} &= \sum\limits_{k = 1}^K R_{k}^{\mathrm{SISO-NOMA}}=  {\ln}\left( 1 + \frac{{P_{\mathrm{max}}}}{K {N_0}} \sum\limits_{k = 1}^K {{\left| {{{h}_k}} \right|}^2} \right) \;\text{and} \label{C3:InstantSumRateSISONOMA}\\
R_{\mathrm{sum}}^{{\mathrm{SISO-OMA}}} &= \sum\limits_{k = 1}^K R_{k}^{\mathrm{SISO-OMA}}= \frac{1}{K}\sum\limits_{k = 1}^K {\ln}\left( 1 + \frac{P_{\mathrm{max}}}{{N_0}} {{\left| {{{h}_k}} \right|}^2} \right), \label{C3:InstantSumRateSISOOMA}
\end{align}
respectively.
%
%It can be observed that in the low-SNR regime with ${P_{\mathrm{max}}} \to 0$, the performance gain of SISO-NOMA over SISO-OMA vanishes.

Given the instantaneous sum-rates in \eqref{C3:InstantSumRateSISONOMA} and \eqref{C3:InstantSumRateSISOOMA}, firstly we have to investigate the channel gain distribution before embarking on the derivation of the corresponding ergodic sum-rates.
Since all the users are scattered uniformly across the pair of concentric rings between the inner radius of $D_0$ and the outer radius of $D$ in Figure \ref{C3:NOMA_Uplink_Model}, the cumulative distribution function (CDF) of the channel gain\footnote{As mentioned before, we assumed that the channel gains of all the users are ordered as $\left| {{{{h}}_1}} \right| \ge \left| {{{{h}}_2}} \right|, \ldots, \ge\left| {{{{h}}_K}} \right|$ in Section \ref{C3:ResourceAllocation}.
However, as shown in \eqref{C3:InstantSumRateSISONOMA}, the system sum-rate for the considered SISO-NOMA system is actually independent of the SIC decoding order.
Therefore, we can safely assume that all the users have i.i.d. channel distribution, which does not affect the performance analysis results.
In the sequel of this chapter, the subscript $k$ is dropped without loss of generality.} ${{\left| {{h}} \right|}^2}$ is given by
\begin{equation}\label{C3:ChannelDistributionCDF}
{F_{{{\left| {{h}} \right|}^2}}}\left( x \right) = \int_{D_0}^D {\left( {1 - {e^{ - \left( {1 + {z^\alpha }} \right)x}}} \right)} {f_{{d}}}\left( z \right)dz,
\end{equation}
where ${f_{{d}}}\left( z \right) = \frac{2z}{D^2 - D_0^2}$, $D_0 \le z \le D$, denotes the probability density function (PDF) for the random distance $d$.
%%
%Despite that the channel gains has been ranked as $\left| {{{{h}}_1}} \right| \ge, \ldots, \ge\left| {{{{h}}_K}} \right|$, we can still assume an independent and identically distributed (i.i.d.) channel distribution for all the $K$ users in \eqref{C3:ChannelDistributionCDF}.
%%
%It is because that the system sum-rates of SISO-NOMA and SISO-OMA in \eqref{C3:InstantSumRate} are irrelevant of the channel gains' rank.
%%
%In the following, we will omit the subscript $k$ for the sake of notation simplicity.
%
With the Gaussian-Chebyshev quadrature approximation\cite{abramowitz1964handbook}, the CDF and PDF of ${{\left| {{h}} \right|}^2}$ can be approximated by
\begin{align}
{F_{{{\left| {{h}} \right|}^2}}}\left( x \right) &\approx 1 - \frac{1}{D+D_0}\sum\limits_{n = 1}^N {{\beta _n}{e^{ - {c_n}x}}} \; \text{and} \label{C3:SISOChannelDistributionPDF}\\
{f_{{{\left| {{h}} \right|}^2}}}\left( x \right) &\approx \frac{1}{D+D_0}\sum\limits_{n = 1}^N {{\beta _n}{c_n}{e^{ - {c_n}x}}}, x \ge 0, \label{C3:SISOChannelDistributionCDF}
\end{align}
respectively, where the parameters in \eqref{C3:SISOChannelDistributionPDF} and \eqref{C3:SISOChannelDistributionCDF} are:
\begin{align}\label{C3:BetaCn}
{\beta_n} &= \frac{\pi }{N}\left| {\sin \frac{{2n \hspace{-1mm}-\hspace{-1mm} 1}}{{2N}}\pi } \right|\left( {\frac{D\hspace{-1mm}-\hspace{-1mm}D_0}{2}\cos \frac{{2n \hspace{-1mm}-\hspace{-1mm} 1}}{{2N}}\pi  + \frac{D\hspace{-1mm}+\hspace{-1mm}D_0}{2}} \right) \;\text{and}\notag\\
{c_n} &= 1 + {\left( {\frac{D\hspace{-1mm}-\hspace{-1mm}D_0}{2}\cos \frac{{2n \hspace{-1mm}-\hspace{-1mm} 1}}{{2N}}\pi  + \frac{D\hspace{-1mm}+\hspace{-1mm}D_0}{2}} \right)^\alpha },
\end{align}
while $N$ denotes the number of terms for integral approximation.
The larger $N$, the higher the approximation accuracy becomes.
%%
%Note that, all the users possess independent and identically distributed (i.i.d.) channels, such that each ${{\left| {{h_k}} \right|}^2}$, $\forall k$, subjects to the distribution in \eqref{C3:SISOChannelDistribution}.
%
%In addition, we note that the following equality is useful in the derivations in this chapter
%\begin{equation}\label{C3:Euqality}
%\sum\limits_{n = 1}^N {{\beta _n}}  = {D+D_0},
%\end{equation}
%which is obtained by using Gaussian-Chebyshev quadrature conversely \cite{abramowitz1964handbook}.

Based on \eqref{C3:InstantSumRateSISONOMA}, the ergodic sum-rate of the SISO-NOMA system considered is defined as:
\begin{align}\label{C3:ErgodicSumRateSISONOMADefine}
\overline{R_{{\mathrm{sum}}}^{{\mathrm{SISO - NOMA}}}} = {{\mathrm{E}}_{\mathbf{H}}}\left\{ {R_{{\mathrm{sum}}}^{{\mathrm{SISO - NOMA}}}} \right\} = {{\mathrm{E}}_{\mathbf{H}}}\left\{ {\ln \left( {1 + \frac{{{P_{{\mathrm{max}}}}}}{{K{N_0}}}\sum\limits_{k = 1}^K {{{\left| {{h_k}} \right|}^2}} } \right)} \right\},
\end{align}
where the expectation ${{\mathrm{E}}_{\mathbf{H}}}\left\{ \cdot \right\}$ is averaged over both the large-scale fading and small-scale fading in the overall channel matrix ${\mathbf{H}}$.
For a large number of users, i.e., $K\rightarrow \infty$, the sum of channel gains of all the users within the $\ln\left(\cdot\right)$ in \eqref{C3:ErgodicSumRateSISONOMADefine} becomes a deterministic value due to the strong law of large number, i.e., $\mathop {\lim } \limits_{K \to \infty } {\frac{1}{K}\sum\limits_{k = 1}^K {{{\left| {{h_k}} \right|}^2}} } = \overline{{{\left| {{h}} \right|}^2}}$, where $\overline{{{\left| {{h}} \right|}^2}}$ denotes the average channel power gain and it is given by
\begin{equation}\label{C3:Mean_ChannelPowerGainSISO}
\overline{{{\left| {{h}} \right|}^2}} = \int_0^\infty  {x{f_{{{\left| {{h}} \right|}^2}}}\left( x \right)} dx
\approx \frac{1}{D+D_0}\sum\limits_{n = 1}^N {\frac{{{\beta _n}}}{{{c_n}}}}.
\end{equation}
Therefore, the asymptotic ergodic sum-rate of the SISO-NOMA system considered is given by
%
%\footnote{The asymptotic ergodic sum-rate is defined by $\mathop {\lim }\limits_{K \to \infty } \overline{R_{\mathrm{sum}}} = \mathop {\lim }\limits_{K \to \infty }{{\mathrm{E}}_{\mathbf{H}}}\left\{ {R_{{\mathrm{sum}}}} \right\}$.
%%
%However, in this chapter, we actually derive the ergodic asymptotic sum-rate $\overline {\mathop {\lim }\limits_{K \to \infty } {R_{\mathrm{sum}}}} = {{\mathrm{E}}_{\mathbf{H}}}\left\{ \mathop {\lim }\limits_{K \to \infty } {R_{{\mathrm{sum}}}} \right\}$ to utilize the asymptotic condition ${K \to \infty }$, such as in \eqref{C3:ErgodicSumRateSISONOMA}.
%%
%.}
\begin{align}\label{C3:ErgodicSumRateSISONOMA}
\hspace{-2mm}\mathop {\lim }\limits_{K \to \infty }  \overline{R_{\mathrm{sum}}^{{\mathrm{SISO-NOMA}}}}
& \mathop  = \limits^{(a)}  {{\mathrm{E}}_{\mathbf{H}}}\left\{ \mathop {\lim }\limits_{K \to \infty } {R_{\mathrm{sum}}^{{\mathrm{SISO-NOMA}}}} \right\}
= {\ln}\left( 1 + \frac{P_{\mathrm{max}}}{N_0} \overline{{{\left| {{h}} \right|}^2}} \right) \notag\\
& \approx {\ln}\left(\hspace{-1mm} {1 \hspace{-1mm}+ \hspace{-1mm} \frac{P_{\mathrm{max}}}{{{\left(D\hspace{-1mm}+\hspace{-1mm}D_0\right)}N_0}}\sum\limits_{n = 1}^N \hspace{-1mm} {\frac{{{\beta _n}}}{{{c_n}}}} } \hspace{-1mm}\right),
\end{align}
where the equality $(a)$ is due to the bounded convergence theorem\cite{bartle2014elements} and owing to the finite channel capacity.
Note that for a finite number of users $K$, the asymptotic ergodic sum-rate in \eqref{C3:ErgodicSumRateSISONOMA} serves as an upper bound for the actual ergodic sum-rate in \eqref{C3:ErgodicSumRateSISONOMADefine}, i.e., we have $\mathop {\lim }\limits_{K \to \infty }  \overline{R_{\mathrm{sum}}^{{\mathrm{SISO-NOMA}}}} \ge \overline{R_{\mathrm{sum}}^{{\mathrm{SISO-NOMA}}}}$, owing to the concavity of the logarithmic function and the Jensen's inequality.
In the Section \ref{C3:Simulations}, we will show that the asymptotic analysis in \eqref{C3:ErgodicSumRateSISONOMA} is also accurate for a finite value of $K$ and becomes tighter upon increasing $K$.

Similarly, based on \eqref{C3:InstantSumRateSISOOMA}, we can obtain the ergodic sum-rate of the SISO-OMA system as follows:
\begin{align}\label{C3:ErgodicSumRateSISOOMA}
\hspace{-2mm}\overline{R_{\mathrm{sum}}^{{\mathrm{SISO-OMA}}}}
& = {{\mathrm{E}}_{\mathbf{H}}}\left\{ \frac{1}{K}\sum\limits_{k = 1}^K {\ln}\left( 1 + \frac{P_{\mathrm{max}}}{{N_0}} {{\left| {{{h}_k}} \right|}^2} \right) \right\} \notag\\
& \mathop = \limits^{(a)} \int_0^\infty  {{{\ln }}\left( {1 + \frac{P_{\mathrm{max}}}{N_0}x} \right){f_{{{\left| {{h}} \right|}^2}}}\left( x \right)} dx \notag\\
& = \frac{1}{\left(D\hspace{-1mm}+\hspace{-1mm}D_0\right)}\sum\limits_{n = 1}^N {{\beta _n}{e^{\frac{{{c_n N_0}}}{{P_{\mathrm{max}}}}}} {\mathcal{E}_1}\left( {\frac{{{c_n N_0}}}{{P_{\mathrm{max}}}}} \right)},
\end{align}
where ${\mathcal{E}_l}\left( x \right) = \int_1^\infty  {\frac{{{e^{ - xt}}}}{t^l}} dt$ denotes the $l$-order exponential integral\cite{abramowitz1964handbook}.
The equality $(a)$ in \eqref{C3:ErgodicSumRateSISOOMA} is obtained since all the users have i.i.d. channel distributions.
Note that in contrast to SISO-NOMA, $\overline{R_{\mathrm{sum}}^{{\mathrm{SISO-OMA}}}}$ in \eqref{C3:ErgodicSumRateSISOOMA} is applicable to SISO-OMA supporting an arbitrary number of users.

\subsection{ESG in Single-antenna Systems}
Comparing \eqref{C3:ErgodicSumRateSISONOMA} and \eqref{C3:ErgodicSumRateSISOOMA}, the asymptotic ESG of SISO-NOMA over SISO-OMA with ${K \rightarrow \infty}$ can be expressed as follows:
\begin{align}\label{C3:EPGSISOERA}
\hspace{-2mm}\mathop {\lim }\limits_{K \rightarrow \infty}  \overline{G^{\mathrm{SISO}}} &= \mathop {\lim }\limits_{K \to \infty }  \overline{R_{\mathrm{sum}}^{{\mathrm{SISO-NOMA}}}} - \overline{R_{\mathrm{sum}}^{{\mathrm{SISO-OMA}}}} \\
&\approx {\ln}\left( {1 + \frac{{P_{\mathrm{max}}}}{{{\left(D+D_0\right)N_0}}}\sum\limits_{n = 1}^N {\frac{{{\beta _n}}}{{{c_n}}}} } \right)
- \frac{1}{\left(D\hspace{-1mm}+\hspace{-1mm}D_0\right)}\sum\limits_{n = 1}^N {{\beta _n}{e^{\frac{{{c_n}N_0}}{{P_{\mathrm{max}}}}}} {\mathcal{E}_1}\left( {\frac{{{c_n}N_0}}{{P_{\mathrm{max}}}}} \right)}.\notag
\end{align}
Then, in the high-SNR regime, we can approximate the asymptotic ESG\footnote{Under the sum-power constraint, the system SNR directly depends on the total system power budget ${P_{\mathrm{max}}}$, and thus the system SNR and ${P_{\mathrm{max}}}$ are interchangeably in this chapter.} in \eqref{C3:EPGSISOERA} by applying $\mathop {\lim }\limits_{x \to 0} {\mathcal{E}_1}\left( x \right)\approx - \ln \left( x \right) - \gamma $ \cite{abramowitz1964handbook} as
\begin{equation}\label{C3:EPGSISOERA2}
\mathop {\lim }\limits_{K \rightarrow \infty, {P_{\mathrm{max}}} \rightarrow \infty}  \overline{G^{\mathrm{SISO}}} \approx \vartheta \left( {D,{D_0}} \right) + \gamma,
\end{equation}
where $\vartheta \left( {D,{D_0}} \right)$ is given by
\begin{equation}\label{C3:NearFarDiversity}
\vartheta \left( {D,{D_0}} \right)
 = \ln \left( {\frac{{\sum\limits_{n = 1}^N {\left( {\frac{1}{{{c_n}}}} \right)\frac{{{\beta _n}}}{{\left( {D + {D_0}} \right)}}} }}{{\mathop \Pi \limits_{n = 1}^N {{\left( {\frac{1}{{{c_n}}}} \right)}^{\frac{{{\beta _n}}}{{\left( {D + {D_0}} \right)}}}}}}} \right)
\end{equation}
and $\gamma = 0.57721$ is the Euler-Mascheroni constant\cite{abramowitz1964handbook}.
Based on the weighted arithmetic and geometric means (AM-GM) inequality\cite{kedlaya1994proof}, we can observe that $\vartheta \left( {D,{D_0}} \right) \ge 0$.
This implies that $\mathop {\lim }\limits_{K \rightarrow \infty, {P_{\mathrm{max}}} \rightarrow \infty}  \overline{G^{\mathrm{SISO}}}>0$ and SISO-NOMA provides a higher asymptotic ergodic sum-rate than SISO-OMA in the system considered.

To further simplify the expression of ESG, we consider path loss exponents $\alpha$ in the range of $\alpha \in \left[3,6\right]$ in \eqref{C3:BetaCn}, which usually holds in urban environments\cite{Access2010}.
As a result, $c_n \gg 1$.
Hence, $\vartheta \left( {D,{D_0}} \right)$ in \eqref{C3:NearFarDiversity} can be further simplified as follows:
\begin{equation}\label{C3:NearFarDiversity2}
\vartheta \left( {D,{D_0}} \right)\approx \vartheta \left( \eta  \right) =\ln \left( {\frac{{\frac{\pi }{{N\left( {1 + \eta } \right)}}\sum\limits_{n = 1}^N {\left[ {{\lambda _n}\left( \eta  \right)} \right]^{1 - \alpha }}\left| {\sin \frac{{2n - 1}}{{2N}}\pi } \right| }}{{\mathop \Pi \limits_{n = 1}^N {{\left[ {{\lambda_n}\left( \eta  \right)} \right]}^{ - \frac{{\alpha \pi {\lambda_n}\left( \eta  \right)}}{{N\left( {1 + \eta } \right)}}\left| {\sin \frac{{2n - 1}}{{2N}}\pi } \right|}}}}} \right),
\end{equation}
where ${\lambda _n}\left( \eta  \right) = \left( {\frac{{\eta - 1}}{2}\cos \left( {\frac{{2n - 1}}{{2N}}\pi } \right) + \frac{{\eta +1}}{2}} \right) \in \left[1, \eta\right)$.
The normalized cell size of $\eta = \frac{D}{D_0} \ge 1$ is the ratio between the outer radius $D$ and the inner radius ${D_0}$, which also serves as a metric of the path loss discrepancy.

We can see that the asymptotic ESG of SISO-NOMA over SISO-OMA in \eqref{C3:EPGSISOERA2} is composed of two components, i.e., $\vartheta \left( {D,{D_0}} \right)$ and $\gamma$.
As observed in \eqref{C3:NearFarDiversity2}, the former component of $\vartheta \left( {D,{D_0}} \right) \approx \vartheta \left( \eta  \right) $ only depends on the normalized cell size of $\eta = \frac{D}{D_0}$ instead of the absolute values of $D$ and ${D_0}$.
In fact, it can characterize the \emph{large-scale near-far gain} attained by NOMA via exploiting the discrepancy in distances among NOMA users.
Interestingly, for the extreme case that all the users are randomly deployed on a circle, i.e., $D = D_0$, we have $\eta = 1$, ${\lambda _n}\left( \eta  \right) = 1$, and $\vartheta \left( \eta  \right) = 0$.
In other words, the large-scale near-far gain disappears, when all the users are of identical distance away from the BS.
With the aid of $\vartheta \left( \eta  \right) = 0$, we can observe in \eqref{C3:EPGSISOERA2} that the performance gain achieved by NOMA is a constant value of $\gamma = 0.57721$ nat/s/Hz.
Since all the users are set to have the same distance when $D = D_0$, the minimum asymptotic ESG $\gamma$ arising from the small-scale Rayleigh fading is named as the \emph{small-scale fading gain} in this chapter.
In fact, in the asymptotic case associated with $K \rightarrow \infty$ and ${P_{\mathrm{max}}} \rightarrow \infty$, SISO-NOMA provides \emph{at least $\gamma =  0.57721$} nat/s/Hz spectral efficiency gain over SISO-OMA for an arbitrary cell size in Rayleigh fading channels.
%
%Otherwise, the asymptotic ESG is larger than $\gamma$ for a general cell deployment $D > D_0$.
%%
Additionally, we can see that the ESG of SISO-NOMA over SISO-OMA is saturated in the high-SNR regime.
This is because the instantaneous sum-rates of both the SISO-NOMA system in \eqref{C3:InstantSumRateSISONOMA} and the SISO-OMA system in \eqref{C3:InstantSumRateSISOOMA} increase logarithmically with ${P_{\mathrm{max}}} \rightarrow \infty$.
%It is because that the logarithm function in Shannon formula leads to a diminishing rate increase with increasing the system SNR.
%

%

%\footnote{In this chapter, we fix the inner radius ${D_0}$ and increase $\eta$ via increasing the outer radius $D$, such that $\eta$ can denote the cell size.}

%In fact, the dominating component of the asymptotic ESG in \eqref{C3:EPGSISOERA2} originates from large-scale near-far gain.
%%
%Even with the further simplification in \eqref{C3:NearFarDiversity2}, the large-scale near-far gain $\vartheta \left( \eta  \right)$ is still a complicated function of cell size $\eta$.
%
To visualize the large-scale near-far gain, we illustrate the asymptotic ESG in \eqref{C3:EPGSISOERA2} versus $D$ and $D_0$ in Figure \ref{C3:ESG_ERA}.
We can observe that when $\eta = 1$,  the large-scale near-far gain disappears and the asymptotic ESG is bounded from below by its minimum value of $\gamma = 0.57721$ nat/s/Hz due to the small-scale fading gain.
Additionally, for different values of $D$ and $D_0$ but with a fixed $\eta = \frac{D}{D_0}$, SISO-NOMA offers the same ESG compared to SISO-OMA.
This is because as predicted in \eqref{C3:NearFarDiversity2}, the large-scale near-far gain only depends on the normalized cell size $\eta$.
More importantly, we can observe that the large-scale near-far gain increases with the normalized cell size $\eta$.
In fact, for a larger normalized cell size $\eta$, the heterogeneity in the large-scale fading among users becomes higher and SISO-NOMA attains a higher near-far gain, hence improving the sum-rate performance.

%
%To characterize the impact of cell size $\eta$ on the ESG of SISO-NOMA over SISO-OMA, based on the converse of Jensen's inequality\cite[Theorem~4]{dragomir1999converse}, we introduce a lower bound for $\vartheta \left( \eta \right)$ as follows:
%\begin{equation}\label{C3:NearFarDiversity2}
%\vartheta \left( \eta  \right) \ge \underline{\vartheta \left( \eta  \right)} = \frac{1}{4}\left( {1 - \frac{1}{{{\eta ^\alpha }}}} \right)\left( {{\eta ^\alpha } - 1} \right) = \frac{1}{4}\left( {1 - \frac{1}{{{\eta ^\alpha }}}} \right)\left( {{\eta ^\alpha } - 1} \right),
%\end{equation}
%
%
%\begin{Thm}\label{C3:JesenDifference}
%The large-scale near-far gain ${\lambda _n}\left( \eta  \right)$ in \eqref{C3:NearFarDiversity2} can be lower bounded by
%\begin{equation}\label{C3:NearFarDiversity2}
%\vartheta \left( \eta  \right) \ge ,
%\end{equation}
%\end{Thm}

\begin{figure}[t]
\centering
\includegraphics[width=4in]{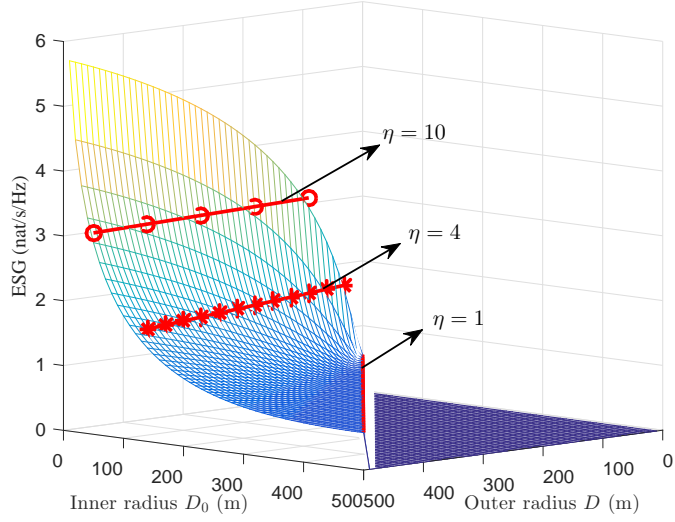}
\caption{The asymptotic ESG in \eqref{C3:EPGSISOERA2} under equal resource allocation versus $D$ and $D_0$ with $K \rightarrow \infty$ and ${P_{\mathrm{max}}} \rightarrow \infty$.}
\label{C3:ESG_ERA}
\end{figure}

\begin{remark}
Note that it has been analytically shown in \cite{Dingtobepublished} that two users with a large distance difference (or equivalently channel gain difference) are preferred to be paired.
This is consistent with our conclusion in this chapter, where a larger normalized cell size $\eta$ enables a higher ESG of NOMA over OMA.
However, \cite{Dingtobepublished} only considered a pair of two NOMA users.
%
%Besides, it did not derive the rate gain of NOMA over OMA explicitly and only relied on simulations to demonstrate this insight.
%
In this chapter, we analytically obtain the ESG of NOMA over OMA for a more general NOMA system supporting a large number of UL users.
More importantly, we identify two kinds of gains in the ESG derived and reveal their different behaviors.
\end{remark}
\section{ESG of MIMO-NOMA over MIMO-OMA}
In this section, the ergodic sum-rates of MIMO-NOMA and MIMO-OMA associated with FDMA-ZF as well as FDMA-MRC are firstly analyzed. Then, the asymptotic ESGs of MIMO-NOMA over MIMO-OMA with FDMA-ZF and FDMA-MRC detection are investigated.
\subsection{Ergodic Sum-rate of MIMO-NOMA with MMSE-SIC}
Let us consider that an $M$-antenna BS serves $K$ single-antenna non-orthogonal users relying on MIMO-NOMA.
The BS employs MMSE-SIC detection for retrieving the messages of all the users.
The instantaneous achievable data rate of user $k$ in the MIMO-NOMA system relying on MMSE-SIC detection\footnote{The derivation of individual rates in \eqref{C3:MIMONOMAIndividualAchievableRate} for MMSE-SIC detection of MIMO-NOMA is based on the matrix inversion lemma: \[\log \left| {{\bf{A}} + {\bf{h}}{{\bf{h}}^{\mathrm{H}}}} \right| - \log \left| {\bf{A}} \right| = \log \left| {1 + {{\bf{h}}^{\mathrm{H}}}{{\bf{A}}^{ - 1}}{\bf{h}}} \right|.\]
	Interested readers are referred to \cite{Tse2005} for a detailed derivation.} is given by\cite{Tse2005}:
\begin{align}\label{C3:MIMONOMAIndividualAchievableRate}
R_{k}^{{\mathrm{MIMO-NOMA}}} = \ln \left| {{{\bf{I}}_M} + \frac{1}{{{N_0}}}\sum\limits_{i = k}^K {{p_i}{{\bf{h}}_i}{\bf{h}}_i^{\mathrm{H}}} } \right|
- \ln \left| {{{\bf{I}}_M} + \frac{1}{{{N_0}}}\sum\limits_{i = k + 1}^K {{p_i}{{\bf{h}}_i}{\bf{h}}_i^{\mathrm{H}}} } \right|.
\end{align}
As a result, the instantaneous sum-rate of MIMO-NOMA is obtained as
\begin{equation}\label{C3:InstantSumRateMIMONOMA}
R_{\mathrm{sum}}^{{\mathrm{MIMO-NOMA}}} = \sum\limits_{k = 1}^K R_{k}^{{\mathrm{MIMO-NOMA}}}
= {\ln}\left| {{{\bf{I}}_M} + \frac{1}{{{N_0}}}\sum\limits_{k = 1}^K p_k {{{\bf{h}}_k}{\bf{h}}_k^{\mathrm{H}}} } \right|.
\end{equation}

In fact, MMSE-SIC is capacity-achieving \cite{Tse2005} and \eqref{C3:InstantSumRateMIMONOMA} is the channel capacity for a given instantaneous channel matrix ${\bf{H}}$\cite{GoldsmithMIMOCapacity2003}.
In general, it is a challenge to obtain a closed-form expression for the instantaneous channel capacity above due to the determinant of the summation of matrices in \eqref{C3:InstantSumRateMIMONOMA}.
To provide some insights, in the following theorem, we consider an asymptotically tight upper bound for the achievable sum-rate in \eqref{C3:InstantSumRateMIMONOMA} associated with $K \to \infty$.

\begin{Thm}\label{C3:Theorem1}
For the MIMO-NOMA system considered in \eqref{C3:MIMONOMASystemModel} relying on MMSE-SIC detection, given any power allocation strategy ${\bf{p}} = \left[ {{p_1}, \ldots ,{p_K}} \right]$, the achievable sum-rate in \eqref{C3:InstantSumRateMIMONOMA} is upper bounded by
\begin{equation}\label{C3:InstantSumRateMIMONOMA_UpperBound}
\hspace{-2mm}R_{\mathrm{sum}}^{{\mathrm{MIMO-NOMA}}} \le M{\ln}\left( {1 + \frac{1}{{M{N_0}}}\sum\limits_{k = 1}^K {{p_k}{{\left\| {{{\bf{h}}_k}} \right\|}^2}} } \right).
\end{equation}
This upper bound is asymptotically tight, when $K \to \infty$, i.e.,
\begin{equation}\label{C3:AsympInstantSumRateMIMONOMA}
\hspace{-2mm}\mathop {\lim }\limits_{K \rightarrow \infty} \hspace{-1mm} R_{\mathrm{sum}}^{{\mathrm{MIMO-NOMA}}}
\hspace{-0.5mm}=\hspace{-0.5mm} \mathop {\lim }\limits_{K \rightarrow \infty} M{\ln}\hspace{-0.5mm}\left( \hspace{-0.5mm}{1 \hspace{-0.5mm}+\hspace{-0.5mm} \frac{1}{{M{N_0}}}\sum\limits_{k = 1}^K {{p_k}{{\left\| {{{\bf{h}}_k}} \right\|}^2}} } \hspace{-0.5mm}\right)\hspace{-0.5mm}.
\end{equation}
\end{Thm}
\begin{proof}
Please refer to Appendix \ref{C3:AppendixA} for the proof of Theorem \ref{C3:Theorem1}.
\end{proof}

Now, given the instantaneous achievable sum-rate obtained in \eqref{C3:AsympInstantSumRateMIMONOMA}, we proceed to calculate the ergodic sum-rate.
Given the distance from a user to the BS as $d$, the channel gain ${{\left\| {{\mathbf{h}}} \right\|}^2}$ follows the Gamma distribution\cite{yang2017noma}, whose conditional PDF and CDF are given by\footnote{Similar to \eqref{C3:SISOChannelDistributionPDF} and \eqref{C3:SISOChannelDistributionCDF}, we can safely assume that all the users have i.i.d. channel distribution within the cell and drop the subscript $k$ in \eqref{C3:MIMOChannelDistribution}, since the system sum-rate in \eqref{C3:InstantSumRateMIMONOMA} is independent of the MMSE-SIC decoding order\cite{Tse2005}.}
\begin{equation}\label{C3:MIMOChannelDistribution}
{f_{{{\left\| {{\mathbf{h}}} \right\|}^2 | d}}}\left( x \right) = {{\mathrm{Gamma}}} \left(M, 1+ d^\alpha,x\right)\; \text{and} \;
{F_{{{\left\| {{\mathbf{h}}} \right\|}^2 | d}}}\left( x \right) = \frac{{\gamma_{L} \left( {M,\left( {1 + d^\alpha } \right)x} \right)}}{\Gamma \left( M \right)},
\end{equation}
respectively, where ${\mathrm{Gamma}} \left(M, \lambda,x\right) = \frac{{{{ \lambda }^M}{x^{M - 1}}{e^{ -  \lambda x}}}}{\Gamma \left( M \right)}$ denotes the PDF of a random variable obeying a Gamma distribution, ${\Gamma \left( M \right)}$ denotes the Gamma function, and ${\gamma_{L} \left( {M,\left( {1 + d^\alpha } \right)x} \right)}$ denotes the lower incomplete Gamma function.
Then, the CDF of the channel gain ${{\left\| {{\mathbf{h}}} \right\|}^2}$ can be obtained by
\begin{equation}
{F_{{{\left\| {{\mathbf{h}}} \right\|}^2}}}\left( x \right) = \int_{D_0}^D \frac{{\gamma_{L} \left( {M,\left( {1 + d^\alpha } \right)x} \right)}}{{\Gamma \left(M\right)}} {f_{{d}}}\left( z \right)dz.
\end{equation}
By applying the Gaussian-Chebyshev quadrature approximation\cite{abramowitz1964handbook}, the CDF and PDF of ${{\left\| {{\mathbf{h}}} \right\|}^2}$ can be written as
\begin{align}
\hspace{-2mm}{F_{{{\left\| {{\mathbf{h}}} \right\|}^2}}}\left( x \right) &\approx 1 - \frac{1}{D+D_0}\sum\limits_{n = 1}^N \frac{{{\beta _n}\gamma_{L} \left( {M,{c_n}x} \right)}}{{\Gamma \left(M\right)}} \; \text{and} \notag\\
\hspace{-2mm}{f_{{{\left\| {{\mathbf{h}}} \right\|}^2}}}\left( x \right) &\approx \frac{1}{D+D_0}\sum\limits_{n = 1}^N {{\beta _n} {\mathrm{Gamma}} \left(M, c_n,x\right)}, x \ge 0,
\end{align}
respectively, where ${\beta _n}$ and ${c_n}$ are given in \eqref{C3:BetaCn}.

According to \eqref{C3:AsympInstantSumRateMIMONOMA}, given the equal resource allocation strategy, i.e., ${{p_{k}}} = \frac{{P_{\mathrm{max}}}}{K}$, the asymptotic ergodic sum-rate of MIMO-NOMA associated with $K\rightarrow \infty$ can be obtained as follows:
\begin{align} \label{C3:ErgodicSumRateMIMONOMA}
\mathop {\lim }\limits_{K \to \infty } \overline{R_{\mathrm{sum}}^{{\mathrm{MIMO-NOMA}}}}
& = \mathop {\lim }\limits_{K \to \infty } {{\mathrm{E}}_{\mathbf{H}}}\left\{ {R_{{\mathrm{sum}}}^{{\mathrm{MIMO - NOMA}}}} \right\} = M{\ln}\left( 1 + \frac{{P_{\mathrm{max}}}}{MN_0} \overline{{{\left\| {{\mathbf{h}}} \right\|}^2}} \right) \\
& \approx M{\ln}\left(\hspace{-1mm}{1+ \frac{{P_{\mathrm{max}}}}{{{\left(D\hspace{-1mm}+\hspace{-1mm}D_0\right)N_0}}}
\sum\limits_{n = 1}^N \hspace{-1mm}{\frac{{{\beta _n}}}{{{c_n}}}} } \hspace{-1mm}\right),\notag
\end{align}
where $\overline{{{\left\| {{\mathbf{h}}} \right\|}^2}}$ denotes the average channel gain, which is given by
\begin{equation}\label{C3:Mean_ChannelPowerGainMIMO}
\overline{{{\left\| {{\mathbf{h}}} \right\|}^2}} = \int_0^\infty  {x{f_{{{\left\| {{\mathbf{h}}} \right\|}^2}}}\left( x \right)} dx
\approx \frac{M}{D+D_0}\sum\limits_{n = 1}^N {\frac{{{\beta _n}}}{{{c_n}}}}.
\end{equation}

\begin{remark} \label{C3:remark2}
Comparing \eqref{C3:ErgodicSumRateSISONOMA} and \eqref{C3:ErgodicSumRateMIMONOMA}, we can observe that for a sufficiently large number of users, the considered MIMO-NOMA system is asymptotically equivalent to a SISO-NOMA system with $M$-fold increases in DoF and an equivalent average channel gain of $\overline{{{\left\| {{{\bf{h}}}} \right\|}^2}}$ in each DoF.
Intuitively, when the number of UL receiver antennas at the BS, $M$, is much smaller than the number of users, $K \to \infty$, which corresponds to the extreme asymmetric case of MIMO-NOMA, the multi-antenna BS behaves asymptotically in the same way as a single-antenna BS.
Additionally, when $K \gg M$, due to the diverse channel directions of all the users, the received signals fully span the $M$-dimensional signal space\cite{WangMUG}.
Therefore, MIMO-NOMA using MMSE-SIC reception can fully exploit the system's spatial DoF, $M$, and its performance can be approximated by that of a SISO-NOMA system with $M$-fold DoF.
\end{remark}

\subsection{Ergodic Sum-rate of MIMO-OMA with FDMA-ZF}
Upon installing more UL receiver antennas at the BS, ZF can be employed for MUD and the MIMO-OMA system using FDMA-ZF can accommodate $M$ users on each frequency subband.
As mentioned before, we adopt a random user grouping strategy for the MIMO-OMA system using FDMA-ZF detection, where we randomly select $M$ users as a group and denote the composite channel matrix of the $g$-th group by ${\bf{H}}_g = \left[ {{{\bf{h}}_{(g-1)M+1}},{{\bf{h}}_{(g-1)M+2}}, \ldots ,{{\bf{h}}_{gM}}} \right] \in \mathbb{C}^{ M \times M}$.
Then, the instantaneous achievable data rate of user $k$ in the MIMO-OMA system is given by
\begin{equation}\label{C3:MIMOOMAZFIndividualAchievableRate}
R_{k,\mathrm{FDMA-ZF}}^{{\mathrm{MIMO-OMA}}} = f_g {\ln}\left(1+ {\frac{{{p_{k }}{{\left| {{{\mathbf{w}_{g,k}^{\mathrm{H}}\mathbf{h}}_k}} \right|}^2}}}{{f_g N_0}}} \right),
\end{equation}
where $f_g$ denotes the normalized frequency allocation for the $g$-th group.
The vector $\mathbf{w}_{g,k} \in \mathbb{C}^{ M \times 1}$ denotes the normalized ZF detection vector for user $k$ with ${\left\| {{{\mathbf{w}_{g,k}}}} \right\|}^2 = 1$, which is obtained based on the pseudoinverse of the composite channel matrix ${\bf{H}}_g$ in the $g$-th user group\cite{Tse2005}.

Given the equal resource allocation strategy, i.e., ${{p_{k}}} = \frac{{P_{\mathrm{max}}}}{K}$ and $f_g = 1/G = \frac{M}{K}$, the instantaneous sum-rate of MIMO-OMA using FDMA-ZF can be formulated as:
\begin{equation}\label{C3:InstantSumRateMIMOOMAZF}
R_{\mathrm{sum,FDMA-ZF}}^{{\mathrm{MIMO-OMA}}}
= \sum\limits_{k = 1}^K R_{k,\mathrm{FDMA-ZF}}^{{\mathrm{MIMO-OMA}}}
= \frac{M}{K}\sum\limits_{k = 1}^K {\ln}\left(1+ \frac{{P_{\mathrm{max}}}}{MN_0}{\left| {{{\mathbf{w}_{g,k}^{\mathrm{H}}\mathbf{h}}_k}} \right|}^2 \right).
\end{equation}
Since ${\left\| {{{\mathbf{w}_{g,k}}}} \right\|}^2 = 1$ and ${\bf{g}}_k \sim \mathcal{CN}\left(\mathbf{0},{{\bf{I}}_M}\right)$, we have ${{{{\mathbf{w}_{g,k}^{\mathrm{H}}\mathbf{g}}_k}}} \sim \mathcal{CN}\left({0},1\right)$ \cite{Tse2005}.
As a result, ${\left| {{{\mathbf{w}_{g,k}^{\mathrm{H}}\mathbf{h}}_k}} \right|^2}$ in \eqref{C3:InstantSumRateMIMOOMAZF} has an identical distribution with ${{\left| {{h}} \right|}^2}$, i.e.,  its CDF and PDF are given by \eqref{C3:SISOChannelDistributionPDF} and \eqref{C3:SISOChannelDistributionCDF}, respectively.
Therefore, the ergodic sum-rate of the MIMO-OMA system considered can be expressed as:
\vspace{-2.5mm}
\begin{align}\label{C3:ErgodicSumRateMIMOOMAZF}
\overline{R_{\mathrm{sum,FDMA-ZF}}^{{\mathrm{MIMO-OMA}}}} &= {{\mathrm{E}}_{\mathbf{H}}}\left\{ {R_{\mathrm{sum,FDMA-ZF}}^{{\mathrm{MIMO-OMA}}}} \right\} = \int_0^\infty  M{{{\ln }}\left( {1 + \frac{P_{\mathrm{max}}}{MN_0}x} \right){f_{{{\left| {{h}} \right|}^2}}}\left( x \right)} dx \notag\\
& = \frac{M}{\left(D+D_0\right)}\sum\limits_{n = 1}^N {{\beta _n}{e^{\frac{{{c_n M N_0}}}{{P_{\mathrm{max}}}}}} {\mathcal{E}_1}\left( {\frac{{{c_n M N_0}}}{{P_{\mathrm{max}}}}} \right)}.
\end{align}
\par\noindent\vspace{-22.5mm}

\subsection{Ergodic Sum-rate of MIMO-OMA with FDMA-MRC}
The instantaneous achievable data rate of user $k$ in the MIMO-OMA system using the FDMA-MRC receiver is given by
\vspace{-2.5mm}
\begin{equation}\label{C3:MIMOOMAMRCIndividualAchievableRate}
R_{k,\mathrm{FDMA-MRC}}^{{\mathrm{MIMO-OMA}}} =
f_k {\ln}\left(1+ {\frac{{{p_{k }}{{\left\| {{{{\mathbf{h}}}_k}} \right\|}^2}}}{{f_k N_0}}} \right).
\end{equation}
Upon adopting the equal resource allocation strategy, i.e., ${{p_{k}}} = \frac{{P_{\mathrm{max}}}}{K}$ and $f_k = 1/K$, the instantaneous sum-rate of MIMO-OMA relying on FDMA-MRC is obtained by
\vspace{-5mm}
\begin{equation}\label{C3:InstantSumRateMIMOOMAMRC}
R_{\mathrm{sum,FDMA-MRC}}^{{\mathrm{MIMO-OMA}}}
= \sum\limits_{k = 1}^K R_{k,\mathrm{FDMA-MRC}}^{{\mathrm{MIMO-OMA}}}
= \frac{1}{K}\sum\limits_{k = 1}^K {\ln}\left(1+ \frac{{P_{\mathrm{max}}}}{N_0}{\left\| {{{\mathbf{h}}_k}} \right\|}^2 \right).
\end{equation}
Averaging $R_{\mathrm{sum,FDMA-MRC}}^{{\mathrm{MIMO-OMA}}}$ over the channel fading, we arrive at the ergodic sum-rate of MIMO-OMA using FDMA-MRC as
\begin{align}\label{C3:ErgodicSumRateMIMOOMAMRC}
\overline{R_{\mathrm{sum,FDMA-MRC}}^{{\mathrm{MIMO-OMA}}}} &= {\mathrm{E}_{\bf{H}}}\left\{ \frac{1}{K}\sum\limits_{k = 1}^K{\ln}\left(1+ \frac{{P_{\mathrm{max}}}}{N_0}{\left\| {{{\mathbf{h}}_k}} \right\|}^2 \right) \right\} \\
& = \int_0^\infty  {{{\ln }}\left( {1 + \frac{{P_{\mathrm{max}}}}{{N_0}}x} \right){f_{{{\left\| {{\mathbf{h}}} \right\|}^2}}}\left( x \right)} dx \notag\\
& = \frac{1}{{\left( {D + {D_0}} \right)}}\sum\limits_{n = 1}^N {{\beta _n}\underbrace {\int_0^\infty  {\ln \left( {1 + \frac{{P_{\mathrm{max}}}}{{N_0}}x} \right)} {{\mathrm{Gamma}}}\left( {M,{c_n}} \right)dx}_{{T_n}}},\notag
\end{align}
with $T_n$ given by
\begin{align}\label{C3:T_N_App}
{T_n} & \mathop = \limits^{(a)} {\int_0^\infty  {\ln \left( 1 + t \right)} {{\mathrm{Gamma}}}\left( {M,\frac{{{N_0}{c_n}}}{{{P_{\mathrm{max}}}}}} \right)dt} \notag\\
& \mathop = \limits^{(b)}  \frac{\left( \frac{{{N_0}{c_n}}}{{{P_{\mathrm{max}}}}} \right)^M}{{\Gamma \left(M\right)}} G_{2,3}^{3,1}\left( {\begin{array}{*{20}{c}}
{ - M, - M + 1}\\
{ - M, - M, 0}
\end{array}\left| {\frac{{{N_0}{c_n}}}{{{P_{\mathrm{max}}}}}} \right.} \right),
%& \mathop \approx \limits^{(c)}_{{P_{\mathrm{max}}} \to \infty }  \psi \left( M \right) + \ln \left( {\frac{{{P_{\mathrm{max}}}}}{{M{c_n}}}} \right) + \mathcal{O}\left( \frac{{M{c_n}}}{{{P_{\mathrm{max}}}}} \right),
\end{align}
%Wikipedia E[ln(X)] X~ gamma(alpha beta)
where $G_{p,q}^{m,n}\left(  \cdot  \right)$ denotes the Meijer G-function.
The equality $(a)$ in \eqref{C3:T_N_App} is obtained due to $t = \frac{{{P_{\mathrm{max}}}}}{{N_0}}x \sim {{\mathrm{Gamma}}}\left( {{M},\frac{{{N_0}{c_n}}}{{{P_{\mathrm{max}}}}}} \right)$ and the equality $(b)$ in \eqref{C3:T_N_App} is based on Equation (3) in \cite{Heath2011}.
Now, the ergodic sum-rate of MIMO-OMA using FDMA-MRC can be written as
\begin{equation}\label{C3:ErgodicSumRateMIMOOMAMRC2}
\overline{R_{\mathrm{sum,FDMA-MRC}}^{{\mathrm{MIMO-OMA}}}}
= \frac{1}{{\left( {D + {D_0}} \right)}}\sum\limits_{n = 1}^N {\beta _n}\left( \frac{\left( \frac{{N_0{c_n}}}{{{P_{\mathrm{max}}}}} \right)^M}{{\Gamma \left(M\right)}} G_{2,3}^{3,1}\left( {\begin{array}{*{20}{c}}
{ - M, - M + 1}\\
{ - M, - M, 0}
\end{array}\left| {\frac{{N_0{c_n}}}{{P_{\mathrm{max}}}}} \right.} \right)  \right).
\end{equation}

Note that, the ergodic sum-rate in \eqref{C3:ErgodicSumRateMIMOOMAMRC2} is applicable to an arbitrary number of users $K$ and an arbitrary SNR, but it is too complicated to offer insights concerning the ESG of MIMO-NOMA over MIMO-OMA.
Hence, based on \eqref{C3:ErgodicSumRateMIMOOMAMRC}, we derive the asymptotic ergodic sum-rate of MIMO-OMA with FDMA-MRC in the low-SNR regime with ${P_{\mathrm{max}}} \rightarrow 0$ as follows:
\begin{equation}\label{C3:ErgodicSumRateMIMOOMAMRC3}
\mathop {\lim }\limits_{{P_{\mathrm{max}}} \rightarrow 0} \overline{R_{\mathrm{sum,FDMA-MRC}}^{{\mathrm{MIMO-OMA}}}} = \frac{{P_{\mathrm{max}}}}{{N_0}} \overline{{{\left\| {{\mathbf{h}}} \right\|}^2}}
= \frac{M{P_{\mathrm{max}}}}{{N_0} \left(D+D_0\right)}\sum\limits_{n = 1}^N {\frac{{{\beta _n}}}{{{c_n}}}}.
\end{equation}
On the other hand, in the high-SNR regime, based on \eqref{C3:ErgodicSumRateMIMOOMAMRC}, the asymptotic ergodic sum-rate of MIMO-OMA using FDMA-MRC is given by
\begin{equation}\label{C3:ErgodicSumRateMIMOOMAMRC4}
\mathop {\lim }\limits_{{P_{\mathrm{max}}} \rightarrow \infty}  \overline{R_{\mathrm{sum,FDMA-MRC}}^{{\mathrm{MIMO-OMA}}}}
= {\ln}\left(\frac{{P_{\mathrm{max}}}}{N_0}\right) + {\mathrm{E}_{\bf{h}}}\left\{ {\ln}\left({\left\| {{{\mathbf{h}}}} \right\|}^2 \right) \right\}.
\end{equation}
\subsection{ESG in Multi-antenna Systems}
By comparing \eqref{C3:ErgodicSumRateMIMONOMA} and \eqref{C3:ErgodicSumRateMIMOOMAZF}, we have the asymptotic ESG of MIMO-NOMA over MIMO-OMA relying on FDMA-ZF as follows:
\begin{align}\label{C3:EPGMIMOERA}
\hspace{-3mm}\mathop {\lim }\limits_{K \rightarrow \infty}  \overline{G^{\mathrm{MIMO}}_{\mathrm{FDMA-ZF}}} &\hspace{-1mm}=
\mathop {\lim }\limits_{K \to \infty }  \overline{R_{\mathrm{sum}}^{{\mathrm{MIMO-NOMA}}}} - \overline{R_{\mathrm{sum,FDMA-ZF}}^{{\mathrm{MIMO-OMA}}}} \\
&\hspace{-1mm}\approx M{\ln}\left( {1 \hspace{-1mm}+\hspace{-1mm} \frac{{P_{\mathrm{max}}}}{{{\left(D\hspace{-1mm}+\hspace{-1mm}D_0\right)N_0}}}\sum\limits_{n = 1}^N {\frac{{{\beta _n}}}{{{c_n}}}} } \right)\hspace{-1mm}- \hspace{-1mm} \frac{M}{\left(D\hspace{-1mm}+\hspace{-1mm}D_0\right)}\hspace{-1mm}\sum\limits_{n = 1}^N \hspace{-1mm}{{\beta _n}{e^{\frac{{{c_n}MN_0}}{{P_{\mathrm{max}}}}}} \hspace{-1mm}{\mathcal{E}_1}\hspace{-1mm}\left( \hspace{-0.5mm} {\frac{{{c_n}MN_0}}{{P_{\mathrm{max}}}}} \hspace{-0.5mm} \right)}.\notag
\end{align}
To unveil some insights, we consider the asymptotic ESG in the high-SNR regime as follows
\begin{equation}\label{C3:EPGMIMOERA2}
\hspace{-2mm}\mathop {\lim }\limits_{K \rightarrow \infty, {P_{\mathrm{max}}} \rightarrow \infty}  \overline{G^{\mathrm{MIMO}}_{\mathrm{FDMA-ZF}}} \approx M  \underbrace{\vartheta \left( {D,{D_0}} \right)}_{\mathrm{large-scale\;near-far\;gain}} + M\ln\left(M\right) + M \underbrace{\gamma}_{\mathrm{small-scale\;fading\;gain}},\hspace{-0.5mm}
\end{equation}
where $\vartheta \left( {D,{D_0}} \right)$ denotes the large-scale near-far gain given in \eqref{C3:NearFarDiversity}.

\begin{remark}
The identified two kinds of gains in ESG of the single-antenna system in \eqref{C3:EPGSISOERA2} are also observed in the ESG of MIMO-NOMA over MIMO-OMA using FDMA-ZF in \eqref{C3:EPGMIMOERA2}.
Moreover, it can be observed that both the large-scale near-far gain $\vartheta \left( {D,{D_0}} \right)$ and the small-scale fading gain $\gamma$ are increased by $M$ times as indicated in \eqref{C3:EPGMIMOERA2}.
In fact, upon comparing \eqref{C3:EPGSISOERA2} and \eqref{C3:EPGMIMOERA2}, we have
\begin{equation}\label{C3:EPGMIMOERA3}
\mathop {\lim }\limits_{K \rightarrow \infty, {P_{\mathrm{max}}} \rightarrow \infty}  \overline{G^{\mathrm{MIMO}}_{\mathrm{FDMA-ZF}}} = M\mathop {\lim }\limits_{K \rightarrow \infty, {P_{\mathrm{max}}} \rightarrow \infty} \overline{G^{{\mathrm{SISO}}}} + M\ln\left(M\right),
\end{equation}
which implies that the asymptotic ESG of MIMO-NOMA over MIMO-OMA is $M$-times of that in single-antenna systems, when there are $M$ UL receiver antennas at the BS.
In fact, for $K \rightarrow \infty$, the heterogeneity in channel directions of all the users allows the received signals to fully span across the $M$-dimensional signal space.
Hence, MIMO-NOMA and MIMO-OMA using FDMA-ZF can fully exploit the system's maximal spatial DoF $M$.
Furthermore, we have an additional power gain of $\ln\left(M\right)$ in the second term in \eqref{C3:EPGMIMOERA3}.
This is due to a factor of $\frac{1}{M}$ average power loss within each group for ZF projection to suppress the IUI in the MIMO-OMA system considered\cite{Tse2005}.
\end{remark}

Comparing \eqref{C3:ErgodicSumRateMIMONOMA} and \eqref{C3:ErgodicSumRateMIMOOMAMRC2}, the asymptotic ESG of MIMO-NOMA over MIMO-OMA with FDMA-MRC is obtained by:
\begin{align}\label{C3:EPGMIMOERAMRC}
\hspace{-3mm}\mathop {\lim }\limits_{K \rightarrow \infty}  \overline{G^{\mathrm{MIMO}}_{\mathrm{FDMA-MRC}}} &\hspace{-1mm}=
\mathop {\lim }\limits_{K \to \infty }  \overline{R_{\mathrm{sum}}^{{\mathrm{MIMO-NOMA}}}} - \overline{R_{\mathrm{sum,FDMA-MRC}}^{{\mathrm{MIMO-OMA}}}} \\
&\hspace{-1mm}\approx M{\ln}\left( {1 \hspace{-1mm}+\hspace{-1mm} \frac{{P_{\mathrm{max}}}}{{{\left(D\hspace{-1mm}+\hspace{-1mm}D_0\right)N_0}}}\sum\limits_{n = 1}^N {\frac{{{\beta _n}}}{{{c_n}}}} } \right) \notag\\ &-\frac{1}{{\left( {D + {D_0}} \right)}}\sum\limits_{n = 1}^N {\beta _n}\left( \frac{\left( \frac{{N_0{c_n}}}{{{P_{\mathrm{max}}}}} \right)^M}{{\Gamma \left(M\right)}} G_{2,3}^{3,1}\left( {\begin{array}{*{20}{c}}
{ - M, - M + 1}\\
{ - M, - M, 0}
\end{array}\left| {\frac{{N_0{c_n}}}{{P_{\mathrm{max}}}}} \right.} \right)  \right).\notag
\end{align}
Then, based on \eqref{C3:ErgodicSumRateMIMONOMA} and \eqref{C3:ErgodicSumRateMIMOOMAMRC3}, the asymptotic ESG of MIMO-NOMA over MIMO-OMA with FDMA-MRC in the low-SNR regime is given by
\begin{align}\label{C3:EPGMIMOERA5}
\mathop {\lim }\limits_{K \rightarrow \infty, {P_{\mathrm{max}}} \rightarrow 0}  \overline{R_{\mathrm{sum,FDMA-MRC}}^{{\mathrm{MIMO-OMA}}}} = 0.
\end{align}
Not surprisingly, the performance gain of MIMO-NOMA over MIMO-OMA with FDMA-MRC vanishes in the low-SNR regime, which has been shown by simulations in existing works, \cite{Ding2014} for example.
In the high-SNR regime, the asymptotic ESG of MIMO-NOMA over MIMO-OMA with FDMA-MRC can be obtained from \eqref{C3:ErgodicSumRateMIMONOMA} and \eqref{C3:ErgodicSumRateMIMOOMAMRC4} by
\begin{align}\label{C3:EPGMIMOERA6}
\mathop {\lim }\limits_{K \rightarrow \infty, {P_{\mathrm{max}}} \rightarrow \infty}  \overline{G^{\mathrm{MIMO}}_{\mathrm{FDMA-MRC}}} &\approx \left( {M - 1} \right)\ln \left( {\frac{{{P_{{\mathrm{max}}}}}}{{\left( {D + {D_0}} \right){N_0}}}\sum\limits_{n = 1}^N {\frac{{{\beta _n}}}{{{c_n}}}} } \right) - \ln \left( M \right) + \Delta,
\end{align}
where $\Delta = \ln \left( {{{\mathrm{E}}_{\bf{h}}}\left\{ {{{\left\| {\bf{h}} \right\|}^2}} \right\}} \right) - {{\mathrm{E}}_{\bf{h}}}\left\{ {\ln \left( {{{\left\| {\bf{h}} \right\|}^2}} \right)} \right\}$ denotes the gap between $\ln \left( {{{\mathrm{E}}_{\bf{h}}}\left\{ {{{\left\| {\bf{h}} \right\|}^2}} \right\}} \right)$ and ${{\mathrm{E}}_{\bf{h}}}\left\{ {\ln \left( {{{\left\| {\bf{h}} \right\|}^2}} \right)} \right\}$.

Although the closed-form ESG of MIMO-NOMA over MIMO-OMA is not available for the case of FDMA-MRC, the third term $\Delta$ in \eqref{C3:EPGMIMOERA6} is a constant for a given outer radius $D$ and inner radius $D_0$.
Besides, it is expected that the first term in \eqref{C3:EPGMIMOERA6} dominates the ESG in the high-SNR regime.
%
%Based on these observations, we reveal some interesting insights in the following.
%
We can observe that the first term in \eqref{C3:EPGMIMOERA6} increases linearly with the system SNR in dB with a slope of $(M-1)$ in the high-SNR regime.
In other words, there is an $(M-1)$-fold DoF gain \cite{DMTadeoff} in the asymptotic ESG of MIMO-NOMA over MIMO-OMA using FDMA-MRC.
In fact, MIMO-NOMA is essentially an $M \times K$ MIMO system on all resource blocks, i.e., time slots and frequency subbands, where the system maximal spatial DoF is limited by $M$ due to $M<K$.
On the other hand, MIMO-OMA using the FDMA-MRC reception is always an $M \times 1$ MIMO system in each resource block, and thus it can only have a  spatial DoF, which is one.
As a result, an $(M-1)$-fold DoF gain can be achieved by MIMO-NOMA compared to MIMO-OMA using FDMA-MRC.
However, MIMO-OMA is only capable of offering a power gain of $\ln\left( M \right)$ owing to the MRC detection utilized at the BS and thus the asymptotic ESG in \eqref{C3:EPGMIMOERA6} suffers from a power reduction by a factor of $\ln\left( M \right)$ in the second term.
\section{ESG of \emph{m}MIMO-NOMA over \emph{m}MIMO-OMA}
In this section, we first derive the ergodic sum-rate of both \emph{m}MIMO-NOMA and \emph{m}MIMO-OMA and then discuss the asymptotic ESG of \emph{m}MIMO-NOMA over \emph{m}MIMO-OMA.

\subsection{Ergodic Sum-rate with $D>D_0$}
Let us now apply NOMA to massive MIMO systems, where a large-scale antenna array ($M \to \infty$) is employed at the BS and all the $K$ users are equipped with a single antenna.
A simple MRC-SIC receiver is adopted at the BS for data detection of \emph{m}MIMO-NOMA.
The instantaneous achievable data rate of user $k$ and the sum-rate of the \emph{m}MIMO-NOMA system using the MRC-SIC reception are given by
\begin{align}
R_{k}^{\mathrm{\emph{m}MIMO-NOMA}} &= {\ln}\left(1+ {\frac{{{p_k}{{\left\| {{{\bf{h}}_k}} \right\|}^2}}}{{\sum\limits_{i = k + 1}^K {p_i}{{\left\| {{{\bf{h}}_i}} \right\|}^2} {{\left| {\bf{e}}_k^{\mathrm{H}}{{\bf{e}}_i} \right|}^2}  + {N_0}}}} \right)\;\text{and} \label{C3:mMIMONOMAIndividualAchievableRate}\\
R_{\mathrm{sum}}^{\mathrm{\emph{m}MIMO-NOMA}} &= \sum\limits_{k = 1}^K R_{k}^{\mathrm{\emph{m}MIMO-NOMA}},\label{C3:mMIMONOMASumRate}
\end{align}
respectively, where ${{\bf{e}}_k} = \frac{{{{\bf{h}}_k}}}{{\left\| {{{\bf{h}}_k}} \right\|}}$ denotes the channel direction of user $k$.
For the massive MIMO system associated with $D>D_0$, the asymptotic ergodic sum-rate of ${K \rightarrow \infty}$ and ${M \rightarrow \infty}$ is given in the following theorem.

\begin{Thm}\label{C3:Theorem2}
For the \emph{m}MIMO-NOMA system considered in \eqref{C3:MIMONOMASystemModel} in conjunction with $D>D_0$ and MRC-SIC detection at the BS, under the equal resource allocation strategy, i.e., ${{p_{k}}} = \frac{{P_{\mathrm{max}}}}{K}$, $\forall k$, the asymptotic ergodic sum-rate can be approximated by
\begin{align}\label{C3:ErgodicSumRatemMIMONOMA}
&\mathop {\lim }\limits_{K \rightarrow \infty, M \rightarrow \infty} \overline{R_{\mathrm{sum}}^{\mathrm{\emph{m}MIMO-NOMA}}} = \mathop {\lim }\limits_{K \rightarrow \infty, M \rightarrow \infty} {{\mathrm{E}}_{\mathbf{H}}}\left\{ {R_{\mathrm{sum}}^{\mathrm{\emph{m}MIMO-NOMA}}} \right\} \\
&\approx \mathop {\lim }\limits_{K \to \infty ,M \to \infty } \sum\limits_{k = 1}^K \left( {\begin{array}{*{20}{c}}
K\\
k
\end{array}} \right){\frac{{k}}{{D + {D_0}}}} \sum\limits_{n = 1}^N {{\beta _n}\ln \left( {1 + \frac{{{\psi _k}}}{{{c_n}}}} \right)} {\left( {\frac{{\phi _n^2 \hspace{-1mm}-\hspace{-1mm} D_0^2}}{{{D^2} \hspace{-1mm}-\hspace{-1mm} D_0^2}}} \right)^{k - 1}}{\left( {\frac{{{D^2} \hspace{-1mm}-\hspace{-1mm} \phi _n^2}}{{{D^2} \hspace{-1mm}-\hspace{-1mm} D_0^2}}} \right)^{K - k}},\notag
\end{align}
with
\begin{align}\label{C3:ParametersformMIMONOMA}
\phi_n &= {\frac{D-D_0}{2}\cos \frac{{2n - 1}}{{2N}}\pi  + \frac{D+D_0}{2}},
{\psi _k} = \frac{{{P_{\mathrm{max}}}M}}{{\sum\limits_{i = k + 1}^K {{P_{\mathrm{max}}}{I_i} + K{N_0}} }}, \text{and} \\
{I_k} & = {{\mathrm{E}}_{{d_k}}}\left\{ {\frac{1}{{1 + d_k^\alpha }}} \right\} = \left( {\begin{array}{*{20}{c}}
K\\
k
\end{array}} \right){\frac{{k}}{{D + {D_0}}}} \sum\limits_{n = 1}^N \frac{{\beta _n}}{c_n} {\left( {\frac{{\phi _n^2 \hspace{-1mm}-\hspace{-1mm} D_0^2}}{{{D^2} \hspace{-1mm}-\hspace{-1mm} D_0^2}}} \right)^{k - 1}}{\left( {\frac{{{D^2} \hspace{-1mm}-\hspace{-1mm} \phi _n^2}}{{{D^2} \hspace{-1mm}-\hspace{-1mm} D_0^2}}} \right)^{K - k}}.\notag
\end{align}
\end{Thm}
\begin{proof}
Please refer to Appendix \ref{C3:AppendixB} for the proof of Theorem \ref{C3:Theorem2}.
\end{proof}

%Note that the asymptotic ergodic sum-rate in \eqref{C3:ErgodicSumRatemMIMONOMA} is applicable to an arbitrary power allocation.
%
%With equal power allocation, i.e., ${{p_{k}}} = \frac{{P_{\mathrm{max}}}}{K}$, the asymptotic ergodic sum-rate is give in \eqref{C3:ErgodicSumRatemMIMONOMA} with
%\begin{equation}
%{\psi _k} = \frac{{{P_{\mathrm{max}}}M}}{{\sum\limits_{i = k + 1}^K {{P_{\mathrm{max}}}{I_i} + K{N_0}} }}.
%\end{equation}

For the \emph{m}MIMO-OMA system using the FDMA-MRC detection, we can allocate more than one user to each frequency subband due to the above-mentioned favorable propagation property\cite{Ngo2013}.
In particular, upon allocating $W = \varsigma M$ users to each frequency subband with $\varsigma = \frac{W}{M} \ll 1$, the orthogonality among channel vectors of the $W$ users holds fairly well, hence the IUI becomes negligible.
Therefore, a random user grouping strategy is adopted, where we randomly select $W = \varsigma M$ users as a group and there are $G = \frac{K}{W}$ groups\footnote{Without loss of generality, we consider that $K$ is an integer multiple of $G$ and $W$.} separated using orthogonal frequency subbands.
%
%Denote ${\bf{H}}_g = \left[ {{{\bf{h}}_{(g-1)W+1}}, \ldots ,{{\bf{h}}_{gW}}} \right] \in \mathbb{C}^{ M \times W}$ as the composite channel matrix for the $g$-th group with $g \in \{1,\ldots,G\}$.
%
In each subband, low-complexity MRC detection can be employed for each individual user and thus the instantaneous achievable data rate of user $k$ can be expressed by
\begin{equation}\label{C3:mMIMOOMAIndividualAchievableRate}
R_{k}^{\mathrm{\emph{m}MIMO-OMA}} = f_g{\ln}\left(1+ {\frac{{{p_k}{{\left\| {{{\bf{h}}_k}} \right\|}^2}}}{{{f_gN_0}}}} \right),
\end{equation}
where $f_g$ denotes the normalized frequency allocation of the $g$-th group.
Note that \eqref{C3:mMIMOOMAIndividualAchievableRate} serves as an upper bound of the instantaneous achievable data rate of user $k$ in the \emph{m}MIMO-OMA system, since we assumed it to be IUI-free.
Then, under the equal resource allocation strategy, i.e., ${{p_{k}}} = \frac{{P_{\mathrm{max}}}}{K}$ and $f_g = 1/G = \frac{W}{K} = \delta\varsigma$, we have the asymptotic ergodic sum-rate of the \emph{m}MIMO-OMA system associated with $D > D_0$ as follows:
\begin{align}\label{C3:ErgodicSumRatemMIMOOMA}
&\mathop {\lim }\limits_{M \rightarrow \infty} \overline{R_{\mathrm{sum}}^{\mathrm{\emph{m}MIMO-OMA}}} = \mathop {\lim }\limits_{M \rightarrow \infty} {{\mathrm{E}}_{\mathbf{H}}}\left\{ {R_{\mathrm{sum}}^{\mathrm{\emph{m}MIMO-OMA}}} \right\} \\
&= \mathop {\lim }\limits_{M \to \infty } \delta\varsigma\sum\limits_{k = 1}^K \left( {\begin{array}{*{20}{c}}
K\\
k
\end{array}} \right){\frac{{k}}{{D + {D_0}}}} \sum\limits_{n = 1}^N {{\beta _n}\ln \left( {1 + \frac{{{\xi}}}{{{c_n}}}} \right)} {\left( {\frac{{\phi _n^2 - D_0^2}}{{{D^2} - D_0^2}}} \right)^{k - 1}}{\left( {\frac{{{D^2} - \phi _n^2}}{{{D^2} - D_0^2}}} \right)^{K - k}},\notag
\end{align}
where $\phi_n$ is given in \eqref{C3:ParametersformMIMONOMA} and ${\xi} = \frac{{{P_{\mathrm{max}}}}}{\varsigma N_0}$.
%
%Note that the approximation in \eqref{C3:ErgodicSumRatemMIMOOMA} becomes exact when there is no grouping in \emph{m}MIMO-OMA, since we can safely ignore the IUI with $W = 1$.

\subsection{Ergodic Sum-rate with $D=D_0$}
We note that the analytical results in \eqref{C3:ErgodicSumRatemMIMONOMA} and \eqref{C3:ErgodicSumRatemMIMOOMA} are only applicable to the system having $D>D_0$.
The asymptotic ergodic sum-rate of the \emph{m}MIMO-NOMA system with $D=D_0$ can be expressed using the following theorem.

\begin{Thm}\label{C3:Theorem3}
With $D = D_0$ and the equal resource allocation strategy, i.e., ${{p_{k}}} = \frac{{P_{\mathrm{max}}}}{K}$ and $f_g = 1/G = \frac{W}{K}= \delta\varsigma$, the asymptotic ergodic sum-rate of the \emph{m}MIMO-NOMA system and of the \emph{m}MIMO-OMA system can be formulated by
\begin{align}
\mathop {\lim }\limits_{K \rightarrow \infty, M \rightarrow \infty} \overline{R_{\mathrm{sum}}^{\mathrm{\emph{m}MIMO-NOMA}}}
&\approx \mathop {\lim }\limits_{K \to \infty ,M \to \infty } \frac{{M}}{\varpi\delta} \left[ \ln \left( 1 + \varpi\delta + \varpi \right) \left( 1 + \varpi\delta + \varpi \right) \right. \notag\\
&- \left. \ln \left( 1 + \varpi\delta \right)\left( 1 + \varpi\delta \right) - \ln \left( 1 + \varpi \right)\left( 1 + \varpi \right)\right] \;\;\text{and}\label{C3:DD0ErgodicSumRatemMIMONOMA}\\
\mathop {\lim }\limits_{M \rightarrow \infty} \overline{R_{\mathrm{sum}}^{\mathrm{\emph{m}MIMO-OMA}}}
&= \mathop {\lim }\limits_{M \to \infty } {\varsigma M} \ln \left( {1 + \frac{\varpi}{\varsigma}} \right),\label{C3:DD0ErgodicSumRatemMIMOOMA}
\end{align}
respectively, where $\delta = \frac{M}{K}$ and $\varsigma = \frac{W}{M}$ are constants and $\varpi = \frac{{P_{\mathrm{max}}}}{\left({1+D_0^{\alpha}}\right) N_0}$ denotes the total average received SNR of all the users.
\end{Thm}
\begin{proof}
Please refer to Appendix \ref{C3:AppendixC} for the proof of Theorem \ref{C3:Theorem3}.
\end{proof}

\subsection{ESG in Massive-antenna Systems}
Based on \eqref{C3:ErgodicSumRatemMIMONOMA} and \eqref{C3:ErgodicSumRatemMIMOOMA}, when $D>D_0$, the asymptotic ESG of \emph{m}MIMO-NOMA over \emph{m}MIMO-OMA associated with ${K \to \infty}$ and ${M \to \infty}$ can be expressed as follows:
\begin{align}\label{C3:DD0EPGmMIMOERA0}
&\mathop {\lim }\limits_{K \to \infty, {M \to \infty}}  \overline{G^{\mathrm{\emph{m}MIMO}}_{D > D_0}} \notag\\
&\approx \mathop {\lim }\limits_{K \to \infty ,M \to \infty } \sum\limits_{k = 1}^K {\left( \hspace{-1mm} {\begin{array}{*{20}{c}}
K\\
k
\end{array}} \hspace{-1mm}\right)} \frac{k}{{D \hspace{-1mm}+\hspace{-1mm} {D_0}}}\sum\limits_{n = 1}^N {{\beta _n}\left[ {\ln \left( {1 \hspace{-1mm}+\hspace{-1mm} \frac{{{\psi _k}}}{{{c_n}}}} \right) - \delta\varsigma\ln \left( {1 \hspace{-1mm}+\hspace{-1mm} \frac{{{\xi}}}{{{c_n}}}} \right)} \right]}  \notag\\
&\times {\left( {\frac{{\phi _n^2 - D_0^2}}{{{D^2} - D_0^2}}} \right)^{k - 1}}{\left( {\frac{{{D^2} - \phi _n^2}}{{{D^2} - D_0^2}}} \right)^{K - k}}.
\end{align}
However, the expression in \eqref{C3:DD0EPGmMIMOERA0} is too complicated and does not provide immediate insights.
Hence, we focus on the case of $D = D_0$ to unveil some important and plausible insights on the ESG of NOMA over OMA in the massive MIMO system.
The simulation results of Section \ref{C3:Simulations} will show that the insights obtained from the case of $D=D_0$ are also applicable to the general scenario of $D>D_0$.

Comparing \eqref{C3:DD0ErgodicSumRatemMIMONOMA} and \eqref{C3:DD0ErgodicSumRatemMIMOOMA}, when $D=D_0$, we have the asymptotic ESG of \emph{m}MIMO-NOMA over \emph{m}MIMO-OMA for ${K \to \infty}$ and ${M \to \infty}$ as follows:
\begin{align}\label{C3:DD0EPGmMIMOERA}
\hspace{-3mm}\mathop {\lim }\limits_{K \to \infty, {M \to \infty}}  \overline{G^{\mathrm{\emph{m}MIMO}}_{D = D_0}} &= \mathop {\lim }\limits_{K \to \infty, {M \to \infty}}  \overline{R_{\mathrm{sum}}^{{\mathrm{\emph{m}MIMO-NOMA}}}} - \overline{R_{\mathrm{sum}}^{{\mathrm{\emph{m}MIMO-OMA}}}} \\
& \approx \mathop {\lim }\limits_{K \to \infty ,M \to \infty } \frac{{M}}{\varpi\delta} \left[ \ln \left( 1 + \varpi\delta + \varpi \right) \left( 1 + \varpi\delta + \varpi \right) \right. \notag\\
& - \left.\ln \left( 1 + \varpi\delta \right)\left( 1 + \varpi\delta \right) - \ln \left( 1 + \varpi \right)\left( 1 + \varpi \right)\right] -  {\varsigma M} \ln \left( {1 + \frac{\varpi}{\varsigma}} \right).\notag
\end{align}
In the low-SNR regime, we can observe that $\mathop {\lim }\limits_{K \to \infty, {M \to \infty}, {P_{\mathrm{max}}} \to 0}  \overline{G^{\mathrm{\emph{m}MIMO}}_{D = D_0}} \to 0$.
This implies that no gain can be achieved by NOMA in the low-SNR regime, which is consistent with \eqref{C3:EPGMIMOERA5}.
By contrast, in the high-SNR regime, we have
\begin{align}
\mathop {\lim }\limits_{K \to \infty, {M \to \infty}, {P_{\mathrm{max}}} \to \infty}  \overline{G^{\mathrm{\emph{m}MIMO}}_{D = D_0}}
&\approx \mathop {\lim }\limits_{K \to \infty ,M \to \infty ,{P_{{\mathrm{max}}}} \to \infty } \frac{{M}}{\delta} \zeta  -  {\varsigma M} \ln \left( {1 + \frac{\varpi}{\varsigma}} \right) \label{C3:DD0EPGmMIMOERA2}\\
&= \mathop {\lim }\limits_{K \to \infty ,M \to \infty ,{P_{{\mathrm{max}}}} \to \infty } K \zeta  -  {\delta\varsigma K} \ln \left( {1 + \frac{\varpi}{\varsigma}} \right),\label{C3:DD0EPGmMIMOERA3}
\end{align}
where $\zeta = \left[ \ln \left( {1 \hspace{-0.5mm}+\hspace{-0.5mm} \varpi \delta  \hspace{-0.5mm}+\hspace{-0.5mm} \varpi } \right)\left( {1 \hspace{-0.5mm}+\hspace{-0.5mm} \delta } \right) \hspace{-0.5mm}-\hspace{-0.5mm} \ln \left( {1 \hspace{-0.5mm}+\hspace{-0.5mm} \varpi \delta } \right)\delta  \hspace{-0.5mm}-\hspace{-0.5mm} \ln \left( {1 \hspace{-0.5mm}+\hspace{-0.5mm} \varpi } \right) \right]$ represents the extra ergodic sum-rate gain upon supporting an extra user by the \emph{m}MIMO-NOMA system considered.
Explicitly, for $K \to \infty$, ${M \to \infty}$, and ${P_{\mathrm{max}}} \rightarrow \infty$, the resultant extra benefit $\zeta$ is jointly determined by the average received sum SNR $\varpi$ and the fixed ratio $\delta$.
Observe in \eqref{C3:DD0EPGmMIMOERA2} and \eqref{C3:DD0EPGmMIMOERA3} that given the average received sum SNR $\varpi$ and the fixed ratios $\delta$ and $\varsigma$, the asymptotic ESG scales linearly with both the number of UL receiver antennas at the BS, $M$ and the number of users, $K$, respectively.
In other words, the asymptotic ESG per user and the asymptotic ESG per antenna of \emph{m}MIMO-NOMA over \emph{m}MIMO-OMA are constant and they are given by
\begin{align}
\mathop {\lim }\limits_{K \to \infty, {M \to \infty}, {P_{\mathrm{max}}} \to \infty}  \frac{\overline{G^{\mathrm{\emph{m}MIMO}}_{D = D_0}}}{K} &= \zeta  -  {\delta\varsigma } \ln \left( {1 + \frac{\varpi}{\varsigma}} \right)\;\text{and} \label{C3:DD0EPGPerUsermMIMOERA2}\\
\mathop {\lim }\limits_{K \to \infty, {M \to \infty}, {P_{\mathrm{max}}} \to \infty}  \frac{\overline{G^{\mathrm{\emph{m}MIMO}}_{D = D_0}}}{M} &= \frac{\zeta}{\delta} -  {\varsigma} \ln \left( {1 + \frac{\varpi}{\varsigma}} \right),\; \label{C3:DD0EPGPerAntennamMIMOERA3}
\end{align}
respectively.
We can explain this observation from the spatial DoF perspective, since it determines the pre-log factor for the ergodic sum-rate of both \emph{m}MIMO-NOMA and \emph{m}MIMO-OMA and thus also determines the pre-log factor of the corresponding ESG.
In particular, the \emph{m}MIMO-NOMA system considered is basically an $(M \times K)$ MIMO system associated with $M<K$, since all the $K$ users transmit their signals simultaneously in the same frequency band.
When scaling up the \emph{m}MIMO-NOMA system while maintaining a fixed ratio $\delta = \frac{M}{K}$, the system's spatial DoF increases linearly  both with $M$ and $K$.
On the other hand, the spatial DoF of the \emph{m}MIMO-OMA system is limited by the group size $W$, since it is always an $(M \times W)$ MIMO system associated with $W \ll M$ in each time slot and frequency subband.
Therefore, the system's spatial DoF increases linearly with both $W$, and $M$, as well as $K$, when scaling up the \emph{m}MIMO-OMA system under fixed ratios of $\delta = \frac{M}{K}$ and $\varsigma = \frac{W}{M}$.
As a result, due to the linear increase of the spatial DoF with $M$ as well as $K$ for both the \emph{m}MIMO-NOMA and \emph{m}MIMO-OMA systems, the asymptotic ESG increases linearly with both $M$ and $K$.
Note that in contrast to \eqref{C3:EPGMIMOERA6}, there is no DoF gain, despite the fact that the asymptotic ESG scales linearly both with $M$ as well as $K$.
This is because the extra benefit $\zeta$ does not increase linearly with the system's SNR in dB.
As a result, the asymptotic ESG of \emph{m}MIMO-NOMA over \emph{m}MIMO-OMA cannot increase linearly with the system's SNR in dB, as it will be shown in Section \ref{C3:Simulations}.

\section{ESG in Multi-cell Systems}\label{C3:DiscussionsMulticell}
\begin{figure}[t]
\centering
\includegraphics[width=4in]{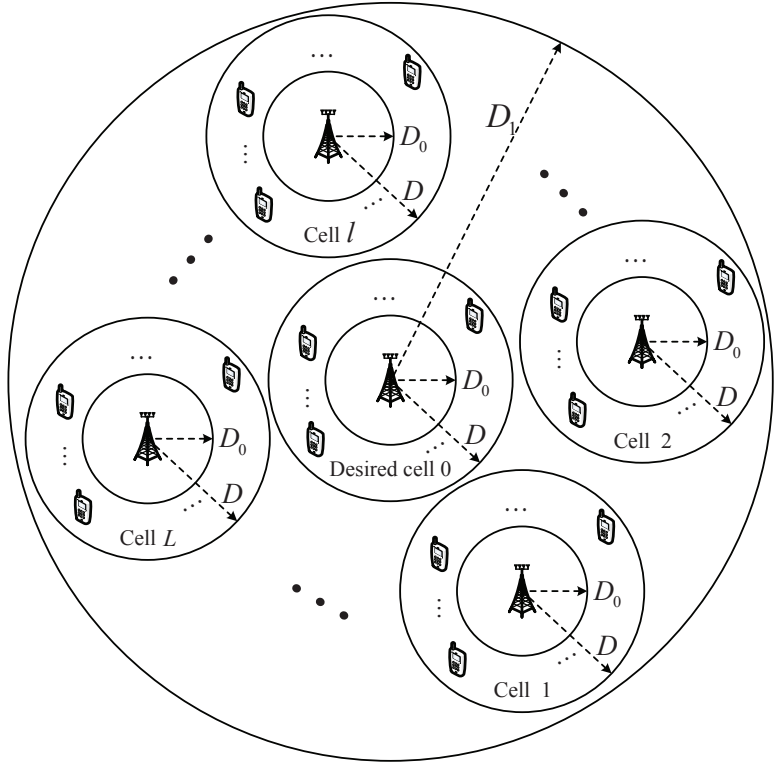}
\caption{The system model of the multi-cell uplink communication with one serving cell and $L$ adjacent cells.}
\label{C3:Multi-cellNOMA_Uplink_Model}
\end{figure}

In previous sections, the performance gain of NOMA over OMA has been investigated in single-cell systems, since these analytical results are easily comprehensible and reveal directly plausible insights.
Naturally, the performance gain of NOMA over OMA in single-cell systems serves as an upper bound on that of non-cooperative multi-cell systems, which can be approached by employing conservative frequency reuse strategy.
In practice, cellular networks consist of multiple cells where the inter-cell interference (ICI) is inevitable.
Furthermore, the characteristics of the ICI for NOMA and OMA schemes are different.
In particular, ICI is imposed by all the users in adjacent cells for NOMA schemes, while only a subset of users inflict ICI in OMA schemes, as an explicit benefit of orthogonal time or frequency allocation.
As a result, NOMA systems face more severe ICI than that of OMA, hence it remains unclear, if applying NOMA is still beneficial in multi-cell systems.
Therefore, in this section, we investigate the ESG of NOMA over OMA in multi-cell systems.
\subsection{Inter-cell Interference in NOMA and OMA Systems}
Consider a multi-cell system having multiple non-overlapped adjacent cells with index $l = 1,\ldots,L$, which are randomly deployed and  surround the serving cell $l=0$, as shown in Figure \ref{C3:Multi-cellNOMA_Uplink_Model}.
We assume that the $L$ interfering cells have the same structure as the serving cell and they are uniformly distributed in a ring-shaped disc of Figure \ref{C3:Multi-cellNOMA_Uplink_Model} having the inner radius of $D$ and outer radius of $D_1$.
Furthermore, we adopt the radical frequency reuse factor of 1, i.e., using the same frequency band for all cells to facilitate the performance analysis\footnote{With a less-aggressive frequency reuse strategy in multi-cell systems, both NOMA and OMA schemes endure less ICI since only the adjacent cells using the same frequency band with the serving cell are taken into account. As a result, the performance analyses derived in this chapter can be extended to the case with a lower frequency reuse ratio by simply decreasing number of adjacent cells $L$. Again, the resultant performance will then approach the performance upper-bound of the single-cell scenario.}.
Again, we are assuming that in each cell there is a single $M$-antenna BS serving $K$ single-antenna users in the UL and thus there are $KL$ users imposing interference on the serving BS.
Additionally, to reduce both the system's overhead and its complexity, no cooperative multi-cell processing is included in our multi-cell system considered.
In the following, we first investigate the resultant ICI distribution and then derive the total received ICI power contaminating over NOMA and OMA systems.

Given the normalized UL receive beamforming vector of user $k$ in the serving cell at the serving BS represented by $\mathbf{w}_k \in \mathbb{C}^{M \times 1}$ with $\left\|\mathbf{w}_k\right\|^2 = 1$, the effective ICI channel spanning from user $k'$ in adjacent cell $l$ to the serving BS can be formulated as:
\begin{equation}\label{C3:MulticelleEffectiveChannelModel}
{h_{k',l}} = {\bf{w}}_k^{\mathrm{H}}{\mathbf{h}_{k',l}} = \frac{{\bf{w}}_k^{\mathrm{H}}{\bf{g}}_{k',l}}{\sqrt{1+d_{k',l}^{\alpha}}},
\end{equation}
where ${\mathbf{h}_{k',l}} = \frac{{\bf{g}}_{k',l}}{\sqrt{1+d_{k',l}^{\alpha}}}$ denotes the channel vector from user $k'$ in adjacent cell $l$ to the serving BS, ${\bf{g}}_{k',l} \in \mathbb{C}^{ M \times 1}$ represents the Rayleigh fading coefficients, i.e., ${\bf{g}}_{k',l} \sim \mathcal{CN}\left(\mathbf{0},{{\bf{I}}_M}\right)$, and $d_{k',l}$ denotes the distance between user $k'$ in adjacent cell $l$ and the serving BS with the unit of meter.
Similar to the single-cell system considered, we assume that the CSI of all the users within the serving cell is perfectly known at the serving BS.
However, the ICI channel is unknown for the serving BS.
Note that the receive beamformer ${\bf{w}}_k$ of the serving BS depends on the instantaneous channel vector of user $k$, ${\mathbf{h}_{k}}$, and/or on the multiple access interference structure $\left[ {{{\bf{h}}_1}, \ldots ,{{\bf{h}}_{k - 1}},{{\bf{h}}_{k + 1}}, \ldots ,{{\bf{h}}_K}} \right]$ in the serving cell.
Therefore, the receive beamformer ${\bf{w}}_k$ of the serving cell is independent of the ICI channel ${\bf{g}}_{k',l}$.
As a result, owing to $\left\|\mathbf{w}_k\right\|^2 = 1$, it can be readily observed that ${{\bf{w}}_k^{\mathrm{H}}{\bf{g}}_{k',l}}$ obeys the circularly symmetric complex Gaussian distribution conditioned on the given $\mathbf{w}_k$, i.e., we have ${\bf{w}}_k^{\mathrm{H}}{{\bf{g}}_{k',l}}\left| {_{{{\bf{w}}_k}}} \right. \sim \mathcal{CN}\left(0,1\right)$.
However, since the resultant distribution $\mathcal{CN}\left(0,1\right)$ is independent of ${\bf{w}}_k$, we can safely drop the condition and directly apply ${\bf{w}}_k^{\mathrm{H}}{{\bf{g}}_{k',l}} \sim \mathcal{CN}\left(0,1\right)$.
Now, based on \eqref{C3:MulticelleEffectiveChannelModel}, we can observe that the effective ICI channel ${h_{k',l}}$ is equivalent to a single-antenna Rayleigh fading channel associated with a distance of $d_{k',l}$, regardless of how many antennas are employed at the serving BS.

Given that each user is equipped with a single-antenna, the transmission of each user is omnidirectional.
Therefore, to facilitate the analysis of the ICI power, we assume that there is no gap between the adjacent cells and that the inner radius of each adjacent cell is zero, i.e., $D_0 = 0$.
Hence, we can further assume that the ICI emanates from $KL$ users uniformly distributed within the ring-shaped disc having the inner radius of $D$ and outer radius of $D_1$.
Similar to \eqref{C3:SISOChannelDistributionPDF} and \eqref{C3:SISOChannelDistributionCDF}, the CDF and PDF of ${{\left| {h_{k',l}} \right|}^2}$ are given by
\begin{align}
{F_{{{\left| {h_{k',l}} \right|}^2}}}\left( x \right) &\approx 1 - \frac{1}{D+D_1}\sum\limits_{n = 1}^N {{\beta' _n}{e^{ - {c'_n}x}}} \; \text{and} \label{C3:MulticellEffectiveChannelDistributionPDF}\\
{f_{{{\left| {h_{k',l}} \right|}^2}}}\left( x \right) &\approx \frac{1}{D+D_1}\sum\limits_{n = 1}^N {{\beta' _n}{c'_n}{e^{ - {c'_n}x}}}, x \ge 0, \forall k',l \label{C3:MulticellEffectiveChannelDistributionCDF}
\end{align}
respectively, with parameters of
\begin{align}\label{C3:MulticellEffectiveChannelBetaCn}
{\beta'_n} &= \frac{\pi }{N}\left| {\sin \frac{{2n \hspace{-1mm}-\hspace{-1mm} 1}}{{2N}}\pi } \right|\left( {\frac{D_1\hspace{-1mm}-\hspace{-1mm}D}{2}\cos \frac{{2n \hspace{-1mm}-\hspace{-1mm} 1}}{{2N}}\pi  + \frac{D_1\hspace{-1mm}+\hspace{-1mm}D}{2}} \right) \;\text{and}\notag\\
{c'_n} &= 1 + {\left( {\frac{D_1\hspace{-1mm}-\hspace{-1mm}D}{2}\cos \frac{{2n \hspace{-1mm}-\hspace{-1mm} 1}}{{2N}}\pi  + \frac{D_1\hspace{-1mm}+\hspace{-1mm}D}{2}} \right)^\alpha}.
\end{align}
Note that all the adjacent cell users have i.i.d. channel distributions since we ignore the adjacent cells' structure.

Due to the ICI encountered in unity-frequency-reuse multi-cell systems, the performance is determined by the signal-to-interference-plus-noise ratio (SINR) instead of the SNR of single-cell systems.
Assuming that the ICI is treated as AWGN by the detector, the system's SINR can be defined as follows:
\begin{equation}\label{C3:MulticellSINR}
{\rm{SINR}_{sum}^{multicell}} = \frac{{P_{\mathrm{max}}}}{{I_{{\mathrm{inter}}}} + N_0} {\overline{{{\left| {{{h}}} \right|}^2}}},
\end{equation}
where ${I_{{\mathrm{inter}}}}$ characterizes the ICI power in multi-cell systems and ${P_{\mathrm{max}}}$ denotes the same system power budget in each single cell.

To facilitate our performance analysis, we assume that the equal resource allocation strategy is adopted in all the adjacent cells, i.e., $p_{k',l} = \frac{{P_{\mathrm{max}}}}{K}$, $\forall k',l$.
When invoking NOMA in a multi-cell system, the ICI power can be modeled as
\begin{equation}\label{C3:MulticellInterference}
{I^{\mathrm{NOMA}}_{{\mathrm{inter}}}} = \sum\limits_{l = 1}^L {\sum\limits_{k' = 1}^K {\frac{{{P_{{\mathrm{max}}}}}}{K}} } {\left| {{h_{k',l}}} \right|^2}.
\end{equation}
For $KL \to \infty$, ${I^{\mathrm{NOMA}}_{{\mathrm{inter}}}}$ becomes a deterministic value, which can be approximated by
\begin{equation}\label{C3:MulticellInterference2}
\mathop {\lim }\limits_{KL \to \infty} {I^{\mathrm{NOMA}}_{{\mathrm{inter}}}} \approx L {P_{{\mathrm{max}}}} \overline{{\left| {{h_{k',l}}} \right|^2}} \approx \frac{L {P_{{\mathrm{max}}}} }{D+D_1}\sum\limits_{n = 1}^N {\frac{{{\beta' _n}}}{{{c'_n}}}}.
\end{equation}
As a result, the SINR of the multi-cell NOMA system considered is given by
\begin{equation}\label{C3:MulticellSINRNOMA}
{\rm{SINR}_{sum,NOMA}^{multicell}} = \frac{{P_{\mathrm{max}}}}{\frac{L {P_{{\mathrm{max}}}} }{D+D_1}\sum\limits_{n = 1}^N {\frac{{{\beta' _n}}}{{{c'_n}}}} + N_0} {\overline{{{\left| {{{h}}} \right|}^2}}}.
\end{equation}

For OMA schemes, we assume that all the $K$ users in each cell are clustered into $G$ groups, with each group allocated to a frequency subband exclusively.
Since only $\frac{1}{G}$ of users in each adjacent cell are simultaneously transmitting their signals in each frequency subband, the ICI power in a multi-cell OMA system can be expressed as:
\begin{equation}\label{C3:MulticellInterference3}
\mathop {\lim }\limits_{KL \to \infty} {I^{\mathrm{OMA}}_{{\mathrm{inter}}}} = \frac{1}{G} \mathop {\lim }\limits_{KL \to \infty} {I^{\mathrm{NOMA}}_{{\mathrm{inter}}}} \approx
\frac{L {P_{{\mathrm{max}}}}}{G(D+D_1)}\sum\limits_{n = 1}^N {\frac{{{\beta' _n}}}{{{c'_n}}}}.
\end{equation}
The SINR of the multi-cell OMA system considered can be written as:
\begin{equation}\label{C3:MulticellSINROMA}
{\rm{SINR}_{sum,OMA}^{multicell}} = \frac{{P_{\mathrm{max}}}}{\frac{L {P_{{\mathrm{max}}}} }{G(D+D_1)}\sum\limits_{n = 1}^N {\frac{{{\beta' _n}}}{{{c'_n}}}} + \frac{1}{G} N_0} {\overline{{{\left| {{{h}}} \right|}^2}}}.
\end{equation}
Note that we have $G = K$ for SISO-OMA and MIMO-OMA with FDMA-MRC, $G = \frac{K}{M}$ for MIMO-OMA with FDMA-ZF, and $G = \frac{K}{W}$ for \emph{m}MIMO-OMA with FDMA-MRC.

\subsection{ESG in Multi-cell Systems}
It can be observed that ${I^{\mathrm{NOMA}}_{{\mathrm{inter}}}}$ in \eqref{C3:MulticellInterference2} and ${I^{\mathrm{OMA}}_{{\mathrm{inter}}}}$ in \eqref{C3:MulticellInterference3} are independent of the number of antennas employed at the serving BS, which is due to the non-coherent combining used at the serving BS ${\bf{w}}_k^{\mathrm{H}}{\bf{g}}_{k',l}$, thereby leading to the effective ICI channel becoming equivalent to a single-antenna Rayleigh fading channel.
Therefore, all the ergodic sum-rates of NOMA in single-antenna, multi-antenna, and massive MIMO single-cell systems are degraded upon replacing the noise power $N_0$ by $({{I^{\mathrm{NOMA}}_{{\mathrm{inter}}}} + N_0})$.
On the other hand, since OMA schemes only face a noise power level of $\frac{1}{G}N_0$ on each subband, all the ergodic sum-rates of the OMA schemes in single-antenna, multi-antenna, and massive MIMO single-cell systems are reduced upon substituting the noise power $\frac{1}{G}N_0$ by ${{I^{\mathrm{OMA}}_{{\mathrm{inter}}}} + \frac{1}{G}N_0} = \frac{1}{G} \left( {{I^{\mathrm{NOMA}}_{{\mathrm{inter}}}} \hspace{-1mm}+\hspace{-1mm} {N_0}} \right)$.

Given the ICI terms ${I^{\mathrm{NOMA}}_{{\mathrm{inter}}}}$ and ${I^{\mathrm{OMA}}_{{\mathrm{inter}}}}$, we have the corresponding asymptotic ESGs in single-antenna, multi-antenna, and massive MIMO multi-cell systems as follows:
\begin{align}
\mathop {\lim }\limits_{K \rightarrow \infty} \overline{G^{{\mathrm{SISO}}}}'
&\approx \ln \left( {1 + \frac{{{P_{{\mathrm{max}}}}}}{{\left( {D + {D_0}} \right)\left( {{I^{\mathrm{NOMA}}_{{\mathrm{inter}}}} \hspace{-1mm}+\hspace{-1mm} {N_0}} \right)}}\sum\limits_{n = 1}^N {\frac{{{\beta _n}}}{{{c_n}}}} } \right)\label{C3:MultiCellEPGSISOERA2}\\
&-\frac{1}{{\left( {D + {D_0}} \right)}}\sum\limits_{n = 1}^N {{\beta _n}{e^{\frac{{{c_n}\left( {{I^{\mathrm{NOMA}}_{{\mathrm{inter}}}} + {N_0}} \right)}}{{{P_{{\mathrm{max}}}}}}}}{{\cal E}_1}\left( {\frac{{{c_n}\left( {{I^{\mathrm{NOMA}}_{{\mathrm{inter}}}} +{N_0}} \right)}}{{{P_{{\mathrm{max}}}}}}} \right)}, \notag\\
\mathop {\lim }\limits_{K \rightarrow \infty}  \overline{G^{{\mathrm{MIMO}}}_{\mathrm{FDMA-ZF}}}'
&\approx M\ln \left( {1 + \frac{{{P_{{\mathrm{max}}}}}}{{\left( {D + {D_0}} \right)\left( {{I^{\mathrm{NOMA}}_{{\mathrm{inter}}}} \hspace{-1mm}+\hspace{-1mm} {N_0}} \right)}}\sum\limits_{n = 1}^N {\frac{{{\beta _n}}}{{{c_n}}}} } \right)\label{C3:MultiCellEPGMIMOERA2}\\
&-\frac{M}{{\left( {D + {D_0}} \right)}}\sum\limits_{n = 1}^N {{\beta _n}{e^{\frac{{{c_n}M\left( {{I^{\mathrm{NOMA}}_{{\mathrm{inter}}}} + {N_0}} \right)}}{{{P_{{\mathrm{max}}}}}}}}{{\cal E}_1}\left( {\frac{{{c_n}M\left( {{I^{\mathrm{NOMA}}_{{\mathrm{inter}}}} \hspace{-1mm}+\hspace{-1mm} {N_0}} \right)}}{{{P_{{\mathrm{max}}}}}}} \right)},\notag\\
\mathop {\lim }\limits_{K \rightarrow \infty}  \overline{G^{\mathrm{MIMO}}_{\mathrm{FDMA-MRC}}}' &\approx M\ln \left( {1 + \frac{{{P_{{\mathrm{max}}}}}}{{\left( {D + {D_0}} \right)\left( {{I^{\mathrm{NOMA}}_{{\mathrm{inter}}}} \hspace{-1mm}+\hspace{-1mm} {N_0}} \right)}}\sum\limits_{n = 1}^N {\frac{{{\beta _n}}}{{{c_n}}}} } \right)\label{C3:MultiCellEPGMIMOERAMRC}\\
&\hspace{-35mm}-\frac{1}{{\left( {D + {D_0}} \right)}}\sum\limits_{n = 1}^N {{\beta _n}} \left( {\frac{{{{\left( {\frac{{\left( {{I^{\mathrm{NOMA}}_{{\mathrm{inter}}}} + {N_0}} \right){c_n}}}{{{P_{{\mathrm{max}}}}}}} \right)}^M}}}{{\Gamma \left( M \right)}}G_{2,3}^{3,1}\left(\hspace{-2mm} {\begin{array}{*{20}{c}}
{ - M, - M + 1}\\
{ - M, - M,0}
\end{array}\hspace{-2mm}\left| {\frac{{\left( {{I^{\mathrm{NOMA}}_{{\mathrm{inter}}}} \hspace{-1mm}+\hspace{-1mm} {N_0}} \right){c_n}}}{{{P_{{\mathrm{max}}}}}}} \right.} \right)} \right),\notag\\
\mathop {\lim }\limits_{K \to \infty, {M \to \infty}}  \overline{G^{\mathrm{\emph{m}MIMO}}_{D > D_0}}' &\approx \mathop {\lim }\limits_{K \to \infty ,M \to \infty } \sum\limits_{k = 1}^K {\left( \hspace{-1mm} {\begin{array}{*{20}{c}}
K\\
k
\end{array}} \hspace{-1mm}\right)} \frac{k}{{D \hspace{-1mm}+\hspace{-1mm} {D_0}}}\sum\limits_{n = 1}^N {\beta _n} \label{C3:MultiCellDD0EPGmMIMOERA0}\\
&\hspace{-25mm}\times\left[ {\ln \left( {1 \hspace{-1mm}+\hspace{-1mm} \frac{{{\psi _k}'}}{{{c_n}}}} \right) - \delta\varsigma\ln \left( {1 \hspace{-1mm}+\hspace{-1mm} \frac{{{\xi}'}}{{{c_n}}}} \right)} \right] {\left( {\frac{{\phi _n^2 - D_0^2}}{{{D^2} - D_0^2}}} \right)^{k - 1}}{\left( {\frac{{{D^2} - \phi _n^2}}{{{D^2} - D_0^2}}} \right)^{K - k}}, \;\text{and} \notag\\
\mathop {\lim }\limits_{K \to \infty, {M \to \infty}}  \overline{G^{\mathrm{\emph{m}MIMO}}_{D = D_0}}'
&\approx \frac{{M}}{\varpi'\delta} \left[ \ln \left( 1 + \varpi'\delta + \varpi' \right) \left( 1 + \varpi'\delta + \varpi' \right) \right. \label{C3:MultiCellDD0EPGmMIMOERA}\\
&\hspace{-20mm} -\left.\ln \left( 1 + \varpi'\delta \right)\left( 1 + \varpi'\delta \right) - \ln \left( 1 + \varpi' \right)\left( 1 + \varpi' \right)\right] - {{\varsigma M}} \ln \left( {1 + \frac{\varpi'}{\varsigma}} \right),\notag
\end{align}
where ${\psi _k}^\prime  = \frac{{{P_{{\mathrm{max}}}}M}}{{\sum\limits_{i = k + 1}^K {{P_{{\mathrm{max}}}}{I_i} + K\left( {{I^{\mathrm{NOMA}}_{{\mathrm{inter}}}} + {N_0}} \right)} }}$, ${{\xi}'} = \frac{{{P_{{\mathrm{max}}}}M}}{{W\left( {{I^{\mathrm{NOMA}}_{{\mathrm{inter}}}} + {N_0}} \right)}}$, and $\varpi ' = \frac{{{P_{{\mathrm{max}}}}}}{{\left( {1 + D_0^\alpha } \right)\left( {{I^{\mathrm{NOMA}}_{{\mathrm{inter}}}} + {N_0}} \right)}}$.

It can be observed that compared to single-cell systems, the ESGs of NOMA over OMA in multi-cell systems are degraded due to the existence of ICI.
In particular, since the OMA schemes endure not only $\frac{1}{G}$ of noise power but also $\frac{1}{G}$ of ICI power, compared to NOMA schemes, the interference plus noise power of $({{I^{\mathrm{NOMA}}_{{\mathrm{inter}}}} + {N_0}})$ in multi-cell systems plays the same role as the noise power ${N_0}$ in single-cell systems.
Therefore, the performance analyses in single-cell systems are directly applicable to multi-cell systems via increasing the noise power ${N_0}$ to the interference plus noise power of $({{I^{\mathrm{NOMA}}_{{\mathrm{inter}}}} + {N_0}})$.
Upon utilizing the coordinate signal processing among multiple cells\cite{ShinMulticellNOMA}, the ICI power can be effectively suppressed, which may prevent the ESG degradation, when extending NOMA from single-cell to multi-cell systems.

\section{Simulations}\label{C3:Simulations}
In this section, we use simulations to evaluate our analytical results.
In the single-cell systems considered, the inner cell radius is $D_0 = 50$ m and the outer cell radius is given by $D = [50, 200, 500]$ m, which corresponds to the cases of normalized cell sizes given by $\eta = [1, 4, 10]$, respectively.
The number of users $K$ ranges from $2$ to $256$ and the number of antennas employed at the BS $M$ ranges from $1$ to $128$.
The path loss exponent is $\alpha = 3.76$ according to the 3GPP path loss model\cite{Access2010}.
The noise power is set as $N_0 = -80$ dBm.
To emphasize the effect of cell size on the ESG of NOMA over OMA, in the simulations of the single-cell systems, we characterize the system's SNR with the aid of the total average received SNR of all the users at the BS as follows\cite{Xu2017}:
\begin{equation}\label{C3:SystemSNR}
{\mathrm{SNR}_{\mathrm{sum}}} = \frac{{P_{\mathrm{max}}}}{N_0} {\overline{{{\left| {{{h}}} \right|}^2}}} = \frac{{P_{\mathrm{max}}}}{N_0} \frac{\overline{{{\left\| {{\mathbf{h}}} \right\|}^2}}}{M},
\end{equation}
where ${\overline{{{\left| {{{h}}} \right|}^2}}}$ and ${\overline{{{\left\| {{\mathbf{h}}} \right\|}^2}}}$ are given by \eqref{C3:Mean_ChannelPowerGainSISO} and \eqref{C3:Mean_ChannelPowerGainMIMO}, respectively.
The total transmit power ${P_{\mathrm{max}}}$ is adjusted adaptively for different cell sizes to satisfy ${\mathrm{SNR}_{\mathrm{sum}}}$ in \eqref{C3:SystemSNR} ranging from $0$ dB to $40$ dB.
In the \emph{m}MIMO-OMA system considered, we set the ratio between the group size and the number of antennas to $\varsigma = \frac{W}{M} = \frac{1}{16}$, hence we can assume that the favorable propagation conditions prevail in the spirit of \cite{Ngo2013}.
Additionally, in the \emph{m}MIMO-NOMA system considered, the ratio between the number of receiver antennas at the BS and the number of serving users is fixed as $\delta = \frac{M}{K} =  \frac{1}{2}$.
The important system parameters adopted in our simulations are summarized in Table \ref{C3:SysParameters}.
The specific simulation setups for each simulation scenario are shown under each figure.
All the simulation results in this chapter are obtained by averaging the system performance over both small-scale fading and large-scale fading.

\begin{table}
	\caption{System Parameters Used in Simulations}
	\centering\small
	\begin{tabular}{l|r}
		\hline
		Inner cell radius, $D_0$                 & 50 m \\\hline
		Outer cell radius, $D$                   & [50, 200, 500] m \\\hline
		Normalized cell size, $\eta$             & [1, 4, 10] \\\hline
		Number of users, $K$                     & 2 $\sim$ 256 \\\hline
		Number of receive antennas at BS, $M$    & 1 $\sim$ 128 \\\hline
		Path loss exponent, $\alpha$         & 3.76 \\\hline
		Noise power, $N_0$                   & -80 dBm \\\hline
		System SNR, ${\mathrm{SNR}_{\mathrm{sum}}}$       & 0 $\sim$ 40 dB \\\hline
		Ratio $\varsigma = \frac{W}{M}$ for \emph{m}MIMO-OMA & $\frac{1}{16}$ \\\hline
		Ratio $\delta = \frac{M}{K}$ for \emph{m}MIMO-NOMA & $\frac{1}{2}$ \\\hline
	\end{tabular}\label{C3:SysParameters}
\end{table}

%\begin{table}
%	\caption{Simulation Parameters For Multi-cell Systems}
%	\centering
%	\begin{tabular}{c|c}
%		\hline
%		Circular area radius $D_1$          & 5 km \\\hline
%		User density $\rho$                 & 1000 devices per $\mathrm{km}^2$ \\\hline
%		Number of receive antennas at BS    & 1 $\sim$ 128 \\\hline
%		Path loss exponent $\alpha$         & 3.76 \\\hline
%		Noise power $N_0$                   & -80 dBm \\\hline
%		System SNR ${\mathrm{SNR}_{\mathrm{sum}}}$       & 0 $\sim$ 40 dB \\\hline
%		Ratio $\varsigma = \frac{W}{M}$ for \emph{m}MIMO-OMA & $\frac{1}{16}$ \\\hline
%		Ratio $\varsigma = \frac{W}{M}$ for \emph{m}MIMO-NOMA & $\frac{1}{2}$ \\\hline
%	\end{tabular}\label{C3:TransceiverProtocol}
%\end{table}

\subsection{ESG versus the Number of Users in Single-cell Systems}

\begin{figure}[ptb]
\centering
\includegraphics[width=4.5in]{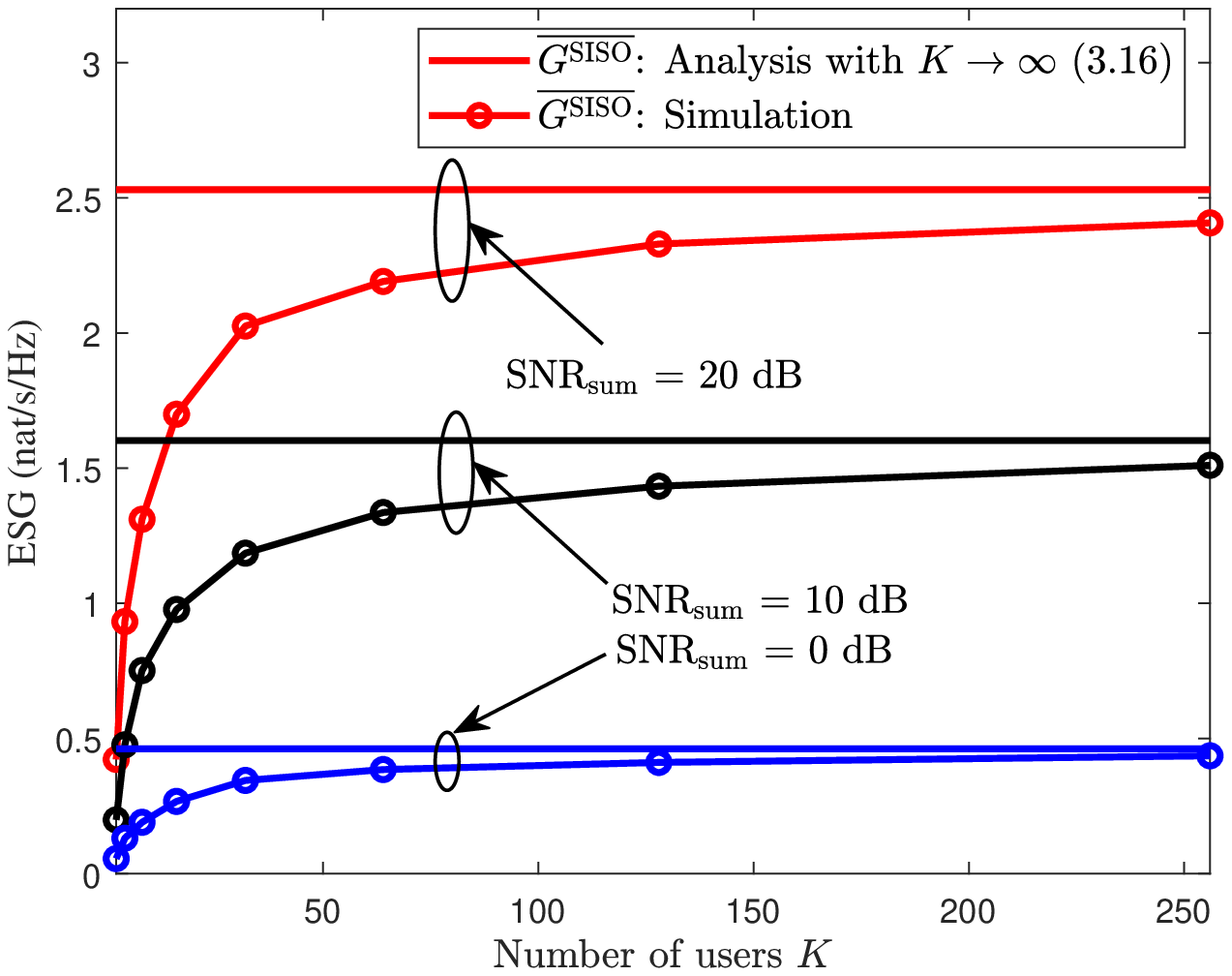}
\caption{The ESG of NOMA over OMA versus the number of users $K$ in single-antenna systems. The normalized cell size is $\eta = 10$.
{The blue, black, and red curves represent the ESG with the average received sum SNR ${\mathrm{SNR}_{\mathrm{sum}}} = [0,10,20]$ dB, respectively.}}%
\label{C3:APGVsK:a}%
\end{figure}

\begin{figure}[ptb]
\centering
\includegraphics[width=4.5in]{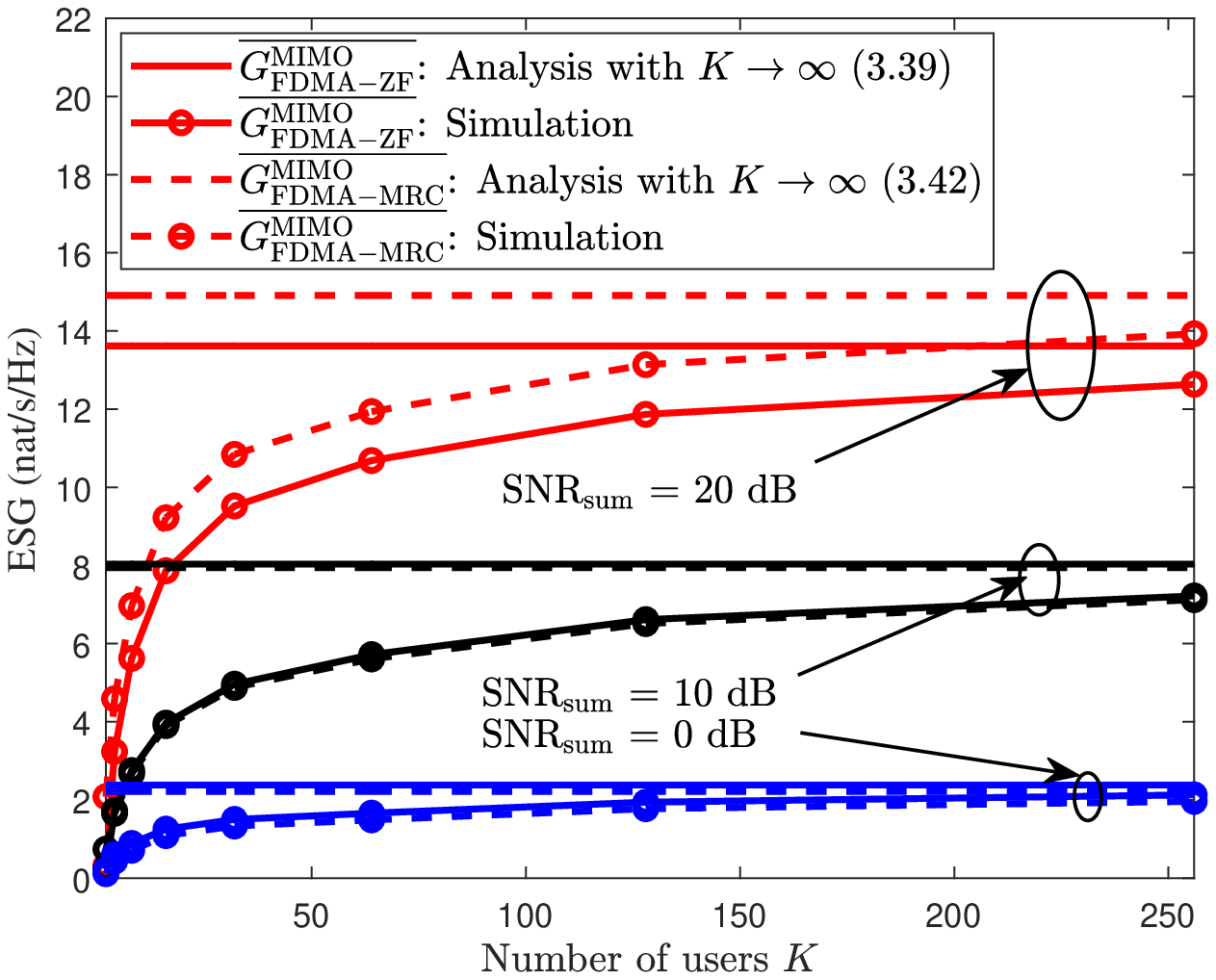}
\caption{The ESG of NOMA over OMA versus the number of users $K$ in multi-antenna systems of $M=4$. The normalized cell size is $\eta = 10$.
{The blue, black, and red curves represent the ESG with the average received sum SNR ${\mathrm{SNR}_{\mathrm{sum}}} = [0,10,20]$ dB, respectively.}}%
\label{C3:APGVsK:b}%
\end{figure}

Figure \ref{C3:APGVsK:a}, Figure \ref{C3:APGVsK:b}, and Figure \ref{C3:APGVsK:c} illustrate the ESG of NOMA over OMA versus the number of users in the single-antenna, multi-antenna, and massive MIMO single-cell systems, respectively.
In both Figure \ref{C3:APGVsK:a} and Figure \ref{C3:APGVsK:b}, we can observe that the ESG increases with the number of users $K$ and eventually approaches the asymptotic results derived for $K\to \infty$.
This is because upon increasing the number of users, the heterogeneity in channel gains among users is enhanced, which leads to an increased near-far gain.
As shown in Figure \ref{C3:APGVsK:c}, for massive MIMO systems, the asymptotic ESG per user derived in \eqref{C3:DD0EPGmMIMOERA0} closely matches with the simulations even for moderate numbers of users and SNRs.
Although \eqref{C3:DD0EPGPerUsermMIMOERA2} is derived for massive MIMO systems with $D=D_0$, in Figure \ref{C3:APGVsK:c}, we can observe a constant ESG per user in massive MIMO systems with $D>D_0$.
\begin{figure}[ptb]
	\centering \includegraphics[width=4.5in]{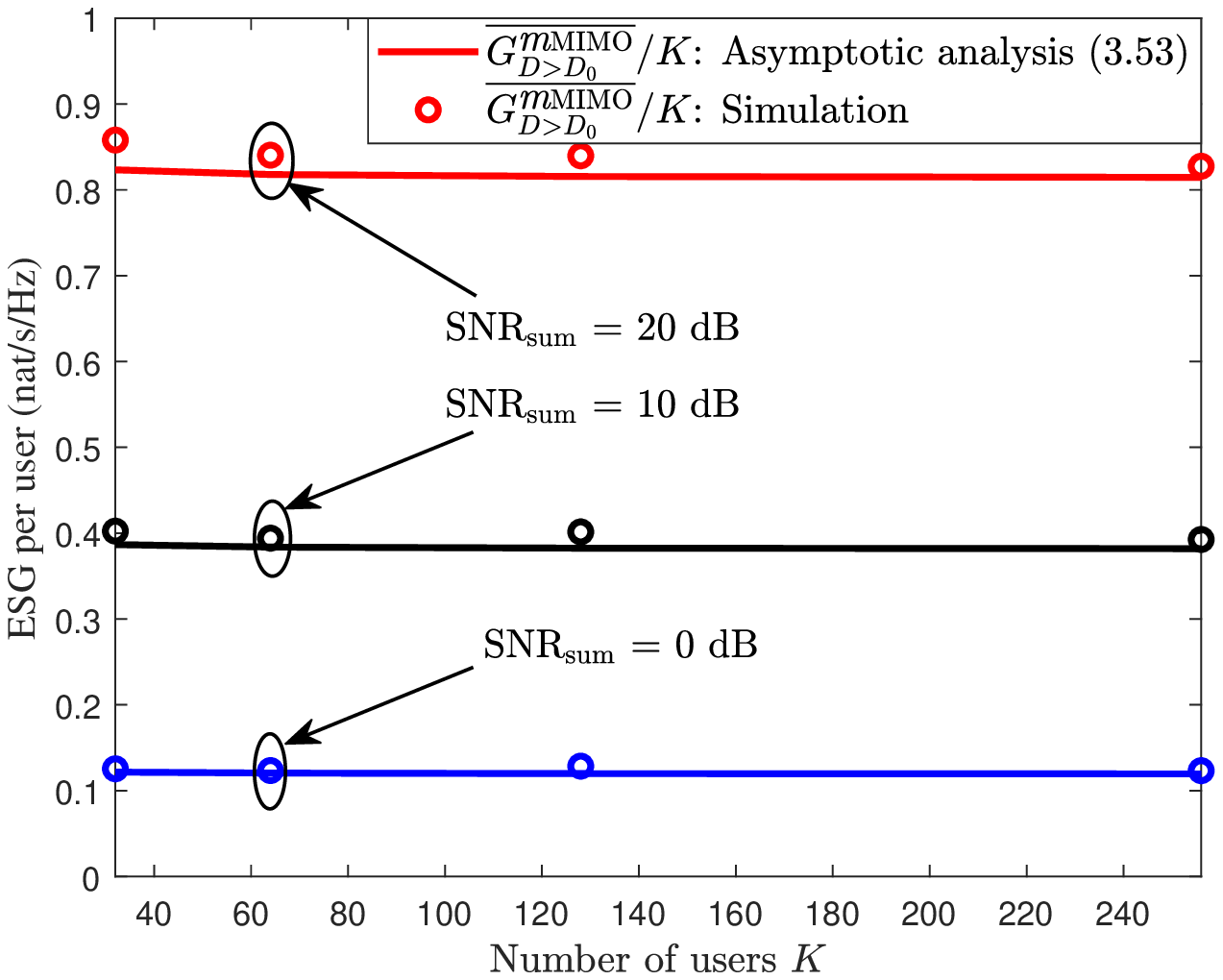}
	\caption{The ESG of NOMA over OMA versus the number of users $K$ in massive-antenna systems. The normalized cell size is $\eta = 10$. The number of antennas equipped at the BS is adjusted according to the number of users $K$ based on $M = {K}{\delta}$ with $\delta = \frac{1}{2}$. {The blue, black, and red curves represent the ESG with the average received sum SNR ${\mathrm{SNR}_{\mathrm{sum}}} = [0,10,20]$ dB, respectively.}}%
	\label{C3:APGVsK:c}%
\end{figure}
In other words, the insights obtained from the massive MIMO systems with $D=D_0$ are also applicable to the scenarios of $D>D_0$.
Compared to the ESG in the single-antenna systems of Figure \ref{C3:APGVsK:a}, the ESG in the multi-antenna systems of Figure \ref{C3:APGVsK:b} is substantially increased due to the extra spatial DoF offered by additional antennas at the BS.
Moreover, it can be observed in Figure \ref{C3:APGVsK:b} that we have $\mathop {\lim }\limits_{K \rightarrow \infty}  \overline{G^{\mathrm{MIMO}}_{\mathrm{FDMA-ZF}}} > \mathop {\lim }\limits_{K \rightarrow \infty}  \overline{G^{\mathrm{MIMO}}_{\mathrm{FDMA-MRC}}}$ in the low-SNR case, while $\mathop {\lim }\limits_{K \rightarrow \infty}  \overline{G^{\mathrm{MIMO}}_{\mathrm{FDMA-ZF}}} < \mathop {\lim }\limits_{K \rightarrow \infty}  \overline{G^{\mathrm{MIMO}}_{\mathrm{FDMA-MRC}}}$ in the high-SNR case.
This is because ZF detection outperforms MRC detection in the high-SNR regime for the MIMO-OMA system considered, while it becomes inferior to MRC detection in the low-SNR regime.
Furthermore, we can observe a higher ESG in Figure \ref{C3:APGVsK:a}, Figure \ref{C3:APGVsK:b}, and Figure \ref{C3:APGVsK:c} for the high-SNR case, e.g. ${\mathrm{SNR}_{\mathrm{sum}}} = 20$ dB.
This is due to the power-domain multiplexing of NOMA, which enables multiple users to share the same time-frequency resource and motivates a more efficient exploitation of the power resource.

\subsection{ESG versus the SNR in Single-cell Systems}
\begin{figure}[ptb]
\centering
\includegraphics[width=4.5in]{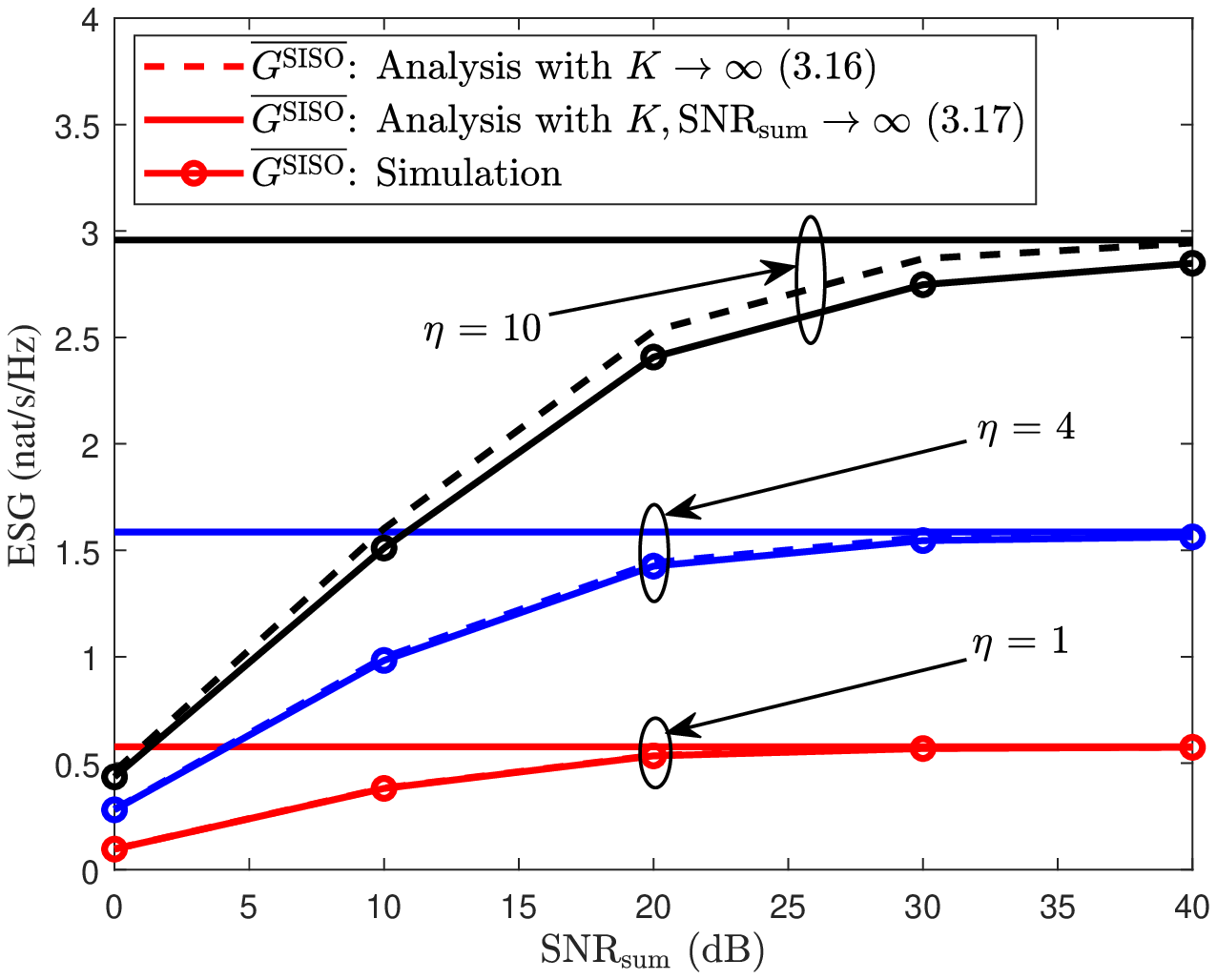}
\caption{The ESG of SISO-NOMA over SISO-OMA versus ${\mathrm{SNR}_{\mathrm{sum}}}$. The number of users is $K = 256$.
{The red, blue, and black curves represent the ESG with the normalized cell sizes $\eta = [1,4,10]$, respectively.}}%
\label{C3:APGVsSNR:a}%
\end{figure}

\begin{figure}[ptb]
\centering
\includegraphics[width=4.5in]{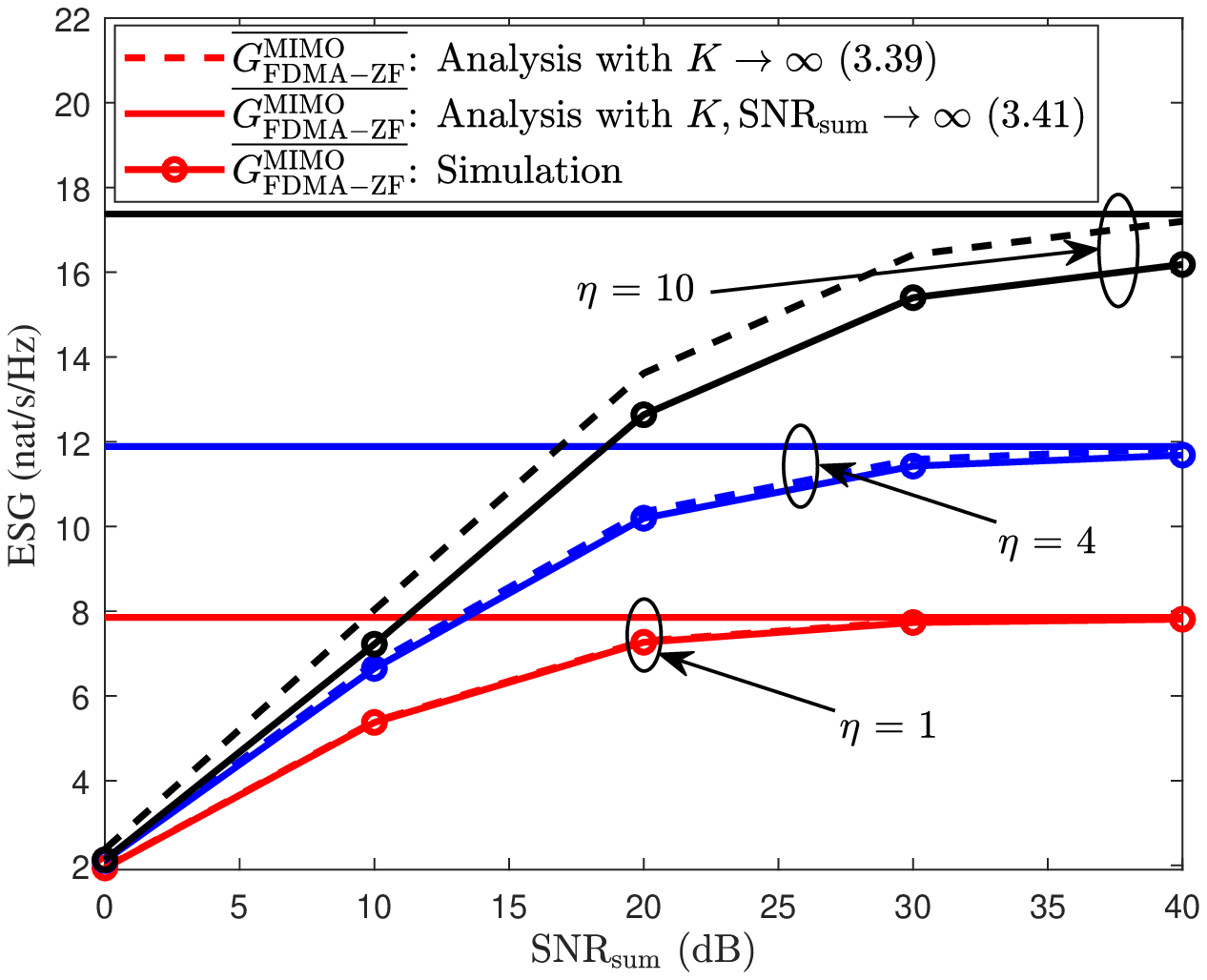}
\caption{The ESG of MIMO-NOMA over MIMO-OMA with FDMA-ZF versus ${\mathrm{SNR}_{\mathrm{sum}}}$. The number of antennas equipped at the BS $M=4$, the number of users is $K = 256$.
{The red, blue, and black curves represent the ESG with the normalized cell sizes $\eta = [1,4,10]$, respectively.}}%
\label{C3:APGVsSNR:b}%
\end{figure}

\begin{figure}[ptb]
	\centering
	\includegraphics[width=4.5in]{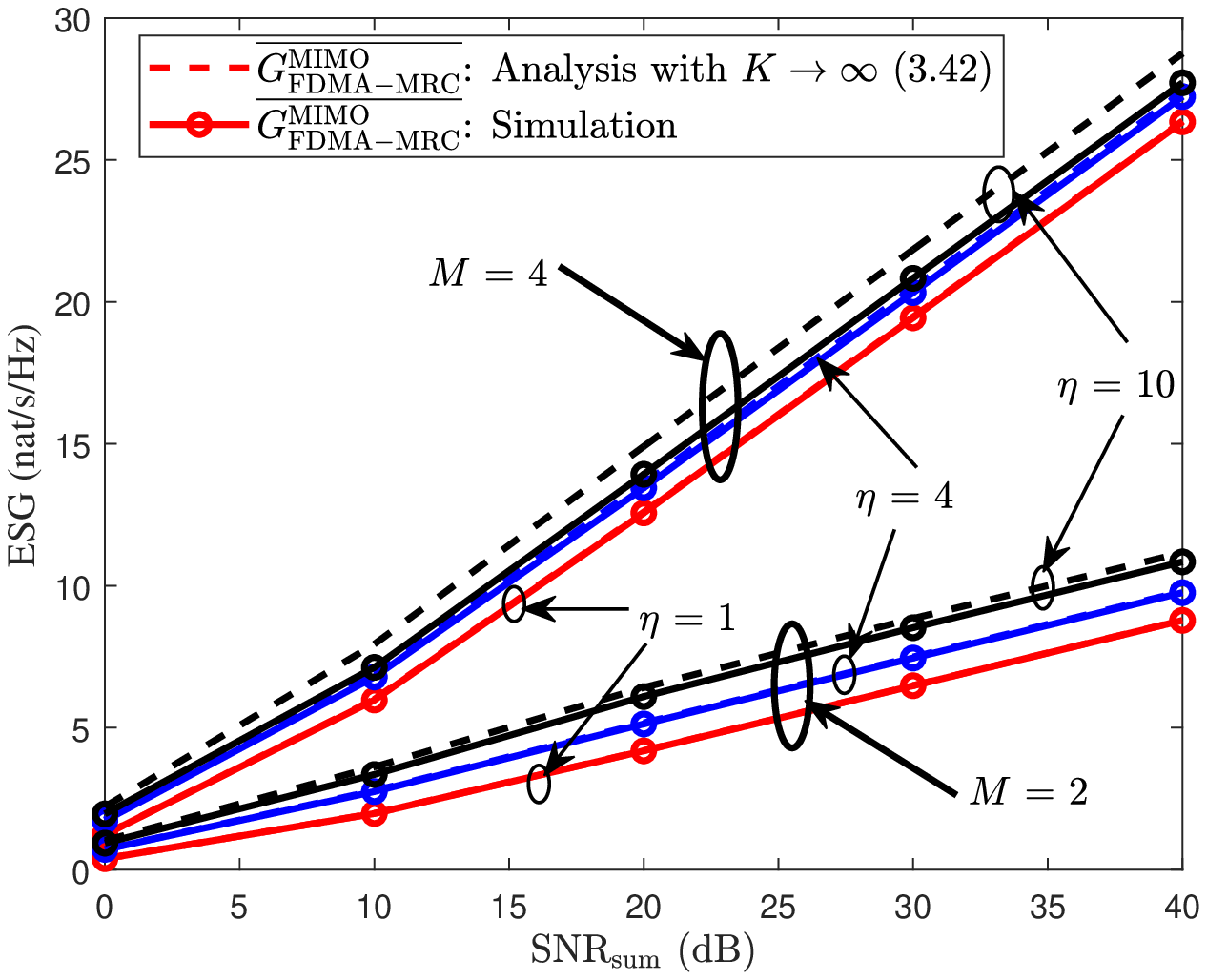}
	\caption{The ESG of MIMO-NOMA over MIMO-OMA with FDMA-MRC versus ${\mathrm{SNR}_{\mathrm{sum}}}$. The number of antennas equipped at the BS $M=[2,4]$, the number of users is $K = 256$.
{The red, blue, and black curves represent the ESG with the normalized cell sizes $\eta = [1,4,10]$, respectively.}}%
	\label{C3:APGVsSNR:c}%
\end{figure}

Figure \ref{C3:APGVsSNR:a}, Figure \ref{C3:APGVsSNR:b}, Figure \ref{C3:APGVsSNR:c}, and Figure \ref{C3:APGVsSNR:d} depict the ESG of NOMA over OMA versus the system's SNR ${\mathrm{SNR}_{\mathrm{sum}}}$ within the range of ${\mathrm{SNR}_{\mathrm{sum}}} = [0, 40]$ dB in the single-antenna, multi-antenna, and massive MIMO single-cell systems, respectively.
We can observe that the simulation results match closely our asymptotic analyses in all the considered cases.
Besides, by increasing the system SNR, the ESGs seen in Figure \ref{C3:APGVsSNR:a} and Figure \ref{C3:APGVsSNR:b} increase monotonically and approach the asymptotic analyses results derived in the high-SNR regime.
\begin{figure}[ptb]
	\centering
	\includegraphics[width=5.0in]{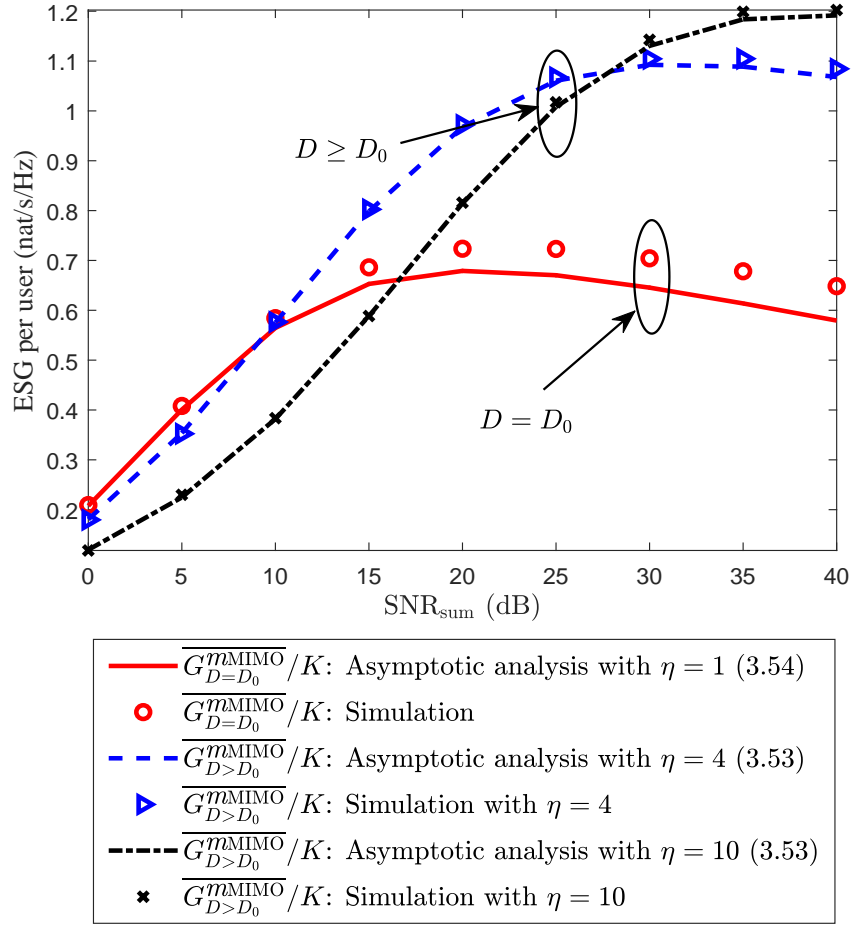}
	\caption{The ESG of \emph{m}MIMO-NOMA over \emph{m}MIMO-OMA versus ${\mathrm{SNR}_{\mathrm{sum}}}$. The number of antennas equipped at the BS $M = 128$, the number of users is $K = 256$, and $\delta = \frac{M}{K} = \frac{1}{2}$.
		The normalized cell size is $\eta = [1,4,10]$.}%
	\label{C3:APGVsSNR:d}%
\end{figure}
In other words, the ESGs seen in Figure \ref{C3:APGVsSNR:a} and Figure \ref{C3:APGVsSNR:b} are bounded from above even if ${P_{{\mathrm{max}}}} \to \infty $.
This is because there is no DoF gain in the ESG of NOMA over OMA in the pair of scenarios considered.
By contrast, as derived in \eqref{C3:EPGMIMOERA6}, the $(M-1)$-fold DoF gain in the ESG of MIMO-NOMA over MIMO-OMA with FDMA-MRC enables the ESG to increase linearly with the system's SNR in dB in the high-SNR regime, as shown in Figure \ref{C3:APGVsSNR:c}.
{Furthermore, employing more antennas provides a larger DoF gain, which leads to a steeper slope $M-1$ of ESG versus the system SNR in dB.
As a result, in the high SNR regime, the ESG with $M=4$ is almost three times of the ESG with $M=2$.}
In contrast to Figure \ref{C3:APGVsSNR:a}, Figure \ref{C3:APGVsSNR:b}, and Figure \ref{C3:APGVsSNR:c}, the ESG of \emph{m}MIMO-NOMA over \emph{m}MIMO-OMA recorded in Figure \ref{C3:APGVsSNR:d} first increases and then decreases with the system SNR, especially for a small normalized cell size.
In fact, the \emph{m}MIMO-NOMA system relying on MRC-SIC detection becomes interference-limited in the high-SNR regime, while the \emph{m}MIMO-OMA system remains interference-free, since favorable propagation conditions prevail for $\varsigma = \frac{W}{M} \ll 1$.
As a result, upon increasing the system SNR, the increased IUI of the \emph{m}MIMO-NOMA system considered neutralizes some of its ESG over the \emph{m}MIMO-OMA system, particularly for a small cell size associated with a limited large-scale near-far gain.

On the other hand, it is worth noticing in Figure \ref{C3:APGVsSNR:a}, that if all the users are randomly distributed on a circle when $D = D_0 = 50$ m, i.e., $\eta = 1$, then we have an ESG of about $0.575$ nat/s/Hz at ${\mathrm{SNR}_{\mathrm{sum}}} = 40$ dB for SISO-NOMA compared to SISO-OMA.
This again verifies the accuracy of the small-scale fading gain $\gamma$ derived in \eqref{C3:EPGSISOERA2}.
Furthermore, we can observe in Figure \ref{C3:APGVsSNR:a}, Figure \ref{C3:APGVsSNR:b}, and Figure \ref{C3:APGVsSNR:c}, that a larger normalized cell size $\eta$ results in a higher performance gain, which is an explicit benefit of the increased large-scale near-far gain $\vartheta \left( \eta \right)$.
By contrast, in Figure \ref{C3:APGVsSNR:d}, a larger cell size facilitates a higher ESG but only in the high-SNR regime, while a smaller cell size can provide a larger ESG in the low to moderate-SNR regime.
In fact, due to the large number of antennas, the IUI experienced in the \emph{m}MIMO-NOMA system is significantly reduced compared to that in single-antenna and multi-antenna systems.
As a result, in the low to moderate-SNR regime, the \emph{m}MIMO-NOMA system considered may be noise-limited rather than interference-limited, which is in line with the single-antenna and multi-antenna systems.
For instance, the noise degrades the achievable rates of the cell-edge users more severely compared to the impact of IUI in the \emph{m}MIMO-NOMA system, especially for large normalized cell sizes.
Therefore, the large-scale near-far gain cannot be fully exploited in the low to moderate-SNR regime in the massive MIMO systems.
Moreover, it can be observed in Figure \ref{C3:APGVsSNR:d} that the ESG increases faster for a larger normalized cell size $\eta$.
This is due to the enhanced large-scale near-far gain observed for a larger cell size, which enables NOMA to exploit the power resource more efficiently.

\subsection{ESG versus the Number of Antennas $M$ in Single-cell Systems}

\begin{figure}[ptb]
\centering
\includegraphics[width=4.5in]{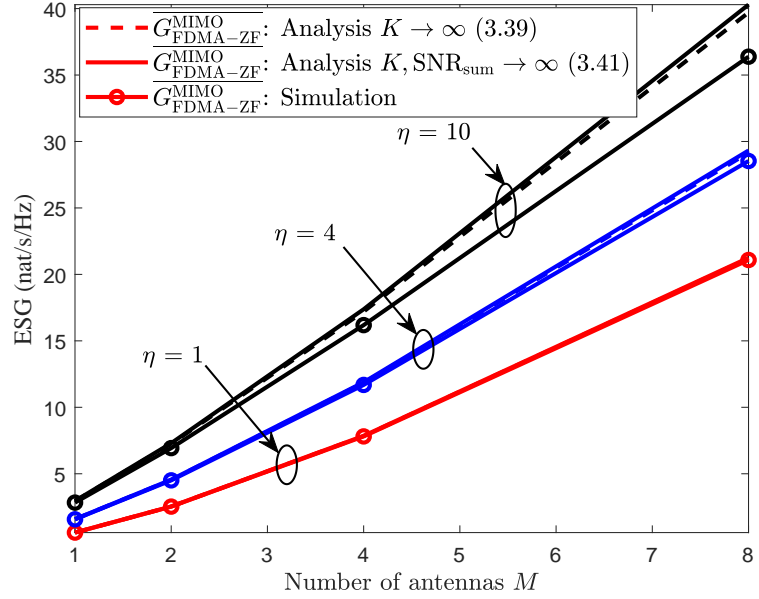}
\caption{The ESG of MIMO-NOMA over MIMO-OMA with FDMA-ZF versus the number of antennas $M$. The number of users is $K = 256$, the number of antennas $M$ equipped at the BS ranges from $1$ to $8$, the average received sum SNR is ${\mathrm{SNR}_{\mathrm{sum}}} = 40$ dB. {The red, blue, and black curves represent the ESG with the normalized cell sizes $\eta = [1,4,10]$, respectively.}}%
\label{C3:MIMOAPGVsM:a}%
\end{figure}

\begin{figure}[ptb]
\centering
\includegraphics[width=4.5in]{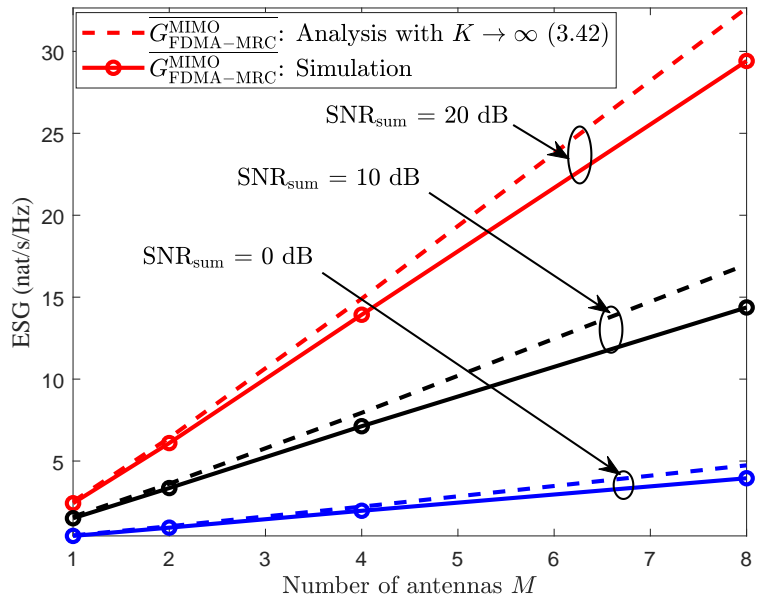}
\caption{The ESG of MIMO-NOMA over MIMO-OMA with FDMA-MRC versus the number of antennas $M$. The number of users is $K = 256$, the number of antennas $M$ equipped at the BS ranges from $1$ to $8$, and the normalized cell size is $\eta = 10$. {The blue, black, and red curves represent the ESG with the average received sum SNR ${\mathrm{SNR}_{\mathrm{sum}}} = [0,10,20]$ dB, respectively.}}%
\label{C3:MIMOAPGVsM:b}%
\end{figure}

\begin{figure}[ptb]
\centering
\includegraphics[width=5.0in]{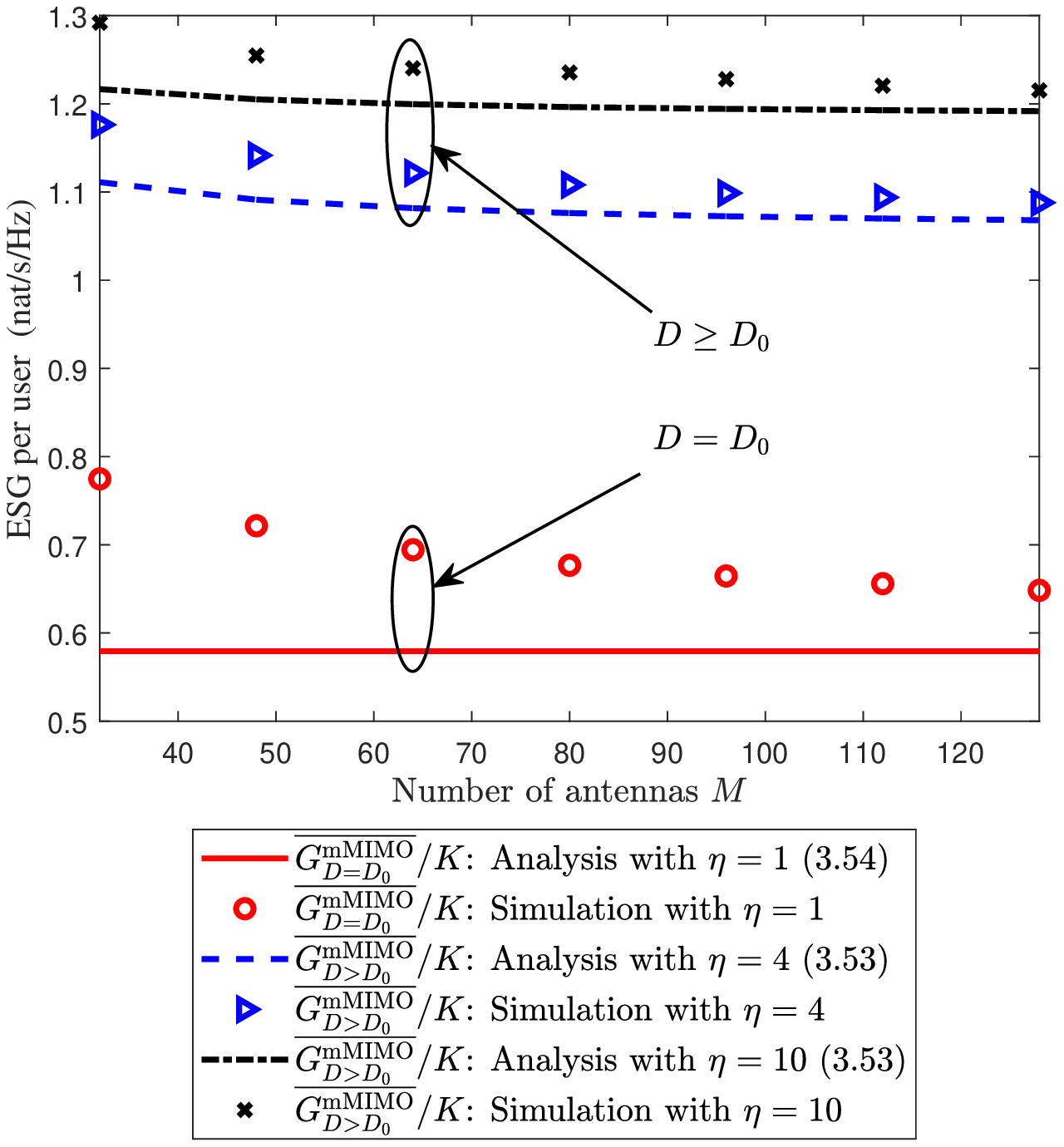}
\caption{The ESG of \emph{m}MIMO-NOMA over \emph{m}MIMO-OMA versus the number of antennas $M$. The average received sum SNR is ${\mathrm{SNR}_{\mathrm{sum}}} = 40$ dB and the normalized cell size is $\eta = [1,4,10]$.
The number of antennas $M$ equipped at the BS ranges ranges from $32$ to $128$ and the number of users $K$ is adjusted according to $M$ based on $K = \frac{M}{\delta}$ with $\delta = \frac{1}{2}$.}%
\label{C3:MIMOAPGVsM:c}%
\end{figure}

Figure \ref{C3:MIMOAPGVsM:a}, Figure \ref{C3:MIMOAPGVsM:b}, and Figure \ref{C3:MIMOAPGVsM:c} illustrate the ESG of NOMA over OMA versus the number of antennas $M$ employed at the BS in multi-antenna and massive MIMO systems, respectively.
It can be observed that the simulation results closely match our asymptotic analyses  for all the simulation scenarios.
In particular, observe for the ESG of MIMO-NOMA over MIMO-OMA with FDMA-ZF in Figure \ref{C3:MIMOAPGVsM:a} that as predicted in \eqref{C3:EPGMIMOERA3}, the asymptotic ESG $\overline{G^{{\mathrm{SISO}}}}$ of single-antenna systems is increased by $M$, when an $M$-antenna array is employed at the BS.
More importantly, a larger normalized cell size $\eta$ enables a steeper slope in the ESG versus the number of antennas $M$, which is due to the increased large-scale near-far gain $\vartheta \left( \eta \right)$, as shown in \eqref{C3:EPGMIMOERA2}.
Apart from the linearly increased component of ESG vesus $M$, an additional power gain factor of $\ln\left(M\right)$ can also be observed in Figure \ref{C3:MIMOAPGVsM:a} as derived in \eqref{C3:EPGMIMOERA3}.
Observe the ESG of MIMO-NOMA over MIMO-OMA with FDMA-MRC in Figure \ref{C3:MIMOAPGVsM:b} that the ESG grows linearly versus $M$ due to the $(M-1)$-fold of DoF gain and the corresponding slope becomes higher for a higher system SNR, as seen in \eqref{C3:EPGMIMOERA6}.
The ESG per user seen in Figure \ref{C3:MIMOAPGVsM:c} for massive MIMO systems remains almost constant upon increasing $M$, which matches for our asymptotic analysis in \eqref{C3:DD0EPGPerUsermMIMOERA2}, and is also consistent with the results of Figure \ref{C3:APGVsK:c} for the fixed ratio $\delta = \frac{M}{K}$.
Furthermore, we can observe that a large cell size offers a higher ESG per user due to the improved large-scale near-far gain.

\subsection{ESG versus the Total Transmit Power in Multi-cell Systems}
In a multi-cell system, we consider a high user density scenario within a large circular area with the radius of $D_1 = 5$ km and the user density of $\rho = 1000$ devices per $\mathrm{km}^{\mathrm{2}}$.
As a result, the total number of users in the multi-cell system considered is $K' = \left\lceil\rho \pi D_1^2\right\rceil$.
Then, the number of users in each cell is given by $K = \left\lceil {\rho \pi {D^2}} \right\rceil$ with $D$ in the unit of km.
Meanwhile, the number of adjacent cells $L$ can be obtained by $L = \left\lceil \frac{K' - K}{K} \right\rceil$, so that all the $K'$ users can be covered.
Furthermore, the $K'$ users in all the cells share a given total transmit power and the total transmit power ${P_{\mathrm{max}}'}$ of $(L+1)$ cells is within the range spanning from $20$ dBm to $60$ dBm\footnote{Since there are a larger number of users deployed in the considered area, we set a large power budget for all the users in the considered multi-cell system.}.
In this section, we also consider an equal power allocation among multiple cells and an equal power allocation among users within each cell, i.e., we have ${P_{\mathrm{max}}} = \frac{P_{\mathrm{max}}'}{L+1}$ and ${p_{k',l}} = \frac{{P_{\mathrm{max}}}}{K}$.
All the other simulation parameters are the same as those adopted in the single-cell systems.

\begin{figure}[t]
\centering
\includegraphics[width=4.5in]{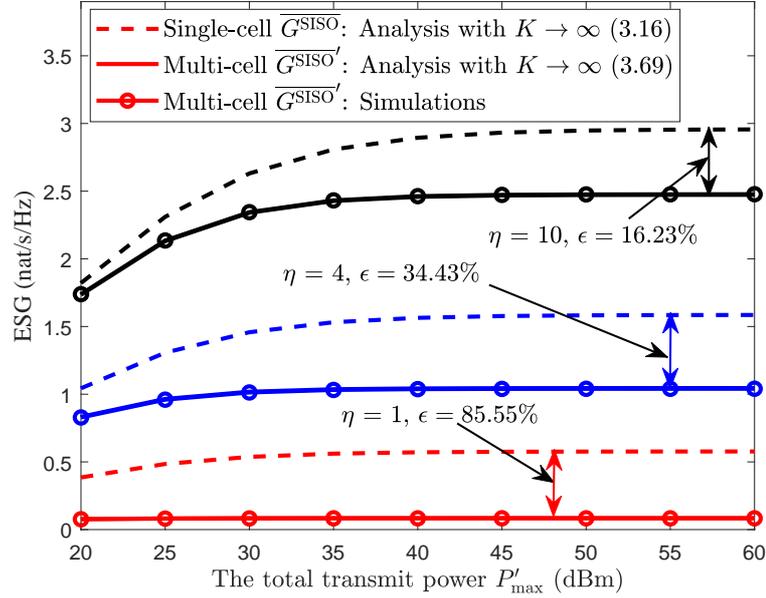}
\caption{The ESG of SISO-NOMA over SISO-OMA versus the total transmit power ${P_{\mathrm{max}}'}$ in multi-cell systems. The number of antennas equipped at each BS is $M = 1$. The ESG degradations due to the ICI are denoted by double-sided arrows. {The red, blue, and black curves represent the ESG with the normalized cell sizes $\eta = [1,4,10]$, respectively.}}%
\label{C3:MultiCell:a}%
\end{figure}

\begin{figure}[t]
\centering
\includegraphics[width=5in,height = 4.25in]{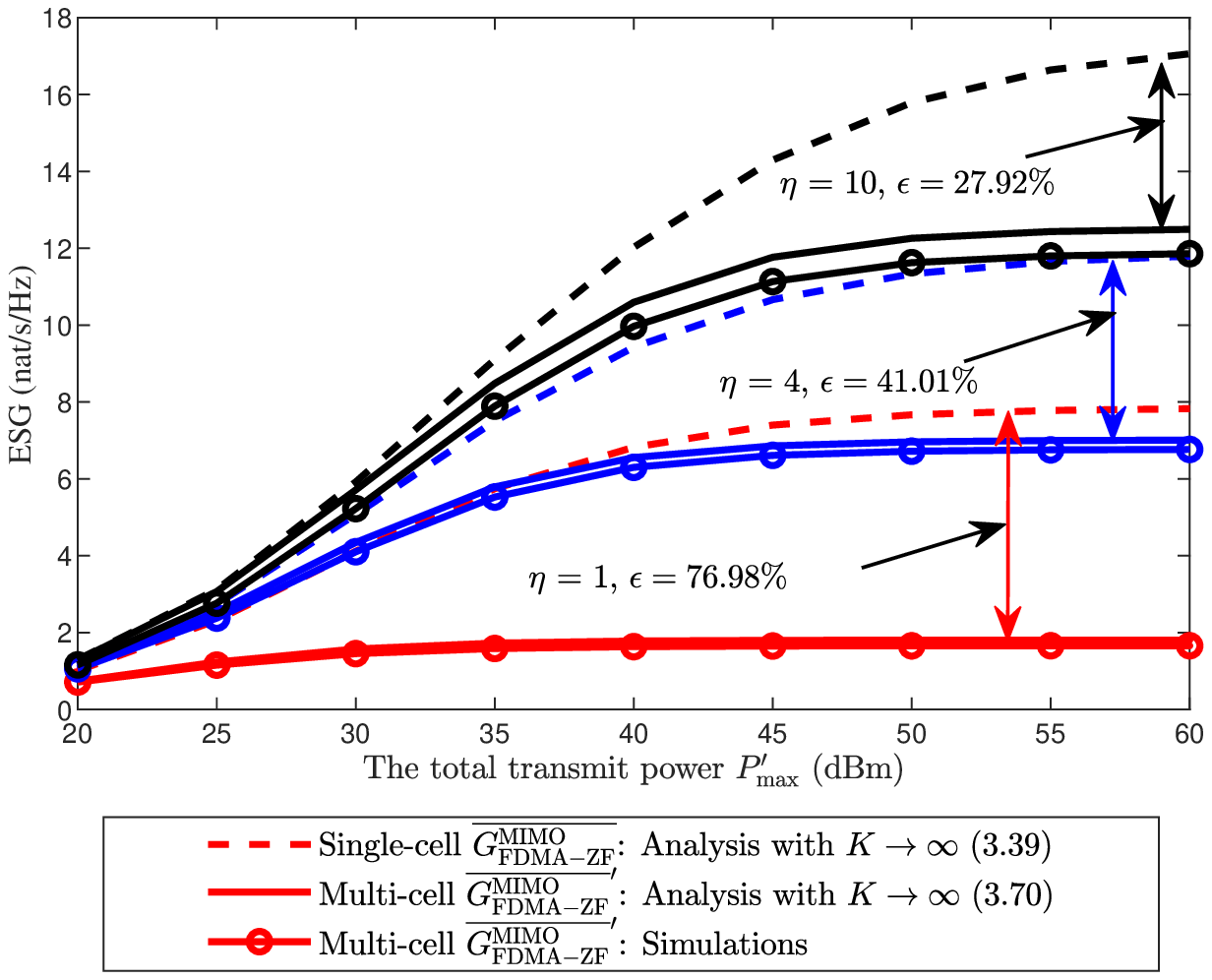}
\caption{The ESG of MIMO-NOMA over MIMO-OMA with FDMA-ZF versus the total transmit power ${P_{\mathrm{max}}'}$ in multi-cell systems. The number of antennas equipped at each BS is $M=4$. The ESG degradations due to the ICI are denoted by double-sided arrows. {The red, blue, and black curves represent the ESG with the normalized cell sizes $\eta = [1,4,10]$, respectively.}}%
\label{C3:MultiCell:b}%
\end{figure}

\begin{figure}[t]
\centering
\includegraphics[width=5in]{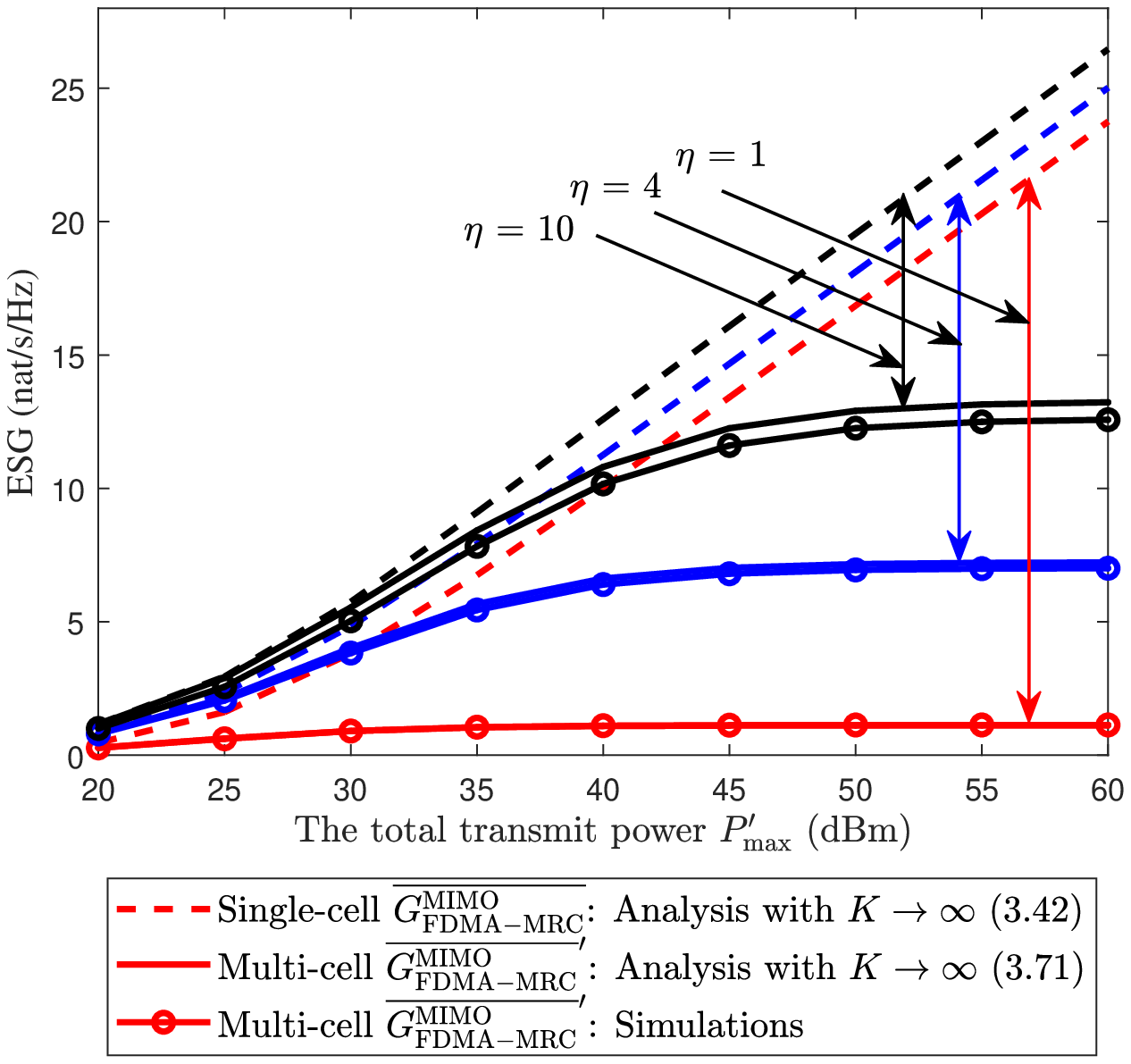}
\caption{The ESG of MIMO-NOMA over MIMO-OMA with FDMA-MRC versus the total transmit power ${P_{\mathrm{max}}'}$ in multi-cell systems. The number of antennas equipped at each BS is $M=4$. The ESG degradations due to the ICI are denoted by double-sided arrows. {The red, blue, and black curves represent the ESG with the normalized cell sizes $\eta = [1,4,10]$, respectively.}}%
\label{C3:MultiCell:c}%
\end{figure}

\begin{figure}[t]
\centering
\includegraphics[width=5in]{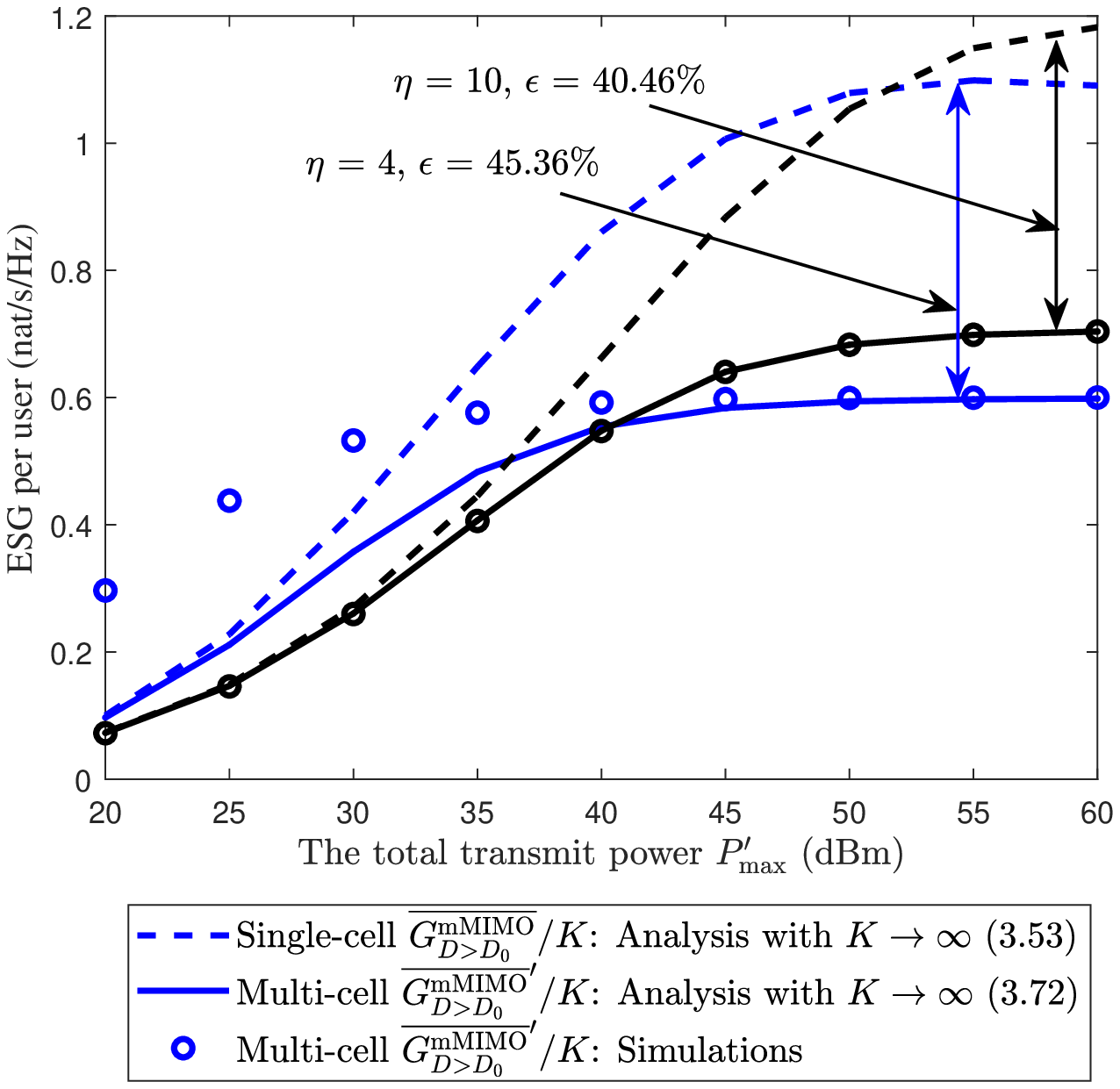}
\caption{The ESG of \emph{m}MIMO-NOMA over \emph{m}MIMO-OMA versus the total transmit power ${P_{\mathrm{max}}'}$ in multi-cell systems. The number of antennas equipped at each BS is adjusted based on the number of users in each cell via $M = \left\lceil {K}{\delta} \right\rceil$ and the group size of the considered \emph{m}MIMO-OMA system is $W = \left\lceil {\varsigma}M \right\rceil$. The ESG degradations due to the ICI are denoted by double-sided arrows. {The blue and black curves represent the ESG with the normalized cell sizes $\eta = [4,10]$, respectively.}}%
\label{C3:MultiCell:d}%
\end{figure}

In contrast to the single-cell systems, the ESG versus the total transmit power ${P_{\mathrm{max}}'}$ trends are more interesting, which is due to the less straightforward impact of ICI on the performance gain of NOMA over OMA in multi-cell systems.
Figure \ref{C3:MultiCell:a}, Figure \ref{C3:MultiCell:b}, Figure \ref{C3:MultiCell:c}, and Figure \ref{C3:MultiCell:d} show the ESG of NOMA over OMA versus the total transmit power ${P_{\mathrm{max}}'}$ in single-antenna, multi-antenna, and massive MIMO\footnote{Note that, for the considered massive MIMO multi-cell system, a small cell size leads to a small number of users $K$ in each cell and thus results in a small number of antennas $M$ due to the fixed ratio $\delta = \frac{M}{K}$. This is contradictory to our assumption of $K \to \infty$ and $M \to \infty$. Therefore, we only consider the normalized cell size of $\eta = [4,10]$ for massive MIMO multi-cell systems in Figure \ref{C3:MultiCell:d}.} multi-cell systems, respectively.
The analytical results in single-cell systems are also shown for comparison.
We can observe that the performance gains of NOMA over OMA are degraded upon extending NOMA from single-cell systems to multi-cell systems.
In fact, NOMA schemes enable all the users in adjacent cells to simultaneously transmit their signals on the same frequency band and thus the ICI level in NOMA schemes is substantially higher than that in OMA schemes, as derived in \eqref{C3:MulticellInterference3}.
For the ease of illustration, we define the normalized performance degradation of the ESG in multi-cell systems compared to that in single-cell systems as $\epsilon = \frac{{\overline G - \overline G'}}{{\overline G}}$, where $\overline G$ denotes the ESG in single-cell systems and $\overline G'$ denotes the ESG in multi-cell systems.
It can be observed that the ESG degradation is more severe for a small normalized cell size $\eta$, because multi-cell systems suffer from a more severe ICI for smaller cell sizes due to a shorter inter-site distance.
Therefore, the system performance becomes saturated even for a moderate system power budget in the case of a smaller cell size.

It is worth noting that the ESG of MIMO-NOMA over MIMO-OMA with FDMA-MRC is saturated in multi-cell systems in the high transmit power regime, as shown in Figure \ref{C3:MultiCell:c}, which is different from the trends seen for single-cell systems in Figure \ref{C3:APGVsSNR:c}.
In fact, the $(M-1)$-fold DoF gain in the ESG of MIMO-NOMA over MIMO-OMA with FDMA-MRC in single-cell systems derived in \eqref{C3:EPGMIMOERA6} can only be achieved in the high-SNR regime.
However, due to the lack of joint multi-cell signal processing to mitigate the ICI, the multi-cell system becomes interference-limited upon increasing the total transmit power.
Therefore, the multi-cell system actually operates in the low-SINR regime, which does not facilitate the exploitation of the DoF gain in single-cell systems.

\begin{table}
	\caption{Comparison on ESG (nat/s/Hz) of NOMA over OMA in the considered scenarios. The system setup is $K = 256$, $D = 200$ m, $\eta = 4$, and $M = 4$ for  multi-antenna systems.}
	\centering\small
	\begin{tabular}{c|cc|cc}
		\hline
		               & \multicolumn{2}{c}{${\mathrm{SNR}_{\mathrm{sum}}} = 0$ dB}  \vline& \multicolumn{2}{c}{${\mathrm{SNR}_{\mathrm{sum}}} = 10$ dB} \\\hline
		               & Single-cell    & Multi-cell                          & Single-cell    & Multi-cell                           \\\hline
		Single-antenna & 0.281          & 0.2639                              & 0.983          & 0.7973                                 \\
		Multi-antenna  & 2.114          & 2.0179                              & 6.65           & 5.4113                                  \\
		Massive MIMO (ESG per user)     & 0.1796         & 0.1702                              & 0.5765         & 0.4490                      \\
		\hline
	\end{tabular}\label{C3:SimulationComparison}
\end{table}

\begin{remark}
The comparison of the ESG (nat/s/Hz) results of NOMA over OMA in all the scenarios considered is summarized in Table \ref{C3:SimulationComparison}.
We consider a practical operation setup with $K = 256$, $D = 200$ m, $\eta = 4$, and $M = 4$ for the multi-antenna systems.
For fair comparison, the total transmit power ${P_{\mathrm{max}}'}$ in multi-cell systems is adjusted for ensuring that the total average received SNR ${\mathrm{SNR}_{\mathrm{sum}}}$ at the serving BS is identical to that in single-cell systems.
Note that the row of massive MIMO in Table \ref{C3:SimulationComparison} quantifies the ESG per user of \emph{m}MIMO-NOMA over \emph{m}MIMO-OMA, which is consistent with Figure \ref{C3:APGVsSNR:d}, Figure \ref{C3:MIMOAPGVsM:c} and Figure \ref{C3:MultiCell:d}.
We can observe that the ESG remains a near-constant at the low SNR of ${\mathrm{SNR}_{\mathrm{sum}}} = 0$ dB when extending NOMA from single-cell systems to multi-cell systems, while the ESG degrades substantially at the high SNR of ${\mathrm{SNR}_{\mathrm{sum}}} = 10$ dB.
In fact, the limited transmit power budget in the low-SNR regime in adjacent cells only leads to a low ICI level at the serving BS, which avoids a significant ESG degradation, when applying NOMA in multi-cell systems.
\end{remark}

\section{Summary}
In this chapter, we investigated the ESG in uplink communications attained by NOMA over OMA in single-antenna, multi-antenna, and massive MIMO systems with both single-cell and multi-cell deployments.
In the single-antenna single-cell system considered, the ESG of NOMA over OMA was quantified and two types of gains were identified in the ESG derived, i.e., the large-scale near-far gain and the small-scale fading gain.
The large-scale near-far gain increases with the cell size, while the small-scale fading gain is a constant of $\gamma = 0.57721$ nat/s/Hz in Rayleigh fading channels.
Additionally, we unveiled that the ESG of SISO-NOMA over SISO-OMA can be increased by $M$ times upon using $M$ antennas at the BS, owing to the extra spatial DoF offered by additional antennas.
In the massive MIMO single-cell system considered, the ESG of NOMA over OMA increases linearly both with the number of users and with the number of antennas at the BS.
The analytical results derived for single-cell systems were further extended to multi-cell systems via characterizing the effective ICI channel distribution and by deriving the ICI power.
We found that a larger cell size is preferred by NOMA for both single-cell and multi-cell systems, due to the enhanced large-scale near-far gain and reduced ICI, respectively.
Extensive simulation results have shown the accuracy of our performance analyses and confirmed the insights provided above.

\chapter{Joint Pilot and Payload Power Control
for Uplink MIMO-NOMA with MRC-SIC Receivers}\label{C4:chapter4}

\section{Introduction}
In the previous chapter, we have theoretically revealed that the performance gain of NOMA over OMA arises from the heterogeneity in channel gains of NOMA users.
Consequently, the performance gain of NOMA might be degraded due to some inevitable channel uncertainty in practice.
Additional, from Table \ref{C3:TransceiverProtocol}, we can observe that successive interference cancelation (SIC) plays an important role for realizing the performance gain of NOMA, which is also sensitive to channel estimation error.
In this chapter, we aim to improve the robustness of NOMA against the channel uncertainty via proposing a power allocation design.

Most of existing works on resource allocation design for NOMA focused on downlink NOMA systems\cite{sun2016optimal,Lei2016NOMA}.
However, NOMA inherently exists in uplink communications\cite{Yang2016NOMA}, since the electromagnetic waves are naturally superimposed at a receiving base station (BS).
A simple back-off power control scheme was proposed for uplink NOMA\cite{Zhangtobepublished}, while an optimal resource allocation algorithm to maximize the system sum-rate was developed in \cite{Al-Imari2014}.
Most recently, multiple-input multiple-output NOMA (MIMO-NOMA) systems are of more interests\cite{DingSignalAlignment,Xu2017}.
In particular, maximal ratio combining with SIC (MRC-SIC) is an appealing reception technique for uplink MIMO-NOMA owing to its low computational complexity.

Despite the fruitful research conducted on NOMA, only payload power allocation and ideal SIC decoding are considered in existing works, e.g.\cite{Al-Imari2014,Wei2017}.
For both single-antenna and multiple-antenna systems, it is well-known that error propagation of SIC decoding limits the promised performance gain brought by NOMA.
In practice, the sources of error propagation are two-fold: one is the channel estimation error (CEE) and the other is the erroneous in data detection.
This chapter focuses on tackling the former issue via exploiting the non-trivial trade-off between the pilot and payload power allocation for uplink MIMO-NOMA systems for a given total energy budget.
Specifically, a higher pilot power yields a better channel estimation but leads to a less payload power for data detection.
In the meantime, the reduced payload power would introduce a lower inter-user interference (IUI) for other users.
Therefore, jointly designing the pilot and payload power allocation is critical for mitigating the error propagation.

In this chapter, to alleviate the error propagation in SIC, we propose a joint pilot and payload power allocation (JPA) scheme for uplink MIMO-NOMA with a MRC-SIC receiver based on a practical minimum mean square error (MMSE) channel estimator.
The average signal-to-interference-plus-noise ratio (SINR) is derived analytically and the JPA design is formulated as a non-convex optimization problem to maximize the minimum weighted average SINR (ASINR).
The globally optimal solution is obtained by geometric programming.
Simulation results demonstrate that the proposed JPA scheme is beneficial to reduce the effect of error propagation, which enhances the data detection performance, especially in the moderate energy budget regime.

\section{System Model}

\subsection{System Model}
\begin{figure}[t]
\centering
\includegraphics[width=4.0in]{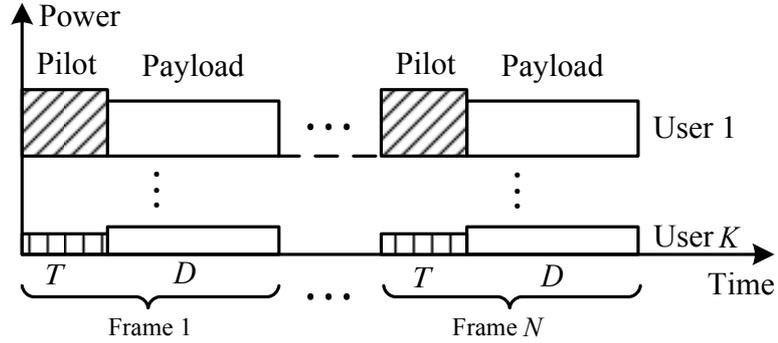}
\caption{An illustration of the frame structure of the uplink transmission.}\label{C4:FrameStructure}
\end{figure}

We consider an uplink MIMO-NOMA communication system in a single-cell with a BS equipped with $M$ antennas serving $K$ single-antenna users.
All the $K$ users are allocated on the same frequency band.
Every user transmits multiple frames over multiple coherence time intervals (CTI) to the BS, where we assume that the duration of each frame is comparable to that of a CTI.
In particular, each frame consists of $T$ pilot symbols and $D$ data symbols consecutive in time, as shown in Figure \ref{C4:FrameStructure}.
We assume that $T$ and $D$ are fixed as it is commonly implemented in practical systems for simplifying time synchronization.
Instead of considering symbol-level SIC as in most of existing works in NOMA \cite{Wei2017,DingSignalAlignment}, we adopt the codeword-level SIC to exploit coding gain.
%
% The number of symbols where channel stay static is determined by the product of coherence time and coherence bandwidth \cite{EmilTENmYCH}.
%
Note that, a codeword is usually much longer than the duration of a CTI and is spread over $N$ frames, which is an important scenario for time-varying channels with a short coherence time.

In frame $n$, the received signal at the BS during pilot transmission and data transmission are given by
\begin{equation}
{{\mathbf{Y}}_{n}^{\mathrm{P}}} = {{\mathbf{H}}_n}{\mathbf{\Lambda T}} + {{\mathbf{Z}}_{n}^{\mathrm{P}}} \; \text{and}\;
{{\mathbf{Y}}_{n}^{\mathrm{D}}} = {{\mathbf{H}}_n}{\mathbf{B}}{{\mathbf{D}}_n} + {{\mathbf{Z}}_{n}^{\mathrm{D}}},
\end{equation}
respectively.
$\mathbf{T} \in \mathbb{C}^{K\times T}$ denotes the pilot matrix and ${\mathbf{D}_n} = \left[{\mathbf{d}_{n,1}},\ldots, \mathbf{d}_{n,K}\right]^{\mathrm{T}}\in \mathbb{C}^{K\times D}$ denotes the data matrix in frame $n$.
The diagonal matrices ${\mathbf{\Lambda}}$ and ${\mathbf{B}}$ are defined by ${\mathbf{\Lambda}} = \text{diag}\left\{\sqrt{\alpha_1}, \ldots, \sqrt{\alpha_K}\right\}$ and ${\mathbf{B}} = \text{diag}\left\{\sqrt{\beta_1}, \ldots, \sqrt{\beta_K}\right\}$, respectively, where $\alpha_k$ and $\beta_k$ denote the pilot and payload power of user $k$, respectively\footnote{The power allocation for pilot and payload are calculated centrally at the BS and to be distributed to all the users through some closed-loop power control scheme in control channels, e.g. \cite{Lee1996CDMA}.}.
We assume that normalized orthogonal pilots are assigned to all the users exclusively, i.e., ${\mathbf{T}}{\mathbf{T}}^{\mathrm{H}} = \mathbf{I}_K$ with $T \ge K$.
%
%In addition, the pilot and payload powers are limited by the same energy budget for each user, i.e., $\alpha_k T + \beta_k D \le E_{\mathrm{max}}$.
%
%The transmitted symbols are assumed to be collected from a normalized finite constellation $\mathcal{A}$, i.e., ${\left\{ {\mathbf{D}}_n \right\}_{ij}} \in \mathcal{A}$ and $E\left\{ {{{\left| {{{\left\{ {\mathbf{D}}_n \right\}}_{ij}}} \right|}^2}} \right\} = 1$, where ${\left\{ {\mathbf{D}}_n \right\}_{ij}}$ denotes the entry in $i$-th row and $j$-th column of ${\mathbf{D}}_n$.
%
%It is different from most of existing works in NOMA with a Gaussian input assumption \cite{Sun2016Fullduplex}.
%
%To save the system overhead for power control in uplink communication, the power allocation matrix $\mathbf{\Lambda}$ and ${\mathbf{B}}$ are kept constant during $N$ frames and adjusted per codeword.
%
The matrices $\mathbf{Z}_{n}^{\mathrm{P}} \in \mathbb{C}^{M\times T}$ and $\mathbf{Z}_{n}^{\mathrm{D}} \in \mathbb{C}^{M\times D}$ denote the additive zero mean Gaussian noise with covariance matrix ${\sigma ^2}\mathbf{I}_M$ during training phase and data transmission phase in frame $n$, respectively.
The matrix $\mathbf{H}_n = \left[{{\mathbf{h}}_{n,1}},\ldots,{{\mathbf{h}}_{n,K}}\right] \in \mathbb{C}^{M\times K}$ contains the channels of all the users in frame $n$, where column $k$ denotes the channel vector of user $k$.
Rayleigh fading assumption is adopted in this chapter, i.e.,
${{\mathbf{h}}_{n,k}} \sim \mathcal{CN}\left( \mathbf{0},\nu_k^2 \mathbf{I}_M \right)$, where $\mathcal{CN}\left( \mathbf{0},\nu_k^2 \mathbf{I}_M \right)$ denotes a circularly symmetric complex Gaussian distribution with zero mean and covariance matrix $\nu_k^2 \mathbf{I}_M$.
Scalar $\nu_k^2$ denotes the large-scale fading of user $k$ capturing the effects of path loss and shadowing.
Since all the users are usually sufficiently separated apart compared to the wavelength, their channels are assumed to be independent with each other.
Therefore, their channel correlation matrix is given by a diagonal matrix ${{{\mathbf{R}}_{\mathbf{H}}}} = M\text{diag} \left\{\nu _1^2,\ldots,\nu _K^2\right\}$.
Without loss of generality, we assume that users are indexed in the descending order of large-scale fading, i.e., $\nu _1^2 \ge \nu _2^2 \ge  \ldots  \ge \nu _K^2$.
In this chapter, we define strong or weak user based on the large-scale fading since it facilitates the characterization of the channel ordering statistically across the codeword, i.e., user 1 is the strongest user, while user $K$ is the weakest user.
As a result, the SIC decoding order is assumed to be the descending order of large-scale fading, i.e., users $1,2,\ldots,K$ are decoded sequentially.
%\subsection{MMSE Channel Estimation}
%The MMSE channel estimation \cite{BigueshMMSE2006} is known to be the optimal with the channels' second order statistical information, which is adopted in this chapter.
%%
%The MMSE channel estimation in frame $n$ is given by
%\begin{equation}
%{{\hat{\mathbf{ H}}}}_n = {{\mathbf{Y}}_{n,t}}{\left( {{{\mathbf{T}}^{\mathrm{H}}}{\mathbf{\Lambda}}{{\mathbf{R}}_{\mathbf{H}}}{\mathbf{\Lambda}}{\mathbf{T}} + {\sigma ^2}M{\mathbf{I}}_{T}} \right)^{ - 1}}{{\mathbf{T}}^{\mathrm{H}}}{\mathbf{\Lambda}}{{\mathbf{R}}_{\mathbf{H}}},
%\end{equation}
%while the channel estimates for user $k$ in frame $n$ is
%\begin{equation}\label{C4:MMSEchannel}
%{{{\mathbf{\hat h}}}_{n,k}} = {{\mathbf{Y}}_{n,t}}{{\mathbf{\Phi }}^{ - 1}}{{\mathbf{A}}_k},
%\end{equation}
%with ${\mathbf{\Phi }} = {{\mathbf{T}}^{\mathrm{H}}}{\mathbf{\Lambda}}{{\mathbf{R}}_{\mathbf{H}}}{\mathbf{\Lambda}}{\mathbf{T}} + {\sigma ^2}M{\mathbf{I}}_{T}$ and ${{\mathbf{A}}_k} = {{\mathbf{T}}^{\mathrm{H}}}{\mathbf{\Lambda}}{\left\{ {{{\mathbf{R}}_{\mathbf{H}}}} \right\}_{:k}}$, where ${\left\{ {{{\mathbf{R}}_{\mathbf{H}}}} \right\}_{:k}}$ denotes the $k$-th column of ${{{\mathbf{R}}_{\mathbf{H}}}}$.
%%
%Note that, due to the randomness of ${{\mathbf{Y}}_{n,t}}$, the channel estimates ${{{\mathbf{\hat h}}}_{n,k}}$ is a random vector.

\section{Performance Analysis on ASINR}
In the $k$-th step of the MRC-SIC decoding, after cancelling the signals of the previous $k-1$ users, the post-processing signal of user $k$ in frame $n$ is given by
\begin{align}\label{C4:MRC-SIC}
{{\mathbf{y}}_{n,k}^{\mathrm{T}}}
&= \underbrace{{\mathbf{\hat h}}_{n,k}^{\mathrm{H}}{{\mathbf{\hat h}}_{n,k}}\sqrt {{\beta _k}} {{\mathbf{d}}_{n,k}}}_{\text{desired signal}} + \underbrace{{\mathbf{\hat h}}_{n,k}^{\mathrm{H}}\sum\limits_{l = 1}^{k} {\boldsymbol{\varepsilon} _{n,l}} \sqrt {{\beta _l}} {{\mathbf{d}}_{n,l}}}_{\text{residual interference}} \notag\\
&+ \underbrace{{\mathbf{\hat h}}_{n,k}^{\mathrm{H}}\sum\limits_{l = k + 1}^K {{{\mathbf{h}}_{n,l}}} \sqrt {{\beta _l}} {{\mathbf{d}}_{n,l}}}_{\text{inter-user interference}}
+ \underbrace{{\mathbf{\hat h}}_{n,k}^{\mathrm{H}}{{\mathbf{Z}}_{n,d}}}_{\text{noise}},
\end{align}
where ${{\mathbf{y}}_{n,k}} \in \mathbb{C}^{T\times 1}$, ${{{\mathbf{\hat h}}}_{n,k}} \in \mathbb{C}^{M\times 1}$ denotes the MMSE channel estimates of user $k$ in frame $n$, and ${\boldsymbol{\varepsilon} _{n,k}} = {{\mathbf{h}}_{n,k}} - {{{\mathbf{\hat h}}}_{n,k}}$ denotes the corresponding CEE.
In \eqref{C4:MRC-SIC}, we assume that the error propagation is only caused by the CEE but not affected by the erroneous in data detection.
It is a reasonable assumption if we can guarantee the ASINR of each user larger than a threshold to maintain the required bit-error-rate (BER) performance through the power control in the following.
To this end, we first define the instantaneous SINR of user $k$ in frame $n$ as:
\begin{equation}\label{C4:SINR}
{\mathrm{SIN}}{{\mathrm{R}}_{n,k}} = \frac{s_{n,k}}{G_{n,k} + Q_{n,k} + {\sigma ^2}}, \forall n,k,
\end{equation}
with $s_{n,k} = {{\mathbf{\hat h}}_{n,k}^{\mathrm{H}}{{{\mathbf{\hat h}}}_{n,k}}{\beta _k}}$, $G_{n,k} = \sum_{l = k + 1}^K {\frac{{{\mathbf{\hat h}}_{n,k}^{\mathrm{H}}{{\mathbf{h}}_{n,l}}{\mathbf{h}}_{n,l}^{\mathrm{H}}{{{\mathbf{\hat h}}}_{n,k}}}}{{{\mathbf{\hat h}}_{n,k}^{\mathrm{H}}{{{\mathbf{\hat h}}}_{n,k}}}}} {\beta _l}$, and
\begin{equation}\label{C4:Qnk}
  Q_{n,k} = \sum_{l = 1}^{k} {\frac{{{\mathbf{\hat h}}_{n,k}^{\mathrm{H}}\left( {{{\mathbf{h}}_{n,l}} - {{{\mathbf{\hat h}}}_{n,l}}} \right){{\left( {{{\mathbf{h}}_{n,l}} - {{{\mathbf{\hat h}}}_{n,l}}} \right)}^{\mathrm{H}}}{{{\mathbf{\hat h}}}_{n,k}}}}{{{\mathbf{\hat h}}_{n,k}^{\mathrm{H}}{{{\mathbf{\hat h}}}_{n,k}}}}} {\beta _l},
\end{equation}
while the ASINR of user $k$ is defined by:
\begin{equation}\label{C4:ASINR}
{\overline{\mathrm{SINR}}_k} = {\mathrm{E}}\left\{ {\frac{{{s_{n,k}}}}{{{G_{n,k}} + {Q_{n,k}} + {\sigma ^2}}}} \right\}, \forall k,
\end{equation}
where $\mathrm{E}\left\{ \cdot \right\}$ denotes the expectation operation.
In fact, for codeword-level SIC, it is the ASINR rather than the instantaneous SINR that determines the detection performance\cite{GaoAESINR}.
Yet, for mathematical tractability, in the sequel, we adopt the lower bound of $\overline{\mathrm{SINR}}_k$ proposed in \cite{GaoAESINR} as
\begin{equation}\label{C4:AESINR}
\mathrm{ASINR}_k = \frac{{\mathrm{E}\left\{ {s_{n,k}}\right\}}}{\mathrm{E}\left\{ {{G_{n,k}}} \right\} + \mathrm{E}\left\{ {{Q_{n,k}}} \right\} + {\sigma ^2}} \le {\overline{\mathrm{SINR}}_k}, \forall k.
\end{equation}
Variable $s_{n,k}$ denotes the desired signal power of user $k$ in frame $n$, ${G_{n,k}}$ denotes the IUI power in the $k$-th step of MRC-SIC, and ${{Q_{n,k}}}$ denotes the residual interference power caused by CEE.
Clearly, $s_{n,k}$, ${G_{n,k}}$, and ${Q_{n,k}}$ are functions of pilot and payload power allocation.
Now, we express the closed-form of \eqref{C4:AESINR} through the following theorem.

\begin{Thm}\label{C4:theorem1}
For two independent random vectors ${\mathbf{x}}, {\mathbf{y}} \in \mathbb{C}^{M\times 1}$ with distribution of ${\mathbf{x}} \sim \mathcal{CN}\left( \mathbf{0},\sigma_x^2 \mathbf{I}_M \right)$, we define a scalar random variable ${\phi} = {{{\mathbf{y}}^{\mathrm{H}}{\mathbf{x}}}}/{{\left| {\mathbf{y}} \right|}}$.
Then, ${\phi}$ is independent with ${\mathbf{y}}$ and it is distributed as a complex Gaussian distribution with zero mean and variance of $\sigma_x^2$, i.e., ${\phi} \sim \mathcal{CN}\left( 0, \sigma_x^2\right)$.
\end{Thm}
\begin{proof}
Please refer to Appendix \ref{C4:AppendixB} for a proof of Theorem \ref{C4:theorem1}.
\end{proof}

%\begin{Cor}
%Assuming that the channel estimation error is distributed as ${\boldsymbol{\varepsilon} _{n,k}} \sim \mathcal{CN}\left( \mathbf{0},\sigma_{k}^2 \mathbf{I}_M \right)$,
%%
%$\mathrm{ASINR}_k$ serves as a lower bound for $\mathrm{E}\left\{ {{\mathrm{SIN}}{{\mathrm{R}}_{n,k}}} \right\}$, i.e., $\mathrm{ASINR}_k \le \mathrm{E}\left\{ {{\mathrm{SIN}}{{\mathrm{R}}_{n,k}}} \right\}$.
%\end{Cor}
%
%\begin{proof}
%Owing to the orthogonal principle of MMSE channel estimation\cite{BigueshMMSE2006}, the channel estimation error ${\boldsymbol{\varepsilon} _{n,k}}$ is independent of the channel estimates ${{{\mathbf{\hat h}}}_{n,k}}$.
%%
%Therefore, according to Theorem \ref{C4:theorem1}, $s_{n,k}$ is independent of $G_{n,k}$, $Q_{n,k}$, and $O_{n,k}$.
%%
%Then, we have
%\vspace{-4mm}
%\begin{align}
%\hspace{-3mm}\mathrm{E}\left\{ {{\mathrm{SIN}}{{\mathrm{R}}_{n,k}}} \right\} & = \mathrm{E}\left\{ {s_{n,k}}\right\}\mathrm{E}\left\{\frac{1}{G_{n,k} + Q_{n,k}+ {O_{n,k}}}\right\} \notag\\[-1mm]
%& \overset{\left(a\right)}{\ge} \frac{\mathrm{E}\left\{ {s_{n,k}}\right\}}{\mathrm{E}\left\{{G_{n,k} + Q_{n,k}+ {O_{n,k}}}\right\}} = \mathrm{ASINR}_k,
%\end{align}
%\par\vspace{-2mm}\noindent
%where the inequality $\left(a\right)$ follows the Jensen's inequality.
%%
%It completes the proof.
%\end{proof}\cite{BigueshMMSE2006}

Now, following the MMSE channel estimation\cite{BigueshMMSE2006} and invoking Theorem \ref{C4:theorem1}, we can easily obtain $\mathrm{E}\left\{ {{s_{n,k}}} \right\} = {{\mathbf{A}}_k^{\mathrm{H}}}{{\mathbf{\Phi }}^{ - 1}}{{\mathbf{A}}_k}$ and $\mathrm{E}\left\{ {{G_{n,k}}} \right\} = \sum\limits_{l = k + 1}^K \nu_l^2 {\beta _l}$,
where ${\mathbf{\Phi }} = {{\mathbf{T}}^{\mathrm{H}}}{\mathbf{\Lambda}}{{\mathbf{R}}_{\mathbf{H}}}{\mathbf{\Lambda}}{\mathbf{T}} + {\sigma ^2}M{\mathbf{I}}_{T}$, ${{\mathbf{A}}_k} = {{\mathbf{T}}^{\mathrm{H}}}{\mathbf{\Lambda}}{\left\{ {{{\mathbf{R}}_{\mathbf{H}}}} \right\}_{:k}}$, and ${\left\{ {{{\mathbf{R}}_{\mathbf{H}}}} \right\}_{:k}}$ denotes the $k$-th column of ${{{\mathbf{R}}_{\mathbf{H}}}}$.
Then, based on the orthogonal principle of the MMSE channel estimation\cite{BigueshMMSE2006} and Theorem \ref{C4:theorem1}, we have $\mathrm{E}\left\{ {{Q_{n,k}}} \right\} = \sum\limits_{l = 1}^{k} \sigma_{l}^2 {\beta _l}$.
Therein, variable $\sigma_{k}^2$ is the variance of the CEE of user $k$, which is obtained based on equation (27) in \cite{BigueshMMSE2006} as $\sigma _k^2 = \nu_k^2 - \frac{1}{M} \mathbf{A}_k^{\mathrm{H}}\mathbf{\Phi}^{-1}\mathbf{A}_k$.
%\vspace{-2mm}
%\begin{equation}\label{C4:variance2}
%\sigma _k^2 = \nu_k^2 - \frac{1}{M} \mathbf{A}_k^{\mathrm{H}}\mathbf{\Phi}^{-1}\mathbf{A}_k.
%\vspace{-1mm}
%\end{equation}

Substituting $\mathrm{E}\left\{ {{s_{n,k}}} \right\}$, $\mathrm{E}\left\{ {{G_{n,k}}} \right\}$, and $\mathrm{E}\left\{ {{Q_{n,k}}} \right\}$ into \eqref{C4:AESINR} and using the matrix inverse lemma on ${{\mathbf{\Phi }}^{ - 1}}$, we have
\begin{equation}\label{C4:AESINR2}
\mathrm{ASINR}_k = \frac{{M\left( {\nu _k^2 - \frac{{{\sigma ^2}\nu _k^2}}{{{\sigma ^2} + {\alpha _k}\nu _k^2}}} \right){\beta _k}}}{{\sum\limits_{l = k + 1}^K {\nu _l^2} {\beta _l} + \sum\limits_{l = 1}^k {\frac{{{\sigma ^2}\nu _l^2}}{{{\sigma ^2} + {\alpha _l}\nu _l^2}}} {\beta _l} + {\sigma ^2}}}.
\end{equation}
We note that ${\mathrm{ASIN}}{{\mathrm{R}}_k}$ increases with $\left\{{\alpha _1}\right.$, $\ldots$, ${\alpha _k}$, $\left.{\beta _k}\right\}$, but decreases with $\left\{{\beta _1}\right.$, $\ldots$, ${\beta _{k-1}}$, ${\beta _{k+1}}$, $\ldots$, $\left.{\beta _K}\right\}$.
In other words, there exists a non-trivial trade-off between the allocation of pilot power and payload power.
Indeed, a higher pilot power ${\alpha _k}$ and payload power ${\beta _k}$ of user $k$ result in a higher ${\mathrm{ASIN}}{{\mathrm{R}}_k}$, while a higher payload power of other users ${\beta _1}$, $\ldots$, ${\beta _{k-1}}$, ${\beta _{k+1}}$, $\ldots$, ${\beta _{K}}$ will introduce more IUI for user $k$.
In contrast, high pilot powers ${\alpha _1}$, $\ldots$, ${\alpha _{k-1}}$ are beneficial to increase ${\mathrm{ASIN}}{{\mathrm{R}}_k}$ since they can reduce the residual interference by improving the quality of channel estimation.

\section{Joint Pilot and Payload Power Allocation}
The JPA design can be formulated to maximize the minimum weighted ${{\mathrm{ASIN}}{{\mathrm{R}}_k}}$ as follows:
\begin{align}
\label{C4:PilotPayloadPowerAllocation}
%\left( {{\mathrm{P}_0}}\right)
&\underset{\left\{{\alpha_1}, \ldots, {\alpha_K}\right\},\left\{{\beta_1}, \ldots, {\beta_K}\right\}}{\maxo} \;\;\underset{k}{\min} \;\;\{{c_k{\mathrm{ASIN}}{{\mathrm{R}}_k}} \} \notag\\
\mbox{s.t.}\;\;\;\;
%%%%%
&\mbox{C1: } \alpha_k T + \beta_k D \le E_{\mathrm{max}}, \forall k, \notag\\
&\mbox{C2: } \alpha_k, \beta_k \ge 0, \forall k, \notag\\
&\mbox{C3: } {\mathrm{ASIN}}{{\mathrm{R}}_k} \ge \gamma, \forall k.
\end{align}
\par\noindent
The constants ${\mathbf{c}} = \left[c_1, \ldots, c_K\right]$ are predefined weights for all the $K$ users.
Constraint C1 limits the pilot power ${\alpha_k}$ and the payload power $\beta_k$ with the maximum energy budget $E_{\mathrm{max}}$ for each user.
Constraint C2 ensures the non-negativity of ${\alpha_k}$ and $\beta_k$.
Constraint C3 requires the ASINR of user $k$ to be larger than a given threshold $\gamma$ to guarantee the data detection performance during SIC decoding.
Note that since the message of each user is decoded only once at the BS for uplink NOMA, a SIC decoding constraint is not required as imposed for downlink NOMA\cite{Hanif2016}.

This max-min problem formulation aims to reduce the effect of error propagation of the MRC-SIC decoding, which is dominated by the user with the minimum ASINR.
Furthermore, since the error propagation caused by the users at the forefront of the MRC-SIC decoding process, e.g. user 1, affects the data
detection of remaining undecoded users, and thus affects the system performance more significantly than other users.
Therefore, we have $c_1 \le c_2, \ldots, \le c_K$ to assign
different priorities to users in maximizing their ASINRs.
The formulated problem in (9) is a non-convex problem, where $\alpha_k$ and $\beta_k$ are coupled with each other severely in ${\mathrm{ASIN}}{{\mathrm{R}}_k}$.
%
%To circumvent the non-convexity, we propose the following transformation.
Defining new optimization variables ${t_k} = \frac{{{\sigma ^2}\nu _k^2}}{{{\sigma ^2} + {\alpha _k}\nu _k^2}}$, $\forall k$, the problem in (9) is equivalent to the following optimization problem \cite{ChiangGP}:
\begin{align} \label{C4:PilotPayloadPowerAllocation2}
%\left( {{\mathrm{P}_0}}\right)
&\underset{\left\{{t_1}, \ldots, {t_K}\right\},\left\{{\beta_1}, \ldots, {\beta_K}\right\},\lambda}{\maxo} \;\;\lambda\\
\mbox{s.t.}\;\;\;\;
%%%%%
&\mbox{C1: } {\sigma ^2}Tt_k^{ - 1} + D{\beta _k} \le {{{\sigma ^2}T}}/{{\nu _k^2}} + {E_{{\mathrm{max}}}}, \forall k, \notag\\
&\mbox{C2: } 0 < t_k \le {{\nu _k^2}}, \beta_k \ge 0, \forall k, \notag\\
&\mbox{C3: } \sum\limits_{l = k + 1}^K {\gamma \nu _l^2{\beta _l}\beta _k^{ - 1}}  + \sum\limits_{l = 1}^k {\gamma {t_l}{\beta _l}\beta _k^{ - 1}}  + \gamma {\sigma ^2}\beta _k^{ - 1} + M{t_k} \le M\nu _k^2, \forall k, \notag\\[-1mm]
&\mbox{C4: } \sum\limits_{l = k + 1}^K {\nu _l^2\lambda {\beta _l}\beta _k^{ - 1}}  + \sum\limits_{l = 1}^k {{t_l}\lambda {\beta _l}\beta _k^{ - 1}}  + {\sigma ^2}\lambda \beta _k^{ - 1} + M{c_k}{t_k} \le M{c_k}\nu _k^2, \forall k,\notag
\end{align}
where $\lambda > 0$ is an auxiliary optimization variable.
We can easily observe that the objective function and the functions on the left side of constraints C1, C3, and C4 in (10) are all valid posynomial functions\cite[Chapter~4]{Boyd2004}.
Therefore, the reformulated problem in (10) is a standard geometric programming (GP) problem\footnote{Actually, the problem transformation in transfoming (9) to (10) is standard and can be found in \cite{ChiangGP}.
However, without the proposed steps to simplify the performance analysis, the GP transformation\cite{ChiangGP} cannot be directly applied to the considered problem.}, which can be solved efficiently by off-the-shelf numerical solvers such as CVX\cite{cvx}.

\section{Simulation Results}
%\begin{table}[t]
%\center
%\caption{System Parameters in Simulations}
%  \begin{tabular}{cc}
%  \hline \hline
%    Cell radius & 400 m \\ \hline
%    Carrier frequency & 2.5 GHz \\ \hline
%    System Bandwidth  & 10 MHz \\ \hline
%    %Path loss model & $36\log_{10}(d)+128.1$\\ \hline
%    Symbols per coherence interval & 100 \\ \hline
%    %Coherence intervals per codeword ($N$) & 10 \\ \hline
%    %Number of antennas at BS ($M$) & 2 \\ \hline
%    %Number of users ($K$) & 4 \\ \hline
%    Energy budget per frame ($E_{\mathrm{max}}$) & $10^{-6}$ J \\ \hline
%    Noise power (${\sigma ^2}$) & -100 dBm \\ \hline
%    %Pilot symbols length ($T$) & 4 \\ \hline
%    %Data symbols length ($D$) & 96 \\ \hline
%    Modulation type & QPSK \\ \hline
%    Weight ${\mathbf{c}}$ & $\left[\frac{1}{8}, \frac{1}{8}, \frac{1}{4}, \frac{1}{2}\right]$ \\
%    \hline\hline
%  \end{tabular}
%\label{C4:SystemParameters}
%\end{table}

We use simulations to verify the developed performance analysis and evaluate the performance of the proposed JPA scheme for both uncoded and coded systems.
Two baseline schemes are introduced for comparison, where the equal power allocation (EPA) scheme sets the equal pilot and payload power, i.e., $\alpha_k = \beta_k = \frac{E_{\mathrm{max}}}{T+D}$, but the payload power allocation (PPA) scheme only fixes the pilot power as $\alpha_k = \frac{E_{\mathrm{max}}}{T+D}$ and optimizes over the payload power $\beta_k$ subject to the same constraint set as in \eqref{C4:PilotPayloadPowerAllocation2}.

In the simulations, we set $M=2$, $T = K = 4$, $D = 96$, ${\mathbf{c}} = \left[\frac{1}{8}, \frac{1}{8}, \frac{1}{4}, \frac{1}{2}\right]$, $\gamma = 5$ dB, $E_{\mathrm{max}} = 20$ J, and ${\sigma ^2} = -100$ dBm.
It is assumed that there are $T+D = 100$ symbols in a CTI.
All the $K$ users are uniformly distributed in a single cell with a cell radius of $400$ m.
The weight ${\mathbf{c}}$ is selected deliberately to alleviate the impact of the error propagation from the previous users during the MRC-SIC decoding, where the optimal weight selection will be considered in future work.
The 3GPP urban path loss model\cite{Access2010} is adopted and quadrature phase shift keying (QPSK) modulation is used for all the simulation cases.
For the coded systems, we adopt the standard turbo code as stated in the 3GPP technical specification\cite{3GPPReportTurboCode2016}.
We assume that one codeword is spread over $N=10$ coherence intervals, which results in a codeword length of $1920$ bits.
%
%All the results shown in the following are averaged over 5000 Monte
%Carlo simulations.

\subsection{Individual ASINR}
\begin{figure}[t!]
\centering
\includegraphics[width=4.5in]{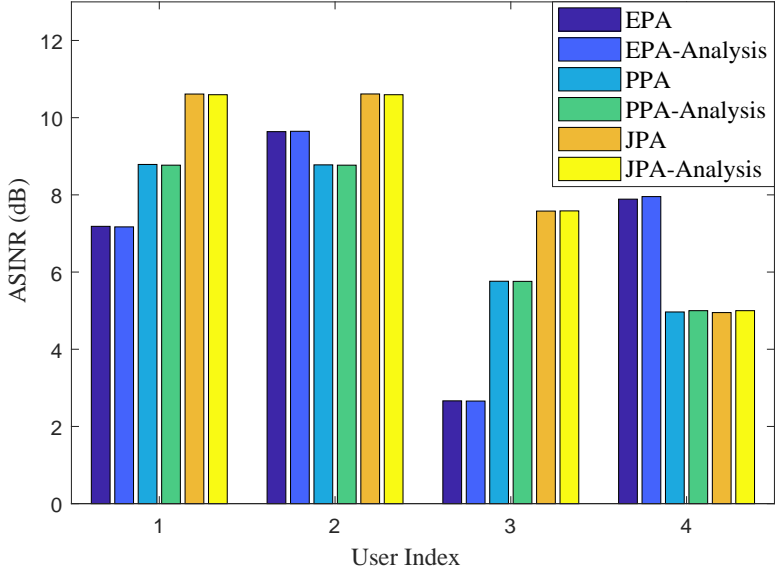}
\caption{Individual ASINR of uplink MIMO-NOMA with a MRC-SIC receiver.}
\label{C4:AESINR_Compare}
\end{figure}

Figure \ref{C4:AESINR_Compare} depicts the individual ASINR for the considered three schemes for uncoded systems.
We can observe that the simulation results match perfectly with the theoretical results in \eqref{C4:AESINR2}.
%
%For the EPA scheme, it can be observed that the strongest user (user 1) does not necessarily result in the highest ASINR and the weakest user (user 4) is not destined to have the lowest ASINR.
%%
%In fact, the users in the front (user 1 and user 2) and the end (user 4) of MRC-SIC decoding process will have a relatively higher ASINR owing to their smaller distance to the BS and the cancelation of interference from the previously decoded users, respectively.
%%
%On the contrary, the users in the middle (user 3) of MRC-SIC decoding usually have a relatively lower ASINR, since they suffer from both a large path loss and the remaining IUIs.
%%
%In fact, it is the worst user (user 3) with the lowest ASINR that introduces a severe error propagation for MRC-SIC decoding.
%
Besides, it can be observed that the lowest ASINR achieved by the PPA scheme and our proposed JPA scheme both occur at $5$ dB for user 4, which is much higher than the minimum ASINR provided by the EPA scheme occurring at $2.6$ dB for user 3.
This is owing to the adopted max-min principle and constraint C3 in the proposed problem formulation in \eqref{C4:PilotPayloadPowerAllocation}.
Nevertheless, it can be observed that our proposed scheme provides a $2$ dB higher ASINR than that of the PPA scheme for users 1, 2, and 3.
This is because our proposed scheme can utilize the energy more efficiently than that of the PPA scheme.
Moreover, our simulation results demonstrate that the optimal power allocation $\alpha^{*}_k$ and $\beta^{*}_k$ can satisfy the energy budget constraint C1 in \eqref{C4:PilotPayloadPowerAllocation}.
%
%In particular, to avoid a large IUI, some user may reduce the payload transmission power, while the saved energy can be utilized for improving the quality of the channel estimation.

Furthermore, the Jain's fairness index (JFI) of the weighted ASINR for the considered three schemes are given by $J_{\mathrm{EPA}} = 0.6174$,
$J_{\mathrm{PPA}} = 0.9436$, and
$J_{\mathrm{JPA}} = 0.9983$,
respectively.
The EPA scheme achieves the lowest JFI, while both the PPA and JPA schemes enjoy a high JFI since they are based on the max-min resource allocation in \eqref{C4:PilotPayloadPowerAllocation}.
In addition, our proposed JPA scheme offers a slightly higher JFI than that of the PPA scheme due to its efficient utilization of energy.

\subsection{Individual BER}

\begin{figure}[t!]
\centering
\includegraphics[width=4.5in]{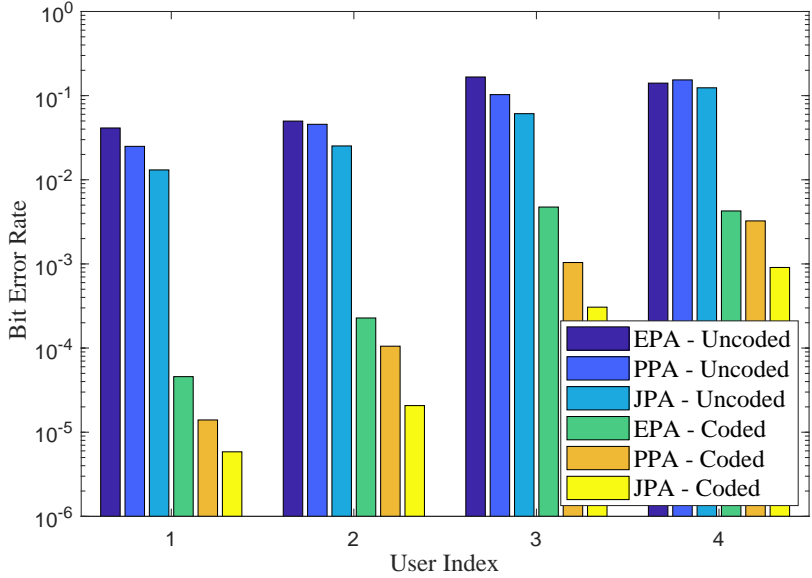}
\caption{Individual BER of uplink MIMO-NOMA with a MRC-SIC receiver.}
\label{C4:BER_Compare}
\end{figure}

Figure \ref{C4:BER_Compare} illustrates the individual BER performance for uncoded and coded systems.
We can observe that the coded system offers much lower BERs than the uncoded system owing to the coding gain.
For the EPA scheme, user 4 endures a high BER as user 3 despite it posses a larger ASINR than user 3, as shown in Figure \ref{C4:AESINR_Compare}.
This reveals the error propagation of the MRC-SIC decoding for the EPA scheme.
The PPA scheme can improve the BER performance for users 1, 2, and 3 compared to the EPA scheme, while it fails to relieve user 4 from high BER.
However, our proposed scheme always enjoys the lowest BER compared to the two baseline schemes for all the users, especially for coded systems.
In fact, our proposed scheme can mitigate the error propagation more efficiently compared to the PPA scheme by optimally balancing the pilot and payload power.
It is worth to note that, with constraint C3 in \eqref{C4:PilotPayloadPowerAllocation}, our proposed JPA scheme can guarantee the BER of all the users to be smaller than $10^{-3}$ for the coded systems, which validates the assumption about the sources of the error propagation in \eqref{C4:MRC-SIC}.

\subsection{BER versus Energy Budget}
\begin{figure}[t!]
\centering
\includegraphics[width=4.5in]{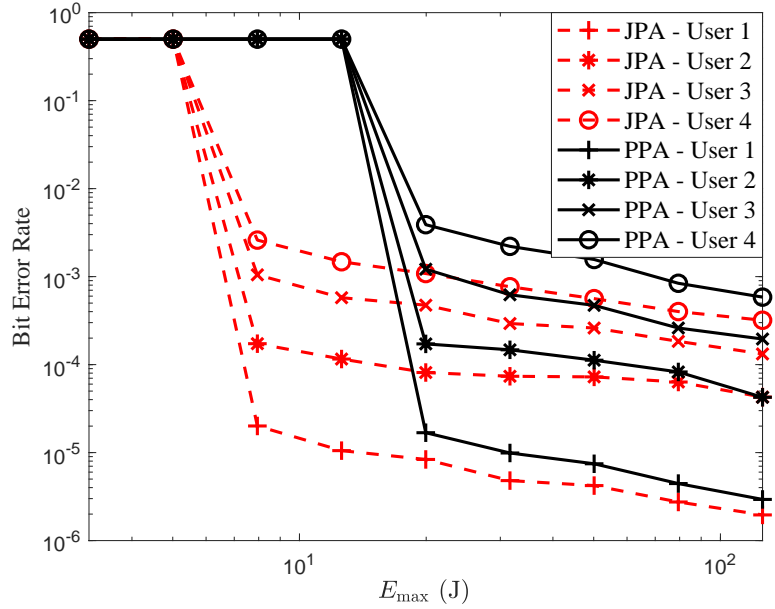}
\caption{BER performance versus energy budget $E_{\max}$ of uplink MIMO-NOMA with a MRC-SIC receiver.}
\label{C4:BER_SNR_Coded}
\end{figure}

For coded systems, Figure \ref{C4:BER_SNR_Coded} shows the BER performance of our proposed scheme over the PPA scheme versus the energy budget $E_{\mathrm{max}}$.
%
%Note that we omit constraint C3 for this simulation since the optimization problem in \eqref{C4:PilotPayloadPowerAllocation2} might be infeasible with $\gamma = 5$ dB when decreasing $E_{\mathrm{max}}$.
%
Note that we set $\text{BER} = 0.5$ if the optimization problem in \eqref{C4:PilotPayloadPowerAllocation} is infeasible to account the penalty of failure.
We can observe that our proposed scheme offers a much lower BER than that of the PPA scheme for all four users.
Interestingly, the BER performance gain is considerable in the moderate $E_{\mathrm{max}}$ regime, while it is marginal in the high $E_{\mathrm{max}}$ regime.
In fact, in the high $E_{\mathrm{max}}$ regime, the residual interference $Q_{n,k}$ vanishes owing to the high channel estimation accuracy.
%
%In other words, there is nearly no error propagation during MRC-SIC decoding.
%
Therefore, our proposed scheme can only offer diminishing gains in alleviating the impact of error propagation for further reducing the BER.
With the moderate $E_{\mathrm{max}}$, our proposed scheme can substantially improve the channel estimation, which can mitigate the residual interference during MRC-SIC decoding, and thus reduce the BER effectively.
In addition, it can be observed that an error floor for both
schemes appears at the BER region ranging from $10^{-2}$ to $10^{-5}$.
This early error floor is due to the joint effect of IUI, CEE, and the error propagation of the MRC-SIC decoding.
Note that an iterative receiver\cite{Xu2017} can be employed to lower the error floor level, which will be considered in our future work.

\section{Summary}
In this chapter, a joint pilot and payload power control scheme was proposed for uplink MIMO-NOMA systems with MRC-SIC receivers to mitigate the error propagation problem.
By taking into account the CEE, we analyzed the ASINR during the MRC-SIC decoding.
The JPA design was formulated as a non-convex optimization problem for maximizing the minimum weighted ASINR and was solved by geometric programming.
Simulation results verified our analysis and demonstrated that our proposed scheme is effective in mitigating the error propagation in SIC which enhances the BER performance, especially in the moderate energy budget regime.

\chapter{Optimal Resource Allocation
for Power-Efficient MC-NOMA with Imperfect Channel State Information}\label{C5:chapter5}

\section{Introduction}
\label{C5:sect1}

The previous two chapters mainly investigated the performance analysis and design for uplink NOMA systems.
Now, we move on to downlink NOMA systems and study the power efficient resource allocation design.

Resource allocation design plays a crucial role in exploiting the potential performance gain of NOMA systems, especially for multicarrier NOMA (MC-NOMA) systems\cite{Lei2016NOMA,Di2016sub,Sun2016Fullduplex}.
However, the existing works in \cite{Lei2016NOMA,Di2016sub,Sun2016Fullduplex,Liu2015b} focused on resource allocation design of NOMA based on the assumption of perfect CSIT.
Unfortunately, it is unlikely to acquire the perfect CSIT due to channel estimation error, feedback delay, and quantization error.
For the case of NOMA with imperfect CSIT, it is difficult for a base station to sort the users' channel gains and to determine the user scheduling strategy and SIC decoding policy.
More importantly, the imperfect CSIT may cause a resource allocation mismatch, which may degrade the system performance.
Therefore, it is interesting and more practical to design robust resource allocation strategy for MC-NOMA systems taking into account of CSIT imperfectness.

In the literature, there are three commonly adopted methods to address the CSIT imperfectness for resource allocation designs, including no-CSIT \cite{JorswieckNOCSI}, worst-case optimization \cite{WangWorstCase}, and stochastic approaches \cite{ZhangStochastic}.
The assumption of no-CSIT usually results in a trivial equal power allocation strategy without any preference in resource allocation \cite{JorswieckNOCSI}.
Besides, it is pessimistic to assume no-CSIT since some sorts of CSIT, e.g. imperfect channel estimates or statistical CSIT, can be easily obtained in practical systems exploiting handshaking signals.
The worst-case based methods guarantee the system performance for the maximal CSIT mismatch \cite{WangWorstCase}.
However, an exceedingly large amount of system resources are exploited for some worst cases that rarely happen.
In particular, for our considered problem, the worst-case method leads to a conservative resource allocation design, which may translate into a higher power consumption.
On the other hand, the stochastic methods aim at modeling the CSIT and/or the channel estimation error according to the long term statistic of the channel realizations \cite{ZhangStochastic}.
It is more meaningful than the no-CSIT method since the statistical CSIT is usually available based on the long term measurements in practical systems.
More importantly, the stochastic methods can guarantee the average system performance over the channel realizations with moderate system resources.
Therefore, in this chapter, we employed the stochastic method to robustly design the resource allocation strategy for MC-NOMA systems under imperfect CSIT.

Recently, green radio design has become an important focus in both academia and industry due to the growing demands of energy consumption and the arising environmental concerns around the world\cite{QingqingEE}.
To address the green radio design for NOMA systems, the authors in \cite{zhang2016energy} proposed an optimal power allocation strategy for a single-carrier NOMA system to maximize the energy efficiency, while a separate subcarrier assignment and power allocation scheme was proposed for MC-NOMA systems in \cite{FangEnergyEfficientNOMAJournal}.
To minimize the total power consumption, the authors in \cite{Lei2016NPM} designed a suboptimal ``relax-then-adjust" algorithm for MC-NOMA systems.
Nevertheless, the existing designs in \cite{zhang2016energy,FangEnergyEfficientNOMAJournal,Lei2016NPM} are based on the assumption of perfect CSIT which may not be applicable to MC-NOMA systems under imperfect CSIT.
To the best of the authors' knowledge, joint design of power allocation, rate allocation, user scheduling, and SIC decoding policy for power-efficient MC-NOMA under imperfect CSIT has not been reported yet.

In this chapter, we study the power-efficient resource allocation design for downlink MC-NOMA systems under imperfect CSIT, where each user imposes its own QoS requirement.
The joint design of power allocation, rate allocation, user scheduling, and SIC decoding policy is formulated as a non-convex optimization problem to minimize the total transmit power.
To facilitate the design of optimal SIC decoding order, we define the \emph{channel-to-noise ratio (CNR) outage threshold}, which includes the joint effect of channel conditions and QoS requirements of users.
Based on the optimal SIC decoding policy, we propose an optimal resource allocation algorithm via the branch-and-bound (B\&B) approach \cite{Konno2000,horst2013global,MARANASProofBB,Androulakis1995}, which serves as a performance benchmark for MC-NOMA systems.
Furthermore, to strike a balance between system performance and computational complexity, we propose a suboptimal iterative resource allocation algorithm based on difference of convex (D.C.) programming\cite{dinh2010local,VucicProofDC}, which has a polynomial time computational complexity and converges quickly to a close-to-optimal solution.
Our simulation results show that the proposed resource allocation schemes enable significant transmit power savings and are robust against channel uncertainty.

\section{System Model and Problem Formulation}
In this section, after introducing the adopted MC-NOMA system model under imperfect CSIT, we define the QoS requirement based on outage probability and formulate the power-efficient resource allocation design as a non-convex optimization problem.

\subsection{System Model}
A downlink MC-NOMA system\footnote{In this chapter, we focus on the power domain NOMA \cite{WeiSurvey2016} for the considered downlink communication scenario. Although the code-domain NOMA, such as sparse code multiple access (SCMA) \cite{Nikopour2013,Wang2015}, may outperform power-domain NOMA, SCMA is more suitable for the uplink communication where the reception complexity for information decoding is more affordable for base stations.} with one base station (BS) and $M$ downlink users is considered and shown in Figure \ref{C5:NOMA_model}.
All transceivers are equipped with single-antennas and there are $N_\mathrm{F}$ orthogonal subcarriers serving the $M$ users.
An overloaded scenario\footnote{Note that the proposed scheme in this chapter can also be applied to underloaded systems where the number of subcarrier $N_\mathrm{F}$ is larger than the number of users $M$, i.e., $N_\mathrm{F} > M$.
For the sake of presentation, we first focus on the overloaded scenario and then apply the proposed resource allocation algorithm to both the overloaded and underloaded systems in the simulations.
Then, our simulation results in Section \ref{C5:SimulationResults} demonstrate that the proposed scheme is more power-efficient than that of the OMA scheme for both overloaded and underloaded systems.} is considered in this chapter, i.e., $N_\mathrm{F}\le M$.
In addition, we assume that each of the $N_\mathrm{F}$ subcarriers can be allocated to at most two users via NOMA to reduce the computational complexity and delay incurred at receivers due to SIC decoding\footnote{In this chapter, we focus on the two-user MC-NOMA system since it is more practical and is more appealing in both industry\cite{Access2015} and academia \cite{Dingtobepublished,LiuSWIPT,ChenTwoUser,Sun2016Fullduplex}.
The generalization of the proposed algorithms to the case of serving multiple users on each subcarrier is left for future work.}.
As a result, we have an implicit condition $\left\lceil {\frac{M}{2}} \right\rceil \le N_\mathrm{F} \le M$ such that the system can serve at least $M$ users.
An example of a downlink MC-NOMA system with two users multiplexed on subcarrier $i$, $i \in \left\{ {1, \ldots ,N_\mathrm{F}} \right\}$, is illustrated in Figure \ref{C5:NOMA_model}.
A binary indicator variable $s_{i,m} \in \{0,1\}$, $i \in \left\{ {1, \ldots ,N_\mathrm{F}} \right\}$, and $m \in \left\{ {1, \ldots ,M} \right\}$, is introduced as the user scheduling variable, where it is one if subcarrier $i$ is assigned to user $m$, and is zero otherwise.
Thus, we have the following constraint for $s_{i,m}$:
\begin{equation}\label{C5:UserSchedulingConstraint}
\sum\limits_{m = 1}^M {{s_{i,m}}}  \le 2,\;\;\forall i.
\end{equation}

\begin{figure}[t]
\centering
\includegraphics[width=5in]{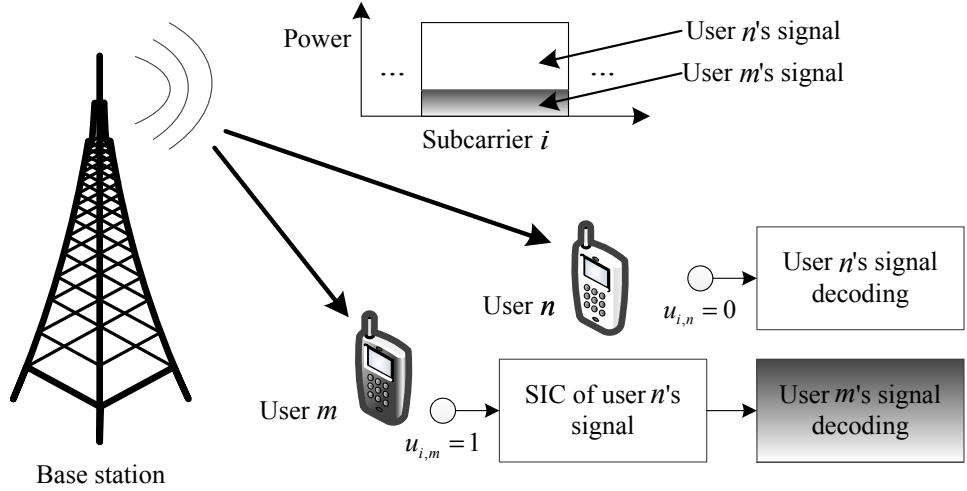}
\caption{A downlink MC-NOMA system where user $m$ and user $n$ are multiplexed on subcarrier $i$. The base station transmits two superimposed signals with different powers. User $m$ is selected to perform SIC, i.e., $u_{i,m} = 1$, while user $n$ is not selected to, i.e., $u_{i,n} = 0$. User $m$ first decodes and removes the signal of user $n$ before decoding its desired signal, while user $n$ directly decodes its own signal with user $m$'s signal treated as noise.}\label{C5:NOMA_model}
\end{figure}

\noindent At the BS side, the transmitted signal on subcarrier $i$ is given by
\begin{equation}\label{C5:Tx}
{x_i} = \sum\limits_{m = 1}^M {{s_{i,m}}\sqrt {{p_{i,m}}} {a_{i,m}}}, \;\; \forall i,
\end{equation}
where $a_{i,m}\in \mathbb{C}$ denotes the modulated symbol for user $m$ on subcarrier $i$ and $p_{i,m}$ is the allocated power for user $m$ on subcarrier $i$.
Different from most of the existing works on resource allocation of NOMA with perfect CSIT, e.g. \cite{sun2016optimal,Lei2015}, our model assumes imperfect CSIT. Under this condition, the BS needs to decide both the SIC decoding order, $u_{i,m} \in \left\{0,1\right\}$, and the rate allocation, $R_{i,m} > 0$, for each user on each subcarrier, which have critical impacts on the system power consumption.
The SIC decoding order variable is defined as follows:
\begin{equation}
u_{i,m} =
\left\{
\begin{array}{ll}
1 & \text{if}\;\text{user}\; m \;\text{on\;subcarrier}\; i \;\text{is selected to perform SIC,} \\
0 & \text{otherwise}.
\end{array}
\right.\label{C5:SICVariable}
\end{equation}

At the receiver side, the received signal at user $m$ on subcarrier $i$ is given by
\begin{equation}\label{C5:Rx}
{y_{i,m}} = {h_{i,m}}\sum\limits_{n = 1}^M {{s_{i,n}}\sqrt {{p_{i,n}}} {a_{i,n}}}  + {z_{i,m}},
\end{equation}
where $z_{i,m} \in \mathbb{C}$ denotes the additive white Gaussian noise (AWGN) for user $m$ on subcarrier $i$ with a zero-mean and variance $\sigma^2_{i,m}$, i.e., $z_{i,m} \sim {\cal CN}(0,\sigma^2_{i,m})$.
Variable ${h}_{i,m} = \frac{{g}_{i,m}}{\sqrt{\text{PL}_m}} \in \mathbb{C}$ denotes the channel coefficient between the BS and user $m$ on subcarrier $i$ capturing the joint effect of path loss and small-scale fading.
In particular, $\text{PL}_m$ denotes the path loss of user $m$, and we assume that the BS can accurately estimate the path loss of each user $\text{PL}_m$, $\forall m$, based on the long term measurements.
On the other hand, ${g}_{i,m} \sim {\cal CN}(0,1) $ denotes the small-scale fading, which is modeled as Rayleigh fading in this chapter \cite{Tse2005}.
Due to the channel estimation error and/or feedback delay, only imperfect CSIT is available for resource allocation.
To capture the channel estimation error, we model the channel coefficient for user $m$ on subcarrier $i$ as
\begin{equation}\label{C5:ChannelModel}
h_{i,m} = \hat{h}_{i,m} + \Delta h_{i,m},
\end{equation}
where $\hat{h}_{i,m}$ denotes the estimated channel coefficient for user $m$ on subcarrier $i$, $\Delta h_{i,m} \sim {\cal CN}(0,\frac{\kappa^2_{i,m}}{\text{PL}_m})$ denotes the corresponding CSIT error, and $\frac{\kappa^2_{i,m}}{\text{PL}_m} >0 $ denotes the variance of the channel estimation error. We assume that the channel estimates $\hat{h}_{i,m}$ and the channel estimation error $\Delta h_{i,m}$ are uncorrelated.
According to the SIC decoding order policy, both multiplexed users will choose to perform SIC or directly decode its own messages.

\subsection{QoS Requirements}
To facilitate our design, we define an outage probability on each subcarrier, which is commonly adopted in the literature for resource allocation design \cite{Zhu2009,Kwan_AF_2010}.
We assume that if the SIC of any user is failed, the user cannot decode its own messages, and thus an outage event occurs\cite{Ding2014}.
Therefore, if user $m$ on subcarrier $i$ is selected to perform SIC, i.e., $u_{i,m} = 1$, we have the outage probability as follows:
\begin{align}\label{C5:OutageProbabilityWithSIC}
&{\mathrm{P}}_{i,m}= {\mathrm{Pr}}\left\{ {C_{i,m}^{{\mathrm{SIC}}} < \sum\limits_{n = 1,n \ne m}^M {{s_{i,n}}R_{i,n}}} \left| {{{\hat h}_{i,m}}}, u_{i,m} = 1 \right.\right\} \notag\\
&+ {\mathrm{ Pr}}\left\{ {C_{i,m}^{{\mathrm{SIC}}} \ge \sum\limits_{n = 1,n \ne m}^M {{s_{i,n}}R_{i,n}},{C_{i,m}^{(1)}} < {R_{i,m}}} \left| {{{\hat h}_{i,m}}}, u_{i,m} = 1\right.\right\},
\end{align}
where ${\mathrm{P}}_{i,m}$ denotes the outage probability of user $m$ on subcarrier $i$ due to the channel uncertainty and $R_{i,m}$ denotes the allocated rate of user $m$ on subcarrier $i$.
Variable $C_{i,m}^{{\mathrm{SIC}}}$ denotes the achievable rate of user $m$ for decoding the interference from the other user on subcarrier $i$ and $C_{i,m}^{(1)}$ denotes the achievable rate of user $m$ on subcarrier $i$ for decoding its own message with a successful SIC. In our model, with two users multiplexing on subcarrier $i$, $C_{i,m}^{{\mathrm{SIC}}}$ and $C_{i,m}^{(1)}$ are given by
\begin{align}
C_{i,m}^{\mathrm{SIC}} &= {\log _2}\left( {1 + \frac{{{{\left| {{h_{i,m}}} \right|}^2}\sum\limits_{n = 1,\;n \ne m}^M {{s_{i,n}}{p_{i,n}}} }}{{{s_{i,m}}{p_{i,m}}{{\left| {{h_{i,m}}} \right|}^2} + \sigma _{i,m}^2}}} \right)\; \text{and} \notag\\
C_{i,m}^{(1)} &= {\log _2}\left( {1 + \frac{{{s_{i,m}}{p_{i,m}}{{\left| {{h_{i,m}}} \right|}^2}}}{{\sigma _{i,m}^2}}} \right),\label{C5:CapacityWithSIC}
\end{align}
respectively. Note that there is only one non-zero entry in the summations in \eqref{C5:OutageProbabilityWithSIC} and \eqref{C5:CapacityWithSIC}, which means that the interference to be cancelled during the SIC processing arises from only one user due to the constraint in \eqref{C5:UserSchedulingConstraint}.

On the other hand, if user $m$ on subcarrier $i$ is not selected to perform SIC, i.e., $u_{i,m} = 0$, we have the outage probability as follows:
\begin{equation}\label{C5:OutageProbabilityWithoutSIC}
{\mathrm{P}}_{i,m}' = {\mathrm{Pr}}\left\{ {C_{i,m}^{(2)} < {R_{i,m}}\left| {{{\hat h}_{i,m}}}, u_{i,m} = 0\right.} \right\},
\end{equation}
where ${\mathrm{P}}_{i,m}'$ denotes the outage probability of user $m$ on subcarrier $i$ due to the channel uncertainty. Variable ${C_{i,m}^{(2)}}$ denotes the achievable rate of user $m$ on subcarrier $i$ for decoding its own messages without performing SIC and it is given by
\begin{equation}\label{C5:CapacityWithoutSIC}
{C_{i,m}^{(2)}} = {\log _2}\left( {1 + \frac{{{s_{i,m}}{p_{i,m}}{{\left| {{h_{i,m}}} \right|}^2}}}{{{{\left| {{h_{i,m}}} \right|}^2}\sum\limits_{n=1,\;n \ne m}^M {{s_{i,n}}{p_{i,n}}}  + \sigma _{i,m}^2}}} \right).
\end{equation}

Now, we define the QoS requirement of user $m$ on subcarrier $i$ as follows:
\begin{equation}\label{C5:QoSConstraint}
{\mathrm{P}}_{i,m}^{{\mathrm{out}}} = s_{i,m} \left\{u_{i,m}{\mathrm{P}}_{i,m} + \left( {1 - {u_{i,m}}} \right){\mathrm{P}}_{i,m}' \right\} \le {\delta _{i,m}},
\end{equation}
where $\delta_{i,m}$, $0 \le \delta_{i,m} \le 1$, denotes the required outage probability of user $m$ on subcarrier $i$.
When user $m$ is not assigned on subcarrier $i$, i.e., $s_{i,m} = 0$, this inequality is always satisfied.

In OMA systems, subcarrier $i$ will be allocated to user $m$ exclusively, i.e., $\sum_{m = 1}^M {{s_{i,m}}}  = 1$.
In this case, we have ${C_{i,m}^{(1)}} = {C_{i,m}^{(2)}}$, and ${\mathrm{P}}_{i,m}'$  denotes the corresponding outage probability of user $m$ on subcarrier $i$.
Also, performing SIC at user $m$ is not required, and thus we have $u_{i,m}=0$,  ${\mathrm{P}}_{i,m}^{{\mathrm{out}}} =  {\mathrm{P}}_{i,m}'$.
In other words, the outage probability defined in the considered MC-NOMA system generalizes that of OMA systems as a subcase.
{Furthermore, although this work considers the robust resource allocation design for MC-NOMA system with imperfect CSIT, one naive option in practice can be performing resource allocation by treating the estimated channel coefficient $\hat{h}_{i,m}$ in \eqref{C5:ChannelModel} as perfect CSIT.}

\subsection{Optimization Problem Formulation}
We aim to jointly design the power allocation, rate allocation, user scheduling, and SIC
decoding policy for minimizing the total transmit power of the considered downlink MC-NOMA system under imperfect CSIT.
The joint resource allocation design can be formulated as the following optimization problem:
\begin{align} \label{C5:P1}
%\left( {{\mathrm{P}_0}}\right)
&\underset{\mathbf{s},\;\mathbf{u},\;\mathbf{p},\;\mathbf{r}}{\mino}\,\, \,\, \sum\limits_{m = 1}^M {\sum\limits_{i = 1}^{N_\mathrm{F}} {s_{i,m}p_{i,m}}}\\
\mbox{s.t.}\;\;
%%%%%
&\mbox{{C1}: } s_{i,m} \in \left\{0,\;1\right\},\;\forall i,m,\;\;\mbox{{C2}: } u_{i,m} \in \left\{0,\;1\right\},\;\forall i,m,\notag\\
&\mbox{{C3}: } p_{i,m} \ge 0,\;\forall i,m,\;\;\hspace{9.5mm}\mbox{{C4}: } R_{i,m} \ge 0,\;\forall i,m,\notag\\
&\mbox{{C5}: } {\mathrm{P}}_{i,m}^{{\mathrm{out}}} \le {\delta _{i,m}},\; \forall i,m,\;\;\hspace{4mm}\mbox{{C6}: } \sum\limits_{m = 1}^M {{s_{i,m}}}  \le 2,\;\forall i,\notag\\
&\mbox{{C7}: } \sum\limits_{i = 1}^{N_{\mathrm{F}}} {{s_{i,m}}{{R}_{i,m}}}  \ge R_m^{\mathrm{total}},\;\forall m,\notag
\end{align}
where $\mathbf{s} \in \mathbb{R}^{N_{\mathrm{F}}M \times 1}$, $\mathbf{u} \in \mathbb{R}^{N_{\mathrm{F}} M\times1}$, $\mathbf{p} \in \mathbb{R}^{N_{\mathrm{F}}M \times1}$, and $\mathbf{r} \in \mathbb{R}^{N_{\mathrm{F}} M \times 1}$ denote the sets of optimization variables $s_{i,m}$, $u_{i,m}$, $p_{i,m}$, and $R_{i,m}$.
Constraints C1 and C2 restrict the user scheduling variables and SIC decoding order variables to be binary, respectively.
Constraints C3 and C4 ensure the non-negativity of the power allocation variables and the rate allocation variables, respectively.
Constraint C5 is the QoS constraint of outage probability for user $m$ on subcarrier $i$. According to \eqref{C5:QoSConstraint}, C5 is inactive when $s_{i,m}=0$.
Constraint C6 is imposed to ensure that at most two users are multiplexed on each subcarrier.
Note that our problem includes OMA as a subcase when $\sum\limits_{m = 1}^M {{s_{i,m}}} = 1$ in C6.
In C7, constant $R_m^{\mathrm{total}} > 0$ denotes the required minimum total data rate of user $m$.
In particular, constraint C7 is introduced for rate allocation such that the required minimum total data rate of user $m$ can be guaranteed.
We note that the rate allocation $R_{i,m}$ for user $m$ on subcarrier $i$ may be very low, but the achievable rate of user $m$ is bounded below by $R_m^{\mathrm{total}}$.

This problem in \eqref{C5:P1} is a mixed combinatorial non-convex optimization problem.
In general, there is no systematic and computational efficient approach to solve \eqref{C5:P1} optimally.
The combinatorial nature comes from the binary constraints C1 and C2 while the non-convexity arises in the QoS constraint in C5.
Besides, the coupling between binary variables and continuous variables in constraint C5 yields an intractable problem.
However, we note that the power allocation variables and SIC decoding order variables are only involved in the QoS constraint C5 among all the constraints.
Therefore, by exploiting this property, we attempt to simplify the optimization problem in \eqref{C5:P1} to facilitate the resource allocation design in the next section.

\section{Problem Transformation}
In this section, we first propose the optimal SIC policy on each subcarrier taking into account the impact of imperfect CSIT.
Then, the minimum total transmit power per subcarrier is derived.
Subsequently, we transform the optimization problem which paves the way for the design of the optimal resource allocation.
\subsection{Optimal SIC Policy Per Subcarrier}
Under perfect CSIT, for a two-user NOMA downlink system, it is well-known that the optimal SIC decoding order is the descending order of channel gains for maximizing the total system sum rate \cite{Dingtobepublished}.
However, under imperfect CSIT, it is not possible to decide the SIC decoding order by comparing the actual channel gains between the multiplexed users.
To facilitate the design of resource allocation for the case of NOMA with imperfect CSIT, we define an \emph{channel-to-noise ratio (CNR) outage threshold} in the following, from which we can decide the optimal SIC decoding order to minimize the total transmit power.

\begin{Def}[CNR Outage Threshold]
For user $m$ on subcarrier $i$ with the estimated channel coefficient $\hat{h}_{i,m}$, the noise power ${{\sigma _{i,m}^2}}$, and the required outage probability ${\delta _{i,m}}$, a CNR outage event occurs when the CNR, $\frac{{{{\left| {{h_{i,m}}} \right|}^2}}}{{\sigma _{i,m}^2}}$, is smaller than a CNR outage threshold ${\beta _{i,m}}$. The CNR outage probability of user $m$ on subcarrier $i$ can be written as follows:
\begin{equation}\label{C5:OutageThreshold1}
{\mathrm{Pr}}\left\{ {\frac{{{{\left| {{h_{i,m}}} \right|}^2}}}{{\sigma _{i,m}^2}} < {\beta _{i,m}}} \left| {{{\hat h}_{i,m}}} \right. \right\} = {\delta _{i,m}},\;\;\forall i,m.
\end{equation}
Therefore, the CNR outage threshold is given by
\begin{equation}\label{C5:OutageThreshold2}
{\beta _{i,m}} = \frac{{F_{\left. {{{\left| {{h_{i,m}}} \right|}^2}} \right|{{\hat h}_{i,m}}}^{ - 1}\left( {{\delta _{i,m}}} \right)}}{{\sigma _{i,m}^2}},\;\;\forall i,m,
\end{equation}
where ${F_{\left. {{{\left| {{h_{i,m}}} \right|}^2}} \right|{{\hat h}_{i,m}}}}\left( x \right)$, $x \ge 0$, is the conditional cumulative distribution function\footnote{We note that if the fading channel is not Rayleigh distributed, the proposed schemes in this chapter are still applicable whereas we only need to update ${F_{\left. {{{\left| {{h_{i,m}}} \right|}^2}} \right|{{\hat h}_{i,m}}}}\left( x \right)$ according to the distribution of small-scale fading.} (CDF) of channel power gain of user $m$ on subcarrier $i$.
In fact, according to the channel model in \eqref{C5:ChannelModel}, ${F^{-1}_{\left. {{{\left| {{h_{i,m}}} \right|}^2}} \right|{{\hat h}_{i,m}}}}\left( x \right)$ is the inverse of a noncentral chi-square CDF\footnote{The inverse function of a noncentral chi-square CDF can be computed efficiently by standard numerical solvers or implemented as a look-up table for implementation.} with degrees of freedom of 2 and a noncentrality parameter of $\frac{{\left| {{\hat h}_{i,m}} \right|^2}}{\kappa^2_{i,m}}{\text{PL}_m}$, ${\kappa^2_{i,m}}>0$. Note that if perfect CSIT is available, i.e., ${\kappa^2_{i,m}}=0$, there is no CNR outage caused by channel estimation error and we have $\beta_{i,m} = \frac{{{{\left| {{\hat h_{i,m}}} \right|}^2}}}{{\sigma _{i,m}^2}} = \frac{{{{\left| {{ h_{i,m}}} \right|}^2}}}{{\sigma _{i,m}^2}}$.
\end{Def}

We note that the CNR outage defined in \eqref{C5:OutageThreshold1} is different from that of the outage events defined in \eqref{C5:OutageProbabilityWithSIC} and \eqref{C5:OutageProbabilityWithoutSIC}.
In particular, the former does not involve the required target data rate of users, while the latter outage event occurs when the achievable rate is smaller than a given target data rate.
In the rest of the chapter, we refer to CNR outage as in \eqref{C5:OutageThreshold1} if it is stated explicitly.
Otherwise, it refers to the outage defined as in \eqref{C5:OutageProbabilityWithSIC} and \eqref{C5:OutageProbabilityWithoutSIC}.
The CNR outage threshold defined in \eqref{C5:OutageThreshold2} involves the estimated channel gain, the channel estimation error distribution, the noise power, and the required outage probability.
It captures the joint effect of channel conditions and QoS requirements in the statistical sense for imperfect CSIT scenarios.
Specifically, the user with higher CNR outage threshold may have better channel condition and/or lower required outage probability, and vice versa.
In fact, the CNR outage threshold serves as a criterion for the optimal SIC decoding order, which is summarized in the following theorem.

\begin{Thm}[Optimal SIC Decoding Order]\label{C5:Theorem1}
Given user $m$ and user $n$ multiplexing on subcarrier $i$,
the optimal SIC decoding order is determined by the CNR outage thresholds as follows:
\begin{equation}
\left(u_{i,m},u_{i,n}\right) =
\left\{
\begin{array}{ll}
\left(1,0\right) & \text{if}\;\beta _{i,m} \ge \beta _{i,n}, \\
\left(0,1\right) & \text{if}\;\beta _{i,m} < \beta _{i,n}.
\end{array}
\right.\label{C5:SICDecodingOrderNOMA}
\end{equation}
which means that the user with a higher CNR outage threshold will perform SIC decoding.
\end{Thm}
\begin{proof}
Please refer to Appendix \ref{C5:AppendixA} for a proof of Theorem \ref{C5:Theorem1}.
\end{proof}

With only a single user allocated on subcarrier $i$, i.e., $\sum_{n = 1}^M {{s_{i,n}}} = 1$, it reduces to the OMA scenario and there is no need of SIC decoding for all the users, and thus we have
\begin{equation}\label{C5:SICDecodingOrderOMA1}
u_{i,m} = 0, \;\forall m,\;\;\text{if}\;\sum\limits_{n = 1}^M {{s_{i,n}}} = 1.
\end{equation}
Note that the optimal SIC decoding policies defined in \eqref{C5:SICDecodingOrderNOMA} and \eqref{C5:SICDecodingOrderOMA1} are conditioned on any given feasible user scheduling strategy satisfying constraints C1, C6, and C7 in (11) for minimizing the total transmit power.
More importantly, given any point in the feasible solution set spanned by C1, C4, C6, and C7 in \eqref{C5:P1}, our proposed optimal SIC decoding order always consumes the minimum total transmit power to satisfy the QoS constraint C5.
By exploiting constraint C5, the underlying relationship between SIC decoding order variables $u_{i,m}$ and user scheduling variables $s_{i,m}$ is revealed for NOMA and OMA scenarios in \eqref{C5:SICDecodingOrderNOMA} and \eqref{C5:SICDecodingOrderOMA1}, respectively.

\begin{remark}\label{C5:remark1}
In Theorem \ref{C5:Theorem1}, it is noteworthy that the optimal SIC decoding order only depends on the CNR outage threshold, $\beta _{i,m}$, and is independent of the target data rates of users.
This observation is reasonable. Let us first recall from the basic principle for SIC decoding for the case of NOMA with perfect CSIT\cite{Tse2005} and then extend it to the case of imperfect CSIT.
Specifically, for perfect CSIT, the strong user (with a higher channel gain) can decode the messages of the weak user (with a lower channel gain), if we can guarantee that the weak user can decode its own messages, no matter who has a higher target data rate.
This is due to the fact that the achievable rate for the strong user to decode the messages of the weak user is always higher than that of the weak user to decode its own.
On the other hand, if the weak user performs SIC, due to its worse channel condition, a higher transmit power is required such that the strong user's messages are decodable at the weak user, no matter whose target data rate is higher.
Similarly, given user $m$ and user $n$ multiplexed on subcarrier $i$ under imperfect CSIT, according to \eqref{C5:OutageProbabilityWithSIC}, \eqref{C5:OutageProbabilityWithoutSIC}, and \eqref{C5:OutageThreshold1}, we have the following implication under the condition of $\beta _{i,m} \ge \beta _{i,n}$:
\begin{equation}
{\mathrm{Pr}}\left\{ {C_{i,n}^{(2)} < {R_{i,n}}\left| {{{\hat h}_{i,n}}}, u_{i,n} = 0\right.} \right\}\le {\delta _{i,n}},
\Rightarrow {\mathrm{Pr}}\left\{ C_{i,m}^{{\mathrm{SIC}}} < {R_{i,n}} \left| {{{\hat h}_{i,m}}}, u_{i,m} = 1 \right.\right\}\le {\delta _{i,m}},
\end{equation}
which means that the SIC process at the user with a higher CNR outage threshold will always satisfy its QoS constraint if the other user's (with a lower CNR outage threshold) decoding process satisfy its own QoS constraint, no matter whose target data rate is higher.
Also, if the user with a lower CNR outage threshold performs SIC, it requires an extra power to guarantee the QoS constraint of the SIC process, no matter who requires a higher data rate.
Therefore, the optimal SIC decoding order is determined according to the CNR outage thresholds, and it is independent of the target data rates.
\end{remark}

\subsection{Minimum Total Transmit Power Per Subcarrier}
In this section, we exploit Theorem \ref{C5:Theorem1} to further simplify the problem in \eqref{C5:P1} via expressing the power allocation variables $p_{i,m}$ in terms of CNR outage threshold and target data rate.
Given user $m$ and user $n$ multiplexed on subcarrier $i$, i.e., $s_{i,m}=s_{i,n}=1$, according to the proof of Theorem \ref{C5:Theorem1} in Appendix \ref{C5:AppendixA}, the minimum total transmit power required on subcarrier $i$ to satisfy the QoS constraint C5 in \eqref{C5:P1} can be generally represented as:
\begin{equation}\label{C5:TotalPowerConsumptionNOMA}
p^{\mathrm{total}}_{(i,m,n)} = \frac{{{\gamma _{i,m}}}}{{{\beta _{i,m}}}} + \frac{{{\gamma _{i,n}}}}{{{\beta _{i,n}}}} + \frac{{{\gamma _{i,m}}{\gamma _{i,n}}}}{{\max \left( {{\beta _{i,m}},{\beta _{i,n}}} \right)}},
\end{equation}
where the subscript $({i,m,n})$ denotes that users $m$ and $n$ are multiplexed on subcarrier $i$, $\gamma _{i,m}$ denotes the required signal-to-interference-plus-noise ratio (SINR) to support the target data rate of user $m$ on subcarrier $i$, and it is given by ${\gamma _{i,m}} = {2^{{R_{i,m}}}} - 1$. Particularly, the power allocations for users $m$ and $n$ are given according to their CNR outage threshold as follows:
\begin{equation}\label{C5:PowerAllocation4}
\left(p_{i,m},\;p_{i,n}\right) =
\left\{
\begin{array}{ll}
\left(\frac{{{\gamma _{i,m}}}}{{{\beta _{i,m}}}},\; \frac{{{\gamma _{i,n}}}}{{{\beta _{i,n}}}} + \frac{{{\gamma _{i,n}}{\gamma _{i,m}}}}{\beta _{i,m}}\right) & \text{if}\; \beta _{i,m} \ge \beta _{i,n}, \\
\left(\frac{{{\gamma _{i,m}}}}{{{\beta _{i,m}}}} + \frac{{{\gamma _{i,m}}{\gamma _{i,n}}}}{\beta _{i,n}},\; \frac{{{\gamma _{i,n}}}}{{{\beta _{i,n}}}}\right) & \text{if}\; \beta _{i,m} < \beta _{i,n},
\end{array}
\right.
\end{equation}

We note that, in the OMA scenario, if $\sum_{n = 1}^M {{s_{i,n}}} = 1$ and $s_{i,m}=1$, the required minimum transmit power to satisfy the QoS constraint C5 in \eqref{C5:P1} is given by
\begin{equation}\label{C5:TotalPowerConsumptionOMA}
p^{\mathrm{total}}_{(i,m)} =p_{i,m}= \frac{{{\gamma _{i,m}}}}{{{\beta _{i,m}}}},
\end{equation}
where the subscript $({i,m})$ denotes that users $m$ is assigned on subcarrier $i$ exclusively.
By combining the user scheduling variables of subcarrier $i$, $s_{i,m},\;\forall m \in \left\{ {1, \ldots ,M} \right\}$, with \eqref{C5:TotalPowerConsumptionNOMA}, the total transmit power on subcarrier $i$ can be generally expressed as:
\begin{equation}\label{C5:TotalPowerConsumptionUniform}
p^{\mathrm{total}}_{i} \left(\mathbf{s},\boldsymbol{\gamma} \right) = \sum\limits_{m = 1}^M {\frac{{{s_{i,m}}{\gamma _{i,m}}}}{{{\beta _{i,m}}}}}  + \sum\limits_{m = 1}^{M-1} {\sum\limits_{n = m + 1}^M {\frac{{{s_{i,m}}{\gamma _{i,m}}{s_{i,n}}{\gamma _{i,n}}}}{{\max \left( {{\beta _{i,m}},{\beta _{i,n}}} \right)}}} },
\end{equation}
where $\boldsymbol{\gamma}  \in \mathbb{R}^{N_{\mathrm{F}} M \times 1} $ denotes the set of $\gamma_{i,m}$. Note that \eqref{C5:TotalPowerConsumptionUniform} subsumes the cases of the power consumption in the OMA scenario in \eqref{C5:TotalPowerConsumptionOMA}.
In fact, equation \eqref{C5:TotalPowerConsumptionUniform} unveils the relationship between the required minimum total transmit power on subcarrier $i$, the user scheduling variables, and the allocated data rates.
Since the SIC decoding process on each subcarrier is independent with each other, we have the total transmit power of all the subcarriers as follows:
\begin{align}\label{C5:TotalPowerConsumptionUniform2}
p^{\mathrm{total}} \left(\mathbf{s},\boldsymbol{\gamma} \right) =  \sum\limits_{i = 1}^{{N_{\mathrm{F}}}} p^{\mathrm{total}}_{i}\left(\mathbf{s},\boldsymbol{\gamma} \right).
\end{align}

Recall that Theorem \ref{C5:Theorem1} is derived from the QoS constraint C5 in \eqref{C5:P1}.
Therefore, given any user scheduling and rate allocation strategy in the feasible solution set spanned by constraints C1, C4, C6, and C7, we obtain the optimal SIC decoding order and power allocation solution from \eqref{C5:SICDecodingOrderNOMA}, \eqref{C5:SICDecodingOrderOMA1}, \eqref{C5:PowerAllocation4}, and \eqref{C5:TotalPowerConsumptionOMA}, which can satisfy the QoS constraint C5 with the minimum power consumption in \eqref{C5:TotalPowerConsumptionUniform2}.
In other words, the problem in \eqref{C5:P1} can be transformed equivalently to a simpler one to minimize the power consumption in \eqref{C5:TotalPowerConsumptionUniform2} w.r.t. the user scheduling and rate allocation variables.

\begin{remark}\label{C5:remark2}
Comparing the minimum total transmit power for NOMA and OMA in \eqref{C5:TotalPowerConsumptionNOMA} and \eqref{C5:TotalPowerConsumptionOMA}, respectively, we obtain $p^{\mathrm{total}}_{(i,m,n)} > p^{\mathrm{total}}_{(i,m)} + p^{\mathrm{total}}_{(i,n)}$.
At the first sight, it seems that OMA is more power-efficient than NOMA as the third term in \eqref{C5:TotalPowerConsumptionNOMA} is the extra power cost for users' multiplexing.
However, this comparison is unfair for NOMA since both schemes require different spectral efficiencies.
To keep the same spectral efficiency in the considered scenarios, the required target data rate for OMA users in \eqref{C5:TotalPowerConsumptionOMA} should be doubled\cite{Ding2015b,sun2016optimal}.
Due to the combinatorial nature of user scheduling, it is difficult to prove that the proposed MC-NOMA scheme is always more power-efficient than the OMA scheme with the optimal resource allocation strategy.
In fact, for a special case with $M = 2N_{\mathrm{F}}$, this conclusion can be proved mathematically based on our previous work in \cite{Wei2016NOMA}.
For the general case, we rely on the simulation results to demonstrate the power savings of our proposed schemes compared to the OMA scheme.
\end{remark}

\subsection{Problem Transformation}
Based on Theorem \ref{C5:Theorem1}, the problem in \eqref{C5:P1} is equivalent to the following optimization problem:
\begin{align}\label{C5:P2}
&\underset{\mathbf{s},\;\boldsymbol{\gamma}}{\mino}\,\, \,\,  p^{\mathrm{total}} \left(\mathbf{s},\boldsymbol{\gamma} \right) \\
\mbox{s.t.}\;\;
%%%%%
&\mbox{C1},\;\mbox{C6},\;\mbox{C4: } \gamma_{i,m} \ge 0,\;\forall i,m,\notag\\
&\mbox{C7: } \sum\limits_{i = 1}^{{N_{\mathrm{F}}}} {s_{i,m}}{{{\log }_2}\left( {1 +{\gamma _{i,m}}} \right)}  \ge {{R}_m^{\mathrm{total}}},\;\forall m,\notag
\end{align}
where the rate allocation variables ${{R}_{i,m}}$ are replaced by their equivalent optimization variables $\gamma_{i,m}$ in C4 and C7.

The reformulated problem in \eqref{C5:P2} is simpler than that of the problem in \eqref{C5:P1}, since the number of optimization variables is reduced and the QoS constraint C5 is safely removed. However, \eqref{C5:P2} is also difficult to solve.
In particular, C1 are binary constraints, and there are couplings between the binary variables and continuous variables in both the objective function and constraint C7.
In fact, the problem in \eqref{C5:P2} is a non-convex mixed-integer nonlinear programming problem (MINLP), which is NP-hard\cite{floudas1995nonlinear} in general. Now, we transform the problem from \eqref{C5:P2} to:
\begin{align}\label{C5:P3}
&\underset{\mathbf{{s}},\;\boldsymbol{{\gamma}}}{\mino}\,\, \,\,  p^{\mathrm{total}}_{\boldsymbol{{\gamma}}} \left(\boldsymbol{\gamma} \right) \\
\notag\mbox{s.t.}\;\;
%%%%%
&\mbox{C1, C4, C6, } \widetilde{\text{C7}}\mbox{: } \sum\limits_{i = 1}^{{N_{\mathrm{F}}}} {{{\log }_2}\left( {1 +{{\gamma} _{i,m}}} \right)}  \ge {{R}_m^{\mathrm{total}}},\;\forall m, \notag\\
&\mbox{C8: } {{\gamma}}_{i,m} = {{s}_{i,m}} {\gamma}_{i,m},\;\forall i,m,\notag
\end{align}
with the new objective function w.r.t. to $\boldsymbol{\gamma}$ as
\begin{equation}
p^{\mathrm{total}}_{\boldsymbol{{\gamma}}} \left(\boldsymbol{\gamma} \right) = \sum\limits_{i = 1}^{{N_{\mathrm{F}}}} {\sum\limits_{m = 1}^M {\frac{{{{\gamma}}_{i,m}}}{{{\beta _{i,m}}}}} }  + \sum\limits_{i = 1}^{{N_{\mathrm{F}}}} {\sum\limits_{m = 1}^{M-1} {\sum\limits_{n = m + 1}^M {\frac{{{{\gamma} _{i,m}}{{\gamma} _{i,n}}}}{{\max \left( {{\beta _{i,m}},{\beta _{i,n}}} \right)}}} } }.
\end{equation}
Constraint C8 is imposed to preserve the original couplings between the binary variables and the continuous variables. Constraint $\widetilde{\text{C7}}$ is obtained from constraint C7 in \eqref{C5:P2} by employing the equality ${s_{i,m}}{{{\log }_2}\left( {1 +{\gamma _{i,m}}} \right)} = {{{\log }_2}\left( {1 +{s_{i,m}}{{\gamma} _{i,m}}} \right)}$ for ${s}_{i,m} \in \left\{0,\;1\right\}$. Clearly, the problem in \eqref{C5:P3} is equivalent to the problem in \eqref{C5:P2}, when we substitute constraint C8 into the objective function and constraint $\widetilde{\text{C7}}$. Therefore, in the sequel, we focus on solving \eqref{C5:P3}.

\section{Optimal Solution}
In this section, to facilitate the design of the optimal resource allocation algorithm, we first relax the binary constraint C1 and augment the coupling constraint C8 into objective function by introducing a penalty factor.
Then, an optimal resource allocation algorithm is proposed based on B\&B approach\cite{Konno2000,horst2013global,MARANASProofBB}.

\subsection{Continuous Relaxation and Penalty Method}
To start with, we relax the binary constraint on $s_{i,m}$ in C1 which yields:
\begin{align}\label{C5:P3Continuous} &\underset{\mathbf{\overline{s}},\;\boldsymbol{{\gamma}}}{\mino}\,\, \,\,  p^{\mathrm{total}}_{\boldsymbol{{\gamma}}} \left(\boldsymbol{\gamma} \right) \\
\mbox{s.t.}\;\;
%%%%%
& \mbox{C4},\;\widetilde{\text{C7}},\; \overline{\text{C1}} \mbox{: } 0 \le \overline{s}_{i,m} \le 1,\;\forall i,m,
\notag\\
&\overline{\text{C6}} \mbox{: } \sum\limits_{m = 1}^M {{\overline{s}_{i,m}}} \le 2,\;\forall i,\;\;\overline{\text{C8}} \mbox{: } {{\gamma}}_{i,m} = {\overline{s}_{i,m}} \notag {\gamma}_{i,m},\;\forall i,m,
\end{align}
where ${\overline{s}_{i,m}}$ denotes the continuous relaxation of the binary variable ${s}_{i,m}$ and $\mathbf{\overline{s}} \in \mathbb{R}^{N_{\mathrm{F}} M \times 1} $ denotes the set of ${\overline{s}_{i,m}}$. $\overline{\text{C1}}$, $\overline{\text{C6}}$, and $\overline{\text{C8}}$ denote the modified constraints for C1, C6, and C8 via replacing ${s}_{i,m}$ with ${\overline{s}_{i,m}}$, correspondingly.
In general, the solution of the constraint relaxed problem in \eqref{C5:P3Continuous} provides a lower bound for the problem in \eqref{C5:P3}.
However, by utilizing the coupling relationship in constraint $\overline{\text{C8}}$, the following theorem states the equivalence between \eqref{C5:P3} and \eqref{C5:P3Continuous}.

\begin{Thm}\label{C5:Theorem2}
The relaxed problem in \eqref{C5:P3Continuous} is equivalent to the problem in \eqref{C5:P3}.
More importantly, for the optimal solution of \eqref{C5:P3Continuous}, $\left({{\overline{s}_{i,m}^*}},{\gamma}_{i,m}^*\right)$, $i \in \left\{1, \ldots ,{N_{\mathrm{F}}}\right\}$, $m \in \left\{ 1, \ldots ,M \right\}$, the optimal solution of \eqref{C5:P3} can be recovered via keeping ${\gamma}_{i,m}^*$ and performing the following mapping:
\begin{equation}\label{C5:OptimalConvert2}
{{{s}_{i,m}^*}} =
\left\{
\begin{array}{ll}
1 & \text{if}\;{\gamma}_{i,m}^* > 0,\\
0 & \text{if}\;{\gamma}_{i,m}^* = 0.
\end{array}
\right.
\end{equation}
\end{Thm}

\begin{proof}
Please refer to Appendix \ref{C5:AppendixB} for a proof of Theorem \ref{C5:Theorem2}.
\end{proof}

Note that the equivalence between $\widetilde{\text{C7}}$ in \eqref{C5:P3Continuous} and C7 in \eqref{C5:P2} still holds at the optimal solution via the mapping relationship in \eqref{C5:OptimalConvert2}.
It is notable that the reformulation in \eqref{C5:P3} not only transforms the couplings between binary variables and continuous variables into a single constraint C8, but also reveals the special structure that enables the equivalence between \eqref{C5:P3} and \eqref{C5:P3Continuous}.
Now, the non-convexity remaining in \eqref{C5:P3Continuous} arises from the product term in both objective function and constraint $\overline{\text{C8}}$.
Thus, we augment $\overline{\text{C8}}$ into the objective function via introducing a penalty factor $\theta$ as follows:
\begin{equation} \label{C5:P3Penalty}
\underset{\mathbf{\overline{s}},\;\boldsymbol{{\gamma}}}{\mino}\,\, \,\, G^{\theta}\left(\mathbf{\overline{s}},{\boldsymbol{{\gamma} }} \right) \;\;\;
\mbox{s.t.}\;\;
%%%%%
\overline{\text{C1}},\; \text{C4},\; \overline{\text{C6}},\;\widetilde{\text{C7}},
\end{equation}
where the new objective function is given by
\begin{equation}\label{C5:NewObj}
G^{\theta}\left(\mathbf{\overline{s}},{\boldsymbol{{\gamma} }} \right) =
\sum\limits_{i = 1}^{{N_{\mathrm{F}}}} {\sum\limits_{m = 1}^M {\frac{{{{\gamma}}_{i,m}}}{{{\beta _{i,m}}}}} }  + \sum\limits_{i = 1}^{{N_{\mathrm{F}}}} {\sum\limits_{m = 1}^{M-1} {\sum\limits_{n = m + 1}^M {\frac{{{{\gamma}}_{i,m}{{\gamma}}_{i,n}}}{{\max \left( {{\beta _{i,m}},{\beta _{i,n}}} \right)}}} } } +\theta \sum\limits_{i = 1}^{{N_{\mathrm{F}}}} {\sum\limits_{m = 1}^M {\left({{\gamma}}_{i,m} - {\overline{s}_{i,m}} {\gamma}_{i,m}\right)} }.
\end{equation}

\begin{Thm}\label{C5:Theorem3}
If the problem in \eqref{C5:P3Continuous} is feasible with a bounded optimal value, the problem in \eqref{C5:P3Penalty} is equivalent to \eqref{C5:P3Continuous} for a sufficient large penalty factor $\theta \gg 1$.
\end{Thm}

\begin{proof}
Please refer to \cite{DerrickFD2016,sun2016optimal} for a proof of Theorem \ref{C5:Theorem3}.
\end{proof}

The problem in \eqref{C5:P3Penalty} is a generalized linear multiplicative programming problem over a compact convex set, where the optimal solution can be obtained via the B\&B method\cite{Konno2000}.

\subsection{B\&B Based Optimal Resource Allocation Algorithm}
The B\&B method has been widely adopted as a partial enumeration strategy for global optimization\cite{horst2013global}.
The basic principle of B\&B relies on a successive subdivision of the original region (Branch) that systematically discards non-promising subregions via employing lower bound or upper bound (Bound).
It has been proved that B\&B can converge to a globally optimal solution in finite numbers of iterations if the branching operation is consistent and the selection operation is bound improving\footnote{We note that the B\&B method cannot be directly used on the problem in \eqref{C5:P2} since a tight convex bounding function for the objective function has not been reported in the literatures and its feasible solution set is not compact. Based on Theorem \ref{C5:Theorem2} and Theorem \ref{C5:Theorem3}, we transform the problem in \eqref{C5:P2} to a equivalent generalized linear multiplicative programming problem on a convex compact feasible solution set in \eqref{C5:P3Penalty}, which can be handled by the B\&B method\cite{Konno2000}.}\cite{horst2013global,MARANASProofBB}.
In this section, we first propose the branching rule and the bounding method for the problem in \eqref{C5:P3Penalty}, and then develop the optimal resource allocation algorithm.

\noindent\textbf{\underline{{Branching Procedure}}}

From constraint $\widetilde{\text{C7}}$ in \eqref{C5:P3Penalty}, it can be observed that ${{\gamma}}_{i,m} > 2^{{R}_m^{\mathrm{total}}}-1$ is not the optimal rate allocation since the objective function is monotonically increasing with ${{\gamma}}_{i,m}$.
Therefore, we rewrite constraint C4 in \eqref{C5:P3Penalty} with a box constraint, C4: $0 \le {\gamma}_{i,m} \le 2^{{R}_m^{\mathrm{total}}}-1$, $\forall i,m$.
As a result, the optimization variables ${\mathbf{\overline{s}}}$ and $\boldsymbol{{\gamma}}$ are defined in a hyper-rectangle, which is spanned by $\overline{\text{C1}}$ and C4.
For notational simplicity, we redefine the optimization variables in \eqref{C5:P3Penalty} as follows:
\begin{equation}
{v_{i,m}} = {{\gamma} _{i,m}}\; \text{and} \; {v_{i,m+M}} = {\overline{s}_{i,m}}, \;\forall{i} \in \left\{1, \ldots ,{N_{\mathrm{F}}}\right\},\;\forall{m} \in \left\{ 1, \ldots ,M \right\}. \label{C5:VariableRedefine1}
\end{equation}
Then the product terms in $G^{\theta}\left(\mathbf{\overline{s}},{\boldsymbol{{\gamma} }} \right)$ in \eqref{C5:P3Penalty}, i.e., ${{{\gamma}}_{i,m}{{\gamma}}_{i,n}}$ and $-{\overline{s}_{i,m}} {\gamma}_{i,m}$, can be generally represented by $a_{i,m,n}{{v_{i,m}v_{i,n}}}$, where $a_{i,m,n} \in \{1,-1\}$ is a constant coefficient.
With the definition in \eqref{C5:VariableRedefine1}, we can use the new variable ${v_{i,m}}$ and the original variable ${\left({{\gamma}}_{i,m},\;{\overline{s}_{i,m}}\right)}$ interchangeably in the rest of the chapter.
Furthermore, the hyper-rectangle spanned by constraints $\overline{\text{C1}}$ and C4 can be presented by $\Phi = \left[ v _{i,m}^{{\mathrm{L}}},v _{i,m}^{{\mathrm{U}}} \right]$, $\forall{i} \in \left\{1, \ldots ,{N_{\mathrm{F}}}\right\}$, $\forall{m} \in \left\{ 1, \ldots ,2M \right\}$, where ${v _{i,m}^{\mathrm{L}}}$ and ${v_{i,m}^{\mathrm{U}}}$ denote the lower bound and upper bound for ${{v _{i,m}}}$, respectively.

\begin{figure}[t]
\centering
\includegraphics[width=5.0in]{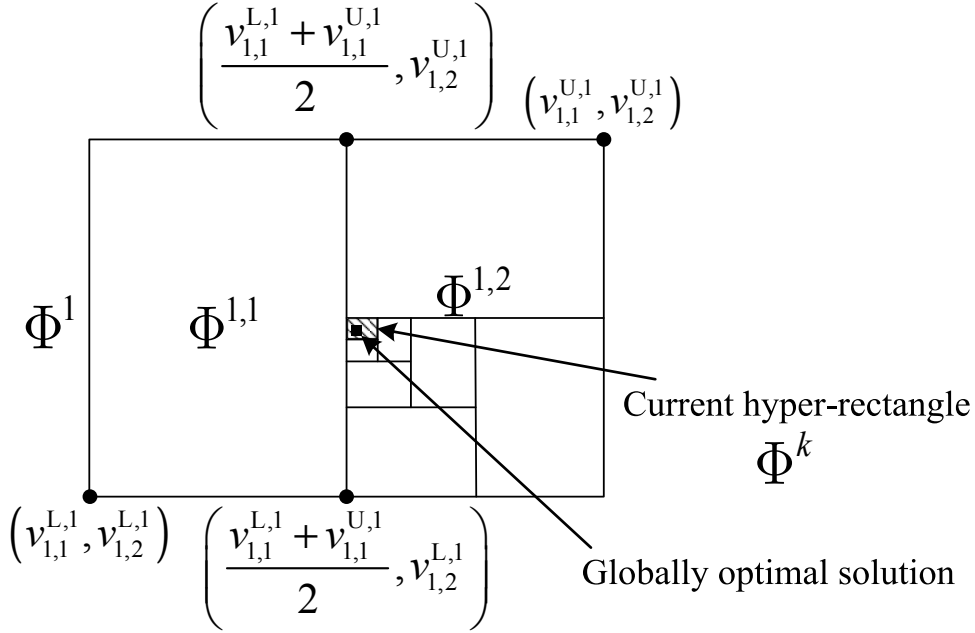}
\caption{An illustration of the successive branching procedure in a two-dimensional space.}
\label{C5:BnB:a}
\end{figure}

An successive branching procedure with bisection on the longest edge of the hyper-rectangle is adopted in this chapter\cite{Konno2000}, as illustrated in Figure \ref{C5:BnB:a}.
Particularly, in the first iteration, according to constraints $\overline{\text{C1}}$ and C4 in \eqref{C5:P3Penalty}, the initial hyper-rectangle $\Phi^1$ is characterized by
\begin{equation}\label{C5:InitialRectangle}
v _{i,m}^{{\mathrm{L}},1} = 0,\;
v _{i,m}^{{\mathrm{U}},1} = 2^{{R}_m^{\mathrm{total}}}-1,\;
v _{i,m+M}^{{\mathrm{L}},1} = 0,\;
\text{and}\;
v _{i,m+M}^{{\mathrm{U}},1} = 1.
\end{equation}
Then, in the $k$-th iteration, the current hyper-rectangle\footnote{The current hyper-rectangle selection rule will be presented in the overall algorithm, cf. \textbf{Algorithm} \ref{C5:alg1}.} $\Phi^k = \left[ v _{i,m}^{{\mathrm{L}},k},v _{i,m}^{{\mathrm{U}},k} \right]$, $\forall{i} \in \left\{1, \ldots ,{N_{\mathrm{F}}}\right\}$, $\forall{m} \in \left\{ 1, \ldots ,2M \right\}$, is partitioned into the following two subrectangles:
\begin{align}\label{C5:Subdivision}
\Phi^{k,1} &= \left[ {\begin{array}{*{20}{c}}
{v_{1,1}^{{\mathrm{L}},k}}&{v_{1,1}^{{\mathrm{U}},k}}\\
 \vdots & \vdots \\
{v_{{i^k},{m^k}}^{{\mathrm{L}},k}}&{\frac{{v_{{i^k},{m^k}}^{{\mathrm{L}},k} + v_{{i^k},{m^k}}^{{\mathrm{U}},k}}}{2}}\\
 \vdots & \vdots \\
{v_{{N_{\mathrm{F}}},2M}^{{\mathrm{L}},k}}&{v_{{N_{\mathrm{F}}},2M}^{{\mathrm{U}},k}}
\end{array}} \right]\; \text{and} \;
\Phi^{k,2} &= \left[ {\begin{array}{*{20}{c}}
{v_{1,1}^{{\mathrm{L}},k}}&{v_{1,1}^{{\mathrm{U}},k}}\\
 \vdots & \vdots \\
{\frac{{v_{{i^k},{m^k}}^{{\mathrm{L}},k} + v_{{i^k},{m^k}}^{{\mathrm{U}},k}}}{2}}&{v_{{i^k},{m^k}}^{{\mathrm{U}},k}}\\
 \vdots & \vdots \\
{v_{{N_{\mathrm{F}}},2M}^{{\mathrm{L}},k}}&{v_{{N_{\mathrm{F}}},2M}^{{\mathrm{U}},k}}
\end{array}} \right],
\end{align}
where ${v _{i,m}^{{\mathrm{L}},k}}$ and ${v_{i,m}^{{\mathrm{U}},k}}$ denote the lower bound and upper bound for ${{v _{i,m}}}$ in the $k$-th iteration. Index $\left({{i^k},{m^k}}\right)$ in \eqref{C5:Subdivision} corresponds to the variable with the longest normalized edge in $\Phi^k$, i.e., $\left({{i^k},{m^k}}\right) = \arg \underset{\left(i,m\right)}{\maxo}\; \frac{{ {v_{i,m}^{{\mathrm{U}},k} - v_{i,m}^{{\mathrm{L}},k}}}}{{ {v_{i,m}^{{\mathrm{U}},1} - v_{i,m}^{{\mathrm{L}},1}} }}$.
Figure \ref{C5:BnB:a} illustrates the successive branching procedure in a two-dimensional space with ${N_{\mathrm{F}}}=1$ and $M=1$.
Firstly, we partitioned the initial hyper-rectangle $\Phi^1$ into $\Phi^{1,1}$ and $\Phi^{1,2}$ via perform a bisection on the edge of $v_{1,1}$.
Then, $\Phi^{1,2}$ is selected as current hyper-rectangle in the 2-th iteration, namely $\Phi^2$, for subsequent branching iterations.
The shadowed region denotes the current hyper-rectangle in the $k$-th iteration, i.e., $\Phi^k$.
Note that the branching procedure is exhaustive due to the finite numbers of optimization variables and the finite volume of initial hyper-rectangle $\Phi^1$, i.e., $\underset{k \to \infty }{\lim}\underset{\left(i,m\right)}{\maxo}\; \left( v _{i,m}^{{\mathrm{U}},k}-v _{i,m}^{{\mathrm{L}},k}\right) = 0.$ Correspondingly, in Figure \ref{C5:BnB:a}, the shadowed region will collapse into a point with $k \to \infty$.

Now, in the $k$-th iteration, we can rewrite the problem in \eqref{C5:P3Penalty} within $\Phi^k$ as follows:
\begin{equation}\label{C5:P3PenaltySubrect} \underset{\left(\mathbf{\overline{s}},\;{\boldsymbol{{\gamma} }} \right) \in \Phi^k}{\mino}\,\, \,\, G^{\theta}\left(\mathbf{\overline{s}},{\boldsymbol{{\gamma} }} \right)\;\;\;
\mbox{s.t.} \;\; \overline{\text{C6}},\;\widetilde{\text{C7}}.
\end{equation}

\noindent\textbf{\underline{{Lower Bound and Upper Bound}}}

\begin{figure}[t]
\centering
\includegraphics[width=5in]{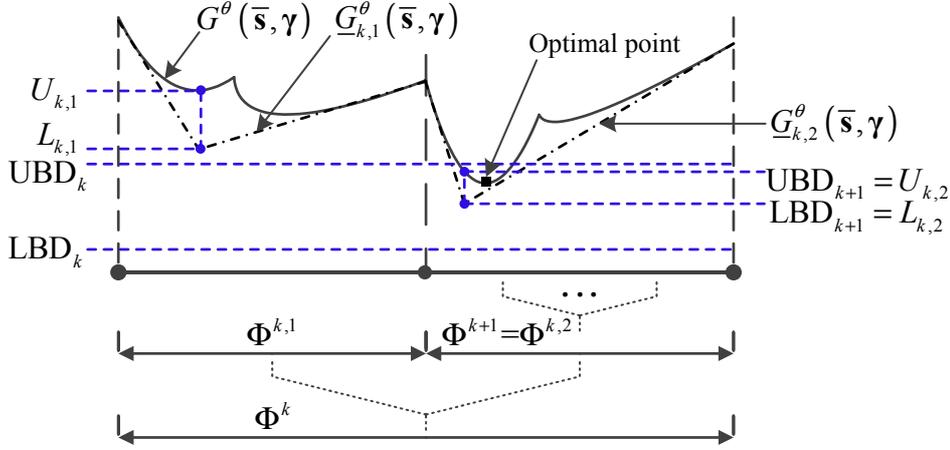}
\caption{An illustration of \textbf{Algorithm} \ref{C5:alg1} from the $k$-th iteration to the $(k+1)$-th iteration in a one-dimensional space.}
\label{C5:BnB:b}
\end{figure}

We first present the lower bound for the problem in \eqref{C5:P3PenaltySubrect}, from which the upper bound can be obtained straightforwardly in the end of this part. As it was shown in \cite{Androulakis1995}, the tightest possible convex lower bound of a product term $a_{i,m,n}{{v_{i,m}v_{i,n}}}$ inside some rectangular region $D_{i,m,n}^k = \left[ v _{i,m}^{{\mathrm{L}},k},v _{i,m}^{{\mathrm{U}},k} \right] \times \left[ v _{i,n}^{{\mathrm{L}},k},v _{i,n}^{{\mathrm{U}},k} \right]$, i.e., convex envelope in the $k$-th iteration, is given by
\begin{equation}\label{C5:ConvexEnvelope}
{l_{a_{i,m,n}}^k}\left( {{v _{i,m}},{v _{i,n}}} \right) =
\left\{
\begin{array}{ll}
a_{i,m,n} \max\left( {v_{i,m,n}^{\mathrm{LL,k}},\;v_{i,m,n}^{\mathrm{UU,k}}
} \right) & \text{if}\;{a_{i,m,n}} > 0,\\
a_{i,m,n} \min\left( {v_{i,m,n}^{\mathrm{LU,k}},\;v_{i,m,n}^{\mathrm{UL,k}}
} \right) & \text{if}\;{a_{i,m,n}} \le 0,
\end{array}
\right.
\end{equation}
where
\begin{align}\label{C5:Parameters}
  v_{i,m,n}^{\mathrm{LL,k}} & = v _{i,m}^{{\mathrm{L}},k}{v _{i,n}} + v_{i,n}^{{\mathrm{L}},k}{v _{i,m}} - v _{i,m}^{{\mathrm{L}},k}v_{i,n}^{{\mathrm{L}},k} \\
  v_{i,m,n}^{\mathrm{UU,k}} & = v_{i,m}^{{\mathrm{U}},k}{v_{i,n}} + v_{i,n}^{{\mathrm{U}},k}{v_{i,m}} - v_{i,m}^{{\mathrm{U}},k} v_{i,n}^{{\mathrm{U}},k} \\
  v_{i,m,n}^{\mathrm{LU,k}} & = v_{i,m}^{{\mathrm{L}},k}{v _{i,n}} + v_{i,n}^{{\mathrm{U}},k}{v _{i,m}} - v_{i,m}^{{\mathrm{L}},k}v_{i,n}^{{\mathrm{U}},k} \\
  v_{i,m,n}^{\mathrm{UL,k}} & = v_{i,m}^{{\mathrm{U}},k}{v_{i,n}} + v_{i,n}^{{\mathrm{L}},k}{v_{i,m}} - v_{i,m}^{{\mathrm{U}},k} v_{i,n}^{{\mathrm{L}},k}
\end{align}
Note that ${l_{a_{i,m,n}}^k}\left( {{v _{i,m}},{v _{i,n}}} \right)$ is a pointwise linear function, which serves as a convex lower bound for $a_{i,m,n}{{v_{i,m}v_{i,n}}}$ in $D_{i,m,n}^k$, i.e., ${l_{a_{i,m,n}}^k}\left( {{v _{i,m}},{v _{i,n}}} \right) \le a_{i,m,n}{{v_{i,m}v_{i,n}}}$.
Further, the maximum separation between $a_{i,m,n}{{v_{i,m}v_{i,n}}}$ and ${l_{a_{i,m,n}}^k}\left( {{v _{i,m}},{v _{i,n}}} \right)$ is equal to one-fourth of the area of rectangular region $D_{i,m,n}^k$ \cite{Androulakis1995}, which is
\begin{equation}\label{C5:MaxSeparation1}
{{\varepsilon}^k} \left( v_{i,m}, v _{i,n} \right) = \frac{1}{4}{\left( {v _{i,m}^{{\mathrm{U}},k} - v _{i,m}^{{\mathrm{L}},k}} \right)\left( {v _{i,n}^{{\mathrm{U}},k} - v _{i,n}^{{\mathrm{L}},k}} \right)}.
\end{equation}
We note that the maximum separation is bounded and ${l_{a_{i,m,n}}^k}\left( {{v _{i,m}},{v _{i,n}}} \right)$ can be arbitrarily close to $a_{i,m,n}{{v_{i,m}v_{i,n}}}$ for a small enough rectangle $D_{i,m,n}^k$.

Based on the convex envelope for product terms in \eqref{C5:P3PenaltySubrect}, in the $k$-th iteration, we have the relaxed convex minimization problem over current hyper-rectangle $\Phi^k$ as follows:
\begin{equation}\label{C5:P3PenaltySubrectConvexLowerBound} \underset{\left(\mathbf{\overline{s}},\;\boldsymbol{{\gamma}}\right) \in \Phi^k }{\mino}\;\underline{G}^{\theta}_k\left(\mathbf{\overline{s}},{\boldsymbol{{\gamma} }} \right)\;\;\;\mbox{s.t.} \;\;\overline{\text{C6}},\;\widetilde{\text{C7}},
\end{equation}
where $\underline{G}^{\theta}_k\left(\mathbf{\overline{s}},{\boldsymbol{{\gamma} }} \right)$ denotes the convex lower bounding function for $G^{\theta}\left(\mathbf{\overline{s}},{\boldsymbol{{\gamma} }} \right)$ in \eqref{C5:P3PenaltySubrect} within $\Phi^k$,
\begin{align}
\underline{G}^{\theta}_k\left(\mathbf{\overline{s}},{\boldsymbol{{\gamma} }} \right) &= \sum\limits_{i = 1}^{{N_{\mathrm{F}}}} {\sum\limits_{m = 1}^M {{{\gamma} _{i,m}}\left( {\frac{1}{{{\beta _{i,m}}}} + \theta } \right)} } + \sum\limits_{i = 1}^{{N_{\mathrm{F}}}} {\sum\limits_{m = 1}^{M - 1} {\sum\limits_{n = m + 1}^M {\frac{{{l_1^k}\left( {{{\gamma} _{i,m}},{{\gamma} _{i,n}}} \right)}}{{\max \left( {{\beta _{i,m}},{\beta _{i,n}}} \right)}}} } }  \\
&+ \theta \sum\limits_{i = 1}^{{N_{\mathrm{F}}}} {\sum\limits_{m = 1}^M {{l_{ - 1}^k}\left( {{{\overline{s}}_{i,m}}{\mathrm{,}}{{\gamma} _{i,m}}} \right)} }.\notag
\end{align}
According to \eqref{C5:MaxSeparation1}, the maximum gap between $G^{\theta}\left(\mathbf{\overline{s}},{\boldsymbol{{\gamma} }} \right)$ and $\underline{G}^{\theta}_k\left(\mathbf{\overline{s}},{\boldsymbol{{\gamma} }} \right)$ within $\Phi^k$ is given by
\begin{equation}\label{C5:MaxGap}
{\Delta _{\max }^k} = \sum\limits_{i = 1}^{{N_{\mathrm{F}}}} {\sum\limits_{m = 1}^{M - 1} {\sum\limits_{n = m + 1}^M {\frac{{{\varepsilon^k}\left( {{\gamma} _{i,m},{\gamma} _{i,n}} \right)}}{{\max \left( {{\beta _{i,m}},{\beta _{i,n}}} \right)}}} } }  +\theta \sum\limits_{i = 1}^{{N_{\mathrm{F}}}} {\sum\limits_{m = 1}^M {{\varepsilon^k}\left( { \overline{s}_{i,m},{\gamma} _{i,m}} \right)} },
\end{equation}
which will vanish when $\Phi^k$ collapses into a point with $k \to \infty$.
Now, the relaxed problem in \eqref{C5:P3PenaltySubrectConvexLowerBound} is a convex programming problem that can be solved efficiently by standard convex program solvers such as CVX \cite{cvx}.
Note that the optimal value of \eqref{C5:P3PenaltySubrectConvexLowerBound} provides a lower bound for \eqref{C5:P3PenaltySubrect} within $\Phi^k$ locally and that \eqref{C5:P3PenaltySubrect} is infeasible if \eqref{C5:P3PenaltySubrectConvexLowerBound} is infeasible.

For the upper bound, it is clear that any feasible solution of the problem in \eqref{C5:P3PenaltySubrect} attains a local upper bound within the hyper-rectangle $\Phi^k$. It can be observed that the optimal solution of the problem in \eqref{C5:P3PenaltySubrectConvexLowerBound}, denoted as $\left(\mathbf{\overline{s}}^*_k, {\boldsymbol{{\gamma} }^*_k}\right)$, is always a feasible point for \eqref{C5:P3PenaltySubrect}, since they share the same feasible solution set. Therefore, an upper bound of \eqref{C5:P3PenaltySubrect} within $\Phi^k$ can be obtained by simply calculating $G^{\theta}\left(\mathbf{\overline{s}}^*_k,{\boldsymbol{{\gamma} }^*_k} \right)$.

\noindent\textbf{\underline{{Overall Algorithm}}}

\begin{table}[ptb]
	\begin{algorithm} [H]                    % enter the algorithm environment
		\caption{{Optimal Resource Allocation Algorithm via B\&B}}     % give the algorithm a caption
		\label{C5:alg1}                          % and a label for \ref{C5:} commands later in the document
		\begin{algorithmic} [1]
			\small          % enter the algorithmic environment
			\STATE {\textbf{Branching on Current Rectangle:}
			Partition current rectangle $\Phi^k$ into two subrectangles $\Phi^{k,1}$ and $\Phi^{k,2}$ through \eqref{C5:Subdivision} and update the unfathomed partition set with $\mathcal{Z} = \mathcal{Z} \bigcup \left\{\Phi^{k,1},\Phi^{k,2}\right\} \setminus \Phi^k$.}
			
			\STATE {\textbf{Local Lower Bound and Upper Bound:}
			Solve the problem in \eqref{C5:P3PenaltySubrectConvexLowerBound} within $\Phi^{k,r}$ ($r = 1,2$). If it is infeasible, delete (fathoming) $\Phi^{k,r}$ from $\mathcal{Z}$. Otherwise, obtain the intermediate optimal solution, $\left(\mathbf{\overline{s}}_{k,r},{\boldsymbol{{\gamma} }_{k,r}} \right)$, and the local lower bound and upper bound within $\Phi^{k,r}$ given by
			$L_{k,r} = \underline{G}^{\theta}_{k,r}\left(\mathbf{\overline{s}}_{k,r},{\boldsymbol{{\gamma} }_{k,r}} \right)$ and
			$U_{k,r} = G^{\theta}\left(\mathbf{\overline{s}}_{k,r},{\boldsymbol{{\gamma} }_{k,r}} \right)$, respectively.}
			
			\STATE {\textbf{Update $\mathcal{Z}$, $\mathcal{W}$, $\mathcal{V}$, $\mathrm{LBD}_k$, and $\mathrm{UBD}_k$:}
			For $L_{k,r} \ge \mathrm{UBD}_k$, delete (fathoming) $\Phi^{k,r}$ from $\mathcal{Z}$.
			If $\mathcal{Z} = \emptyset$, then stop and return the incumbent point. Update parameters based on \eqref{C5:UpdateParameters1}, \eqref{C5:UpdateParameters2}, and \eqref{C5:UpdateParameters3}.}
			\STATE {\textbf{Update Current Point, Incumbent Point, and Current Rectangle:}
			$k = k+1$. Update the current point with     $\left(\mathbf{\overline{s}}_{k}^{\mathrm{cur}},{\boldsymbol{{\gamma} }_{k}^{\mathrm{cur}}} \right) = \left(\mathbf{\overline{s}}_{k^*,r^*},{\boldsymbol{{\gamma} }_{k^*,r^*}} \right)$ and the incumbent point with
			$\left(\mathbf{\overline{s}}_{k}^{\mathrm{inc}},{\boldsymbol{{\gamma} }_{k}^{\mathrm{inc}}} \right) = \left(\mathbf{\overline{s}}_{k^\circ,r^\circ},{\boldsymbol{{\gamma} }_{k^\circ,r^\circ}} \right)$,
			where $\left(k^*,r^*\right)$ and $\left(k^\circ,r^\circ\right)$ are obtained from \eqref{C5:UpdateParameters2} and \eqref{C5:UpdateParameters3}, respectively.
			Correspondingly, current rectangle is selected as $\Phi^k = \Phi^{k^*,r^*}$.}
			
			\STATE {\textbf{Convergence Check:}
			If $\left(\mathrm{UBD}_k-\mathrm{LBD}_k\right)>\epsilon$, then go to \textbf{Step} 1.
			Otherwise, ${\epsilon}\text{-}$convergence solution has been attained and return the incumbent point.}
		\end{algorithmic}
	\end{algorithm}
\end{table}

Based on the proposed branching procedure and bounding
methods, we develop the B\&B resource allocation algorithm to obtain the globally optimal solution for the problem in \eqref{C5:P3Penalty}, cf. \textbf{Algorithm} \ref{C5:alg1}.
In our algorithm, the current subrectangle is the one which possesses the minimum local lower bound among all the unfathomed subrectangles.
The current point denotes the intermediate optimal solution within the current subrectangle.
In addition, the current point denotes the intermediate optimal solution within the current subrectangle.

{The proposed algorithm is initialized as follows.
Set the convergence tolerance $\epsilon$, the iteration counter $k=1$, and initialize current subrectangle $\Phi^k$ through \eqref{C5:InitialRectangle}.
Solve the problem in \eqref{C5:P3PenaltySubrectConvexLowerBound} within $\Phi^k$ to obtain the intermediate optimal solution $\left(\mathbf{\overline{s}}_k,{\boldsymbol{{\gamma} }_k}\right)$.
Then, define the current point $\left(\mathbf{\overline{s}}_k^{\mathrm{cur}},{\boldsymbol{{\gamma} }_k^{\mathrm{cur}}}\right) = \left(\mathbf{\overline{s}}_k,{\boldsymbol{{\gamma} }_k}\right)$, the incumbent point $\left(\mathbf{\overline{s}}_k^{\mathrm{inc}},{\boldsymbol{{\gamma} }_k^{\mathrm{inc}}}\right) = \left(\mathbf{\overline{s}}_k,{\boldsymbol{{\gamma} }_k}\right)$, the local lower bound $L_k = \underline{G}^{\theta}_k\left(\mathbf{\overline{s}}_k,{\boldsymbol{{\gamma} }_k} \right)$, and the local upper bound $U_k = G^{\theta}\left(\mathbf{\overline{s}}_k,{\boldsymbol{{\gamma} }_k} \right)$.
Initialize the global lower bound and upper bound with $\mathrm{LBD}_{k} = L_k$ and $\mathrm{UBD}_{k} = U_k$, respectively.
Initialize the unfathomed partition set with $\mathcal{Z} = \left\{ \Phi^k  \right\}$.
Correspondingly, define the local lower bound set and the local upper bound set for all unfathomed rectangles in $\mathcal{Z}$ with $\mathcal{W}$ and $\mathcal{V}$, respectively.
Initialize them with $\mathcal{W} = \left\{L_k\right\}$ and $\mathcal{V} = \left\{U_k\right\}$, respectively.
Additionally, in Line 5 of Algorithm \ref{C5:alg1}, the parameters are updated as follows:
\begin{align}
\mathcal{W} &= \mathcal{W} \bigcup L_{k,r},\;
\mathcal{V} = \mathcal{V} \bigcup U_{k,r},\; \text{if} \; \Phi^{k,r} \in \mathcal{Z}, \label{C5:UpdateParameters1}\\
\mathrm{LBD}_{k+1} &= L_{k^*,r^*} = \underset{k',r'}{\mino}\;\left(\mathcal{W}\right),\; r'\in \left\{1,2\right\},\; k'\in \left\{1, \ldots , k\right\},\;\text{and} \label{C5:UpdateParameters2}\\
\mathrm{UBD}_{k+1} &= U_{k^\circ,r^\circ} = \underset{k',r'}{\mino}\;\left(\mathcal{V}\right),\;r'\in \left\{1,2\right\},\; k'\in \left\{1, \ldots , k\right\}.\label{C5:UpdateParameters3}
\end{align}}

Accordingly, Figure \ref{C5:BnB:b} illustrates a simple example of the developed algorithm in a one-dimensional space, where $\underline{G}^{\theta}_{k,r}\left(\mathbf{\overline{s}},{\boldsymbol{{\gamma} }} \right)$, $r = 1,2$, denotes the lower bounding function for $G^{\theta}\left(\mathbf{\overline{s}},{\boldsymbol{{\gamma} }} \right)$ in \eqref{C5:P3PenaltySubrect} within $\Phi^{k,r}$.
In Step 4, if $L_{k,r} \ge \mathrm{UBD}_k$, the optimal solution must not locate in $\Phi^{k,r}$, and thus we discard it from $\mathcal{Z}$, such as $\Phi^{k,1}$ in Figure \ref{C5:BnB:b}.
Note that the lower bound within subrectangle $\Phi^{k,r}$ is always larger than that within $\Phi^{k}$ since the feasible solution set becomes smaller, i.e., $L_{k,r} \ge L_{k}$.
Therefore, the global lower bound update and current rectangle selection operation in Steps 4 and 5 can generate a \emph{non-decreasing} sequence for $\mathrm{LBD}_k$.
On the other hand, the global upper bound update operation can generate a \emph{non-increasing} sequence for $\mathrm{UBD}_k$.
For example, in Figure \ref{C5:BnB:b}, we can easily observe that $\mathrm{UBD}_{k+1}\le\mathrm{UBD}_k$ and $\mathrm{LBD}_{k+1}\ge\mathrm{LBD}_k$.
It can be proved that the proposed branch and bound algorithm converges to the globally optimal solution in finite number of iterations based on the sufficient conditions stated in \cite{horst2013global}.
The proof of convergence for the adopted B\&B algorithm can be found in \cite{MARANASProofBB}.
The convergence speed of our proposed algorithm will be verified by simulations in Section \ref{C5:SimulationResults}.

\section{Suboptimal Solution} \label{C5:Suboptimal}
Compared to the brute-force search method, the proposed B\&B algorithm provides a systematic approach by exploiting the structure of the problem in \eqref{C5:P3Penalty}.
It saves a large amount of computational complexity since it discards the non-promising subrectangles\cite{horst2013global}.
More importantly, it serves as a performance benchmark for any suboptimal algorithm.
However, it has a non-polynomial time computational complexity\cite{Konno2000}.
In this section, we present a suboptimal solution for the problem in \eqref{C5:P2} by exploiting the D.C. programming\cite{dinh2010local}, which only requires a polynomial time computational complexity.

To start with, we aim to circumvent the coupling between binary variables ${s}_{i,m}$ and continuous variables ${\gamma} _{i,m}$ in \eqref{C5:P2}. We define an auxiliary variable $\widetilde{\gamma}_{i,m} = {\gamma} _{i,m}{s}_{i,m}$ and adopt the big-M formulation\cite{DerrickFD2016} to equivalently transform the problem in \eqref{C5:P2} as follows:
\begin{align}\label{C5:P4}
&\underset{\mathbf{s},\boldsymbol{\gamma},\widetilde{\boldsymbol{\gamma}}}{\mino}\,\, \,\,  p^{\mathrm{total}} \left(\mathbf{s},\widetilde{\boldsymbol{\gamma}} \right) \\
\mbox{s.t.}\;\;
%%%%%
&\mbox{C1},\;\mbox{C4},\;\mbox{C6},\notag\\
&\mbox{C7: } \sum\limits_{i = 1}^{{N_{\mathrm{F}}}} {s_{i,m}}{{{\log }_2}\left( {1 +\frac{\widetilde{\gamma} _{i,m}}{{s_{i,m}}}} \right)}  \ge {{R}_m^{\mathrm{total}}},\;\forall m,\notag\\
&\mbox{C9: } {{\widetilde {\gamma} }_{i,m}} \ge 0,\;\forall i,m,
\;\;\;\mbox{C10: } {{\widetilde {\gamma}}_{i,m}} \le {\gamma _{i,m}},\;\forall i,m,\notag\\
&\mbox{C11: } {{\widetilde {\gamma} }_{i,m}} \le {s_{i,m}} \left(2^{R_m^{{\mathrm{total}}}}-1\right),
\;\forall i,m,\notag\\
&\mbox{C12: } {{\widetilde {\gamma} }_{i,m}} \ge {\gamma _{i,m}} - \left( {1 - {s_{i,m}}} \right) \left(2^{R_m^{{\mathrm{total}}}}-1\right),\;\forall i,m,\notag
\end{align}
where $\widetilde{\boldsymbol{\gamma}} \in \mathbb{R}^{N_{\mathrm{F}}M \times 1}$ denotes the set of the auxiliary variables $\widetilde{\gamma}_{i,m}$ and constraints C9-C12 are imposed additionally following the big-M formulation\cite{DerrickFD2016}.
Besides, the binary constraints in C1 are another major obstacle for the design of a computationally efficient resource allocation algorithm. Hence, we rewrite the binary constraint C1 in its equivalent form:
\begin{align}\label{C5:C1AB}
&\mbox{C1a: }0 \le {s}_{i,m} \le 1 \; \text{and} \;\notag\\
&\mbox{C1b: }\sum\limits_{i = 1}^{{N_\mathrm{F}}} {\sum\limits_{m = 1}^M {{{s}_{i,m}}} }  - \sum\limits_{i = 1}^{{N_\mathrm{F}}} {\sum\limits_{m = 1}^M {s_{i,m}^2} }  \le 0, \;\;\forall i,m.
\end{align}

Furthermore, we rewrite ${{\widetilde{\gamma} _{i,m}}{\widetilde{\gamma} _{i,n}}} = \frac{1}{2}{{\left( {{{\widetilde \gamma }_{i,m}} + {{\widetilde \gamma }_{i,n}}} \right)}^2} - \frac{1}{2}{\left( {\widetilde \gamma _{i,m}^2 + \widetilde \gamma _{i,n}^2} \right)}$ and augment the D.C. constraint C1b into the objective function via a penalty factor $\eta \gg 1$. The problem in \eqref{C5:P4} can be rewritten in the canonical form of D.C. programming as follows:
\begin{align}\label{C5:P4DC}
&\underset{\mathbf{s},\boldsymbol{\gamma},\widetilde{\boldsymbol{\gamma}}}{\mino}\,\, \,\, G_{1}^{\eta}\left(\mathbf{{s}},{\boldsymbol{\widetilde{\gamma} }} \right) - G_{2}^{\eta}\left(\mathbf{{s}},{\boldsymbol{\widetilde{\gamma} }} \right) \\
\mbox{s.t.}\;\;&\mbox{C1a},\;\mbox{C4},\;\mbox{C7},\;\mbox{C6},\;\mbox{C9-C12}, \notag
\end{align}
where
\begin{align}
G_{1}^{\eta}\left(\mathbf{{s}},{\boldsymbol{\widetilde{\gamma} }} \right) &= \sum\limits_{i = 1}^{{N_\mathrm{F}}} {\sum\limits_{m = 1}^M {\frac{{{{\widetilde \gamma }_{i,m}}}}{{{\beta _{i,m}}}}} }  + \eta \sum\limits_{i = 1}^{{N_\mathrm{F}}} {\sum\limits_{m = 1}^M {{s_{i,m}}} }+\frac{1}{2}\sum\limits_{i = 1}^{{N_\mathrm{F}}} {\sum\limits_{m = 1}^{M - 1} {\sum\limits_{n = m + 1}^M {\frac{{{{\left( {{{\widetilde \gamma }_{i,m}} + {{\widetilde \gamma }_{i,n}}} \right)}^2}}}{{\max \left( {{\beta _{i,m}},{\beta _{i,n}}} \right)}}} } },\label{C5:DC1}\\
G_{2}^{\eta}\left(\mathbf{{s}},{\boldsymbol{\widetilde{\gamma} }} \right) &= \frac{1}{2}\sum\limits_{i = 1}^{{N_\mathrm{F}}} {\sum\limits_{m = 1}^{M - 1} {\sum\limits_{n = m + 1}^M {\frac{{\left( {\widetilde \gamma _{i,m}^2 + \widetilde \gamma _{i,n}^2} \right)}}{{\max \left( {{\beta _{i,m}},{\beta _{i,n}}} \right)}}} } } + \eta \sum\limits_{i = 1}^{{N_\mathrm{F}}} {\sum\limits_{m = 1}^M {s_{i,m}^2} } \label{C5:DC2}.
\end{align}

According to Theorem \ref{C5:Theorem3}, the problem in \eqref{C5:P4DC} is equivalent to the problem in \eqref{C5:P4}. Note that ${G_{1}^{\eta}} \left(\mathbf{{s}},{\boldsymbol{\widetilde{\gamma} }} \right)$ and ${G_{2}^{\eta}} \left(\mathbf{{s}},{\boldsymbol{\widetilde{\gamma} }} \right)$ are differentiable convex functions w.r.t. $s_{i,m}$ and ${{\widetilde \gamma }_{i,m}}$. Therefore, for any feasible point $\left(\mathbf{{s}}_{k},{\boldsymbol{\widetilde{\gamma} }}_{k} \right)$, we can define the global underestimator for ${G_{2}^{\eta}} \left(\mathbf{{s}},{\boldsymbol{\widetilde{\gamma} }} \right)$ based on its first order Taylor's expansion at $\left(\mathbf{{s}}_{k},{\boldsymbol{\widetilde{\gamma} }}_{k} \right)$ as follows:
\begin{align}\label{C5:Taylor1}
G_2^\eta\left( {{\mathbf{{s}}},{\boldsymbol{\widetilde{\gamma} }}} \right) &\ge G_2^\eta\left( {{{\mathbf{{s}}}_k},{\boldsymbol{\widetilde{\gamma} }_k}} \right) + {\nabla _{{\mathbf{{s}}}}}G_2^\eta{\left( {{{{\mathbf{{s}}}}_k},{{\boldsymbol{\widetilde{\gamma} }}_k}} \right)^{\mathrm{T}}}\left( {{\mathbf{{s}}} - {{{\mathbf{{s}}}}_k}} \right) \notag\\
&+ {\nabla _{\boldsymbol{\widetilde{\gamma} }}}G_2^\eta{\left( {{{{\mathbf{{s}}}}_k},{{\boldsymbol{\widetilde{\gamma} }}_k}} \right)^{\mathrm{T}}}\left( {{\boldsymbol{\widetilde{\gamma} }} - {{\boldsymbol{\widetilde{\gamma} }}_k}} \right),
\end{align}
where ${\nabla _{{\mathbf{{s}}}}}G_2^\eta{\left( {{{{\mathbf{{s}}}}_k},{{\boldsymbol{\widetilde{\gamma} }}_k}} \right)}$ and ${\nabla _{\boldsymbol{\widetilde{\gamma} }}}G_2^\eta{\left( {{{{\mathbf{{s}}}}_k},{{\boldsymbol{\widetilde{\gamma} }}_k}} \right)}$ denote the gradient vectors of ${G_2^\eta} \left(\mathbf{{s}},{\boldsymbol{\widetilde{\gamma} }} \right)$ at $\left( {{{{\mathbf{{s}}}}_k},{{\boldsymbol{\widetilde{\gamma} }}_k}} \right)$ w.r.t. ${{\mathbf{{s}}}}$ and ${\boldsymbol{\widetilde{\gamma} }}$, respectively. Then, we obtain an upper bound for the problem in \eqref{C5:P4DC} by solving the following convex optimization problem:
\begin{eqnarray}\label{C5:P4DCUpperBound} &&\underset{\mathbf{{s}},\boldsymbol{{\gamma}},\boldsymbol{\widetilde{\gamma}}}{\mino}\,\, G_1^\eta\left(\mathbf{{s}},{\boldsymbol{\widetilde{\gamma} }} \right)- G_2^\eta\left( {{{\mathbf{{s}}}_k},{\boldsymbol{\widetilde{\gamma} }_k}} \right) - {\nabla _{{\mathbf{{s}}}}}G_2^\eta{\left( {{{{\mathbf{{s}}}}_k},{{\boldsymbol{\widetilde{\gamma} }}_k}} \right)^{\mathrm{T}}}\left( {{\mathbf{{s}}} - {{{\mathbf{{s}}}}_k}} \right) \notag\\
&&- {\nabla _{\boldsymbol{\widetilde{\gamma} }}}G_2^\eta{\left( {{{{\mathbf{{s}}}}_k},{{\boldsymbol{\widetilde{\gamma} }}_k}} \right)^{\mathrm{T}}}\left( {{\boldsymbol{\widetilde{\gamma} }} - {{\boldsymbol{\widetilde{\gamma} }}_k}} \right)\\
\mbox{s.t.}&&\mbox{C1a},\;\mbox{C4},\;\mbox{C7},\;\mbox{C6},\;\mbox{C9-C12},\notag
\end{eqnarray}
where ${\nabla _{{\mathbf{{s}}}}}G_2^\eta{\left( {\mathbf{{s}_k},{{\boldsymbol{\widetilde{\gamma} }}_k}} \right)^{\mathrm{T}}}\left( {{\mathbf{{s}}} - {{{\mathbf{{s}}}}_k}} \right) =  2\eta \sum\limits_{i = 1}^{{N_{\mathrm{F}}}} {\sum\limits_{m = 1}^M { {s}_{i,m}^k } \left( {{{{s}}_{i,m}} - {s}_{i,m}^k} \right)}$ and
${\nabla _{\boldsymbol{\widetilde{\gamma} }}}G_2^\eta{\left( {{{{\mathbf{{s}}}}_k},{{\boldsymbol{\widetilde{\gamma} }}_k}} \right)^{\mathrm{T}}} \left( {{\boldsymbol{\widetilde{\gamma} }} - {{\boldsymbol{\widetilde{\gamma} }}_k}} \right) = \sum\limits_{i = 1}^{{N_{\mathrm{F}}}} {\sum\limits_{m = 1}^M {\sum\limits_{n \ne m}^M {\frac{{\gamma _{i,m}^k\left( {{\gamma _{i,m}} - \gamma _{i,m}^k} \right)}}{{\max \left( {{\beta _{i,m}},{\beta _{i,n}}} \right)}}} } }.$

\begin{table}[ptb]
\begin{algorithm} [H]                    % enter the algorithm environment
\caption{Suboptimal Resource Allocation Algorithm}     % give the algorithm a caption
\label{C5:alg2}                             % and a label for \ref{C5:} commands later in the document
\begin{algorithmic} [1]
\small          % enter the algorithmic environment
\STATE \textbf{Initialization}\\
Initialize the convergence tolerance $\epsilon$, the maximum number of iterations $K_\mathrm{max}$, the iteration counter $k = 1$, and the initial feasible solution $\left(\mathbf{{s}}_k,{\boldsymbol{\widetilde{\gamma} }_k} \right)$.

\REPEAT
\STATE Solve \eqref{C5:P4DCUpperBound} for a given $\left(\mathbf{{s}}_k,{\boldsymbol{\widetilde{\gamma} }_k} \right)$ and obtain the intermediate resource allcation policy $\left(\mathbf{{s}}',{\boldsymbol{\widetilde{\gamma} }'} \right)$
\STATE Set $k=k+1$ and $\left(\mathbf{{s}}_k,{\boldsymbol{\widetilde{\gamma} }_k} \right) = \left(\mathbf{{s}}',{\boldsymbol{\widetilde{\gamma} }'} \right)$
\UNTIL
$k = K_\mathrm{max}$ or $\max \left\{ {\left\| {\left(\mathbf{{s}}_k,{\boldsymbol{\widetilde{\gamma} }_k} \right) - \left(\mathbf{{s}}_{k-1},{\boldsymbol{\widetilde{\gamma} }_{k-1}} \right)} \right\|_2} \right\} \le \epsilon$
\STATE Return the optimal solution $\left(\mathbf{{s}}^*,{\boldsymbol{\widetilde{\gamma} }^*} \right) = \left(\mathbf{{s}}_k,{\boldsymbol{\widetilde{\gamma} }_k} \right)$
\end{algorithmic}
\end{algorithm}
\end{table}

Now, the problem in \eqref{C5:P4DCUpperBound} is a convex programming problem which can be easily solved by CVX \cite{cvx}.
The solution of the problem in \eqref{C5:P4DCUpperBound} provides an upper bound for the problem in \eqref{C5:P4DC}.
To tighten the obtained upper bound, we employ an iterative algorithm to generate a sequence of feasible solution successively, cf. \textbf{Algorithm} \ref{C5:alg2}.
In \textbf{Algorithm} \ref{C5:alg2}, the initial feasible solution $\left(\mathbf{{s}}_1,{\boldsymbol{\widetilde{\gamma} }_1} \right)$ is obtained via solving the problem in \eqref{C5:P4DC} with $G_{1}^{\eta}\left(\mathbf{{s}},{\boldsymbol{\widetilde{\gamma} }} \right)$
as the objective function\footnote{Note that with $G_{1}^{\eta}\left(\mathbf{{s}},{\boldsymbol{\widetilde{\gamma} }} \right)$
as the objective function, the problem in (41) is a convex problem. Therefore, the initial feasible solution can be found via existing algorithms \cite{Boyd2004} for solving convex problems with a polynomial time computational complexity.}.
The problem in \eqref{C5:P4DCUpperBound} is updated with the intermediate solution from the last iteration and is solved to generate a feasible solution for the next iteration.
Such an iterative procedure will stop when the maximum iteration number is reached or the change of optimization variables is smaller than a predefined convergence tolerance.
It has been shown that the proposed suboptimal algorithm converges a stationary point with a polynomial time computational complexity for differentiable ${G_{\eta}^{1}} \left(\mathbf{{s}},{\boldsymbol{\widetilde{\gamma} }} \right)$ and ${G_{\eta}^{2}} \left(\mathbf{{s}},{\boldsymbol{\widetilde{\gamma} }} \right)$ \cite{VucicProofDC}.
Noted that there is no guarantee that \textbf{Algorithm} \ref{C5:alg2} can
converge to a globally optimum of the problem in \eqref{C5:P2}. However, our simulation results in the next section will demonstrate its close-to-optimal performance.

In practice, different numerical methods\cite{Boyd2004} can be used to solve the convex problem in \eqref{C5:P4DCUpperBound}.
Particularly, the computational complexity of proposed suboptimal algorithm implemented by the primal-dual path-following interior-point method \cite{Nemirovski2004IPM} is given by
\begin{equation}\label{C5:Complexity}
\mathcal{O}\left(K_{\mathrm{max}}\underbrace{\left(\sqrt{8N_{\mathrm{F}}M+N_{\mathrm{F}}+M}\ln(\frac{1}{\Delta})\right)}_{\text{Number of Newton iterations}}
\underbrace{\left( \left(N_{\mathrm{F}}+M\right)(N_{\mathrm{F}}M)^2+35(N_{\mathrm{F}}M)^3\right)}_{\text{Complexity per Newton iteration}}\right),
\end{equation}
for a given solution accuracy $\Delta>0$ of the adopted numerical solver, where $\mathcal{O}(\cdot)$ is the big-O notation.
On the other hand, for the proposed optimal algorithm in \textbf{Algorithm} \ref{C5:alg1}, we note that although the B\&B algorithm is guaranteed to find the optimal solution, the required computational complexity in the worst-case is as high as that of an exhaustive search.
The computational complexity of an exhaustive search for the problem in \eqref{C5:P2} is $\mathcal{O}\Bigg(2^{N_{\mathrm{F}}M}\Big(\prod_{m=1}^{M}\frac{2^{R_m^{\mathrm{total}}}-1}{\Delta}\Big)^{N_{\mathrm{F}}}\Bigg),$
for a given solution accuracy $\Delta>0$.
Therefore, the proposed suboptimal algorithm provides a substantial saving in computational complexity compared to the exhaustive search approach.
We note that proposed suboptimal algorithm with a polynomial time computational complexity is desirable for real time implementation \cite{thomas2001introduction}.

\section{Simulation Results} \label{C5:SimulationResults}
In this section, we evaluate the performance of our proposed resource allocation algorithms through simulations.
Unless specified otherwise, the system parameters used in the our simulations are given as follows.
A single-cell with a BS located at the center with a cell size of $500$ m is considered.
The carrier center frequency is $1.9$ GHz and the bandwidth of each subcarrier is $15$ kHz.
There are $M$ users randomly and uniformly distributed between $30$ m and $500$ m, i.e., $d_i \sim U[30,\;500]$ m, and their target data rates are generated by ${R}_m^{\mathrm{total}} \sim U[1,\;10]$ bit/s/Hz.
The required outage probability of each user on each subcarrier is generated by $\delta_{i,m} \sim U[10^{-5},\;10^{-1}]$.
The user noise power on each subcarrier is $\sigma_{i,m}^2=-128$ dBm and the variance of channel estimation error is $\kappa^2_{i,m} = 0.1$, $\forall i,m$.
The 3GPP urban path loss model with a path loss exponent of $3.6$ \cite{Access2010} is adopted in our simulations.
For our proposed iterative optimal and suboptimal resource allocation algorithms, the maximum error tolerance is set to $\epsilon = 0.01$ and the penalty factors is set to a sufficiently large number such that the value of the penalty term comparable to the value of the objective function\cite{DerrickFD2016}.
The simulations shown in the sequel are obtained by averaging the results over different realizations of different user distances, target data rates, multipath fading coefficients, and outage probability requirements.

For comparison, we consider the performance of the following three baseline schemes.
For baseline 1, the conventional MC-OMA scheme is considered where each subcarrier can only be allocated to at most one user.
To support all $M$ active users and to have a fair comparison, the subcarrier spacing is changed by a factor of $\frac{N_\mathrm{F}}{M}$ to generate $M$ subcarriers.
Since our proposed scheme subsumes the MC-OMA scheme as a subcase, the minimum power consumption for baseline scheme 1 can be obtained by solving the problem in \eqref{C5:P2} by replacing C6 with $\sum\limits_{m = 1}^M {{s_{i,m}}} = 1$.
For baseline 2, a scheme of MC-NOMA with random scheduling is considered where the paired users on each subcarrier is randomly selected\cite{Wei2016NOMA}.
The minimum power consumption for baseline scheme 2 is obtained via solving the problem in \eqref{C5:P4DC} with a given random user scheduling policy $\mathbf{{s}}$ and $\eta = 0$.
For baseline 3, a scheme of MC-NOMA with an equal rate allocation is studied where the target data rate of each user is assigned equally on its allocated subcarriers\cite{Wei2016NOMA}.
Based on the equal rate allocation, the problem in \eqref{C5:P2} can be transformed to a mixed integer linear program, which can be solved by standard numerical integer program solvers, such as Mosek\cite{Mosek2010}, via some non-polynomial time algorithms.

\subsection{Convergence of Proposed Algorithms}

\begin{table}[ptb]
\center\small
\caption{Optimal Solution of \eqref{C5:P1} for A Single Channel Realization with $N_{\mathrm{F}} = 4$ and $M = 7$ in Figure \ref{C5:CaseI_Convergence}}
  \begin{tabular}{ccccccccc}
  \hline
    Subcarrier index & \multicolumn{2}{c}{1} & \multicolumn{2}{c}{2} & \multicolumn{2}{c}{3} & \multicolumn{2}{c}{4} \\ \hline
    Paired user index & 2 & 5 & 5 & 7 & 1 & 4 & 3 & 6 \\
    $\beta_{i,m}$ & 783.39 & 39.99 & 30.92 & 520.27 & 8.57 & 269.80 & 9.59 & 1349.80 \\
    $R_{i,m}$ (bit/s/Hz)& 8 & 2.03 & 4.97 & 3 & 1 & 7 & 3 & 4\\
%    Power allocation $p_{i,m}$ (mW)& 6.0830 & 0.0698 & 27.9943 & 1.0704 & 0.0325 & 2.1557 & 30.4580 & 0.0128 \\
    $p_{i,m}$ (dBm)& 25.13 & 30.35 & 31.42 & 11.29 & 27.69 & 26.73 & 29.07 & 10.46 \\
    SIC decoding   & \Checkmark & - & - & \Checkmark & - & \Checkmark & - & \Checkmark \\
    \hline
  \end{tabular}
\label{C5:PairedUsers}
\end{table}

Figure \ref{C5:CaseI_Convergence} illustrates the convergence of our proposed optimal and suboptimal resource allocation algorithms for different values\footnote{Since the computational complexity of the B\&B approach is high, we adopt small values for $M$ and $N_\mathrm{F}$ to compare the gap between the proposed optimal algorithm and the suboptimal algorithm. We note that our proposed suboptimal resource allocation algorithm is computational efficient compared to the optimal one, which can apply to scenarios with more users and subcarriers, such as the simulation case in Section \ref{C5:SimulationD}. In fact, the number of subcarriers in this chapter can be viewed as the number of resource blocks in LTE standard\cite{sesia2015lte}, where the user scheduling are performed on resource block level.} of $N_{\mathrm{F}}$ and $M$.
For the first case with $N_{\mathrm{F}} = 4$ and $M = 7$, we observe that the optimal algorithm generates a non-increasing upper bound and a non-decreasing lower bound when the number of iterations increases.
Besides, the optimal solution is found when the two bounds meet after $600$ iterations on average.
More importantly, our proposed suboptimal algorithm can converge to the optimal value within $80$ iterations on average.
For the second case with $N_{\mathrm{F}} = 8$ and $M = 15$, it can be observed that the optimal algorithm converges after $4500$ iterations on average.
In fact, the computational complexity of the proposed optimal algorithm increases exponentially w.r.t. the number of optimization variables, and thus the convergence speed is relatively slow for a larger problem size.
However, it can be observed that the suboptimal algorithm converges faster when the numbers of subcarriers and users increase, and it can achieve the optimal value with only 25 iterations in the second case on average.
This is because the time-sharing condition \cite{Yu2006Dual,DerrickLimitedBackhaul} is satisfied with a larger number of subcarriers.
In this case, the optimization problem in \eqref{C5:P4DC} tends to be convexified leading to a higher chance of holding strong duality \cite{DerrickLimitedBackhaul}.
Further, our proposed suboptimal scheme is able to exploit the ``convexity" inherent in large-scale optimization problem via the successive convex approximation while the optimal one cannot with relying on feasible set partitioning, cf. Figure \ref{C5:BnB:a}.
Therefore, the proposed suboptimal algorithm converges faster with a larger number of subcarriers in the second case.
To obtain further insight, Table \ref{C5:PairedUsers}  shows the solution of our original formulated problem in \eqref{C5:P1} via following the proposed optimal SIC decoding policy for a single channel realization with $N_{\mathrm{F}} = 4$ and $M = 7$ in Figure \ref{C5:CaseI_Convergence}.
The tick denotes that the user is selected to perform SIC.
It can be observed that two users with distinctive CNR outage thresholds are preferred to be paired together (users 1 and 4, users 3 and 6).
Also, the users (users 2, 4, 6, 7) with higher CNR outage thresholds are selected to perform SIC and only a fraction of power are allocated to them owing to their better channel conditions or non-stringent QoS requirements.
These observations are in analogy to the conclusions for the case of NOMA with perfect CSIT, where users with distinctive channel gains are more likely to be paired, more power is allocated to the weak user, and the strong user is selected to perform SIC\cite{Dingtobepublished}.
Therefore, the defined CNR outage threshold in this chapter serves as a metric for determining the optimal SIC decoding policy and resource allocation design for MC-NOMA systems with imperfect CSIT.
\begin{figure}[t]
\centering
\includegraphics[width=4.5in]{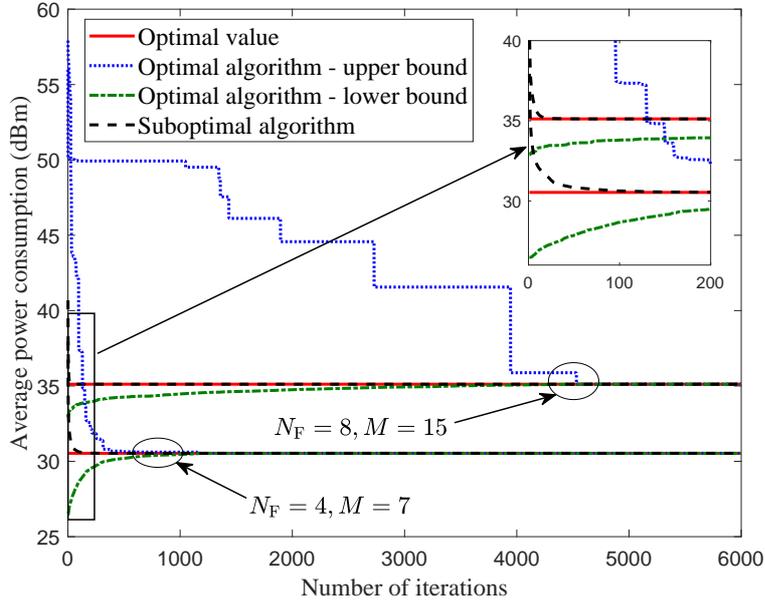}
\caption{Convergence of the proposed optimal and suboptimal resource allocation algorithms.}
\label{C5:CaseI_Convergence}
\end{figure}

\subsection{Power Consumption versus Target Data Rate}
\begin{figure}[t]
\centering
\includegraphics[width=4.5in]{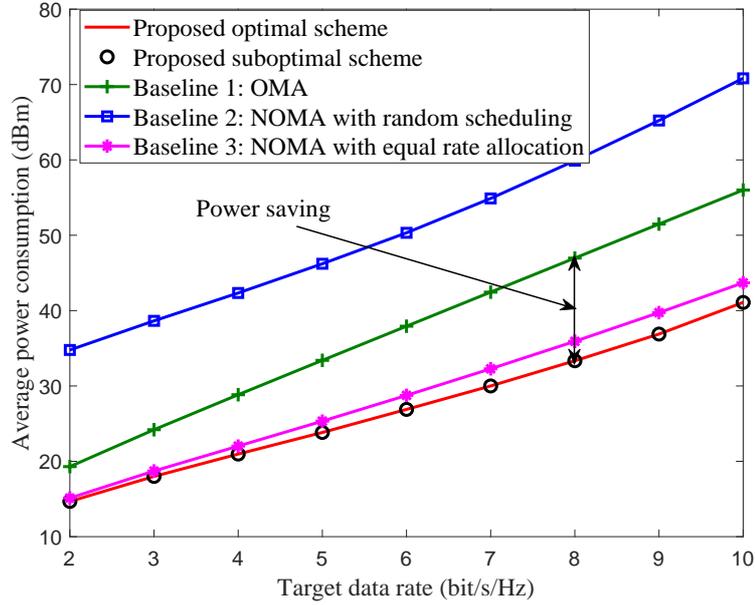}
\caption{Average power consumption (dBm) versus target data rate with $N_\mathrm{F} = 8$ and $M = 12$. The power saving gain achieved by the proposed schemes over baseline scheme 1 is denoted by the double-side arrow.}
\label{C5:CaseII_TargetRate}
\end{figure}
In Figure \ref{C5:CaseII_TargetRate}, we investigate the power consumption versus the target data rate with $N_{\mathrm{F}} = 8$ and $M = 12$.
In this simulation, all the users have an identical target data rates ${R}_m^{\mathrm{total}}$ and they are set to be from $1$ bit/s/Hz to $10$ bit/s/Hz.
The three baseline schemes are also included for comparison.
As can be observed from Figure \ref{C5:CaseII_TargetRate}, the power consumption increases monotonically with the target data rate for all the schemes.
Clearly, the BS is required to transmit with a higher power to support a more stringent data rate requirement.
Besides, our proposed optimal and suboptimal resource allocation schemes provide a significant power reduction compared to the baseline schemes.
Specifically, baseline scheme 1 requires a higher power consumption (about $3\sim15$ dB) compared to proposed schemes.
This is attributed to the fact that the proposed NOMA schemes are able to distribute the required target data rate of each user over multiple subcarriers efficiently since they admit multiplexing multiple users on each subcarrier.
For OMA schemes, the power consumption increases exponentially with the target data rate requirement since only one subcarrier is allocated to each user in the overloaded scenario.
As a result, the performance gain of NOMA over OMA in terms of power consumption becomes larger when the target data rate increases.
For baseline scheme 2, it can be observed that NOMA is very sensitive to the user scheduling strategy where NOMA with suboptimal random scheduling even consumes more power than that of OMA schemes.
Therefore, a cautiously design of the user scheduling strategy for MC-NOMA systems is fundamentally important in practice.
For baseline scheme 3, the power consumption is slightly higher than that of the proposed schemes but the performance gap is enlarged with an increasing target data rate.
In fact, baseline scheme 3 shares the target data rate equally across the allocated subcarriers of a user, which can realize most of the performance gain of NOMA in low target data rate regimes.
However, our proposed schemes consume less transmit power for high target data rate by exploiting the frequency diversity, where higher rates are allocated to the subcarriers with better channel conditions.
{We can observe that baseline scheme 3 outperforms substantially baseline scheme 2 and is closest to the proposed scheme. This is because user scheduling is more important than rate allocation for exploiting the performance gain of NOMA. For instance, when the scheduled two NOMA users possess the equal channel gain, a significant inter-user interference exists and thus degrades the system power efficiency. Therefore, baseline scheme 3 performs user scheduling with a fixed equal rate allocation strategy, which is substantially better than that of baseline scheme 2, where a random user scheduling strategy is adopted.}

\subsection{Power Consumption versus Channel Estimation Error}
Figure \ref{C5:CaseIV_ChannelError} depicts the power consumption versus the variance of channel estimation error with $N_{\mathrm{F}} = 8$, $M = 12$, and ${R}_m^{\mathrm{total}} \sim U[1,\;10]$ bit/s/Hz.
The variance of channel estimation error $\kappa^2_{i,m}$ increasing from $0$ to $0.5$, where $\kappa^2_{i,m} = 0$ denotes that perfect CSIT is available for resource allocation.
It can be observed that the power consumption increases monotonically with $\kappa^2_{i,m}$ for all the schemes.
It is expected that a higher transmit power is necessary to cope with a larger channel uncertainty to satisfy its required outage probability.
Particularly, for our proposed schemes, baseline scheme 1, and baseline scheme 3, a $6$ dB of extra power is required to handle the channel estimation error when $\kappa^2_{i,m}$ increases from $0$ to $0.5$.
However, our proposed schemes are the most power-efficient among all the schemes.
Furthermore, compared to Figure \ref{C5:CaseII_TargetRate} with identical target data rates, the gap of power consumption between baseline scheme 3 and our proposed schemes at $\kappa^2_{i,m} = 0.1$ is enlarged.
Also, the performance gain of our proposed schemes over baseline scheme 1 is larger than that of Figure \ref{C5:CaseII_TargetRate}.
In fact, our proposed schemes can exploit the heterogeneity of the target data rates via users multiplexing and rate allocation. Particularly, users multiplexing of NOMA enables rate splitting onto multiple subcarriers in the overloaded scenario. Moreover, instead of equal rate allocation, our proposed schemes are more flexible to combat the large dynamic range of target data rates via exploiting the frequency diversity. Therefore, for random target data rate, our proposed schemes are more efficient to reduce the power consumption.
\begin{figure}[t]
\centering
\includegraphics[width=4.5in]{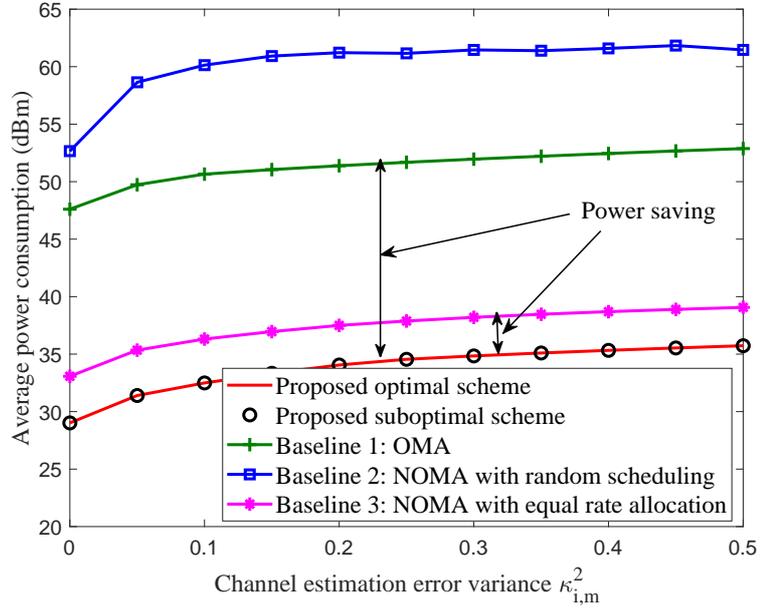}
\caption{Average power consumption (dBm) versus channel estimation error variance with $N_\mathrm{F} = 8$ and $M = 12$. The power saving gain achieved by the proposed schemes over baseline scheme 1 and baseline scheme 3 are denoted by the double-side arrows.}
\label{C5:CaseIV_ChannelError}
\end{figure}

\subsection{Power Consumption versus Number of Users}\label{C5:SimulationD}
\begin{figure}[t]
\centering
\includegraphics[width=4.5in]{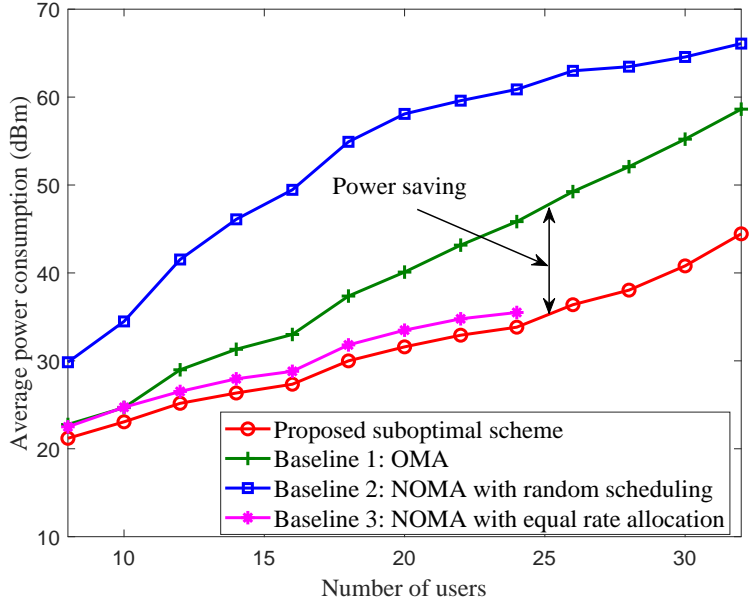}
\caption{Average power consumption (dBm) versus number of users with $N_\mathrm{F} = 16$ and ${R}_m^{\mathrm{total}} = 8$ bit/s/Hz. The power saving gain achieved by the proposed scheme over baseline scheme 1 is denoted by the double-side arrow.}
\label{C5:CaseIII_NumUser}
\end{figure}
Figure \ref{C5:CaseIII_NumUser} illustrates the power consumption versus the number of users with $N_\mathrm{F} = 16$ and ${R}_m^{\mathrm{total}} = 8$ bit/s/Hz, $\forall m$.
The proposed optimal scheme is not included here due to its exponentially computational complexity.
We observe that our proposed scheme is also applicable to underloaded systems with $N_\mathrm{F} > M$, and it is more power-efficient than that of the OMA scheme in both overloaded and underloaded systems.
Furthermore, it can be seen that the power consumption increases with the number of users for all the considered schemes.
This is because a higher power consumption is required when there are more users requiring stringent QoSs.
Besides, our proposed scheme is the most power-efficient among all the schemes.
In particular, compared to the proposed suboptimal scheme, baseline scheme 2 requires a substantially higher power consumption since NOMA requires a careful design of user scheduling to cope with the inherent interference.
On the contrary, baseline scheme 3 needs a slightly higher power than the proposed suboptimal scheme.
As mentioned before, baseline scheme 3 can exploit most of the performance gain of NOMA via enabling multiuser multiplexing with equal rate allocation.
{Note that the computational complexity of baseline scheme 3 is extremely high when the number of users larger than 25. Thus, the simulation results for baseline scheme 3 with M>25 is not shown in Fig \ref{C5:CaseIII_NumUser}.}

Compared to baseline scheme 1, we can observe that the power saving brought by our proposed suboptimal scheme increases with the number of users.
This can be attributed to the spectral efficiency gain \cite{Ding2015b} and multiuser diversity gain \cite{Sun2016Fullduplex} of NOMA.
On the one hand, NOMA allows multiuser multiplexing on each subcarrier, which provides higher spectral efficiency than that of OMA.
As a result, a smaller amount of power is able to support the NOMA users' QoS requirements than the users using OMA.
Besides, the proposed scheme can efficiently exploit the spectral efficiency gain to reduce the power consumption compared to baseline scheme 1.
In particular, with an increasing number of users, the spectrum available to each user in baseline scheme 1 becomes less due to the exclusive subcarrier allocation constraint, while relatively more spectrum is available in the proposed MC-NOMA scheme owing to the power domain multiplexing.
Consequently, the proposed scheme can save more power compared to the baseline scheme 1.
On the other hand, NOMA possesses a higher capability in exploiting the multiuser diversity than that of OMA.
Particularly, instead of scheduling a single user on each subcarrier in OMA, NOMA enables multiuser multiplexing on each subcarrier, which promises more degrees of freedom for user selection and power allocation to exploit the multiuser diversity.
Therefore, our proposed NOMA scheme with the suboptimal resource allocation design can effectively utilize the multiuser diversity to reduce the total transmit power.
In fact, in the considered MC-NOMA systems, the multiuser diversity comes from the heterogeneity of CNR outage thresholds.
The CNR outage thresholds become more heterogeneous for an increasing number of users. Thus, the power saving gain brought by the proposed NOMA scheme over the OMA scheme increases with the number of users.
\subsection{Outage Probability}
\begin{figure}[ptb]
\centering
\subfigure[Outage probability for all the users with $\kappa^2_{i,m} = 0.1$.]
{\label{C5:CaseV_OutageProbability:a} %% label for first subfigure
\includegraphics[width=4.5in]{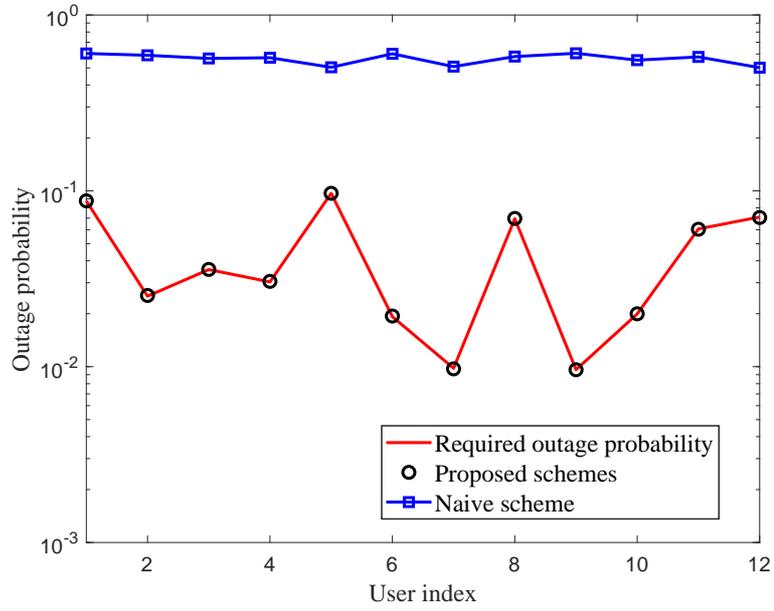}}
\subfigure[Outage probability versus $\kappa^2_{i,m}$ for user 9.]
{\label{C5:CaseV_OutageProbability:b} %% label for second subfigure
\includegraphics[width=4.5in]{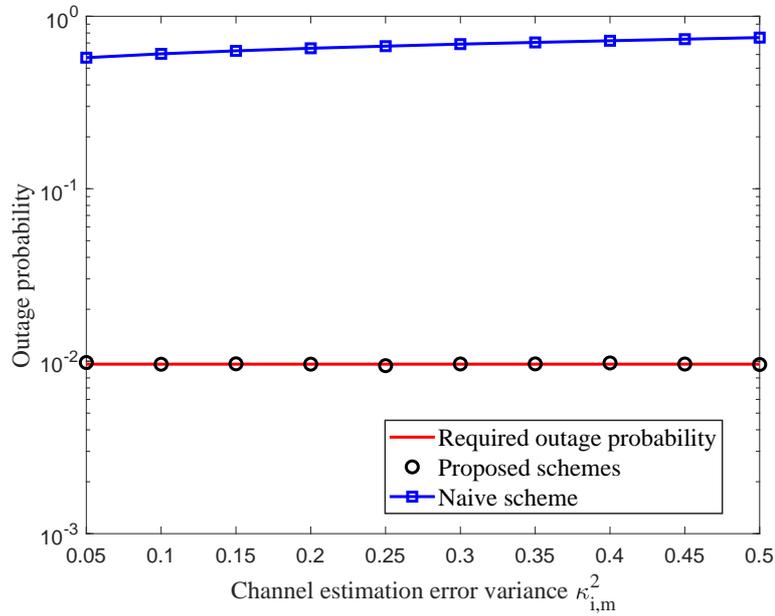}}
\caption{Outage probability of our proposed scheme and a naive scheme with $N_\mathrm{F} = 8$ and $M = 12$.}
\label{C5:CaseV_OutageProbability}%
\end{figure}

In this simulation, we introduce a naive scheme where the resource allocation is performed by treating the estimated channel coefficient $\hat{h}_{i,m}$ as perfect CSIT.
Figures \ref{C5:CaseV_OutageProbability:a} and \ref{C5:CaseV_OutageProbability:b} compare the outage probability for our proposed schemes and the naive scheme with $N_\mathrm{F} = 8$ and $M = 12$.
Figure \ref{C5:CaseV_OutageProbability:a} illustrates the outage probability for all the users with channel estimation error variance $\kappa^2_{i,m} = 0.1$.
It can be observed that our proposed schemes can satisfy the required outage probability of all the users while the naive scheme leads to a significantly higher outage probability than the required.
Figure \ref{C5:CaseV_OutageProbability:b} shows the outage probability versus $\kappa^2_{i,m}$ for user 9.
It can be observed that our proposed scheme can always satisfy the required outage probability, despite $\kappa^2_{i,m}$ increases from $0.05$ to $0.5$.
In contrast, the outage probability for the naive scheme increases with $\kappa^2_{i,m}$ due to the deteriorated quality of channel estimates.
In fact, our resource allocation design can guarantee the required outage probability, at the expense of a slightly higher transmit power compared to the case of perfect CSIT, cf. Figure \ref{C5:CaseIV_ChannelError} for $\kappa^2_{i,m} = 0$.

\section{Summary}
In this chapter, we studied the power-efficient resource allocation algorithm design for MC-NOMA systems. The resource allocation algorithm design was formulated as a non-convex optimization problem and it took into account the imperfect CSIT and heterogenous QoS requirements. We proposed an optimal resource allocation algorithm, in which the optimal SIC decoding policy was determined by the CNR outage threshold. Furthermore, a suboptimal resource allocation scheme was proposed based on D.C. programming, which can converge to a close-to-optimal solution rapidly. Simulation results showed that our proposed resource allocation schemes provide significant transmit power savings and enhanced robustness against channel uncertainty via exploiting the heterogeneity of channel conditions and QoS requirements of users in MC-NOMA systems.

\chapter{Multi-Beam NOMA for Hybrid mmWave Systems}\label{C6:chapter6}

\section{Introduction}
The previous chapters focused on non-orthogonal multiple access (NOMA) in microwave communications.
However, the spectrum resources is very scarce in microwave frequency band, which creates a fundamental bottleneck for capacity improvement and sustainable system evolution\cite{Rappaport2013,Andrews2014}.
As a result, higher frequency bands with a wider available frequency bandwidth, such as mmWave bands \cite{Rappaport2013,ZhaommWaveMulticell}, are highly desirable for future wireless networks.
From this chapter, we start to explore the application of NOMA in millimeter wave (mmWave) communications.
In this chapter, we propose a multi-beam NOMA scheme for hybrid mmWave communication systems.

As mentioned in the Chapter 1, in hybrid mmWave systems, the limited number of RF chains restricts the number of users that can be served simultaneously via orthogonal multiple access (OMA), i.e., one RF chain can serve at most one user.
In particular, in overloaded scenarios, i.e., the number of users is larger than the number of RF chains, OMA fails to accommodate all the users.
This motivates us attempting to overcome the limitation incurred by the small number of RF chains via introducing the concept of NOMA into hybrid mmWave systems.
In particular, the users within the same analog beam can be clustered as a NOMA group\cite{Ding2017RandomBeamforming} and share a RF chain, which is named as single-beam mmWave-NOMA scheme in this chapter.
However, due to the narrow analog beamwidth in hybrid mmWave systems, the number of users that can be served concurrently by the single-beam mmWave-NOMA scheme is very limited and it depends on the users' angle-of-departure (AOD) distribution.
This reduces the potential performance gain brought by NOMA in hybrid mmWave systems.
Therefore, we propose a multi-beam NOMA scheme for hybrid mmWave communication systems, which provides a higher flexibility in user clustering and thus can efficiently exploit the potential multi-user diversity.

In this chapter, we propose a multi-beam NOMA framework for a multiple RF chain hybrid mmWave system and study the resource allocation design for the proposed multi-beam mmWave-NOMA scheme.
Specifically, all the users are clustered into several NOMA groups and each NOMA group is associated with a RF chain.
Then, multiple analog beams are formed for each NOMA group to facilitate downlink NOMA transmission by exploiting the channel sparsity and the large-scale antenna array at the BS.
To this end, a novel \emph{beam splitting} technique is proposed, which dynamically divides the whole antenna array associated with a RF chain into multiple subarrays to form multiple beams.
Compared to the conventional single-beam mmWave-NOMA scheme\cite{Ding2017RandomBeamforming,Cui2017Optimal, WangBeamSpace2017}, our proposed multi-beam NOMA scheme offers a higher flexibility in serving multiple users with an arbitrary AOD distribution.
As a result, our proposed scheme can form more NOMA groups, which facilitates the exploitation of multi-user diversity to achieve a higher spectral efficiency.

To improve the performance of our proposed multi-beam mmWave-NOMA scheme, we further propose a two-stage resource allocation design.
In the first stage, given the equal power allocation and applying an identity matrix as an initial and simple digital precoder, the joint design of user grouping and antenna allocation is formulated as an optimization problem to maximize the conditional system sum-rate.
We recast the formulated problem as a coalition formation game \cite{SaadCoalitionalGame,SaadCoalitional2012,WangCoalitionNOMA,SaadCoalitionOrder,Han2012} and develop an algorithmic solution for user grouping and antenna allocation.
In the second stage, based on the obtained user grouping and antenna allocation strategy, the power allocation problem is formulated to maximize the system sum-rate by taking into account the quality of service (QoS) requirement.
A suboptimal power allocation algorithm is obtained based on the difference of convex (D.C.) programming technique.
Compared to the optimal benchmarks in the two stages, we show that our proposed resource allocation design can achieve a close-to-optimal performance in each stage in terms of the system sum-rate.
Simulation results demonstrate that the proposed multi-beam mmWave-NOMA scheme can achieve a higher spectral efficiency than that of the conventional single-beam mmWave-NOMA scheme and mmWave-OMA scheme.

\section{System Model}

\begin{figure}[t]
\centering
\subfigure[The proposed multi-beam mmWave-NOMA scheme.]
{\label{C6:MultiBeamNOMA:a} %% label for first subfigure
\includegraphics[width=0.95\textwidth]{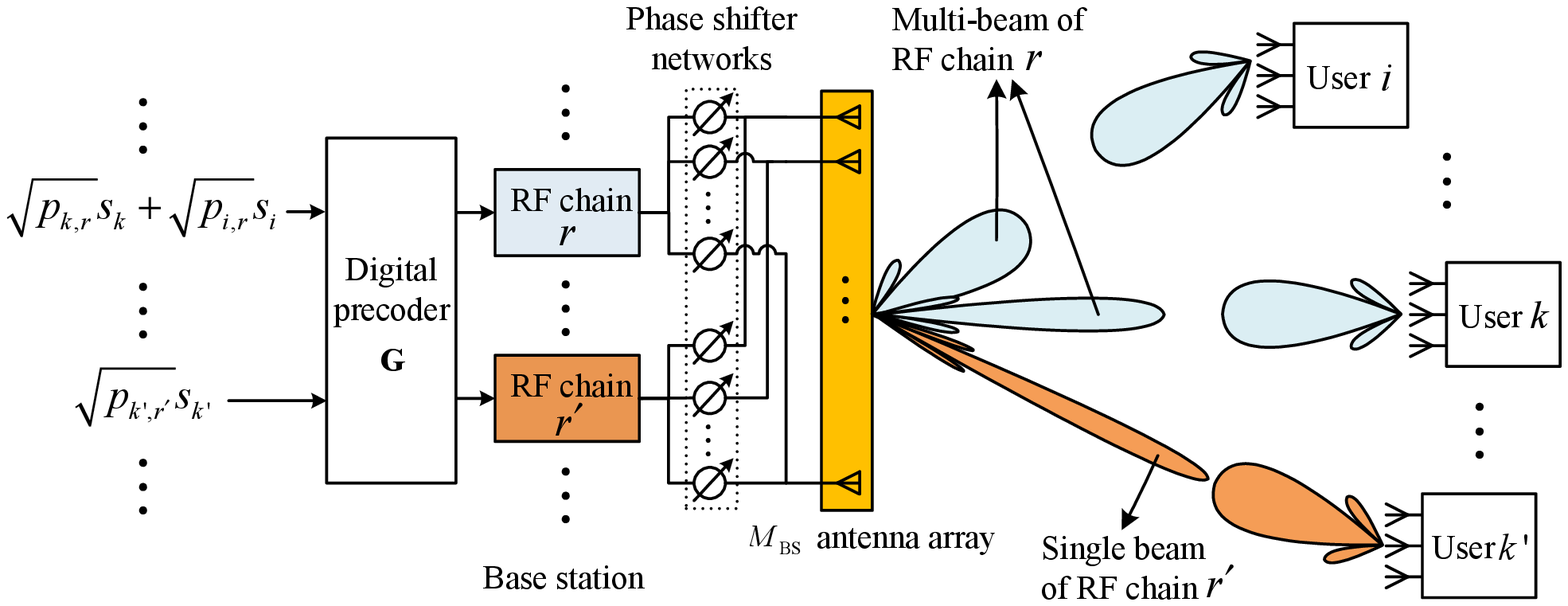}}
\subfigure[Hybrid structure receiver at users.]
{\label{C6:MultiBeamNOMA:b} %% label for second subfigure
\includegraphics[width=0.5\textwidth]{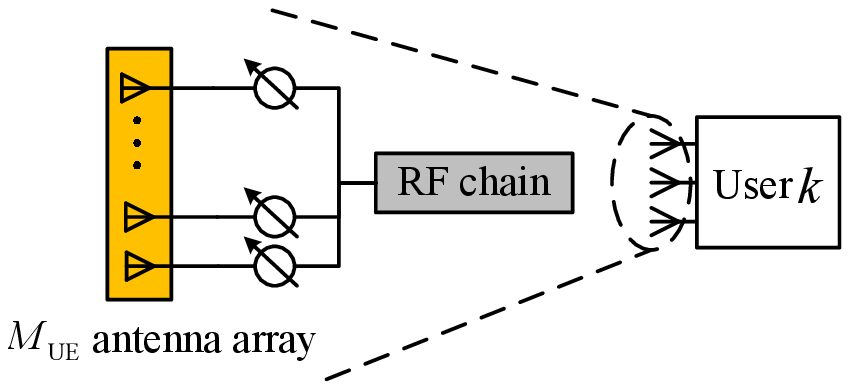}}
\caption{System model of the proposed multi-beam NOMA scheme for hybrid mmWave systems.}
\label{C6:MultiBeamNOMA}%
\end{figure}

We consider the downlink hybrid mmWave communication in a single-cell system with one base station (BS) and $K$ users, as shown in Figure \ref{C6:MultiBeamNOMA}.
In this work, we adopt a fully access hybrid structure\footnote{To simplify the presentation in this chapter, the proposed scheme and resource allocation design are based on the fully access hybrid structures as an illustrative example, while they can be easily extended to subarray hybrid structures.} to illustrate the proposed multi-beam NOMA framework for hybrid mmWave systems\cite{zhao2017multiuser,GaoSubarray,lin2016energy}.
In particular, the BS is equipped with $M_{\mathrm{BS}}$ antennas but only connected to $N_{\mathrm{RF}}$ RF chains with $M_{\mathrm{BS}} \gg N_{\mathrm{RF}}$.
We note that each RF chain can access all the $M_{\mathrm{BS}}$ antennas through $M_{\mathrm{BS}}$ phase shifters, as shown in Figure \ref{C6:MultiBeamNOMA:a}.
Besides, each user is equipped with $M_{\mathrm{UE}}$ antennas connected via a single RF chain, as shown in Figure \ref{C6:MultiBeamNOMA:b}.
We employ the commonly adopted uniform linear array (ULA) structure\cite{zhao2017multiuser} at both the BS and user terminals.
We assume that the antennas at each transceiver are deployed and separated with equal-space of half wavelength with respect to the neighboring antennas.
In this work, we focus on the overloaded scenario with $K \ge N_{\mathrm{RF}}$, which is fundamentally different from existing works in hybrid mmWave communications, e.g. \cite{zhao2017multiuser,GaoSubarray,lin2016energy}.
In fact, our considered system model is a generalization of that in existing works\cite{zhao2017multiuser,GaoSubarray,lin2016energy}.
For example, the considered system can be degenerated to the conventional hybrid mmWave systems when $K \le N_{\mathrm{RF}}$ and each NOMA group contains only a single user.

We use the widely adopted the Saleh-Valenzuela model \cite{WangBeamSpace2017} as the channel model for our considered mmWave communication systems.
In this model, the downlink channel matrix of user $k$, ${{\bf{H}}_k} \in \mathbb{C}^{{ M_{\mathrm{UE}} \times M_{\mathrm{BS}}}}$, can be represented as
\begin{equation}\label{C6:ChannelModel1}
{{\bf{H}}_k} = {\alpha _{k,0}}{{\bf{H}}_{k,0}} + \sum\limits_{l = 1}^L {{\alpha _{k,l}}{{\bf{H}}_{k,l}}},
\end{equation}
where ${\mathbf{H}}_{k,0} \in \mathbb{C}^{ M_{\mathrm{UE}} \times M_{\mathrm{BS}} }$ is the line-of-sight (LOS) channel matrix between the BS and user $k$ with ${\alpha _{k,0}}$ denoting the LOS complex path gain, ${\mathbf{H}}_{k,l} \in \mathbb{C}^{ M_{\mathrm{UE}} \times M_{\mathrm{BS}} }$ denotes the $l$-th non-line-of-sight (NLOS) path channel matrix between the BS and user $k$ with ${\alpha _{k,l}}$ denoting the corresponding NLOS complex path gains, $1 \le l \le L$, and $L$ denoting the total number of NLOS paths\footnote{If the LOS path is blocked, we treat the strongest NLOS path as ${\mathbf{H}}_{k,0}$ and all the other NLOS paths as ${\mathbf{H}}_{k,l}$.}.
In particular, ${\mathbf{H}}_{k,l}$, $\forall l \in \{0,\ldots,L\}$, is given by
\begin{equation}
{\mathbf{H}}_{k,l} = {\mathbf{a}}_{\mathrm{UE}} \left(  \phi _{k,l} \right){\mathbf{a}}_{\mathrm{BS}}^{\mathrm{H}}\left( \theta _{k,l} \right),
\end{equation}
with ${\mathbf{a}}_{\mathrm{BS}}\left( \theta _{k,l} \right) = \left[ {1, \ldots ,{e^{ - j{\left({M_{{\mathrm{BS}}}} - 1\right)}\pi  \cos \theta _{k,l} }}}\right]^{\mathrm{T}}$
denoting the array response vector \cite{van2002optimum} for the AOD of the $l$-th path ${\theta _{k,l}}$ at the BS and ${\mathbf{a}}_{\mathrm{UE}}\left( \phi _{k,l} \right) = \left[ 1, \ldots ,{e^{ - j{\left({M_{{\mathrm{UE}}}} - 1\right)}\pi \cos \phi _{k,l} }} \right] ^ {\mathrm{T}}$ denoting the array response vector for the angle-of-arrival (AOA) of the $l$-th path ${\phi _{k,l}}$ at user $k$.
Besides, we assume that the LOS channel state information (CSI), including the AODs ${\theta _{k,0}}$ and the complex path gains ${{\alpha _{k,0}}}$ for all the users, is known at the BS owing to the beam tracking techniques \cite{BeamTracking}.
For a similar reason, the AOA ${\phi _{k,0}}$ is assumed to be known at user $k$, $\forall k$.

\section{The Multi-beam NOMA scheme}

\begin{figure}[t]
\centering
\includegraphics[width=4.5in]{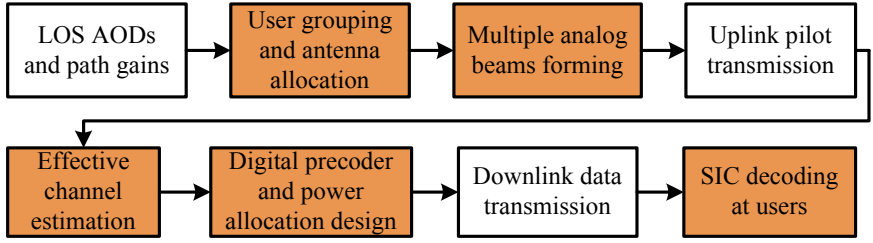}
\caption{The proposed multi-beam NOMA framework for hybrid mmWave systems. The shaded blocks are the design focuses of this chapter.}
\label{C6:MultiBeamNOMAScheme}
\end{figure}

The block diagram of the proposed multi-beam NOMA framework for the considered hybrid mmWave system is shown in Figure \ref{C6:MultiBeamNOMAScheme}.
Based on the LOS CSI, we cluster users as multiple NOMA groups and perform antenna allocation among users within a NOMA group.
Then, we control the phase shifters based on the proposed beam splitting technique to generate multiple analog beams.
Effective channel is estimated at the BS based on the uplink pilot transmission and the adopted analog beamformers.
Then, according to the effective channel, the digital precoder and power allocation are designed for downlink data transmission.
Since superposition transmission is utilized within a NOMA group, SIC decoding will be performed at the strong users as commonly adopted in traditional downlink NOMA protocol\cite{WeiSurvey2016}.
The shaded blocks are the design focuses of this chapter, which are detailed in the sequel.

\subsection{User Grouping and Antenna Allocation}
Based on the LOS AODs of all the users $\left\{ {\theta _{1,0}}, \ldots ,{\theta _{K,0}} \right\}$ and their path gains $\left\{ {\alpha _{1,0}}, \ldots ,{\alpha _{K,0}} \right\}$, we first perform user grouping and antenna allocation.
In particular, multiple users might be allocated with the same RF chain to form a NOMA group.
We can define the user scheduling variable as follows:
\begin{equation}
{u_{k,r}} = \left\{ {\begin{array}{*{20}{c}}
1,&{{\mathrm{user}}\;k\;{\mathrm{is}}\;{\mathrm{allocated}}\;{\mathrm{to}}\;{\mathrm{RF}}\;{\mathrm{chain}}\;r},\\
0,&{{\mathrm{otherwise}}}.
\end{array}} \right.
\end{equation}
To reduce the computational complexity and time delay of SIC decoding within the NOMA group, we assume that at most $2$ users\footnote{This assumption is commonly adopted in the NOMA literature, e.g. \cite{WeiTCOM2017,Sun2016Fullduplex}, to reduce the computational complexity and signal processing delay incurred in SIC decoding within the NOMA group.
However, the proposed algorithms in this chapter can be generalized to the scenarios with more than two users in each NOMA group at the expense of a more involved notations.} can be allocated with the same RF chain, i.e., $\sum\limits_{k = 1}^K {{u_{k,r}}}  \le 2$, $\forall r$.
In addition, due to the limited number of RF chains in the considered hybrid systems, we assume that each user can be allocated with at most one RF chain\footnote{Since each user terminal only equips with a single RF chain, allocating more than one RF chains to a user can only provide a limited power gain.}, i.e., $\sum\limits_{r = 1}^{{N_{{\mathrm{RF}}}}} {{u_{k,r}}}  \le 1$, $\forall k$.
The beam splitting technique proposed in this chapter involves antenna allocation within each NOMA group.
Denote $M_{k,r}$ as the number of antennas allocated to user $k$ associated with RF chain $r$, we have $\sum\limits_{k = 1}^K u_{k,r}{M_{k,r}}  \le M_{\mathrm{BS}}$, $\forall r$.

\subsection{Multiple Analog Beams with Beam Splitting}
In the conventional single-beam mmWave-NOMA schemes\cite{Ding2017RandomBeamforming,Cui2017Optimal, WangBeamSpace2017}, there is only a single beam for each NOMA group.
However, as mentioned before, the beamwidth is usually very narrow in mmWave frequency bands and a single beam can rarely cover multiple NOMA users, which restricts the potential performance gain brought by NOMA.
Therefore, we aim to generate multiple analog beams for each NOMA group, wherein each beam is steered to a user within the NOMA group.
To this end, we propose the beam splitting technique, which separates adjacent antennas to form multiple subarrays creating an analog beam via each subarray\footnote{In this chapter, the proposed beam splitting technique splits the original single analog beam into two beams since the NOMA group size is limited as two. For a more general case with a larger NOMA group size, the original analog beam can be split into more beams in a similar way.}.
For instance, in Figure \ref{C6:MultiBeamNOMA:a}, user $k$ and user $i$ are scheduled to be served by RF chain $r$ at the BS, where their allocated number of antennas are ${M_{k,r}}$ and ${M_{i,r}}$, respectively, satisfying ${M_{k,r}} + {M_{i,r}} \le {M_{\mathrm{BS}}}$.
Then, the analog beamformer for the ${M_{k,r}}$ antennas subarray is given by
\begin{equation}\label{C6:SubarayWeight1}
{\mathbf{w}}\left( {M_{k,r}},{\theta _{k,0}} \right)=
{\frac{1}{{\sqrt {{M_{BS}}} }}{\left[ {1, \ldots ,{e^{j\left( {{M_{k,r}} - 1} \right)\pi \cos  {{\theta _{k,0}}} }}} \right]} ^ {\mathrm{T}}},
\end{equation}
and the analog beamformer for the ${M_{i,r}}$ antennas subarray is given by
\begin{equation}\label{C6:SubarayWeight2}
{\mathbf{w}}\left( {M_{i,r},{\theta _{i,0}}} \right) =
{\frac{{e^{j{M_{k,r}}\pi \cos {{\theta _{k,0}}}}}}{\sqrt{M_{BS}}}{\left[ {1, \ldots ,{e^{j\left( {{M_{i,r}} - 1} \right)\pi \cos {{\theta _{i,0}}} }}} \right]} ^ {\mathrm{T}}},
\end{equation}
where $j$ is the imaginary unit, ${\mathbf{w}}\left( M_{k,r},{\theta _{k,0}} \right) \in \mathbb{C}^{ M_{k,r} \times 1}$, and ${\mathbf{w}}\left( M_{i,r},\theta _{i,0} \right) \in \mathbb{C}^{ M_{i,r} \times 1}$.
The same normalized factor $\frac{1}{\sqrt {M_{BS}} }$ is introduced in \eqref{C6:SubarayWeight1} and \eqref{C6:SubarayWeight2} to fulfill the constant modulus constraint \cite{Sohrabi2016} of phase shifters.
As a result, the analog beamformer for RF chain $r$ is given by
\begin{equation}\label{C6:ArrayWeight1}
{\mathbf{w}}_r = \left[ {
{{\mathbf{w}}^{\mathrm{T}}\left( {M_{k,r}},{\theta _{k,0}} \right)}, {{{\mathbf{w}}^{\mathrm{T}}}\left( {M_{i,r}},{\theta _{i,0}} \right)}
} \right]^{\mathrm{T}}.
\end{equation}
Note that the phase shift term ${e^{j{M_{k,r}}\pi \cos  {{\theta _{k,0}}} }}$ in \eqref{C6:SubarayWeight2} is introduced to synchronize the phases between two subarrays.
In particular, when user $k$ and user $i$ have the same AOD, i.e., ${\theta _{k,0}} = {\theta _{i,0}}$, the analog beamformer in \eqref{C6:ArrayWeight1} is degenerated to the single-beam case, i.e., ${\mathbf{w}}_r = {\frac{1}{{\sqrt {{M_{BS}}} }}{\left[ {1, \ldots ,{e^{j\left( {{M_{BS}} - 1} \right)\pi \cos  {{\theta _{k,0}}} }}} \right]} ^ {\mathrm{T}}}$.
Besides, since the phase shift term is imposed on all the elements in the ${M_{i,r}}$ antennas subarray, it does not change the beam pattern response for the ${M_{i,r}}$ antennas subarray.
On the other hand, if user $k'$ is allocated with RF chain $r'$ exclusively, then all the ${M_{\mathrm{BS}}}$ antennas of RF chain $r'$ will be allocated to user $k'$.
In this situation, the analog beamformer for user $k'$ is identical to the conventional analog beamformer in hybrid mmWave systems, i.e., no beam splitting, and it is given by
\begin{equation}\label{C6:ArrayWeight2}
{{\mathbf{w}}_{r'}} = {\mathbf{w}}\left( {{M_{\mathrm{BS}}},{\theta _{k',0}}} \right) =
{\frac{1}{{\sqrt {{M_{BS}}} }}\left[ {1, \ldots ,{e^{j{\left({M_{\mathrm{BS}}} - 1\right)}\pi \cos {{\theta _{k',0}}} }}} \right] ^ {\mathrm{T}}}.
\end{equation}
Note that, compared to the single-beam mmWave-NOMA schemes\cite{Ding2017RandomBeamforming,Cui2017Optimal, WangBeamSpace2017}, the AODs of the LOS paths ${\theta _{k,0}}$ and ${\theta _{i,0}}$ in the proposed scheme are not required to be in the same analog beam.
In other words, the proposed multi-beam NOMA scheme provides a higher flexibility in user grouping than that of the single-beam NOMA schemes.

\begin{figure}[t]
\centering
\includegraphics[width=4.5in]{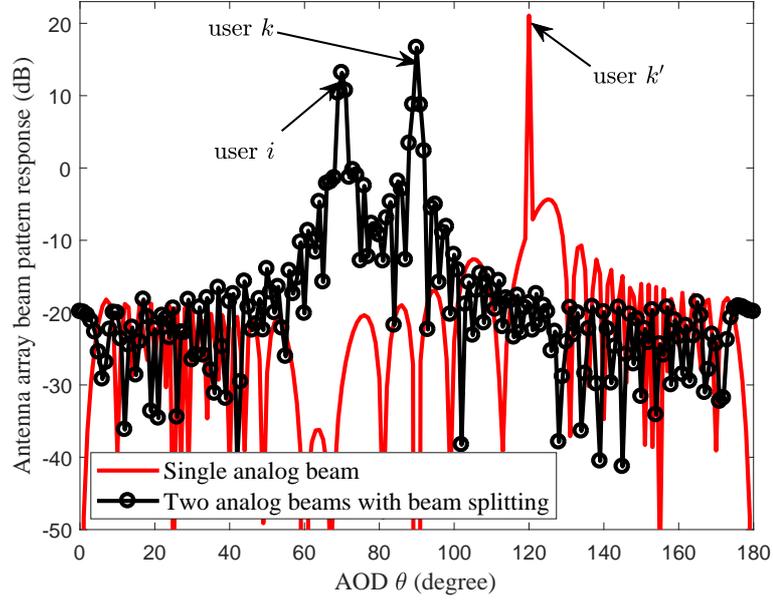}
\caption{Antenna array beam pattern response for ${{\bf{w}}_r}$ in \eqref{C6:ArrayWeight1} and ${{\bf{w}}_r'}$ in \eqref{C6:ArrayWeight2} at the BS via the proposed beam splitting technique. We assume ${M_{\mathrm{BS}}} = 128$, ${M_{i,r}} = 50$, ${M_{k,r}} = 78$, ${\theta _i} = 70^{\circ}$, ${\theta _k} = 90^{\circ}$, and ${\theta _{k'}} = 120^{\circ}$.}
\label{C6:MultiBeamNOMA2}
\end{figure}

Based on the analog beamformers ${{\mathbf{w}}_r}$ and ${{\mathbf{w}}_{r'}}$, RF chain $r$ generates two analog beams steering toward user $k$ and user $i$, respectively, while RF chain $r'$  forms a single analog beam steering to user $k'$.
The antenna array beam pattern responses for ${{\mathbf{w}}_r}$ and ${{\mathbf{w}}_{r'}}$ are shown in Figure \ref{C6:MultiBeamNOMA2} to illustrate the multiple analog beams generated via beam splitting.
Compared to the single analog beam for user $k'$, we can observe that the magnitude of the main beam response decreases and the beamwidth increases for both the two analog beams for users $k$ and $i$.
Although forming multiple analog beams via beam splitting sacrifices some beamforming gain of the original single analog beam, it can still improve the system performance via accommodating more users on each RF chain.

Now, we integrate the users scheduling variables ${u_{k,r}}$ with ${{\mathbf{w}}_r}$ as follows
\begin{equation}\label{C6:Weight1}
{{\mathbf{w}}_r} = {\left[
{{{\mathbf{w}}^{\mathrm{T}}}\left( {{u_{1,r}},{M_{1,r}},{\theta _{1,0}}} \right)}, \ldots ,{{{\mathbf{w}}^{\mathrm{T}}}\left( {{u_{K,r}},{M_{K,r}},{\theta _{K,0}}} \right)}
\right]^{\mathrm{T}}},
\end{equation}
with
\begin{equation}\label{C6:Weight2}
{\mathbf{w}}\left( {{u_{k,r}},{M_{k,r}},{\theta _{k,0}}} \right) = \left\{ {\begin{array}{*{20}{c}}
\emptyset, &{{u_{k,r}} = 0},\\
{{{{e^{j\sum\limits_{d = 1}^{k - 1} {{u_{d,r}}{M_{d,r}}\pi \cos \left( {{\theta _{d,0}}} \right)} }}}}{\mathbf{w}}\left( {{M_{k,r}},{\theta _{k,0}}} \right)}, &{{u_{k,r}} = 1}.
\end{array}} \right.
\end{equation}
It can be observed in \eqref{C6:Weight2} that ${\mathbf{w}}\left( {{u_{k,r}},{M_{k,r}},{\theta _{k,0}}} \right)$ is an empty set $\emptyset$ when ${{u_{k,r}} = 0}$, $\forall k$, and ${{\mathbf{w}}_r}$ consists of the analog beamformers for the users allocated with RF chain $r$, i.e., ${{u_{k,r}} = 1}$, $\forall k$.

\subsection{Effective Channel Estimation}
For a given user grouping strategy, antenna allocation, and multiple analog beamforming, all the users transmit their unique orthogonal pilots to the BS in the uplink to perform effective channel estimation.
In this chapter, we assume the use of time division duplex (TDD) and exploit the channel reciprocity, i.e., the estimated effective channel in the uplink can be used for digital precoder design in the downlink.
The effective channel of user $k$ on RF chain $r$ at the BS is given by
\begin{equation}\label{C6:EffectiveChannel1}
{\widetilde h_{k,r}} = {{\mathbf{v}}_k^{\mathrm{H}}}{{\mathbf{H}}_{k}}{\mathbf{w}}_r,\;\forall k,r,
\end{equation}
where ${{\mathbf{v}}_k}$ and ${\mathbf{w}}_r$ denote the analog beamformers adopted at user $k$ and the RF chain $r$ at the BS, respectively\footnote{In the proposed multi-beam NOMA framework, any existing channel estimation scheme can be used to estimate the effective channel in \eqref{C6:EffectiveChannel1}. For illustration, we adopted the strongest LOS based channel estimation scheme in \cite{AlkhateebPrecoder2015,zhao2017multiuser}.
In particular, we set the receiving analog beamformer as ${{\mathbf{v}}_k} = \frac{1}{{\sqrt {{M_{{\mathrm{UE}}}}} }}{{\mathbf{a}}_{{\mathrm{UE}}}}\left( {{\phi _{k,0}}} \right)$ and transmitting analog beamformer of RF chain $r$ ${\mathbf{w}}_r$ given by \eqref{C6:Weight1}, and utilize the minimum mean square error (MMSE) channel estimation to estimate the effective channels of all the users.}.
In the following, we denote the effective channel vector of user $k$ as ${{{\bf{\widetilde h}}}_k} = {\left[ {
{{{\widetilde h}_{k,1}}}, \ldots ,{{{\widetilde h}_{k,{N_{{\mathrm{RF}}}}}}}
} \right]}^{\mathrm{T}} \in \mathbb{C}^{{ N_{\mathrm{RF}} \times 1}}$ and denote the effective channel matrix between the BS and the $K$ users as ${\bf{\widetilde H}} = {\left[ {
{\widetilde{{\bf{h}}}_1}, \ldots ,{\widetilde{{\bf{h}}}_K}
} \right]} \in \mathbb{C}^{{ N_{\mathrm{RF}}\times K }}$.

\subsection{Digital Precoder and Power Allocation Design}
Given the estimated effective channel matrix ${{\mathbf{\widetilde H}}}$, the digital precoder and the power allocation can be designed accordingly.
With the proposed user grouping design, there are totally $N_{\mathrm{RF}}$ NOMA groups to be served in the proposed multi-beam NOMA scheme and the users within each NOMA group share the same digital precoder.
Assuming that the adopted digital precoder is denoted as ${\mathbf{G}} = \left[ {
{{{\mathbf{g}}_1}},\ldots ,{{{\mathbf{g}}_{{N_{{\mathrm{RF}}}}}}}
} \right] \in \mathbb{C}^{{ N_{\mathrm{RF}} \times N_{\mathrm{RF}}}}$, where ${{\mathbf{g}}_r}$ with ${\left\| {{{\bf{g}}_r}} \right\|^2} = 1$ denotes the digital precoder for the NOMA group associated\footnote{Note that the concept of RF chain association is more clear for a pure analog mmWave system, i.e., ${\mathbf{G}} = {{\mathbf{I}}_{{N_{{\mathrm{RF}}}}}}$, where the signal of a NOMA group or an OMA user is transmitted through its associated RF chain, as shown in Figure \ref{C6:MultiBeamNOMA:a}.
In hybrid mmWave systems with a digital precoder, ${\mathbf{G}}$, the signals of NOMA groups and OMA users are multiplexed on all the RF chains.
In this case, the RF chain allocation is essentially the spatial DoF allocation, where a NOMA group or an OMA user possesses one spatial DoF.
However, we still name it with RF chain association since each associated RF chain generates multiple analog beams for a NOMA group or a single beam for an OMA user, as shown in Figure \ref{C6:MultiBeamNOMA:a}.} with RF chain $r$.
In addition, denoting the power allocation for user $k$ associated with RF chain $r$ as $p_{k,r}$, we have the sum power constraint $ \sum_{k = 1}^K \sum_{r = 1}^{{N_{{\mathrm{RF}}}}} {{u_{k,r}}{p_{k,r}}}  \le p_{\mathrm{BS}}$, where $p_{\mathrm{BS}}$ is the power budget for the BS.
Then, the received signal at user $k$ is given by
\begin{align}\label{C6:DLRx1}
{y_k}=& {{{\mathbf{\widetilde h}}}_k^{\mathrm{H}}}{{\mathbf{G}}\mathbf{t}} + {z_k} = {{{\mathbf{\widetilde h}}}_k^{\mathrm{H}}}\sum\limits_{r = 1}^{{N_{{\mathrm{RF}}}}} {{{\mathbf{g}}_r}{t_r}}  + {z_k} \\
=& \underbrace{{{{\mathbf{\widetilde h}}}_k^{\mathrm{H}}} {{\mathbf{g}}_r} \sqrt {{p_{k,r}}} {s_k}}_{\mathrm{Desired\;signal}} + \underbrace{{{{\mathbf{\widetilde h}}}_k^{\mathrm{H}}} {{\mathbf{g}}_r}\sum\limits_{d \ne k}^K {{u_{d,r}}\sqrt {{p_{d,r}}} {s_{d}}}}_{\mathrm{Intra-group\;interference}} + \underbrace{{{{\mathbf{\widetilde h}}}_k^{\mathrm{H}}}\sum\limits_{r' \ne r}^{{N_{{\mathrm{RF}}}}} {{{\mathbf{g}}_{r'}}\sum\limits_{d = 1}^K {{u_{d,r'}}} \sqrt {{p_{d,r'}}} {s_{d}}}}_{\mathrm{Inter-group\;interference}}  + {z_k},\notag
\end{align}
where ${t_r} = \sum\limits_{k = 1}^K {{u_{k,r}}} \sqrt{p_{k,r}}{s_k}$ denotes the superimposed signal of the NOMA group associated with RF chain $r$ and ${\mathbf{t}} = {\left[
{{t_1}}, \ldots ,{{t_{{N_{{\mathrm{RF}}}}}}}
\right]^{\mathrm{T}}} \in \mathbb{C}^{{ N_{\mathrm{RF}} \times 1}}$.
Variable ${s_k} \in \mathbb{C}$ denotes the modulated symbol for user $k$ and $z_k \sim {\cal CN}(0,\sigma^2)$ is the additive white Gaussian noise (AWGN) at user $k$, where $\sigma^2$ is the noise power.
For instance, in Figure \ref{C6:MultiBeamNOMA:a}, if user $k$ and user $i$ are allocated to RF chain $r$ and user $k'$ is allocated to RF chain $r'$, we have ${t_r} = \sqrt{p_{k,r}}{s_k} + \sqrt{p_{i,r}}{s_i}$ and ${t_{r'}} = \sqrt{p_{k',r'}}{s_{k'}}$.
In \eqref{C6:DLRx1}, the first term represents the desired signal of user $k$, the second term denotes the intra-group interference caused by the other users within the NOMA group associated with RF chain $r$, and the third term is the inter-group interference originated from all the other RF chains.

\subsection{SIC Decoding at Users}
At the user side, as the traditional downlink NOMA schemes\cite{WeiTCOM2017}, SIC decoding is performed at the strong user within one NOMA group, while the weak user directly decodes the messages by treating the strong user's signal as noise.
In this chapter, we define the strong or weak user by the LOS path gain.
Without loss of generality, we assume that the users are indexed in the descending order of LOS path gains, i.e., ${\left| {{\alpha _{1,0}}} \right|^2} \ge {\left| {{\alpha _{2,0}}} \right|^2} \ge , \ldots , \ge {\left| {{\alpha _{K,0}}} \right|^2}$.

According to the downlink NOMA protocol\cite{WeiTCOM2017}, the individual data rate of user $k$ when associated with RF chain $r$ is given by
\begin{equation}\label{C6:DLIndividualRate1}
{R_{k,r}} = {\log _2}\left( {1 + \frac{{{u_{k,r}}{p_{k,r}}{{\left| {{{\mathbf{\widetilde h}}}_k^{\mathrm{H}}}{{\mathbf{g}}_r} \right|}^2}}}{{I_{k,r}^{{{\mathrm{ inter}}}} + I_{k,r}^{{{\mathrm{intra}}}} + \sigma^2}}} \right),
\end{equation}
with $I_{k,r}^{{{\mathrm{inter}}}} = \sum\limits_{r' \ne r}^{{N_{{\mathrm{RF}}}}} {{{ \left| { {{{\mathbf{\widetilde h}}}_k^{\mathrm{H}}}{{\mathbf{g}}_{r'}} } \right|}^2}\sum\limits_{d = 1}^K {{u_{d,r'}}{p_{d,r'}}} }$ and $I_{k,r}^{{{\mathrm{intra}}}} = {{\left| {{{\mathbf{\widetilde h}}}_k^{\mathrm{H}}}{{\mathbf{g}}_r} \right|}^2}\sum\limits_{d = 1}^{k-1} {{u_{d,r}}{p_{d,r}}}$
denoting the inter-group interference power and intra-group interference power, respectively.
Note that with the formulation in \eqref{C6:DLIndividualRate1}, we have ${R_{k,r}} = 0$ if ${u_{k,r}} = 0$.
If user $k$ and user $i$ are scheduled to form a NOMA group associated with RF chain $r$, $\forall i > k$, user $k$ first decodes the messages of user $i$ before decoding its own information and the corresponding achievable data rate is given by
\begin{equation}\label{C6:IndividualRate2}
R_{k,i,r} = {\log _2}\left( {1 + \frac{{{u_{i,r}}{p_{i,r}}{{\left| {{{\mathbf{\widetilde h}}}_k^{\mathrm{H}}}{{\mathbf{g}}_r} \right|}^2}}}{{I_{k,r}^{{\mathrm{inter}}} + I_{k,i,r}^{{\mathrm{intra}}} + \sigma^2}}} \right),
\end{equation}
where $I_{k,i,r}^{{\mathrm{intra}}} = {{\left| {{{\mathbf{\widetilde h}}}_k^{\mathrm{H}}}{{\mathbf{g}}_r} \right|}^2}\sum\limits_{d = 1}^{i - 1} {{u_{d,r}}{p_{d,r}}}$ denotes the intra-group interference power when decoding the message of user $i$ at user $k$.
To guarantee the success of SIC decoding, we need to maintain the rate condition as follows \cite{Sun2016Fullduplex}:
\begin{equation}\label{C6:DLSICDecoding}
R_{k,i,r} \ge {R_{i,r}}, \forall i > k.
\end{equation}
Note that, when user $i$ is not allocated with RF chain $r$, we have $R_{k,i,r} = {R_{i,r}} =0$ and the condition in \eqref{C6:DLSICDecoding} is always satisfied.
Now, the individual data rate of user $k$ is defined as $R_{k} = {\sum\limits_{r = 1}^{{N_{\mathrm{RF}}}} {{R_{k,r}}}}$, $\forall k$, and the system sum-rate is given by
\begin{equation}\label{C6:SumRate}
{R_{{\mathrm{sum}}}} = \sum\limits_{k = 1}^K {\sum\limits_{r = 1}^{{N_{RF}}} {{{\log }_2}\left( {1 + \frac{{{u_{k,r}}{p_{k,r}}{{\left| {{{\mathbf{\widetilde h}}}_k^{\mathrm{H}}}{{\mathbf{g}}_r} \right|}^2}}}{{I_{k,r}^{{{\mathrm{inter}}}} + I_{k,r}^{{{\mathrm{intra}}}} + \sigma^2}}} \right)} }.
\end{equation}

\begin{remark}
In summary, the key idea of the proposed multi-beam NOMA framework\footnote{Note that this chapter proposes a multi-beam NOMA framework for hybrid mmWave communication systems, where different analog beamformer designs, channel estimation methods, and digital precoder designs can be utilized in the proposed framework.} is to multiplex the signals of multiple users on a single RF chain via beam splitting, which generates multiple analog beams to facilitate non-orthogonal transmission for multiple users.
Intuitively, compared to the natural broadcast in conventional NOMA schemes considered in microwave frequency bands\cite{WeiSurvey2016,Ding2015b}, the proposed multi-beam NOMA scheme generates multiple non-overlapped virtual tunnels in the beam-domain and broadcast within the tunnels for downlink NOMA transmission.
It is worth to note that the beam splitting technique is essentially an allocation method of array beamforming gain.
In particular, allocating more antennas to a user means allocating a higher beamforming gain for this user, and vice versa.
Specifically, apart from the power domain multiplexing in conventional NOMA schemes\cite{WeiSurvey2016,Ding2015b}, the proposed multi-beam NOMA scheme further exploits the beam-domain for efficient multi-user multiplexing.
Besides, the proposed multi-beam NOMA scheme only relies on AODs of the LOS paths, ${\theta _{k,0}}$, and the complex path gains, ${{\alpha _{k,0}}}$, which is different from the existing fully digital MIMO-NOMA schemes\cite{Hanif2016}.
Clearly, the performance of the proposed multi-beam mmWave-NOMA scheme highly depends on the user grouping, antenna allocation, power allocation, and digital precoder design, which will be detailed in the next section.
\end{remark}

\section{Resource Allocation Design}
In this section, we focus on resource allocation design for the proposed multi-beam mmWave-NOMA scheme.
As shown in Figure \ref{C6:MultiBeamNOMAScheme}, the effective channel seen by the digital precoder depends on the structure of analog beamformers, which is determined by the user grouping and antenna allocation.
In other words, the acquisition of effective channel is coupled with the user grouping and antenna allocation.
In fact, this is fundamentally different from the resource allocation design of the fully digital MIMO-NOMA schemes\cite{Sun2017MIMONOMA,Hanif2016} and single-beam mmWave-NOMA schemes\cite{Ding2017RandomBeamforming,Cui2017Optimal, WangBeamSpace2017}.
As a result, jointly designing the user grouping, antenna allocation, power allocation, and digital precoder for our proposed scheme is very challenging and generally intractable.
To obtain a computationally efficient design, we propose a two-stage design method, which is commonly adopted for the hybrid precoder design in the literature\cite{AlkhateebPrecoder2015,Mumtaz2016mmwave}.
Specifically, in the first stage, we design the user grouping and antenna allocation based on the coalition formation game theory \cite{SaadCoalitionalGame,SaadCoalitional2012,WangCoalitionNOMA,SaadCoalitionOrder,Han2012} to maximize the conditional system sum-rate assuming that only analog beamforming is available.
In the second stage, based on the obtained user grouping and antenna allocation strategy, we adopt a ZF digital precoder to manage the inter-group interference and formulate the power allocation design as an optimization problem to maximize the system sum-rate while taking into account the QoS constraints.

\subsection{First Stage: User Grouping and Antenna allocation}

\noindent\textbf{\underline{{Problem Formulation}}}

In the first stage, we assume that only analog beamforming is available with ${\mathbf{G}} = {{\mathbf{I}}_{{N_{{\mathrm{RF}}}}}}$ and equal power allocation ${p_{k,r}} = \frac{p_{\mathrm{BS}}}{K}$, $\forall k,r$, i.e., each RF chain serves its associated NOMA group correspondingly.
The joint user grouping and antenna allocation to maximize the achievable sum-rate can be formulated as the following optimization problem:
\begin{align} \label{C6:ResourceAllocation}
&\underset{{u_{k,r}},{M_{k,r}}}{\maxo} \;\;\overline{R}_{\mathrm{sum}} = \sum\limits_{k = 1}^K \sum\limits_{r = 1}^{{N_{RF}}} {\overline{R}_{k,r}} \\
\mbox{s.t.}\;\;
%%%%%
&\mbox{C1: } {u_{k,r}} \in \left\{ {0,1} \right\}, {M_{k,r}} \in \mathbb{Z}{^ + }, \forall k,r,  \notag\\
&\mbox{C2: } \sum\limits_{k = 1}^K {{u_{k,r}}}  \le 2, \forall r, \;\;\mbox{C3: } \sum\limits_{r = 1}^{{N_{{\mathrm{RF}}}}} {{u_{k,r}}}  \le 1, \forall k, \notag\\
&\mbox{C4: } \sum\limits_{k = 1}^K {{u_{k,r}}{M_{k,r}}}  \le {M_{{\mathrm{BS}}}}, \forall r, \notag\\
&\mbox{C5: } {u_{k,r}}{M_{k,r}} \ge {u_{k,r}}{M_{\min }}, \forall k,r, \notag
\end{align}
where the user scheduling variable ${u_{k,r}}$ and antenna allocation variable ${M_{k,r}}$ are the optimization variables.
The objective function $\overline{R}_{\mathrm{sum}}$ denotes the conditional system sum-rate which can be given with ${R}_{\mathrm{sum}}$ in \eqref{C6:SumRate} by substituting ${p_{k,r}} = \frac{p_{\mathrm{BS}}}{K}$ and ${\mathbf{G}} = {{\mathbf{I}}_{{N_{{\mathrm{RF}}}}}}$.
Similarly, ${\overline{R}_{k,r}}$ denotes the conditional individual data rate of user $k$ associated with RF chain $r$ which can be obtained with ${{R}_{k,r}}$ in \eqref{C6:DLIndividualRate1} by substituting ${p_{k,r}} = \frac{p_{\mathrm{BS}}}{K}$ and ${\mathbf{G}} = {{\mathbf{I}}_{{N_{{\mathrm{RF}}}}}}$.
Constraint C2 restricts that each RF chain can only serve at most two users.
Constraint C3 is imposed such that one user can only be associated with at most one RF chain.
Constraint C4 guarantees that the number of all the allocated antennas on RF chain $r$ cannot be larger than $M_{\mathrm{BS}}$.
Constraint C5 is imposed to keep that the minimum number of antennas allocated to user $k$ is ${M_{\mathrm{min}}}$ if it is associated with RF chain $r$.
The reasons for introducing constraint C5 is two-fold: 1) we need to prevent the case where a user is not served by any antenna for a fair resource allocation; 2) constraint C5 can prevent the formation of high sidelobe beam, which introduces a high inter-group interference and might sacrifice the performance gain of multi-beam NOMA.
Note that constraint C5 are inactive when user $k$ is not assigned on RF chain $r$, i.e., $u_{k,r} = 0$.

The formulated problem is a non-convex integer programming problem, which is very difficult to find the globally optimal solution.
In particular, the effective channel gain in the objective function in \eqref{C6:ResourceAllocation} involves a periodic trigonometric function of ${M_{k,r}}$, which cannot be solved efficiently with existing convex optimization methods\cite{Boyd2004}.
In general, the exhaustive search can optimally solve this problem, but its computational complexity is prohibitively high, which is given by
\begin{equation}\label{C6:Complexity}
\mathcal{O}\left(({M_{{\mathrm{BS}}} - 2{M_{\min }}})^{{K - {N_{{\mathrm{RF}}}}}}\left( {\begin{array}{*{20}{c}}
K\\
{2{N_{\mathrm{RF}}} -K}
\end{array}} \right)\mathop \Pi \limits_{r = 1}^{K - {N_{{\mathrm{RF}}}}} \left( {2r - 1} \right)\right),
\end{equation}
where $\mathcal{O}\left(\cdot\right)$ is the big-O notation.
As a compromise solution, we recast the formulated problem as a coalition formation game \cite{SaadCoalitionalGame,SaadCoalitional2012,WangCoalitionNOMA,SaadCoalitionOrder,Han2012} and propose a coalition formation algorithm for obtaining an efficient suboptimal solution of the user grouping and antenna allocation.

\noindent\textbf{\underline{{The Proposed Coalition Formation Algorithm}}}

In this section, to derive an algorithm for the user grouping and antenna allocation with a low computational complexity, we borrow the concept from the coalitional game theory \cite{SaadCoalitionalGame,SaadCoalitional2012,WangCoalitionNOMA,SaadCoalitionOrder,Han2012} to achieve a suboptimal but effective solution.
Note that, although the game theory is commonly introduced to deal with the distributed resource allocation design problem\cite{SaadCoalitional2012}, it is also applicable to the scenario with a centralized utility \cite{SaadCoalitionOrder,WangCoalitionNOMA,Han2012}.
Besides, it is expected to achieve a better performance with centralized utility compared to the distributed resource allocation due to the availability of the overall resource allocation strategy and the system performance.

\textbf{Basic concepts:} We first introduce the notions from the coalitional game theory \cite{SaadCoalitionalGame,SaadCoalitional2012,WangCoalitionNOMA,SaadCoalitionOrder,Han2012} to reformulate the problem in \eqref{C6:ResourceAllocation}.
In particular, we view all the $K$ users as a set of cooperative \emph{players}, denoted by $\mathcal{K} = \{1,\ldots,K\}$, who seek to form \emph{coalitions} $S$ (NOMA groups in this chapter), $S \subseteq \mathcal{K} $, to maximize the conditional system sum-rate.
The \emph{coalition value}, denoted by ${V} \left(S\right)$, quantifies the payoff of a coalition in a game.
In this chapter, we characterize the payoff of each player as its individual conditional data rate, since we aim to maximize the conditional system sum-rate in \eqref{C6:ResourceAllocation}.
In addition, for our considered problem, since the payoff of each player in a coalition $S$ depends on the antenna allocation within the coalition and thus cannot be divided in any manner between the coalition members, our considered game has a nontransferable utility (NTU) property \cite{Han2012}.
Therefore, the coalition value is a set of payoff vectors, ${V} \left(S\right) \subseteq \mathbb{R}^{\left|S\right|}$, where each element represents the payoff of each player in $S$.
Furthermore, due to the existence of inter-group interference in our considered problem, the coalition value ${V} \left(S\right)$ depends not only on its own coalition members in $S$, but also on how the other players in $\mathcal{K}\backslash S$ are structured.
As a result, the considered game falls into the category of \emph{coalition game in partition form}, i.e., coalition formation game\cite{SaadCoalitionalGame}.

\textbf{Coalition value:} More specifically, given a coalitional structure $\mathcal{B}$, defined as a partition of $\mathcal{K}$, i.e., a collection of coalitions $\mathcal{B} = \{S_1,\ldots,S_{\left|\mathcal{B}\right|}\}$, such that $\forall r \ne r'$, $S_{r} \bigcap S_{r'} = \emptyset$, and $\cup_{r = 1}^{\left|\mathcal{B}\right|} S_{r} = \mathcal{K}$, the value of a coalition $S_{r}$ is defined as ${V} \left(S_{r},\mathcal{B}\right)$.
Moreover, based on \eqref{C6:DLIndividualRate1}, we can observe that ${\overline{R}_{k,r}}$ is affected by not only the antenna allocation strategy $\mathcal{M}_r = \{M_{k,r} | \forall k \in S_{r}\}$ in the coalition $S_{r}$, but also the antenna allocation strategy $\mathcal{M}_{r'} = \{M_{k,{r'}} | \forall k \in S_{{r'}}\}$ in other coalitions $S_{r'}$, $r' \neq r$.
Therefore, the coalition value is also a function of the antenna allocation strategy $\mathcal{M}_r$ of the coalition $S_{r}$ and the overall antenna allocation strategy $\Pi = \{\mathcal{M}_1,\ldots,\mathcal{M}_{\left|\mathcal{B}\right|}\}$, i.e., ${V} \left(S_{r},\mathcal{M}_r,\mathcal{B},\Pi\right)$.
Then, the coalition value of $S_{r}$ in place of $\mathcal{B}$ can be defined as
\begin{equation}
{V} \left(S_{r},\mathcal{M}_r,\mathcal{B},\Pi\right) = \{v_k\left(S_{r},\mathcal{M}_r,\mathcal{B},\Pi\right) | \forall k \in S_{r} \},
\end{equation}
with the payoff of user $k$ as
\begin{equation}\label{C6:PayoffUserk}
v_k \left(S_{r},\mathcal{M}_r,\mathcal{B},\Pi\right) = \left\{ {\begin{array}{*{20}{c}}
\frac{{N_{\mathrm{RF}}}}{{\left|\mathcal{B}\right|}}{\overline{R}_{k,r}},&{{\left|\mathcal{B}\right|} > {N_{\mathrm{RF}}}},\\
{\overline{R}_{k,r}},&{{\left|\mathcal{B}\right|} \le {N_{\mathrm{RF}}}},
\end{array}} \right.
\end{equation}
where $\frac{{N_{\mathrm{RF}}}}{{\left|\mathcal{B}\right|}}$ is the time-sharing factor since the number of coalitions is larger than the number of system RF chains.
To facilitate the solution design, based on \eqref{C6:DLIndividualRate1}, the conditional individual data rate of user $k$ associated with RF chain $r$, ${\overline{R}_{k,r}}$, can be rewritten with the notations in coalition game theory as follows
\begin{equation}\label{C6:DLIndividualRate1Game}
{\overline{R}_{k,r}} = \left\{ {\begin{array}{*{20}{c}}
{\log _2}\left( {1 + {\frac{{{p_{k,r}}{{\left| {{{\mathbf{\widetilde h}}}_k^{\mathrm{H}}}{{\mathbf{g}}_r} \right|}^2}}}{{I_{k,r}^{{{\mathrm{inter}}}}  + I_{k,r}^{{{\mathrm{intra}}}} + {\sigma ^2}}}}} \right), &{k \in S_{r}},\\
0,&{\mathrm{otherwise}},
\end{array}} \right.
\end{equation}
where $I_{k,r}^{{{\mathrm{inter}}}} = \sum\limits_{r' \ne r, S_{r'} \in \mathcal{B}} {{{\left| {{{\mathbf{\widetilde h}}}_k^{\mathrm{H}}}{{\mathbf{g}}_{r'}} \right|}^2}\sum\limits_{d \in S_{r'}} {{p_{d,r}}}}$ and $I_{k,r}^{{{\mathrm{intra}}}} = {{\left| {{{\mathbf{\widetilde h}}}_k^{\mathrm{H}}}{{\mathbf{g}}_r} \right|}^2}\sum\limits_{d < k, d \in S_{r}} {{p_{d,r}}}$ with ${p_{k,r}} = \frac{p_{\mathrm{BS}}}{K}$ and ${\mathbf{G}} = {{\mathbf{I}}_{{N_{{\mathrm{RF}}}}}}$.

Based on the aforementioned coalitional game theory notions, the proposed coalition formation algorithm essentially forms coalitions among players iteratively to improve the system performance.
To facilitate the presentation of the proposed coalition formation algorithm, we introduce the following basic concepts.

\begin{Def}
Given two partitions $\mathcal{B}_1$ and $\mathcal{B}_2$ of the set of users $\mathcal{K}$ with two antenna allocation strategies $\Pi_1$ and $\Pi_2$, respectively, for two coalitions ${S_r} \in \mathcal{B}_1$ and ${S_{r'}} \in \mathcal{B}_2$ including user $k$, i.e., $k \in {S_r}$ and $k \in {S_{r'}}$, the \emph{preference relationship} for user $k \in \mathcal{K}$ is defined as
\begin{equation}\label{C6:PreferenceRelationship}
\left( {{S_r},\mathcal{B}_1} \right) { \preceq _k} \left( {{S_{r'}},\mathcal{B}_2} \right)
\Leftrightarrow \mathop {\max }\limits_{{\mathcal{M}_r}}   {\overline{R}_{{\mathrm{sum}}}}\left( {{\cal B}_1,\Pi_1} \right) \le \mathop {\max } \limits_{{\mathcal{M}_{r'}}} {\overline{R}_{{\mathrm{sum}}}} \left( {{\cal B}_2,\Pi_2} \right),
\end{equation}
with
\begin{equation}\label{C6:SystemSumRateCoalitionGame}
{\overline{R}_{{\mathrm{sum}}}} \left( {{\cal B},\Pi} \right) = \sum\limits_{{S_r} \in {\cal B}} {\sum\limits_{k \in {S_r}} {{v_k}\left( {{S_r},{\mathcal{M}_r},{\cal B},\Pi} \right)} }.
\end{equation}
The notation $\left( {{S_r},\mathcal{B}_1} \right){ \preceq _k} \left( {{S_{r'}},\mathcal{B}_2} \right)$ means that user $k$ prefers to be part of coalition ${S_{r'}}$ when $\mathcal{B}_2$ is in place, over being part of coalition ${S_{r}}$ when $\mathcal{B}_1$ is in place, or at least prefers both pairs of coalitions and partitions equally.
As shown in \eqref{C6:PreferenceRelationship} and \eqref{C6:SystemSumRateCoalitionGame}, we have $\left( {{S_r},\mathcal{B}_1} \right){ \preceq _k} \left( {{S_{r'}},\mathcal{B}_2} \right)$ if and only if the resulting conditional system sum-rate with $\mathcal{B}_2$ and $\Pi_2$ is larger than or at least equal to that with $\mathcal{B}_1$ and $\Pi_1$, with the optimized antenna allocation over ${\mathcal{M}_r}$ and ${\mathcal{M}_{r'}}$, respectively.
Furthermore, we denote their asymmetric counterpart ${ \prec_k}$ and $<$ to indicate the \emph{strictly preference relationship}.
\end{Def}

Note that, the maximization $\mathop {\max }  \limits_{{\mathcal{M}_r}}   {\overline{R}_{{\mathrm{sum}}}}\left( {{\cal B}_1,\Pi_1} \right)$ in \eqref{C6:PreferenceRelationship} is only over ${\mathcal{M}_r}$ for the involved coalition ${S_r}$ via fixing all the other antenna allocation strategies in $\{\Pi_1 \backslash {\mathcal{M}_r}\}$ when $\mathcal{B}_1$ is in place.
Similarly, the maximization $\mathop {\max }  \limits_{{\mathcal{M}_{r'}}}   {\overline{R}_{{\mathrm{sum}}}}\left( {{\cal B}_2,\Pi_2} \right)$ in \eqref{C6:PreferenceRelationship} is only over ${\mathcal{M}_{r'}}$ for the involved coalition ${S_{r'}}$ via fixing all the other antenna allocation strategies in $\{\Pi_2 \backslash {\mathcal{M}_{r'}}\}$ when $\mathcal{B}_2$ is in place.
Let ${\mathcal{M}_r^{*}} = \mathop {\arg } \max \limits_{{\mathcal{M}_r}}   {\overline{R}_{{\mathrm{sum}}}}\left( {{\cal B}_1,\Pi_1} \right)$ and ${\mathcal{M}_{r'}^{*}} = \mathop {\arg } \max\limits_{{\mathcal{M}_{r'}}} {\overline{R}_{{\mathrm{sum}}}} \left( {{\cal B}_2,\Pi_2} \right)$ denote the optimal solutions for the maximization problems in both sides of  \eqref{C6:PreferenceRelationship}.
To find the optimal ${\mathcal{M}_r^{*}}$ and ${\mathcal{M}_{r'}^{*}}$ in \eqref{C6:PreferenceRelationship}, we use the full search method over ${\mathcal{M}_r}$ and ${\mathcal{M}_{r'}}$ within the feasible set of C4 and C5 in \eqref{C6:ResourceAllocation}.
Recall that there are at most two members in each coalition, the computational complexity for solving the antenna allocation in each maximization problem is acceptable via a one-dimensional full search over the integers in set $[{M_{\min }},{M_{{\mathrm{BS}}} - {M_{\min }}}]$.
Based on the defined preference relationship above, we can define the \emph{strictly preferred leaving and joining operation} in the $\mathrm{iter}$-th iteration as follows.
\begin{Def}\label{C6:LeaveAndJoin}
In the $\mathrm{iter}$-th iteration, given the current partition $\mathcal{B}_{\mathrm{iter}} = \{S_1,\ldots,S_{\left|\mathcal{B}_{\mathrm{iter}}\right|}\}$ and its antenna allocation strategy $\Pi_{\mathrm{iter}} = \{\mathcal{M}_1,\ldots,\mathcal{M}_{\left|\mathcal{B}_{\mathrm{iter}}\right|}\}$ of the set of users $\mathcal{K}$, user $k$ will \emph{leave} its current coalition ${S_{r}}$ and \emph{joint} another coalition ${S_{r'}} \in \mathcal{B}_{\mathrm{iter}}$, $r' \neq r$, if and only if it is a \emph{strictly preferred leaving and joining operation}, i.e.,
\begin{equation}\label{C6:StrictPreferredLaJ}
\left( {{S_r},\mathcal{B}_{\mathrm{iter}}} \right) { \prec _k} \left( {{S_{r'}}  \cup  \{k\} ,\mathcal{B}_{\mathrm{iter}+1}} \right),
\end{equation}
where new formed partition is $\mathcal{B}_{\mathrm{iter}+1} = \left\{ {\mathcal{B}_{\mathrm{iter}}\backslash \left\{ {{S_r},{S_{r'}}} \right\}} \right\} \cup \left\{ {{S_r}\backslash \left\{ k \right\},{S_{r'}} \cup \left\{ k \right\}} \right\}$.
Let ${{\cal M}_r^\circ}$ denote the optimal antenna allocation strategy for the coalition $\{{S_r}\backslash \left\{ k \right\}\}$ if user $k$ leaves the coalition ${S_{r}}$.
Due to ${\left|S_r\right|} \le 2$, there are at most one member left in $\{{S_r}\backslash \left\{ k \right\}\}$, and thus we have
\begin{equation}
{{\cal M}_r^\circ} = \left\{ {\begin{array}{*{20}{c}}
M_{\mathrm{BS}}, & \left|{S_r}\backslash \left\{ k \right\}\right| = 1,\\
0,           & \left|{S_r}\backslash \left\{ k \right\}\right| = 0.
\end{array}} \right.
\end{equation}
On the other hand, we denote ${\cal M}_{r'}^*$ as the optimal antenna allocation to maximize the conditional system sum-rate for the new formed coalition ${S_{r'}} \cup \left\{ k \right\}$ if user $k$ joins the coalition ${S_{r'}}$.
Similarly, we note that the maximization is only over ${\cal M}_{r'}$ via fixing all the others antenna strategies $\left\{ {\Pi_{\mathrm{iter}}\backslash \left\{ {{\cal M}_r,{{\cal M}_{r'}}} \right\}} \right\} \cup \left\{ {{{\cal M}_r^\circ} } \right\}$.
Then, we can define the new antenna allocation strategy for the new formed partition $\mathcal{B}_{\mathrm{iter}+1}$ as $\Pi_{\mathrm{iter}+1} = \left\{ {\Pi_{\mathrm{iter}}\backslash \left\{ {{\cal M}_r,{{\cal M}_{r'}}} \right\}} \right\} \cup \left\{ {{{\cal M}_r^\circ}, {{\cal M}_{r'}^*} } \right\}$.
In other words, we have the update rule as $\left\{ {{S_r},{S_{r'}}} \right\} \to \left\{ {{S_r}\backslash \left\{ k \right\},{S_{r'}} \cup \left\{ k \right\}} \right\}$, $\left\{ {{\cal M}_r,{{\cal M}_{r'}}} \right\}\to \left\{ {{{\cal M}_r^\circ},{\cal M}_{r'}^* } \right\}$, $\mathcal{B}_{\mathrm{iter}} \to \mathcal{B}_{\mathrm{iter}+1}$, and $\Pi_{\mathrm{iter}} \to \Pi_{\mathrm{iter}+1}$.
\end{Def}

The defined \emph{strictly preferred leaving and joining operation} above provides a mechanism to decide whether user $k$ should move from ${S_r}$ to ${S_{r'}}$, given that the coalition and partition pair $\left( {{S_{r'}}  \cup  \{k\} ,\mathcal{B}_{\mathrm{iter}+1}} \right)$ is strictly preferred over $\left( {{S_r},\mathcal{B}_{\mathrm{iter}}} \right)$.
However, as we mentioned before, there is a constraint on the size of coalition, i.e., ${\left|S_r\right|} \le 2$, $\forall r$.
As a result, we should prevent the size of the new coalition being larger than 2, i.e., $\left|\{{S_{r'}} \cup \left\{ k \right\}\}\right| > 2$.
To this end, we need to introduce the concept of \emph{strictly preferred switch operation} \cite{SaadCoalitional2012} as follows to enable user $k \in {S_{r}}$ and user $k' \in {S_{r'}}$ switch with each other, such that the new formed coalitions satisfy $\left| \{{S_{r}} \backslash \{k\} \cup \left\{ k' \right\}\} \right| \le 2$ and $\left| \{{S_{r'}} \backslash \{k'\} \cup \left\{ k \right\}\} \right| \le 2$.

\begin{Def}\label{C6:SwitchOperation}
In the $\mathrm{iter}$-th iteration, given a partition $\mathcal{B}_{\mathrm{iter}} = \{S_1,\ldots,S_{\left|\mathcal{B}_{\mathrm{iter}}\right|}\}$ and an corresponding antenna allocation strategy $\Pi_{\mathrm{iter}} = \{\mathcal{M}_1,\ldots,\mathcal{M}_{\left|\mathcal{B}_{\mathrm{iter}}\right|}\}$ of the set of users $\mathcal{K}$, user $k \in {S_{r}}$ and user $k' \in {S_{r'}}$ will switch with each other, if and only if it is a \emph{strictly preferred switch operation}, i.e.,
\begin{align}\label{C6:SwitchRule}
&\left( {{S_r},{S_r'},\mathcal{B}_{\mathrm{iter}}} \right){ \prec ^{k'}_k} \left({S_{r}} \backslash \{k\} \cup \left\{ k' \right\}, {{S_{r'}} \backslash \{k'\} \cup \left\{ k \right\} ,\mathcal{B}_{\mathrm{iter}+1}} \right) \notag\\
&\Leftrightarrow \mathop {\max }\limits_{{\mathcal{M}_r},{\mathcal{M}_r'}}  {\overline{R}_{{\mathrm{sum}}}}\left( {{\cal B}_{\mathrm{iter}},\Pi_{\mathrm{iter}}} \right)< \mathop {\max } \limits_{{\mathcal{M}_r},{\mathcal{M}_r'}} {\overline{R}_{{\mathrm{sum}}}}\left( {{\cal B}_{\mathrm{iter}+1},\Pi_{\mathrm{iter}+1}} \right),
\end{align}
where new formed partition is
\begin{equation}
	\mathcal{B}_{\mathrm{iter}+1} = \left\{ {{\cal B}_{\mathrm{iter}}\backslash \left\{ {{S_r},{S_{r'}}} \right\}} \right\} \cup \left\{ {\{{S_{r}} \backslash \{k\} \cup \left\{ k' \right\}\},\{{S_{r'}} \backslash \{k'\} \cup \left\{ k \right\}\}} \right\}.
\end{equation}
Let ${\cal M}_r^ +$ and ${\cal M}_{r'}^ +$ denote the optimal solutions for the maximization on the right hand side of \eqref{C6:SwitchRule}.
Then, we can define the new antenna allocation strategy for the new formed partition $\mathcal{B}_{\mathrm{iter}+1}$ as ${\Pi_{\mathrm{iter}+1} } = \left\{ {\Pi_{\mathrm{iter}}\backslash \left\{ {{\cal M}_r,{\cal M}_{r'}} \right\}} \right\} \cup \left\{ {{\cal M}_r^ + ,{\cal M}_{r'}^ + } \right\}$.
In other words, the update rules are $\left\{ {{S_r},{S_{r'}}} \right\} \to \left\{ {{S_{r}} \backslash \{k\} \cup \left\{ k' \right\},{S_{r'}} \backslash \{k'\} \cup \left\{ k \right\}} \right\}$, $\left\{ {{\cal M}_r,{{\cal M}_{r'}}} \right\}\to \left\{ {{{\cal M}_r^{+}},{\cal M}_{r'}^{+} } \right\}$, $\mathcal{B}_{\mathrm{iter}} \to \mathcal{B}_{\mathrm{iter}+1}$, and $\Pi_{\mathrm{iter}} \to \Pi_{\mathrm{iter}+1}$.
\end{Def}

Again, the maximization in \eqref{C6:SwitchRule} is only over ${\cal M}_r$ and ${\cal M}_{r'}$ within the feasible set of C4 and C5 in \eqref{C6:ResourceAllocation} for the involved coalitions ${S_{r}}$ and ${S_{r'}}$ in the switch operation, via fixing all the other antenna allocation strategies $\left\{ {\Pi_{\mathrm{iter}}\backslash \left\{ {{\cal M}_r,{\cal M}_{r'}} \right\}} \right\}$.
Also, the computational complexity for each maximization is acceptable via a two-dimensional full search\footnote{Since the NOMA group size is limited to two, there is only one antenna allocation variable for each coalition.
As a result, a two-dimensional search for antenna allocation within the involved two coalitions can be utilized to improve the system sum-rate.} over the integers in $[{M_{\min }},{M_{{\mathrm{BS}}} - {M_{\min }}}]$.
From \eqref{C6:SwitchRule}, we can observe that a switch operation is a strictly preferred switch operation of users $k$ and $k'$ if and only if when this switch operation can strictly increase the conditional system sum-rate.
The defined \emph{strictly preferred switch operation} above provides a mechanism to decide whether to switch user $k$ and user $k'$ with each other, given that it is a strictly preferred switch operation by users $k$ and $k'$.

\begin{table}
\begin{algorithm} [H]                    % enter the algorithm environment
\caption{User Grouping and Antenna Allocation Algorithm}
\label{C6:alg1}                             % and a label for \ref{C6:} commands later in the document
\begin{algorithmic} [1]
\small          % enter the algorithmic environment
\STATE \textbf{Initialization}\\
Initialize the iteration index $\mathrm{iter} = 0$.
The partition is initialized by $\mathcal{B}_0 = \mathcal{K} = \{S_1,\ldots,S_{K}\}$ with $S_k = {k}$, $\forall k$, i.e., OMA.
Correspondingly, the antenna allocation is initialized with $\Pi_0 = \{\mathcal{M}_1,\ldots,\mathcal{M}_K\}$ with $\mathcal{M}_k = \{{M_{\mathrm{BS}}}\}$, $\forall k$.
\REPEAT
\FOR{$k$ = $1$:$K$}
    \STATE User $k \in {S_{r}}$ visits each existing coalitions ${S_{r'}} \in \mathcal{B}_{\mathrm{iter}}$ with ${S_{r'}} \neq {S_{r}}$.
    \IF {$\left|\{{S_{r'}} \cup \left\{ k \right\}\}\right| \le 2$}
        \IF {$\left( {{S_r},\mathcal{B}_{\mathrm{iter}}} \right) { \prec _k} \left( {{S_{r'}}  \cup  \{k\} ,\mathcal{B}_{\mathrm{iter}+1}} \right)$}
            \STATE Execute the leaving and joining operation in Definition \ref{C6:LeaveAndJoin}.
            \STATE $\mathrm{iter} = \mathrm{iter} + 1$.
        \ENDIF
    \ELSE
        \IF {$\left( {{S_r},{S_r'},\mathcal{B}_{\mathrm{iter}}} \right) { \prec ^{k'}_k} \left({S_{r}} \backslash \{k\} \cup \left\{ k' \right\}, {{S_{r'}} \backslash \{k'\} \cup \left\{ k \right\} ,\mathcal{B}_{\mathrm{iter}+1}} \right)$}
            \STATE Execute the switch operation in Definition \ref{C6:SwitchOperation}.
            \STATE $\mathrm{iter} = \mathrm{iter} + 1$.
        \ENDIF
    \ENDIF
\ENDFOR
\UNTIL {No strictly preferred leaving and joining operation or switch operation can be found.}
\RETURN {$\mathcal{B}_{\mathrm{iter}}$ and $\Pi_{\mathrm{iter}}$}
\end{algorithmic}
\end{algorithm}
\end{table}

Now, the proposed coalition algorithm for the user grouping and antenna allocation problem in \eqref{C6:ResourceAllocation} is shown in \textbf{Algorithm} \ref{C6:alg1}.
The algorithm is initialized with each user as a coalition, i.e., OMA, and each user is allocated with the whole antenna array.
In each iteration, all the users visit all the potential coalitions except its own coalition in current coalitional structure, i.e., ${S_{r'}} \in \mathcal{B}_{\mathrm{iter}}$ and ${S_{r'}} \neq {S_{r}}$.
Then, each user checks and executes the leaving and joining operation or the switch operation based on the Definition \ref{C6:LeaveAndJoin} and Definition \ref{C6:SwitchOperation}, respectively.
The iteration stops when no more preferred operation can be found.

\noindent\textbf{\underline{{Effectiveness, Stability, and Convergence}}}

In the following, we briefly discuss the effectiveness, stability, and convergence for our proposed coalition formation algorithm.
Interested readers are referred to \cite{SaadCoalitional2012,WangCoalitionNOMA,SaadCoalitionOrder} for detailed proofs.
From the Definition \ref{C6:LeaveAndJoin} and Definition \ref{C6:SwitchOperation}, we can observe that every executed operation in \textbf{Algorithm} \ref{C6:alg1} increases the conditional system sum-rate.
In other words, \textbf{Algorithm} \ref{C6:alg1} can effectively increase the conditional system sum-rate in \eqref{C6:ResourceAllocation}.
The stability of \textbf{Algorithm} \ref{C6:alg1} can be proved by contradiction, cf. \cite{SaadCoalitional2012,WangCoalitionNOMA,SaadCoalitionOrder}.
Firstly, we need to note that each user has an incentive to leave its current coalition only and if only when this operation can increase the conditional system sum-rate, i.e., a strictly preferred leaving and joining operation or switch operation.
Then, we can define a stable coalitional structure state as a final state $\mathcal{B}^{*}$ after \textbf{Algorithm} \ref{C6:alg1} terminates where no user has an incentive to leave its current coalition.
If there exists a user $k$ want to leave its current coalition $S_r$ in the final coalitional structure $\mathcal{B}^{*}$, it means that \textbf{Algorithm} \ref{C6:alg1} will not terminate and $\mathcal{B}^{*}$ is not a final state, which causes a contradiction.

Now, the convergence of \textbf{Algorithm} \ref{C6:alg1} can be proved with the following logic.
Since the number of feasible combinations of user grouping and antenna allocation in \eqref{C6:ResourceAllocation} is finite, the number of strictly preferred operations is finite.
Moreover, according to \textbf{Algorithm} \ref{C6:alg1}, the conditional system sum-rate increases after each approved operation.
Since the conditional system sum-rate is upper bounded by above due to the limited number of RF chains and time resources, \textbf{Algorithm} \ref{C6:alg1} terminates when the conditional system sum-rate is saturated.
In other words, \textbf{Algorithm} \ref{C6:alg1} converges to the final stable coalitional structure $\mathcal{B}^{*}$ within a limited number of iterations.

\noindent\textbf{\underline{{Computational Complexity}}}

Assume that, in the $\mathrm{iter}$-th iteration,  a coalitional structure $\mathcal{B}_{\mathrm{iter}}$ consists of $\left|\mathcal{B}^{\mathrm{I}}_{\mathrm{iter}}\right|$ single-user coalitions and $\left|\mathcal{B}^{\mathrm{II}}_{\mathrm{iter}}\right|$ two-users coalitions with $\left|\mathcal{B}_{\mathrm{iter}}\right| = \left|\mathcal{B}^{\mathrm{I}}_{\mathrm{iter}}\right| + \left|\mathcal{B}^{\mathrm{II}}_{\mathrm{iter}}\right|$.
For user $k$, the computational complexity to locate a strictly preferred leaving and joining operation is $\mathcal{O}\left(2(\left|\mathcal{B}^{\mathrm{I}}_{\mathrm{iter}}\right| + 1) ({M_{{\mathrm{BS}}} - 2{M_{\min }}})\right)$ in the worst case and the counterpart to locate a strictly preferred switch operation is $\mathcal{O}\left(4\left|\mathcal{B}^{\mathrm{II}}_{\mathrm{iter}}\right| ({M_{{\mathrm{BS}}} - 2{M_{\min }}})^2\right)$ in the worst case.
As a result, the computational complexity in each iteration of our proposed coalition formation algorithm is $\mathcal{O}\left( 2(\left|\mathcal{B}^{\mathrm{I}}_{\mathrm{iter}}\right| + 1) ({M_{{\mathrm{BS}}}- 2{M_{\min }}})  + 4\left|\mathcal{B}^{\mathrm{II}}_{\mathrm{iter}}\right| ({M_{{\mathrm{BS}}} - 2{M_{\min }}})^2\right)$ in the worst case, which is substantially low compared to that of the exhaustive search. i.e., \eqref{C6:Complexity}.

\subsection{Second Stage: Digital Precoder and Power Allocation Design}

\noindent\textbf{\underline{{ZF Digital Precoder}}}

Given the obtained user grouping strategy $\mathcal{B}^{*} = \{S_1^{*},\ldots,S_{N_{\mathrm{RF}}}^{*}\}$ and the antenna allocation strategy $\Pi^{*}$ in the first stage, we can obtain the effective channel matrix ${{\mathbf{\widetilde H}}} \in \mathbb{C}^{{N_{\mathrm{RF}} \times K}}$ via uplink pilot transmission.
We adopt a ZF digital precoder to suppress the inter-group interference.
Since there might be more than one users in each group $S_r^{*} \in \mathcal{B}^{*}$, we perform singular value decomposition (SVD) on the equivalent channel for each NOMA group $S_r^{*}$ with $\left|S_r^{*}\right| = 2$.
In particular, let ${{\mathbf{\widetilde H}}}_r \in \mathbb{C}^{{N_{\mathrm{RF}} \times \left|S_r^{*}\right|}}$, $\forall r$, denotes the effective channel matrix for all the $\left|S_r^{*}\right|$ users in the coalition $S_r^{*}$.
We have the SVD for ${{\mathbf{\widetilde H}}}_r$ as follows:
\begin{equation}
{{\mathbf{\widetilde H}}}_r^{\mathrm{H}} = {\mathbf{U}_r}{\mathbf{\Sigma} _r}\mathbf{V}_r^{\mathrm{H}},
\end{equation}
where ${\mathbf{U}_r} = {\left[{{\bf{u}}_{r,{1}}}, \cdots ,{{\bf{u}}_{r,{\left|S_r\right|}}}\right]} \in \mathbb{C}^{{ \left|S_r^{*}\right| \times \left|S_r^{*}\right|}}$ is the left singular matrix, ${\mathbf{\Sigma} _r} \in \mathbb{R}^{{ \left|S_r^{*}\right| \times N_{\mathrm{RF}}}}$ is the singular value matrix with its diagonal entries as singular values in descending order, and ${\mathbf{V}_r} \in \mathbb{C}^{{ N_{\mathrm{RF}} \times N_{\mathrm{RF}}}}$ is the right singular matrix.
Then, the equivalent channel vector of the NOMA group $S_r^{*}$ is given by
\begin{equation}
{{\mathbf{\hat h}}}_r = {{\mathbf{\widetilde H}}}_r {{\bf{u}}_{r,{1}}} \in \mathbb{C}^{{N_{\mathrm{RF}} \times 1}},
\end{equation}
where ${{\bf{u}}_{r,{1}}} \in \mathbb{C}^{{ \left|S_r^{*}\right| \times 1}}$ is the first left singular vector corresponding to the maximum singular value.
Note that, the equivalent channel for a single-user coalition with $\left|S_r^{*}\right| = 1$ can be directly given with its effective channel, i.e., ${{\mathbf{\hat h}}}_r = {{\mathbf{\widetilde h}}}_r$.
Now, the equivalent channel for all the coalitions on all the RF chains can be given by
\begin{equation}
{\bf{\hat H}} = {\left[ {
{{\bf{\hat{h}}}_1}, \ldots,{{\bf{\hat{h}}}_{N_{\mathrm{RF}}}}
} \right]} \in \mathbb{C}^{{ {N_{\mathrm{RF}}} \times N_{\mathrm{RF}}}}.
\end{equation}
Furthermore, the ZF digital precoder can be obtained by
\begin{equation}
{\mathbf{G}} = {\bf{\hat H}}^{\mathrm{H}}\left({\bf{\hat H}}{\bf{\hat H}}^{\mathrm{H}}\right)^{-1} \in \mathbb{C}^{{ N_{\mathrm{RF}} \times {N_{\mathrm{RF}}} }},
\end{equation}
where ${\mathbf{G}} = \left[ {
{{{\mathbf{g}}_1}},\ldots ,{{{\mathbf{g}}_{{N_{\mathrm{RF}}}}}}
} \right]$ and ${{{\mathbf{g}}_r}}$ denotes the digital precoder shared by all the user in $S_r^{*}$.

\noindent\textbf{\underline{{Power Allocation Design}}}

Given the effective channel matrix ${{\mathbf{\widetilde H}}}_r$ and the digital precoder ${\mathbf{G}}$, the optimal power allocation can be formulated as the following optimization problem:
\begin{align} \label{C6:ResourceAllocation_Power}
&\underset{{p_{k,r}}}{\maxo} \;\;R_{\mathrm{sum}}\\
\mbox{s.t.}\;\;
%%%%%
&\mbox{C1: } {p_{k,r}} \ge 0, \forall k,r,
\;\;\mbox{C2: } \sum\limits_{k = 1}^K\sum\limits_{r = 1}^{N_{{\mathrm{RF}}}} {{u_{k,r}^{*}}{p_{k,r}}}  \le {p_{{\mathrm{BS}}}}, \notag\\
&\mbox{C3: } {u_{k,r}^{*}}R_{k,i,r} \ge {u_{k,r}^{*}}R_{i,r}, \forall i > k, \forall r,\;\;\mbox{C4: } \sum\limits_{r = 1}^{{N_{{\mathrm{RF}}}}} R_{k,r} \ge R_{\mathrm{min}}, \forall k, \notag
\end{align}
where $R_{k,r}$, $R_{k,i,r}$, and $R_{\mathrm{sum}}$ are given by \eqref{C6:DLIndividualRate1}, \eqref{C6:IndividualRate2}, and \eqref{C6:SumRate} with replacing ${u_{k,r}}$ with ${u_{k,r}^{*}}$, respectively.
Note that the user scheduling ${u_{k,r}^{*}}$ can be easily obtained by the following mapping:
\begin{equation}
{u_{k,r}^{*}} = \left\{ {\begin{array}{*{20}{c}}
1,&{{\mathrm{if}}\;k \in S_r^{*}},\\
0,&{{\mathrm{otherwise}}}.
\end{array}} \right.
\end{equation}
Constraint C2 is the total power constraint at the BS.
Constraint C3 is introduced to guarantee the success of SIC decoding.
Note that constraint C3 are inactive when $u_{k,r}^{*} = 0$ or $u_{i,r}^{*} = 0$.
Constraint C4 is imposed to guarantee a minimum rate requirement for each user.

The formulated problem is a non-convex optimization, but can be equivalently transformed to a canonical form of D.C. programming \cite{WeiTCOM2017} as follows:
\begin{align} \label{C6:ResourceAllocation_Power1}
&\underset{{p_{k,r}}}{\mino} \;\;{{H_1}\left( {\bf{p}} \right) - {H_2}\left( {\bf{p}} \right)}\\
\mbox{s.t.}\;\;
%%%%%
&\mbox{C1: } {p_{k,r}} \ge 0, \forall k,r,
\;\;\mbox{C2: } \sum\limits_{k = 1}^K\sum\limits_{r = 1}^{N_{{\mathrm{RF}}}}  {{u_{k,r}^{*}}{p_{k,r}}}  \le {p_{{\mathrm{BS}}}}, \notag\\
&\mbox{C3: } u_{k,r}^*u_{i,r}^*{{{{\left| {{{{\bf{\widetilde h}}}_i^{\mathrm{H}}}{{\bf{g}}_r}} \right|}^2}}} D_2^{k,i,r}\left( {\bf{p}} \right) \le u_{k,r}^*u_{i,r}^*{{{{\left| {{{{\bf{\widetilde h}}}_k^{\mathrm{H}}}{{\bf{g}}_r}} \right|}^2}}}D_2^{i,i,r}\left( {\bf{p}} \right), \forall i > k, \forall r,\notag\\
&\mbox{C4: } {u_{k,r}^*{p_{k,r}}{{{{\left| {{{{\bf{\widetilde h}}}_k^{\mathrm{H}}}{{\bf{g}}_r}} \right|}^2}}} \ge \left( {{2^{u_{k,r}^*{R_{{\mathrm{min}}}}}} - 1} \right)D_2^{k,k,r}\left( {\bf{p}} \right)}, \forall k,\notag
\end{align}
where ${\bf{p}} \in \mathbb{R}^{{ K {N_{\mathrm{RF}}}} \times 1}$ denotes the collection of ${p_{k,r}}$, ${H_1}\left( {\bf{p}} \right)$ and ${H_2}\left( {\bf{p}} \right)$ are given by
\begin{equation}
{H_1}\left( {\bf{p}} \right)= -\sum\limits_{k = 1}^K {\sum\limits_{r = 1}^{{N_{RF}}} {{{\log }_2}\left( D_1^{k,k,r}\left( {\bf{p}} \right) \right)} } \;\;\mathrm{and} \;\;
{H_2}\left( {\bf{p}} \right)= -\sum\limits_{k = 1}^K {\sum\limits_{r = 1}^{{N_{RF}}} {{{\log }_2}\left( D_2^{k,k,r}\left( {\bf{p}}\right) \right)} },
\end{equation}
respectively, and $D_1^{k,i,r}\left( {\bf{p}} \right)$ and $D_2^{k,i,r}\left( {\bf{p}} \right)$ are given by
\begin{align}
D_1^{k,i,r}\left( {\bf{p}} \right) =& \sum\limits_{r' \ne r}^{{N_{{\mathrm{RF}}}}} {{{\left| {{{{\bf{\widetilde h}}}_k^{\mathrm{H}}}{{\bf{g}}_{r'}}} \right|}^2}\sum\limits_{d = 1}^K {u_{d,r'}^*{p_{d,r'}}} }  + \sum\limits_{d = 1}^i {u_{d,r}^*{p_{d,r}}} {\left| {{{{\bf{\widetilde h}}}_k^{\mathrm{H}}}{{\bf{g}}_r}} \right|^2} + {\sigma ^2} \;\mathrm{and} \notag\\
D_2^{k,i,r}\left( {\bf{p}} \right) =& \sum\limits_{r' \ne r}^{{N_{{\mathrm{RF}}}}} {{{\left| {{{{\bf{\widetilde h}}}_k^{\mathrm{H}}}{{\bf{g}}_{r'}}} \right|}^2}\sum\limits_{d = 1}^K {u_{d,r'}^*{p_{d,r'}}} } + \sum\limits_{d = 1}^{i - 1} {u_{d,r}^*{p_{d,r}}} {\left| {{{{\bf{\widetilde h}}}_k^{\mathrm{H}}}{{\bf{g}}_r}} \right|^2} + {\sigma ^2},
\end{align}
respectively.
Note that ${H_1}\left( {\bf{p}} \right)$ and ${H_2}\left( {\bf{p}} \right)$ are differentiable convex functions with respect to ${\bf{p}}$.
Thus, for any feasible solution ${\bf{p}}^{\mathrm{iter}}$ in the $\mathrm{iter}$-th iteration, we can obtain a lower bound for ${H_2}\left( {\bf{p}} \right)$, which is given by
\begin{equation}
{H_2}\left( {\bf{p}} \right) \ge {H_2}\left( {\bf{p}}^{\mathrm{iter}} \right) + {\nabla _{\bf{p}}}{H_2}{\left( {\bf{p}}^{\mathrm{iter}} \right)^{\mathrm{T}}}\left( {\bf{p}} - {\bf{p}}^{\mathrm{iter}} \right),
\end{equation}
with ${\nabla _{\bf{p}}}{H_2}\left( {{{\bf{p}}^{{\mathrm{iter}}}}} \right) = \left\{ {\frac{{\partial {H_2}\left( {\bf{p}} \right)}}{{\partial {p_{k,r}}}}\left| {_{{{\bf{p}}^{{\mathrm{iter}}}}}} \right.} \right\}_{k = 1,r = 1}^{k = K,r = {N_{{\mathrm{RF}}}}} \in \mathbb{R}^{{ K {N_{\mathrm{RF}}}} \times 1}$ denoting the gradient of ${H_2}\left( \cdot\right)$ with respect to ${\bf{p}}$ and
\begin{equation}
{\frac{{\partial {H_2}\left( {\bf{p}} \right)}}{{\partial {p_{k,r}}}}\left| {_{{{\bf{p}}^{{\mathrm{iter}}}}}} \right.} = -\frac{1}{\log(2)} \sum\limits_{k' = 1}^K {\sum\limits_{r' \ne r}^{{N_{RF}}} {\frac{{{{\left| {{{{\bf{\widetilde h}}}_{k'}^{\mathrm{H}}}{{\bf{g}}_r}} \right|}^2}u_{k,r}^*}}{D_2^{k',k',r'}\left( {{{\bf{p}}^{{\mathrm{iter}}}}} \right)}} }- \frac{1}{\log(2)} \sum\limits_{k' = k + 1}^K {\frac{{{{\left| {{{{\bf{\widetilde h}}}_{k'}^{\mathrm{H}}}{{\bf{g}}_r}} \right|}^2}u_{k,r}^*}}{{D_2^{k',k',r}\left( {{{\bf{p}}^{{\mathrm{iter}}}}} \right)}}}.
\end{equation}
Then, we obtain an upper bound for the minimization problem in \eqref{C6:ResourceAllocation_Power1} by solving the following convex optimization problem:
\begin{align} \label{C6:ResourceAllocation_Power2}
&\underset{{p_{k,r}}}{\mino} \;\;{H_1}\left( {\bf{p}} \right) - {H_2}\left( {\bf{p}}^{\mathrm{iter}} \right) -{\nabla _{\bf{p}}}{H_2}{\left( {\bf{p}}^{\mathrm{iter}} \right)^{\mathrm{T}}}\left( {\bf{p}} - {\bf{p}}^{\mathrm{iter}} \right)\\
\mbox{s.t.}\;\;
%%%%%
&\mbox{C1: } {p_{k,r}} \ge 0, \forall k,r,
\;\;\mbox{C2: } \sum\limits_{k = 1}^K \sum\limits_{r = 1}^{N_{{\mathrm{RF}}}}  {{u_{k,r}^{*}}{p_{k,r}}}  \le {p_{{\mathrm{BS}}}}, \forall r, \notag\\
&\mbox{C3: } u_{k,r}^*u_{i,r}^*{{{{\left| {{{{\bf{\widetilde h}}}_i^{\mathrm{H}}}{{\bf{g}}_r}} \right|}^2}}} D_2^{k,i,r}\left( {\bf{p}} \right) \le u_{k,r}^*u_{i,r}^*{{{{\left| {{{{\bf{\widetilde h}}}_k^{\mathrm{H}}}{{\bf{g}}_r}} \right|}^2}}}D_2^{i,i,r}\left( {\bf{p}} \right), \forall i > k, \notag\\
&\mbox{C4: } {u_{k,r}^*{p_{k,r}}{{{{\left| {{{{\bf{\widetilde h}}}_k^{\mathrm{H}}}{{\bf{g}}_r}} \right|}^2}}} \ge \left( {{2^{u_{k,r}^*{R_{{\mathrm{min}}}}}} - 1} \right)D_2^{k,k,r}\left( {\bf{p}} \right)}, \forall k.\notag
\end{align}

\begin{table}
	\begin{algorithm} [H]                    % enter the algorithm environment
		\caption{Power Allocation Algorithm}     % give the algorithm a caption
		\label{C6:alg2}                             % and a label for \ref{C6:} commands later in the document
		\begin{algorithmic} [1]
			\small          % enter the algorithmic environment
			\STATE \textbf{Initialization}\\
			Initialize the convergence tolerance $\epsilon$, the maximum number of iterations $\mathrm{iter}_\mathrm{max}$, the iteration index $\mathrm{iter} = 1$, and the initial feasible solution $\mathbf{{p}}^{\mathrm{iter}}$.
			\REPEAT
			\STATE Solve \eqref{C6:ResourceAllocation_Power2} for a given $\mathbf{{p}}^{\mathrm{iter}}$ to obtain the power allocation $\mathbf{{p}}^{\mathrm{iter}+1}$.
			\STATE Set $\mathrm{iter}=\mathrm{iter}+1$.
			\UNTIL
			$\mathrm{iter} = \mathrm{iter}_\mathrm{max}$ or ${\left\| {\mathbf{{p}}^{\mathrm{iter}} - \mathbf{{p}}^{\mathrm{iter}-1}} \right\|}\le \epsilon$.
			\STATE Return the solution $\mathbf{{p}}^{*} = \mathbf{{p}}^{\mathrm{iter}}$.
		\end{algorithmic}
	\end{algorithm}
\end{table}

Now, the problem in \eqref{C6:ResourceAllocation_Power2} is a convex programming problem which can be solved efficiently by standard convex problem solvers, such as CVX \cite{cvx}.
Based on D.C. programming \cite{VucicProofDC}, an iteration algorithm is developed to tighten the obtained upper bound in \eqref{C6:ResourceAllocation_Power2}, which is shown in \textbf{Algorithm} \ref{C6:alg2}.
The power allocation algorithm is initialized with ${\bf{p}}^{1}$, which is obtained by solving the problem in \eqref{C6:ResourceAllocation_Power2} with ${H_1}\left( {\bf{p}} \right)$ as the objective function.
In the $\mathrm{iter}$-th iteration, the updated solution $\mathbf{{p}}^{\mathrm{iter}+1}$ is obtained by solving the problem in \eqref{C6:ResourceAllocation_Power2} with $\mathbf{{p}}^{\mathrm{iter}}$.
The algorithm will terminate when the maximum iteration number is reached, i.e., $\mathrm{iter} = \mathrm{iter}_\mathrm{max}$, or the change of power allocation solutions between adjacent iterations becomes smaller than a given convergence tolerance, i.e., ${\left\| {\mathbf{{p}}^{\mathrm{iter}} - \mathbf{{p}}^{\mathrm{iter}-1}} \right\|}\le \epsilon$.
Note that, with differentiable convex functions ${H_1}\left( {\bf{p}} \right)$ and ${H_2}\left( {\bf{p}} \right)$, the proposed power allocation algorithm converges to a stationary point with a polynomial time computational complexity\cite{VucicProofDC}.

\section{Simulation Results}

In this section, we evaluate the performance of our proposed multi-beam mmWave-NOMA scheme via simulations.
Unless specified otherwise, the simulation setting is given as follows.
We consider an mmWave system with carrier frequency at $28$ GHz.
There are one LOS path and $L = 10$ NLOS paths for the channel model in \eqref{C6:ChannelModel1} and the path loss models for LOS and NLOS paths follow Table I in \cite{AkdenizChannelMmWave}.
A single hexagonal cell with a BS located at the cell center with a cell radius of $200$ m is considered.
All the $K$ users are randomly and uniformly distributed in the cell unless further specified.
The maximum transmit power of the BS is 46 dBm, i.e., $p_{\mathrm{BS}} \le 46$ dBm and the noise power at all the users is assumed identical with $\sigma^2 = -80$ dBm.
We assume that there are $M_{\mathrm{BS}} = 100$ antennas equipped at the BS and $M_{\mathrm{UE}} = 10$ antennas equipped at each user terminals.
The minimum number of antennas allocated to each user is assumed as 10\% of $M_{\mathrm{BS}}$, i.e., ${M_{\min }} = \frac{1}{10} M_{\mathrm{BS}}$.
The minimum rate requirement ${R_{\min }}$ is selected from a uniform distributed random variable with the range of $(0,5]$ bit/s/Hz.
The simulation results shown in the sequel are averaged over different realizations of the large scaling fading, the small-scale fading, and the minimum rate data requirement.

To show the effectiveness of our proposed two-stage resource allocation design, we compare the performance in the two stages to their corresponding optimal benchmark, respectively.
In particular, we compare the performance of our proposed coalition formation algorithm to the optimal exhaustive user grouping and antenna allocation.
Given the same obtained user grouping and antenna allocation strategy in the first stage, the performance of our proposed digital precoder and power allocation is compared to the optimal dirty chapter coding (DPC) scheme\cite{VishwanathDuality} in the second stage.
On the other hand, to demonstrate the advantages of our proposed multi-beam mmWave-NOMA scheme, we consider two baseline schemes in our simulations.
For baseline 1, the conventional mmWave-OMA scheme is considered where each RF chain can be allocated to at most one user.
To accommodate all the $K$ users on $N_{\mathrm{RF}}$ RF chains with $N_{\mathrm{RF}} \le K \le 2N_{\mathrm{RF}}$, users are scheduled into two time slots.
In each time slot, we perform the downlink mmWave-OMA transmission via a ZF digital precoder and a power allocation design without intra-group interference terms and constraint C3 in \eqref{C6:ResourceAllocation_Power}.
Note that, for a fair comparison, a user can be allocated to at most one time slot in our considered mmWave-OMA scheme since a user can only be associated with at most one RF chain in our proposed multi-beam mmWave-NOMA scheme.
For baseline 2, the single-beam mmWave-NOMA scheme is considered where only the users' LOS AOD within the same $-3$ dB main beamwidth can be considered as a NOMA group\footnote{It has been demonstrated that, in mmWave systems, the angle difference based user pairing \cite{zhouperformanceII} is superior to the channel gain based user pairing as in conventional NOMA schemes\cite{Dingtobepublished}.}.
If the resulting number of users and single-beam NOMA groups is equal or smaller than $N_{\mathrm{RF}}$, mmWave-OMA transmission is used for spatial multiplexing.
Otherwise, all the users and single-beam NOMA groups are scheduled into two time slots and then we perform mmWave-OMA transmission on each time slot.
For a fair comparison, both our proposed scheme and the baseline schemes are based on the LOS CSI only, i.e., $\left\{ {\theta _{1,0}}, \ldots ,{\theta _{K,0}} \right\}$ and $\left\{ {\alpha _{1,0}}, \ldots ,{\alpha _{K,0}} \right\}$.

\subsection{Convergence of the Proposed Coalition Formation Algorithm}
\begin{figure}[t!]
\centering
\includegraphics[width=4.5in]{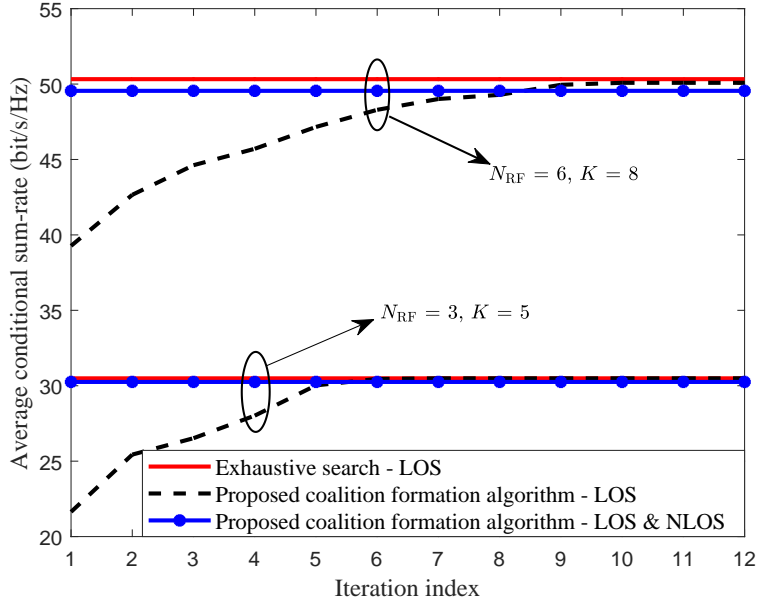}
\caption{Convergence of our proposed coalition formation algorithm in the first stage.}
\label{C6:Convergence}
\end{figure}

Figure \ref{C6:Convergence} illustrates the average conditional system sum-rate in the first stage versus the iteration index to show the convergence of our proposed coalition formation algorithm for user grouping and antenna allocation.
The performance of the exhaustive search on user grouping and antenna allocation is also shown as a benchmark.
Due to the large computational complexity of the optimal exhaustive search, we consider two simulation cases with ${N_{{\mathrm{RF}}}} = 3, K = 5$ and ${N_{{\mathrm{RF}}}} = 6, K = 8$.
The BS transmit power is set as $p_{\mathrm{BS}} = 30$ dBm.
Note that our proposed coalition formation algorithm is applicable to the case with a larger number of RF chains and users as shown in the following simulation cases.
We can observe that the average conditional system sum-rate of our proposed coalition formation algorithm monotonically increases with the iteration index.
Particularly, it can converge to a close-to-optimal performance compared to the exhaustive search within only $10$ iterations on average.
This demonstrates the fast convergence and the effectiveness of our proposed coalition formation algorithm.
With the user grouping and antenna allocation strategy obtained from our proposed coalition formation algorithm, the average conditional sum-rate performance of a multipath channel with both LOS and NLOS paths is also shown.
It can be observed that the performance degradation due to the ignorance of NLOS paths in the first stage is very limited, especially for the case with small numbers of RF chains and users.
In fact, the channel gain of the LOS path is usually much stronger than that of the NLOS paths in mmWave frequency bands due to the high attenuation in reflection and penetration, c.f. \cite{Rappaport2013,WangBeamSpace2017}.
With increasing ${N_{{\mathrm{RF}}}}$ and $K$, the probability that a NLOS path of a user lying in the beam of another user increases, which leads to a more severe inter-beam interference and a further degraded system sum-rate.
Besides, the analog beamforming of the massive antennas array at the BS can focus the energy on the LOS AOD and reduce the signal leakage to the NLOS AODs, and hence further reduce the impact of NLOS paths on the system performance.

\begin{figure}[t!]
\centering
\includegraphics[width=4.5in]{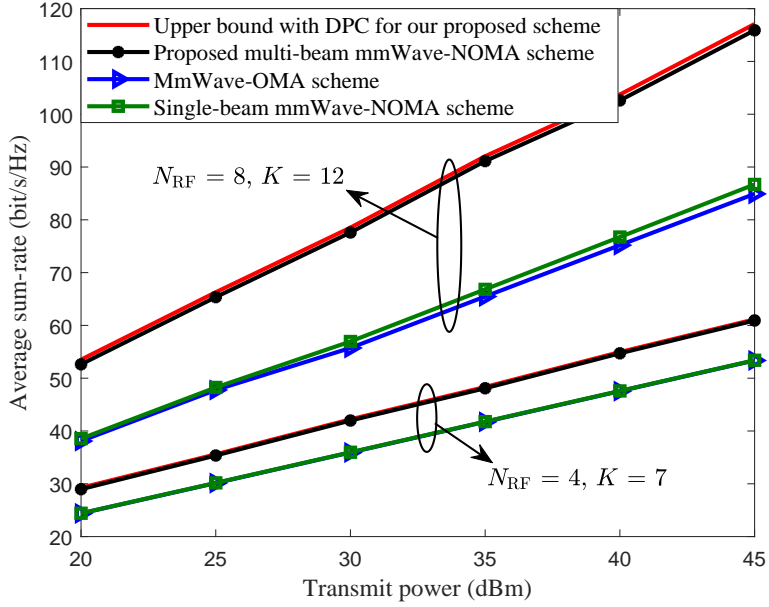}
\caption{Average system sum-rate (bit/s/Hz) versus the transmit power (dBm) at the BS.}
\label{C6:Case2}
\end{figure}

\subsection{Average System Sum-rate versus Transmit Power at the BS}
Figure \ref{C6:Case2} illustrates the average system sum-rate versus the total transmit power  $p_{\mathrm{BS}}$ at the BS.
The performance for our proposed scheme with the optimal DPC in the second stage is also shown as the performance benchmark.
Two baseline schemes are considered for comparison and two simulation cases with ${N_{{\mathrm{RF}}}} = 4, K = 7$ and ${N_{{\mathrm{RF}}}} = 8, K = 12$ are included.
We observe that the average system sum-rate monotonically increases with the transmit power since the proposed algorithm can efficiently allocate the transmit power when there is a larger transmit power budget.
Besides, it can be observed that the performance of our proposed resource allocation scheme can approach its upper bound achieved by DPC in the second stage.
It is owing to the interference management capability of our proposed resource allocation design.
In particular, the designed user grouping and antenna allocation algorithm is able to exploit the users' AOD distribution to avoid a large inter-group interference.
Besides, the adopted ZF digital precoder can further suppress the inter-group interference.
Within each NOMA group, the intra-group interference experienced at the strong user can be controlled with the SIC decoding and the intra-group interference at the weak user is very limited owing to our proposed power allocation design.

Compared to the existing single-beam mmWave-NOMA scheme and the mmWave-OMA scheme, the proposed multi-beam mmWave-NOMA scheme can provide a higher spectral efficiency.
This is because that the proposed multi-beam mmWave-NOMA scheme is able to pair two NOMA users with arbitrary AODs, which can generate more NOMA groups and exploit the multi-user diversity more efficiently.
We note that the performance of the single-beam mmWave-NOMA scheme is only slightly better than that of the mmWave-OMA scheme.
It is due to the fact that the probability of multiple users located in the same analog beam is very low\cite{zhouperformanceII}.
Note that, the average system sum-rate of mmWave systems is much larger than the typical value in microwave systems\cite{WeiTCOM2017}.
It is due to the array gain brought by the large number of antennas equipped at both the BS and user terminals\footnote{We note that the average sum-rate per user obtained in our simulations is comparable to the simulation results in the literature in the field of mmWave communications \cite{AlkhateebPrecoder2015,zhao2017multiuser}.}.

\begin{figure}[t!]
\centering
\includegraphics[width=4.5in]{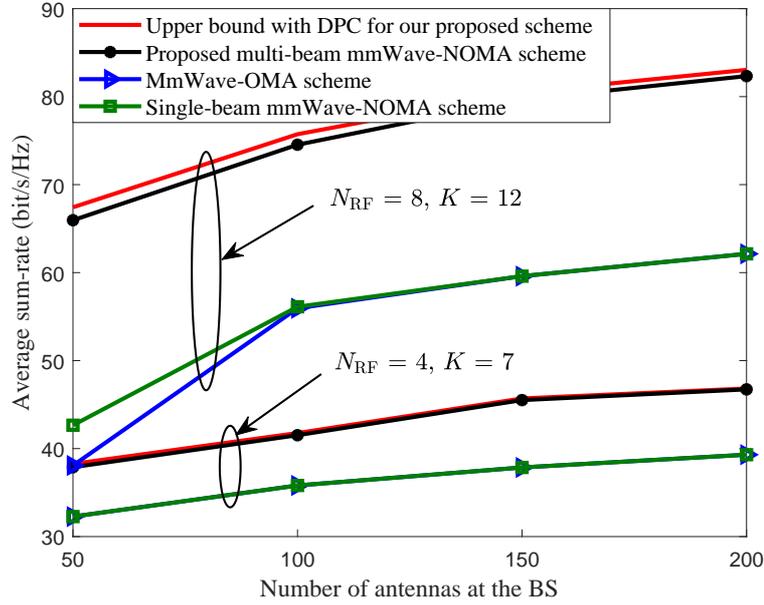}
\caption{Average system sum-rate (bit/s/Hz) versus the number of antennas equipped at the BS.}
\label{C6:Case3}
\end{figure}
\subsection{Average System Sum-rate versus Number of Antennas at the BS}

Figure \ref{C6:Case3} illustrates the average system sum-rate versus the number of antennas  $M_{\mathrm{BS}}$ equipped at the BS.
Note that, we fix the number of RF chains and only vary the number of antennas equipped at the BS ranging from $50$ to $200$ for the considered hybrid mmWave system.
The simulation setup is the same as Figure \ref{C6:Case2}, except that we fix the transmit power as $p_{\mathrm{BS}} = 30$ dBm.
We observe that the average system sum-rate increases monotonically with the number of antennas equipped at the BS due to the increased array gain.
Compared to the two baseline schemes, a higher spectral efficiency can be achieved by the proposed multi-beam NOMA scheme due to its higher flexibility of user pairing and enabling a more efficient exploitation of multi-user diversity.
In addition, it can be observed that the performance of the single-beam mmWave-NOMA scheme is almost the same as that of the mmWave-OMA scheme and it is only slightly better than that of the mmWave-OMA scheme when there is a small number of antennas.
In fact, the beamwidth is larger for a small number of antennas, which results in a higher probability for serving multiple NOMA users via the same analog beam.
Therefore, only in the case of a small number of antennas, the single-beam mmWave-NOMA scheme can form more NOMA groups and can provide a higher spectral efficiency than the mmWave-OMA scheme.

\subsection{Average System Sum-rate versus User Density}

\begin{figure}[t!]
\centering
\includegraphics[width=2.5in]{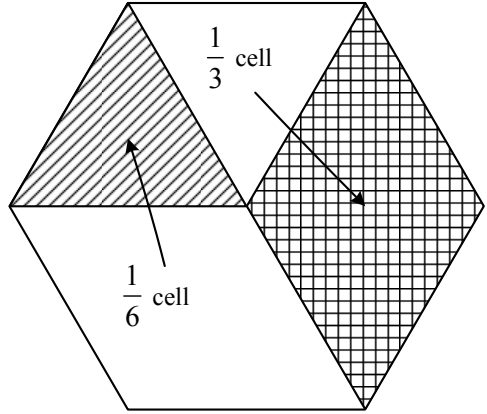}
\caption{An illustration of the selected portions of a cell.}
\label{C6:SelectedPortions}
\end{figure}

\begin{figure}[t!]
\centering
\includegraphics[width=4.5in]{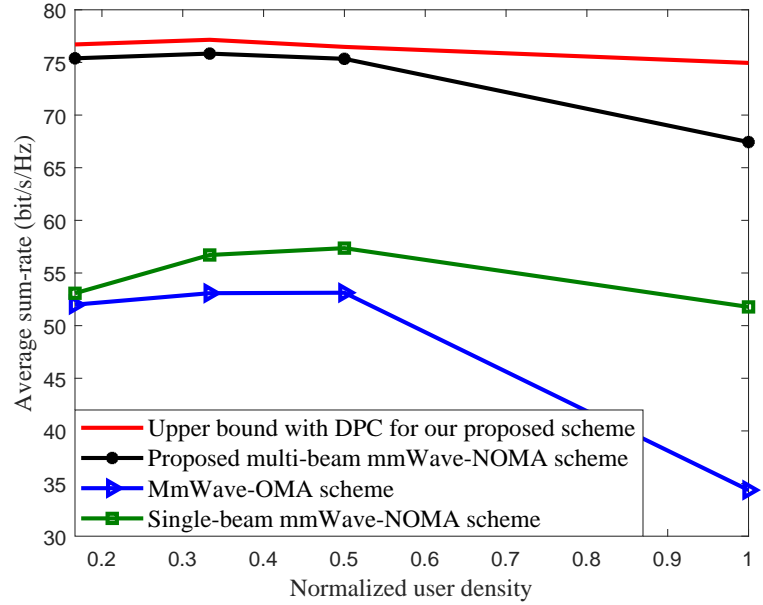}
\caption{Average system sum-rate (bit/s/Hz) versus the normalized user density.}
\label{C6:Case4}
\end{figure}

Recall that the performance of the baseline single-beam mmWave-NOMA scheme is related to all the users' AOD distribution depending on the user density.
It is interesting to investigate the relationship between the system performance and the user density $\rho = \frac{K}{\mathrm{Area}}$, where $\mathrm{Area}$ denotes the area of the selected portions in the cell.
Figure \ref{C6:Case4} illustrates the average system sum-rate versus the normalized user density.
To facilitate the simulations, we keep the number of RF chains as ${N_{{\mathrm{RF}}}} = 8$ and the number of users as $K = 12$ and change the area by selecting different portions of the cell, as illustrated in Figure \ref{C6:SelectedPortions}.
In the selected area, all the $K$ users are uniformly and randomly deployed.
For instance, the minimum normalized user density is obtained with all the $K$ users randomly scattered in the whole hexagonal cell, while the maximum normalized user density is obtained with all the $K$ users randomly deployed in the $\frac{1}{6}$ cell.
The total transmit power at the BS is $p_{\mathrm{BS}} = 30$ dBm.
We can observe that the average system sum-rate of the mmWave-OMA scheme decreases with the user density due to the high channel correlation among users, which introduces a higher inter-user interference.
In addition, the average system sum-rate of the single-beam mmWave-NOMA scheme firstly increases and then decreases with user density.
It is because that, in the low density regime, increasing the user density can provide a higher probability of multiple users located in the same analog beam and thus more NOMA groups can be formed.
On the other hand, in the high density regime, the inter-group interference becomes more severe with the increasing user density.
Besides, it can be observed that our proposed multi-beam mmWave-NOMA scheme can offer a higher average system sum-rate compared to the baseline schemes in the whole user density range.
Note that although the performance of the proposed scheme decreases with the user density, it decreases much slower than that of the mmWave-OMA scheme.
In fact, the inter-group interference becomes severe when more than two users are located in the same analog beam due to the high user density.
However, the proposed scheme can still exploit the multi-user diversity and provide a substantial system performance gain compared to the mmWave-OMA scheme.

\section{Summary}
In this chapter, we proposed a multi-beam NOMA framework for hybrid mmWave systems and studied the resource allocation design for the proposed multi-beam mmWave-NOMA scheme.
In particular, a beam splitting technique was proposed to generate multiple analog beams to serve multiple users for NOMA transmission.
Our proposed multi-beam mmWave-NOMA scheme is more practical than the conventional single-beam mmWave-NOMA schemes, which can flexibly pair NOMA users with an arbitrary AOD distribution.
As a result, the proposed scheme can generate more NOMA groups and hence is able to explore the multi-user diversity more efficiently.
To unlock the potential of the proposed multi-beam NOMA scheme, a two-stage resource allocation was designed to improve the system performance.
{More specifically}, a coalition formation algorithm based on coalition formation game theory was developed for user grouping and antenna allocation in the first stage, while a power allocation based on a ZF digital precoder was proposed to maximize the system sum-rate in the second stage.
Simulation results demonstrated that the proposed resource allocation can achieve a close-to-optimal performance in each stage and our proposed multi-beam mmWave NOMA scheme can offer a substantial system sum-rate improvement compared to the conventional mmWave-OMA scheme and the single-beam mmWave-NOMA scheme.

\chapter{NOMA for Hybrid MmWave Communication Systems with Beamwidth Control}\label{C7:chapter7}

\section{Introduction}
In the last chapter, we proposed a multi-beam non-orthogonal multiple access (NOMA) scheme for hybrid millimeter wave (mmWave) communication systems and studied its resource allocation design to maximize the system sum-rate.
Due to the prohibitively large power consumption of RF chains, even with hybrid architectures, the energy efficiency of mmWave systems, which is defined as the average number of bits delivered by consuming one joule of energy (bit/J) is generally limited \cite{GaoSubarray,lin2016energy} and remains to be further improved.
Hence, it is a fundamentally important issue to be tackled for realizing mmWave communications in future wireless networks, which is the main focus of this chapter.

{Applying NOMA in hybrid mmWave communications is potential to improve the system energy efficiency, especially considering the inevitable blockage effects in mmWave channels.}
In particular, when the line-of-sight (LOS) path of a user is blocked, the radio frequency (RF) chain serving this user can only rely on non-line-of-sight (NLOS) paths to provide communications.
However, the channel gain of the NLOS paths is usually much weaker than that of the LOS path in mmWave frequency bands, due to the high attenuation in reflection and penetration\cite{Rappaport2013}.
As a result, the dedicated RF chain serving this user wastes a lot of system resources for achieving only a low data rate, which translates into a low energy efficiency.
As demonstrated in Chapter 3, the performance gain of NOMA over OMA originates from the near-far diversity, i.e., the discrepancy in channel gains among NOMA users \cite{Wei2018PerformanceGain}.
In general, when the LOS path of a user is blocked, it should be treated as a weak user in the context of NOMA.
Thus, the weak user can be clustered with a strong user possessing a LOS link to form a NOMA group serving by only one RF chain to exploit the near-far diversity.
Meanwhile, the original RF chain and its associated phase shifters (PSs) dedicated to the weak user become idle, which can substantially save the associated circuit power consumption.
Inspired by these observations, in this chapter, we further investigate applying NOMA in hybrid mmWave communications and study its energy-efficient resource allocation design.

Owing to the high carrier frequency in mmWave frequency band, massive numbers of antennas can be equipped at transceivers and thus the beamwidths of the associated analog beams are typically narrow \cite{BusariSurveyonMMWave}.
The narrow beams impose a fundamental limit when applying NOMA in hybrid mmWave communication systems.
In particular, only the users within the main lobe of an analog beam can be clustered as a NOMA group.
Otherwise, the user located in the sidelobe suffers a dramatically signal power attenuation and cannot maintain a sustainable communication link.
{As a result, only few users with similar angles-of-departure (AODs) can be clustered as a NOMA group and thus the potential energy efficiency gain brought by NOMA is very limited.}
When the beamwidth can be controlled through an appropriate analog beamformer design, widening the analog beamwidth can effectively increase the probability of two users locating in the same analog beam.
For instance, consider a system with an $N = 128$ antenna uniform linear array equipped at the base station (BS) and assume independent and identically distributed (i.i.d.) spatial directions of all the $K = 12$ users.
According to Lemma 1 in \cite{XinyuGaoLetter}, doubling the beamwidth can increase the probability of the existence of multiple users sharing the same analog beam from 0.41 to 0.67.
Therefore, our proposed beamwidth control can increase the number of potential NOMA groups, which facilitates the design of an efficient NOMA scheme to exploit the NOMA gain.

{However, beamwidth control introduces an inevitable power loss at the main beam direction.
Note that a larger main lobe power loss could degrade the achievable data rate for a given transmit power, which leads to a lower system energy efficiency.
As a result, there is a non-trivial trade-off between the beamwidth and the system performance.
Yet, such a study has not been reported in the literature.
Therefore, this chapter proposes a beamwidth control-based hybrid mmWave NOMA scheme to overcome the limitation of narrow analog beams.
In particular, we aim to control the analog beamwidth adaptively to accommodate more NOMA groups, which enables a higher system energy efficiency since more RF chains are switched to idle to reduce the system energy consumption.}

In this chapter, we propose an mmWave NOMA scheme using beamwidth control and design the resource allocation to maximize the system energy efficiency.
Our main contributions are summarized as follows:
\begin{itemize}
\item We propose two beamwidth control methods, i.e., conventional beamforming (CBF)-based and Dolph-Chebyshev beamforming (DCBF)-based \cite{Dolph1946,Duhamel1953,Drane1968DC} beamwidth control, which are applicable to the cases with constant-modulus PSs and amplitude-adjustable PSs, respectively.
    More importantly, we characterize their main lobe power losses due to beamwidth control.
\item Based on the two beamwidth control methods, we first study an asymptotically optimal analog beamformer design to minimize the main lobe power loss in a single-RF chain system, which serves as a foundation for the design of analog beamformer in multi-RF chain systems.
\item We develop a NOMA user scheduling algorithm for the proposed hybrid mmWave NOMA scheme adopting the coalition formation game theory\cite{SaadCoalitional2012,WangCoalitionNOMA} to achieve a stable user grouping strategy.
\item Adopting the obtained NOMA user grouping strategy, the energy-efficient digital precoder design is formulated as a non-convex optimization problem to maximize the system energy efficiency taking into account the minimum individual user data rate requirement.
    Utilizing the quadratic transformation\cite{ShenFP}, the considered problem is transformed to its equivalent form which facilitates the application of an iterative block coordinate ascent algorithm to achieve a locally optimal solution.
\item {We have conducted extensive simulations to demonstrate that the proposed scheme with appropriate beamwidth control offers a substantial energy efficiency gain over the conventional OMA and NOMA schemes without beamwidth control.
    In addition, we have also demonstrated that the proposed scheme with widen analog beamwidth is more robust against the beam
    alignment error compared to the baseline schemes.}
\end{itemize}

\section{System Model}
\subsection{System Model}
\begin{figure}[t]
\centering
\includegraphics[width=2.5in]{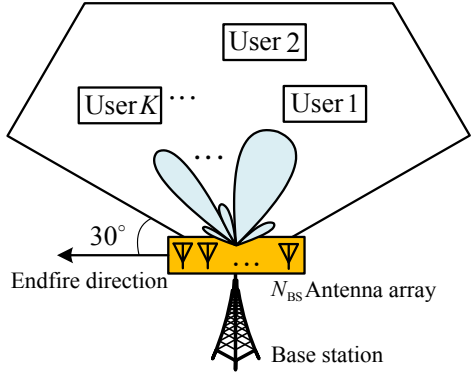}
\caption{The system model for downlink mmWave communications in a single-cell system.}
\label{C7:PortionCell}
\end{figure}

Consider downlink mmWave communications in a single-cell system with one base station (BS) and $K$ downlink users, as shown in Figure \ref{C7:PortionCell}.
The BS is located at the cell center with a cell radius of $D$ meters.
A commonly adopted uniform linear array (ULA)\cite{zhao2017multiuser} is employed at both the BS and user terminals, where the BS is equipped with $N_{\mathrm{BS}}$ antennas and each user is equipped with $N_{\mathrm{UE}}$ antennas.
To prevent an endfire beamforming \cite{van2002optimum} that sacrifices a large beamforming gain, we only consider the communication for users uniformly distributed in the $\frac{1}{3}$ cell, i.e., their angles of departure (AODs) range from 30 degrees to 150 degrees compared to endfire direction of the ULA at the BS, as illustrated in Figure \ref{C7:PortionCell}.
A hybrid fully-access architecture\footnote{In this chapter, we focus on the fully-connected hybrid architecture since it can provide a highly directional beamforming to compensate the severe path loss in mmWave frequency bands.
Besides, the analog beamformer design of fully-connected structure is a generalization of the other commonly used subarray architectures.
For example, the beamforming matrix of the partially-connected structure can be obtained by forcing the entry in the $n$-th row and the $r$-th column of the beamforming matrix of the fully-connected structure as zero if the $n$-th antenna is not connected with the $r$-th RF chain.} is adopted at both the BS and users, as shown in Figure \ref{C7:HybridNOMAStructure}.
In particular, each RF chain equipped at the BS can access all the $N_{\mathrm{BS}}$ antennas through $N_{\mathrm{BS}}$ PSs.
On the other hand, for the sake of a simple receiver design, each user has a single RF chain and connects to all its $N_{\mathrm{UE}}$ antennas via $N_{\mathrm{UE}}$ PSs.

\begin{figure}[t]
\centering
\subfigure[The transmitter structure at the BS.]
{\label{C7:HybridNOMAStructure:a} %% label for first subfigure
\includegraphics[width=0.9\textwidth]{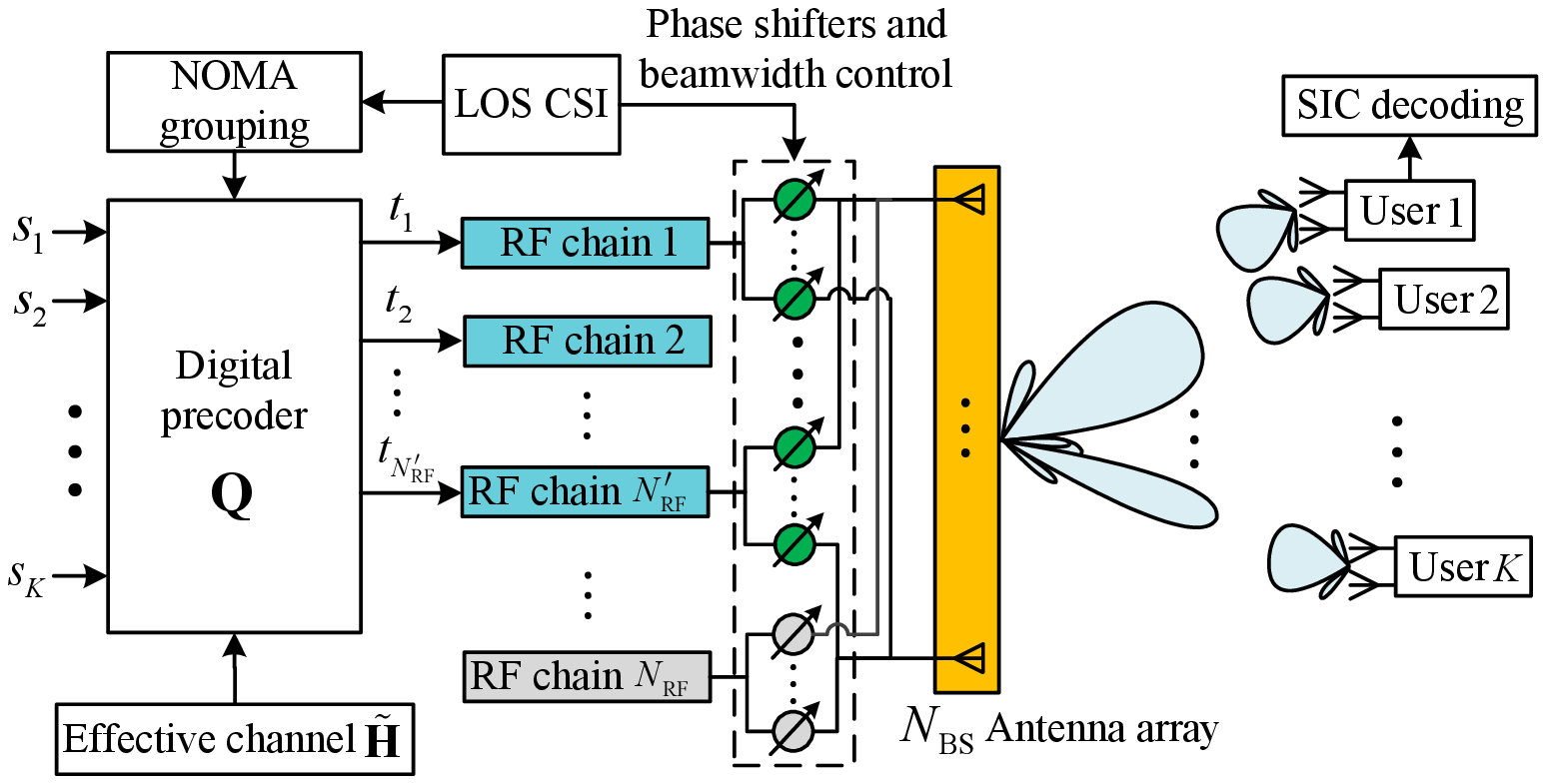}}
\subfigure[The receiver structure at user $k$.]
{\label{C7:HybridNOMAStructure:b} %% label for second subfigure
\includegraphics[width=0.5\textwidth]{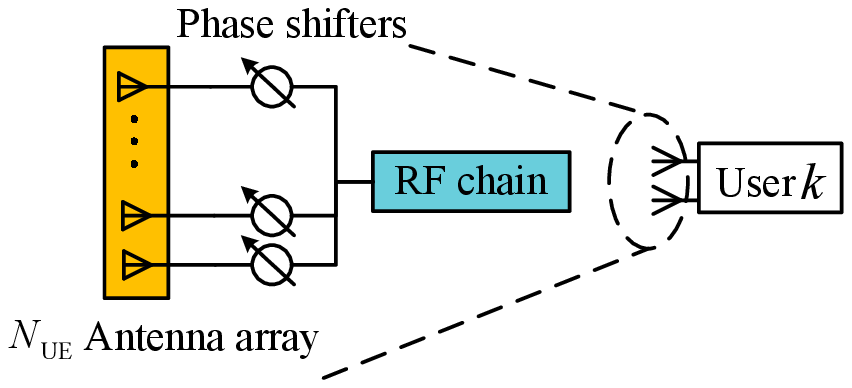}}
\caption{The transceiver structure of NOMA mmWave systems with the proposed beamwidth control. The green colored PSs in Figure \ref{C7:HybridNOMAStructure:a} denote their distinctive capability of beamwidth control compared to the PSs in Figure \ref{C7:HybridNOMAStructure:b}, where the detailed structures will be shown in Figure \ref{C7:BeamwidthControl}. The grey shadowed components are inactive in our proposed hybrid mmWave NOMA scheme for improving system efficiency.}
\label{C7:HybridNOMAStructure}%
\end{figure}

As shown in Figure \ref{C7:HybridNOMAStructure:a}, the key idea of the proposed NOMA scheme is to cluster two users to form a NOMA group according to their LOS channel state information (CSI) and to employ some beamwidth control such that an analog beam can cover the NOMA group.
As a result, the signals of the two users within a NOMA group can be superimposed for non-orthogonal transmission through a single RF chain, instead of utilizing two RF chains as in conventional OMA scheme.
Note that, due to the SIC decoding complexity and time delay, in this chapter, we assume that at most two users can be grouped as a NOMA group, as commonly adopted in the literature\cite{WeiTCOM2017,Sun2016Fullduplex,QiuLOMA}.
In addition, we assume that only the two users within the half-power beamwidth can be clustered as a NOMA group.
Otherwise, there is a substantial power loss, at least 3 dB, for the user outside the half-power beamwidth and its communication is almost impossible to guarantee.
Thanks to the beamwidth control method and the NOMA user grouping proposed in this chapter, we can assume that all the $K$ users are clustered as ${N'_{{\mathrm{RF}}}} \le {N_{{\mathrm{RF}}}}$ NOMA groups.
Correspondingly, we only need to activate ${N'_{{\mathrm{RF}}}}$ RF chains to serve ${N'_{{\mathrm{RF}}}}$ NOMA groups.
In other words, some RF chains are switched to idle which can substantially save the system circuit power consumption.
Let $u_{k,r} = 1$ denote that user $k$ is allocated to NOMA group $r \in \{1,2,\ldots,{N'_{{\mathrm{RF}}}}\}$. Otherwise, $u_{k,r} = 0$.
The signals going through all activated RF chains are given by
\begin{equation}
{\bf{t}} = {\bf{Qs}} = [t_1,t_2,\ldots,t_{N'_{{\mathrm{RF}}}}]^{\mathrm{T}} \in \mathbb{C}^{{ {N'_{{\mathrm{RF}}}} \times 1}},
\end{equation}
where $\textbf{s} = [s_1,s_2,\ldots,s_{K}]^{\mathrm{T}} \in \mathbb{C}^{{ K \times 1}}$ collects the normalized modulated symbols of all the $K$ users with ${\mathrm{E}}\left\{ {{{\left| {{s_k}} \right|}^2}} \right\} = 1$.
The digital precoder for the proposed NOMA scheme is ${\mathbf{Q}} = \left[ {
{{{\mathbf{q}}_1}},{{{\mathbf{q}}_2}},\ldots ,{{{\mathbf{q}}_{K}}}
} \right] \in \mathbb{C}^{{ {N'_{{\mathrm{RF}}}} \times K}}$, where ${{\mathbf{q}}_{k}} = \left[ {q}_{k,1},{q}_{k,2},\ldots ,{{q}}_{k,{N'_{{\mathrm{RF}}}}} \right]^{\mathrm{T}} \in \mathbb{C}^{{ {N'_{{\mathrm{RF}}}} \times 1}}$ denotes the digital precoder of user $k$.
For notational simplicity, we denote $p_k = \left\|{{{\mathbf{q}}_k}}\right\|^2$ as the power allocation for user $k$ satisfying $\sum\limits_{k = 1}^K \left\|{{{\mathbf{q}}_k}}\right\|^2  \le {p_{{\mathrm{BS}}}}$ with ${p_{{\mathrm{BS}}}}$ denoting the maximum transmit power budget for the BS.
Note that the digital precoded signal $t_r$ on RF chain $r$ contains the multiplexed signals of all the users and all the NOMA groups, owing to the adopted digital precoder ${\mathbf{Q}}$.
In fact, sharing a RF chain by a NOMA group is essentially sharing a spatial degree of freedom.
For pure analog beamforming mmWave systems, the digital precoder component in Figure \ref{C7:HybridNOMAStructure:a} can be replaced by superposition coding.
Each entry of the digital precoder ${\mathbf{Q}}$ can be obtained by
\begin{equation}\label{C7:PureAnalog1}
{q}_{k,r} = u_{k,r}\sqrt{p_{k}}, \forall k,r,
\end{equation}
and the superimposed signal on RF chain $r$ for NOMA group $r$ can be explicitly denoted by
\begin{equation}\label{C7:PureAnalog2}
t_{r} = \sum \limits_{k=1}^{K} u_{k,r} \sqrt{p_{k}} s_k,
\end{equation}
where $\sum\limits_{k = 1}^K {{p_k}}  \le {p_{{\mathrm{BS}}}}$.

At the user side, the received signal at user $k$ after processed by receiving analog beamforming is given by
\begin{equation}\label{C7:SystemModel}
{y_k} = {\bf{v}}_k^{\mathrm{H}}{{\bf{H}}_k}{\bf{Wt}} + {\bf{v}}_k^{\mathrm{H}}{{\bf{z}}_k}
=\widetilde{\mathbf{h}}^{\mathrm{H}}_k{\bf{t}} + {\bf{v}}_k^{\mathrm{H}}{{\bf{z}}_k}, \forall k,
\end{equation}
where channel matrix ${{\bf{H}}_k}$ and effective channel vector $\widetilde{\mathbf{h}}_k$ will be detailed in the channel model section.
The analog beamformer for the proposed NOMA scheme is denoted by matrix $\mathbf{W} = \left[{\mathbf{w}_1},{\mathbf{w}_2},\ldots,{\mathbf{w}_{N'_{{\mathrm{RF}}}}}\right] \in \mathbb{C}^{{{N_{{\mathrm{BS}}}} \times {N'_{{\mathrm{RF}}}}}}$, where ${\textbf{w}_{r}} \in \mathbb{C}^{{{N_{{\mathrm{BS}}}} \times 1}}$ denotes the analog beamformer for RF chain $r$ or equivalently for NOMA group $r$, with ${\left\| {{{\bf{w}}_{r}}} \right\|^2} = 1$.
Vector $\textbf{v}_k \in \mathbb{C}^{{N_{\mathrm{UE}} \times 1}}$ denotes the adopted normalized analog beamformer of user $k$ with ${\left\| {{{\bf{v}}_k}} \right\|^2} = 1$.
Vector $\textbf{z}_k \in \mathbb{C}^{{N_{\mathrm{UE}} \times 1}}$ is the additive white Gaussian noise (AWGN) at receiving antenna array of user $k$, i.e., $\textbf{z}_k \sim {\cal CN}(0,N_{0} \textbf{I}_{N_{\mathrm{UE}}})$, where $N_{0}$ is the noise power spectral density.

\subsection{Channel Model}
Matrix ${{\bf{H}}_k} \in \mathbb{C}^{{ {N_{{\mathrm{UE}}}} \times {N_{{\mathrm{BS}}}}}}$ denotes the channel matrix between the BS and user $k$.
In this chapter, we apply the widely adopted Saleh-Valenzuela model \cite{WangBeamSpace2017} for our considered mmWave system.
In particular, ${{\bf{H}}_k}$ is given by
\begin{equation}\label{C7:ChannelModel1}
{{\bf{H}}_k} = {\alpha _{k,0}}{{\bf{H}}_{k,0}} + \sum\limits_{l = 1}^L {{\alpha _{k,l}}{{\bf{H}}_{k,l}}},
\end{equation}
where ${\alpha _{k,0}}$ denotes the LOS complex path gain and ${{\bf{H}}_{k,0}}  \in \mathbb{C}^{ N_{\mathrm{UE}} \times N_{\mathrm{BS}} }$ is the LOS channel matrix between the BS and user $k$.
In \eqref{C7:ChannelModel1}, matrix ${\mathbf{H}}_{k,l} \in \mathbb{C}^{ N_{\mathrm{UE}} \times N_{\mathrm{BS}} }$ denotes the $l$-th NLOS channel matrix between the BS and user $k$, ${\alpha _{k,l}}$ denotes the corresponding $l$-th NLOS complex path gain, $1 \le l \le L$, and $L$ denotes the total number of NLOS paths.
The channel matrix ${\mathbf{H}}_{k,l}$, $\forall l \in \{0,\ldots,L\}$, can be generally given by ${\mathbf{H}}_{k,l} = {\mathbf{a}}_{\mathrm{UE}} \left(  \phi _{k,l}, {N_{{\mathrm{UE}}}} \right){\mathbf{a}}_{\mathrm{BS}}^{\mathrm{H}}\left( \theta _{k,l}, {N_{{\mathrm{BS}}}} \right)$,
with ${\mathbf{a}}_{\mathrm{BS}}\left( \theta _{k,l}, {N_{{\mathrm{BS}}}}  \right)  =  \left[ {1, {e^{ - j \zeta\left(\theta_{k,l}\right) }}, \ldots , }{e^{ - j\left({N_{{\mathrm{BS}}}} - 1\right)\zeta\left(\theta_{k,l}\right) }}\right]^{\mathrm{T}}$ $ \in \mathbb{C}^{{ {N_{{\mathrm{BS}}}} \times 1}}$
denoting the array response vector \cite{van2002optimum} for the $l$-th path of user $k$ with AOD ${\theta _{k,l}}$ at the BS and
${\mathbf{a}}_{\mathrm{UE}}\left( \phi _{k,l},{N_{{\mathrm{UE}}}} \right) = \left[ 1, {e^{ - j\zeta\left(\phi_{k,l}\right) }}, \ldots ,{e^{ - j{\left({N_{{\mathrm{UE}}}} - 1\right)}\zeta\left(\phi_{k,l}\right) }} \right] ^ {\mathrm{T}} \in \mathbb{C}^{{ {N_{{\mathrm{UE}}}} \times 1}}$ denoting the array response vector for the $l$-th path with angle of arrival (AOA) ${\phi _{k,l}}$ at user $k$.
Note that $\zeta\left(x\right) = \frac{2\pi d}{\lambda} \cos\left(x\right)$, $\lambda$ denotes the wavelength at the carrier frequency, and $d = \frac{\lambda}{2}$ denotes the space between adjacent antennas.

In this chapter, we define the strong or weak user according to their LOS path gains.
Without loss of generality, we assume that the users are indexed in the descending order of LOS path gains, i.e., ${\left| {{\alpha _{1,0}}} \right|} \ge {\left| {{\alpha _{2,0}}} \right|} \ge , \ldots , \ge {\left| {{\alpha _{K,0}}} \right|}$.
In this work, we consider a fixed SIC decoding order according to the LOS path gain order $1,2,\ldots,K$ to facilitate a tractable resource allocation design and to save the system overhead.
In particular, if user $k$ and user $i$ are scheduled to form a NOMA group associated with RF chain $r$, $\forall k<i$, user $k$ first decodes the messages of user $i$ before decoding its own information.
In contrast, user $i$ directly decodes its own information via treating the signals of user $k$ as interference.

As shown in \eqref{C7:SystemModel}, the effective channel captures both effects of the channel matrix and the analog beamformer.
The design of the analog beamformer ${\bf{W}}$ at the BS will be shown in the latter of this chapter.
In addition, every user would like to steer its receive analog beamformer to its LOS AOA to maximize the received power and thus we have ${\bf{v}}_k = {\mathbf{a}}_{\mathrm{UE}} \left(  \phi _{k,0}, {N_{{\mathrm{UE}}}} \right)$, $\forall k$.
Adopting the designed analog beamforming, the effective channel matrix can be denoted by
\begin{equation}\label{C7:EffectiveChannel}
{\widetilde{{\bf{H}}}} = \left[\widetilde{\mathbf{h}}_1,\widetilde{\mathbf{h}}_2,\ldots,\widetilde{\mathbf{h}}_K\right]\in \mathbb{C}^{{N'_{{\mathrm{RF}}}} \times K},
\end{equation}
where $\widetilde{\mathbf{h}}_k = \left[\widetilde{{{h}}}_{k,1},\widetilde{{{h}}}_{k,2},\ldots,\widetilde{{{h}}}_{k,{N'_{{\mathrm{RF}}}}}\right]^{\mathrm{T}} = \left({\bf{v}}_k^{\mathrm{H}}{{\bf{H}}_k}{\bf{W}}\right)^{\mathrm{H}} \in \mathbb{C}^{{N'_{{\mathrm{RF}}}} \times 1}$ denotes the effective channel vector between user $k$ and all the activated RF chains at the BS.

One important feature of mmWave communications is their vulnerability to blockage due to its higher penetration loss and deficiency of diffraction in mmWave frequency band\cite{Andrews2017}.
To capture the blockage effects in mmWave systems, a probabilistic model was proposed in \cite{Andrews2017}.
The probability for an arbitrary link with a distance $\varrho$ having LOS is given by
\begin{equation}
{P_{{\mathrm{LOS}}}}\left( \varrho \right) = {e^{{-\varrho}/{C}}},
\end{equation}
where $C = 200$ meters.
Note that, in this chapter, if the LOS path is blocked, we treat the strongest NLOS path as ${\mathbf{H}}_{k,0}$ and all the other NLOS paths as ${\mathbf{H}}_{k,l}$.

In the following, we introduce the CSI requirement and the adopted channel acquisition technique of our proposed scheme.
In addition, although our proposed design is based on the assumption of perfect channel estimation, the channel estimation error will be taken into account in the simulations and its model is introduced here to facilitate the presentation.

\noindent\underline{\textbf{CSI Requirement for the Proposed Scheme}}

The proposed NOMA user grouping and analog beamformer design requires the LOS CSI, including the AODs ${\theta _{k,0}}$, the AOAs ${\phi _{k,0}}$, and the complex path gains ${{\alpha _{k,0}}}$ of all the users.
Then, based on the obtained user grouping strategy and analog beamformer, the proposed digital precoder design depends on the estimated effective channel matrix ${\widetilde{{\bf{H}}}}$.

\noindent\underline{\textbf{Three-stage Channel Estimation}}

In this work, we assume that the channel acquisition technique in our previous work \cite{zhao2017multiuser} is adopted.
Since this work mainly proposes the beamwidth control-based hybrid mmWave NOMA scheme and focuses on the energy-efficient resource allocation design, we briefly discuss the adopted channel estimation method in the following.
A three-stage channel estimation scheme was proposed in \cite{zhao2017multiuser}, where the AODs ${\hat{\theta} _{k,0}}$ and the complex path gains ${{\hat{\alpha} _{k,0}}}$ are estimated in the first stage and the AOAs ${\hat{\phi} _{k,0}}$ is estimated in the second stage through exhaustive beam scanning.
In the third stage, each user transmits a unique orthogonal pilot to the BS and the effective channel matrix ${\widetilde{{\bf{H}}}}$ is estimated with classic minimum mean square error (MMSE) channel estimation \cite{BigueshMMSE2006}.
In this chapter, we assume the use of time division duplex (TDD) and exploit the channel reciprocity, i.e., the estimated effective channel matrix $\hat{\widetilde{{\bf{H}}}}$ can be used for digital precoder design.

\noindent\underline{\textbf{Channel Estimation Error}}

To account for the impacts of imperfect CSI on the performance of the proposed scheme, we need to consider the estimation error in all the three stages.
In particular, for the first stage, the AOD and path gain of user $k$ can be modeled as ${{\theta} _{k,0}} = {\hat{\theta} _{k,0}} + \Delta_{{\theta} _{k,0}}$ and
${{\alpha} _{k,0}} = {\hat{\alpha} _{k,0}} + \Delta_{{\alpha} _{k,0}}$, $\forall k$, respectively, where ${\hat{\theta} _{k,0}}$ denotes the estimated AOD of user $k$ at the BS, ${\hat{\alpha}_{k,0}}$ denotes the estimated LOS path gain of user $k$, $\Delta_{{\theta} _{k,0}} \sim {\cal CN}(0,\sigma^2_{{\theta} _{k,0}})$ and $\Delta_{{\alpha} _{k,0}} \sim {\cal CN}(0,\sigma^2_{{\alpha} _{k,0}})$ denote the corresponding AOD and LOS path gain estimation error, respectively.
Variables $\sigma^2_{{\theta} _{k,0}}$ and $\sigma^2_{{\alpha} _{k,0}}$ denote the variance of the AOD and LOS path gain estimation error, respectively.
In the second stage, the AOA of the BS at user $k$ can be modeled by ${{\phi} _{k,0}} = {\hat{\phi} _{k,0}} + \Delta_{{\phi} _{k,0}}$, $\forall k$,
where ${\hat{\phi} _{k,0}}$ denotes the estimated AOA of the BS at user $k$, $\Delta_{{\phi} _{k,0}} \sim {\cal CN}(0,\sigma^2_{{\phi} _{k,0}})$ denotes the AOA estimation error, and $\sigma^2_{{\phi} _{k,0}}$ denotes the variance of AOA estimation error.
In the third stage, the effective channel coefficient between user $k$ and RF chain $r$ at the BS can be modeled by ${{{\widetilde{{{h}}}_{k,r}}}} = \hat{\widetilde{{{h}}}}_{k,r} + \Delta_{{\widetilde{{{h}}}_{k,r}}}$, $\forall k,r$,
where $\hat{\widetilde{{{h}}}}_{k,r}$ denotes the estimated effective channel coefficient, $\Delta_{{\widetilde{{{h}}}_{k,r}}} \sim {\cal CN}(0,\sigma^2_{{\widetilde{{{h}}}_{k,r}}})$ denotes the  corresponding effective channel estimation error, and $\sigma^2_{{\widetilde{{{h}}}_{k,r}}}$ denotes the variance of the effective channel estimation error.
\subsection{Power Consumption Model}
The power consumption of the considered hybrid mmWave communication system mainly consists of the RF chain power, the PS driven power, as well as the transmit power for communications.
Therefore, we model the system power dissipation as follows\cite{GaoSubarray}:
\begin{equation}\label{C7:PowerConsumption}
{U_{{\mathrm{P}}}} = {N^{\mathrm{Active}}_{{\mathrm{RF}}}}{P_{{\mathrm{RF}}}} + {N^{\mathrm{Active}}_{{\mathrm{PS}}}}{P_{{\mathrm{PS}}}} + \frac{1}{\rho} \sum\limits_{k = 1}^{K} \left\|{{{\mathbf{q}}_k}}\right\|^2,
\end{equation}
where ${P_{{\mathrm{RF}}}}$ and ${P_{{\mathrm{PS}}}}$ denotes the power consumption of a RF chain and a PS, respectively.
Variables ${N^{\mathrm{Active}}_{{\mathrm{RF}}}}$ and ${N^{\mathrm{Active}}_{{\mathrm{PS}}}}$ denote the number of active RF chains and active PSs at the BS, respectively.
Besides, $0 < \rho \le 1$ is a constant which accounts for the efficiency of the power amplifier at the BS.

\section{Beamwidth Control and Analog Beamformer Design}
Clearly, the key challenge of the proposed scheme is to control the beamwidth with an effective analog beamformer design.
In this chapter, we first propose two types of beamwidth control methods based on conventional beamforming (CBF) and Dolph-Chebyshev beamforming (DCBF), for the cases with constant-modulus PSs and amplitude-adjustable PSs, respectively.
However, controlling the beamwidth will inevitably change the main lobe response.
In particular, a larger main lobe power loss could degrade the system sum-rate for a fixed transmit power, which translates into a lower system energy efficiency.
Therefore, through characterizing relationship between the main lobe power loss and the desired beamwidth, we propose an effective analog beamformer design to minimize the main lobe power loss.
\subsection{Two Beamwidth Control Methods}\label{C7:BeamWidthControlSection}

\begin{figure}[t]
\centering
\subfigure[CBF-based]
{\label{C7:BeamwidthControl:a} %% label for first subfigure
\includegraphics[width=0.45\textwidth]{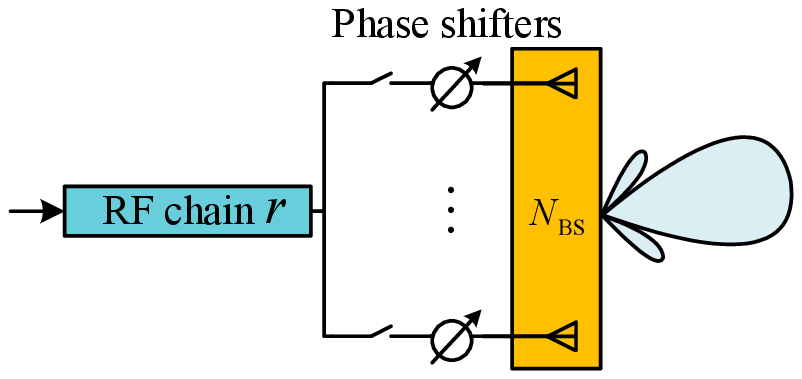}}
\subfigure[DCBF-based]
{\label{C7:BeamwidthControl:b} %% label for second subfigure
\includegraphics[width=0.45\textwidth]{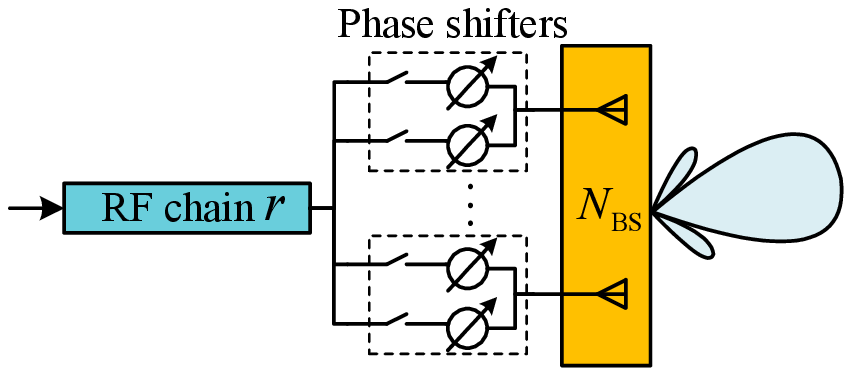}}
\caption{The PSs structures for the proposed two kinds of beamwidth control methods. The switches can activate or idle an antenna and its PS based on the adopted beamwidth control method.}
\label{C7:BeamwidthControl}%
\end{figure}

\noindent\textbf{\underline{CBF-based Beamwidth Control}}

The most natural way of beamwidth control is to reduce the number of active antennas to widen the analog beamwidth with CBF in the hybrid architecture with constant-modulus PSs, as shown in Figure \ref{C7:BeamwidthControl:a}.
Consider an ULA with $N_{\mathrm{BS}}$ antennas but only $N'_{\mathrm{BS}}$ consecutive antennas\footnote{Note that, to maintain the equal distances between adjacent antennas, we put the antennas idle one by one from the end of an antenna array.} of them are active.
The conventional beamformer to generate an analog beam steering towards AOD $\theta_0$ can be given by
\begin{equation}\label{C7:CBFWeight}
{{\bf{w}}_{{\mathrm{CBF}}}}\left(\theta_0,N'_{\mathrm{BS}}\right) =\frac{1}{\sqrt{N'_{\mathrm{BS}}}}\left[ {1,{e^{ - j\zeta\left(\theta_0\right) }},\ldots,{e^{ - j\left(N'_{\mathrm{BS}} - 1\right)\zeta\left(\theta_0\right)  }}},0,\ldots,0 \right]^{\mathrm{T}},
\end{equation}
where ${{{\bf{w}}_{{\mathrm{CBF}}}}} \in \mathbb{C}^{{N_{\mathrm{BS}}} \times 1}$, ${\left\| {{{\bf{w}}_{{\mathrm{CBF}}}}} \right\|^2} = 1$, and its last ${N_{\mathrm{BS}}} - N'_{\mathrm{BS}}$ entries are zero means that those antennas are inactive.
The main lobe direction $\theta_0$ and the number of active antennas $N'_{\mathrm{BS}}$ will be determined in Section \ref{C7:AnalogBeamformer}.
The beam response function characterizes the beamforming gain for the signal departing from an arbitrary direction $\theta$.
In particular, the beam response function of CBF can be obtained by\cite{van2002optimum}
\begin{equation}\label{C7:BeamResponseFunction}
\left|{B_{{\mathrm{CBF}}}}\left( \psi  \right)\right|^2 = \frac{1}{{N'_{\mathrm{BS}}}}\frac{{\sin^2 \left( {\frac{{N'_{\mathrm{BS}}\psi }}{2}} \right)}}{{\sin^2 \left( {\frac{\psi }{2}} \right)}},
\end{equation}
where $\psi = \frac{2\pi d}{\lambda}\left(\cos\left(\theta\right) - \cos\left(\theta_0\right)\right)$ denotes the phase difference\footnote{Note that, the beam response function in the angle domain $\left|{B_{{\mathrm{CBF}}}}\left( \theta,\theta_0  \right)\right|$  can be easily obtained via the mapping relationship $\psi = \frac{2\pi d}{\lambda}\left(\cos\left(\theta\right) - \cos\left(\theta_0\right)\right)$.} between the signals departing from directions $\theta_0$ and $\theta$.
In addition, the half-power beamwidth (i.e., -3 dB beamwidth) can be approximated by\cite{van2002optimum}
\begin{equation}\label{C7:Beamwidth3dBCBF}
{\psi^{\mathrm{CBF}}_{\mathrm{H}}} \approx 0.891\frac{2\pi}{N'_{\mathrm{BS}}}.
\end{equation}

Although the CBF-based beamwidth control is easy to implement due to its simplicity, reducing the array size decreases the beamforming gain, which translates into a lower main lobe response and a higher sidelobe power level, as shown in Figure \ref{C7:BeamPatternWithBeamWidthControl}.
Both the lower main lobe response and the higher sidelobe power level have negative impacts on the system sum-rate and thus decrease the system energy efficiency.
In particular, a lower main lobe response implies a lower power transmission efficiency which yields a lower data rate for the users at the main beam direction.
On the other hand, a higher sidelobe power level introduces a higher inter-beam interference to other users as well as NOMA groups.
Therefore, the DCBF-based beamwidth control is proposed in next section.

\begin{figure}[t]
\centering
\includegraphics[width=4.5in]{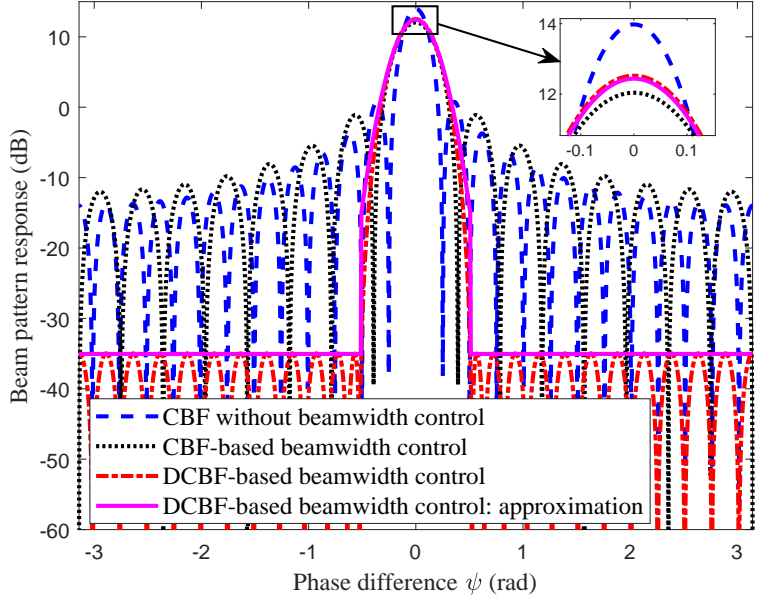}
\caption{The beam pattern response with beamwidth control for $N_{\mathrm{BS}} = 25$ and $\theta_0 = 90$ degrees. The desired half-power beamwidth of CBF-based and DCBF-based beamwidth control are the same and are 1.5 times of that for CBF without beamwidth control.}
\label{C7:BeamPatternWithBeamWidthControl}
\end{figure}

\noindent\textbf{\underline{DCBF-based Beamwidth Control}}

When the PS can adjust its amplitude adaptively, it enables the design of an optimal analog beamformer, e.g. DCBF \cite{Dolph1946}.
Note that, the capability of amplitude adjusting can be realized by
deploying \emph{one pair} of PSs at each antenna and following a similar method as in \cite{Bogale2016RFChainNumber}, as shown in Figure \ref{C7:BeamwidthControl:b}.
Given a ULA with $N_{\mathrm{BS}}$ antennas, the analog beamformer for DCBF \cite{Dolph1946} is given by
\begin{equation}\label{C7:DCBFWeight}
{{\bf{w}}_{{\mathrm{DC}}}}\left( {{\theta _0},N_{\mathrm{BS}},\Theta} \right) =  \frac{1}{{\sqrt {N_{\mathrm{BS}}} }} \left[ {\varpi _1}\left(\Theta\right),{\varpi _2}\left(\Theta\right){e^{ - j\zeta\left(\theta_0\right) }}, \ldots ,{\varpi _{N_{\mathrm{BS}}}}\left(\Theta\right){e^{ - j\left( {N_{\mathrm{BS}} -  1} \right)\zeta\left(\theta_0\right) }} \right]^{\mathrm{T}},
\end{equation}
where the optimal amplitudes $\left[{\varpi _1}\left(\Theta\right),{\varpi _2}\left(\Theta\right), \ldots ,\right.$ $\left.{\varpi _{N_{\mathrm{BS}}}}\left(\Theta\right)\right]$ across $N_{\mathrm{BS}}$ antennas are all positive real numbers.
Given $N_{\mathrm{BS}}$, each optimal amplitude coefficient ${\varpi _n}\left(\Theta\right)$ on the $n$-th antenna, $\forall n = \{1,2,\ldots,{N_{\mathrm{BS}}}\}$, is a function of the main-to-side-lobe ratio (MSLR) in voltage as $\Theta>1$.
In particular, ${\varpi _n}\left(\Theta\right)$ can be obtained by Equation (16) in \cite{Duhamel1953} according to the selected MSLR $\Theta$\footnote{Since the excitation coefficients of DCBF beamforming only depend on the selected MLSR $\Theta$ for a given $N_{\mathrm{BS}}$, one can compute the excitation coefficients in advance and refer to a look-up table for practical implementation.}.
Hence, they are treated as constants in this work with a given $N_{\mathrm{BS}}$ and MSLR $\Theta$, where the selection of $\Theta$ depends on the desired half-power beamwidth and will be detailed in Section \ref{C7:MLPvsBeamWidth}.
Since we consider the normalized analog beamformer, i.e., ${\left\| {{{\bf{w}}_{{\mathrm{DC}}}}} \right\|^2} = 1$, we have $\frac{{\varpi _n}\left(\Theta\right)}{{\sqrt {N_{\mathrm{BS}}} }} <= 1$, $\forall n = \{1,2,\ldots,{N_{\mathrm{BS}}}\}$.
According to Equation (11) in \cite{Bogale2016RFChainNumber}, the $n$-th entry of ${{{\bf{w}}_{{\mathrm{DC}}}}}$ can be decomposed as
\begin{equation}
\frac{{\varpi _n}\left(\Theta\right)}{{\sqrt {N_{\mathrm{BS}}} }}{e^{ - j\left( {n - 1} \right)\zeta\left(\theta_0\right) }} = {e^{ - j\left[ {\left( {n - 1} \right)\zeta\left(\theta_0\right) - {{\cos }^{ - 1}}\left( {\frac{{{\varpi _n}\left( \Theta  \right)}}{{2\sqrt {{N_{{\mathrm{BS}}}}} }}} \right)} \right]}} + {e^{ - j\left[ {\left( {n - 1} \right)\zeta\left(\theta_0\right) + {{\cos }^{ - 1}}\left( {\frac{{{\varpi _n}\left( \Theta  \right)}}{{2\sqrt {{N_{{\mathrm{BS}}}}} }}} \right)} \right]}}.
\end{equation}
We can observe that changing the amplitude of each antenna can be implemented via setting the phases for a pair of constant-modulus PSs attached with this antenna.
Note that DCBF-based beamwidth control does not idle any antennas.
Instead, it can change the beamwidth through varying the amplitude distribution ${\varpi _n}\left(\Theta\right)$ across the antenna array\footnote{We note that the expense of adjusting the amplitude of PSs is to double the power cost for each activated antenna since it is driven by a pair of PSs.
However, the power consumption of a PS is generally much smaller than that of a RF chain, e.g. 250 times smaller \cite{GaoSubarray}.
Therefore, controlling the beamwidth via DCBF {is} still a viable method to improve the system energy-efficiency.} based on the selected $\Theta$ \cite{Duhamel1953}.
According to \cite{Drane1968DC}, the normalized beam response function of DCBF is given by
\begin{equation}\label{C7:BeamFunctionDC}
\left|B_{\mathrm{DC}}\left( \psi  \right)\right|^2 = \frac{{{T^2_M}\left( {a\cos \psi  + b} \right)}}{{{T^2_M}\left( {{z_0}} \right)}} = \frac{{{T^2_M}\left( {a\cos \psi  + b} \right)}}{\Theta^2},
\end{equation}
with ${T_M}\left( \cdot \right)$ denoting the $M$-th degree Chebyshev polynomial of the first kind and is given by
\begin{equation}
{T_M}\left( x \right) = \left\{ {\begin{array}{*{20}{c}}
{\cos \left( {M{{\cos }^{ - 1}}\left( x \right)} \right)},&{\left| x \right| \le 1},\\
{\cosh \left( {M{{\cosh }^{ - 1}}\left( x \right)} \right)},&{\left| x \right| > 1},
\end{array}} \right.
\end{equation}
where $a = \frac{z_0+1}{2} > 1$, $b = \frac{z_0-1}{2} >0$, $N_{\mathrm{BS}} = 2M+1$, and\footnote{To facilitate the presentation, we consider a symmetric ULA containing an odd number of antennas in the sequel.
For the case with an even number of antennas, i.e., $N_{\mathrm{BS}} = 2M$, the beam response function is slightly changed and the details can be found in \cite{Dolph1946}.}
\begin{equation}\label{C7:SLL}
{z_0} = {{T^{-1}_M}\left( {\Theta} \right)} = \cosh \left( {\frac{1}{M}{{\cosh }^{ - 1}}\left( \Theta \right)} \right) > 1.
\end{equation}
When ${T^2_M}\left( {a\cos \psi  + b} \right) = \frac{\Theta^2}{{2}}$, we have $\left|B_{\mathrm{DC}}\left( \psi  \right)\right|^2 = \frac{1}{{2 }}$ and thus the half-power beamwidth of DCBF is given as follows:
\begin{equation}\label{C7:Beamwidth3dBDC}
{\psi^{\mathrm{DC}}_{\mathrm{H}}} = 2 {{{\cos }^{ - 1}}\left\{ {\frac{1}{a}\cosh \left[ {\frac{1}{M}{{\cosh }^{ - 1}}\left( {\frac{\Theta}{{\sqrt 2 }}} \right)} \right] - \frac{b}{a}} \right\}}.
\end{equation}
Similarly, through ${{T^2_M}\left( {a\cos \psi  + b} \right)} = 1$, the main lobe beamwidth\footnote{The main lobe beamwidth is defined as the beamwidth from the first left peak sidelobe to the first right peak sidelobe in the beam pattern response.} of DCBF can be obtained by
\begin{equation}\label{C7:MainBeamWidthDC}
{\psi^{\mathrm{DC}}_{\mathrm{\Theta}}} = 2 {{{\cos }^{ - 1}}\left( {\frac{{3 - {z_0}}}{{1 + {z_0}}}} \right)}.
\end{equation}
The beam pattern response of DCBF-based beamwidth control is illustrated in Figure \ref{C7:BeamPatternWithBeamWidthControl}.
With the same required half-power beamwidth, we can observe a higher main lobe response and a substantially lower sidelobe power level for the DCBF-based beamwidth control compared to the CBF-based one.
In fact, it has been proved that the DCBF is the optimal analog beamforming \cite{Dolph1946} in the sense that: 1) given the specified sidelobe power level, the beamwidth is the narrowest; or 2) given the fixed beamwidth, the sidelobe power level is minimized.
Therefore, the DCBF-based beamwidth control is a promising technique to accompany a NOMA group and to reduce the interference signal leakage.

Although the two proposed beamwidth control methods both incur a main lobe power loss, our formulated energy-efficient resource allocation design can still guarantee the quality of service (QoS) of each user within the expected fixed communication range, by considering the minimum data rate requirement C2 in \eqref{C7:ResourceAllocationOriginal} in Section \ref{C7:EEMax}.
Compared to conventional mmWave NOMA scheme, the proposed beamwidth control scheme can facilitate the formation of more NOMA groups.
As a result, more RF chains can be set idle which saves circuit power consumption for improving the system energy efficiency.

\subsection{Main Lobe Response versus the Half-power Beamwidth}\label{C7:MLPvsBeamWidth}

\noindent\textbf{\underline{{CBF-based Beamwidth Control}}}\

According to \eqref{C7:BeamResponseFunction}, the main lobe response for CBF can be easily obtained by
\begin{equation}\label{C7:MLR}
G_{\mathrm{CBF}} = \left|{B_{{\mathrm{CBF}}}}\left( 0 \right)\right|^2 = N'_{\mathrm{BS}}.
\end{equation}
Note that \eqref{C7:MLR} is obtained via $\mathop {\lim }\limits_{x \to 0} \mathrm{sinc}\left( x \right) = 1$.
Combining \eqref{C7:MLR} with \eqref{C7:Beamwidth3dBCBF}, we have
\begin{equation}\label{C7:MLRVersus3dBBeamwdith}
G_{\mathrm{CBF}} \approx 0.891\frac{2\pi}{\psi^{\mathrm{CBF}}_{\mathrm{H}}}.
\end{equation}
It can be observed from \eqref{C7:MLRVersus3dBBeamwdith} that the main lobe response $G_{\mathrm{CBF}}$ is a monotonically decreasing function with an increasing half-power beamwidth for CBF.
In addition, given the desired half-power beamwidth ${\psi^{\mathrm{CBF}}_{\mathrm{H}}}$ for CBF-based beamwidth control, we can determine the number of activated antennas via combining \eqref{C7:MLR} and \eqref{C7:MLRVersus3dBBeamwdith}.

\noindent\textbf{\underline{{DCBF-based Beamwidth Control}}}

Equation \eqref{C7:BeamFunctionDC} only provides a normalized beam response function for DCBF, i.e., $\left|{B_{{\mathrm{DC}}}}\left( 0  \right) \right|^2= 1$ with $\psi = 0$ and ${\theta = \theta_0}$, and its absolute main lobe response ${G_{\mathrm{DC}}}$ is difficult to obtained.
Therefore, we follow the practical 3GPP two-dimension directional antenna pattern approximation model\cite{MIMO3GPP} to approximate $\left|B_{\mathrm{DC}}\left( \psi  \right)\right|^2$ as follows:
\begin{equation}\label{C7:BeamResponseFunctionApproximation}
{\left|B'_{\mathrm{DC}}\left( \psi  \right)\right|^2}  = \left\{ {\begin{array}{*{20}{c}}
{{G_{\mathrm{DC}}}\cdot{{10}^{ - \frac{3}{{10}}{{\left( {\frac{{2\psi }}{{{\psi^{\mathrm{DC}} _{\mathrm{H}}}}}} \right)}^2}}},}&{{\left| \psi  \right|}  \le \frac{{{\psi^{\mathrm{DC}} _{\mathrm{\Theta}}}}}{2},}\\
{\frac{{G_{\mathrm{DC}}}}{\Theta^2},}&{{\left| \psi  \right|}  > \frac{{{\psi^{\mathrm{DC}} _{\mathrm{\Theta}}}}}{2}}
\end{array}} \right.,
\end{equation}
where an exponential function is used to approximate the beam response function in the main lobe and a constant level is adopted for modeling the sidelobe.
Note that, one important feature of DCBF is its constant sidelobe power level\cite{Dolph1946} due to the special structure of the Chebyshev polynomial.
Therefore, one can imagine that the adopted approximation model in \eqref{C7:BeamResponseFunctionApproximation} can provide a high approximation accuracy.
The main lobe response ${G_{\mathrm{DC}}}$ for DCBF can be determined by utilizing the fact that the total radiated powers of CBF and DCBF are the same, i.e.,
\begin{align}
&\int_{ - \pi }^\pi  {{{\left| {{B'_{{\mathrm{DC}}}}\left( \psi  \right)} \right|}^2}} d\psi
= \int_{ - \pi }^\pi  {{{\left| {{B_{{\mathrm{CBF}}}}\left( \psi  \right)} \right|}^2}} d\psi  = {{2\pi}}\label{C7:RadiatedPowerDCBF}\\
& \approx {G_{\mathrm{DC}}}\left[  {\frac{{2\pi  - {\psi^{\mathrm{DC}}_{\mathrm{\Theta}}}}}{{{\Theta^2}}} + \sqrt {\frac{{5\pi \left(\psi^{{\mathrm{DC}}}_{\mathrm{H}}\right)^2}}{{6\ln 10}}} }{{\mathrm{erf}}\left( {\sqrt {\frac{{6\ln 10}}{{5 \left(\psi^{{\mathrm{DC}}}_{\mathrm{H}}\right)^2}}} + \frac{{{\psi^{\mathrm{DC}}_{\mathrm{\Theta}}}}}{2}} \right)} \right],\notag
\end{align}
where $\mathrm{erf}\left( \cdot \right)$ denotes the error function.
Now, we have
\begin{equation}\label{C7:MainBeamResponseDC}
{G_{\mathrm{DC}}} \approx {\frac{{2\pi}}{{\left[  {\frac{{2\pi  - {\psi^{\mathrm{DC}}_{\mathrm{\Theta}}}}}{{{\Theta^2}}} + \sqrt {\frac{{5\pi \left(\psi^{{\mathrm{DC}}}_{\mathrm{H}}\right)^2}}{{6\ln 10}}} }{{\mathrm{erf}}\left( {\sqrt {\frac{{6\ln 10}}{{5 \left(\psi^{{\mathrm{DC}}}_{\mathrm{H}}\right)^2}}}  + \frac{{{\psi^{\mathrm{DC}}_{\mathrm{\Theta}}}}}{2}} \right)} \right]}}}.
\end{equation}

Given the desired half-power beamwidth ${\psi^{\mathrm{DC}}_{\mathrm{H}}}$ for DCBF-based beamwidth control, we can obtain beam parameters, such as the main lobe beamwidth ${\psi^{\mathrm{DC}}_{\mathrm{\Theta}}}$, the MSLR $\Theta$, and the main lobe response ${G_{\mathrm{DC}}}$, via combining \eqref{C7:SLL}, \eqref{C7:Beamwidth3dBDC}, \eqref{C7:MainBeamWidthDC}, and \eqref{C7:MainBeamResponseDC}.
As a result, the beam response approximation function of DCBF in \eqref{C7:BeamResponseFunctionApproximation} can be determined, as shown in Figure \ref{C7:BeamPatternWithBeamWidthControl}.
We can observe a tight approximation for the beam pattern response of DCBF in both the main lobe and the sidelobe.

\begin{figure}[t]
\centering
\includegraphics[width=4.5in]{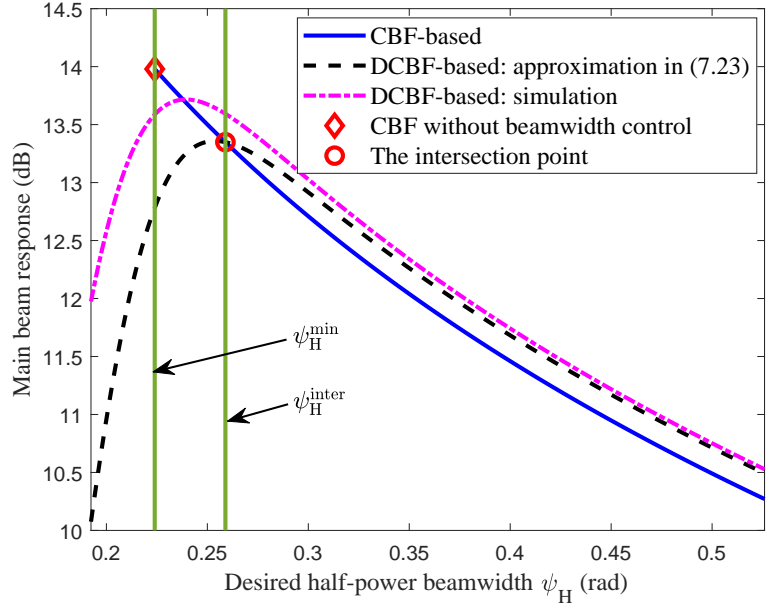}
\caption{The main lobe response versus the half-power beamwidth with $N_{\mathrm{BS}} = 25$ and $\theta_0 = 90$ degrees.}
\label{C7:MainBeamResponseBeamWidthControl}
\end{figure}

\subsection{Beamwidth Control Methods Comparison}
The trade-off between the main lobe response and the desired half-power beamwidth has been analyzed in Equations \eqref{C7:MLRVersus3dBBeamwdith} and \eqref{C7:MainBeamResponseDC} for the proposed CBF-based and DCBF-based beamwidth control, respectively.
To further visualize this trade-off, the main lobe responses versus the desired half-power beamwidth ${\psi_{\mathrm{H}}}$ of CBF-based and DCBF-based beamwidth control are compared in Figure \ref{C7:MainBeamResponseBeamWidthControl}.
We can observe that the main lobe response for the CBF-based beamwidth control $G_{\mathrm{CBF}}\left({\psi_{\mathrm{H}}}\right)$ is monotonically decreasing with an increasing ${\psi_{\mathrm{H}}}$. %
In contrast, for DCBF-based beamwidth control, ${G_{\mathrm{DC}}}\left({\psi_{\mathrm{H}}}\right)$ first increases and then decreases with ${\psi_{\mathrm{H}}}$.
In fact, decreasing the number of active antennas for the CBF-based beamwidth control always introduces a power loss.
The wider the beam, the larger the power loss, as predicted in \eqref{C7:MLRVersus3dBBeamwdith}.
On the other hand, forcing the system to shape a beam with an exceedingly narrow or wide beamwidth reduces the flexibility in DCBF analog beamformer design.
Hence, a relatively high power loss is anticipated in these two cases.
Furthermore, it can be seen from Figure \ref{C7:MainBeamResponseBeamWidthControl} that a high accuracy can be achieved in approximating the main lobe response as a function of desired half-power beamwidth for DCBF-based beamwidth control in \eqref{C7:MainBeamResponseDC}, especially for the regime with a large half-power beamwidth.
It is worth pointing out that for the proposed two beamwidth control methods, the beam gain reduction due to beamwidth widening can be as large as $3 \sim 4$ dB as shown in Figure \ref{C7:MainBeamResponseBeamWidthControl}.
However, broadening the beamwidth can increase the number of served NOMA groups, which enables a higher system energy efficiency since more RF chains are switched to idle to reduce the system energy consumption, as demonstrated in Section \ref{C7:SimulationResults}.

There are two important pieces of information in Figure \ref{C7:MainBeamResponseBeamWidthControl}, which will be adopted to serve as guidelines for the analog beamformer design in the next section.
The first important point of the CBF-based beamwidth control (marked with $\lozenge$) is obtained when all the ${N_{\mathrm{BS}}}$ antennas are active, i.e., $\psi _{\mathrm{H}}^{\mathrm{min}} = 0.891\frac{2\pi}{N_{\mathrm{BS}}}$.
We denote the half-power beamwidth associated with this point with a superscript ${\mathrm{min}}$, since this is the minimum beamwidth can be created without applying any beamwidth control.
Another point\footnote{{Note that, the beamwidth $\psi _{\mathrm{H}}^{\mathrm{min}}$ and $\psi _{\mathrm{H}}^{\mathrm{inter}}$ at the considered two points in Figure \ref{C7:MainBeamResponseBeamWidthControl} can be obtained explicitly from Equations \eqref{C7:Beamwidth3dBCBF} and \eqref{C7:RootInterSect}, respectively.}} (marked with $\circ$) is the intersection point $\psi _{\mathrm{H}}^{\mathrm{inter}}$ where CBF and DCBF have the same main lobe response, i.e., $G_{\mathrm{DC}}\left({\psi _{\mathrm{H}}}\right) = G_{\mathrm{CBF}}\left({\psi _{\mathrm{H}}}\right)$.
When the desired half-power beamwidth is smaller than $\psi _{\mathrm{H}}^{\mathrm{inter}}$, i.e., ${{\psi _{\mathrm{H}}}} < {\psi _{\mathrm{H}}^{\mathrm{inter}}}$, the CBF-based beamwidth control outperforms the DCBF-based one in terms of the power of the main lobe response.
On the other hand, the DCBF-based beamwidth control outperforms the CBF-based one when ${{\psi _{\mathrm{H}}}} > {\psi _{\mathrm{H}}^{\mathrm{inter}}}$.
In the following theorem, we show that the intersection point always exists and there is only one intersection point.
\begin{Thm}\label{C7:Theorem1}
The intersection point between $G_{\mathrm{DC}}\left({\psi _{\mathrm{H}}}\right)$ and $G_{\mathrm{CBF}}\left({\psi _{\mathrm{H}}}\right)$ can be obtained by solving the root of the following equation:
\begin{equation}\label{C7:RootInterSect}
{\frac{{2\pi}}{{\left[ {\frac{{2\pi  - {\psi^{\mathrm{DC}}_{\mathrm{\Theta}}}}}{{{\Theta^2}}} + \sqrt { \frac{{5\pi \psi^{2}_{\mathrm{H}}}}{{6\ln {10}}}} }{{\mathrm{erf}}\left( {\sqrt { \frac{{6\ln {10}}}{{5 \psi^{2}_{\mathrm{H}}}}}  +  \frac{{{\psi^{\mathrm{DC}}_{\mathrm{\Theta}}}}}{2}} \right)} \right]}}} =  0.891\frac{2\pi}{\psi _{\mathrm{H}}}.
\end{equation}
There exists only one root of the equation in \eqref{C7:RootInterSect}.
\end{Thm}
\begin{proof}
Please refer to Appendix \ref{C7:AppendixD} for a proof of Theorem \ref{C7:Theorem1}.
\end{proof}

The existence of the unique intersection point $\psi _{\mathrm{H}}^{\mathrm{inter}}$ implies that there is always a threshold which helps to determine the superior beamwidth control method in terms of the main lobe power gain based on the AODs of the two NOMA users.

\subsection{Analog Beamformer Design}\label{C7:AnalogBeamformer}
In a multi-RF chain system, inter-beam interference exists due to concurrent transmission.
In practice, the interference can be suppressed via using a digital precoder ${\bf{Q}}$, which will be designed in Section \ref{C7:EEMax}.
Moreover, we note that, for multi-RF chain systems, the inter-beam interference diminishes with increasing the number of antennas at the BS due to the array gain.
In other words, for ${N_{\mathrm{BS}}} \to \infty$, a BS equipped with multiple independent single-RF chains can mimic a BS with multi-RF chains.
Therefore, in this section, we study an effective analog beamformer design to minimize the main lobe power loss of the single-RF chain system, which serves as a guideline for the design of efficient analog beamformers in multi-RF chain systems\footnote{In fact, the optimal analog beamformer design in a multi-RF chain system needs to take into account both the inter-beam interference and user scheduling, which is generally intractable.}.
As shown in \eqref{C7:CBFWeight} and \eqref{C7:DCBFWeight} in Section \ref{C7:BeamWidthControlSection}, the analog beamformer depends on the main beam direction $\theta_0$ and the desired half-power beamwidth $\psi_{\mathrm{H}}$.
In the following, we first deal with searching for an efficient beam direction for a given arbitrary analog beamformer.
Then, the optimal desired half-power beamwidth $\psi_{\mathrm{H}}$ can be obtained to minimize the main lobe power loss\footnote{In this chapter, after user ordering, we first perform beam rotation for a given arbitrary analog beamformer and then control the beamwidth. Note that swapping the order of beam rotation and beamwidth control does not change the result of the analog beamformer design and does not incur a system performance degradation, since they are independent with each other.}.

\noindent\textbf{\underline{{Main Beam Direction}}}

Although the beamwidth control methods have been proposed, characterized, and compared, how to select an appropriate main beam direction $\theta_0$ for the generated analog beam is still unknown and it is an important issue for the proposed scheme.
In particular, rotating the analog beam of a NOMA group changes the effective channel gains of the two NOMA users, and thus affects the SIC decoding order and the sum-rate of this NOMA group.
Since a normalized analog beamformer is employed in this chapter, i.e., ${\left\| {{{\bf{w}}_{{\mathrm{CBF}}}}} \right\|^2} = {\left\| {{{\bf{w}}_{{\mathrm{DC}}}}} \right\|^2} = 1$, changing the main beam direction does not change the system power consumption.
As a result, in this section, we aim to design the asymptotically optimal main beam direction to maximize the sum-rate of a NOMA group, so as to improve the system energy efficiency.

Since only LOS CSI is required for analog beamformer design in our proposed scheme, we assume a pure LOS mmWave channel in this section, i.e., ${{\bf{H}}_k} = {\alpha _{k,0}}{{\bf{H}}_{k,0}}$.
Note that it is a reasonable assumption since the LOS path gain is much stronger than the NLOS path gain \cite{Rappaport2013}.
In this section, for the considered single-RF chain system in pure LOS channel, the RF chain index $r$ is omitted and the index of LOS path $l = 0$ is fixed.

To design the optimal beam direction of a NOMA group, consider a single-RF chain system serving two downlink NOMA users, i.e., user 1 and user 2, where user 1 is the strong user and user 2 is the weak user, i.e., ${{\left| {{\alpha _{2,0}}} \right|}}\le {{\left| {{\alpha _{1,0}}} \right|}}$.
Through beamwidth control, an analog beam is generated and widened such that the two users are covered and NOMA transmission is performed within the beam.
In the proposed scheme, changing its main beam direction $\theta_0$ of the generated analog beam may reverse the effective channel gain order of the two users, which in turn changes the SIC decoding order and the sum-rate.
For the first case, if a selected main beam direction $\widetilde{\theta}_0$ does not change the channel gain order of the two users, i.e., $\left| {\widetilde h_1} \left(\widetilde{\theta}_0\right) \right| \ge \left| {\widetilde h_2} \left(\widetilde{\theta}_0\right) \right|$, where ${\widetilde{h}_1}\left(\widetilde{\theta}_0\right)$ and ${\widetilde{h}_2}\left(\widetilde{\theta}_0\right)$ denote the effective channels of user 1 and user 2, respectively.
According to the downlink NOMA protocol\cite{WeiTCOM2017}, the achievable sum-rate of both users is given by
\begin{equation}\label{C7:SumRate1}
{R}_{{\mathrm{sum}}}\left(\widetilde{\theta}_0\right) = \mathcal{W}{\log _2}\left( {1 + \frac{{{p_1}{{\left| {\widetilde h_1} \left(\widetilde{\theta}_0\right) \right|}^2}}}{{\mathcal{W}{N_{0}}}}} \right) + \mathcal{W}{\log _2}\left( {1 + \frac{{{p_2}{{\left| {\widetilde h_2} \left(\widetilde{\theta}_0\right) \right|}^2}}}{{{p_1}{{\left| {\widetilde h_2}\left(\widetilde{\theta}_0\right) \right|}^2} + \mathcal{W}{N_{0}}}}} \right),
\end{equation}
where the two terms denote the achievable data rates of user 1 and user 2, respectively\footnote{In the considered single-RF chain system, the SIC decoding at the strong user, i.e., user 1, can be guaranteed when $\left| {\widetilde h_1} \right| \ge \left| {\widetilde h_2} \right|$\cite{Tse2005}.}, and $\mathcal{W}$ denotes the system bandwidth.
On the other hand, in the second case, if a selected main beam direction $\overline{\theta}_0$ does reverse the channel gain order of the two users, i.e., $\left| {\widetilde h_2}\left(\overline{\theta}_0\right) \right| \ge \left| {\widetilde h_1} \left(\overline{\theta}_0\right)\right|$, similarly, the achievable sum-rate can be written as:
\begin{equation}\label{C7:SumRate2}
\overline{{R}}_{{\mathrm{sum}}} \left(\overline{\theta}_0\right) = \mathcal{W}{\log _2}\left( {1 + \frac{{{p_2}{{\left| {\widetilde h_2}\left(\overline{\theta}_0\right) \right|}^2}}}{{\mathcal{W}{N_{0}}}}} \right)+ \mathcal{W}{\log _2}\left( {1 + \frac{{{p_1}{{\left| { \widetilde h_1}\left(\overline{\theta}_0\right) \right|}^2}}}{{{p_2}{{\left| { \widetilde h_1}\left(\overline{\theta}_0\right) \right|}^2} + \mathcal{W}{N_{0}}}}} \right).
\end{equation}
Now, we study the asymptotically optimal main beam direction in the large number of antennas regime.
The main results are summarized in the following theorem.

\begin{Thm}\label{C7:BeamDirection}
Considering a single-RF chain BS equipped with a massive antenna array serving two downlink NOMA users.
For a given arbitrary analog beamformer and a fixed total transmit power: 1): the main beam direction keeping the original channel gain order outperforms the one that reversing the channel gain order in terms of system sum-rate, i.e., $\mathop {\max }\limits_{{\theta}_0}  \mathop {\lim }\limits_{N_{\mathrm{BS}} \to \infty } {R}_{{\mathrm{sum}}}\left( {\theta}_0  \right) \ge \mathop {\max }\limits_{{\theta}_0}  \mathop {\lim }\limits_{N_{\mathrm{BS}} \to \infty } \overline{{R}}_{{\mathrm{sum}}}\left( {\theta}_0  \right)$;  2): steering the beam to the strong user is asymptotically optimal in terms of the system sum-rate in the large number of antennas regime, i.e., $\theta_{1,0} = \arg\mathop {\max }\limits_{{\theta}_0}  \mathop {\lim }\limits_{N_{\mathrm{BS}} \to \infty } R_{{\mathrm{sum}}}\left( {\theta}_0  \right)$.
\end{Thm}
\begin{proof}
Please refer to Appendix \ref{C7:AppendixD2} for a proof of Theorem \ref{C7:BeamDirection}.
\end{proof}

\begin{remark}
In the literature, e.g. \cite{Dingtobepublished}, it has been analytically shown that the performance gain of NOMA over OMA is enlarged when the paired two NOMA users have distinctive channel gains.
It is consistent with our conclusion shown in Theorem \ref{C7:BeamDirection}, where steering the beam to the strong user will maximize the difference of effective channel gains between the two NOMA users.
Note that the optimal main beam direction derived in Theorem \ref{C7:BeamDirection} can be easily verified via simulations using the exhaustive-based search algorithm.
\end{remark}

\noindent\textbf{\underline{{The Optimal Desired Half-power Beamwidth}}}

After determining the asymptotically optimal main beam direction for an arbitrary analog beam, the optimal desired half-power beamwidth can be obtained based on the phase difference $\psi_{12} = \frac{2\pi d}{\lambda}\left(\cos\left(\theta_{1,0}\right) - \cos\left(\theta_{2,0}\right)\right)$ between user 1 and user 2.
In particular, according to the comparison between $\psi_{12}$ and $\psi _{\mathrm{H}}^{\mathrm{inter}}$ as well as ${\psi^{\mathrm{min}} _{\mathrm{H}}}$, we discuss three cases for analog beamformer design to minimize the main lobe power loss in the following:
\begin{enumerate}[(a)]
  \item When $\left|{\psi _{12}}\right| \le \frac{{\psi^{\mathrm{min}} _{\mathrm{H}}}}{2}$, increasing the beamwidth is not necessary.
      The analog beamformer is given by CBF with the whole antenna array ${\bf{w}} = {{\bf{w}}_{{\mathrm{CBF}}}}\left( {{\theta _{1,0}},N_{{\mathrm{BS}}}} \right)$ and there is no power loss at the main beam direction.
  \item When $\frac{{\psi^{\mathrm{min}} _{\mathrm{H}}}}{2} < \left|{\psi _{12}}\right| \le \frac{\psi _{\mathrm{H}}^{\mathrm{inter}}}{2}$, CBF-based beamwidth control incurs less power loss compared to the DCBF-based one.
      Therefore, we need to reduce the array size such that $\frac{\psi _{\mathrm{H}}}{2} = 0.891\frac{\pi}{N'_{{\mathrm{BS}}}} \ge \left|{\psi _{12}}\right|$.
      As shown in \eqref{C7:MLRVersus3dBBeamwdith}, the main lobe response of CBF is monotonically decreased with increasing $\psi _{\mathrm{H}}$.
      The optimal number of active antennas for the CBF-based beamwidth control should be ${N'_{\mathrm{BS}}} = \lfloor 0.891\frac{\pi}{\left|{\psi _{12}}\right|}\rfloor$ to minimize the main lobe power loss and the analog beamformer is given by ${\bf{w}} = {{\bf{w}}_{{\mathrm{CBF}}}}\left( {{\theta _{1,0}},N'_{\mathrm{BS}}} \right)$.
  \item When $\left|{\psi _{12}}\right| > \frac{\psi _{\mathrm{H}}^{\mathrm{inter}}}{2}$, performing DCBF to control the beamwidth is preferred due to its larger main lobe response compared to the CBF-based one in this regime.
      Also, as shown in Figure \ref{C7:MainBeamResponseBeamWidthControl}, we can observe that the main lobe response of DCBF is monotonically decreasing with the desired half-power beamwidth $\psi _{\mathrm{H}}$ for $\left|{\psi _{12}}\right| > \frac{\psi _{\mathrm{H}}^{\mathrm{inter}}}{2}$.
      As a result, we set $\psi _{\mathrm{H}} = 2\left|{\psi _{12}}\right|$ to minimize the main lobe power loss.
      The analog beamformer is given by ${\bf{w}} = {{\bf{w}}_{{\mathrm{DC}}}}\left( {{\theta _{1,0}},{N_{{\mathrm{BS}}}},\Theta} \right)$, where $\Theta$ is obtained based on \eqref{C7:Beamwidth3dBDC} for $\psi _{\mathrm{H}} = 2\left|{\psi _{12}}\right|$.
\end{enumerate}

Note that the proposed analog beamformer design exploits the advantages of both the two proposed beamwidth control methods in the whole phase difference regime.
Furthermore, as the proposed analog beamformer design depends on the phase difference between two NOMA users, the user scheduling strategy would affect the analog beamformer design when extending to multi-RF chain systems.
\section{Energy-efficient Resource Allocation Design}\label{C7:EEMax}
In this section, the energy-efficient resource allocation strategy is designed for the proposed beamwidth-controlled mmWave NOMA scheme in a multi-RF chain system.
\subsection{Problem Formulation}

\noindent\textbf{\underline{Analog Beamformer in Multi-RF Chain Systems}}

Adopting the analog beamformer design for a single-RF chain system, there are four options for the analog beamformer of RF chain $r$ to be selected according to the user scheduling variable $u_{k,r}$.
These four options are listed as follows:
\begin{align}\label{C7:DCBFAnalogBeamformer}
{{\bf{w}}_r}\left({\mathbf{u} _r}\right) = \left\{  \begin{array}{ll}
{{\bf{w}}_{{\mathrm{CBF}}}}\left( {{\theta _{k',0}},{N_{{\mathrm{BS}}}}} \right),&\text{if}\;\sum\limits_{k = 1}^K {{u_{k,r}}}  = 1\;\text{and}\;{u_{k',r}} = 1\\[0mm]
{{\bf{w}}_{{\mathrm{CBF}}}}\left( {{\theta _{k',0}},{N_{{\mathrm{BS}}}}} \right),&\text{if}\;\sum\limits_{k = 1}^K {{u_{k,r}}}  = 2,{u_{k',r}} = {u_{j',r}} = 1,\\
&k' < j',\text{and}\;\left| {{\psi _{k'j'}}} \right| \le \frac{{\psi _{\mathrm{H}}^{{\mathrm{min}}}}}{2},\\[0mm]
{{\bf{w}}_{{\mathrm{CBF}}}}\left( {{\theta _{k',0}},N'_{r,\mathrm{BS}} } \right),&
\text{if}\;\sum\limits_{k = 1}^K {{u_{k,r}}}  = 2,{u_{k',r}} = {u_{j',r}} = 1,\\
&k' < j',\text{and}\;\frac{{\psi _{\mathrm{H}}^{{\mathrm{min}}}}}{2} < \left| {{\psi _{k'j'}}} \right| \le \frac{{\psi _{\mathrm{H}}^{{\mathrm{inter}}}}}{2}, \\[0mm]
{\bf{w}}_{{\mathrm{DC}}}\left( {{\theta _{k',0}},{N_{{\mathrm{BS}}}},\Theta } \right),&
\text{if}\;\sum\limits_{k = 1}^K {{u_{k,r}}}  = 2,{u_{k',r}} = {u_{j',r}} = 1, \\
&k' < j',\text{and}\;\left| {{\psi _{k'j'}}} \right| > \frac{{\psi _{\mathrm{H}}^{{\mathrm{inter}}}}}{2},
\end{array} \right.
\end{align}
where ${\mathbf{u} _r} =  \left[ {{{u} _{1,r}},{{u} _{2,r}}, \ldots ,{{u} _{K,r}}} \right]^{\mathrm{T}} \in \mathbb{B}^{K\times 1}$ collects the scheduling variables associated with RF chain $r$.
In \eqref{C7:DCBFAnalogBeamformer}, the first option means that CBF with the whole antenna array is employed for the case of OMA, i.e., $\sum\limits_{k = 1}^K {{u_{k,r}}}  = 1$.
The second, third, and fourth options follow the analog beamformer design for a NOMA group as discussed in the single-RF chain system.
Note that, for the third option, the number of active antennas associated with RF chain $r$ is $N'_{r,\mathrm{BS}} = \lfloor 0.891\frac{\pi}{\left|{\psi _{k'j'}}\right|}\rfloor$.
The calculation of MSLR $\Theta$ is based on \eqref{C7:Beamwidth3dBDC} for ${{\psi _{\mathrm{H}}} = 2\left| {{\psi _{k'j'}}} \right|}$ in the fourth option.

\noindent\textbf{\underline{{Energy Efficiency in Multi-RF Chain Systems}}}

The system energy efficiency of the proposed mmWave NOMA scheme in a multi-RF chain system is formulated in the following.
A fixed SIC decoding order $1,2,\ldots,K$ is adopted.
If user $k$ and user $i$ are scheduled to form a NOMA group associated with RF chain $r$, $\forall k<i$, user $k$ first decodes the messages of user $i$ before decoding its own information and the corresponding achievable data rate is given by
\begin{equation}\label{C7:IndividualRate2}
R_{k,i,r}=  \mathcal{W}{\log _2}\left( {1 +  \frac{{{u_{i,r}}{u_{k,r}}{{\left| {{{ {\widetilde{{\mathbf{h}}}_k}}^{\mathrm{H}}}{\bf{q}}_i} \right|}^2}}}{{I_{k,r}^{{\mathrm{inter}}} + I_{k,i,r}^{{\mathrm{intra}}} + \mathcal{W}{N_{0}}}}} \right),\forall k<i,r,
\end{equation}
where $I_{k,r}^{{\mathrm{inter}}} = \sum\limits_{r' \ne r}^{{N'_{{\mathrm{RF}}}}} {\sum\limits_{d = 1}^K {{u_{d,r'}}{{\left| {{{ {\widetilde{\bf{ h}}_k} }^{\mathrm{H}}}{\bf{q}}_d} \right|}^2}} }$ denotes the inter-beam interference power faced by user $k$ associated with RF chain $r$ and $I_{k,i,r}^{{\mathrm{intra}}} = \sum\limits_{d = 1}^{i - 1} {{u_{d,r}}{{\left| {{{ {\widetilde{\bf{ h}}_k}}^{\mathrm{H}}}{\bf{q}}_d} \right|}^2}}$ denotes the intra-beam interference power when decoding the message of user $i$ at user $k$.
After SIC decoding, the achievable rate of user $k$ associated with RF chain $r$ is given by:
\begin{equation}\label{C7:DLIndividualRate1}
R_{k,k,r} = \mathcal{W}{\log _2}\left( {1 + \frac{{{u_{k,r}}{{\left| {{{ {\widetilde{\bf{ h}}_k}}^{\mathrm{H}}}{\bf{q}}_k} \right|}^2}}}{{I_{k,r}^{{\mathrm{inter}}} +  I_{k,k,r}^{{\mathrm{intra}}} + \mathcal{W}{N_{0}}}}} \right),\forall k,r,
\end{equation}
where $I_{k,k,r}^{{\mathrm{intra}}} = \sum\limits_{d = 1}^{k - 1} {{u_{d,r}}{{\left| {{{ {\widetilde{\bf{ h}}_k}}^{\mathrm{H}}}{\bf{q}}_d} \right|}^2}}$ denotes the intra-beam interference power when user $k$ decodes its own information.
As a result, the achievable rate of user $i$ associated with RF chain $r$ is generally given by
\begin{equation}\label{C7:DLIndividualRateMin}
{R_{i,r}} = \mathop {\min }\limits_{\forall k < i, {u_{k,r}} = 1}  \left\{ {{R_{k,i,r}},{R_{i,i,r}}} \right\},\forall i,r,
\end{equation}
where $\min\left\{ \cdot \right\}$ is over the users stronger than user $i$ allocated to RF chain $r$, i.e., $\forall k < i$ and ${u_{k,r}} = 1$.
We can observe that if user $i$ is the strong user in NOMA group $r$, i.e., ${u_{k,r}} = 0$, $\forall k < i$, we have ${R_{k,i,r}} = 0$ in \eqref{C7:IndividualRate2} and thus ${R_{i,r}} = {R_{i,i,r}}$ in \eqref{C7:DLIndividualRateMin}.
On the other hand, if user $i$ is the weak user in NOMA group $r$, ${R_{i,r}}$ is obtained by the minimum value among the achievable rates of decoding user $i$ at stronger users  allocated to RF chain $r$, i.e., $\forall k < i$ and ${u_{k,r}} = 1$, and at user $i$ itself.
Therefore, the achievable rate formulation in \eqref{C7:DLIndividualRateMin} guarantees that the weaker user's message is always decodable at the strong user during SIC decoding, even adopting the fixed SIC decoding order.
The achievable system sum-rate for the proposed NOMA scheme can be obtained by
\begin{equation}\label{C7:SumRate}
R_{{\mathrm{sum}}}\left(\mathbf{U},\mathbf{Q}\right) = \sum\limits_{k = 1}^K {\sum\limits_{r = 1}^{{N'_{{\mathrm{RF}}}} } {R_{k,r}} },
\end{equation}
and the system power consumption is given by
\begin{equation}
 {U_{\mathrm{P}}} \left(\mathbf{U},\mathbf{Q}\right) =  {N'_{{\mathrm{RF}}}}{P_{{\mathrm{RF}}}} + \sum\limits_{r = 1}^{{N'_{{\mathrm{RF}}}}}  {{N_{{\mathrm{PS}}}^{r}}}\left(\mathbf{u}_r\right){P_{{\mathrm{PS}}}} +  \frac{1}{\rho} \sum\limits_{k = 1}^K  \left\|{{{\mathbf{q}}_k}}\right\|^2,
\end{equation}
where ${{N_{{\mathrm{PS}}}^{r}}}\left(\mathbf{u}_r\right)$ denotes the number of active PSs attached with RF chain $r$ and $\mathbf{U} =  \left[ {{\mathbf{u} _1},{\mathbf{u} _2}, \ldots ,{\mathbf{u} _{N'_{{\mathrm{RF}}}}}} \right] \in \mathbb{B}^{K\times {N'_{{\mathrm{RF}}}}}$.
We note that ${{N_{{\mathrm{PS}}}^{r}}}\left(\mathbf{u}_r\right)$ depends on the employed analog beamformer for RF chain $r$ and thus is affected by its user scheduling strategy $\mathbf{u}_r$.
In particular, as shown in \eqref{C7:DCBFAnalogBeamformer}, we have ${{N_{{\mathrm{PS}}}^{r}}} = N_{\mathrm{BS}}$ for options 1 and 2, ${{N_{{\mathrm{PS}}}^{r}}} = {N'_{r,\mathrm{BS}}}$ for option 3, and ${{N_{{\mathrm{PS}}}^{r}}} = 2{N_{\mathrm{BS}}}$ for option 4.
Now, the energy-efficient resource allocation design of the proposed scheme in the multi-RF chain system can be formulated as the following optimization problem\footnote{{Due to the existence of intra-beam interference and inter-beam interference among users, the analytical expression of the communication range cannot be obtained in a closed-from for a given  target minimum data rate and total transmit power budget.
Instead, we rely on the proposed optimization framework to design the user grouping policy and the digital precoder under a total transmit power constraint to satisfy the QoS requirements of all the users within an expected communication range.}}:
\begin{align} \label{C7:ResourceAllocationOriginal}
&\underset{\mathbf{U},\mathbf{Q}}{\maxo} \;\;\ {\mathrm{EE}}\left(\mathbf{U},\mathbf{Q}\right) = \frac{R_{{\mathrm{sum}}}\left(\mathbf{U},\mathbf{Q}\right)}{{U_{\mathrm{P}}}\left(\mathbf{U},\mathbf{Q}\right)}  \\[-1mm]
\mbox{s.t.}\;\;
%%%%%
&\mbox{C1: } \sum\limits_{k = 1}^K \left\|{{{\mathbf{q}}_k}}\right\|^2 \le {p_{{\mathrm{BS}}}},\notag\\[-1mm]
%&\mbox{C2: } {u^*_{k,r}}R_{k,i,r} \ge {u^*_{k,r}}R_{i,r}, \forall i > k,\forall r, \notag\\
&\mbox{C2: } \sum\limits_{r = 1}^{N'_{{\mathrm{RF}}}} {R_{k,r} \ge R_{\min }}, \forall k,\;\;\mbox{C3: } u_{k,r} \in \{0,1\}, \forall k,r,\notag
\end{align}
where ${\mathrm{EE}}$ denotes the system energy efficiency, C1 is the transmit power constraint for the BS\footnote{{As shown in \cite{Vaezi2018non}, allocating a higher power to the user with the worse channel is not necessarily required in NOMA.
Therefore, we do not impose any power order constraint in our problem formulation in \eqref{C7:ResourceAllocationOriginal}.}}, and $R_{\min }$ in C2 denotes the minimum data rate requirement of each user.

In the considered system, the user scheduling strategy affects the analog beamformer on each RF chain, c.f. \eqref{C7:DCBFAnalogBeamformer}, and thus determines the effective channel in \eqref{C7:EffectiveChannel}.
In particular, combining \eqref{C7:EffectiveChannel} and \eqref{C7:DCBFAnalogBeamformer}, the effective channel gain in the rate expression in \eqref{C7:ResourceAllocationOriginal} is actually a piecewise function of user scheduling strategy $\mathbf{U}$, which leads to an intractable problem formulation.
As an alternative, we aim to design a computationally efficient resource allocation.
In the following, we first handle the user grouping design and then propose a digital precoder to maximize the system energy efficiency.

\subsection{NOMA User Grouping Algorithm}
To design an efficient NOMA user grouping method, we borrow the concept of coalition formation from game theory\cite{SaadCoalitional2012,WangCoalitionNOMA}.
The basic idea of the proposed user grouping algorithm\footnote{{Note that the formulated energy-efficient resource allocation deign is adaptive to the possible change of the order in effective channel gains with taking into account the analog beamforming and digital precoding.
As a result, the user ordering constraint is not considered in the user grouping process.
}} is to form coalitions among users iteratively to improve the system energy efficiency with equal power allocation in a pure analog beamforming mmWave system, as shown in \eqref{C7:PureAnalog1} and \eqref{C7:PureAnalog2} with $p_k = \frac{p_{{\mathrm{BS}}}}{K}$, $\forall k$.
Before presenting the algorithm, we need to define the following two preferred operations, i.e., \emph{strictly preferred leaving and joining operation} and \emph{strictly preferred switching operation}, based on the coalitional game theory\cite{SaadCoalitional2012}.
In particular, we view all the $K$ users as a set of cooperative \emph{players}, denoted by $\mathcal{K} = \{1,\ldots,K\}$, who seek to form \emph{coalitions} $S$ (NOMA groups in this chapter), $S \subseteq \mathcal{K} $, to maximize the system energy efficiency.
More specifically, we define a coalitional structure $\mathcal{B}$, as a partition of $\mathcal{K}$, i.e., a collection of coalitions $\mathcal{B} = \{S_1,\ldots,S_{{\left|\mathcal{B}\right|}}\}$, such that $\forall r \ne r'$, $S_{r} \bigcap S_{r'} = \emptyset$, and $\cup_{r = 1}^{\left|\mathcal{B}\right|} S_{r} = \mathcal{K}$.
Note that, given an arbitrary coalitional structure $\mathcal{B} = \{S_1,\ldots,S_{{\left|\mathcal{B}\right|}}\}$ with ${\left|\mathcal{B}\right|} = {N'_{{\mathrm{RF}}}}$, the user scheduling variable ${u_{k,r}}$ can be easily obtained by the following mapping:
\begin{equation}\label{C7:MappingB2U}
{u_{k,r}} = \left\{ {\begin{array}{*{20}{c}}
1&{{\mathrm{if}}\;k \in S_r},\\
0&{{\mathrm{otherwise}}}.
\end{array}} \right.
\end{equation}
As a result, the system energy efficiency of the formed coalitional structure ${\mathrm{EE}}\left( \mathcal{B} \right)$ can be obtained combining \eqref{C7:PureAnalog1}, \eqref{C7:MappingB2U}, and the objective function in \eqref{C7:ResourceAllocationOriginal}.

\begin{table}
\begin{algorithm} [H]                    % enter the algorithm environment
\caption{User Grouping Algorithm}
\label{C7:alg1}                             % and a label for \ref{C7:} commands later in the document
\begin{algorithmic} [1]
\small          % enter the algorithmic environment
\STATE \textbf{Initialization}
Initialize the iteration index $\mathrm{iter} = 0$.
The partition is initialized by $\mathcal{B}_0 = \mathcal{K} = \{S_1,\ldots,S_{K}\}$ with $S_k = {k}$, $\forall k$, i.e., OMA.
\REPEAT
\FOR{$k$ = $1$:$K$}
    \STATE User $k \in {S_{r}}$ visits coalitions ${S_{r'}} \in \mathcal{B}_{\mathrm{iter}}$ with $r' \neq r$.
    \IF {$\left|\{{S_{r'}} \cup \left\{ k \right\}\}\right| \le 2$}
        \IF {$\left( {{S_r},\mathcal{B}_{\mathrm{iter}}} \right) { \prec _k} \left( {{S_{r'}}  \cup  \{k\} ,\mathcal{B}_{\mathrm{iter}+1}} \right)$}
            \STATE Execute the leaving and joining operation in Definition \ref{C7:PairOperation}.
            \STATE $\mathrm{iter} = \mathrm{iter} + 1$.
        \ENDIF
    \ELSE
        \IF {$\left( {{S_r},{S_r'},\mathcal{B}_{\mathrm{iter}}} \right) { \prec ^{k'}_k}  \left({S_{r}}   \backslash \{k\}  \cup   \left\{ k' \right\}, {{S_{r'}} \backslash \{k'\}  \cup   \left\{ k \right\},\mathcal{B}_{\mathrm{iter}+1}} \right)$}
            \STATE Execute the switch operation in Definition \ref{C7:SwitchOperation}.
            \STATE $\mathrm{iter} = \mathrm{iter} + 1$.
        \ENDIF
    \ENDIF
\ENDFOR
\UNTIL {No strictly preferred operation can be found.}
\RETURN {$\mathcal{B}_{\mathrm{iter}}$}
\end{algorithmic}
\end{algorithm}
\end{table}

\begin{Def}\label{C7:PairOperation}
In the $\mathrm{iter}$-th iteration, given the current partition $\mathcal{B}_{\mathrm{iter}} = \{S_1,\ldots,S_{\left|\mathcal{B}_{\mathrm{iter}}\right|}\}$ of the set of users $\mathcal{K}$, user $k$ would \textbf{leave} its current coalition ${S_{r}}$ and \textbf{join} another coalition ${S_{r'}} \in \mathcal{B}_{\mathrm{iter}}$, $r' \neq r$, if and only if it is a \emph{strictly preferred leaving and joining operation}, i.e.,
\begin{equation}\label{C7:StrictPreferredLaJ}
\left( {{S_r},\mathcal{B}_{\mathrm{iter}}} \right) { \prec _k} \left( {{S_{r'}}  \cup  \{k\} ,\mathcal{B}_{\mathrm{iter}+1}} \right)\Leftrightarrow {\mathrm{EE}}\left( \mathcal{B}_{\mathrm{iter}} \right) < {\mathrm{EE}}\left( \mathcal{B}_{\mathrm{iter}+1} \right),
\end{equation}
where new formed partition is $\mathcal{B}_{\mathrm{iter}+1} = \left\{ {\mathcal{B}_{\mathrm{iter}}\backslash \left\{ {{S_r},{S_{r'}}} \right\}} \right\} \cup \left\{ {{S_r}\backslash \left\{ k \right\},{S_{r'}} \cup \left\{ k \right\}} \right\}$.
\end{Def}

\begin{Def}\label{C7:SwitchOperation}
In the $\mathrm{iter}$-th iteration, given a partition $\mathcal{B}_{\mathrm{iter}} = \{S_1,\ldots,S_{\left|\mathcal{B}_{\mathrm{iter}}\right|}\}$ of the set of users $\mathcal{K}$, user $k \in {S_{r}}$ and user $k' \in {S_{r'}}$ would switch with each other, if and only if it is a \emph{strictly preferred switch operation}, i.e.,
\begin{align}\label{C7:SwitchRule}
\left( {{S_r},{S_r'},\mathcal{B}_{\mathrm{iter}}} \right) &{ \prec ^{k'}_k} \left({S_{r}} \backslash \{k\} \cup \left\{ k' \right\}, {{S_{r'}} \backslash \{k'\} \cup \left\{ k \right\} ,\mathcal{B}_{\mathrm{iter}+1}} \right) \notag\\[0mm]
& \Leftrightarrow {\mathrm{EE}}\left( \mathcal{B}_{\mathrm{iter}} \right) < {\mathrm{EE}}\left( \mathcal{B}_{\mathrm{iter}+1} \right),
\end{align}
where new formed partition is
\begin{equation}\label{C7:ColiationStructureB}
  \mathcal{B}_{\mathrm{iter}+1} = \left\{ {{\cal B}_{\mathrm{iter}}\backslash \left\{ {{S_r},{S_{r'}}} \right\}} \right\} \cup \left\{ {\{{S_{r}} \backslash \{k\} \cup \left\{ k' \right\}\},\{{S_{r'}} \backslash \{k'\} \cup \left\{ k \right\}\}} \right\}.
\end{equation}
\end{Def}

Note that the equivalence in \eqref{C7:StrictPreferredLaJ} and  \eqref{C7:SwitchRule} imply that a \emph{strictly preferred operation} can strictly improve the system energy efficiency.
Based on the defined strictly preferred operations, the proposed user grouping algorithm is presented in \textbf{Algorithm} \ref{C7:alg1}.
The algorithm is initialized with each user as a coalition, i.e., OMA, and all the PSs of all the RF chains are active.
In each iteration, every user $k \in {S_{r}}$ visits all the potential coalitions except its own coalition in current coalitional structure, i.e., ${S_{r'}} \in \mathcal{B}_{\mathrm{iter}}$ and $r' \neq r$.
Then, each user checks and executes the strictly preferred operations based on the Definitions \ref{C7:PairOperation} and  \ref{C7:SwitchOperation}.
The iteration stops when no more preferred operation can be found, i.e., it converges to a final stable coalitional structure $\mathcal{B}^{*} = \{S_1^{*},\ldots,S_{N'_{{\mathrm{RF}}}}^{*}\}$.
The effectiveness, stability, and convergence of the proposed algorithm are omitted here and the details can be found in \cite{wei2018multibeam,SaadCoalitional2012}.

\subsection{Energy-efficient Digital Precoder Design}
Given the converged stable coalitional structure $\mathcal{B}^{*}$, we can obtain the corresponding user scheduling variable $u_{k,r}$, $\forall k,r$, and the analog beamformer ${{\bf{w}}}_r$, $\forall r$, based on \eqref{C7:MappingB2U} and \eqref{C7:DCBFAnalogBeamformer}, respectively.
Now, the energy-efficient digital precoder design can be formulated as the following optimization problem:
\begin{equation} \label{C7:ResourceAllocation00}
\underset{\mathbf{Q}}{\maxo} \;\;\ \frac{R_{{\mathrm{sum}}}\left(\mathbf{Q}\right)}{{U_{\mathrm{P}}}\left(\mathbf{Q}\right)} \;\;\mbox{s.t.} \;\; \mbox{C1},\mbox{C2}.
\end{equation}
The formulated problem is a non-convex optimization problem, which in general cannot be solved with a systematic and computationally efficient approach optimally.
In particular, the non-convexity arises from the inter-beam and intra-beam interferences in the rate functions.
The non-linear fractional objective function is also an obstacle in designing the digital precoder.
In this chapter, we adopt a novel quadratic transformation \cite{ShenFP} to transform the problem in \eqref{C7:ResourceAllocation00} into an equivalent optimization problem, which facilitates the use of the iterative block coordinate ascent algorithm to achieve a locally optimal solution with a low computational complexity.

Let us define $\mathbf{Q}_k = {{{\mathbf{q}}_k}}{{{\mathbf{q}}^{\mathrm{H}}_k}}$ and ${\widetilde{\bf{ H}}_k} = {{\widetilde{{\mathbf{h}}}_k}}{{\widetilde{{\mathbf{h}}}^{\mathrm{H}}_k}}$.
Also, we introduce the SINR variable of user $i$ associated with RF chain $r$ as
\begin{equation}\label{C7:GammaTransform}
  \gamma_{i,r} = \min \{{\frac{{{ {\Tr \left( {\widetilde{\bf{ H}}_k} {\bf{Q}}_i \right)} }}}{{I_{k,r}^{{\mathrm{inter}}} + I_{k,i,r}^{{\mathrm{intra}}} + \mathcal{W}{N_{0}}}}}, {\frac{{{ {\Tr \left( {\widetilde{\bf{ H}}_i} {\bf{Q}}_i \right)} }}}{{I_{i,r}^{{\mathrm{inter}}} + I_{i,i,r}^{{\mathrm{intra}}} + \mathcal{W}{N_{0}}}}} \}, \forall k<i,r.
\end{equation}
The objective function in \eqref{C7:ResourceAllocation00} can be rewritten as
$V_{\mathrm{I}}\left(\mathcal{Q},\mathbf{\Upsilon}\right) = \frac{R_{{\mathrm{sum}}}\left(\mathbf{\Upsilon}\right)}{{U_{\mathrm{P}}}\left(\mathcal{Q}\right)}$,
where
\vspace{-5mm}
\begin{equation}\label{C7:RsumTransform}
  {R_{{\mathrm{sum}}}\left(\mathbf{\Upsilon}\right)} = {\sum\limits_{k = 1}^K  \sum\limits_{r = 1}^{N'_{{\mathrm{RF}}}} \mathcal{W} {{\log _2}\left( {1 + u_{k,r}\gamma_{k,r}} \right) }}
  \vspace{-5mm}
\end{equation}
and
\vspace{-5mm}
\begin{equation}\label{C7:UPTransform}
    {{U_{\mathrm{P}}}\left(\mathcal{Q}\right)} = {\left( {N'_{{\mathrm{RF}}}}{P_{{\mathrm{RF}}}}  +  \sum\limits_{r = 1}^{{N'_{{\mathrm{RF}}}}} {{N_{{\mathrm{PS}}}^{r}}}\left(\mathbf{u}_r\right){P_{{\mathrm{PS}}}} + \frac{1}{\rho } \sum\limits_{k = 1}^K { { \Tr \left({\bf{Q}}_k\right)} } \right)}\vspace{-5mm}
\end{equation}
with $\mathcal{Q} = \{{\mathbf{Q}_{k}}\}_{k=1}^{K}$ denoting the digital precoder set and $\mathbf{\Upsilon} \in \mathbb{R}^{K \times {N'_{{\mathrm{RF}}}}}$ collecting all the $\gamma_{k,r}$.
Now, the formulated problem in \eqref{C7:ResourceAllocation00} can be equivalently rewritten as
\vspace{-10mm}
\begin{align} \label{C7:ResourceAllocation33}
&\underset{\mathcal{Q},\mathbf{\Upsilon}}{\maxo} \;\;\  V_{\mathrm{I}}\left(\mathcal{Q},\mathbf{\Upsilon}\right) \\
\mbox{s.t.}\;\;
%%%%%
&\mbox{C1: } \sum\limits_{k = 1}^K {{\Tr \left({\bf{Q}}_k\right)} }  \le {p_{{\mathrm{BS}}}},\;\;\mbox{C2: } {{{{u_{k,r}} \gamma_{k,r}}} \ge {u_{k,r}}\left(2^{R_{\min }} - 1\right)}, \forall k,r,\notag \\
&{\mbox{C3}}\mbox{: }  {\frac{{{u_{k,r}}{u_{i,r}}{ {\Tr \left( {\widetilde{\bf{ H}}_k} {\bf{Q}}_i \right)} }}}{{I_{k,r}^{{\mathrm{inter}}}\left({\mathcal{{Q}}}\right) + I_{k,i,r}^{{\mathrm{intra}}}\left({\mathcal{{Q}}}\right) + \mathcal{W}{N_{0}}}}} \ge {u_{k,r}}{u_{i,r}}\gamma_{i,r}, \forall i >  k,\forall r, \notag\\
&\mbox{C4: } {\frac{{u_{k,r}}{{ {\Tr \left( {\widetilde{\bf{ H}}_k} {\bf{Q}}_k \right)} }}}{{I_{k,r}^{{\mathrm{inter}}}\left({\mathcal{{Q}}}\right) + I_{k,k,r}^{{\mathrm{intra}}}\left({\mathcal{{Q}}}\right) + \mathcal{W} {N_{0}}}}} \ge {u_{k,r}}\gamma_{k,r},\forall k,r,\;\mbox{C5: } {\mathrm{Rank}}\left({\mathbf{Q}_{k}}\right) \le 1, \forall k,\notag
\end{align}
where  $I_{k,r}^{{\mathrm{inter}}}\left({\mathcal{{Q}}}\right) = \sum\limits_{r' \ne r}^{{N'_{{\mathrm{RF}}}}} {\sum\limits_{d = 1}^K {{u_{d,r'}}{\Tr \left( {\widetilde{\bf{ H}}_k} {\bf{Q}}_d \right)}} }$ and $I_{k,i,r}^{{\mathrm{intra}}}\left(\mathcal{Q}\right) = \sum\limits_{d = 1}^{i - 1} {{u_{d,r}}{\Tr \left( {\widetilde{\bf{ H}}_k} {\bf{Q}}_d \right)}}$.
Constraints C3 and C4 follow the definition of $\gamma_{i,r}$ and constraint C5 is imposed to guarantee that $\mathbf{Q}_k = {{{\mathbf{q}}_k}}{{{\mathbf{q}}^{\mathrm{H}}_k}}$ holds after optimization.

With adopting the quadratic transform in \cite{ShenFP}, the problem in \eqref{C7:ResourceAllocation33} can be equivalently\footnote{In this chapter, ``equivalent'' means that both problem formulations lead to the same digital precoder and the same system energy efficiency.} transformed as
\begin{align} \label{C7:ResourceAllocation55} &\underset{\mathcal{Q},\mathbf{\Upsilon},\mathbf{{A}},\mathbf{{B}},\delta}{\maxo} \;\; V_{\mathrm{II}} \left( \mathcal{Q}, \mathbf{\Upsilon}, \mathbf{{A}}, \mathbf{{B}}, \delta \right)  =  2\delta \sqrt{ R_{{\mathrm{sum}}} \left( \mathbf{\Upsilon} \right)} - \delta^2{{U_{\mathrm{P}}} \left( \mathcal{Q} \right) } \\
\mbox{s.t.}\;\;
%%%%%
&\mbox{C1},\;\mbox{C2},\;\mbox{C5},\;\overline{{\mbox{C3}}}\mbox{: }  u_{k,r}u_{i,r}g_1\left(a_{k,i,r},{\mathcal{{Q}}}\right) \ge u_{k,r}u_{i,r}{\gamma _{i,r}},\forall i > k, r, \notag\\
&\overline{{\mbox{C4}}}\mbox{: }  u_{k,r}g_2\left(b_{k,r},{\mathcal{{Q}}}\right) \ge u_{k,r}{\gamma _{k,r}},\forall k,r, \notag
\end{align}
where $g_1\left(a_{k,i,r},{\mathcal{{Q}}}\right) = 2{a_{k,i,r}}\sqrt{ {\Tr}\left( {\widetilde{\bf{H}}_k{\bf{Q}}_i} \right)} - a_{k,i,r}^2\left( {{I_{k,r}^{{\mathrm{inter}}}\left({\mathcal{{Q}}}\right) + I_{k,i,r}^{{\mathrm{intra}}}\left({\mathcal{{Q}}}\right) + \mathcal{W}{N_{0}}}} \right)$
and $g_2\left(b_{k,r},{\mathcal{{Q}}}\right) = \left(2{{b}_{k,r}}\sqrt{ {\Tr}\left( {\widetilde{\bf{H}}_k{\bf{Q}}_k} \right)}  - b_{k,r}^2\left( {{I_{k,r}^{{\mathrm{inter}}}\left({\mathcal{{Q}}}\right) + I_{k,k,r}^{{\mathrm{intra}}}\left({\mathcal{{Q}}}\right) + \mathcal{W} {N_{0}}}} \right)\right)$.
Variables ${a_{k,i,r}} \in \mathbb{R}$, ${b_{k,r}} \in \mathbb{R}$, and $\delta \in \mathbb{R}$ are auxiliary variables introduced to handle the fractional functions in constraints C3, C4, and the objective function in \eqref{C7:ResourceAllocation33}, respectively.
To facilitate the presentation, we denote $\mathbf{{A}} \in \mathbb{R} ^ {K \times K \times {N'_{{\mathrm{RF}}}}}$ and $\mathbf{B} \in \mathbb{R} ^ {K \times {N'_{{\mathrm{RF}}}}}$ as the collections of all the ${a_{k,i,r}}$ and ${b_{k,r}}$, respectively, and denote $V_{\mathrm{II}}\left(\mathcal{Q},\mathbf{\Upsilon},\mathbf{{A}},\mathbf{{B}},\delta\right)$ as the objective function in \eqref{C7:ResourceAllocation55}.

\begin{proof}
Please refer to Appendix \ref{C7:AppendixD3} for a proof of the equivalence between the problems in \eqref{C7:ResourceAllocation33} and \eqref{C7:ResourceAllocation55}\footnote{The proof in \cite{ShenFP} cannot be directly applied to the considered problems in \eqref{C7:ResourceAllocation33} and \eqref{C7:ResourceAllocation55}.}.
\end{proof}
\begin{table}
\begin{algorithm} [H]                    % enter the algorithm environment
\caption{Iterative Digital Precoder Design Algorithm}     % give the algorithm a caption
\label{C7:alg2}                             % and a label for \ref{C7:} commands later in the document
\begin{algorithmic} [1]
\small          % enter the algorithmic environment
\STATE \textbf{Initialization}\\
Initialize the convergence tolerance $\epsilon$, the maximum number of iterations $\mathrm{iter}_\mathrm{max}$, the iteration index ${\mathrm{iter}} = 1$, and the initial feasible solution $\left(\mathcal{Q}^{\mathrm{iter}},\mathbf{\Upsilon}^{\mathrm{iter}}\right)$.
\REPEAT
\STATE
Update the auxiliary variables $\left(\mathbf{{A}}^{\mathrm{iter}},\mathbf{{B}}^{\mathrm{iter}},\delta^{\mathrm{iter}}\right)$ by \eqref{C7:UpdateSlack1}, \eqref{C7:UpdateSlack2}, and \eqref{C7:UpdateSlack3} in Appendix \ref{C7:AppendixD3} with a fixed $\left(\mathcal{Q}^{\mathrm{iter}},\mathbf{\Upsilon}^{\mathrm{iter}}\right)$, respectively.
\STATE
Update the primal variables $\left(\mathcal{Q}^{\mathrm{iter} + 1},\mathbf{\Upsilon}^{\mathrm{iter} + 1}\right)$ by solving the problem in \eqref{C7:ResourceAllocation55} with a fixed $\left(\mathbf{{A}}^{\mathrm{iter}},\mathbf{{B}}^{\mathrm{iter}},\delta^{\mathrm{iter}}\right)$.
\STATE $\mathrm{iter} = \mathrm{iter} + 1$
\UNTIL $\mathrm{iter} = \mathrm{iter}_\mathrm{max}$ or $\frac{\left| {\delta^{\mathrm{iter}} - \delta^{\mathrm{iter-1}}} \right|}{\delta^{\mathrm{iter}}} \le \epsilon$
\end{algorithmic}
\end{algorithm}
\end{table}

After applying the proposed transformation, a block coordinate ascent algorithm can be utilized to achieve a locally optimal solution \cite{ShenFP} via optimizing primal variables $\left(\mathcal{Q},\mathbf{\Upsilon}\right)$ and auxiliary variables $\left(\mathbf{{A}},\mathbf{{B}},\delta\right)$ iteratively, as shown in \textbf{Algorithm} \ref{C7:alg2}.
Specifically, $\left(\mathcal{Q},\mathbf{\Upsilon}\right)$ in line 4 is updated by solving the non-convex problem in \eqref{C7:ResourceAllocation55} with a fixed $\left(\mathbf{{A}},\mathbf{{B}},\delta\right)$, where the non-convexity is due to the rank-one constraint C5.
We adopt the semi-definite programming (SDP) relaxation by removing constraint C5 in \eqref{C7:ResourceAllocation55}, which results in a convex SDP given by
\begin{equation}
\underset{\mathcal{Q},\mathbf{\Upsilon}}{\maxo} \;\;\  V_{\mathrm{II}}\left(\mathcal{Q},\mathbf{\Upsilon},\mathbf{{A}},\mathbf{{B}},\delta\right)
\;\;\mbox{s.t.}\;\;
%%%%%
\mbox{C1},\;\mbox{C2},\;\overline{{\mbox{C3}}},\;\overline{{\mbox{C4}}}.\label{C7:ResourceAllocation77}
\end{equation}
Now, the problem in \eqref{C7:ResourceAllocation77} is convex which can be solved efficiently with standard numerical convex program solvers, such as CVX\cite{cvx}.
Note that, for ${p_{{\mathrm{BS}}}}>0$, the optimal solution ${\bf{Q}}_k^*$, $\forall k$, of \eqref{C7:ResourceAllocation77} is guaranteed to be a rank-one matrix, which can be proved via a similar approach as in \cite{SunMISONOMA}.
In \textbf{Algorithm} \ref{C7:alg2}, the primal variables $\left(\mathcal{Q},\mathbf{\Upsilon}\right)$ are initialized by solving the convex problem in \eqref{C7:ResourceAllocation33} with all the users having the same minimum data rate, i.e., ${{{\gamma_{k,r}}} = \left(2^{R_{\min }} - 1\right)}$, $\forall k,r$.
The algorithm terminates when the maximum number of iterations is reached or the change of $\delta^{\mathrm{iter}}$ is smaller than a predefined convergence tolerance.

\section{Simulation Results}\label{C7:SimulationResults}
In this section, we evaluate the performance of our proposed beamwidth control-based mmWave NOMA scheme and the proposed resource allocation algorithm.
According to \cite{GaoSubarray}, we set ${P_{{\mathrm{RF}}}} = 250$ mW and ${P_{{\mathrm{PS}}}} = 1$ mW.
Besides, we assume that the power amplifier efficiency coefficient is $\rho = 0.8$.
Unless specified otherwise, the simulation setting is given as follows.
We consider a hybrid mmWave communication system with a carrier frequency at $28$ GHz and a system bandwidth $\mathcal{W} = 100$ MHz.
There are one LOS path and $L = 10$ NLOS paths for the channel model in \eqref{C7:ChannelModel1} and the path loss models for LOS and NLOS paths follow Table I in \cite{AkdenizChannelMmWave}.
All the $K$ users are randomly and uniformly distributed in the $\frac{1}{3}$ cell, as shown in Figure \ref{C7:PortionCell}, with a cell radius of $D = 200$ meters.
The number of RF chains at the BS is the same as the number of users, i.e., $N_{\mathrm{RF}} = K$.
For the proposed iterative digital precoder design, the convergence tolerance is $\epsilon = 0.01$ and the maximum number of iterations is $\mathrm{iter}_\mathrm{max} = 100$.
The maximum transmit power of the BS is 40 dBm, i.e., ${p_{{\mathrm{BS}}}} \le 40$ dBm and the noise power spectral density at all the users is assumed identical with $N_{0} = -168$ dBm/Hz.
The number of antennas equipped at the BS\footnote{As mentioned, we consider a symmetric ULA containing an odd number of antennas at the BS.} is $N_{\mathrm{BS}} = 129$ and the number of antennas equipped at each user terminal is $N_{\mathrm{UE}} = 16$.
The minimum data rate is set as ${R_{\min }} = 1$ Mbit/s.

To show the effectiveness of our proposed mmWave NOMA with beamwidth control, we consider two baseline schemes in our simulations.
For baseline 1, the conventional mmWave OMA scheme is considered where each user is served by a single RF chain.
As a result, we can simply set ${N'_{{\mathrm{RF}}}} = {N_{{\mathrm{RF}}}} = K$ and $\mathbf{U} = \mathbf{I}_{N_{{\mathrm{RF}}}}$ for the proposed iterative digital precoder design in \textbf{Algorithm} \ref{C7:alg2} to obtain the algorithm for the baseline OMA scheme.
For baseline 2, the conventional mmWave NOMA scheme is considered where only the users within a half-power beamwidth can be considered as a NOMA group.
Thus, the proposed user grouping algorithm in \textbf{Algorithm} \ref{C7:alg1} is also applicable to the baseline NOMA scheme and the corresponding analog beamformer is obtained by the first two lines in \eqref{C7:DCBFAnalogBeamformer}.
After obtaining the user grouping and analog beamforming strategy, \textbf{Algorithm} \ref{C7:alg2} can be used to design the digital precoder for the baseline NOMA scheme.
The simulation results shown in the sequel are obtained by averaging the system energy efficiency over multiple channel realizations.

\subsection{Convergence of the Proposed Algorithms}

\begin{figure}[t]
\centering
\subfigure[User grouping algorithm.]
{\label{C7:EE_Versus_IterationIndex} %% label for first subfigure
\includegraphics[width=2.5in]{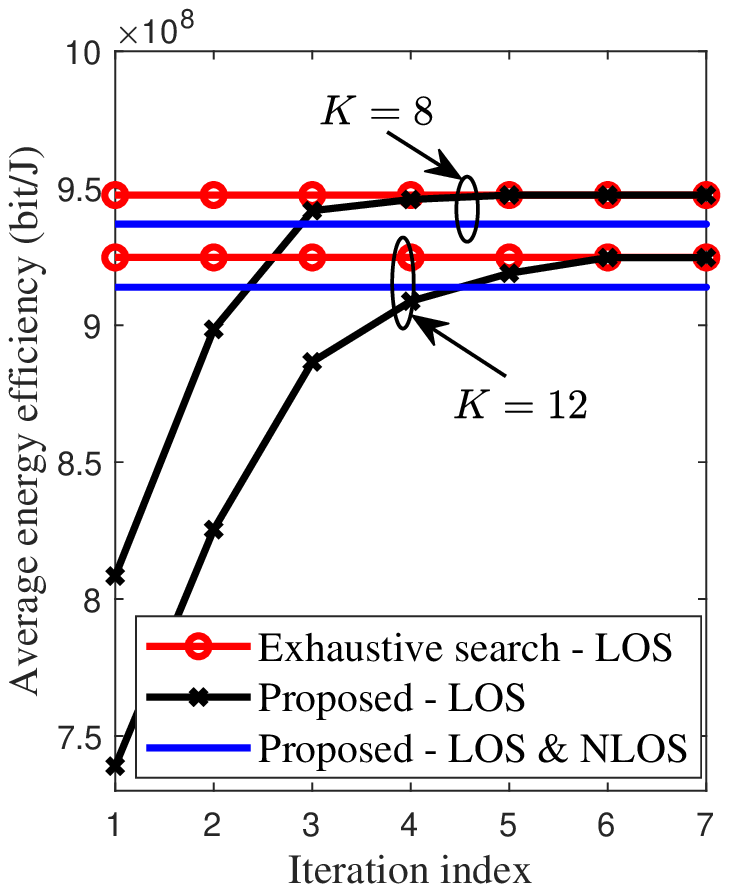}}
\subfigure[Digital precoder design algorithm.]
{\label{C7:EE_Versus_IterationIndex_Digital} %% label for second
\includegraphics[width=2.5in]{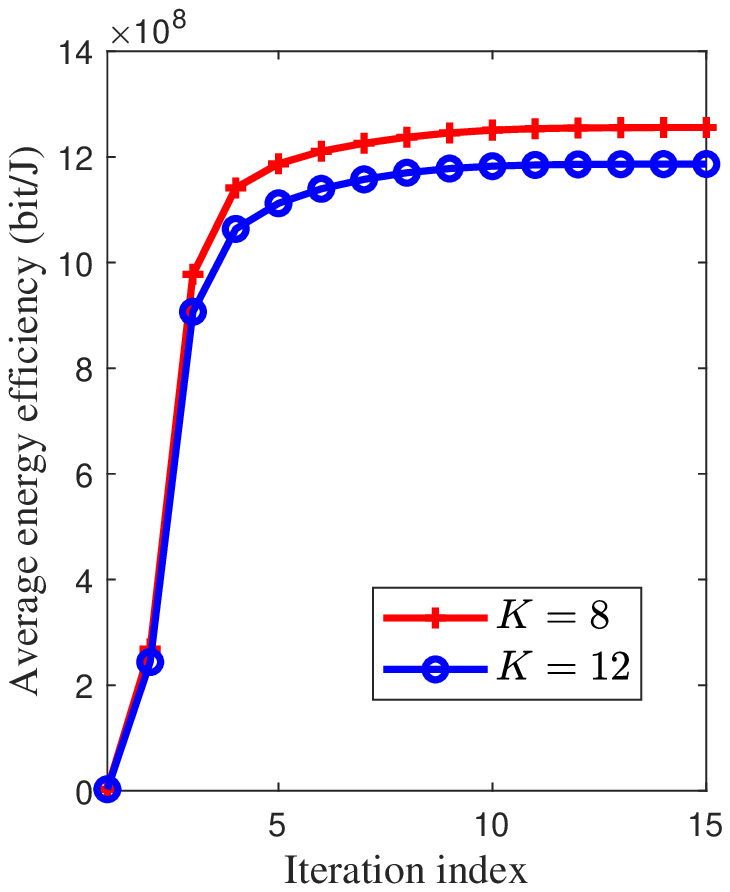}}
\caption{The convergence of our proposed resource allocation algorithms.}
\label{C7:Convergence}%
\end{figure}

Figure \ref{C7:Convergence} illustrates the convergence of our proposed resource allocation algorithms.
Two simulation cases with $K = 8$ and $K = 12$ are shown with the same transmit power ${p_{{\mathrm{BS}}}} = 30$ dBm at the BS.
In legend, ``LOS" denotes the pure LOS mmWave channels, while ``LOS \& NLOS" denotes the case with both LOS path and NLOS paths.
The exhaustive search for user grouping is also shown as a benchmark for \textbf{Algorithm} \ref{C7:alg1} in Figure \ref{C7:EE_Versus_IterationIndex}.
It can be observed that our proposed user grouping algorithm can perform closely to the optimal exhaustive search within only 6 iterations on average, which demonstrates the fast convergence and the effectiveness of our proposed user grouping algorithm.
Although the proposed user grouping and analog beamformer are designed exploiting only the LOS CSI, the system energy efficiency is still very close to that with pure LOS only scenario, even if there exists $L$ NLOS paths in mmWave channels.
It means that the performance degradation of the proposed scheme due to the lack of NLOS CSI at the BS is negligible.
This is because the LOS path is much stronger than the NLOS paths in mmWave channels\cite{Rappaport2013}.
In Figure \ref{C7:EE_Versus_IterationIndex_Digital}, we can observe that the system energy efficiency of the proposed iterative digital precoder design algorithm improves monotonically with the number of iterations.
Comparing Figure \ref{C7:EE_Versus_IterationIndex} and Figure \ref{C7:EE_Versus_IterationIndex_Digital}, we can observe that the digital precoder design can further improve the system energy efficiency of the proposed mmWave NOMA scheme, compared to the pure analog beamforming system.
This is attributed to the interference management capability of our proposed digital precoder design.

\subsection{Energy Efficiency versus the Number of Users}

\begin{figure}[t]
\centering
\includegraphics[width=4.5in]{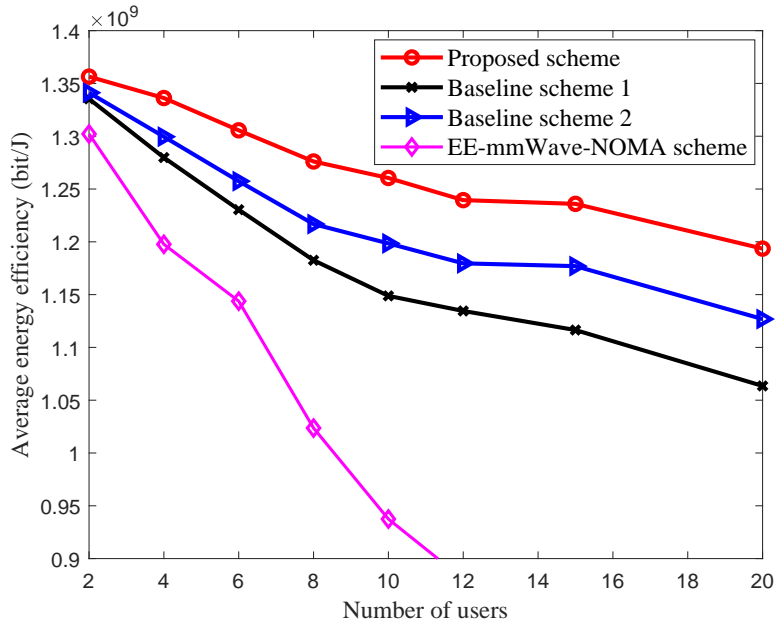}
\caption{Average energy efficiency (bit/J) versus the number of users $K$.}
\label{C7:EE_Versus_K}
\end{figure}

Figure \ref{C7:EE_Versus_K} shows the average system energy efficiency versus the number of users $K$.
We can observe that the energy efficiencies of all the schemes decrease with the increasing number of users $K$.
In fact, introducing more users makes the minimum data rate constraints more stringent, which leads to a smaller feasible solution set for solving \eqref{C7:ResourceAllocationOriginal}.
In other words, the system is less flexible in resource allocation to improve the system energy efficiency with large numbers of users $K$.
In addition, it can be observed that proposed scheme outperforms the two baseline schemes.
In fact, our proposed scheme can widen the analog beamwidth which increases the probability of forming a NOMA group and thus admits more NOMA groups.
As a result, we can put more RF chains idle which translates into a higher system energy efficiency compared to the baseline schemes.
More importantly, we can observe that the performance gain of our proposed scheme over the baseline NOMA and OMA schemes is enlarged with increasing the number of users.
This is because our proposed scheme with beamwidth control can effectively exploit the users' AOD distribution to improve the system energy efficiency.

The simulation results of the proposed ``EE-mmWave-NOMA'' scheme in \cite{HaoEEmmWaveNOMA} is reproduced in Figure \ref{C7:EE_Versus_K} for comparison.
Compared to the EE-mmWave-NOMA scheme in \cite{HaoEEmmWaveNOMA}, our proposed scheme can achieve a substantially higher system energy efficiency.
This is because the proposed beamwidth control enables the system to generate more NOMA groups which substantially saves the circuit power consumption of RF chains.
Besides, our proposed resource allocation design outperforms the EE-mmWave-NOMA scheme in \cite{HaoEEmmWaveNOMA} in the following aspects.
Firstly, instead of adopting a suboptimal correlation-based NOMA user grouping strategy as in \cite{HaoEEmmWaveNOMA}, our proposed user grouping algorithm utilizes the coalition formation game theory to maximize the system energy efficiency taking into account the inter-cluster interference.
In fact, as shown in Figure \ref{C7:EE_Versus_IterationIndex}, the proposed user grouping algorithm performs closely to the optimal exhaustive search.
Secondly, a zero-forcing (ZF) digital precoder was adopted in \cite{HaoEEmmWaveNOMA}, which may lead to an unsatisfactory system energy efficiency, especially when users' channels are correlated and a significant power loss occurs due to the channel matrix inversion.
In contrast, our proposed digital precoder design is based on the optimization framework and the ZF digital precoder in \cite{HaoEEmmWaveNOMA} is only a subcase of the optimization solution set.
Thirdly, the proposed EE-mmWave-NOMA scheme in \cite{HaoEEmmWaveNOMA} optimized the power allocation for each NOMA cluster separately.
However, our proposed scheme jointly optimizes the digital precoder of all the users.

\subsection{Energy Efficiency versus the Total Transmit Power}

\begin{figure}[t]
\centering
\includegraphics[width=4.5in]{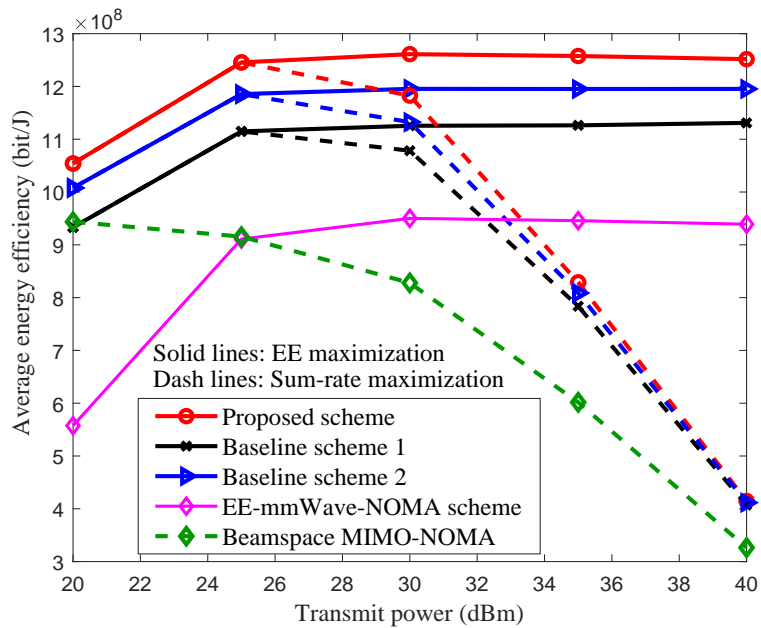}
\caption{Average energy efficiency (bit/J) versus the total transmit power.}
\label{C7:EE_Versus_Pmax}
\end{figure}

Figure \ref{C7:EE_Versus_Pmax} depicts the average system energy efficiency versus the total maximum transmit power budget at the BS with $K = 10$ users.
We can observe that the proposed scheme with beamwidth control outperforms the two baseline schemes.
Note that the considered two baseline schemes do not suffer from the main lobe power loss as the proposed beamwidth control-based scheme.
However, the proposed beamwidth control can alleviate the system performance bottleneck created by exceedingly narrow analog beams in mmWave communication systems via forming more NOMA groups, which reduces the total circuit power consumption of RF chains and thus achieves a higher system energy efficiency.
In Figure \ref{C7:EE_Versus_Pmax}, the results of the proposed EE-mmWave-NOMA scheme \cite{HaoEEmmWaveNOMA} is also shown for comparison with our proposed scheme.
It can be observed that our proposed scheme significantly outperforms the EE-mmWave-NOMA scheme in \cite{HaoEEmmWaveNOMA} since the proposed beamwidth control can facilitate the efficient exploitation of NOMA gain.
Due to our superior resource allocation design, the two baseline schemes also outperform the EE-mmWave-NOMA scheme in \cite{HaoEEmmWaveNOMA}.

In addition, we show the energy efficiency of all the schemes with the objective for maximizing the system sum-rate.
It can be observed that, in the low SNR regime, the sum-rate maximization achieves the same energy efficiency as the proposed ${\mathrm{EE}}$ maximization scheme.
Indeed, transmitting with the maximum available power is the most energy-efficient option in the low SNR regime.
However, with increasing the system transmit power budget, the energy efficiency of sum-rate maximization schemes decreases dramatically while that of the proposed ${\mathrm{EE}}$ maximization schemes saturates.
In fact, in the high SNR regime, there is a diminishing return in spectral efficiency when allocating more transmit power.
Hence, the energy consumption in the system outweighs the spectral efficiency gain in the large transmit power regime.
We note that the energy efficiency gain of the proposed scheme over the baseline NOMA scheme decreases slightly in the low SNR regime.
This is because, in the power-limited region, the system performance is more sensitive to the power loss introduced by the proposed beamwidth control scheme.
For comparison, the simulation results of the proposed ``beamspace MIMO-NOMA'' scheme in \cite{WangBeamSpace2017} is reproduced and shown in Figure \ref{C7:EE_Versus_Pmax}.
We can observe a substantially higher energy efficiency for the proposed scheme compared to the proposed beamspace MIMO-NOMA scheme in \cite{WangBeamSpace2017}.
Moreover, our two considered baseline schemes outperform the beamspace MIMO-NOMA scheme owing to the superior performance of our proposed digital precoder design than that of the adopted ZF digital precoder \cite{WangBeamSpace2017}.

\subsection{Energy Efficiency versus the Number of Antennas}

\begin{figure}[t]
\centering
\includegraphics[width=4.5in]{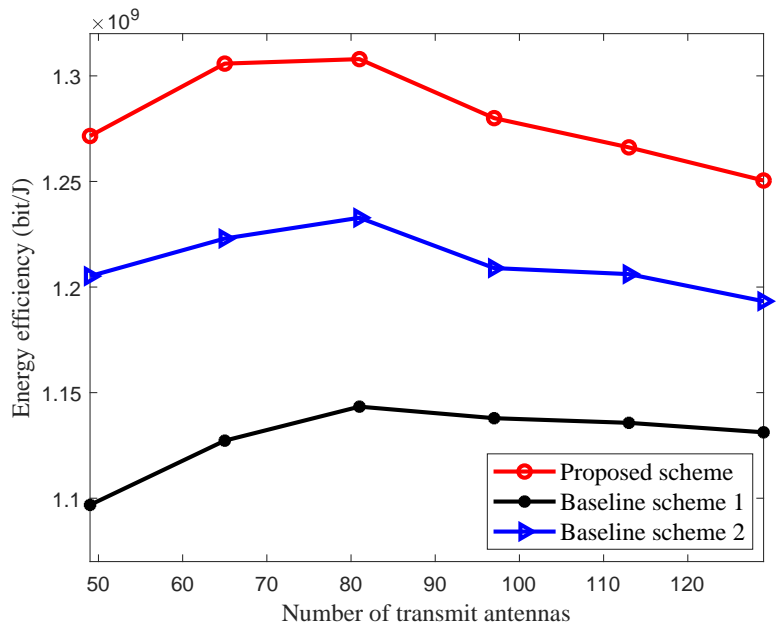}
\caption{Average energy efficiency (bit/J) versus the number of transmit antennas.}
\label{C7:EE_Versus_NBS}
\end{figure}

Figure \ref{C7:EE_Versus_NBS} demonstrates the average system energy efficiency versus the total number of antennas equipped at the BS with the number of users $K = 10$.
We can observe that the system energy efficiencies of all the schemes first increase and then decrease slightly with the increasing number of antennas $N_{\mathrm{BS}}$.
In fact, employing more antennas at the BS can increase the beamforming gain at the expense of more power consumption for PSs.
As a result, for small $N_{\mathrm{BS}}$, increasing the number of antennas can improve the system sum-rate substantially so as to increase the system energy efficiency.
However, for large $N_{\mathrm{BS}}$, the sum-rate increment of introducing extra antennas becomes diminished but the energy consumption driving the massive antenna array is still significant.
In addition, it can be observed that the energy efficiency gain of the baseline NOMA scheme without beamwidth control over the conventional OMA scheme decreases with increasing $N_{\mathrm{BS}}$.
Indeed, without beamwidth control, the large number of antennas in baseline 2 can only create a very narrow analog beam which restricts the numbers of formed NOMA groups to improve the system energy efficiency.
In contrast to the baseline 2, our proposed scheme can provide a higher energy efficiency gain over the conventional OMA scheme via utilizing beamwidth control.
We note that, with increasing $N_{\mathrm{BS}}$, the energy efficiency gain of the proposed beamwidth control scheme over the two baseline schemes also decreases.
This is due to the fact that the analog beamwidth becomes narrower with increasing the array size.
As a result, widening the beamwidth of a large-scale array to cover a NOMA group introduces a larger main lobe power loss, which reduces the energy efficiency gain brought by NOMA.

\subsection{Energy Efficiency versus the Channel Estimation Error}

\begin{figure}[t]
\centering  \includegraphics[width=4.5in]{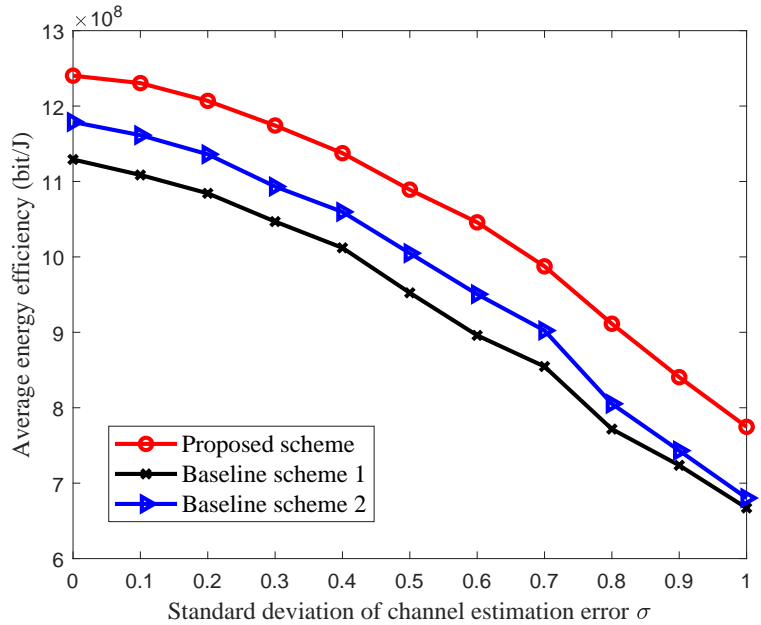}
\caption{Average energy efficiency (bit/J) versus the channel estimation error in terms of the standard deviation $\sigma$.}
\label{C7:EE_Versus_Error}
\end{figure}

In this section, we investigate the impact of imperfect CSI on the system performance achieved by the proposed scheme and algorithms through simulations.
For simplicity, we adopt a single parameter, the standard deviation $\sigma$, to measure the amount of channel estimation errors during the adopted three-stage channel estimation procedure.
In particular, we assume that $\sigma_{{\theta} _{k,0}} = \frac{\sigma}{N_{\mathrm{BS}}}$, $\sigma_{{\phi} _{k,0}} = \frac{\sigma}{N_{\mathrm{UE}}}$, ${\sigma_{{\alpha} _{k,0}}} = \sigma {\abs{{\alpha} _{k,0}}}$, and ${\sigma_{{\widetilde{{{h}}}_{k,r}}}} = \sigma{|{{{\widetilde{{{h}}}_{k,r}}}}|}$, $\forall k,r$.
The larger $\sigma$, the larger channel estimation error, and vice versa.
When $\sigma = 0$, it degenerates to the case without channel estimation error.
Please note that according to the array signal processing theory \cite{van2002optimum}, the angular resolution is inversely proportional to the array size.
Therefore, we assume that the standard deviations of AOD and AOA estimation error are given by $\sigma_{{\theta} _{k,0}} = \frac{\sigma}{N_{\mathrm{BS}}}$ and $\sigma_{{\phi} _{k,0}} = \frac{\sigma}{N_{\mathrm{UE}}}$, respectively.
	
Figure \ref{C7:EE_Versus_Error} demonstrates the average energy efficiency versus the channel estimation error of our proposed scheme.
We can observe that the system energy efficiency decreases when the channel estimation error becomes larger.
Besides, it can be observed that the proposed scheme outperforms the two considered baseline schemes in the considered channel estimation error regime.
With the increasing channel estimation error, the energy efficiency gain of our proposed scheme over the baseline NOMA scheme keeps almost a constant.
In contrast, the energy efficiency gain of the baseline NOMA scheme over the baseline OMA scheme diminishes with increasing $\sigma$.
In fact, the narrower the beamwidth, the more sensitive the performance to the beam alignment error.
Therefore, the proposed scheme with widened analog beamwidth is more robust against the beam alignment error compared to the baseline NOMA scheme without beamwidth control.

\section{Summary}

In this chapter, we proposed a novel beamwidth control-based mmWave NOMA scheme and studied its energy-efficient resource allocation design.
In particular, the proposed scheme overcomes the fundamental limit of narrow beams in hybrid mmWave systems through beamwidth control, which facilitates the exploitation of NOMA transmission to enhance the system energy efficiency.
We proposed two types of beamwidth control methods and characterized their main lobe power losses.
To facilitate the design of analog beamformers for multi-RF chain systems, we first studied the asymptotically optimal analog beamforming design for a single-RF chain system.
Subsequently, the obtained analog beamforming design was adopted for formulating the problem of energy-efficient resource allocation design in multi-RF chain systems.
A user grouping algorithm based on coalition formation game theory was developed to achieve a stable user grouping strategy, which can perform closely to the optimal exhaustive search.
In addition, adopting the quadratic transformation, a low-complexity iterative digital precoder design algorithm was proposed to converge to a locally optimal solution.
Simulation results demonstrated that the proposed scheme with beamwidth control offers a substantial energy efficiency gain over the conventional OMA and NOMA schemes without beamwidth control.

\chapter{Thesis Conclusions and Future Works}\label{C9:chapter9}

{In this chapter, we first conclude this thesis and then outline some future research directions arising from our works. }

\section{Conclusions}
In this thesis, we have studied and addressed the performance analysis and designs for NOMA in wireless communication systems.
Specifically, we have analyzed the performance gain of NOMA over OMA and have  revealed its distinctive behaviors in different scenarios.
In addition, we have proposed practical schemes and resource allocation designs for NOMA in microwave and millimeter wave communication systems.
We conclude this thesis in the following by summarizing our main contributions.

In Chapter 1, we have presented the motivation, the literature review, the outline, and the main contributions of this thesis.
In Chapter 2, we have reviewed some fundamental and related background knowledge of the thesis that are useful in the later chapters.

Then, we have presented our unified performance analysis on ESG of NOMA over OMA in single-antenna, multi-antenna and massive MIMO systems considering both single-cell and multi-cell deployments.
The corresponding ESGs were quantified via employing the asymptotic analysis and their different behaviors under different scenarios were revealed.
Our numerical results have verified the accuracy of the analytical results derived and have confirmed the insights revealed about the ESG of NOMA over OMA in different scenarios.

To enhance the robustness against the channel uncertainty, we have proposed an JPA scheme for uplink MIMO-NOMA systems with a MRC-SIC receiver.
We first analyzed ASINR of each user during the MRC-SIC decoding by taking into account the error propagation due to the channel estimation error.
The JPA design was formulated as a non-convex optimization problem to maximize the minimum weighted ASINR.
Then, the formulated problem was transformed into an equivalent geometric programming problem and was solved optimally.
Our simulation results have demonstrated that our proposed scheme can effectively alleviate the error propagation of MRC-SIC and enhance the detection performance, especially for users with moderate energy budgets.

Chapter 5 has been devoted to the power-efficient resource allocation design for downlink MC-NOMA systems with taking into account the imperfection of CSI at transmitter and QoS requirements of users.
The resource allocation design was formulated as a non-convex optimization problem which jointly designs the power allocation, rate allocation, user scheduling, and SIC decoding policy for minimizing the total transmit power.
A globally optimal solution was obtained via employing the branch-and-bound approach and a suboptimal iterative resource allocation algorithm was developed based on difference of convex programming.
Our simulation results have demonstrated that the suboptimal scheme achieves a close-to-optimal performance rapidly and the significant power savings of our proposed scheme.

Different from the previous chapters, in Chapter 6, we have further investigated applying NOMA in hybrid mmWave communications.
We have proposed a multi-beam NOMA scheme for hybrid mmWave systems and have studied its resource allocation.
In contrast to the recently proposed single-beam mmWave-NOMA scheme which can only serve multiple NOMA users within the same analog beam, the proposed scheme can perform NOMA transmission for the users with an arbitrary angle-of-departure distribution.
A suboptimal two-stage resource allocation design has been proposed for maximizing the system sum-rate.
Our simulation results have demonstrated that our designed resource allocation can achieve a close-to-optimal performance in each stage and the proposed scheme offers a substantial spectral efficiency improvement compared to that of the single-beam mmWave-NOMA and the mmWave OMA schemes.

A further work on beamwidth control-enabled mmWave NOMA scheme has been presented in Chapter 7 and we have studied its energy-efficient resource allocation design.
In particular, the proposed beamwidth control can increase the number of served NOMA groups by widening the analog beamwidth.
We have formulated the energy-efficient resource allocation design as a non-convex optimization problem which takes into account the minimum required user data rate.
A NOMA user grouping algorithm based on the coalition formation game theory was developed and a low-complexity iterative digital precoder design was proposed to achieve a locally optimal solution.
Our numerical results have verified the fast convergence and effectiveness of our proposed algorithms and have demonstrated the superior energy efficiency achieved by our proposed scheme, compared to the conventional orthogonal multiple access and NOMA schemes without beamwidth control.

\section{Future Works}
{The explosive growth of traffic demand keeps imposing unprecedentedly challenges for the development of future wireless communication systems, including supporting massive connectivity, energy efficiency as well as spectral efficiency improvement, ultra-reliable low-latency communications (URLLC), and so on.
This thesis has addressed some of these challenges via applying the concept of NOMA and improving the system performance with resource allocation design.
However, there are still many research issues to be addressed.
In the following, we propose some future research directions arising from the work presented in this thesis.}

\subsection{ESG Analysis with Imperfect CSI and SIC Decoding Error}
{In Chapter \ref{C3:chapter3}, as a first attempt to unveil fundamental insights on the performance gain of NOMA over OMA, we considered the ideal case associated with perfect CSI and error propagation-free SIC detection at the BS.
In practice, it is difficult to acquire the perfect CSI due to channel estimation errors, feedback delays, and/or quantization errors.
Similarly, the error propagation during SIC decoding is usually inevitable in practice.
Therefore, it is important to further investigate the ESG of NOMA over OMA both in the face of imperfect CSI and error propagation during SIC detection.}
\subsection{The Robust Resource Allocation Design for NOMA in High Mobility Scenario}
{In this thesis, the resource allocation designs have assumed slow fading channels, whether or not there is perfect CSIT, statistical CSIT, or partial CSIT.
However, in practice, in high mobility scenarios, NOMA may suffer from the severe Doppler spread, which results in a fast fading channel.
Therefore, it is worth to investigate the robust resource allocation design for NOMA in fast fading channels.}
\subsection{URLLC Design for NOMA Systems}
{The presented works in this thesis mainly focused on improving the NOMA systems' power efficiency, spectral efficiency, and energy efficiency, with NOMA's inherent nature of accommodating more users with limited system DoFs.
However, the future wireless networks are expected to provide services for latency sensitive devices for applications in factory automation, autonomous driving, and remote surgery.
Therefore, cross-layer designs and network architecture designs should be taken into account for NOMA systems to support URLLC in future works.}

\appendix
\onehalfspacing
\fancyhead[CE]{\leftmark}
\fancyhead[CO]{\leftmark}
\chapter{Proof of Theories of Chapter 3}\label{C3:appendix_2a}

\section{Proof of Theorem \ref{C3:Theorem1}}\label{C3:AppendixA}
To facilitate the proof, we first consider a virtual system whose capacity serves as an upper bound to that of the system in \eqref{C3:MIMONOMASystemModel}.
In particular, the virtual system is the uplink of a $K$-user $M \times M$ MIMO system with $M$ antennas employed at each user and the BS.
We assume that, in the virtual $K$-user $M \times M$ MIMO system, each user faces $M$ parallel subchannels with identical subchannel gain ${\left\| {{{\bf{h}}_k}} \right\|}$, i.e., the channel matrix between user $k$ and the BS is ${\left\| {{{\bf{h}}_k}} \right\|}{{\bf{I}}_M}$.
As a result, the signal received at the BS is given by
\begin{equation}\label{C3:MIMONOMASystemModelRelaxed}
{\bf{\tilde y}} = \sum\limits_{k = 1}^K \sqrt {{p_k}} {\left\| {{{\bf{h}}_k}} \right\|}{{\bf{I}}_M}{{{\bf{\tilde x}}}_k} + {\bf{v}},
\end{equation}
where ${{{\bf{\tilde x}}}_k} = {{\bf{u}}_k}{x_k} \in \mathbb{C}^{M \times 1}$ denotes the transmitted signal after preprocessing by a precoder ${{\bf{u}}_k}\in \mathbb{C}^{M \times 1}$.
We note that the precoder should satisfy the constraint $ \mathrm{Tr} \left({{{\bf{u}}_k}} {{{\bf{u}}_k^{\mathrm{H}}}} \right)\le 1$, so that ${\mathrm{E}}\left\{ {{\bf{\tilde x}}_k^{\mathrm{H}}{{{\bf{\tilde x}}}_k}} \right\} \le {\mathrm{E}}\left\{ {x_k^{\mathrm{2}}} \right\} = 1$.
Additionally, in the virtual $K$-user $M \times M$ MIMO system, the subchannel gain between user $k$ and the BS is forced to be identical as ${\left\| {{{\bf{h}}_k}} \right\|}$, where ${\left\| {{{\bf{h}}_k}} \right\|}$ is the corresponding channel gain value between user $k$ and the BS in the original $K$-user $1 \times M$ MIMO system in \eqref{C3:MIMONOMASystemModel}.
Furthermore, we consider an arbitrary but the identical power allocation strategy ${\bf{p}} = \left[ {{p_1}, \ldots ,{p_K}} \right]$ as that of our original system in \eqref{C3:MIMONOMASystemModel} during the following proof.
Upon comparing \eqref{C3:MIMONOMASystemModel} and \eqref{C3:MIMONOMASystemModelRelaxed}, we can observe that the specific choice of the precoder ${{{\bf{u}}_k}} = \frac{{{{\bf{h}}_k}}}{\left\| {{{\bf{h}}_k}} \right\|}$ in \eqref{C3:MIMONOMASystemModelRelaxed} would result in an equivalent system to that in \eqref{C3:MIMONOMASystemModel}.
In other words, the capacity of the system in \eqref{C3:MIMONOMASystemModelRelaxed} serves as an upper bound to that of the system in \eqref{C3:MIMONOMASystemModel}, i.e., we have:
\begin{align}\label{C3:UpperBoundMIMONOMA}
R_{\mathrm{sum}}^{{\mathrm{MIMO-NOMA}}} \mathop  =\limits^{(a)} C\left( {M, K, {\bf{p}},{\bf{H}}} \right) &\le {C}\left( {M^2, K,{\bf{p}},\widetilde{{\bf{H}}}} \right)\notag\\
&= \mathop {\max }\limits_{{\mathrm{Tr}}\left( {{{\bf{u}}_k}{\bf{u}}_k^{\mathrm{H}}} \right) \le 1} {\ln}\left| {{{\bf{I}}_M} + \frac{1}{{{N_0}}}\sum\limits_{k = 1}^K {{p_k}{{\left\| {{{\bf{h}}_k}} \right\|}^2}{{\bf{I}}_M}{{\bf{u}}_k}{\bf{u}}_k^{\mathrm{H}}{\bf{I}}_M^{\mathrm{H}}} } \right| \notag\\
&= M{\ln}\left( {1 + \frac{1}{{M{N_0}}}\sum\limits_{k = 1}^K {{p_k}{{\left\| {{{\bf{h}}_k}} \right\|}^2}} } \right),
\end{align}
where $C\left( {M, K, {\bf{p}},{\bf{H}}} \right)$ denotes the capacity for the uplink $K$-user $1 \times M$ MIMO system in \eqref{C3:MIMONOMASystemModel} for a channel matrix ${\bf{H}} = \left[ {{{\bf{h}}_1}, \ldots ,{{\bf{h}}_K}} \right]$ and power allocation ${\bf{p}}$.
Furthermore, $C\left( {M^2, K, {\bf{p}},\widetilde{{\bf{H}}}} \right)$ denotes the capacity of the virtual $K$-user $M \times M$ MIMO system in \eqref{C3:MIMONOMASystemModelRelaxed} associated with a channel matrix $\widetilde{{\bf{H}}} = \left[ {{\left\| {{{\bf{h}}_1}} \right\|}{\bf{I}}_M, \ldots ,{\left\| {{{\bf{h}}_K}} \right\|}{\bf{I}}_M} \right]$, while ${\bf{p}}$ is the value as in \eqref{C3:MIMONOMASystemModel}.
The achievable sum-rate $R_{\mathrm{sum}}^{{\mathrm{MIMO-NOMA}}}$  is given in \eqref{C3:InstantSumRateMIMONOMA} and the equality $(a)$ in \eqref{C3:UpperBoundMIMONOMA} is obtained by a capacity-achieving MMSE-SIC\cite{Tse2005}.

Now, to prove the asymptotic tightness of the upper bound considered in \eqref{C3:UpperBoundMIMONOMA}, we have to consider a lower bound of the achievable sum-rate in \eqref{C3:InstantSumRateMIMONOMA} and prove that asymptotically the upper bound and the lower bound converge to the same expression.
For the uplink $K$-user $1 \times M$ MIMO system in \eqref{C3:MIMONOMASystemModel}, we assume that all the users transmit their signals subject to the power allocation ${\bf{p}} = \left[ {{p_1}, \ldots ,{p_K}} \right]$ and the BS utilizes an MRC-SIC receiver to retrieve the messages of all the $K$ users.
Then the achievable rate for user $k$ of the MIMO-NOMA system using the MRC-SIC receiver is given by:
\begin{align}\label{C3:MIMONOMAMRCSICIndividualAchievableRate}
R_{k,\mathrm{MRC-SIC}}^{{\mathrm{MIMO-NOMA}}} = {\ln}\left(1+ {\frac{{{p_k}{{\left\| {{{\bf{h}}_k}} \right\|}^2}}}{{\sum\limits_{i = k + 1}^K {p_i}{{\left\| {{{\bf{h}}_i}} \right\|}^2} {{\left| {\bf{e}}_k^{\mathrm{H}}{{\bf{e}}_i} \right|}^2}  + {N_0}}}} \right),
\end{align}
where ${{\bf{e}}_k} = \frac{{{{\bf{h}}_k}}}{{\left\| {{{\bf{h}}_k}} \right\|}}$ denotes the channel direction of user $k$.
Then, it becomes clear that the achievable sum-rate of the MIMO-NOMA system using the MRC-SIC receiver serves as a lower bound to the channel capacity in \eqref{C3:InstantSumRateMIMONOMA}, i.e., we have
\begin{align}\label{C3:LowerBoundMIMONOMA2}
R_{\mathrm{sum,MRC-SIC}}^{{\mathrm{MIMO-NOMA}}} = \sum \limits_{k=1}^{K} R_{k,\mathrm{MRC-SIC}}^{{\mathrm{MIMO-NOMA}}} \le R_{\mathrm{sum}}^{{\mathrm{MIMO-NOMA}}}.
\end{align}

Through the following theorem and corollaries, we first characterize the statistics of ${{\bf{e}}_k}$ as well as ${{\left| {\bf{e}}_k^{\mathrm{H}}{{\bf{e}}_i} \right|}^2}$ and derive the asymptotic achievable sum-rate of MIMO-NOMA employing an MRC-SIC receiver.
Then, we show that the upper bound considered in \eqref{C3:UpperBoundMIMONOMA} and the lower bound of \eqref{C3:LowerBoundMIMONOMA2} will asymptotically converge to the same limit for $K \to \infty$.

\begin{Lem}\label{C3:UnitSphere}
For ${\bf{h}}_k \sim \mathcal{CN}\left(\mathbf{0},\frac{1}{1+d_k^{\alpha}}{{\bf{I}}_M}\right)$, the normalized random vector (channel direction) ${{\bf{e}}_k} = \frac{{{{\bf{h}}_k}}}{{\left\| {{{\bf{h}}_k}} \right\|}}$ is uniformly distributed on a unit sphere in $\mathbb{C}^{M}$.
\end{Lem}
\begin{proof}
According to the system model of ${{{\bf{h}}_k}} = \frac{{\bf{g}}_k}{\sqrt{1+d_k^{\alpha}}}$ with ${\bf{g}}_k \sim \mathcal{CN}\left(\mathbf{0},{{\bf{I}}_M}\right)$, we have ${\bf{h}}_k \sim \mathcal{CN}\left(\mathbf{0},\frac{1}{1+d_k^{\alpha}}{{\bf{I}}_M}\right)$.
The distribution of ${{\bf{e}}_k}$ can be proven by exploiting the orthogonal-invariance of the multivariate normal distribution.
In particular, for any orthogonal matrix $\mathbf{Q}$, we have $\mathbf{Q}{\bf{h}}_k \sim \mathcal{CN}\left(\mathbf{0},\frac{1}{1+d_k^{\alpha}}{{\bf{I}}_M}\right)$, which means that the distribution of ${\bf{h}}_k$ is invariant to rotations (orthogonal transform).
Then, ${{\bf{e}}_k} = \frac{{{\mathbf{Q}{\bf{h}}_k}}}{{\left\| {{\mathbf{Q}{\bf{h}}_k}} \right\|}} = \frac{{{\mathbf{Q}{\bf{h}}_k}}}{{\left\| {{{\bf{h}}_k}} \right\|}}$ is also invariant to rotation.
Meanwhile, we have ${{\left\| {{{\bf{e}}_k}} \right\|}} = 1$ for sure.
Therefore, ${{\bf{e}}_k}$ must be uniformly distributed on a unit sphere on $\mathbb{C}^{M}$.
\end{proof}

\begin{Cor}\label{C3:Corollary1}
The channel direction of user $k$, ${{\bf{e}}_k}$, is independent of its channel gain ${{\left\| {{{\bf{h}}_k}} \right\|}}$.
\end{Cor}
\begin{proof}
According to Lemma \ref{C3:UnitSphere}, the channel direction ${{\bf{e}}_k}$ is uniformly distributed on a unit sphere on $\mathbb{C}^{M}$, regardless of the value of ${{\left\| {{{\bf{h}}_k}} \right\|}}$.
Therefore, ${{\bf{e}}_k}$ is independent of ${{\left\| {{{\bf{h}}_k}} \right\|}}$.
\end{proof}

\begin{Cor}
The mean and covariance matrix of ${{\bf{e}}_k}$ are given by
\begin{align}
{\mathrm{E}}\left\{ {{{\bf{e}}_k}} \right\} = {\bf{0}}\; \text{and} \;
{\mathrm{E}}\left\{ {{{\bf{e}}_k}{\bf{e}}_k^{\mathrm{H}}} \right\} = \frac{1}{M}{{\bf{I}}_M},
\end{align}
respectively.
\end{Cor}
\begin{proof}
Due to the symmetry of the uniform spherical distribution, ${{{\bf{e}}_k}}$ and $-{{{\bf{e}}_k}}$ have the same distribution and thus we have ${\mathrm{E}}\left\{ {{{\bf{e}}_k}} \right\} =  {\mathrm{E}}\left\{ -{{{\bf{e}}_k}} \right\}$ and hence ${\mathrm{E}}\left\{ {{{\bf{e}}_k}} \right\} = {\bf{0}}$.
For the reason of symmetry, ${{{\bf{e}}_k}} = \left[ {{e_{k,1}}, \ldots ,{e_{k,m}}, \ldots ,{e_{k,M}}} \right]$ and ${{{\bf{e}}_k'}} = \left[ {{e_{k,1}}, \ldots ,-{e_{k,m}}, \ldots ,{e_{k,M}}} \right]$ have the same distribution, where ${e_{k,m}}$ denotes the $m$-th entry in ${{{\bf{e}}_k}}$.
Therefore, we have
\begin{equation}
{\mathrm{E}}\left\{ {{e_{k,m}}e_{k,n}^*} \right\} = {\mathrm{E}}\left\{ { - {e_{k,m}}e_{k,n}^*} \right\} =  - {\mathrm{E}}\left\{ {{e_{k,m}}e_{k,n}^*} \right\}, \forall m \neq n,
\end{equation}
which implies that the covariance terms are zero, i.e., ${\mathrm{E}}\left\{ {{e_{k,m}}e_{k,n}^*} \right\} = 0$, $\forall m \neq n$.
Note that, the zero covariance terms only reflect the lack of correlation between ${e_{k,m}}$ and $e_{k,n}$, but not their independence.
In fact, the entries of ${{{\bf{e}}_k}}$ are dependent on each other, i.e., increasing one entry will decrease all the other entries due to ${{\left\| {{{\bf{e}}_k}} \right\|}} = 1$.
As for the variance, since ${{{\bf{e}}_k}}$ has been normalized, we have
\begin{equation}
\sum\limits_{m = 1}^M {{\mathrm{E}}\left\{ {e_{k,m}^2} \right\}}  = {\mathrm{E}}\left\{ {\sum\limits_{m = 1}^M {e_{k,m}^2} } \right\} = 1.
\end{equation}
Again, based on the symmetry of the uniform spherical distribution, we have ${{\mathrm{E}}\left\{ {e_{k,m}^2} \right\}} = {{\mathrm{E}}\left\{ {e_{k,n}^2} \right\}}$, $\forall m,n$, and hence we have ${{\mathrm{E}}\left\{ {e_{k,m}^2} \right\}} = \frac{1}{M}$ and ${\mathrm{E}}\left\{ {{{\bf{e}}_k}{\bf{e}}_k^{\mathrm{H}}} \right\} = \frac{1}{M}{{\bf{I}}_M}$.
This completes the proof.
\end{proof}

Let us now define a scalar random variable as ${\nu_{k,i}} = {{\bf{e}}_k^{\mathrm{H}}{{\bf{e}}_i}} \in \mathbb{C}$, which denotes the projection of channel direction of user $k$ on the channel direction of user $i$.
Note that the random variable ${\nu_{k,i}}$ can characterize the IUI during MRC in \eqref{C3:MIMONOMAMRCSICIndividualAchievableRate}.
Additionally, thanks to the independence between ${{\bf{e}}_k}$ and ${{\left\| {{{\bf{h}}_k}} \right\|}}$, ${\nu_{k,i}}$ is independent of ${{\left\| {{{\bf{h}}_k}} \right\|}}$ and ${{\left\| {{{\bf{h}}_i}} \right\|}}$.
The following Lemma characterizes the mean and variance of ${\nu_{k,i}}$.

\begin{Lem}\label{C3:Lemma2}
For ${\bf{h}}_k \sim \mathcal{CN}\left(\mathbf{0},\frac{1}{1+d_k^{\alpha}}{{\bf{I}}_M}\right)$ and ${{\bf{e}}_k} = \frac{{{{\bf{h}}_k}}}{{\left\| {{{\bf{h}}_k}} \right\|}}$, the random variable ${\nu_{k,i}} = {{\bf{e}}_k^{\mathrm{H}}{{\bf{e}}_i}}$ has a zero mean and variance of $\frac{1}{M}$.
\end{Lem}
\begin{proof}
In fact, ${\nu_{k,i}}$ denotes the projection of ${\bf{e}}_k$ on ${{\bf{e}}_i}$, where ${\bf{e}}_k$ and ${{\bf{e}}_i}$ are uniformly distributed in a unit sphere on $\mathbb{C}^{M}$.
Upon fixing one channel direction ${{\bf{e}}_k}$, the conditional mean and variance of ${\nu_{k,i}}$ are given by
\begin{equation}\label{C3:ConditionalMeanVariance}
{\mathrm{E}}\left\{ {{\nu_{k,i}}\left| {{{\bf{e}}_k}} \right.} \right\} = {\bf{e}}_k^{\mathrm{H}}{\mathrm{E}}\left\{ {{{\bf{e}}_i}} \right\} =0 \;\text{and}\;
{\mathrm{E}}\left\{ {\left|{\nu_{k,i}}\right|^2\left| {{{\bf{e}}_k}} \right.} \right\} = {{\bf{e}}_k^{\mathrm{H}}{\mathrm{E}}\left\{ {{{\bf{e}}_i}{\bf{e}}_i^{\mathrm{H}}} \right\}{{\bf{e}}_k}} = \frac{1}{M},
\end{equation}
respectively.
Since ${{\bf{e}}_k}$ is uniformly distributed, the integral over ${{\bf{e}}_k}$ will not change the mean and variance.
Therefore, the mean and variance of ${\nu_{k,i}}$ are given by
\begin{align}\label{C3:ConditionalMeanVarianceII}
{\mathrm{E}}\left\{ {\nu_{k,i}} \right\} = 0 \;\text{and}\;
{\mathrm{E}}\left\{ \left|{\nu_{k,i}}\right|^2 \right\} = \frac{1}{M},
\end{align}
respectively, which completes the proof.
\end{proof}

Now, based on \eqref{C3:MIMONOMAMRCSICIndividualAchievableRate}, we have the asymptotic achievable data rate of user $k$ as follows:
\begin{align}\label{C3:AsymptoticSumRateMIMONOMA}
\mathop {\lim }\limits_{K \to \infty } R_{k,\mathrm{MRC-SIC}}^{{\mathrm{MIMO-NOMA}}} & = \mathop {\lim }\limits_{K \to \infty } {\ln}\left(1+ {\frac{{{p_k}{{\left\| {{{\bf{h}}_k}} \right\|}^2}}}{{\sum\limits_{i = k + 1}^K {p_i}{{\left\| {{{\bf{h}}_i}} \right\|}^2} \left|{\nu_{k,i}}\right|^2  + {N_0}}}} \right) \notag\\
& \mathop  =\limits^{(a)} \mathop {\lim }\limits_{K \to \infty } {\ln}\left(1+ {\frac{{{p_k}{{\left\| {{{\bf{h}}_k}} \right\|}^2}}}{{\sum\limits_{i = k + 1}^K {p_i}{{\left\| {{{\bf{h}}_i}} \right\|}^2} \frac{1}{M}  + {N_0}}}} \right) \notag\\
& \mathop  = \limits^{(b)} \mathop {\lim }\limits_{K \to \infty } M{\ln}\left(1+ {\frac{{{p_k}{{\left\| {{{\bf{h}}_k}} \right\|}^2} \frac{1}{M}}}{{\sum\limits_{i = k + 1}^K {p_i}{{\left\| {{{\bf{h}}_i}} \right\|}^2} \frac{1}{M}  + {N_0}}}} \right).
\end{align}
Note that the equality in $(a)$ holds asymptotically by applying Corollary \ref{C3:Corollary1} and Lemma \ref{C3:Lemma2} with $K \to \infty$.
In addition, the equality in $(b)$ holds with $K \to \infty$ since $\mathop {\lim }\limits_{x \to 0} {\ln}\left( {1 + Mx} \right) = \mathop {\lim }\limits_{x \to 0} M{\ln}\left( {1 + x} \right)$.
As a result, the asymptotic achievable sum-rate in \eqref{C3:LowerBoundMIMONOMA2} can be obtained by
\begin{equation}\label{C3:AchievableRateMIMONOMA2}
\mathop {\lim }\limits_{K \to \infty } R_{\mathrm{sum,MRC-SIC}}^{{\mathrm{MIMO-NOMA}}}
= \mathop {\lim }\limits_{K \to \infty } M{\ln}\left( {1 + \frac{{1}}{{M{N_0}}}\sum\limits_{k = 1}^K {p_k{{\left\| {{{\bf{h}}_k}} \right\|}^2}} } \right).
\end{equation}

Now, upon comparing \eqref{C3:UpperBoundMIMONOMA}, \eqref{C3:LowerBoundMIMONOMA2}, and \eqref{C3:AchievableRateMIMONOMA2}, it can be observed that the upper bound and the lower bound considered converge when $K \to \infty$.
In other words, for any given power allocation strategy ${\bf{p}} = \left[ {{p_1}, \ldots ,{p_K}} \right]$, the upper bound in \eqref{C3:UpperBoundMIMONOMA} is asymptotically tight.
It completes the proof.

\section{Proof of Theorem \ref{C3:Theorem2}}\label{C3:AppendixB}
Based on \eqref{C3:AsymptoticSumRateMIMONOMA} in the proof of Theorem \ref{C3:Theorem1} in Appendix \ref{C3:AppendixA}, under the equal resource allocation strategy, i.e., ${{p_{k}}} = \frac{{P_{\mathrm{max}}}}{K}$, $\forall k$, the asymptotic individual rate of user $k$ of the \emph{m}MIMO-NOMA system with the MRC-SIC detection in \eqref{C3:mMIMONOMAIndividualAchievableRate} can be obtained by
\begin{equation}\label{C3:mMIMONOMAIndividualAchievableRate2}
\mathop {\lim }\limits_{K \rightarrow \infty} R_{k}^{\mathrm{\emph{m}MIMO-NOMA}} = \mathop {\lim }\limits_{K \rightarrow \infty} {\ln}\left(1+ {\frac{{{P_{\mathrm{max}}}{{\left\| {{{\bf{h}}_k}} \right\|}^2}}}{{\sum\limits_{i = k + 1}^K {P_{\mathrm{max}}}{{\left\| {{{\bf{h}}_i}} \right\|}^2} \frac{1}{M}  + {KN_0}}}} \right).
\end{equation}
With the aid of a large-scale antenna array, i.e., $M \to \infty$, the fluctuation imposed by the small-scale fading on the channel gain can be averaged out as a benefit of channel hardening\cite{HochwaldMassiveMIMO}.
Therefore, the channel gain is mainly determined by the large-scale fading asymptotically as follows:
\begin{equation}
\mathop {\lim }\limits_{M \rightarrow \infty}\frac{{\left\| {{{\bf{h}}_k}} \right\|}^2}{M} = \frac{1}{1+d_k^{\alpha}}.
\end{equation}
As a result, the asymptotic data rate of user $k$ in \eqref{C3:mMIMONOMAIndividualAchievableRate2} is given by:
\begin{equation}\label{C3:mMIMONOMAIndividualAchievableRate3}
\mathop {\lim }\limits_{K \rightarrow \infty, M \rightarrow \infty} R_{k}^{\mathrm{\emph{m}MIMO-NOMA}} = \mathop {\lim }\limits_{K \rightarrow \infty, M \rightarrow \infty} {\ln}\left(1+ {\frac{M{{P_{\mathrm{max}}}\frac{1}{1+d_k^{\alpha}}}}{{\sum\limits_{i = k + 1}^K {P_{\mathrm{max}}}\frac{1}{1+d_i^{\alpha}}  + {KN_0}}}} \right).
\end{equation}

Based on the theory of order statistics\cite{David2003order}, the PDF of $d_k$ is given by
\begin{equation}\label{C3:OrderedDistancePDF}
{f_{{d_k}}}\left( x \right) = k\left( {\begin{array}{*{20}{c}}
{K}\\
{k}
\end{array}} \right) {F_{{d}}}^{k-1}\left( x \right) \left(1-{F_{{d}}}\left( x \right)\right)^{K-k} {f_{{d}}}\left( x \right), D_0 \le z \le D.
\end{equation}
Thus, the mean of the large-scale fading of user $k$ can be written as
\begin{align}\label{C3:OrderedDistanceMean}
I_k &= {{\mathrm{E}}_{{d_k}}}\left\{ {\frac{1}{{1 + d_k^\alpha }}} \right\} = \int_{{D_0}}^D {\frac{1}{{1 + {x^\alpha }}}} {f_{{d_k}}}\left( x \right)dx  \notag\\
& \approx \left( {\begin{array}{*{20}{c}}
K\\
k
\end{array}} \right){\frac{{k}}{{D + {D_0}}}} \sum\limits_{n = 1}^N \frac{{\beta _n}}{c_n} {\left( {\frac{{\phi _n^2 - D_0^2}}{{{D^2} - D_0^2}}} \right)^{k - 1}}{\left( {\frac{{{D^2} - \phi _n^2}}{{{D^2} - D_0^2}}} \right)^{K - k}},
\end{align}
with $\phi_n = {\frac{D-D_0}{2}\cos \frac{{2n - 1}}{{2N}}\pi  + \frac{D+D_0}{2}}$.
For a large number of users, i.e., $K \to \infty$, the random IUI term in \eqref{C3:mMIMONOMAIndividualAchievableRate3} can be approximated by a deterministic value given by
\begin{equation}\label{C3:DeterministicInterference}
\mathop {\lim }\limits_{K \rightarrow \infty} \sum\limits_{i = k + 1}^K {P_{\mathrm{max}}}\frac{1}{1+d_i^{\alpha}} \approx \mathop {\lim }\limits_{K \rightarrow \infty} \sum\limits_{i = k + 1}^K {P_{\mathrm{max}}} I_i.
\end{equation}
Now, the asymptotic ergodic data rate of user $k$ can be approximated by
\begin{align}\label{C3:mMIMONOMAErgodicRate}
&\mathop {\lim }\limits_{K \rightarrow \infty, M \rightarrow \infty} \overline{R_{k}^{\mathrm{\emph{m}MIMO-NOMA}}}= \mathop {\lim }\limits_{K \rightarrow \infty, M \rightarrow \infty} \int_{{D_0}}^D {\ln \left( {1 + \frac{{{\psi _k}}}{{1 + {x^\alpha }}}} \right){f_{{d_k}}}\left( x \right)dx} \notag\\
& \approx \mathop {\lim }\limits_{K \to \infty ,M \to \infty } \left( {\begin{array}{*{20}{c}}
K\\
k
\end{array}} \right){\frac{{k}}{{D + {D_0}}}} \sum\limits_{n = 1}^N {{\beta _n}\ln \left( {1 + \frac{{{\psi _k}}}{{{c_n}}}} \right)} {\left( {\frac{{\phi _n^2 - D_0^2}}{{{D^2} - D_0^2}}} \right)^{k - 1}}{\left( {\frac{{{D^2} - \phi _n^2}}{{{D^2} - D_0^2}}} \right)^{K - k}},
\end{align}
with ${\psi _k} = \frac{{{P_{\mathrm{max}}}M}}{{\sum\nolimits_{i = k + 1}^K {{P_{\mathrm{max}}}{I_i} + {KN_0}} }}$.
Substituting \eqref{C3:mMIMONOMAErgodicRate} into \eqref{C3:mMIMONOMASumRate} yields the asymptotic ergodic sum-rate of the \emph{m}MIMO-NOMA system with the MRC-SIC detection as in \eqref{C3:ErgodicSumRatemMIMONOMA}, which completes the proof.

\section{Proof of Theorem \ref{C3:Theorem3}}\label{C3:AppendixC}
With $D=D_0$, based on the channel hardening property\cite{HochwaldMassiveMIMO}, the channel gain can be asymptotically formulated as:
\begin{equation}\label{C3:ChannelHardening}
\mathop {\lim }\limits_{M \rightarrow \infty}\frac{{\left\| {{{\bf{h}}_k}} \right\|}^2}{M} = \frac{1}{1+D_0^{\alpha}}, \forall k.
\end{equation}
Substituting \eqref{C3:ChannelHardening} into \eqref{C3:mMIMONOMAIndividualAchievableRate2}, the asymptotic individual rate of user $k$ of the \emph{m}MIMO-NOMA system with $D=D_0$ is obtained by
\begin{align}\label{C3:DD0mMIMONOMAIndividualAchievableRate2}
\mathop {\lim }\limits_{K \rightarrow \infty, M \rightarrow \infty} R_{k}^{\mathrm{\emph{m}MIMO-NOMA}} &= \mathop {\lim }\limits_{K \rightarrow \infty, M \rightarrow \infty} {\ln}\left(1+ M{\frac{{{P_{\mathrm{max}}}\frac{1}{1+D_0^{\alpha}}}}{{\sum\limits_{i = k + 1}^K {P_{\mathrm{max}}} \frac{1}{1+D_0^{\alpha}}  + K{N_0}}}} \right) \notag\\
& = \mathop {\lim }\limits_{K \rightarrow \infty, M \rightarrow \infty} \ln \left( {1 + \frac{{\delta \varpi }}{{\left( {1 - \frac{k}{K}} \right)\varpi  + 1}}} \right),
\end{align}
where $\delta = \frac{M}{K}$ and $\varpi = \frac{{P_{\mathrm{max}}}}{\left({1+D_0^{\alpha}}\right) N_0}$.
We can observe that the asymptotic individual rate of user $k$ in \eqref{C3:DD0mMIMONOMAIndividualAchievableRate2} becomes a deterministic value for $K \rightarrow \infty$ and $M \rightarrow \infty$ due to the channel hardening property.
As a result, we have $\mathop {\lim }\limits_{K \rightarrow \infty, M \rightarrow \infty} R_{k}^{\mathrm{\emph{m}MIMO-NOMA}} = \mathop {\lim }\limits_{K \rightarrow \infty, M \rightarrow \infty} \overline{R_{k}^{\mathrm{\emph{m}MIMO-NOMA}}}$.

Now, the asymptotic ergodic sum-rate of the \emph{m}MIMO-NOMA system with MRC-SIC receiver can be obtained by
\begin{align}\label{C3:DD0mMIMONOMAErgodicRate}
&\mathop {\lim }\limits_{K \rightarrow \infty, M \rightarrow \infty} \overline{R_{\mathrm{sum}}^{\mathrm{\emph{m}MIMO-NOMA}}} = \mathop {\lim }\limits_{K \to \infty ,M \to \infty } \sum\limits_{k = 1}^K \ln \left( {1 + \frac{{\delta \varpi }}{{\left( {1 - \frac{k}{K}} \right)\varpi  + 1}}} \right) \\
&= \mathop {\lim }\limits_{K \to \infty ,M \to \infty } K {{\int_{0}^1 {\ln \left( {1 + \frac{{\delta \varpi }}{{\left( {1 - x} \right)\varpi  + 1}}} \right) dx} } } \notag\\
&= \mathop {\lim }\limits_{K \to \infty ,M \to \infty } \frac{{M}}{\varpi\delta} \left[ \ln \left( 1 \hspace{-1mm}+\hspace{-0.5mm} \varpi\delta \hspace{-0.5mm}+\hspace{-0.5mm} \varpi \right) \left( 1 \hspace{-0.5mm}+\hspace{-0.5mm} \varpi\delta \hspace{-0.5mm}+\hspace{-0.5mm} \varpi \right) - \ln \left( 1 \hspace{-0.5mm}+\hspace{-0.5mm} \varpi\delta \right)\left( 1 \hspace{-0.5mm}+\hspace{-0.5mm} \varpi\delta \right) - \ln \left( 1 \hspace{-0.5mm}+\hspace{-0.5mm} \varpi \right)\left( 1 \hspace{-0.5mm}+\hspace{-0.5mm} \varpi \right) \right], \notag
\end{align}
which completes the proof of \eqref{C3:DD0ErgodicSumRatemMIMONOMA}.

On the other hand, under the equal resource allocation strategy, the asymptotic individual rate of user $k$ of the \emph{m}MIMO-OMA system with the MRC detection in \eqref{C3:mMIMOOMAIndividualAchievableRate} can be approximated by
\begin{equation}\label{C3:DD0mMIMOOMAIndividualAchievableRate}
R_{k}^{\mathrm{\emph{m}MIMO-OMA}} \approx \delta\varsigma{\ln}\left(1+ {\frac{{{P_{\mathrm{max}}}{{\left\| {{{\bf{h}}_k}} \right\|}^2}}}{{{\varsigma M N_0}}}} \right).
\end{equation}
Exploiting the channel hardening property as stated in \eqref{C3:ChannelHardening}, the individual rate of user $k$ in \eqref{C3:DD0mMIMOOMAIndividualAchievableRate} can be approximated by a deterministic value and we have the asymptotic ergodic sum-rate of the  \emph{m}MIMO-OMA system considered as
\begin{equation}\label{C3:DD0mMIMOOMAErgodicRate}
\mathop {\lim }\limits_{M \rightarrow \infty} \overline{R_{\mathrm{sum}}^{\mathrm{\emph{m}MIMO-OMA}}}
\approx \mathop {\lim }\limits_{M \to \infty } {\varsigma M} \ln \left( {1 + \frac{\varpi}{\varsigma}}\right),
\end{equation}
which completes the proof of \eqref{C3:DD0ErgodicSumRatemMIMOOMA}.

\chapter{Proof of Theories of Chapter 4}\label{C4:appendix_2b}

\section{Proof of Theorem \ref{C4:theorem1}}\label{C4:AppendixB}
It is clear that the random variable ${\phi}$ for a given ${\mathbf{y}}$ is complex Gaussian distributed, where its conditional mean and variance are given by
\begin{align}
\mathrm{E}\left\{ {{\phi}|{\mathbf{y}}} \right\} &= {{\mathbf{y}}^{\mathrm{H}}}\mathrm{E}\left\{ {\mathbf{x}} \right\}/{{\left| {\mathbf{y}} \right|}} = 0 \;\text{and}\notag\\
\mathrm{Var}\left\{ {{\phi}|{\mathbf{y}}} \right\} &= {{{\mathbf{y}}^{\mathrm{H}}\mathrm{E}\left\{ {{\mathbf{x}}{{\mathbf{x}}^{\mathrm{H}}}} \right\}{\mathbf{y}}}}/{{{{\left| {\mathbf{y}} \right|}^2}}} = \sigma_x^2,
\end{align}
respectively.
Since the conditional mean and variance of ${\phi}$ are uncorrelated with ${\mathbf{y}}$, hence ${\phi}$ is independent of ${\mathbf{y}}$ with zero mean and variance of $\sigma_x^2$.
This completes the proof.

\chapter{Proof of Theories of Chapter 5}\label{C5:appendix_2b}

\section{Proof of Theorem \ref{C5:Theorem1}}\label{C5:AppendixA}
In the following, we prove Theorem \ref{C5:Theorem1} by comparing the total transmit power for four kinds of SIC policies. Given user $m$ and user $n$ multiplexed on subcarrier $i$, there are following four possible cases on SIC decoding order:

\begin{itemize}
  \item Case I: $u_{i,m} = 1$, $u_{i,n} = 0$,
  \item Case II: $u_{i,m} = 0$, $u_{i,n} = 1$,
  \item Case III: $u_{i,m} = 1$, $u_{i,n} = 1$,
  \item Case IV: $u_{i,m} = 0$, $u_{i,n} = 0$,
\end{itemize}
\noindent where Case I and Case II correspond to selecting only user $m$ or user $n$ to perform SIC, respectively. Case III and Case IV correspond to selecting both users or neither user to perform SIC, respectively.
In Case I, user $m$ is selected to perform SIC and user $n$ directly decodes its own message. According to \eqref{C5:OutageProbabilityWithSIC} and \eqref{C5:OutageProbabilityWithoutSIC}, the outage probabilities of users $m$ and $n$ are given by
\begin{align}
    {\mathrm{P}}_{i,m}^{{\mathrm{out}}} &= {\mathrm{Pr}}\left\{ {\frac{{{{\left| {{h_{i,m}}} \right|}^2}}}{{\sigma _{i,m}^2}}< \max \left( {\frac{{{\gamma _{i,n}}}}{{{p_{i,n}} - {p_{i,m}}{\gamma _{i,n}}}},\frac{{{\gamma _{i,m}}}}{{{p_{i,m}}}}} \right)} \right\}\; \text{and} \;\; \notag\\
    {\mathrm{P}}_{i,n}^{{\mathrm{out}}} &= {\mathrm{Pr}}\left\{ \frac{{{{\left| {{h_{i,n}}} \right|}^2}}}{{\sigma _{i,n}^2}} < \frac{{{\gamma _{i,n}}}}{{{p_{i,n}} - {p_{i,m}}{\gamma _{i,n}}}}  \right\},\label{C5:OutageProbabilityWithoutSIC2}
\end{align}
respectively. Note that a prerequisite $p_{i,n} - {p_{i,m}}{\gamma _{i,n}} > 0$ should be satisfied, otherwise the SIC will never be successful, i.e., ${\mathrm{P}}_{i,m}^{{\mathrm{out}}}=1$.

Combining the threshold definition in \eqref{C5:OutageThreshold1} and the QoS constraint C5 in \eqref{C5:P1}, the feasible solution set spanned by \eqref{C5:OutageProbabilityWithoutSIC2} can be characterized by the following equations:
\begin{align}
    p_{i,n} - {p_{i,m}}{\gamma _{i,n}} > 0,\; \frac{{{\gamma _{i,n}}}}{{{p_{i,n}} - {p_{i,m}}{\gamma _{i,n}}}} &\le {\beta _{i,n}},\; \text{and}\label{C5:FeasibleSolution13a}\\
    \max \left( {\frac{{{\gamma _{i,n}}}}{{{p_{i,n}} - {p_{i,m}}{\gamma _{i,n}}}},\frac{{{\gamma _{i,m}}}}{{{p_{i,m}}}}} \right) &\le  {\beta _{i,m}}. \label{C5:FeasibleSolution13b}
\end{align}
Then, recall that ${\beta _{i,m}} \ge {\beta _{i,n}}$, we have
${p_{i,m}} \ge \frac{{{\gamma _{i,m}}}}{{{\beta _{i,m}}}} \;\; \text{and}\;\;
{p_{i,n}} \ge \frac{{{\gamma _{i,n}}}}{{{\beta _{i,n}}}} + \frac{{{\gamma _{i,m}}{\gamma _{i,n}}}}{{{\beta _{i,m}}}}$.
Then, the optimal power allocation for user $m$ and user $n$ on subcarrier $i$ in Case I are given by
\begin{equation}\label{C5:PowerAllocation12}
    {p_{i,m}^{\mathrm{I}}} = \frac{{{\gamma _{i,m}}}}{{{\beta _{i,m}}}} \;\;  \text{and}\;\;
    {p_{i,n}^{\mathrm{I}}} = \frac{{{\gamma _{i,n}}}}{{{\beta _{i,n}}}} + \frac{{{\gamma _{i,m}}{\gamma _{i,n}}}}{{{\beta _{i,m}}}},
\end{equation}
respectively, with the total transmit power
\begin{equation}\label{C5:PowerComsumptionPerSubcarrier1}
    p^{\mathrm{total}}_{\mathrm{I}} = \frac{{{\gamma _{i,m}}}}{{{\beta _{i,m}}}} + \frac{{{\gamma _{i,n}}}}{{{\beta _{i,n}}}} + \frac{{{\gamma _{i,m}}{\gamma _{i,n}}}}{{{\beta _{i,m}}}}.
\end{equation}

Similarly, we can derive the total transmit power for Cases II, III, and IV as follows:
\begin{align}
\hspace{-1mm}p^{\mathrm{total}}_{\mathrm{II}}& \hspace{-0.5mm}=\hspace{-0.5mm} \frac{{{\gamma _{i,n}}}}{{{\beta _{i,n}}}} + \frac{{{\gamma _{i,m}}}}{{{\beta _{i,n}}}} + \frac{{{\gamma _{i,m}}{\gamma _{i,n}}}}{{{\beta _{i,n}}}},\label{C5:PowerComsumptionPerSubcarrier2}\\
\hspace{-1mm}p^{\mathrm{total}}_{\mathrm{III}}&\hspace{-0.5mm}=\hspace{-0.5mm} \frac{1}{{1 \hspace{-0.75mm}-\hspace{-0.75mm} {\gamma _{i,m}}{\gamma _{i,n}}}}\hspace{-0.75mm}\left( \frac{{{\gamma _{i,m}}} \hspace{-0.75mm}+\hspace{-0.75mm} {{\gamma _{i,m}}{\gamma _{i,n}}}}{{{\beta _{i,n}}}} \hspace{-0.75mm}+\hspace{-0.75mm} \frac{{{\gamma _{i,n}}}\hspace{-0.75mm}+\hspace{-0.75mm}{{\gamma _{i,n}}{\gamma _{i,m}}}}{{{\beta _{i,m}}}} \right)\hspace{-1mm}, \; \text{and}\label{C5:PowerComsumptionPerSubcarrier3}\\
\hspace{-1mm}p^{\mathrm{total}}_{\mathrm{IV}}&\hspace{-0.5mm}=\hspace{-0.5mm} \frac{1}{{1 \hspace{-0.75mm}-\hspace{-0.75mm} {\gamma _{i,m}}{\gamma _{i,n}}}}\hspace{-0.75mm}\left( \frac{{{\gamma _{i,m}}} \hspace{-0.75mm}+\hspace{-0.75mm} {{\gamma _{i,m}}{\gamma _{i,n}}}}{{{\beta _{i,m}}}} \hspace{-0.75mm}+\hspace{-0.75mm} \frac{{{\gamma _{i,n}}}\hspace{-0.75mm}+\hspace{-0.75mm}{{\gamma _{i,n}}{\gamma _{i,m}}}}{{{\beta _{i,n}}}} \right)\hspace{-1mm},\label{C5:PowerComsumptionPerSubcarrier4}
\end{align}
where in \eqref{C5:PowerComsumptionPerSubcarrier3} and \eqref{C5:PowerComsumptionPerSubcarrier4}, it is required that $0<1-\gamma _{i,m}\gamma _{i,n}<1$ is satisfied, otherwise no feasible power allocation can satisfy the QoS constraint. Besides, two prerequisites for \eqref{C5:PowerComsumptionPerSubcarrier3} are
\begin{align}
p_{i,m} &= \frac{1}{{1 - {\gamma _{i,m}}{\gamma _{i,n}}}}\left( \frac{{{\gamma _{i,m}}}}{{{\beta _{i,n}}}} + \frac{{{\gamma _{i,m}}{\gamma _{i,n}}}}{{{\beta _{i,m}}}}\right) \ge \frac{{{\gamma _{i,m}}}}{{{\beta _{i,m}}}}\; \text{and}\notag\\
p_{i,n} &= \frac{1}{{1 - {\gamma _{i,m}}{\gamma _{i,n}}}}\left( \frac{{{\gamma _{i,n}}}}{{{\beta _{i,m}}}} + \frac{{{\gamma _{i,m}}{\gamma _{i,n}}}}{{{\beta _{i,n}}}}\right) \ge \frac{{{\gamma _{i,n}}}}{{{\beta _{i,n}}}}, \label{C5:CaseIII1}
\end{align}
respectively, otherwise there is no feasible solution.
From \eqref{C5:PowerComsumptionPerSubcarrier1}, \eqref{C5:PowerComsumptionPerSubcarrier2}, \eqref{C5:PowerComsumptionPerSubcarrier3}, and \eqref{C5:PowerComsumptionPerSubcarrier4}, we obtain
\begin{equation}
p^{\mathrm{total}}_{\mathrm{II}} \ge p^{\mathrm{total}}_{\mathrm{I}},\;
p^{\mathrm{total}}_{\mathrm{III}} > p^{\mathrm{total}}_{\mathrm{I}},\; \text{and}\;
p^{\mathrm{total}}_{\mathrm{IV}} > p^{\mathrm{total}}_{\mathrm{I}},\label{C5:Conclusion3}
\end{equation}
which means that the SIC decoding order in Case I is optimal for minimizing the total transmit power. Note that the relationship of $p^{\mathrm{total}}_{\mathrm{III}} > p^{\mathrm{total}}_{\mathrm{I}}$ can be easily obtained by lower bounding $p_{i,n}$ by $\frac{{{\gamma _{i,n}}}}{{{\beta _{i,n}}}}$ according to \eqref{C5:CaseIII1}. Interestingly, we have $p^{\mathrm{total}}_{\mathrm{II}} = p^{\mathrm{total}}_{\mathrm{I}}$ for ${\beta _{i,m}} = {\beta _{i,n}}$, which means that it will consume the same total transmit power for Case I and Cases II when both users have the same CNR outage threshold.

\section{Proof of Theorem \ref{C5:Theorem2}} \label{C5:AppendixB}
The constraint relaxed problem in \eqref{C5:P3Continuous} is equivalent to the problem in \eqref{C5:P3} if the optimal solution of \eqref{C5:P3Continuous} still satisfies the relaxed constraints in \eqref{C5:P3}.
For the optimal solution of \eqref{C5:P3Continuous}, $\left({{\overline{s}_{i,m}^*}},{\gamma}_{i,m}^*\right)$, $i \in \left\{1, \ldots ,{N_{\mathrm{F}}}\right\}$, $m \in \left\{ 1, \ldots ,M \right\}$, and the corresponding optimal value, $p^{\mathrm{total}}_{\boldsymbol{{\gamma}}} \left({\boldsymbol{{\gamma} }^*} \right)$, we have the following relationship according to constraint $\overline{\text{C8}}$ in \eqref{C5:P3Continuous}:
\begin{equation}\label{C5:C5Again}
{\overline{s}_{i,m}^*} =
\left\{
\begin{array}{ll}
1 & \text{if}\;{\gamma}_{i,m}^* > 0,\\
\overline{s}_{i,m}^* \in \left[ {0,\;1} \right] & \text{if}\;{\gamma}_{i,m}^* = 0,
\end{array}
\right.
\end{equation}
where ${\boldsymbol{{\gamma} }^*} \in \mathbb{R}^{N_{\mathrm{F}} M \times 1}$ denotes the set of ${\gamma}_{i,m}^*$.
Note that reducing $\overline{s}_{i,m}^*$ to zero where ${\gamma}_{i,m}^* = 0$ will not change the optimal value $p^{\mathrm{total}}_{\boldsymbol{{\gamma}}} \left({\boldsymbol{{\gamma} }^*} \right)$ and will not violate constraints $\overline{\text{C1}}$, $\overline{\text{C6}}$, and $\overline{\text{C8}}$ in \eqref{C5:P3Continuous}.
Through the mapping relationship \eqref{C5:OptimalConvert2}, we have
\begin{align}
{s}_{i,m}^* \in \left\{0,\;1\right\} &\subseteq \left[ {0,\;1} \right],\;{{\gamma}}_{i,m}^* = {{s}_{i,m}^*} {\gamma}_{i,m}^*, \; \text{and} \\
\sum\limits_{m = 1}^M {{{s}_{i,m}^*}} &\le \sum\limits_{m = 1}^M {{\overline{s}_{i,m}^*}} \le 2,
\end{align}
which implies that $\left({{{s}_{i,m}^*}},{\gamma}_{i,m}^* \right)$ is also the optimal solution of \eqref{C5:P3Continuous} with the same optimal objective value $p^{\mathrm{total}}_{\boldsymbol{{\gamma}}} \left({\boldsymbol{{\gamma} }^*} \right)$. More importantly, $\left({{{s}_{i,m}^*}},{\gamma}_{i,m}^* \right)$ can also satisfy the constraints C1, C6, and C8 in \eqref{C5:P3}.
Therefore, the problem in \eqref{C5:P3Continuous} is equivalent to the problem in \eqref{C5:P3}, and the optimal solution of \eqref{C5:P3} can be obtained via the mapping relationship in \eqref{C5:OptimalConvert2} from the optimal solution of \eqref{C5:P3Continuous}.

\chapter{Proof of Theories of Chapter 7}\label{C7:appendix_2d}

\section{Proof of Theorem \ref{C7:Theorem1}}\label{C7:AppendixD}
The equation in \eqref{C7:RootInterSect} can be rewritten as
\begin{align}\label{C7:RootInterSect2}
0.891 \frac{{2\pi  - {\psi^{\mathrm{DC}}_{\mathrm{\Theta}}}}}{{{\Theta^2}}}  =
\left[{\psi _{\mathrm{H}}} - 0.891\sqrt{\frac{{5\pi {\psi^2 _{\mathrm{H}}} }}{{6\ln 10}}} {{\mathrm{erf}}\left( {\sqrt { \frac{{6\ln {10}}}{{5 {\psi^2 _{\mathrm{H}}}}}}  + \frac{{{\psi^{\mathrm{DC}}_{\mathrm{\Theta}}}}}{2}} \right)} \right].
\end{align}
By definitions, the main lobe beamwidth cannot be smaller than the half-power beamwidth, i.e., ${\psi^{\mathrm{DC}}_{\mathrm{\Theta}}} \ge {\psi _{\mathrm{H}}}$.
Therefore, we have ${\frac{1}{\psi _{\mathrm{H}}}\sqrt {\frac{{6\ln 10}}{{5}}}  + \frac{{{\psi^{\mathrm{DC}}_{\mathrm{\Theta}}}}}{2}} \ge {\frac{1}{\psi _{\mathrm{H}}}\sqrt {\frac{{6\ln 10}}{{5}}}  + \frac{{\psi _{\mathrm{H}}}}{2}}$.
For an arbitrary half-power beamwidth ${\psi _{\mathrm{H}}} \in \left[0,2\pi\right]$, we can observe that ${\mathrm{erf}}\left( {\frac{1}{\psi _{\mathrm{H}}}\sqrt {\frac{{6\ln 10}}{{5}}}  + \frac{\psi _{\mathrm{H}}}{2}} \right) \approx 1$.
Since the error function ${\mathrm{erf}}\left(\cdot\right)$ should be smaller than 1 and it is an increasing function of the input, we have ${\mathrm{erf}}\left( {\frac{1}{\psi _{\mathrm{H}}}\sqrt {\frac{{6\ln 10}}{{5}}}  + \frac{\psi _{\Theta}}{2}} \right) \approx 1$.
As a result, we have
\begin{equation}\label{C7:RootInterSect3}
0.891\frac{2\pi  - {\psi^{\mathrm{DC}}_{\mathrm{\Theta}}}}{\Theta^2} \approx  \psi_{\mathrm{H}} \left[ 1- 0.891\sqrt{\frac{5\pi}{6\ln 10}} \right],
\end{equation}
and the equation in \eqref{C7:RootInterSect3} can be further rewritten as
\begin{equation}\label{C7:RootInterSect4}
17.844 \left({2\pi  - {\psi^{\mathrm{DC}}_{\mathrm{\Theta}}}}\right) \approx \psi_{\mathrm{-3dB}}{\Theta^2}.
\end{equation}
Note that both ${\psi^{\mathrm{DC}}_{\mathrm{H}}}$ and ${\psi^{\mathrm{DC}}_{\mathrm{\Theta}}}$ are monotonically increasing functions with respect to (w.r.t.) MSLR $\Theta$, since a lower sidelobe power level always leads to a wider beam for DCBF \cite{Dolph1946}.
Therefore, the right hand side (RHS) in \eqref{C7:RootInterSect4} is monotonically increased and the left hand side (LHS) is monotonically decreased with increasing $\Theta$.
With $\Theta \to \infty$, the RHS goes to infinity while the LHS is finite.
On the other hand, when $\Theta = \sqrt{2}$, the sidelobe power level is exactly half of the main lobe response and thus we have ${\psi _{\mathrm{H}}} = {\psi_{\mathrm{\Theta}}} \ll 2\pi$.
Therefore, the LHS is larger than the RHS in \eqref{C7:RootInterSect4} at $\Theta = \sqrt{2}$.
As a result, there exists a unique root for the equation in \eqref{C7:RootInterSect4}.
In other words, there exists only one root of the equation in \eqref{C7:RootInterSect}, which completes the proof.

\section{Proof of Theorem \ref{C7:BeamDirection}}\label{C7:AppendixD2}

Firstly, it is clearly that the main beam direction should be located within the range of angles spanned by the two NOMA users, i.e., $\theta_0 \in \left[\theta_{1,0},\theta_{2,0}\right]$, as steering the beam outside this range degrades both users' effective channel gains.
Since we have assumed that only the two users within the half-power beamwidth can be clustered as a NOMA group, in the large number of antennas regime, we have $\mathop {\lim }\limits_{N_{\mathrm{BS}} \to \infty } \frac{\left| {\widetilde h_1}\left({\theta}_0\right) \right|^2}{\mathcal{W}{N_{0}}} \to \infty$ and $\mathop {\lim }\limits_{N_{\mathrm{BS}} \to \infty } \frac{\left| {\widetilde h_2}\left({\theta}_0\right) \right|^2}{\mathcal{W}{N_{0}}} \to \infty$, $\forall \theta_0 \in \left[\theta_{1,0},\theta_{2,0}\right]$.
As a result, the asymptotic sum-rates for both cases in \eqref{C7:SumRate1} and \eqref{C7:SumRate2} can be rewritten as
\begin{align}
\mathop {\lim }\limits_{N_{\mathrm{BS}} \to \infty } R_{{\mathrm{sum}}} \left(\widetilde{\theta}_0\right) &=\mathcal{W}{\log _2}\left( {\frac{{{p_{\max }}{{\left| {\widetilde h_1} \left(\widetilde{\theta}_0\right) \right|}^2}}}{{\mathcal{W}{N_{0}}}}} \right) \;\text{and}\label{C7:HighSNRRateNOMApROOF1}\\[-1mm]
\mathop {\lim }\limits_{N_{\mathrm{BS}} \to \infty } \overline{R}_{{\mathrm{sum}}} \left(\overline{\theta}_0\right) &=\mathcal{W}{\log _2}\left( {\frac{{{p_{\max }}{{\left| {\widetilde h_2}\left(\overline{\theta}_0\right) \right|}^2}}}{{\mathcal{W}{N_{0}}}}} \right),\label{C7:HighSNRRateNOMApROOF2}
\end{align}
respectively.
For both cases, we can observe that the sum-rate of the two NOMA users only depends on strength of the larger effective channel gain after analog beamforming in the large number of antennas regime.
Therefore, we have
\begin{equation}
\theta_{1,0} = \arg\mathop { \max}\limits_{{\theta}_0}  \mathop {\lim }\limits_{N_{\mathrm{BS}} \to \infty } R_{{\mathrm{sum}}}\left( {\theta}_0  \right) \;\text{and}\;
\theta_{2,0} = \arg\mathop { \max}\limits_{{\theta}_0}  \mathop {\lim }\limits_{N_{\mathrm{BS}} \to \infty } \overline{R}_{{\mathrm{sum}}}\left( {\theta}_0  \right),\label{C7:OptimalMainBeamDirection}
\end{equation}
as $\theta_{1,0}$ and $\theta_{2,0}$ can maximize the effective channel gains ${\left| {\widetilde h_1} \left({\theta}_0\right) \right|^2}$ and ${\left| {\widetilde h_2} \left({\theta}_0\right) \right|^2}$ in \eqref{C7:HighSNRRateNOMApROOF1} and \eqref{C7:HighSNRRateNOMApROOF2}, respectively.
The maximum effective channel gains of user 1 and user 2 in the considered two cases can be obtained by
\begin{equation}
{\left| {\widetilde h_1}  \left({\theta}_{0,1}\right) \right|^2} = {\left|{\alpha _{1,0}}\right|}^2 {{N_{{\mathrm{UE}}}}G} \;\;\text{and}\;\; {\left| {\widetilde h_2}  \left({\theta}_{0,2}\right) \right|^2} =  {\left|{\alpha _{2,0}}\right|}^2 {{N_{{\mathrm{UE}}}}G},
\end{equation}
respectively, where the receiving beamforming gain ${N_{{\mathrm{UE}}}}$ is obtained via ${\bf{v}}_k = {\mathbf{a}}_{\mathrm{UE}} \left(  \phi _{k,0}, {N_{{\mathrm{UE}}}} \right)$, $\forall k = \{1,2\}$.
Here, $G$ denotes the main lobe response of transmitting analog beamforming, which depends on the adopted beamwidth control method and the desired half-power beamwidth.
Moreover, rotating a beam does not change its main lobe response.
Hence, we have the same $G$ for the analog beam steered to user 1 and user 2.
Furthermore, due to ${{\left| {{\alpha _{2,0}}} \right|}} \le {{\left| {{\alpha _{1,0}}} \right|}}$, $\mathop {\max }\limits_{{\theta}_0}  \mathop {\lim }\limits_{N_{\mathrm{BS}} \to \infty } {R}_{{\mathrm{sum}}}\left( {\theta}_0  \right) \ge \mathop {\max }\limits_{{\theta}_0}  \mathop {\lim }\limits_{N_{\mathrm{BS}} \to \infty } \overline{{R}}_{{\mathrm{sum}}}\left( {\theta}_0  \right)$ holds, which completes the proof of the first part of Theorem \ref{C7:BeamDirection}.
Then, combining this result with \eqref{C7:OptimalMainBeamDirection}, we can conclude that steering the generated analog beam to the strong user, i.e., user 1, can maximize the system sum-rate in the large number of antennas regime, which completes the proof of the second part of Theorem \ref{C7:BeamDirection}.

\section{Proof of the Equivalence between \eqref{C7:ResourceAllocation33} and \eqref{C7:ResourceAllocation55}}\label{C7:AppendixD3}
We first prove that the variables $\left(\mathcal{Q}^*,\mathbf{\Upsilon}^*\right)$ can optimize the problem in \eqref{C7:ResourceAllocation33} if and only if the variables $\left(\mathcal{Q}^*,\mathbf{\Upsilon}^*,\mathbf{{A}}^*,\mathbf{{B}}^*,\delta^*\right)$ are able to optimize the problem in \eqref{C7:ResourceAllocation55}.
Also, both problem formulations achieve the same optimal value, i.e., $V_{\mathrm{I}}\left(\mathcal{Q}^*,\mathbf{\Upsilon}^*\right) =V_{\mathrm{II}}\left(\mathcal{Q}^*,\mathbf{\Upsilon}^*,\mathbf{{A}}^*,\mathbf{{B}}^*,\delta^*\right)$.

Given any $\left(\mathcal{Q},\mathbf{\Upsilon}\right)$, we can observe that
\begin{align}
{a_{k,i,r}} & = \frac{\sqrt{\Tr\left( {{\widetilde{{\bf{H}}}_k}{\bf{Q}}_i} \right)}}{{I_{k,r}^{{\mathrm{inter}}}\left({\mathcal{{Q}}}\right) + I_{k,i,r}^{{\mathrm{intra}}}\left({\mathcal{{Q}}}\right) + \mathcal{W}{N_{0}}}},\forall i > k, r, \;\text{and}\label{C7:UpdateSlack1}\\
b_{k,r} &= \frac{\sqrt{\Tr\left( {{\widetilde{{\bf{H}}}_k}{\bf{Q}}_k} \right)}}{{I_{k,r}^{{\mathrm{inter}}}\left({\mathcal{{Q}}}\right) + I_{k,k,r}^{{\mathrm{intra}}}\left({\mathcal{{Q}}}\right) + \mathcal{W} {N_{0}}}},\forall k,r,\label{C7:UpdateSlack2}
\end{align}
can maximize $g_1\left(a_{k,i,r},{\mathcal{{Q}}}\right)$ and $g_2\left(b_{k,r},{\mathcal{{Q}}}\right)$ in $\overline{{\mbox{C3}}}$ and $\overline{{\mbox{C4}}}$ in \eqref{C7:ResourceAllocation55}, respectively, which can maximize the feasible solution set of the problem in \eqref{C7:ResourceAllocation55}.
In other words, ${a_{k,i,r}}$ in \eqref{C7:UpdateSlack1} and $b_{k,r}$ in \eqref{C7:UpdateSlack2} are the optimal solution of \eqref{C7:ResourceAllocation55} given any $\left(\mathcal{Q},\mathbf{\Upsilon}\right)$ and we can observe that they lead to exactly the same feasible solution set as the problem in \eqref{C7:ResourceAllocation33}.
Furthermore, for given $\left(\mathcal{Q},\mathbf{\Upsilon}\right)$, we can observe that
\begin{equation}\label{C7:UpdateSlack3}
{\delta} = \frac{\sqrt{R_{{\mathrm{sum}}}\left(\mathbf{\Upsilon}\right)}}{{U_{\mathrm{P}}}\left(\mathcal{Q}\right)}
\end{equation}
maximizes the objective function in \eqref{C7:ResourceAllocation55}.
Since the feasible solution sets of \eqref{C7:ResourceAllocation33} and \eqref{C7:ResourceAllocation55} are identical, the auxiliary variables' update in  \eqref{C7:UpdateSlack1}, \eqref{C7:UpdateSlack2}, and \eqref{C7:UpdateSlack3} yields the same objective value, i.e.,
\begin{equation}\label{C7:SameObj}
V_{\mathrm{I}}\left(\mathcal{Q},\mathbf{\Upsilon}\right) =V_{\mathrm{II}}\left(\mathcal{Q},\mathbf{\Upsilon},\mathbf{{A}},\mathbf{{B}},\delta\right).
\end{equation}

Now, denoting the optimal solution of the problem in \eqref{C7:ResourceAllocation33} as $\left(\mathcal{Q}^*,\mathbf{\Upsilon}^*\right)$, we can obtain $\mathbf{{A}}^*$, $\mathbf{{B}}^*$, and $\delta^*$ according to \eqref{C7:UpdateSlack1}, \eqref{C7:UpdateSlack2}, and \eqref{C7:UpdateSlack3}, respectively.
Then, we have the following implications:
\begin{equation}\label{C7:InequalityI}
V_{\mathrm{II}}\left(\mathcal{Q}^*,\mathbf{\Upsilon}^*,\mathbf{{A}}^*,\mathbf{{B}}^*,\delta^*\right) \mathop= \limits^{\mathrm{\left( a \right)}}V_{\mathrm{I}}\left(\mathcal{Q}^*,\mathbf{\Upsilon}^*\right)
\mathop  \ge \limits^{\mathrm{\left( b\right)}}
V_{\mathrm{I}}\left(\mathcal{Q},\mathbf{\Upsilon}\right) \mathop  = \limits^{\mathrm{\left( c \right)}}V_{\mathrm{II}}\left(\mathcal{Q},\mathbf{\Upsilon},\mathbf{{A}},\mathbf{{B}},\delta\right),
\end{equation}
where the equalities ${\mathrm{\left( a \right)}}$ and ${\mathrm{\left( c \right)}}$ are obtained from \eqref{C7:SameObj} and the inequality ${\mathrm{\left( b \right)}}$ is obtained due to $\left(\mathcal{Q}^*,\mathbf{\Upsilon}^*\right)$ as the optimal solution of \eqref{C7:ResourceAllocation33}.
As a result, $\left(\mathcal{Q}^*,\mathbf{\Upsilon}^*,\mathbf{{A}}^*,\mathbf{{B}}^*,\delta^*\right)$ is the optimal solution of \eqref{C7:ResourceAllocation55} and both problems in \eqref{C7:ResourceAllocation33} and \eqref{C7:ResourceAllocation55} attain the same optimal value.
Similarly, assuming that $\left(\mathcal{Q}^*,\mathbf{\Upsilon}^*,\mathbf{{A}}^*,\mathbf{{B}}^*,\delta^*\right)$ can optimize the problem in \eqref{C7:ResourceAllocation55}, we have
\begin{equation}\label{C7:InequalityII}
V_{\mathrm{I}}\left(\mathcal{Q}^*,\mathbf{\Upsilon}^*\right)
= V_{\mathrm{II}}\left(\mathcal{Q}^*,\mathbf{\Upsilon}^*,\mathbf{{A}}^*,\mathbf{{B}}^*,\delta^*\right) \ge V_{\mathrm{II}}\left(\mathcal{Q},\mathbf{\Upsilon},\mathbf{{A}},\mathbf{{B}},\delta\right)
= V_{\mathrm{I}}\left(\mathcal{Q},\mathbf{\Upsilon}\right),
\end{equation}
which means that $\left(\mathcal{Q}^*,\mathbf{\Upsilon}^*\right)$ is the optimal solution of \eqref{C7:ResourceAllocation33} and both problems in \eqref{C7:ResourceAllocation33} and \eqref{C7:ResourceAllocation55} attain the same optimal value.
It completes the proof.

\clearpage{\pagestyle{empty}\cleardoublepage}

% references
\renewcommand{\bibname}{Bibliography}
\addcontentsline{toc}{chapter}{\protect\numberline{}{Bibliography}}
%\bibliographystyle{splncs}
%\singlespacing

\bibliographystyle{IEEEtran}
\bibliography{NOMA}

% Generated by IEEEtran.bst, version: 1.14 (2015/08/26)
\begin{thebibliography}{100}
\providecommand{\url}[1]{#1}
\csname url@samestyle\endcsname
\providecommand{\newblock}{\relax}
\providecommand{\bibinfo}[2]{#2}
\providecommand{\BIBentrySTDinterwordspacing}{\spaceskip=0pt\relax}
\providecommand{\BIBentryALTinterwordstretchfactor}{4}
\providecommand{\BIBentryALTinterwordspacing}{\spaceskip=\fontdimen2\font plus
\BIBentryALTinterwordstretchfactor\fontdimen3\font minus
  \fontdimen4\font\relax}
\providecommand{\BIBforeignlanguage}[2]{{%
\expandafter\ifx\csname l@#1\endcsname\relax
\typeout{** WARNING: IEEEtran.bst: No hyphenation pattern has been}%
\typeout{** loaded for the language `#1'. Using the pattern for}%
\typeout{** the default language instead.}%
\else
\language=\csname l@#1\endcsname
\fi
#2}}
\providecommand{\BIBdecl}{\relax}
\BIBdecl

\bibitem{Andrews2014}
J.~Andrews, S.~Buzzi, W.~Choi, S.~Hanly, A.~Lozano, A.~Soong, and J.~Zhang,
  ``What will {5G} be?'' \emph{IEEE J. Select. Areas Commun.}, vol.~32, no.~6,
  pp. 1065--1082, Jun. 2014.

\bibitem{wong2017key}
V.~W. Wong, R.~Schober, D.~W.~K. Ng, and L.-C. Wang, \emph{Key Technologies for
  {5G} Wireless Systems}.\hskip 1em plus 0.5em minus 0.4em\relax Cambridge
  University Press, 2017.

\bibitem{Zorzi2010}
M.~Zorzi, A.~Gluhak, S.~Lange, and A.~Bassi, ``From today's {INTRAnet} of
  things to a future {INTERnet} of things: a wireless- and mobility-related
  view,'' \emph{IEEE Wireless Commun.}, vol.~17, no.~6, pp. 44--51, Dec. 2010.

\bibitem{QualComm}
``The {5G} unified air interface: Scalable to an extreme variation of
  requirements,'' Qualcomm Technologies Inc., Tech. Rep., Nov. 2015.

\bibitem{SunALOHA}
Z.~Sun, Y.~Xie, J.~Yuan, and T.~Yang, ``Coded slotted {ALOHA} for erasure
  channels: Design and throughput analysis,'' \emph{IEEE Trans. Commun.},
  vol.~65, no.~11, Nov. 2017.

\bibitem{PopovskiAccess}
P.~{Popovski}, K.~F. {Trillingsgaard}, O.~{Simeone}, and G.~{Durisi}, ``{5G}
  wireless network slicing for {eMBB}, {URLLC}, and {mMTC}: {A}
  communication-theoretic view,'' \emph{IEEE Access}, vol.~6, pp.
  55\,765--55\,779, Sep. 2018.

\bibitem{SunJUICE}
Z.~Sun, Z.~Wei, L.~Yang, J.~Yuan, X.~Cheng, and L.~Wan, ``Joint user
  identification and channel estimation in massive connectivity with
  transmission control,'' in \emph{Proc. IEEE Intern. Sympos. on Turbo Codes
  Iterative Information Process.}, 2018, pp. 1--5.

\bibitem{SunPNC}
Z.~Sun, L.~Yang, J.~Yuan, and D.~W.~K. Ng, ``Physical-layer network coding
  based decoding scheme for random access,'' \emph{IEEE Trans. Veh. Technol.},
  vol.~68, no.~4, pp. 3550--3564, Apr. 2019.

\bibitem{liu2019deep}
B.~Liu, Z.~Wei, J.~Yuan, and M.~Pajovic, ``Deep learning assisted user
  identification in massive machine-type communications,'' \emph{arXiv preprint
  arXiv:1907.09735}, 2019.

\bibitem{SunCSTCOM}
Z.~{Sun}, Z.~{Wei}, L.~{Yang}, J.~{Yuan}, X.~{Cheng}, and L.~{Wan},
  ``Exploiting transmission control for joint user identification and channel
  estimation in massive connectivity,'' \emph{IEEE Trans. Commun.}, pp. 1--1,
  May early access, 2019.

\bibitem{QingqingEE}
Q.~{Wu}, G.~Y. {Li}, W.~{Chen}, D.~W.~K. {Ng}, and R.~{Schober}, ``An overview
  of sustainable green {5G} networks,'' \emph{IEEE Wireless Commun.}, vol.~24,
  no.~4, pp. 72--80, Aug. 2017.

\bibitem{WuEESE2018}
Q.~{Wu}, W.~{Chen}, D.~W.~K. {Ng}, and R.~{Schober}, ``Spectral and
  energy-efficient wireless powered {IoT} networks: {NOMA} or {TDMA}?''
  \emph{IEEE Trans. Veh. Technol.}, vol.~67, no.~7, pp. 6663--6667, Jul. 2018.

\bibitem{DerrickEEOFDMA}
D.~W.~K. Ng, E.~S. Lo, and R.~Schober, ``Energy-efficient resource allocation
  in {OFDMA} systems with large numbers of base station antennas,'' \emph{IEEE
  Trans. Wireless Commun.}, vol.~11, no.~9, pp. 3292--3304, Sep. 2012.

\bibitem{DerrickEESWIPT}
------, ``Wireless information and power transfer: Energy efficiency
  optimization in {OFDMA} systems,'' \emph{IEEE Trans. Wireless Commun.},
  vol.~12, no.~12, pp. 6352--6370, Dec. 2013.

\bibitem{DerrickEERobust}
------, ``Robust beamforming for secure communication in systems with wireless
  information and power transfer,'' \emph{IEEE Trans. Wireless Commun.},
  vol.~13, no.~8, pp. 4599--4615, Aug. 2014.

\bibitem{DerrickLimitedBackhaul}
------, ``Energy-efficient resource allocation in multi-cell {OFDMA} systems
  with limited backhaul capacity,'' \emph{IEEE Trans. Wireless Commun.},
  vol.~11, no.~10, pp. 3618--3631, Oct. 2012.

\bibitem{NgDUALII}
D.~W.~K. {Ng}, E.~S. {Lo}, and R.~{Schober}, ``Energy-efficient resource
  allocation for secure {OFDMA} systems,'' \emph{IEEE Trans. Veh. Technol.},
  vol.~61, no.~6, pp. 2572--2585, Jul. 2012.

\bibitem{cai2019energy}
Y.~Cai, Z.~Wei, R.~Li, D.~W.~K. Ng, and J.~Yuan, ``Energy-efficient resource
  allocation for secure {UAV} communication systems,'' in \emph{Proc. IEEE
  Wireless Commun. and Networking Conf.}, 2019, pp. 1--8.

\bibitem{ClimateHN}
C.~H. News, ``Tsunami of data could consume one-fifth of global electricity by
  2025,'' Online: \url{https://www.climatechangenews.com}.

\bibitem{Marzetta2010}
T.~Marzetta, ``Noncooperative cellular wireless with unlimited numbers of base
  station antennas,'' \emph{IEEE Trans. Wireless Commun.}, vol.~9, no.~11, pp.
  3590--3600, Nov. 2010.

\bibitem{ngo2013energy}
H.~Q. Ngo, E.~G. Larsson, and T.~L. Marzetta, ``Energy and spectral efficiency
  of very large multiuser {MIMO} systems,'' \emph{IEEE Trans. Commun.},
  vol.~61, no.~4, pp. 1436--1449, Apr. 2013.

\bibitem{BoshkovskaEE}
E.~{Boshkovska}, D.~W.~K. {Ng}, N.~{Zlatanov}, and R.~{Schober}, ``Practical
  non-linear energy harvesting model and resource allocation for {SWIPT}
  systems,'' \emph{IEEE Commun. Lett.}, vol.~19, no.~12, pp. 2082--2085, Dec.
  2015.

\bibitem{AhmedEH2013}
I.~{Ahmed}, A.~{Ikhlef}, D.~W.~K. {Ng}, and R.~{Schober}, ``Power allocation
  for an energy harvesting transmitter with hybrid energy sources,'' \emph{IEEE
  Trans. Wireless Commun.}, vol.~12, no.~12, pp. 6255--6267, Dec. 2013.

\bibitem{Zlatanov2017}
N.~{Zlatanov}, D.~W.~K. {Ng}, and R.~{Schober}, ``Capacity of the two-hop relay
  channel with wireless energy transfer from relay to source and energy
  transmission cost,'' \emph{IEEE Trans. Wireless Commun.}, vol.~16, no.~1, pp.
  647--662, Jan. 2017.

\bibitem{Dai2015}
L.~Dai, B.~Wang, Y.~Yuan, S.~Han, I.~Chih-Lin, and Z.~Wang, ``Non-orthogonal
  multiplelarge-scale underlay access for {5G}: solutions, challenges,
  opportunities, and future research trends,'' \emph{IEEE Commun. Mag.},
  vol.~53, no.~9, pp. 74--81, Sep. 2015.

\bibitem{Ding2015b}
Z.~Ding, Y.~Liu, J.~Choi, Q.~Sun, M.~Elkashlan, C.~L. I, and H.~V. Poor,
  ``Application of non-orthogonal multiple access in {LTE} and {5G} networks,''
  \emph{IEEE Commun. Mag.}, vol.~55, no.~2, pp. 185--191, Feb. 2017.

\bibitem{WeiSurvey2016}
Z.~Wei, Y.~Jinhong, D.~W.~K. Ng, M.~Elkashlan, and Z.~Ding, ``A survey of
  downlink non-orthogonal multiple access for {5G} wireless communication
  networks,'' \emph{ZTE Commun.}, vol.~14, no.~4, pp. 17--25, Oct. 2016.

\bibitem{QiuLOMA}
M.~Qiu, Y.~Huang, S.~Shieh, and J.~Yuan, ``A lattice-partition framework of
  downlink non-orthogonal multiple access without {SIC},'' \emph{IEEE Trans.
  Commun.}, vol.~66, no.~6, pp. 2532--2546, Jun. 2018.

\bibitem{QiuLOMASlowFading}
M.~Qiu, Y.~Huang, J.~Yuan, and C.~Wang, ``Lattice-partition-based downlink
  non-orthogonal multiple access without {SIC} for slow fading channels,''
  \emph{IEEE Trans. Commun.}, vol.~67, no.~2, pp. 1166--1181, Feb. 2019.

\bibitem{Rappaport2013}
T.~Rappaport, S.~Sun, R.~Mayzus, H.~Zhao, Y.~Azar, K.~Wang, G.~Wong, J.~Schulz,
  M.~Samimi, and F.~Gutierrez, ``Millimeter wave mobile communications for {5G}
  cellular: {It} will work!'' \emph{IEEE Access}, vol.~1, pp. 335--349, May
  2013.

\bibitem{XiaoMing2017}
M.~Xiao, S.~Mumtaz, Y.~Huang, L.~Dai, Y.~Li, M.~Matthaiou, G.~K. Karagiannidis,
  E.~Bj{\"{o}}rnson, K.~Yang, C.~L. I, and A.~Ghosh, ``Millimeter wave
  communications for future mobile networks,'' \emph{IEEE J. Select. Areas
  Commun.}, vol.~35, no.~9, pp. 1909--1935, Sep. 2017.

\bibitem{zhao2017multiuser}
L.~Zhao, D.~W.~K. Ng, and J.~Yuan, ``Multi-user precoding and channel
  estimation for hybrid millimeter wave systems,'' \emph{IEEE J. Select. Areas
  Commun.}, vol.~35, no.~7, pp. 1576--1590, Jul. 2017.

\bibitem{wei2018multibeam}
Z.~Wei, L.~Zhao, J.~Guo, D.~W.~K. Ng, and J.~Yuan, ``Multi-beam {NOMA} for
  hybrid mmwave systems,'' \emph{IEEE Trans. Commun.}, vol.~67, no.~2, pp.
  1705--1719, Feb. 2019.

\bibitem{WEIMultiBeamFramework}
Z.~{Wei}, L.~{Zhao}, J.~{Guo}, D.~W.~K. {Ng}, and J.~{Yuan}, ``A multi-beam
  {NOMA} framework for hybrid {mmWave} systems,'' in \emph{Proc. IEEE Intern.
  Commun. Conf.}, May 2018, pp. 1--7.

\bibitem{WeiBeamWidthControl}
Z.~{Wei}, D.~W. {Kwan Ng}, and J.~{Yuan}, ``{NOMA} for hybrid mmwave
  communication systems with beamwidth control,'' \emph{IEEE J. Select. Topics
  Signal Process.}, vol.~13, no.~3, pp. 567--583, Jun. 2019.

\bibitem{wei2019beamwidthConference}
Z.~Wei, D.~W.~K. Ng, and J.~Yuan, ``Beamwidth control for {NOMA} in hybrid
  {mmWave} communication systems,'' \emph{arXiv preprint arXiv:1902.04227},
  2019.

\bibitem{Andrews2012}
J.~Andrews, H.~Claussen, M.~Dohler, S.~Rangan, and M.~Reed, ``{Femtocells:
  Past, Present, and Future},'' \emph{IEEE J. Select. Areas Commun.}, vol.~30,
  no.~3, pp. 497--508, Apr. 2012.

\bibitem{Andrews2013}
J.~Andrews, ``Seven ways that {HetNets} are a cellular paradigm shift,''
  \emph{IEEE Commun. Mag.}, vol.~51, no.~3, pp. 136--144, Mar. 2013.

\bibitem{Ramasamy2013}
D.~Ramasamy, R.~Ganti, and U.~Madhow, ``On the capacity of picocellular
  networks,'' in \emph{Proc. IEEE Intern. Sympos. on Inf. Theory}, Jul. 2013,
  pp. 241--245.

\bibitem{Timotheou2015}
S.~Timotheou and I.~Krikidis, ``Fairness for non-orthogonal multiple access in
  {5G} systems,'' \emph{IEEE Signal Process. Lett.}, vol.~22, no.~10, pp.
  1647--1651, Oct. 2015.

\bibitem{Cover1991}
T.~M. Cover and J.~A. Thomas, \emph{Elements of Information Theory}.\hskip 1em
  plus 0.5em minus 0.4em\relax John Wiley \& Sons, Inc., 1991.

\bibitem{Tse2005}
D.~Tse and P.~Viswanath, \emph{Fundamentals of wireless communication}.\hskip
  1em plus 0.5em minus 0.4em\relax Cambridge university press, 2005.

\bibitem{WangPowerEfficiency}
P.~Wang, J.~Xiao, and P.~Li, ``Comparison of orthogonal and non-orthogonal
  approaches to future wireless cellular systems,'' \emph{IEEE Veh. Technol.
  Mag.}, vol.~1, no.~3, pp. 4--11, Sep. 2006.

\bibitem{Ding2014}
Z.~Ding, Z.~Yang, P.~Fan, and H.~Poor, ``On the performance of non-orthogonal
  multiple access in {5G} systems with randomly deployed users,'' \emph{IEEE
  Signal Process. Lett.}, vol.~21, no.~12, pp. 1501--1505, Dec. 2014.

\bibitem{Wei2017}
Z.~Wei, D.~W.~K. Ng, J.~Yuan, and H.~M. Wang, ``Optimal resource allocation for
  power-efficient {MC-NOMA} with imperfect channel state information,''
  \emph{IEEE Trans. Commun.}, vol.~65, no.~9, pp. 3944--3961, Sep. 2017.

\bibitem{Dingtobepublished}
Z.~Ding, P.~Fan, and H.~V. Poor, ``Impact of user pairing on {5G} nonorthogonal
  multiple-access downlink transmissions,'' \emph{IEEE Trans. Veh. Technol.},
  vol.~65, no.~8, pp. 6010--6023, Aug. 2016.

\bibitem{Xu2017}
C.~Xu, Y.~Hu, C.~Liang, J.~Ma, and P.~Li, ``Massive {MIMO}, non-orthogonal
  multiple access and interleave division multiple access,'' \emph{IEEE
  Access}, vol.~5, pp. 14\,728--14\,748, Jul. 2017.

\bibitem{Wei2018PerformanceGain}
Z.~Wei, L.~Yang, D.~W.~K. Ng, and J.~Yuan, ``On the performance gain of {NOMA}
  over {OMA} in uplink single-cell systems,'' in \emph{Proc. IEEE Global
  Commun. Conf.}, 2018, pp. 1--7.

\bibitem{Lei2016NOMA}
L.~Lei, D.~Yuan, C.~K. Ho, and S.~Sun, ``Power and channel allocation for
  non-orthogonal multiple access in {5G} systems: Tractability and
  computation,'' \emph{IEEE Trans. Wireless Commun.}, vol.~15, no.~12, pp.
  8580--8594, Dec. 2016.

\bibitem{Di2016sub}
B.~Di, L.~Song, and Y.~Li, ``Sub-channel assignment, power allocation, and user
  scheduling for non-orthogonal multiple access networks,'' \emph{IEEE Trans.
  Wireless Commun.}, vol.~15, no.~11, pp. 7686--7698, Nov. 2016.

\bibitem{Sun2016Fullduplex}
Y.~Sun, D.~W.~K. Ng, Z.~Ding, and R.~Schober, ``Optimal joint power and
  subcarrier allocation for full-duplex multicarrier non-orthogonal multiple
  access systems,'' \emph{IEEE Trans. Commun.}, vol.~65, no.~3, pp. 1077--1091,
  Mar. 2017.

\bibitem{Wei2016NOMA}
Z.~Wei, D.~W.~K. Ng, and J.~Yuan, ``Power-efficient resource allocation for
  {MC-NOMA} with statistical channel state information,'' in \emph{Proc. IEEE
  Global Commun. Conf.}, Dec. 2016, pp. 1--7.

\bibitem{Nara2005Error}
R.~Narasimhan, ``Error propagation analysis of {V-BLAST} with
  channel-estimation errors,'' \emph{IEEE Trans. Commun.}, vol.~53, no.~1, pp.
  27--31, Jan. 2005.

\bibitem{verdu1998multiuser}
S.~Verdu \emph{et~al.}, \emph{Multiuser detection}.\hskip 1em plus 0.5em minus
  0.4em\relax Cambridge university press, 1998.

\bibitem{IDMAPingLI}
P.~Li, ``Interleave-division multiple access and chip-by-chip iterative
  multi-user detection,'' \emph{IEEE Commun. Mag.}, vol.~43, no.~6, pp. 19--23,
  Jun. 2005.

\bibitem{ChenNOMAScheme}
Y.~Chen, A.~Bayesteh, Y.~Wu, B.~Ren, S.~Kang, S.~Sun, Q.~Xiong, C.~Qian, B.~Yu,
  Z.~Ding, S.~Wang, S.~Han, X.~Hou, H.~Lin, R.~Visoz, and R.~Razavi, ``Toward
  the standardization of non-orthogonal multiple access for next generation
  wireless networks,'' \emph{IEEE Commun. Mag.}, vol.~56, no.~3, pp. 19--27,
  Mar. 2018.

\bibitem{Verdu1999}
S.~Verdu and S.~Shamai, ``Spectral efficiency of {CDMA} with random
  spreading,'' \emph{IEEE Trans. Inf. Theory}, vol.~45, no.~2, pp. 622--640,
  Mar. 1999.

\bibitem{YangCDMA}
L.-L. Yang and L.~Hanzo, ``Multicarrier {DS-CDMA}: A multiple access scheme for
  ubiquitous broadband wireless communications,'' \emph{IEEE Commun. Mag.}, pp.
  116--124, Oct. 2003.

\bibitem{Hanzo2003}
L.~Hanzo, L.-L. Yang, E.~Kuan, and K.~Yen, \emph{Single- and Multi-Carrier
  {DS-CDMA}: Multi-USer Detection, Space-Time Spreading, Synchronisation,
  Standards and Networking}.\hskip 1em plus 0.5em minus 0.4em\relax John Wiley
  \& Sons, Aug. 2003.

\bibitem{Ping2006IDMA}
P.~Li, L.~Liu, K.~Wu, and W.~K. Leung, ``Interleave division multiple-access,''
  \emph{IEEE Trans. Wireless Commun.}, vol.~5, no.~4, pp. 938--947, Apr. 2006.

\bibitem{Access2015}
``Study on downlink multiuser supersition transmission ({MUST}) for {LTE}
  ({R}elease 13),'' 3GPP TR 36.859, Tech. Rep., Dec. 2015.

\bibitem{NOMADOCOMO}
\BIBentryALTinterwordspacing
DOCOMO and MediaTek. (2017) {DOCOMO} and {MediaTek} achieve world's first
  successful {5G} trial using smartphone-sized {NOMA} chipset-embedded device
  to increase spectral efficiency. [Online]. Available:
  \url{https://www.nttdocomo.co.jp/english/info/media_center/pr/2017/1102_02.html}
\BIBentrySTDinterwordspacing

\bibitem{Hoshyar2008}
R.~Hoshyar, F.~P. Wathan, and R.~Tafazolli, ``Novel low-density signature for
  synchronous {CDMA} systems over {AWGN} channel,'' \emph{IEEE Trans. Signal
  Process.}, vol.~56, no.~4, pp. 1616--1626, Apr. 2008.

\bibitem{Hoshyar2010}
R.~Hoshyar, R.~Razavi, and M.~Al-Imari, ``{LDS-OFDM} an efficient multiple
  access technique,'' in \emph{Proc. IEEE Veh. Techn. Conf.}, May 2010, pp.
  1--5.

\bibitem{Razavi2012}
R.~Razavi, M.~Al-Imari, M.~A. Imran, R.~Hoshyar, and D.~Chen, ``On receiver
  design for uplink low density signature {OFDM (LDS-OFDM)},'' \emph{IEEE
  Trans. Commun.}, vol.~60, no.~11, pp. 3499--3508, Nov. 2012.

\bibitem{HuangTaoLDS}
T.~{Huang}, J.~{Yuan}, X.~{Cheng}, and W.~{Lei}, ``Design of degrees of
  distribution of {LDS-OFDM},'' in \emph{Proc. IEEE Intern. Conf. on Signal
  Process. and Commun. Syst.}, Dec. 2015, pp. 1--6.

\bibitem{Nikopour2013}
H.~Nikopour and H.~Baligh, ``Sparse code multiple access,'' in \emph{Proc. IEEE
  Personal, Indoor and Mobile Radio Commun. Sympos.}, Sep. 2013, pp. 332--336.

\bibitem{Dai2014PDMA}
X.~Dai, S.~Chen, S.~Sun, S.~Kang, Y.~Wang, Z.~Shen, and J.~Xu, ``Successive
  interference cancelation amenable multiple access {(SAMA)} for future
  wireless communications,'' in \emph{Proc. IEEE Intern. Commun. Conf.}, Nov.
  2014, pp. 222--226.

\bibitem{Hanif2016}
M.~F. Hanif, Z.~Ding, T.~Ratnarajah, and G.~K. Karagiannidis, ``A
  minorization-maximization method for optimizing sum rate in the downlink of
  non-orthogonal multiple access systems,'' \emph{IEEE Trans. Signal Process.},
  vol.~64, no.~1, pp. 76--88, Jan. 2016.

\bibitem{sun2016optimal}
Y.~Sun, D.~W.~K. Ng, Z.~Ding, and R.~Schober, ``Optimal joint power and
  subcarrier allocation for {MC-NOMA} systems,'' in \emph{Proc. IEEE Global
  Commun. Conf.}, Dec. 2016, pp. 1--6.

\bibitem{Sun2015a}
Q.~Sun, S.~Han, C.-L. I, and Z.~Pan, ``Energy efficiency optimization for
  fading {MIMO} non-orthogonal multiple access systems,'' in \emph{Proc. IEEE
  Intern. Commun. Conf.}, Jun. 2015, pp. 2668--2673.

\bibitem{zhang2016energy}
Y.~Zhang, H.~M. Wang, T.~X. Zheng, and Q.~Yang, ``Energy-efficient transmission
  design in non-orthogonal multiple access,'' \emph{IEEE Trans. Veh. Technol.},
  vol.~66, no.~3, pp. 2852--2857, Mar. 2016.

\bibitem{Zhang2018NOMA}
H.~Zhang, F.~Fang, J.~Cheng, K.~Long, W.~Wang, and V.~C.~M. Leung,
  ``Energy-efficient resource allocation in {NOMA} heterogeneous networks,''
  \emph{IEEE Wireless Commun.}, vol.~25, no.~2, pp. 48--53, Apr. 2018.

\bibitem{Fuhui2018NOMA}
F.~Zhou, Y.~Wu, R.~Q. Hu, Y.~Wang, and K.~K. Wong, ``Energy-efficient {NOMA}
  enabled heterogeneous cloud radio access networks,'' \emph{IEEE Network},
  vol.~32, no.~2, pp. 152--160, Mar. 2018.

\bibitem{LiuFairnessNOMA}
Y.~Liu, M.~Elkashlan, Z.~Ding, and G.~K. Karagiannidis, ``Fairness of user
  clustering in {MIMO} non-orthogonal multiple access systems,'' \emph{IEEE
  Commun. Lett.}, vol.~20, no.~7, pp. 1465--1468, Jul. 2016.

\bibitem{wei2017fairness}
Z.~Wei, J.~Guo, D.~W.~K. Ng, and J.~Yuan, ``Fairness comparison of uplink
  {NOMA} and {OMA},'' in \emph{Proc. IEEE Veh. Techn. Conf.}, 2017, pp. 1--6.

\bibitem{ProceedingLiu}
Y.~Liu, Z.~Qin, M.~Elkashlan, Z.~Ding, A.~Nallanathan, and L.~Hanzo,
  ``Nonorthogonal multiple access for {5G} and beyond,'' \emph{Proceedings of
  the IEEE}, vol. 105, no.~12, pp. 2347--2381, Dec. 2017.

\bibitem{DaiCST2018}
L.~Dai, B.~Wang, Z.~Ding, Z.~Wang, S.~Chen, and L.~Hanzo, ``A survey of
  non-orthogonal multiple access for {5G},'' \emph{IEEE Commun. Surveys Tuts.
  Mag.}, vol.~20, no.~3, pp. 2294--2323, May 2018.

\bibitem{Otao2012}
N.~Otao, Y.~Kishiyama, and K.~Higuchi, ``Performance of non-orthogonal access
  with {SIC} in cellular downlink using proportional fair-based resource
  allocation,'' in \emph{Proc. IEEE Intern. Sympos. on Wireless Commun.
  Systems}, Aug. 2012, pp. 476--480.

\bibitem{Saito2013}
Y.~Saito, A.~Benjebbour, Y.~Kishiyama, and T.~Nakamura, ``System-level
  performance evaluation of downlink non-orthogonal multiple access {(NOMA)},''
  in \emph{Proc. IEEE Personal, Indoor and Mobile Radio Commun. Sympos.}, Sep.
  2013, pp. 611--615.

\bibitem{Saito2013a}
Y.~Saito, Y.~Kishiyama, A.~Benjebbour, T.~Nakamura, A.~Li, and K.~Higuchi,
  ``Non-orthogonal multiple access {(NOMA)} for cellular future radio access,''
  in \emph{Proc. IEEE Veh. Techn. Conf.}, Jun. 2013, pp. 1--5.

\bibitem{Xu2015}
P.~Xu, Z.~Ding, X.~Dai, and H.~V. Poor, ``A new evaluation criterion for
  non-orthogonal multiple access in {5G} software defined networks,''
  \emph{IEEE Access}, vol.~3, pp. 1633--1639, Sep. 2015.

\bibitem{Chen2017}
Z.~Chen, Z.~Ding, X.~Dai, and R.~Zhang, ``An optimization perspective of the
  superiority of {NOMA} compared to conventional {OMA},'' \emph{IEEE Trans.
  Signal Process.}, vol.~65, no.~19, pp. 5191--5202, Oct. 2017.

\bibitem{Yang2016}
Z.~Yang, Z.~Ding, P.~Fan, and G.~K. Karagiannidis, ``On the performance of
  non-orthogonal multiple access systems with partial channel information,''
  \emph{IEEE Trans. Commun.}, vol.~64, no.~2, pp. 654--667, Feb. 2016.

\bibitem{DingSignalAlignment}
Z.~Ding, R.~Schober, and H.~V. Poor, ``A general {MIMO} framework for {NOMA}
  downlink and uplink transmission based on signal alignment,'' \emph{IEEE
  Trans. Wireless Commun.}, vol.~15, no.~6, pp. 4438--4454, Jun. 2016.

\bibitem{ChenQuasiDegradation}
Z.~Chen, Z.~Ding, X.~Dai, and G.~K. Karagiannidis, ``On the application of
  quasi-degradation to {MISO-NOMA} downlink,'' \emph{IEEE Trans. Signal
  Process.}, vol.~64, no.~23, pp. 6174--6189, Dec. 2016.

\bibitem{ZhangDiMassive}
D.~Zhang, K.~Yu, Z.~Wen, and T.~Sato, ``Outage probability analysis of {NOMA}
  within massive {MIMO} systems,'' in \emph{Proc. IEEE Veh. Techn. Conf.}, May
  2016, pp. 1--5.

\bibitem{DingMassive}
Z.~Ding and H.~V. Poor, ``Design of {massive MIMO-NOMA} with limited
  feedback,'' \emph{IEEE Signal Process. Lett.}, vol.~23, no.~5, pp. 629--633,
  May 2016.

\bibitem{ZhangDimmWave}
D.~Zhang, Z.~Zhou, C.~Xu, Y.~Zhang, J.~Rodriguez, and T.~Sato, ``Capacity
  analysis of {NOMA} with {mmWave} massive {MIMO} systems,'' \emph{IEEE J.
  Select. Areas Commun.}, vol.~35, no.~7, pp. 1606--1618, Jul. 2017.

\bibitem{WeiTCOM2017}
Z.~Wei, D.~W.~K. Ng, J.~Yuan, and H.~M. Wang, ``Optimal resource allocation for
  power-efficient {MC-NOMA} with imperfect channel state information,''
  \emph{IEEE Trans. Commun.}, vol.~65, no.~9, pp. 3944--3961, May 2017.

\bibitem{ShinMulticellNOMA}
W.~Shin, M.~Vaezi, B.~Lee, D.~J. Love, J.~Lee, and H.~V. Poor, ``Non-orthogonal
  multiple access in multi-cell networks: Theory, performance, and practical
  challenges,'' \emph{IEEE Commun. Mag.}, vol.~55, no.~10, pp. 176--183, Oct.
  2017.

\bibitem{YouMulticellNOMA}
L.~You, D.~Yuan, L.~Lei, S.~Sun, S.~Chatzinotas, and B.~Ottersten, ``Resource
  optimization with load coupling in multi-cell {NOMA},'' \emph{IEEE Trans.
  Wireless Commun.}, vol.~17, no.~7, pp. 4735--4749, Jul. 2018.

\bibitem{FuMulticellNOMA}
Y.~Fu, Y.~Chen, and C.~W. Sung, ``Distributed power control for the downlink of
  multi-cell {NOMA} systems,'' \emph{IEEE Trans. Wireless Commun.}, vol.~16,
  no.~9, pp. 6207--6220, Sep. 2017.

\bibitem{NguyenMulticellNOMA}
V.~Nguyen, H.~D. Tuan, T.~Q. Duong, H.~V. Poor, and O.~Shin, ``Precoder design
  for signal superposition in {MIMO-NOMA} multicell networks,'' \emph{IEEE J.
  Select. Areas Commun.}, vol.~35, no.~12, pp. 2681--2695, Dec. 2017.

\bibitem{Hojeij2015}
M.~R. Hojeij, J.~Farah, C.~Nour, and C.~Douillard, ``Resource allocation in
  downlink non-orthogonal multiple access {(NOMA)} for future radio access,''
  in \emph{Proc. IEEE Veh. Techn. Conf.}, May 2015, pp. 1--6.

\bibitem{Hojeij2015a}
------, ``New optimal and suboptimal resource allocation techniques for
  downlink non-orthogonal multiple access,'' \emph{Wireless Personal Commun.},
  vol.~87, no.~3, pp. 837--867, May 2015.

\bibitem{Di2015}
B.~Di, S.~Bayat, L.~Song, , and Y.~Li, ``Radio resource allocation for downlink
  non-orthogonal multiple access {(NOMA)} networks using matching theory,'' in
  \emph{Proc. IEEE Global Commun. Conf.}, Dec. 2015, pp. 1--6.

\bibitem{Lei2015}
L.~Lei, D.~Yuan, C.~K. Ho, and S.~Sun, ``Joint optimization of power and
  channel allocation with non-orthogonal multiple access for {5G} cellular
  systems,'' in \emph{Proc. IEEE Global Commun. Conf.}, Dec. 2015, pp. 1--6.

\bibitem{Liu2014}
Y.-F. Liu and Y.-H. Dai, ``On the complexity of joint subcarrier and power
  allocation for multi-user {OFDMA} systems,'' \emph{IEEE Trans. Signal
  Process.}, vol.~62, no.~3, pp. 583--596, Feb. 2014.

\bibitem{Kim2004PF}
H.~Kim, K.~Kim, Y.~Han, and S.~Yun, ``A proportional fair scheduling for
  multicarrier transmission systems,'' in \emph{Proc. IEEE Veh. Techn. Conf.},
  vol.~1, Sep. 2004, pp. 409--413.

\bibitem{Wengerter2005PF}
C.~Wengerter, J.~Ohlhorst, and A.~G.~E. von Elbwart, ``Fairness and throughput
  analysis for generalized proportional fair frequency scheduling in {OFDMA},''
  in \emph{Proc. IEEE Veh. Techn. Conf.}, vol.~3, May 2005, pp. 1903--1907.

\bibitem{Liu2015b}
F.~Liu, P.~Mahonen, and M.~Petrova, ``Proportional fairness-based user pairing
  and power allocation for non-orthogonal multiple access,'' in \emph{Proc.
  IEEE Personal, Indoor and Mobile Radio Commun. Sympos.}, Aug. 2015, pp.
  1127--1131.

\bibitem{Shi2015}
S.~{Shi}, L.~{Yang}, and H.~{Zhu}, ``Outage balancing in downlink nonorthogonal
  multiple access with statistical channel state information,'' \emph{IEEE
  Trans. Wireless Commun.}, vol.~15, no.~7, pp. 4718--4731, Jul. 2016.

\bibitem{DerrickFD2016}
D.~W.~K. Ng, Y.~Wu, and R.~Schober, ``Power efficient resource allocation for
  full-duplex radio distributed antenna networks,'' \emph{IEEE Trans. Wireless
  Commun.}, vol.~15, no.~4, pp. 2896--2911, Apr. 2016.

\bibitem{FangEnergyEfficientNOMAJournal}
F.~Fang, H.~Zhang, J.~Cheng, and V.~C.~M. Leung, ``Energy-efficient resource
  allocation for downlink non-orthogonal multiple access network,'' \emph{IEEE
  Trans. Commun.}, vol.~64, no.~9, pp. 3722--3732, Sep. 2016.

\bibitem{Lei2016NPM}
L.~Lei, D.~Yuan, and P.~V{\"{a}}rbrand, ``On power minimization for
  non-orthogonal multiple access {(NOMA)},'' \emph{IEEE Commun. Lett.},
  vol.~20, no.~12, pp. 2458--2461, Dec. 2016.

\bibitem{Kim2013}
B.~Kim, S.~Lim, H.~Kim, S.~Suh, J.~Kwun, S.~Choi, C.~Lee, S.~Lee, and D.~Hong,
  ``Non-orthogonal multiple access in a downlink multiuser beamforming
  system,'' in \emph{Proc. IEEE Mil. Commun. Conf.}, Nov. 2013, pp. 1278--1283.

\bibitem{Choi2015}
J.~Choi, ``Minimum power multicast beamforming with superposition coding for
  multiresolution broadcast and application to {NOMA} systems,'' \emph{IEEE
  Trans. Commun.}, vol.~63, no.~3, pp. 791--800, Mar. 2015.

\bibitem{Sun2015}
Q.~Sun, S.~Han, C.-L. I, and Z.~Pan, ``On the ergodic capacity of {MIMO NOMA}
  systems,'' \emph{IEEE Wireless Commun. Lett.}, vol.~4, no.~4, pp. 405--408,
  Aug. 2015.

\bibitem{Sun2015b}
Q.~Sun, S.~Han, Z.~Xu, S.~Wang, I.~Chih-Lin, and Z.~Pan, ``Sum rate
  optimization for {MIMO} non-orthogonal multiple access systems,'' in
  \emph{Proc. IEEE Wireless Commun. and Networking Conf.}, Mar. 2015, pp.
  747--752.

\bibitem{Dingtobepublisheda}
Z.~Ding, F.~Adachi, and H.~V. Poor, ``The application of {MIMO} to
  non-orthogonal multiple access,'' \emph{IEEE Trans. Wireless Commun.},
  vol.~15, no.~1, pp. 537--552, Jan. 2016.

\bibitem{SunMISONOMA}
Y.~{Sun}, D.~W.~K. {Ng}, J.~{Zhu}, and R.~{Schober}, ``Robust and secure
  resource allocation for full-duplex {MISO} multicarrier {NOMA} systems,''
  \emph{IEEE Trans. Commun.}, vol.~66, no.~9, pp. 4119--4137, Sep. 2018.

\bibitem{ZhaommWaveMulticell}
L.~{Zhao}, Z.~{Wei}, D.~W.~K. {Ng}, J.~{Yuan}, and M.~C. {Reed}, ``Multi-cell
  hybrid millimeter wave systems: Pilot contamination and interference
  mitigation,'' \emph{IEEE Trans. Commun.}, vol.~66, no.~11, pp. 5740--5755,
  Nov. 2018.

\bibitem{ZhaoMulticellMMwAVE}
------, ``Mitigating pilot contamination in multi-cell hybrid millimeter wave
  systems,'' in \emph{Proc. IEEE Intern. Commun. Conf.}, May 2018, pp. 1--7.

\bibitem{Mao2018}
Y.~Mao, B.~Clerckx, and V.~O. Li, ``Rate-splitting multiple access for downlink
  communication systems: bridging, generalizing, and outperforming {SDMA} and
  {NOMA},'' \emph{EURASIP Journal on Wireless Communications and Networking},
  vol. 2018, no.~1, pp. 133--187, May 2018.

\bibitem{zhao2019distributed}
L.~Zhao, J.~Guo, Z.~Wei, D.~W.~K. Ng, and J.~Yuan, ``A distributed {multi-RF}
  chain hybrid mmwave scheme for small-cell systems,'' \emph{arXiv preprint
  arXiv:1902.08354}, 2019.

\bibitem{GaoSubarray}
X.~Gao, L.~Dai, S.~Han, C.~L. I, and R.~W. Heath, ``Energy-efficient hybrid
  analog and digital precoding for {MmWave} {MIMO} systems with large antenna
  arrays,'' \emph{IEEE J. Select. Areas Commun.}, vol.~34, no.~4, pp.
  998--1009, Apr. 2016.

\bibitem{lin2016energy}
C.~Lin and G.~Y. Li, ``Energy-efficient design of indoor mmwave and {sub-THz}
  systems with antenna arrays,'' \emph{IEEE Trans. Wireless Commun.}, vol.~15,
  no.~7, pp. 4660--4672, Jul. 2016.

\bibitem{Dai2017}
M.~Dai and B.~Clerckx, ``Multiuser millimeter wave beamforming strategies with
  quantized and statistical {CSIT},'' \emph{IEEE Trans. Wireless Commun.},
  vol.~16, no.~11, pp. 7025--7038, Nov. 2017.

\bibitem{van2002optimum}
H.~L. Van~Trees, \emph{Optimum array processing: Part {IV} of detection,
  estimation and modulation theory}.\hskip 1em plus 0.5em minus 0.4em\relax
  Wiley Online Library, 2002, vol.~1.

\bibitem{DingFRAB}
Z.~Ding, L.~Dai, R.~Schober, and H.~V. Poor, ``{NOMA} meets finite resolution
  analog beamforming in massive {MIMO} and millimeter-wave networks,''
  \emph{IEEE Commun. Lett.}, vol.~21, no.~8, pp. 1879--1882, Aug. 2017.

\bibitem{Ding2017RandomBeamforming}
Z.~Ding, P.~Fan, and H.~V. Poor, ``Random beamforming in millimeter-wave {NOMA}
  networks,'' \emph{IEEE Access}, vol.~5, pp. 7667--7681, Feb. 2017.

\bibitem{Cui2017Optimal}
J.~Cui, Y.~Liu, Z.~Ding, P.~Fan, and A.~Nallanathan, ``Optimal user scheduling
  and power allocation for millimeter wave {NOMA} systems,'' \emph{IEEE Trans.
  Wireless Commun.}, vol.~17, no.~3, pp. 1502--1517, Mar. 2018.

\bibitem{WangBeamSpace2017}
B.~Wang, L.~Dai, Z.~Wang, N.~Ge, and S.~Zhou, ``Spectrum and energy efficient
  beamspace {MIMO-NOMA} for millimeter-wave communications using lens antenna
  array,'' \emph{IEEE J. Select. Areas Commun.}, vol.~35, no.~10, pp.
  2370--2382, Oct. 2017.

\bibitem{BusariSurveyonMMWave}
S.~A. Busari, K.~M.~S. Huq, S.~Mumtaz, L.~Dai, and J.~Rodriguez,
  ``Millimeter-wave massive {MIMO} communication for future wireless systems: A
  survey,'' \emph{IEEE Commun. Surveys Tuts. Mag.}, vol.~20, no.~2, pp.
  836--869, Dec. 2018.

\bibitem{Shafi2018microwave}
M.~Shafi, J.~Zhang, H.~Tataria, A.~F. Molisch, S.~Sun, T.~S. Rappaport,
  F.~Tufvesson, S.~Wu, and K.~Kitao, ``Microwave vs. millimeter-wave
  propagation channels: Key differences and impact on {5G} cellular systems,''
  \emph{IEEE Commun. Mag.}, vol.~56, no.~12, pp. 14--20, Dec. 2018.

\bibitem{Goldsmith2005wireless}
A.~Goldsmith, \emph{Wireless communications}.\hskip 1em plus 0.5em minus
  0.4em\relax Cambridge university press, 2005.

\bibitem{AkdenizChannelMmWave}
M.~R. Akdeniz, Y.~Liu, M.~K. Samimi, S.~Sun, S.~Rangan, T.~S. Rappaport, and
  E.~Erkip, ``Millimeter wave channel modeling and cellular capacity
  evaluation,'' \emph{IEEE J. Select. Areas Commun.}, vol.~32, no.~6, pp.
  1164--1179, Jun. 2014.

\bibitem{HemadehmmWaveChannelModels}
I.~A. {Hemadeh}, K.~{Satyanarayana}, M.~{El-Hajjar}, and L.~{Hanzo},
  ``Millimeter-wave communications: Physical channel models, design
  considerations, antenna constructions, and link-budget,'' \emph{IEEE Commun.
  Surveys Tuts. Mag.}, vol.~20, no.~2, pp. 870--913, Dec. 2018.

\bibitem{wei2019performanceGainJournal}
Z.~Wei, L.~Yang, D.~W.~K. Ng, J.~Yuan, and L.~Hanzo, ``On the performance gain
  of {NOMA} over {OMA} in uplink communication systems,'' \emph{IEEE Trans.
  Commun.}, minor revision, 25th Aug. 2019.

\bibitem{Access2010}
``Evolved universal terrestrial radio access: Further advancements for {E-UTRA}
  physical layer aspects,'' 3GPP TR 36.814, Tech. Rep., 2010.

\bibitem{XinyuGaoLetter}
X.~Gao, L.~Dai, Z.~Chen, Z.~Wang, and Z.~Zhang, ``Near-optimal beam selection
  for beamspace mmwave massive {MIMO} systems,'' \emph{IEEE Commun. Lett.},
  vol.~20, no.~5, pp. 1054--1057, May 2016.

\bibitem{Brady2013}
J.~{Brady}, N.~{Behdad}, and A.~M. {Sayeed}, ``Beamspace {MIMO} for
  millimeter-wave communications: System architecture, modeling, analysis, and
  measurements,'' \emph{IEEE Trans. Antennas Propag.}, vol.~61, no.~7, pp.
  3814--3827, Jul. 2013.

\bibitem{Andrews2017}
J.~G. Andrews, T.~Bai, M.~N. Kulkarni, A.~Alkhateeb, A.~K. Gupta, and R.~W.
  Heath, ``Modeling and analyzing millimeter wave cellular systems,''
  \emph{IEEE Trans. Commun.}, vol.~65, no.~1, pp. 403--430, Jan. 2017.

\bibitem{CoverBC}
T.~Cover, ``Broadcast channels,'' \emph{IEEE Trans. Inf. Theory}, vol.~18,
  no.~1, pp. 2--14, Jan. 1972.

\bibitem{Al-Imari2014}
M.~Al-Imari, P.~Xiao, M.~A. Imran, and R.~Tafazolli, ``Uplink non-orthogonal
  multiple access for {5G} wireless networks,'' in \emph{Proc. IEEE Intern.
  Sympos. on Wireless Commun. Systems}, Aug. 2014, pp. 781--785.

\bibitem{Al-Imari2015}
M.~Al-Imari, P.~Xiao, and M.~A. Imran, ``Receiver and resource allocation
  optimization for uplink {NOMA} in {5G} wireless networks,'' in \emph{Proc.
  IEEE Intern. Sympos. on Wireless Commun. Systems}, Aug. 2015, pp. 151--155.

\bibitem{Yang2016NOMA}
Z.~Yang, Z.~Ding, P.~Fan, and N.~Al-Dhahir, ``A general power allocation scheme
  to guarantee quality of service in downlink and uplink {NOMA} systems,''
  \emph{IEEE Trans. Wireless Commun.}, vol.~15, no.~11, pp. 7244--7257, Nov.
  2016.

\bibitem{Liu2016}
Y.~{Liu}, Z.~{Ding}, M.~{Elkashlan}, and J.~{Yuan}, ``Nonorthogonal multiple
  access in large-scale underlay cognitive radio networks,'' \emph{IEEE Trans.
  Veh. Technol.}, vol.~65, no.~12, pp. 10\,152--10\,157, Dec. 2016.

\bibitem{Vaezi2018non}
M.~Vaezi, R.~Schober, Z.~Ding, and H.~V. Poor, ``Non-orthogonal multiple
  access: Common myths and critical questions,'' \emph{arXiv preprint
  arXiv:1809.07224}, 2018.

\bibitem{PatelCDMA}
P.~Patel and J.~Holtzman, ``Analysis of a simple successive interference
  cancelation scheme in a {DS/CDMA} system,'' \emph{IEEE J. Select. Areas
  Commun.}, vol.~12, no.~5, pp. 796--807, Jun. 1994.

\bibitem{WolnianskyVBLAST}
P.~W. Wolniansky, G.~J. Foschini, G.~D. Golden, and R.~A. Valenzuela,
  ``{V-BLAST}: an architecture for realizing very high data rates over the
  rich-scattering wireless channel,'' in \emph{URSI Symp. on Signals, Syst, and
  Electronics}, Oct. 1998, pp. 295--300.

\bibitem{Cover:2006:EIT:1146355}
T.~M. Cover and J.~A. Thomas, \emph{Elements of Information Theory}.\hskip 1em
  plus 0.5em minus 0.4em\relax New York, NY, USA: Wiley-Interscience, 2006.

\bibitem{Kwan_AF_2010}
D.~W.~K. Ng and R.~Schober, ``Cross-layer scheduling for {OFDMA}
  amplify-and-forward relay networks,'' \emph{IEEE Trans. Veh. Technol.},
  vol.~59, no.~3, pp. 1443--1458, Mar. 2010.

\bibitem{DerrickOFDMARelay}
D.~W.~K. {Ng} and R.~{Schober}, ``Resource allocation and scheduling in
  multi-cell {OFDMA} systems with decode-and-forward relaying,'' \emph{IEEE
  Trans. Wireless Commun.}, vol.~10, no.~7, pp. 2246--2258, Jul. 2011.

\bibitem{li2018joint}
R.~Li, Z.~Wei, L.~Yang, D.~W.~K. Ng, N.~Yang, J.~Yuan, and J.~An, ``Joint
  trajectory and resource allocation design for {UAV} communication systems,''
  in \emph{Proc. IEEE Global Commun. Conf.}, 2018, pp. 1--7.

\bibitem{Shannon1948mathematical}
C.~E. Shannon, ``A mathematical theory of communication,'' \emph{Bell system
  technical journal}, vol.~27, no.~3, pp. 379--423, 1948.

\bibitem{HasanGreen}
Z.~{Hasan}, H.~{Boostanimehr}, and V.~K. {Bhargava}, ``Green cellular networks:
  {A} survey, some research issues and challenges,'' \emph{IEEE Commun. Surveys
  Tuts. Mag.}, vol.~13, no.~4, pp. 524--540, Nov. 2011.

\bibitem{LengPowerEFFICIENT}
S.~{Leng}, D.~W.~K. {Ng}, and R.~{Schober}, ``Power efficient and secure
  multiuser communication systems with wireless information and power
  transfer,'' in \emph{Proc. IEEE Intern. Commun. Conf.}, Jun. 2014, pp.
  800--806.

\bibitem{DerrickPowerEFFCIEINT2015}
D.~W.~K. {Ng}, Y.~{Sun}, and R.~{Schober}, ``Power efficient and secure
  full-duplex wireless communication systems,'' in \emph{Proc. IEEE Conf. on
  Commun. and Network Security}, Sep. 2015, pp. 1--6.

\bibitem{Sun2016MultipleObj}
Y.~Sun, D.~W.~K. Ng, J.~Zhu, and R.~Schober, ``Multi-objective optimization for
  robust power efficient and secure full-duplex wireless communication
  systems,'' \emph{IEEE Trans. Wireless Commun.}, vol.~15, no.~8, pp.
  5511--5526, Aug. 2016.

\bibitem{Boshkovska2018}
E.~{Boshkovska}, D.~W.~K. {Ng}, L.~{Dai}, and R.~{Schober}, ``Power-efficient
  and secure {WPCNs} with hardware impairments and non-linear {EH} circuit,''
  \emph{IEEE Trans. Commun.}, vol.~66, no.~6, pp. 2642--2657, Jun. 2018.

\bibitem{Boyd2004}
S.~Boyd and L.~Vandenberghe, \emph{Convex optimization}.\hskip 1em plus 0.5em
  minus 0.4em\relax Cambridge university press, 2004.

\bibitem{horst2013global}
R.~Horst and H.~Tuy, \emph{Global optimization: Deterministic
  approaches}.\hskip 1em plus 0.5em minus 0.4em\relax Springer Science \&
  Business Media, 2013.

\bibitem{ZhangStochastic}
X.~Zhang, D.~P. Palomar, and B.~Ottersten, ``Statistically robust design of
  linear {MIMO} transceivers,'' \emph{IEEE Trans. Signal Process.}, vol.~56,
  no.~8, pp. 3678--3689, Aug. 2008.

\bibitem{abramowitz1964handbook}
M.~Abramowitz and I.~A. Stegun, \emph{Handbook of mathematical functions: with
  formulas, graphs, and mathematical tables}.\hskip 1em plus 0.5em minus
  0.4em\relax Courier Corporation, 1964, vol.~55.

\bibitem{Vishwanath2003}
S.~Vishwanath, N.~Jindal, and A.~Goldsmith, ``Duality, achievable rates, and
  sum-rate capacity of gaussian {MIMO} broadcast channels,'' \emph{IEEE Trans.
  Inf. Theory}, vol.~49, no.~10, pp. 2658--2668, Oct. 2003.

\bibitem{WangMUG}
P.~Wang and P.~Li, ``On maximum eigenmode beamforming and multi-user gain,''
  \emph{IEEE Trans. Inf. Theory}, vol.~57, no.~7, pp. 4170--4186, Jul. 2011.

\bibitem{wei2017performance}
Z.~Wei, L.~Dai, D.~W.~K. Ng, and J.~Yuan, ``Performance analysis of a hybrid
  downlink-uplink cooperative {NOMA} scheme,'' in \emph{Proc. IEEE Veh. Techn.
  Conf.}, Jun. 2017, pp. 1--7.

\bibitem{WeiLetter2018}
Z.~Wei, D.~W.~K. Ng, and J.~Yuan, ``Joint pilot and payload power control for
  uplink {MIMO-NOMA} with {MRC-SIC} receivers,'' \emph{IEEE Commun. Lett.},
  vol.~22, no.~4, pp. 692--695, Apr. 2018.

\bibitem{Ngo2013}
H.~Q. Ngo, E.~G. Larsson, and T.~L. Marzetta, ``Energy and spectral efficiency
  of very large multiuser {MIMO} systems,'' \emph{IEEE Trans. Commun.},
  vol.~61, no.~4, pp. 1436--1449, Apr. 2013.

\bibitem{bartle2014elements}
R.~G. Bartle, \emph{The elements of integration and Lebesgue measure}.\hskip
  1em plus 0.5em minus 0.4em\relax John Wiley \& Sons, Inc., 2014.

\bibitem{kedlaya1994proof}
K.~Kedlaya, ``Proof of a mixed arithmetic-mean, geometric-mean inequality,''
  \emph{The American Mathematical Monthly}, vol. 101, no.~4, pp. 355--357,
  1994.

\bibitem{GoldsmithMIMOCapacity2003}
A.~Goldsmith, S.~A. Jafar, N.~Jindal, and S.~Vishwanath, ``Capacity limits of
  {MIMO} channels,'' \emph{IEEE J. Select. Areas Commun.}, vol.~21, no.~5, pp.
  684--702, Jun. 2003.

\bibitem{yang2017noma}
Q.~Yang, H.-M. Wang, D.~W.~K. Ng, and M.~H. Lee, ``{NOMA} in downlink {SDMA}
  with limited feedback: Performance analysis and optimization,'' \emph{IEEE J.
  Select. Areas Commun.}, vol.~35, no.~10, pp. 2281--2294, Oct. 2017.

\bibitem{Heath2011}
R.~W. Heath, T.~Wu, Y.~H. Kwon, and A.~C.~K. Soong, ``Multiuser {MIMO} in
  distributed antenna systems with out-of-cell interference,'' \emph{IEEE
  Trans. Signal Process.}, vol.~59, no.~10, pp. 4885--4899, Oct. 2011.

\bibitem{DMTadeoff}
L.~Zheng and D.~N.~C. Tse, ``Diversity and multiplexing: a fundamental tradeoff
  in multiple-antenna channels,'' \emph{IEEE Trans. Inf. Theory}, vol.~49,
  no.~5, pp. 1073--1096, May 2003.

\bibitem{Zhangtobepublished}
N.~Zhang, J.~Wang, G.~Kang, and Y.~Liu, ``Uplink non-orthogonal multiple access
  in {5G} systems,'' \emph{IEEE Commun. Lett.}, vol.~20, no.~3, pp. 458--461,
  Mar. 2016.

\bibitem{Lee1996CDMA}
C.-C. Lee and R.~Steele, ``Closed-loop power control in {CDMA} systems,''
  \emph{IEE Proc. - Commun.}, vol. 143, no.~4, pp. 231--239, Aug. 1996.

\bibitem{GaoAESINR}
F.~Gao, R.~Zhang, and Y.~C. Liang, ``Optimal channel estimation and training
  design for two-way relay networks,'' \emph{IEEE Trans. Commun.}, vol.~57,
  no.~10, pp. 3024--3033, Oct. 2009.

\bibitem{BigueshMMSE2006}
M.~Biguesh and A.~B. Gershman, ``Training-based {MIMO} channel estimation: a
  study of estimator tradeoffs and optimal training signals,'' \emph{IEEE
  Trans. Signal Process.}, vol.~54, no.~3, pp. 884--893, Mar. 2006.

\bibitem{ChiangGP}
M.~Chiang, C.~W. Tan, D.~P. Palomar, D.~O'neill, and D.~Julian, ``Power control
  by geometric programming,'' \emph{IEEE Trans. Wireless Commun.}, vol.~6,
  no.~7, pp. 2640--2651, Jul. 2007.

\bibitem{cvx}
M.~Grant and S.~Boyd, ``{CVX}: Matlab software for disciplined convex
  programming, version 2.1,'' \url{http://cvxr.com/cvx}, Mar. 2014.

\bibitem{3GPPReportTurboCode2016}
\BIBentryALTinterwordspacing
``Multiplexing and channel coding ({FDD}) (release 14),'' 3GPP TS 25.212, Tech.
  Rep., Dec. 2016. [Online]. Available:
  \url{http://www.3gpp.org/DynaReport/25212.htm}
\BIBentrySTDinterwordspacing

\bibitem{JorswieckNOCSI}
E.~A. Jorswieck and H.~Boche, ``Optimal transmission strategies and impact of
  correlation in multiantenna systems with different types of channel state
  information,'' \emph{IEEE Trans. Signal Process.}, vol.~52, no.~12, pp.
  3440--3453, Dec. 2004.

\bibitem{WangWorstCase}
J.~Wang and D.~P. Palomar, ``Worst-case robust {MIMO} transmission with
  imperfect channel knowledge,'' \emph{IEEE Trans. Signal Process.}, vol.~57,
  no.~8, pp. 3086--3100, Aug. 2009.

\bibitem{Konno2000}
H.~Konno and K.~Fukaishi, ``A branch and bound algorithm for solving low rank
  linear multiplicative and fractional programming problems,'' \emph{Journal of
  Global Optimization}, vol.~18, no.~3, pp. 283--299, Nov. 2000.

\bibitem{MARANASProofBB}
C.~D. Maranas and C.~A. Floudas, ``Global optimization in generalized geometric
  programming,'' \emph{Computers \& Chemical Engineering}, vol.~21, no.~4, pp.
  351--369, Dec. 1997.

\bibitem{Androulakis1995}
I.~P. Androulakis, C.~D. Maranas, and C.~A. Floudas, ``$\alpha${BB}: A global
  optimization method for general constrained nonconvex problems,''
  \emph{Journal of Global Optimization}, vol.~7, no.~4, pp. 337--363, Dec.
  1995.

\bibitem{dinh2010local}
Q.~T. Dinh and M.~Diehl, ``Local convergence of sequential convex programming
  for nonconvex optimization,'' in \emph{Recent Advances in Optimization and
  its Applications in Engineering}.\hskip 1em plus 0.5em minus 0.4em\relax
  Springer, 2010, pp. 93--102.

\bibitem{VucicProofDC}
N.~Vucic, S.~Shi, and M.~Schubert, ``{DC} programming approach for resource
  allocation in wireless networks,'' in \emph{Proc. Int. Symp. Model. Optim.
  Mobile Ad Hoc Wireless Netw.}, May 2010, pp. 380--386.

\bibitem{Wang2015}
B.~Wang, K.~Wang, Z.~Lu, T.~Xie, and J.~Quan, ``Comparison study of
  non-orthogonal multiple access schemes for {5G},'' in \emph{Proc. IEEE Inter.
  Sympo. on Broadband Multimedia Syst. and Broadcast.}, Jun. 2015, pp. 1--5.

\bibitem{LiuSWIPT}
Y.~Liu, Z.~Ding, M.~Elkashlan, and H.~V. Poor, ``Cooperative non-orthogonal
  multiple access with simultaneous wireless information and power transfer,''
  \emph{IEEE J. Select. Areas Commun.}, vol.~34, no.~4, pp. 938--953, Apr.
  2016.

\bibitem{ChenTwoUser}
Z.~Chen, Z.~Ding, P.~Xu, and X.~Dai, ``Optimal precoding for a {QoS}
  optimization problem in two-user {MISO-NOMA} downlink,'' \emph{IEEE Commun.
  Lett.}, vol.~20, no.~6, pp. 1263--1266, Jun. 2016.

\bibitem{Zhu2009}
H.~Zhu and J.~Wang, ``Chunk-based resource allocation in {OFDMA} systems -
  {Part I}: chunk allocation,'' \emph{IEEE Trans. Commun.}, vol.~57, no.~9, pp.
  2734--2744, Sep. 2009.

\bibitem{floudas1995nonlinear}
C.~A. Floudas, \emph{Nonlinear and mixed-integer optimization: fundamentals and
  applications}.\hskip 1em plus 0.5em minus 0.4em\relax Oxford University Press
  on Demand, 1995.

\bibitem{Nemirovski2004IPM}
A.~Nemirovski, ``Interior point polynomial time methods in convex
  programming,'' \emph{Lecture Notes}, 2004.

\bibitem{thomas2001introduction}
T.~H. Cormen, C.~E. Leiserson, R.~L. Rivest, and C.~Stein, \emph{Introduction
  to algorithms}.\hskip 1em plus 0.5em minus 0.4em\relax MIT press Cambridge,
  2001, vol.~6.

\bibitem{Mosek2010}
A.~Mosek, ``The {MOSEK} optimization software,'' \emph{Online at http://www.
  mosek. com}, vol.~54, pp. 1--2, 2010.

\bibitem{sesia2015lte}
S.~Sesia, I.~Toufik, and M.~Baker, \emph{{LTE} - the {UMTS} long term
  evolution}.\hskip 1em plus 0.5em minus 0.4em\relax Wiley Online Library,
  2015.

\bibitem{Yu2006Dual}
W.~Yu and R.~Lui, ``Dual methods for nonconvex spectrum optimization of
  multicarrier systems,'' \emph{IEEE Trans. Commun.}, vol.~54, no.~7, pp.
  1310--1322, Jul. 2006.

\bibitem{SaadCoalitionalGame}
W.~Saad, Z.~Han, M.~Debbah, A.~Hjorungnes, and T.~Basar, ``Coalitional game
  theory for communication networks,'' \emph{IEEE Signal Process. Mag.},
  vol.~26, no.~5, pp. 77--97, Sep. 2009.

\bibitem{SaadCoalitional2012}
W.~Saad, Z.~Han, R.~Zheng, A.~Hjorungnes, T.~Basar, and H.~V. Poor,
  ``Coalitional games in partition form for joint spectrum sensing and access
  in cognitive radio networks,'' \emph{IEEE J. Select. Topics Signal Process.},
  vol.~6, no.~2, pp. 195--209, Apr. 2012.

\bibitem{WangCoalitionNOMA}
K.~Wang, Z.~Ding, and W.~Liang, ``A game theory approach for user grouping in
  hybrid non-orthogonal multiple access systems,'' in \emph{Proc. IEEE Intern.
  Sympos. on Wireless Commun. Systems}, Sep. 2016, pp. 643--647.

\bibitem{SaadCoalitionOrder}
W.~Saad, Z.~Han, M.~Debbah, and A.~Hjorungnes, ``A distributed coalition
  formation framework for fair user cooperation in wireless networks,''
  \emph{IEEE Trans. Wireless Commun.}, vol.~8, no.~9, pp. 4580--4593, Sep.
  2009.

\bibitem{Han2012}
Z.~Han, D.~Niyato, W.~Saad, T.~Baar, and A.~Hjrungnes, \emph{Game Theory in
  Wireless and Communication Networks: Theory, Models, and Applications}.\hskip
  1em plus 0.5em minus 0.4em\relax New York, NY, USA: Cambridge University
  Press, 2012.

\bibitem{BeamTracking}
V.~Va, H.~Vikalo, and R.~W. Heath, ``Beam tracking for mobile millimeter wave
  communication systems,'' in \emph{Proc. IEEE Global Conf. on Signal and Inf.
  Process.}, Dec. 2016, pp. 743--747.

\bibitem{Sohrabi2016}
F.~Sohrabi and W.~Yu, ``Hybrid digital and analog beamforming design for
  large-scale antenna arrays,'' \emph{IEEE J. Select. Areas Commun.}, vol.~10,
  no.~3, pp. 501--513, Apr. 2016.

\bibitem{AlkhateebPrecoder2015}
A.~Alkhateeb, G.~Leus, and R.~W. Heath, ``Limited feedback hybrid precoding for
  multi-user millimeter wave systems,'' \emph{IEEE Trans. Wireless Commun.},
  vol.~14, no.~11, pp. 6481--6494, Nov. 2015.

\bibitem{Sun2017MIMONOMA}
Y.~Sun, D.~W.~K. Ng, and R.~Schober, ``Optimal resource allocation for
  multicarrier {MISO-NOMA} systems,'' in \emph{Proc. IEEE Intern. Commun.
  Conf.}, May 2017, pp. 1--7.

\bibitem{Mumtaz2016mmwave}
S.~Mumtaz, J.~Rodriguez, and L.~Dai, \emph{MmWave Massive {MIMO}: {A} Paradigm
  for {5G}}.\hskip 1em plus 0.5em minus 0.4em\relax Academic Press, 2016.

\bibitem{VishwanathDuality}
S.~Vishwanath, N.~Jindal, and A.~Goldsmith, ``Duality, achievable rates, and
  sum-rate capacity of gaussian {MIMO} broadcast channels,'' \emph{IEEE Trans.
  Inf. Theory}, vol.~49, no.~10, pp. 2658--2668, Oct. 2003.

\bibitem{zhouperformanceII}
Y.~Zhou, V.~W.~S. Wong, and R.~Schober, ``Performance analysis of millimeter
  wave {NOMA} networks with beam misalignment,'' in \emph{Proc. IEEE Intern.
  Commun. Conf.}, May 2018, pp. 1--7.

\bibitem{Dolph1946}
C.~L. Dolph, ``A current distribution for broadside arrays which optimizes the
  relationship between beam width and side-lobe level,'' \emph{Proceedings of
  the IRE}, vol.~34, no.~6, pp. 335--348, Jun. 1946.

\bibitem{Duhamel1953}
R.~H. Duhamel, ``Optimum patterns for endfire arrays,'' \emph{Proceedings of
  the IRE}, vol.~41, no.~5, pp. 652--659, May 1953.

\bibitem{Drane1968DC}
C.~J. Drane, ``Useful approximations for the directivity and beamwidth of large
  scanning dolph-chebyshev arrays,'' \emph{Proceedings of the IEEE}, vol.~56,
  no.~11, pp. 1779--1787, Nov. 1968.

\bibitem{ShenFP}
K.~Shen and W.~Yu, ``Fractional programming for communication systems part {I}:
  Power control and beamforming,'' \emph{IEEE Trans. Signal Process.}, vol.~66,
  no.~10, pp. 2616--2630, May 2018.

\bibitem{Bogale2016RFChainNumber}
T.~E. Bogale, L.~B. Le, A.~Haghighat, and L.~Vandendorpe, ``On the number of
  {RF} chains and phase shifters, and scheduling design with hybrid
  analog-digital beamforming,'' \emph{IEEE Trans. Wireless Commun.}, vol.~15,
  no.~5, pp. 3311--3326, May 2016.

\bibitem{MIMO3GPP}
``Universal mobile telecommunications system ({UMTS}): Spatial channel model
  for multiple input multiple output ({MIMO}) simulations,'' 3GPP TR 25.996,
  Tech. Rep., 2014.

\bibitem{HaoEEmmWaveNOMA}
W.~Hao, M.~Zeng, Z.~Chu, and S.~Yang, ``Energy-efficient power allocation in
  millimeter wave massive {MIMO} with non-orthogonal multiple access,''
  \emph{IEEE Wireless Commun. Lett.}, vol.~6, no.~6, pp. 782--785, Dec. 2017.

\bibitem{HochwaldMassiveMIMO}
B.~M. Hochwald, T.~L. Marzetta, and V.~Tarokh, ``Multiple-antenna channel
  hardening and its implications for rate feedback and scheduling,'' \emph{IEEE
  Trans. Inf. Theory}, vol.~50, no.~9, pp. 1893--1909, Sep. 2004.

\bibitem{David2003order}
H.~A. David and H.~N. Nagaraja, ``Order statistics,'' \emph{NJ: John Wiley \&
  Sons}, vol.~7, pp. 159--61, 2003.

\end{thebibliography}

\end{document}